\documentclass[english,11pt,a4paper]{book}
\usepackage[utf8]{inputenc}
\usepackage[T1]{fontenc}
\DeclareUnicodeCharacter{2212}{-} 
\usepackage[english]{babel}
\usepackage{amsmath}
\usepackage{amsfonts}
\usepackage{fancyhdr}
\usepackage{amssymb}
\usepackage{color} 
\definecolor{Prune}{RGB}{99,0,60}
\usepackage{mdframed}
\usepackage{multirow} 
\usepackage{multicol} 
\usepackage{scrextend} 
\usepackage{tikz}
\usepackage{graphicx, subcaption}
\usepackage[absolute]{textpos} 
\usepackage{colortbl}
\usepackage{array}
\usepackage{geometry}
\usepackage{amsthm}
\usepackage[hidelinks]{hyperref}
\usepackage{bbm}
\usepackage{bm}
\usepackage{dsfont}
\usepackage{lipsum}
\usepackage{nicematrix}
\usepackage{emptypage}
\usepackage{pdfpages}

\usepackage[
  backend=biber,
  style=numeric-comp,
  sorting=none,
  url=false
]{biblatex}

\AtEveryBibitem{%
  \clearfield{urldate}%
  \clearfield{abstract}
}

\DeclareMathOperator{\Tr}{Tr}
\DeclareMathOperator{\TR}{TR}
\newtheorem{theorem}{Lemma}
\newcommand{\grafe}[1]{\left\{ #1 \right\}}
\newcommand{\tonde}[1]{\left( #1 \right)}
\newcommand{\quadre}[1]{\left[ #1 \right]}

\begin{document}



**Titre :** Exploration des paysages aléatoires de haute dimension : des verres de spins aux matrices aléatoires, en passant par des systèmes chaotiques simples

**Mots clés :** paysages d'énergie, verres de spin, matrices aléatoires, systèmes chaotiques, réseaux neuronaux aléatoires, théorie du champ moyen dynamique

**Résumé :** Les paysages aléatoires de haute dimension sous-tendent des phénomènes aussi divers que la physique des verres et l'optimisation en apprentissage automatique. Même leurs modèles jouets les plus simples présentent déjà un comportement d'une richesse extraordinaire. Cette thèse vise à approfondir notre compréhension de ce comportement, en combinant des approches statiques, via le formalisme de Kac–Rice, et des approches dynamiques, avec attention particulière aux systèmes à interactions réciproques et non réciproques. Après avoir passé en revue les techniques et résultats fondamentaux à travers le modèle sphérique à p-spin, cette thèse apporte trois avancées principales : (i) une comparaison exacte entre les approches dynamiques et statiques dans une classe résoluble de modèles avec interactions non-réciproques, mettant en évidence leurs différences et similitudes ; (ii) un calcul du nombre moyen de points fixes (i.e. *annealed complexity*) du réseau neuronal aléatoire de Sompolinsky–Crisanti–Sommers, pour tout degré de non-réciprocité ; (iii) deux approches pour calculer les barrières et la distribution des minima locaux profonds dans le paysage du modèle p-spin ; (iv) quelques résultats sur le produit scalaire entre vecteurs propres de matrices aléatoires corrélées et deformées, utiles pour explorer la géométrie des paysages d'énergie. Pris ensemble, ces résultats affinent notre compréhension de ces systèmes, tout en fournissant de nouveaux outils et en ouvrant de nouvelles perspectives de recherche.

**Title :** Exploring high-dimensional random landscapes: from spin glasses to random matrices, passing through simple chaotic systems

**Keywords :** energy landscapes, spin glasses, random matrices, chaotic systems, random neural networks, dynamical mean-field theory

**Abstract :** High-dimensional random landscapes underlie phenomena as diverse as glassy physics and optimization in machine learning, and even their simplest toy models already display extraordinarily rich behavior. This thesis aims to deepen our understanding of that behavior, by combining landscape based approaches, via the Kac–Rice formalism, with dynamical approaches, paying special attention to both systems with reciprocal and with non-reciprocal interactions. After surveying core techniques and results through the spherical p-spin model, this thesis delivers three main advances: (i) exact dynamic–static comparison in a solvable class of models with non-reciprocal interactions, pinpointing differences and similarities of the two approaches; (ii) a stability-based calculation of the mean number of fixed points (i.e. annealed complexity) of the Sompolinsky–Crisanti–Sommers random neural network, for any level of non-reciprocity; (iii) two approaches to probe the barriers and the distribution of deep local minima in the landscape of the p-spin model; (iv) some results on the overlaps among eigenvectors of spiked, correlated random matrices, which are useful to explore the geometry of energy landscapes. Together, these results sharpen our understanding of these systems, while providing new tools and opening new doors for future research directions.

# Contents



\tableofcontents

\chapter*{Aknowledgements}
\addcontentsline{toc}{chapter}{Aknowledgements}
\markboth{}{} 
Writing a thesis is like leaping into the unknown—you have no idea what awaits you at the bottom. If you did, it wouldn’t be called research. For this reason, it is essential to have a guide: a mentor, a reference point, someone who gives you wings to land gently instead of facing an almost certain crash.

Having someone who listens to you—and to whom you must also listen—should be a non-negotiable part of a PhD. For me, those people are Alberto and Valentina. I am more grateful to them than to anyone else, for their constant support, insightful feedback, dedication, and curiosity.

A thesis, however, is not just an individual effort—it is built on the strength of a community of students and researchers. First and foremost, I want to thank my collaborators: Giulio, Alessia, Samantha, and Pierfrancesco. Working with Samantha and Pierfrancesco over these last few months has been both productive and fun. I have learned so much from them, and their knowledge has shaped many pages of this manuscript. 

I would also like to thank Yan Fyodorov and Federico Ricci-Tersenghi for their careful reading of this manuscript and for their valuable advice, both concerning the improvement of the present thesis and possible directions for future research. I am also grateful to Cécile Monthus and Eric Bertin for their insightful comments and constructive discussions, which have highlighted very interesting perspectives for future work.

My warm thanks also go to all the members of the LPTMS. In particular, Claudine, Delphine, and Véronique deserve special recognition for their tireless help in navigating French bureaucracy. For their guidance and mentorship, I am especially grateful to Satya, Christophe, Martin, and Raoul. I also want to thank all the past and present PhD students and postdocs—Charbel, Ivan, Thomas, Marco, Federico G., Romain, Giuseppe, Pawat, Maximilien, Oscar, Gabriele, Benoit, Pietro, Lukas, Giovanni, Dimitrios, Vincent, Andrey, Luca, Alberto, Florent, Gianluca, Alice, Vincenzo, Sapt, Lara, Federico L., Ana.

A thesis would also be impossible without family—the people who listen patiently to both your complaints and your achievements. I am deeply thankful to my parents Alberto and Maria, for preparing the best \textit{pot de thèse} ever. I also thank my grandparents, Ivana, Salvatore, Fulvio, Fernanda and my closest relatives, Cinzia, Paola, Silvia, Fabio, Jacopo, Daniele, Simone, Cristian, Giulia, for always being there for me. 

To my sister Elena, thank you for all the laughter we have shared over these years - it has been truly precious. And the most special thanks goes to my wife Sara, whose emotional support and advice have been the most precious part throughout this journey.

Friends have been equally essential, especially on those gray and rainy Parisian November days. Thank you to Ayoub, Carlos, Yazmina, Emanuele, Sébastien, Irène, Alexis, Mo, Enel, and Jeanne. I am especially grateful to Matteo and Megi, who have felt like a brother and sister, and to Marco, for sharing this adventure and for always making lab days brighter (you know how).

Finally, a heartfelt thanks to those friends who remain close despite the distance\\—Michele R., Giacomo, Alberto, Michele D., Nicola, Camille, Raymond, Matteo, Paolo, Cossu, Alex, Guido. Your presence, even from afar, has always been felt.

\cleardoublepage

\chapter*{List of Publications}
\addcontentsline{toc}{chapter}{List of Publications}
The present thesis is based on the following articles:
\begin{itemize}
\item \cite{pacco2024curvature} Alessandro Pacco, Giulio Biroli, and Valentina Ros. Curvature-driven pathways interpolating between stationary points: the case of the pure spherical 3-spin model.
Journal of Physics A: Mathematical and Theoretical, 57(7):07LT01, 2024\\
\item \cite{paccoros} Alessandro Pacco and Valentina Ros. Overlaps between eigenvectors of spiked, correlated random matrices: From matrix principal component analysis to random gaussian landscapes. Phys. Rev. E, 108:024145, Aug 2023.\\
\item \cite{pacco_triplets_2025} Alessandro Pacco, Alberto Rosso, Valentina Ros. Triplets of local minima in a high-dimensional random landscape: correlations, clustering, and memoryless activated jumps. J. Stat. Mech. (2025) 033302.\\
\item \cite{us_non_reciprocal_2025} Samantha J. Fournier*, Alessandro Pacco*, Valentina Ros and Pierfrancesco Urbani. Non-reciprocal interactions and high-dimensional chaos: comparing dynamics and statistics of equilibria in a solvable model. Arxiv  2503.20908\\
\item \cite{pacco_scs_2025} Alessandro Pacco et al. Landscape analysis of random neural networks with excitatory
interactions: dynamics, complexity, marginal states and topology trivializations. \textit{In preparation}

\item \cite{pacco_quenched_triplets_2025} Alessandro Pacco, Alberto Rosso, Valentina Ros. Three-point complexity of the pure spherical $p$-spin model: the quenched calculation. \textit{In preparation}
\end{itemize}

\chapter*{Thesis summary - english}
\addcontentsline{toc}{chapter}{Thesis summary - english}
The goal of this thesis is to advance our understanding of high-dimensional random landscapes. These have become paradigmatic in the description of a large variety of complex phenomena, where the collective interaction of many units gives rise to interesting macroscopic behaviors. These landscapes find broad applications, including in the physics of glasses, theoretical neuroscience, machine learning, economics and theoretical ecology. Therefore, understanding the simplest toy models is a fundamental building block to achieve a complete theory of the phenomenology associated to high-dimensional random landscapes. This thesis is therefore about deepening our understanding of such simple models, that already show a rich and broad phenomenology. \\ 

In the first Chapter we introduce the physics of high-dimensional random landscapes, and of systems with non-reciprocal interactions. This is done by first introducing the notorious Kac-Rice formula, through a bird-eye-view of the main models to which it has been applied, notably random manifolds, spin glasses, complex interacting ecosystems, random neural networks, problems of inference, optimization and machine learning. We then review in quite detail a prototypical glassy model, i.e. the pure spherical $p-$spin. We conclude the first Chapter by summarizing all of our main contributions.\\

In the second Chapter we investigate a class of simple models with non-reciprocal interactions, by comparing thoroughly the solution of the dynamics with the distribution of equilibria. We show that these models are exactly solvable, thus providing the perfect playground to investigate the relationship between the landscape and the dynamics. \\

In the third Chapter we continue on the same venue of the previous one, by thoroughly analyzing the notorious "SCS" model of randomly interacting neurons. We propose a non-linear activation function that allows us to classify the number of equlibira in terms of their stability, for any degree of non-reciprocity. As above, we compare the dynamical with the statical results, showing where they agree, and where they don't. \\

In the fourth Chapter we explore the landscape of the pure spherical $p$-spin model. This model has been studied extensively in the past few decades. We expand upon those works, by proposing two approaches to probe the barriers and the distribution of fixed points in the deep part of the landscape, which is dominated by exponentially many minima.\\

In the fifth Chapter we study a spiked matrix problem that arises in the context of determining the barriers between local minima in Chapter 4. The goal is to compute the overlaps between eigenvectors of spiked, correlated GOE random matrices.


\chapter*{Résumé de la thèse – français}
\addcontentsline{toc}{chapter}{Résumé de la thèse – français}

L’objectif de cette thèse est de faire progresser notre compréhension des paysages aléatoires de haute dimension. Ces systèmes sont devenus paradigmatiques dans la description d’un large éventail de phénomènes complexes, où l’interaction collective d’un grand nombre d’unités engendre des comportements macroscopiques riches et non triviaux. Ces paysages trouvent des applications variées, notamment en physique des verres, en neurosciences théoriques, en apprentissage automatique, en économie et en écologie théorique. La compréhension des modèles jouets les plus simples constitue ainsi une étape fondamentale vers une théorie complète de la phénoménologie associée aux paysages aléatoires de haute dimension. Cette thèse vise donc à approfondir l’étude de tels modèles simples, qui présentent déjà une phénoménologie particulièrement riche et diversifiée. \\

Dans le premier chapitre, nous introduisons la physique des paysages aléatoires de haute dimension ainsi que celle des systèmes à interactions non réciproques. Nous présentons tout d’abord la célèbre formule de Kac-Rice, à travers une vue d’ensemble des principaux modèles auxquels elle a été appliquée : variétés aléatoires, verres de spins, écosystèmes complexes, réseaux neuronaux aléatoires, ainsi que divers problèmes d’inférence, d’optimisation et d’apprentissage automatique. Nous passons ensuite en revue, de manière détaillée, un modèle prototypique : le modèle sphérique pur à $p$-spins. Le chapitre se conclut par un résumé de nos principales contributions. \\

Le deuxième chapitre est consacré à l’étude d’une classe de modèles simples à interactions non réciproques, pour lesquels nous comparons de manière approfondie la solution dynamique à la distribution des équilibres. Nous démontrons que ces modèles sont exactement solubles, constituant ainsi un terrain idéal pour explorer la relation entre la structure du paysage et la dynamique. \\

Dans le troisième chapitre, nous poursuivons cette même approche en analysant en détail le célèbre modèle "SCS" de neurones en interaction aléatoire. Nous proposons une fonction d’activation non linéaire permettant de classifier le nombre d’équilibres en fonction de leur stabilité, pour tout degré de non-réciprocité. Comme précédemment, nous comparons les résultats dynamiques et statiques, en mettant en évidence leurs points d’accord et de divergence. \\

Le quatrième chapitre est dédié à l’exploration du paysage du modèle sphérique pur à $p$-spins, un modèle abondamment étudié au cours des dernières décennies. Nous étendons ces travaux en introduisant deux approches distinctes permettant de sonder les barrières énergétiques et la distribution des points fixes dans les régions profondes du paysage, dominées par un nombre exponentiel de minima locaux. \\

Enfin, dans le cinquième chapitre, nous étudions un problème de matrice aléatoire déformée (« spiked matrix problem »), qui émerge dans le cadre de l’analyse des barrières entre minima locaux présentée au chapitre 4. L’objectif est de calculer le produit scalaire entre vecteurs propres de matrices aléatoires GOE corrélées et déformées.

\chapter{Introduction}
\label{chapter:intro}
This Chapter serves as an accessible introduction to the physics of high-dimensional random landscapes. In Sec.~\ref{sec:intro_motivation} we present a high-level motivation to the field of random landscapes and give some historical accounts on its origins, points of interest, and recent developments. In Sec.~\ref{sec:topo_complexity} we give more insight into the main object of study of this thesis, the complexity, and introduce the main tools for its computation. In Sec.~\ref{sec:p_spin_model} we develop the theory of the simplest high-dimensional random energy landscape, the \textit{pure spherical p-spin model}, and we comment on mixed models, which are generalizations of it. Finally in Sec.~\ref{sec:sumamry_contributions} we resume our main contributions of the manuscript and point out open problems.

\section{Motivation: counting equilibria}
\label{sec:intro_motivation}
The common object that brings together all the topics treated in this thesis is the Kac-Rice formula. We shall therefore devote this first section to it. In short, the Kac-Rice formula is a way of counting the number of solutions to a system of $N\in\mathbb{N}$ equations of the variable ${\bf x}\in\mathcal{M}$, with $\mathcal{M}\subset\mathbb{R}^N$ a manifold. The first appearance dates back to Marc Kac in 1943 \cite{Kac_roots_43}, where he introduced a way to count the expected number of real roots $E_N$ for a polynomial of the form:
\begin{align}
p_N(x):=\sum_{k=0}^Na_kx^k,\quad a_k\sim\mathcal{N}(0,1),
\end{align}
with $\mathcal{N}(\mu,\sigma)$ denoting the Gaussian distribution of mean $\mu$ and variance $\sigma^2$. More recent work \cite{Edelman_roots_95} gives the asymptotic expression of $E_N$ up to order $1/N^2$:
\begin{align}
    E_N=\frac{2}{\pi}\log(N) + C +\frac{2}{N\pi}+\mathcal{O}(1/N^2),\quad \quad C=0.6257358072...
\end{align}
which is already very accurate for small values of $N$, as we can see from Fig.~\ref{fig:roots_poly}.

\begin{figure}[t!]
\centering
\includegraphics[width=0.7\textwidth]{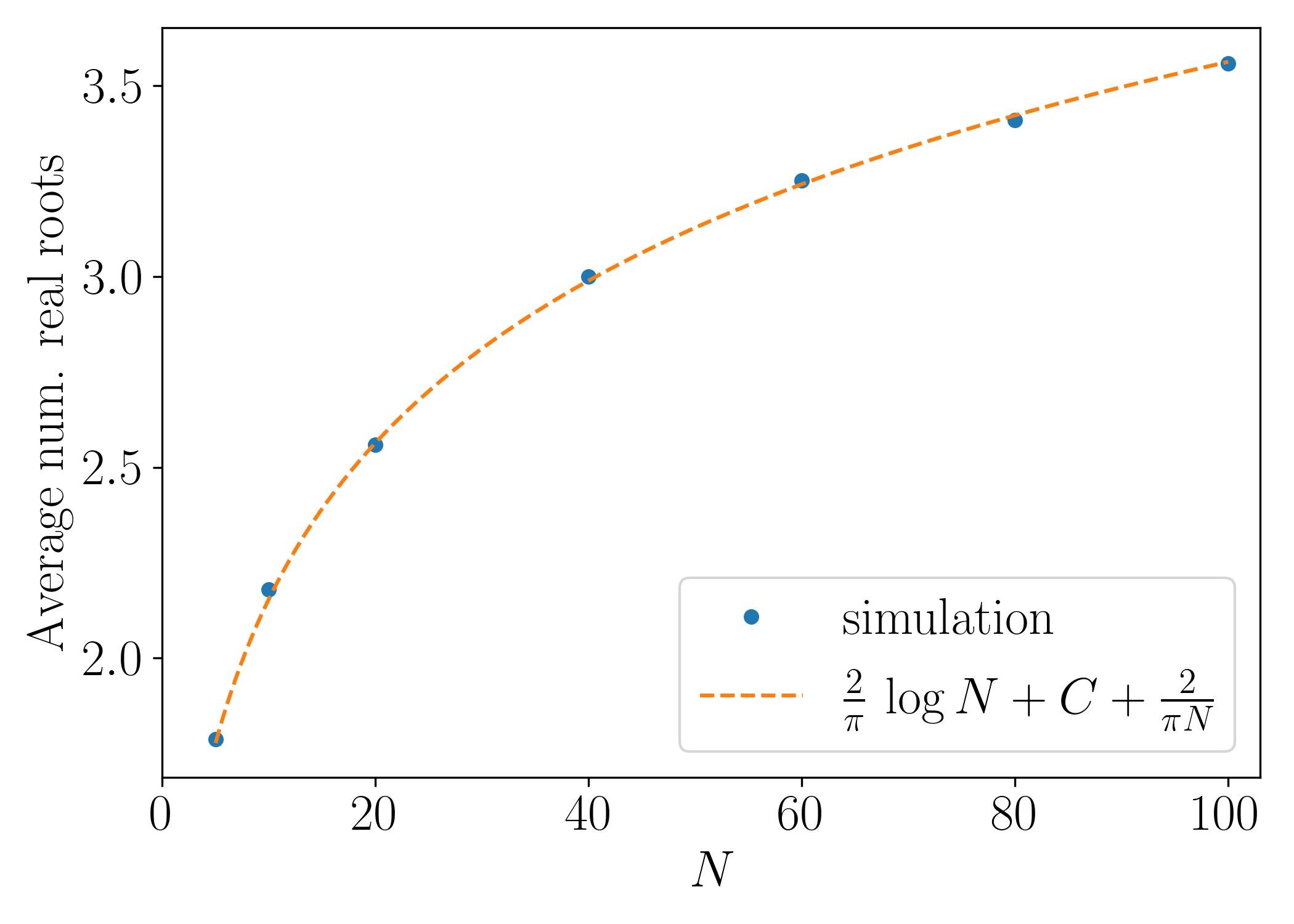}\\
\caption{Comparison of $E_N$ with numerical simulations of the number of real roots of random polynomials $p_N(x)$. For each $N$ we sample $4000$ different polynomials and average their number of real roots (found with Python's method \textit{numpy.root}).}
\label{fig:roots_poly}
\end{figure}

\noindent According to \cite{brezin_kac_2022}, the name "Kac-Rice" has been introduced only later in Ref.\cite{Cramer_stationary_67}, after works of Stephen O. Rice \cite{Rice_44} on the mean number of zeros of random Gaussian function of a single variable. The Kac-Rice formula is the subject of intense work in the mathematics community \cite{brezin_kac_2022, random_fields_adler}, and rigorous proofs often require many pages of hard work. However, for the present analyses we are interested in the theoretical physics approach, where the Kac-Rice formula is used mainly to count stationary points of high-dimensional (often Gaussian) random fields; see also these reviews \cite{ros2022high, Fyod_highd_rmt_2013, ros_lecture_2025}. In particular, consider random functions of the following form: 
\begin{align}
    {\bf F}=(F_1,\ldots,F_N),\quad\quad F_i:\mathbb{R}^N\to \mathbb{R},
\end{align}
then we are interested in counting the number of solutions ${\bf x}\in\mathcal{M}\subseteq\mathbb{R}^N$ to ${\bf F}({\bf x})={\bf 0}$ for $\mathcal{M}$ a manifold. If we denote by $D$ the set of solutions and by $\mathcal{N}_N:=|D|$ its number, then we can write:\footnote{here we are tacitly assuming that the set of solutions of the system of equations consists of isolated points in $\mathbb{R}^N$. This is usually the case for the Gaussian random systems considered in the following, but the presence of symmetries may invalidate such assumption; see for example in the case of neural networks \cite{berfin_geometry_2021, martinelli_flat_channels_2025}.}
\begin{align}
\begin{split}
\mathcal{N}_N&=\sum_{{\bf x}_\alpha\in D}1=\sum_{{\bf x}_\alpha\in D}\int_{\mathcal{M}}d{\bf x}\,\delta({\bf x}-{\bf x}_\alpha)=\int_{\mathcal{M}}d{\bf x}\,\sum_{{\bf x}_\alpha\in D}\delta({\bf x}-{\bf x}_\alpha)\\
&=\int_{\mathcal{M}}d{\bf x}\,\delta({\bf F}({\bf x}))\,|\det J({\bf x})|
\end{split}
\end{align}
where we used the Dirac delta's property $\delta({\bf F}({\bf x}))\,|\det J({\bf x})|=\sum_{{\bf x}_\alpha\in D}\delta({\bf x}-{\bf x}_\alpha)$ with $[J({\bf x})]_{ij}=\partial_iF_j({\bf x})$ the Jacobian. The name "Kac-Rice" is commonly used as a reference to the formula for the first moment $\mathbb{E}[\mathcal{N}_N]$ (with $\mathbb{E}$ the average over the randomness), which can be expressed as:
\begin{align}
\label{eq:def_annealed_kr_intro}
\begin{split}
\mathbb{E}[\mathcal{N}_N]&=\int_{\mathcal{M}}d{\bf x}\,\mathbb{E}[\delta({\bf F}({\bf x}))]\,\,\mathbb{E}\left[|\det J({\bf x})| \bigg| f_i({\bf x})=0,\forall i \right]\\
&=\int_{\mathcal{M}}d{\bf x}\,P_{{\bf F}({\bf x})}({\bf 0})\,\,\mathbb{E}\left[|\det J({\bf x})| \bigg| f_i({\bf x})=0,\forall i \right]\\
\end{split}
\end{align}
with $P_{{\bf F}({\bf x})}({\bf 0})$ the joint density of the variables $F_i({\bf x})$ evaluated at 0, and $\mathbb{E}[\cdot|\cdot]$ a conditional expectation value. \\

\noindent A commonly studied object in the physics and mathematics literature is the \textit{annealed topological complexity} (or simply annealed complexity), defined as 
\begin{align}
\label{eq:motivation_ann_kac}
    \Sigma_A:=\lim_{N\to\infty}\frac{\log \mathbb{E}[\mathcal{N}_N]}{N}
\end{align}
where for $\log$ is intended the natural logarithm. This is often compared with the so called \textit{quenched complexity}, that we will introduce in Sec.~\ref{sec:topo_complexity}. Physicists are interested in what are normally referred to as \textit{energy landscapes}, which are characterized by an energy (or Hamiltonian) $\mathcal{E}$ that is often modeled as a disordered potential or as the energy resulting from the (random) interaction of many units (or spins) in high-dimensions. In the study of Gradient Descent or Langevin Dynamics, the gradient $\nabla \mathcal{E}$ represents a fundamental object of interest, since minima of the energy are basins of attraction of the dynamics. The Kac-Rice formula is then important as a tool to count and classify such attractors of the dynamics, by choosing ${\bf F}=\nabla \mathcal{E}$ in Eq.~\eqref{eq:def_annealed_kr_intro}. Those points where $\nabla \mathcal{E}({\bf x})={\bf 0}$ are referred to as \textit{stationary} (or \textit{fixed}) points, or sometimes as \textit{equilibria}. In particular, a positive value of $\Sigma_A$ indicates that the \textit{mean} value of the number of stationary points grows exponentially with $N$ as $\mathbb{E}[\mathcal{N}_N]\sim e^{N\Sigma_A}$. The \textit{mean} value may differ from the \textit{typical} value of $\mathcal{N}_N$ when $N$ is large (see Sec.~\ref{sec:topo_complexity}), however many of the systems we will introduce below and study in the rest of the manuscript possess an exponential, both mean and typical, number of stationary points. When counting stationary points of an energy landscape, the Jacobian $J$ presented above simply reduces to the \textit{Hessian} at that point, which encodes for the stability of the point, making it either a minimum, a saddle or a maximum. In most situations one can introduce an index (say $\lambda$) in the Kac-Rice formula, used to classify the various stationary points of the energy landscape by their instability, that is, by the number of negative eigenvalues of the Hessian. Counting and classifying the stationary points of an energy landscape is an interesting mathematical question for its own sake, but with important applications in several fields of research such as glasses and spin glasses, machine learning and deep learning, inference and optimization, economics, neuroscience, ecology and condensed matter theory. In short, every system whose dynamics can be seen as an optimization or motion in an energy (or cost or loss or fitness) landscape is a subject of interest and study for the Kac-Rice formula. Although the complexity is not the only interesting quantity about an energy landscape, it surely is an important indicator of its \textit{ruggedness} (or \textit{glassiness}), that is, of its non-convexity. This has important consequences for optimization algorithms, since finding the global minimum can be hard in a rugged landscape, but easy in a convex one. \\

\noindent Below we give a bird-eye-view of the field of random landscapes, giving particular attention to the uses of the Kac-Rice formula in statistical physics, from its first uses to its most recent developments, expanding on the ideas given above. Hopefully, the list of references is representative, although certainly not exhaustive given the breadth of these fields of research. \\

\subsection{The Kac-Rice for energy landscapes}
\noindent \textbf{Spin glasses.} We have already mentioned that the first appearance of formulas for counting zeros of random functions dates back to the '40s, by seminal works of Kac and Rice \cite{Kac_roots_43, Rice_44, Kac_roots_48}. It seems that the first appearance of the "Rice formula" in the context of physics dates back to 1956 in Ref.~\cite{Higgins}, where the author studies the statistical properties of a random, moving Gaussian surface, as a model of ocean waves. However the Kac-Rice formula has gained a lot of popularity in the statistical physics community, and in particular in spin glass theory, where it has become the paradigmatic tool to count minima of high-d energy landscapes. The first appearance of the (configurational) \textit{complexity} was due to Bray and Moore \cite{bray1980metastable}, to count the number of stable solutions of the TAP free energy of the Sherrington-Kirkpatrick (SK) model \cite{SK75, TAP_77} at fixed values of the free energy and of the temperature. A problem that was noticed in those earlier works concerns the absolute value of the determinant present in Eq.~\eqref{eq:def_annealed_kr_intro}; indeed without it, the expression provides a topological invariant called Euler characteristics from Morse theory \cite{random_fields_adler}. Despite this, in Ref.~\cite{bray1980metastable} the authors were able to overcome this issue by noting that in some models an omission of the modulus may be legit in the large $N$ limit at low energies, where the landscape is expected to be dominated by local minima. The determinant was then computed by integral representations. 
The first time that the complexity appeared as a property of the topology of the landscape (and thus, temperature independent) was in a paper by Cavagna, Giardina, Parisi \cite{cavagna1998stationary}. This paper considers the landscape of the simplest prototypical Gaussian landscape, the \textit{pure spherical p-spin model} \cite{Rem_1, gross1984simplest, Thirum_pspin_87, crisanti1992sphericalp, Crisanti_TAP_pspin_95}, putting on firmer grounds the definition of complexity. The determinant is still computed by means of integral representations, but the authors remark that the Hessian spectrum can be computed from random matrix theory (RMT). In particular, their approach allows to characterize local minima but also saddles with a finite number of negative eigenvalues, thus improving previous works \cite{Crisanti_TAP_pspin_95, Rieger_TAP_92, kurchan_barriers_93}. The picture emerging from \cite{cavagna1998stationary} is that below a certain \textit{threshold} energy the total complexity actually corresponds to the complexity of local minima, while saddle points with a \textit{finite} number of negative modes in the Hessian still have a positive (but lower) complexity. Since the number of stationary points counted grows exponentially with the complexity, this explains why below the threshold the complexity is dominated by the local minima. Instead, above the threshold, the landscape is dominated by saddles with an extensive number of negative modes \footnote{that is, $\sim\mathcal{O}(N)$} in their Hessian. In another article \cite{cavagna1999quenched} Cavagna and collaborators also compute the quenched complexity of the $p$-spin model with external magnetic field. In \cite{cavagna1997investigation} the \textit{constrained complexity} (or two-point complexity) is introduced, a tool that allows to obtain the structure of local minima of the landscape, by computing the distribution of local minima at overlap $q$ (a measure of distance that we will explore later) from a reference one. This followed previous works in which the same question was answered, but among equilibrium configurations at a given temperature \cite{franz1995recipes, Barrat_bifurcation_95, franz1998effective, cavagna1997structure}. Two decades later, the constrained complexity is still the object of numerous works \cite{pacco_triplets_2025, pacco2024curvature,ros2019complexity, ros2020distribution, ros2021dynamical, kent2024arrangement}, and we will talk about it extensively in Chapter~\ref{chapter:energy_landscapes}, when we discuss the three-point complexity \cite{pacco_triplets_2025}. \\

\noindent Pure models are the simplest instances of a class of \textit{mixed} spherical models \cite{MixedModelNieuwenhuizen95,Barrat_bifurcation_95, Leuzzi_mixed_2004, Leuzzi_mixed_2006, benarous_mixed_2019, folena2020rethinking}, which can be seen as sums of pure $p$-spin models with different $p$'s. In the physics community, the study of the energy landscape of mixed models has gained recent attention ~\cite{Kent_Dobias_typical_2023, kent2024arrangement, kent_rsb_2023}, since the publication of Ref.~\cite{folena2020rethinking}, where the authors compare the dynamics of these models with the complexity of the energy landscape, finding striking differences from pure models. \\

\noindent \textbf{Random manifolds.} While the links between random energy landscapes and RMT were already pointed out in \cite{cavagna1998stationary}, it was Fyodorov \cite{Fyod2004} to first cast the computation of the complexity for a certain class of random Gaussian landscapes into a RMT problem. The random matrix is, clearly, the Jacobian (or Hessian for energy landscapes) in Eq.~\eqref{eq:def_annealed_kr_intro}. The novel approach consists in evaluating directly the expected value of the determinant by leveraging the probability distribution of the random matrix, without using any integral representations. This usually proves simpler in the annealed case, since the conditioning does not affect the distribution. The model studied in \cite{Fyod2004} consists of a particle moving in a quadratically confined potential with Gaussian disorder, through the energy function:
\begin{align}
\label{eq:def_d=0_random_mfd}
    \mathcal{E}({\bf x})=\frac{\mu}{2}{\bf x}^2+V({\bf x}) ,\quad{\bf x}\in\mathbb{R}^N,\quad\,\mu>0
\end{align}
with $V$ a zero-mean isotropic Gaussian potential with correlations that depend only on the (Euclidean) distance:
\begin{align}
    \mathbb{E}[V({\bf x}_1)V({\bf x}_2)]=Nf\left( \frac{1}{2N}||{\bf x}_1-{\bf x}_2||^2\right),\quad f''(0)>0.
\end{align}
The result for the annealed complexity of stationary points as $N\to\infty$ reads:\footnote{"A" stands for annealed}
\begin{align}
    \Sigma_A(\mu)=H(\mu-J)\left[\frac{1}{2}\left(\frac{\mu^2}{J^2}-1 \right)-\ln(\mu/J)\right],\quad J=\sqrt{f''(0)},
\end{align}
with $H$ the Heaviside step function. This result indicates that the landscape has an exponential number of stationary points for $\mu<\mu_c\equiv J$ (it is said to be \textit{glassy} or \textit{hard}) and it is said to be \textit{easy} for $\mu\geq J$ (that is, the number of stationary points is sub-exponential since the complexity is zero). This transition is commonly referred to as \textit{topology trivialization} \cite{Fyod_simplest_optimization_2013, Fyodorov_2016}, and it will be important when we discuss the results of Chapters \ref{chapter:non_reciprocal} and \ref{chapter:scs}. This analysis has been extended to count stationary points with fixed energy and index (that is, number of negative Hessian eigenvalues) in the case $\mu=0$ \cite{Bray_large_dim_2007}. These results were then improved in \cite{Fyodorov_Williams_2007}, where the authors compute the complexity of stationary points with a given index for $\mu>0$. Moreover, when $f$ is such that it decays to zero for large arguments, they identify the vanishing of the complexity with the instability of the replica symmetric (RS) solution of the associated statistical mechanical problem. In a subsequent work \cite{Fyod_nadal_2012} the scaling and behavior of the annealed complexity of minima close to the trivialization transition is also considered, by leveraging results on the large deviations of the largest eigenvalue of GOE (Gaussian Orthogonal Ensemble) matrices \cite{Tracy_Widom_1996, Borot_max_eval_2011, Nadal_right_tail_TW_2012}. Notice that the statistics of extrema can also be studied in the 1-dimensional case (i.e. $N=1$) \cite{Monthus_1d_landscape_2003}. \\

\noindent Let us stress that the RMT approach has now become very popular in the statistical physics community, as it allows one to extract information on isolated eigenvalues that appear when constraining two local minima to be at fixed overlap \cite{ros2019complex, ros2019complexity, pacco2024curvature}. A particularly important recent work that merges calculations of the complexity of energy landscapes with RMT is Ref.~\cite{ros2019complex}. The model they study is the energy landscape of the pure spherical $p$-spin model plus a term that favors all configurations that are
close to a given one \cite{Gillin_Sherrington_2000}, a model that they name "generalized spiked tensor model", since it has as a special case the spiked-tensor model studied in the context of inference and high-dimensional statistics \cite{Montanari_tensor_pca_2014, Montanari_landscape_spiked_2018}.\\

\noindent The landscape defined in Eq.~\eqref{eq:def_d=0_random_mfd} falls into the category of \textit{random elastic manifolds} \cite{Rosso_review_2021}. The most general Hamiltonian describing an elastic manifold in a random potential can be written as \cite{ros2022high, Cugliandolo_mfd_out_eq_1996, Mezard_Parisi_mfd_1991}:
\begin{align}
&\mathcal{E}[{\bf u}]=\int_\mathcal{M} d{\bf x}\left[V({\bf u}({\bf x}),{\bf x})+\frac{\kappa}{2}(\nabla {\bf u}({\bf x}))^2+\frac{\mu}{2}{\bf u}({\bf x})^2\right]\\
&\mathbb{E}[V({\bf u},{\bf x})]=0,\quad\mathbb{E}[V({\bf u},{\bf x})V({\bf u}',{\bf x}')]=N\delta({\bf x}-{\bf x}')f\left(\frac{||{\bf u}-{\bf u}'||^2}{N}\right)
\end{align}
where $\mathcal{M}\subset \mathbb{R}^d$, ${\bf u}({\bf x})\in\mathbb{R}^N$. Physically, this represents the generalization to arbitrary dimensions of an elastic line in a random potential, where the function ${\bf u}$ represents the displacement field at every position in space. The term multiplying $\kappa$ represents an elastic energy, the $\mu$ term is a confining potential and $V$ is the quenched disorder, modeled by a random potential energy which couples directly to ${\bf u}$. The case $N\to\infty $ with $d=0$ corresponds exactly to \eqref{eq:def_d=0_random_mfd}, which can be seen as a "toy model" of a single particle in a random potential, thus without elastic interaction. The annealed complexity 
of this model for $d=N=1$ was computed in \cite{fyod_doussal_texier_mfd_2018}, and for $d\leq4$ with $N\to\infty$ in \cite{fyod_doussal_mfd_2020}. In \cite{fyod_doussal_texier_mfd_2018} the authors consider, moreover, as an application of their results the depinning transition in the presence of a uniform applied force to the manifold. The critical value of the depinning threshold is identified as the force beyond which no metastable state survives. This can therefore be put in the framework of topology trivialization phenomena. However, the authors of \cite{fyod_doussal_texier_mfd_2018} show that the critical force predicted from the annealed computation is only an upper bound to the value inferred from numerical simulations. \\


\noindent Let us finally mention other directions where the Kac-Rice formula has been employed to study energy landscapes in recent years. In \cite{Lacroix_Fyodorov_Fedeli_2022} the authors compute the annealed complexity for a new class of models of the form in Eq.~\eqref{eq:def_d=0_random_mfd} where $V$ is made from a superposition of random plane waves: the model does not present a topology trivialization transition for any finite value of the strength of confinement but decreases rapidly as it is increased. The complexity of models of superposition of random plane waves was then extended in \cite{lacroix_fyod_waves_2024}. In \cite{kent_complex_complex_2021} a generalization of the complexity of spherical spin glasses is given in the context of complex variables, and in \cite{fyod_non_herm_2025} the Kac-Rice is used as an alternative to the Hermitization method to compute statistics of eigenvalues and eigenvectors of non-Hermitian random matrices. Let us mention also that the theory of complexity of random energy landscapes developed in the statistical physics community has gained a lot of recent interest in the mathematics community, with extensive and frequently successful efforts to put many of these results on rigorous grounds \cite{auffinger_complexity_2013, Auffinger_Arous_Černý_spin_2012, subag2017complexity, benarous_mfd_2024, Auffinger_saddles_2020, benarous_detemrinants_2022, Subag_Zeitouni_2021}. \\

\noindent \textbf{Optimization, inference, machine learning.}
Above, we have discussed the use of the Kac-Rice formula to count stationary points of energy landscapes, and we have mainly mentioned physics-inspired problems, such as spin glasses and random manifolds. However, the framework of energy landscapes is rather generic and includes any possible high-dimensional function that has to be optimized. In particular, these ideas are useful in the context of optimization, inference and, more recently, machine learning. The application of counting techniques to the "simplest random optimization" problem was done first in \cite{Fyod_simplest_optimization_2013}, where the authors study the topology trivialization of the number of equilibria in various scaling regimes the spherical $2-$spin model with magnetic field. In the context of neural networks instead, one can see the loss function as an energy landscape, with the data being drawn from a probability distribution, and thus corresponding to the quenched disorder. The problem is sometimes cast into the so-called "teacher-student" framework, where the student (neural network 2) has to learn the weights of the teacher (neural network 1). What makes the complexity calculation hard is that, even if the data is treated as Gaussian, the non-linearity usually adopted as an activation function in neural network renders the problem non-Gaussian. Recent attempts were made in the case of perceptron models \cite{tsironic_perceptron_kac_2025, maillard_landscape_2020}, and implicit formulas for the annealed complexity can be obtained. However, how to obtain explicit formulas and study the complexity numerically remains an open problem. We mention that advances have been made in simpler optimization problems, such as in constrained random least-square problems \cite{Fyodorov_least_sqaure_2022}, simple non-Gaussian problems \cite{kent_algo_marg_2024}, and in inference problems such as matrix and tensor spiked models \cite{mannelli_afraid_2019, mannelli_spiked_matrix_2019, ros2019complex, Montanari_landscape_spiked_2018}. To the best of our knowledge, a calculation of the complexity of, say, two layer fully connected networks with Gaussian data and non-linear activation functions is still missing (although it might appear soon \cite{montanari_erm_localmin_2025}).

\subsection{The Kac-Rice for non-gradient systems}
\label{subsec:kac_rice_non_grad}
\noindent Another important branch of statistical physics that has attracted a lot of attention, especially in neuroscience, ecology, biology, and economics, concerns complex systems that have non-reciprocal interactions or, more generally, whose driving force does not come from the gradient of an energy function. These are naturally occurring in nature, where various agents (neurons, cells, species, firms etc.) interact in an asymmetric way. In their simplest approximation, these systems are modeled as SDEs (Stochastic differential equations) of the form:
\begin{align}
\label{eq:intro_non-gradient}
    \frac{d{\bf x}}{dt}={\bf F}({\bf x})+{\bm\eta}, \quad {\bf x}\in\mathbb{R}^N,
\end{align}
where ${\bf F}$ is a (non-linear) Gaussian random field and ${\bm\eta}$ an external noise. When ${\bf F}$ can be written as the gradient of some energy function, we can speak of energy landscapes, since the force pushes the system towards minima of the energy function. However, even when ${\bf F}$ does not come from the gradient of an energy, the system may still possess many fixed points such that $d{\bf x}/dt=0$, but whether (or rather, when) these are actually attractors of the dynamics is debated \cite{wainrib2013topological, AnnibaleDynamics2024, HeliasFP2022, HuangSteadyState2024,fyodorov2016nonlinear, ben2021counting, Fyodorov_resilient_2021}. This is a question that we will tackle in Chapters \ref{chapter:non_reciprocal} and \ref{chapter:scs} by comparing explicitly the Kac-Rice complexity with a Dynamical Mean-Field Theory approach, see \cite{us_non_reciprocal_2025, pacco_scs_2025}.\\

\noindent One of the first works on large complex systems with random interactions dates back to the 1970s, from the seminal work by Robert May \cite{MAY_1972}, who wanted to show that, depending on the strength of (random) interaction among $N$ agents, one can have a transition from a stable regime to an unstable regime of the system. In general, this is the consequence of the interaction matrix being random and of the \textit{self-averaging} property of the spectrum. Indeed, (common) random matrix ensembles have spectra that converge to well-defined, bounded regions in the complex plane \cite{mehta_rmt, tao_circularlaw, potters_bouchaud_2020, sommers1988spectrum, Khoru_Sommers_non_herm_2009}. By changing a control parameter (such as the interaction strength among agents), one can shift the asymptotic values of the minimal (and maximal) eigenvalues, eventually crossing the instability threshold. \\

\noindent These large complex systems are generally referred to as \textit{non-gradient systems} or \textit{systems with non-reciprocal interactions} because it can be shown that, in general, they arise when the various units interact in an asymmetric way. The literature on these types of models is vast, and we shall only concentrate on those works where the Kac-Rice formula to count stationary points is used. To the best of our knowledge the first account on the use of the Kac-Rice formula to compute the complexity of stationary points of systems with asymmetric interactions is Ref.~\cite{wainrib2013topological}, where the authors analyze the model of random neural network proposed by Sompolinsky, Crisanti, Sommers in \cite{ChaosSompo88}, with a force ${\bf F}$ given by:\footnote{see also Chapter~\ref{chapter:scs}.}
\begin{align}
    {\bf F}({\bf x})=-{\bf x}+g\sum_{j}\tanh(x_j)J_{ij}, \quad J_{ij}\overset{iid}{\sim}\mathcal{N}(0,1/N).
\end{align}
The authors in \cite{wainrib2013topological} were not able to obtain an explicit formula for the complexity, but just an approximation, due to the difficulty in treating the associated random matrix problem as well as the non-linearity present in the model. However, the authors observe a topology trivialization of the annealed complexity as the interaction strength (i.e. $g$) crosses the specific value $g=1$, and they argue that the emergence of chaos (as observed in \cite{ChaosSompo88}) precisely corresponds to this transition in the explosion of unstable stationary points. \\

\noindent Shortly later, another work by Fyodorov and Khoruzhenko \cite{fyodorov2016nonlinear} considers an extension of May's work, by introducing a simple nonlinear model of large ecosystem, where the force is now made up of a sum of a conservative irrotational component (curl-free) and a curl field (divergence-free). Their model corresponds to Eq.~\eqref{eq:intro_non-gradient} with ${\bf F}({\bf x})=-\mu\,{\bf x}-\nabla V({\bf x})+\nabla\cdot{\bf A}({\bf x})/\sqrt{N}$ with $V$ a scalar and ${\bf A}$ a (antisymmetric) vector potential, both assumed to be independent Gaussian (zero-mean) random fields with a homogeneous and isotropic covariance structure. The authors were able to compute the mean of the total number of fixed points, and observe again a trivialization phenomenon: such systems exhibit a transition from a trivial
phase with a single stable equilibrium to one characterized by exponentially many. The analysis of \cite{fyodorov2016nonlinear} was extended in \cite{ben2021counting} to account for a statistical analysis of stability properties of equilibria. They count both the mean total number of stationary points and the total number of stable ones. Depending on the values of two control parameters, they obtain a regime of \textit{absolute stability}, where a single stable equilibrium is found; a regime of \textit{relative instability}, where there are exponentially many stable equilibria, which are, however, exponentially rare among all equilibria; and a regime of \textit{absolute instability}, with exponentially many unstable equilibria and exponentially rare stable ones. We have a similar portrait in Chapter~\ref{chapter:non_reciprocal} (in particular Sec.~\ref{sec:rnn_alpha>0}), where moreover we compare with dynamical results. Extensions of this work also include \cite{Fyodorov_resilient_2021, lacroix2022counting}.
\noindent Besides this, the analysis of both equilibria and dynamics of large complex ecosystems has drawn a lot of attention recently, especially in the context of Lotka-Volterra equations \cite{ros2023generalized, RosEcoQuenched2023, arnoulx2024many, BiroliLVMarginal2018}.\\

\noindent In the context of random neural networks, many works concentrate on the study of the dynamical properties \cite{ChaosSompo88, crisanti_path_2018, MastrogiuseppeLink2018, Sompo_transition_2015, Abbott_lyapunov_2023, HeliasMemory18}, but little work has been done on the Kac-Rice computation of equilibria; to the best of our knowledge only \cite{wainrib2013topological, HeliasFP2022} and our soon to be published work \cite{pacco_scs_2025}, presented in Chapter~\ref{chapter:scs}, where we count fixed points as a function of several order parameters, including the (extensive) instability index. Besides this, other works concentrate on simpler models, such as \cite{Fyodorov_2016}, where a non-gradient spherical model with an external magnetic field is considered (see also earlier works \cite{CugliandoloNonrelax97, crisanti1987dynamics}) and the annealed complexity is computed by incorporating the Lagrange multiplier into the Kac-Rice framework; see also subsequent works \cite{Kivimae2024, garcia2017numberequilibriagivennumber}. One of the main results of Ref.~\cite{Fyodorov_2016} is that also this model presents a topology trivialization transition as the variance of the (random) external field is increased, from a regime with an exponential number of stationary points to just two. In our recent works \cite{us_non_reciprocal_2025, pacco_scs_2025} we try to expand upon the aforementioned literature: we consider general scenarios, and make explicit comparisons with the corresponding solutions of the dynamical mean-field theory. An extension of \cite{Fyodorov_2016} accompanied by comparisons with the dynamics is done in Chapter~\ref{chapter:non_reciprocal}, while in Chapter~\ref{chapter:scs} we overcome the obstacles found in \cite{wainrib2013topological}, by using a specific realization of the non-linearity that allows to obtain explicit results as a function of the instability index for both the annealed and quenched (Replica Symmetric) complexities. \\




\section{Topological complexity of landscapes}
\label{sec:topo_complexity}
In the previous section we have already given a flavor on what the topological complexity represents, see Eq.~\ref{eq:motivation_ann_kac}. Here we shall be more precise, summarizing some concepts also found in recent reviews~\cite{ros2022high,ros_lecture_2025}. As we have already mentioned, what motivates us is a set of problems related to counting fixed points of random functions. We can identify this set of problems under the name of the \textit{landscape paradigm} \cite{ros2022high}. By this we mean to understand the properties of the dynamics of complex systems by looking at the underlying \textit{energy landscape} (or loss, cost, fitness, etc.) which has to be optimized according to some algorithm, the most common being gradient descent. As already anticipated, we shall also discuss
problems that go beyond the landscape paradigm, by studying fixed points of non-gradient complex systems.

\subsubsection{Average vs typical, annealed vs quenched, replicas}

\noindent One central object of analysis is the so-called topological complexity (or simply complexity for the present manuscript), which we have already introduced as a fundamental quantity to describe equilibria of dynamical systems. In order to compute the complexity, one first considers the number of solutions $\mathcal{N}_N(\Delta)$ to the (random) equation ${\bf F}({\bf x})={\bf 0}$ with ${\bf F}:\mathbb{R}^N\to \mathbb{R}^N$ such that ${\bf x}$ also satisfies a set of constraints encoded in $\Delta$ (including, for example, magnetization, energy and self-overlap\footnote{i.e. $q={\bf x}^2/N$}: $\Delta=\{m,\epsilon,q,\ldots\}$). It is also common to include a parameter, referred to as (extensive) instability index, that accounts for the fraction of negative modes of the Jacobian of the fixed points counted. Then, only after studying the solutions at fixed disorder, does one average over the randomness. However, in general, one has that $\mathcal{N}_N$ itself is not \textit{self-averaging} for large $N$, meaning that its typical value for large $N$ is not the average. In fact, it usually scales exponentially: $\mathcal{N}_N\sim e^{N\Sigma}$, with $\Sigma$ a \textit{self-averaging} random variable, which we define a posteriori as the complexity, and can be seen as an entropy of stationary points. By taking the $\log$ \footnote{natural logarithm}, the complexity is then defined as
\begin{align}
\label{eq:intro_quenched_compelxity_def}
    \Sigma(\Delta):=\lim_{N\to\infty}\frac{\mathbb{E}[\log\mathcal{N}_N(\Delta)]}{N}
\end{align}
From this we see that a positive complexity is related to an exponential abundance of stationary points that satisfy the imposed constraints encoded in $\Delta$. Although a random energy landscape need not have a positive complexity (see for instance the spherical $2$-spin model, cf. Chapter~\ref{chapter:rmt_} or Ref.~\cite{ros_lecture_2025}), the models considered in this manuscript all have a positive complexity for some values of the control parameters. \\

\noindent The quantity defined in \eqref{eq:intro_quenched_compelxity_def} is the \textit{quenched complexity} or simply the complexity, different from the \textit{annealed complexity}, introduced previously in \eqref{eq:motivation_ann_kac}, and which is obtained by interchanging the $\log$ with the average:
\begin{align}
\Sigma_A(\Delta):=\lim_{N\to\infty}\frac{\log\mathbb{E}[\mathcal{N}_N(\Delta)]}{N}.
\end{align}
In some situations the two computations can coincide \cite{Crisanti_TAP_pspin_95,cavagna1998stationary, ros2019complexity, subag2017complexity, pacco2024curvature, HeliasFP2022}, while in others they don't \cite{ros2019complex, cavagna1999quenched, pacco_triplets_2025, kent2024arrangement, RosEcoQuenched2023, us_non_reciprocal_2025}. Often they differ when the system is subject to an external field, or a preferred direction in configuration space that breaks isotropy.
Hence, in the most general case, it is the quenched complexity that controls the typical behavior of $\mathcal{N}_N$ for large $N$: $\mathcal{N}_N(\Delta)\sim e^{N\Sigma(\Delta)}$, which in general differs from the average behavior $\mathbb{E}[\mathcal{N}_N(\Delta)]\sim e^{N\Sigma_A(\Delta)}$. Due to the concavity of the logarithm, one has that the following inequality holds:
\begin{align}
\Sigma_A(\Delta)\geq \Sigma(\Delta),
\end{align}
meaning that the annealed complexity can be used as an upper bound on the true value of the complexity. This implies additionally that $\mathbb{E}[\mathcal{N}_N]$ is bigger or equal than the typical value, and (when it is bigger) it is dominated by rare realizations. A simple example to show that average and typical properties differ is given in Ref.~\cite{ros_lecture_2025}, page 31.\\

\noindent Let us now introduce the use of replicas as a non-rigorous way to compute the quenched complexity. Since in general we do not know how to compute Eq.~\eqref{eq:intro_quenched_compelxity_def} directly, we can use the following identity:
\begin{align}
\log\mathcal{N}_N=\lim_{n\to 0}\frac{[\mathcal{N}_N]^n-1}{n},
\end{align}
which applied to Eq.~\eqref{eq:intro_quenched_compelxity_def} gives:
\begin{align}
\Sigma(\Delta)=\lim_{N\to\infty}\lim_{n\to 0}\frac{\mathbb{E}[\mathcal{N}_N^n(\Delta)]-1}{nN}.
\end{align}
Let us now present the main steps of the replicated Kac-Rice method in an informal but intuitive way:\\

\noindent\textit{Step 1}\\
Treat $n$ as an integer and compute $\mathbb{E}[\mathcal{N}_N^n(\Delta)]$ via the Kac-Rice formula (or via the Boltzmann-Gibbs measure for thermodynamic quantities). By doing so we introduce $n$ replicas, one for each $\mathcal{N}_N$ that appears in the expectation:
    \begin{align}
        \mathbb{E}[\mathcal{N}_N^n(\Delta)]=\int \prod_{a=1}^n d{\bf x}^a \delta(\Delta({\bf x}^a))\,\mathbb{E}\left[\prod_a\delta({\bf F}({\bf x}^a))\right]\mathbb{E}\left[\prod_a|\det J({\bf x}^a)|\bigg|{\bf F}({\bf x}^a)={\bf 0},\forall a\right],
    \end{align}
where we used the law of total expectation to divide the expected value of the product of deltas and determinants. To do this we introduced a conditional expectation (denoted by $\mathbb{E}[\,\cdot\,|\,\cdot\,]$ ) of the product of determinants, conditioned on the fields being zero at the points ${\bf x}^a,\,\forall a=1,\ldots,n$. We also abusively denote by $\delta(\Delta({\bf x}))=\delta({\bf x}^2-Nq)\,\delta({\bf x}\cdot{\bf 1}-Nm)\cdots$ the product of Dirac deltas that constrain the system at a specific value of the order parameters of our choice. From here we can open up the delta functions using their Fourier representation. \\

\noindent\textit{Step 2}\\
Introduce a matrix of overlaps between the replicas: $\hat{Q}_{ab}={\bf x}^a\cdot{\bf x}^b/N$ (more overlaps might be needed depending on the problem at hand), and enforce them with delta functions:
\begin{align}
    1\propto\int \prod_{a<b}dQ_{ab}\,\delta(Q_{ab}-{\bf x}^a\cdot{\bf x}^b/N).
\end{align}
Then open up the delta functions using their Fourier representation with conjugate variables $\Lambda_{ab}$, and write the $n$-th moment as:
\begin{align}
    \mathbb{E}[\mathcal{N}_N^n(\Delta)]\propto \int \prod_{a<b}dQ_{ab}\,d\Lambda_{ab}\,e^{NS(\Delta,\hat{Q},\hat{\Lambda})}
\end{align}
where (for the computation of the complexity) we do not care about sub-exponential terms in $N$, since in the end we want to use a saddle point method. \\

\noindent\textit{Step 3}\\
This is the first of the two most heuristic steps of the method: we want to use a saddle point to reduce the integral to an optimization over the maximum of $S$, and to do this we make the assumption that we can exchange the $N$ and $n$ limits. Then, we assume a structure on the set of overlaps that maximize the exponential. Physicists have developed the theory of mean-field glassy systems by considering that the overlap matrix (and its conjugate) has a very particular structure, which can be either Replica Symmetric (RS) or Replica Symmetry Breaking (RSB), where multiple steps of RSB could be considered \cite{spin_glass_beyond_86}. Often (but not always \cite{kent_rsb_2023}), replicated Kac-Rice calculations are done within the RS ansatz since the calculations already prove to be quite complicated. For an example of 1RSB ansatz see Sec.~\ref{sec:p_spin_free_energy}, while for an RS example see Sec.~\ref{sec:rnn_topo_general}\\

\noindent\textit{Step 4}\\
The second heuristic step consists in using the saddle point method for $N\to \infty$ and fixed $n\in \mathbb{N}$, and then assuming that we can take an analytic continuation of $n\to 0$ (in particular, assuming that the two limits can be exchanged). By doing this, only the term proportional to $Nn$ in the exponential survives:
\begin{align}
\Sigma_A(\Delta)=\text{extr}_{\tilde{Q},\tilde{\Lambda}}S^{(1)}(\Delta,\tilde{Q},\tilde{\Lambda})
\end{align}
where $S^{(1)}$ is the term proportional to $n$ in the Taylor expansion of $S$ around $n=0$, and $\tilde{Q},\tilde{\Lambda}$ represent the set of parameters of the assumed ansatz.

\noindent On the historical side, let us just state that the use of the replica trick in the physics community has a long history (see Ref.~\cite{charobonneau_history_replica_2022} for historical accounts) that dates back to Marc Kac in 1968 \cite{charobonneau_history_replica_2022}. The method was then brought to glory by Giorgio Parisi's works \cite{parisi_order_para_spin_83, Parisi_sk_solution_1980, spin_glass_beyond_86}, when in the 80's he formulated the fullRSB solution to the Sherrington Kirkpatrick model introduced in 1975 \cite{SK75}. \\

\subsubsection{A note on the sign of the complexity}
When the complexity $\Sigma$ is positive, we clearly have that typically, for large $N$, there is an exponential number (in $N$) of fixed points. When the complexity vanishes, that is, $\Sigma=0$, we can only say that typically there is a sub-exponential number of fixed points. This means that, for large $N$, there could be a finite number (two, for example) as well as a number polynomial in $N$. If one is interested in the precise number, an exact calculation of $\mathbb{E}[\mathcal{N}]$ at finite $N$ should be done. When it becomes negative, i.e. $\Sigma<0$, then it means that only in exponentially rare realizations the number of points is bigger
than zero.

\subsubsection{Topological vs Geometrical properties}
The complexity is a topological property which tells us information on the number of fixed points and their stability. This is important since it provides us with criteria to identify \textit{dynamical transitions}, by telling us how far the system can go deep before encountering local minima. On the other side, it also helps us identify \textit{topology trivialization transitions} \cite{Fyod2004}, where the number of fixed points crosses from being exponential to sub-exponential in the system's size, as a control parameter is varied. With the same tools introduced above we can also analyze more geometrical properties of the landscape. These include, for example, the correlations between different stationary points, and the distribution of barriers. We will see in Chapter~\ref{chapter:energy_landscapes} that, similarly as above, we can compute a constrained complexity, by extracting first one fixed point, and then studying the distribution of local minima and saddles around it. In the same Chapter we will also use Kac-Rice approaches to study barriers along specific paths in phase space.

\section{The $p$-spin: a prototypical glassy model}
\label{sec:p_spin_model}
The goal of this section is to provide a "cookbook" of basic ingredients to analyze the simplest random energy landscape: the pure spherical $p-$spin model. This is a Gaussian model because the energy function $\mathcal{E}$ is a Gaussian random field, and it is mean-field because it has \textit{all-to-all} interactions with no underlying spatial structure. With $N$ we denote the (large) number of degrees of freedom (or spins), which are continuous vectors on the hypersphere of radius $\sqrt{N}$ in $\mathbb{R}^N$, denoted as $\mathcal{S}_N(\sqrt{N})$. More precisely this is an isotropic and homogeneous random landscape defined as follows: 
\begin{align}
\label{eq:chap1_def_tensor_S}
\mathcal{E}:\mathcal{S}_N(\sqrt{N})\to\mathbb{R},\quad\quad\mathbb{E}[\mathcal{E}({\bf x})]=0,\quad\mathbb{E}[\mathcal{E}({\bf x})\mathcal{E}({\bf y})]=\frac{N}{2}\left(\frac{{\bf x}\cdot{\bf y}}{N}\right)^p.
\end{align}
 This model for $p\geq 3$ has become the prototypical model of glassiness (in a way that will become clear later), since it retains the important phenomenology of glassy systems while being analytically tractable. A particular way to realize the landscape above is to introduce a symmetric (i.e. constant by permutation of the indices) Gaussian tensor $S_{i_1\ldots i_p}$ such that
\begin{align}
\mathbb{E}[S_{i_1\ldots i_p}]=1\quad\quad\mathbb{E}[S_{i_1\ldots i_p}S_{j_1\ldots j_p}]=\sum_{\pi\in G_p}\delta_{i_1j_{\pi(1)}}\ldots\delta_{i_p j_{\pi(p)}}
\end{align}
with $G_p$ the permutation group of $p$ elements, and to write 
\begin{align}
\mathcal{E}({\bf s})=\sqrt{\frac{1}{2p!\,N^{p-1}}}\sum_{i_1,\ldots ,i_p}S_{i_1\ldots i_p}x_{i_1}\cdots x_{i_p}.
\end{align}
In physics this model is generally defined with the tensor $S$ being zero whenever two indices are equal \cite{Castellani_2005}, in which case the formula \ref{eq:chap1_def_tensor_S} is only valid up to leading order in $N$.

From this representation of the model we see that each spin $x_i$ interacts with all other spins, and the model is therefore devoid of a spatial structure. This model is called "pure", to differentiate with "mixed" models, introduced in Sec.~\ref{sec:mixed_models}, where the covariance is a generic polynomial. The tensor $S$ is commonly referred to as "quenched disorder", because it represents the inherent disorder that creates the landscape, and it is quenched because we study properties of the landscape at fixed $S$, eventually taking an average in the end. Indeed, we expect that in the limit of large $N$, relevant observables become independent on the particular realization of $S$ (in the same way as we expect properties of materials to not depend on the particular impurities of each sample), and are thus \textit{self-averaging}. A natural question that one can ask is about the behavior of gradient descent (GD) as one tries to minimize this landscape, aiming for the global minima (or ground states). Thus, we are interested in studying:
\begin{align}
\frac{d{\bf x}}{dt}=-\lambda({\bf x}){\bf x}-\nabla \mathcal{E}({\bf x}),
\end{align}
where $\lambda$ is a Lagrange multiplier used to impose the spherical constraint:
\begin{align}
&0=\frac{1}{2}\frac{d{\bf x}^2}{dt}={\bf x}\frac{d{\bf x}}{dt}=-\lambda({\bf x}){\bf x}^2-{\bf x}\cdot \nabla \mathcal{E}({\bf x})\quad\Rightarrow \quad\lambda({\bf x})=-\frac{{\bf x}\cdot\nabla \mathcal{E}({\bf x})}{N}.
\end{align}
Let us remark that by $\nabla$ we mean the gradient on $\mathbb{R}^N$ (and not restricted to the sphere), which can be used thanks to the introduction of the Lagrange multiplier. For simplicity we refer to $\nabla$ as the "unconstrained" operator. It is not too hard to prove that:
\begin{align}
\label{eq:intro_homog_pspin}
    \nabla\mathcal{E}({\bf x})\cdot{\bf x}=p\,\mathcal{E}({\bf x}),
\end{align}
a property that follows from the homogeneity of $\mathcal{E}$ (i.e. that $\mathcal{E}(\alpha\,{\bf x})=\alpha^p\mathcal{E}({\bf x})$).  
This fact is crucial for the analysis of the pure $p$-spin model, as it relates the Lagrange multiplier to the energy:
\begin{align}
\label{eq:mu_energy_T=0}
    \lambda({\bf x})=-\frac{p\,\mathcal{E}({\bf x})}{N}
\end{align}
a property that is peculiar to the pure model. As we will describe later on, this fact does not hold true for mixed models, and this has profound consequences on the properties of local minima of the landscape. Similarly, we can also prove that 
\begin{align}
    \nabla^2\mathcal{E}({\bf x})\cdot{\bf x}=(p-1)\nabla\mathcal{E}({\bf x})
\end{align}
where $\nabla^2\mathcal{E}({\bf x})$ is the Hessian matrix on $\mathbb{R}^N$. These properties are important in the "algebra" of the pure $p$-spin, especially when studying the so called "two-point" and "three-point" complexities, see Chapter~\ref{chapter:energy_landscapes}. 

\subsubsection{A bit of history}
Let us give a brief historical account of the pure spherical $p$-spin model, and cite the major references with lecture notes. The first appearance of a spherical model of ferromagnet dates back to Berlin and Kac in the 50's \cite{kac_spherical_ferro_52}. The $2$-spherical model was introduced by Kosterlitz, Thouless, Jones \cite{kosterlitz1976spherical}, and also studied by Sompolinsky and Zippelius \cite{Sompolinsky_zippelius_ed_82}, and later by Kirkpatrick and Thirumalai for the case $p=2+\epsilon$ \cite{Thirum_pspin_87, Thirum_pspin_87_part2}. The model with $p$-Ising spins (i.e. $\pm 1$) was introduced by Derrida \cite{Rem_1} and studied for $p\to\infty$, and later by Gross and Mézard in \cite{gross1984simplest}. The model as we present it here was properly analyzed in the 90's by Crisanti, Sommers (and Horner) in \cite{crisanti1992sphericalp} (\cite{crisanti1993spherical}), who showed that the free energy presents a 1RSB structure below a certain temperature $T_s$ and that ergodicity is broken at a temperature $T_d>T_s$, for $p>2$. The TAP free energy \footnote{a free energy constrained to specific values of magnetization of each spin, see Sec.~\ref{sec:tap_approach}} was then computed by Crisanti and Sommers in \cite{Crisanti_TAP_pspin_95}, while the analytical solution to the out-of-equilibrium dynamics below $T_d$ was due to Cugliandolo and Kurchan \cite{cugliandolo1993analytical, Cugliandolo_1995_weak}. In '95 Franz and Parisi introduce "the potential", a free energy of two equilibrium configurations at fixed overlap extracted one after the other, from which one can bridge thermodynamical and dynamical approaches \cite{franz1995recipes, franz1998effective, Barrat_bifurcation_95, cavagna1997structure}. In the 2000's, the concept of complexity was put on firmer grounds by Cavagna, Giardina, Parisi \cite{cavagna1998stationary}. More recently, new research on the structure of the energy landscape and dynamics has received a significant boost \cite{ros2019complex, ros2019complexity, ros2020distribution, ros2021dynamical, pacco2024curvature, pacco_triplets_2025, folena_rare_2025, rizzo2021path, ghimenti_accelerating_2022}. Moreover, the ideas and results found by physicists are being formalized and further developed in recent years by a community of mathematicians: \cite{Talagrand_pspin_2005, auffinger_complexity_2013, subag2017complexity, Auffinger_Arous_Černý_spin_2012, Auffinger_saddles_2020, Subag_tap_2023, Subag_Zeitouni_2021} to cite a few. \\

\noindent Let us indicate some resources where the pure spherical $p$-spin is explained \cite{Castellani_2005, cugliandolo2002dynamics, DeDominicis_Giardina_book_2006, zamponi_mean_field_notes_2014, Folena_notes_2023, folena_these_2020, barra_pspin_97}, and other resources on the general tools used in this field: \cite{spin_glass_beyond_86, urbani_notes_2017}.


\subsection{Topological complexity}
\label{sec:chap1_topo_compl}
We show with good detail (for a physicist) the derivation of the complexity, while the other tools will be presented in a more informal and direct way.\\

\noindent As we have seen in Sec.~\ref{sec:topo_complexity}, the topological complexity of a landscape $\mathcal{E}$ counts (on an exponential scale) the number of stationary points (including local minima and saddles). In the case of the pure spherical $p-$spin model defined above, one writes the complexity as:
\begin{align}
\label{eq:intro_topo_compl_pspin_quenched}
    \Sigma(\epsilon)=\lim_{N\to\infty}\frac{1}{N}\mathbb{E}[\log\mathcal{N}(\epsilon)]
\end{align}
where $\mathcal{N}(\epsilon)$ counts the number of stationary points ${\bf x}$ at fixed energy density $\epsilon\equiv \mathcal{E}({\bf x})/N$. To be consistent with the notation used in recent papers on this topic \cite{ros2019complex, ros2019complexity, ros2021dynamical}, we define the rescaled field:
\begin{align}
    {\bm\sigma}:=\frac{{\bf x}}{\sqrt{N}},\quad h({\bm\sigma}):=\sqrt{\frac{2}{N}}\mathcal{E}\left({\bm\sigma}\sqrt{N}\right).
\end{align}
Since we are working on the hypersphere, let us define the tangent plane of a certain vector ${\bm\sigma}$ by $\tau[{\bm\sigma}]$ (that is, the vector space orthogonal to ${\bm\sigma}$). With this notation, we can moreover define a local (orthonormal) basis dependent on ${\bm\sigma}$ as:
\begin{align}
\label{eq:local_basis}
    \mathcal{B}[{\bm \sigma}]&:=\bigg\{\underbrace{{\bf e}_1({\bm \sigma}),\ldots,{\bf e}_{N-1}({\bm\sigma})}_{\tau[{\bm \sigma}]},{\bf e}_N({\bm \sigma}):={\bm \sigma}\bigg\}
\end{align}
where it is intended that $\tau[{\bm\sigma}]=\text{Span}({\bf e}_1({\bm\sigma}),\ldots,{\bf e}_{N-1}({\bm\sigma}))$. We also refer to a generic orthonormal basis of $\mathbb{R}^N$ as $\mathcal{C}:=\{{\bf x}_1,\ldots,{\bf x}_N\}$. The relation \eqref{eq:intro_homog_pspin} now reads:
\begin{align}
    \nabla h({\bm\sigma})\cdot{\bm\sigma}=p\,h({\bm\sigma}).
\end{align}
This identity implies that we can write
\begin{align}
\label{app:eq:second_grad_id}
\nabla h({\bm\sigma})={\bf g}({\bm \sigma})+ph({\bm\sigma}){\bm \sigma}
\end{align}
where ${\bf g}({\bm \sigma})$ has components in the $\mathcal{B}[{\bm \sigma}]$ basis given by
\begin{align*}
&[{\bf g}({\bm \sigma})]_{\alpha<N}=\nabla h({\bm\sigma})\cdot {\bf e}_\alpha({\bm\sigma})\\
&[{\bf g}({\bm \sigma})]_N=0.
\end{align*}
The Riemannian gradient on the hypersphere $\nabla_\perp h({\bm \sigma})$ is the $(N-1)-$dimensional vector obtained projecting ${\bf g}({\bm \sigma})$  on $\tau[{\bm \sigma}]$, that is, neglecting the (last) null component. In the following, we denote also with ${\bf g}({\bm \sigma})$ this $(N-1)-$ dimensional projection as well, with a slight abuse of notation. \\

\noindent Similarly, we define by $\nabla^2 h({\bm \sigma})$ the $N \times N$-dimensional Hessian matrix with components in $\mathcal{C}$:
\begin{equation}
\label{app:eq:def_hes}
  [\nabla^2 h({\bm \sigma})]_{ij} :=  {\bf x}_i \cdot  \nabla^2 h({\bm \sigma}) \cdot {\bf x}_j= \frac{\partial^2 h({\bm \sigma})}{\partial \sigma_i \partial \sigma_j}.
\end{equation}
As before, we have that 
\begin{align}
\label{app:eq:first_hes_id}
    \nabla^2 h({\bm \sigma})\cdot{\bm \sigma}=(p-1)\nabla h({\bm \sigma}).
\end{align}

\noindent The Riemannian Hessian on the hypersphere, denoted by $\nabla^2_\perp h({\bm\sigma})$, is conveniently found by means of the Lagrange multiplier introduced above. Given $h_{\lambda}({\bm\sigma}):=h({\bm\sigma})-\frac{\lambda}{2}({\bm\sigma}^2-N)$, then $\nabla h_\lambda({\bm\sigma})=\nabla h({\bm\sigma})-\lambda{\bm\sigma}\overset{!}{=}0\Rightarrow \lambda={\bm\sigma}\cdot\nabla h({\bm\sigma})=p\,h({\bm\sigma})$. Therefore, in the basis of $\tau[{\bm\sigma}]$, we have:
\begin{align}
\label{app:eq:def_hes_riem}
[\nabla^2_\perp h({\bm \sigma})]_{\alpha\beta}\equiv \mathbf{e}_{\alpha}({\bm\sigma})^\top \nabla^2 h_\lambda({\bm\sigma})  \mathbf{e}_{\beta}({\bm \sigma})=\mathbf{e}_{\alpha}({\bm \sigma})^\top\nabla^2 h({\bm\sigma}) \mathbf{e}_{\beta}({\bm\sigma})-p h({\bm\sigma})\delta_{\alpha \beta}.
\end{align}
where $ \alpha,\beta \leq N-1$.
Eq.~\eqref{app:eq:def_hes_riem} shows that the Riemannian Hessian is derived from the unconstrained Hessian by shifting with a diagonal matrix proportional to $p h({\bm\sigma})$, and by projecting onto the local tangent plane. Thus, working with either the unconstrained or Riemannian Hessian is essentially equivalent, as long as the shift is taken into account. We note that, with this new notation, the energy density is expressed as $\epsilon=h({\bm\sigma})/\sqrt{2N}$.\\

\noindent We have now set the stage for the computation of the complexity. It has been seen in \cite{cavagna1998stationary}, and later shown rigorously in \cite{subag2017complexity}, that one actually has
\begin{align}
\Sigma(\epsilon)=\lim_{N\to\infty}\frac{1}{N}\log\mathbb{E}[\mathcal{N}(\epsilon)],
\end{align}
meaning that the annealed and quenched averages match for this model. However, bear in mind that in general this might not be the case, as for example happens in the presence of a magnetic field \cite{cavagna1999quenched}. Physically, the difference is essentially related to the time scale of fluctuations of $S_{i_1,\ldots i_p}$ and $\bm\sigma$. In the quenched case, the tensor $S$ is quenched (that is, fixed) and all averages over configurations are taken with this disorder being fixed. Then, only after averaging over $\bm\sigma$ does one average over $S$ quantities that are \textit{self-averaging}, namely those observable physical quantities that are the same regardless of the specific realization of the disorder, when $N$ is large. Such quantities are therefore \textit{typical}, in the sense that, as $N$ becomes larger and larger, they correspond to what one would observe with an experiment or simulation at a single instance of disorder.  In the annealed computation, instead, disorder and configurational averages are on the same footing.  As we have already seen in Sec.~\ref{sec:topo_complexity}, the annealed quantity might provide \textit{atypical} results, by giving more weight to less probable instances. \\

\noindent The calculation of the quenched complexity can be done by introducing replicas, and while here we restrict to the annealed one, we have performed several quenched calculations along this manuscript, see Chapters~\ref{chapter:non_reciprocal},\ref{chapter:scs},\ref{chapter:energy_landscapes}, with a detailed presentation for Chapter~\ref{chapter:non_reciprocal} in Appendix.~\ref{app:rnn_quenched}. \\

\noindent The number of stationary points of $\mathcal{E}$ can be written as an integral over configurations that satisfy certain constraints:\footnote{one could also implement directly the Lagrange multiplier and consider an integration over $\mathbb{R}^N$ as in \cite{Fyodorov_2016}, but here we wanted to leverage the nice simplifications of the pure model, due to its isotropy and homogeneity.}
\begin{align*}
    \mathcal{N}(\epsilon)&=\int_{\mathcal{S}_N(1)}d{\bm \sigma}\,\delta\left(h({\bm\sigma})-\sqrt{2N}\epsilon\right)\delta({\bf g}({\bm\sigma}))\,\left|\det \nabla^2_\perp h({\bm\sigma})\right|.
\end{align*}
Since ultimately we will take the log of the expected value of this quantity (and divide by $N$), we can drop any prefactor that is not exponential in $N$. The Kac-Rice formula for the average gives:
\begin{align}
\label{eq:pspin:avg_EN}
    \mathbb{E}[\mathcal{N}(\epsilon)]=\int_{S_N(1)}d{\bm\sigma}\,\,\mathbb{E}[\delta(h({\bm\sigma})-\sqrt{2N}\epsilon)\delta({\bf g}({\bm\sigma}))]\,\cdot\mathbb{E}\left[|\det\nabla^2_\perp h({\bm \sigma})|\bigg|
    \grafe{
    \begin{subarray}{l}
    h({\bm\sigma}) = \sqrt{2N}\epsilon\\
   {\bf g}({\bm\sigma})={\bf 0}\end{subarray}}
    \right].
\end{align}
By virtue of the isotropy of this landscape, one can show (see below) that the expected values inside \eqref{eq:pspin:avg_EN} do not depend on ${\bm\sigma}$. Hence, by neglecting sub-exponential contributions, we have:
\begin{align}
    \mathbb{E}[\mathcal{N}(\epsilon)]\propto V\cdot P(\epsilon)\cdot H(\epsilon)
\end{align}
where
\begin{align}
    V=\int_{\mathcal{S}_N(1)} d{\bm\sigma},\quad \quad P(\epsilon)=\mathbb{E}\left[\delta\left(h({\bm\sigma})-\sqrt{2N}\epsilon\right)\delta({\bf g}({\bm\sigma}))\right]
\end{align}
and 
\begin{align}
    H(\epsilon)=\mathbb{E}\left[|\det\nabla^2_\perp h({\bm \sigma})|\bigg|
    \grafe{
    \begin{subarray}{l}
    h({\bm\sigma}) = \sqrt{2N}\epsilon\\
   {\bf g}({\bm\sigma})={\bf 0}\end{subarray}}
    \right].
\end{align}
Therefore, the complexity reads:
\begin{align}
\Sigma(\epsilon)=\lim_{N\to\infty}\frac{1}{N}[\log V+\log P(\epsilon)+\log H(\epsilon)].
\end{align}
Before proceeding, we see that it is important to compute the statistics of fields that depend on the disorder: $h, \nabla h,\nabla^2 h$. Since the disorder is Gaussian and these fields are built as sums of Gaussian random variables, they are also Gaussian. Hence, we only need to determine their mean and covariance. Once this is done we can, for example, compute conditional distributions of such quantities. These fields are, by construction, of zero mean. Their covariance can be computed directly from Eq.~\ref{eq:chap1_def_tensor_S} by differentiating with respect to the right variables. This has been done in \cite{ros2019complex, ros2019complexity}, and we simply state here the results, since the computation is not conceptually difficult. Consider two configurations ${\bm\sigma}_a$ and ${\bm\sigma}_b$, and denote for simplicity $h({\bm\sigma}_a)=h^a$, and same for $b$. Consider also 4 arbitrary vectors ${\bf u}_i$. Then we have:
\begin{align}
\label{eq:pspin_compl_correl_1}
\begin{split}
    \mathbb{E}\left[h^a h^b \right]=({\bm\sigma}_a\cdot {\bm\sigma}_b)^p,\quad\mathbb{E}\left[ \tonde{{\bm \nabla} h^a \cdot {\bf u}_1} h^b  \right]= p ({\bm \sigma}_a \cdot {\bm \sigma}_b)^{p-1} \tonde{{\bf u}_1 \cdot {\bm \sigma}_b}\\
\end{split}
\end{align}
and 
\begin{align}
    \mathbb{E}\left[\tonde{{\bf u}_1 \hspace{-0.1cm}\cdot \hspace{-0.1cm} {\bm \nabla}^2 h^a \hspace{-0.1cm} \cdot \hspace{-0.05cm} {\bf u}_2}  h^b \right]=p(p\hspace{-0.05cm}-\hspace{-0.05cm}1)({\bm \sigma}_a \hspace{-0.05cm}\cdot\hspace{-0.05cm} {\bm \sigma}_b)^{p-2}({\bf u}_1 \hspace{-0.05cm}\cdot\hspace{-0.05cm} {\bm \sigma}_b) ({\bf u}_2 \hspace{-0.05cm}\cdot\hspace{-0.05cm} {\bm \sigma}_b).
\end{align}
Between the gradient components we have:
\begin{equation}\label{app:eq:CovGradComp}
\begin{split}
  \mathbb{E}\left[ \tonde{{\bm \nabla} h^a \cdot {\bf u}_1} \tonde{{\bm \nabla} h^b \cdot {\bf u}_2} \right]&=p ({\bm \sigma}_a \cdot {\bm \sigma}_b)^{p-1} \tonde{{\bf u}_1 \cdot {\bf u}_2}\\
  &+p(p-1)({\bm \sigma}_a \cdot {\bm \sigma}_b)^{p-2} \tonde{{\bf u}_2 \cdot {\bm \sigma}_a} \tonde{{\bf u}_1 \cdot {\bm \sigma}_b}.
\end{split}
\end{equation}
 For what concerns the Hessians, one gets:
\begin{equation}\label{app:eq:HessTang}
\begin{split}
 \mathbb{E}&\left[\tonde{ {\bf u}_1 \cdot  {\bm \nabla}^2 h^a \cdot {\bf u}_2} \tonde{ {\bf u}_3 \cdot {\bm \nabla}^2 h^b \cdot {\bf u}_4} \right]=\frac{p! ({\bm \sigma}_a \cdot {\bm \sigma}_b)^{p-4}}{(p-4)!}({\bf u}_1 \cdot {\bm \sigma}_b) ({\bf u}_2\cdot {\bm \sigma}_b) ({\bf u}_3\cdot {\bm \sigma}_a)( {\bf u}_4 \cdot {\bm \sigma}_a)\\&+
 \frac{p!}{(p-3)!}({\bm \sigma}_a \cdot {\bm \sigma}_b)^{p-3} \,{({\bf u}_1 \cdot {\bf u}_4 )( {\bf u}_2 \cdot {\bm \sigma}_b)( {\bf u}_3 \cdot {\bm \sigma}_a)}\\&+\frac{p!}{(p-3)!}({\bm \sigma}_a \cdot {\bm \sigma}_b)^{p-3} \,{({\bf u}_2 \cdot {\bf u}_4)( {\bf u}_1 \cdot {\bm \sigma}_b)( {\bf u}_3\cdot {\bm \sigma}_a)}\\
  &+\frac{p!}{(p-3)!}({\bm \sigma}_a \cdot {\bm \sigma}_b)^{p-3}\, {({\bf u}_1 \cdot {\bf u}_3)( {\bf u}_2 \cdot {\bm \sigma}_b)( {\bf u}_4 \cdot {\bm \sigma}_a)}\\&+\frac{p!}{(p-3)!}({\bm \sigma}_a \cdot {\bm \sigma}_b)^{p-3}\, {({\bf u}_2 \cdot {\bf u}_3)( {\bf u}_1 \cdot {\bm \sigma}_b)( {\bf u}_4\cdot {\bm \sigma}_a)}\\
 &+\frac{p! ({\bm \sigma}_a \cdot {\bm \sigma}_b)^{p-2} }{(p-2)!}\quadre{( {\bf u}_1\cdot  {\bf u}_3)( {\bf u}_2\cdot  {\bf u}_4)+( {\bf u}_1\cdot  {\bf u}_4)( {\bf u}_2\cdot  {\bf u}_3)}.
 \end{split}
\end{equation}
Finally, the correlations between Hessians and gradients read: 
\begin{equation}\label{app:eq:CorelationsHessianGrad}
 \begin{split}
& \mathbb{E}\left[   \tonde{ {\bf u}_1 \cdot  {\bm \nabla}^2 h^a \cdot {\bf u}_2} \tonde{{\bm \nabla} h^b \cdot  {\bf u}_3 } \right]
 =p(p-1)(p-2) ({\bm \sigma}_a \cdot {\bm \sigma}_b)^{p-3} ({\bf u}_1 \cdot {\bm \sigma}_b) ({\bf u}_2 \cdot {\bm \sigma}_b)  ({\bf u}_3\cdot {\bm \sigma}_a)\\&+
  p(p-1) ({\bm \sigma}_a \cdot {\bm \sigma}_b)^{p-2} ({\bf u}_1 \cdot {\bf u}_3) ({\bf u}_2 \cdot {\bm \sigma}_b)+  p(p-1) ({\bm \sigma}_a \cdot {\bm \sigma}_b)^{p-2}({\bf u}_2 \cdot {\bf u}_3) ({\bf u}_1 \cdot {\bm \sigma}_b).
 \end{split}
\end{equation}
Since here the computation is annealed, we will only be interested in ${\bm\sigma}_a={\bm\sigma}_b$. However, these expressions will be very useful later in our analysis in Chapter~\ref{chapter:energy_landscapes}, when considering the complexity of triplets of stationary points.

\noindent Let us now compute the various pieces that compose the complexity.

\subsubsection{Phase space factor.}
Although we abusively denote this factor by $V$ to indicate a volume (indeed, in general we might consider non-spherically constrained models), in this present case it represents the surface of $\mathcal{S}_N(1)$:
\begin{align*}
    V&=\frac{2\pi^{\frac{N}{2}}}{\Gamma(\frac{N}{2})}\underset{N>>1}{\sim} e^{\frac{N}{2}\left[1+\log(\frac{2\pi}{N})\right]+o(N)}
\end{align*}
where to express its asymptotic value for $N$ large we used Stirling's approximation: $\ln(N!)\sim N\log(N)-N$.

\subsubsection{Joint probability of energy and gradients.}
We now turn our attention to the computation of $P(\epsilon)$. By leveraging the fact that
\begin{align*}
    \nabla h({\bm\sigma})={\bf g}({\bm\sigma})+p\,h({\bm\sigma}){\bm\sigma}
\end{align*}
we can write the joint probability of ${\bf g}$ and $h$ as the probability density of the unconstrained gradient $\nabla h$:
\begin{align}
    P(\epsilon)=P_{{\bf g},h}\left({\bf 0},\sqrt{2N}\epsilon\right)=P_{\nabla h}\left(p\sqrt{2N}\epsilon\,{\bm\sigma}\right).
\end{align}
Taking this into account, we can express $P(\epsilon)$ as the joint probability density of the (unconstrained) gradient components evaluated at $p\sqrt{2N}\epsilon\,{\bm\sigma}$. More precisely we have:
\begin{align}
P(\epsilon)=\frac{e^{-\frac{1}{2}\nabla h^\top \,\hat{C}^{-1}\,\nabla h}}{(2\pi)^{\frac{N}{2}}(\det\hat{C})^{\frac{1}{2}}}\Bigg|_{\nabla h={\bf 0}+p\sqrt{2N}\epsilon{\bm\sigma}}=\frac{e^{-p^2N\epsilon^2\,{\bm\sigma}^\top \,\hat{C}^{-1}\,{\bm\sigma}}}{(2\pi)^{\frac{N}{2}}(\det\hat{C})^{\frac{1}{2}}}
\end{align}
with $\hat{C}_{ij}=\mathbb{E}[(\nabla h({\bm\sigma}))_i(\nabla h({\bm\sigma}))_j]=p\,\delta_{ij}+p(p-1)\sigma_i\sigma_j$ which implies that $\hat{C}=p\mathbb{I}+p(p-1){\bm\sigma}{\bm\sigma}^\top$. It is easy to see that, in the local (orthonormal) basis $\mathcal{B}[{\bm\sigma}]$, the matrix $\hat{C}$ reads
\begin{align}
\hat{C}=
\begin{bmatrix}
p & 0 & 0 & \cdots & 0 \\
0 & p & 0 & \cdots & 0 \\
0 & 0 & p & \cdots & 0 \\
\vdots & \vdots & \vdots & \ddots & \vdots \\
0 & 0 & 0 & \cdots & p^2
\end{bmatrix}.
\end{align}
This result immediately implies that
\begin{align}
    {\bm\sigma}^\top\,\hat{C}^{-1}\,{\bm\sigma}=\frac{1}{p^2},\quad\quad \det\hat{C}=p^{N+1}.
\end{align}
Putting all together, we obtain:
\begin{align}
P(\epsilon)=e^{-N\left[\epsilon^2 +\frac{1}{2}\log(2\pi p)\right]}
\end{align}
where we neglected the additional prefactor $p$, that will not contribute when taking the log and the limit of $N\to\infty$. 

\subsubsection{Hessian}
The Hessian term $H(\epsilon)$ is usually not hard to compute when the Hessian is a random matrix that follows either the Wigner's semicircle law or the elliptic law \cite{sommers1988spectrum} for $N\to\infty$. In more complex non-Gaussian models, such as those arising when modeling neural networks \cite{maillard_landscape_2020}, the Hessian becomes more complicated. \\

\noindent Here we need to study the statistics of $\nabla^2_\perp h$, which we recall reads $[\nabla^2_\perp h({\bm\sigma})]_{ab}=[\nabla^2h({\bm\sigma})]_{ab}-p\,h({\bm\sigma})\delta_{ab}$ where $a,b\leq N-1$ represent the local basis $\mathcal{B}[{\bm\sigma}]$. We could express the problem in any basis, but we see that in the local basis the expression for the Riemannian Hessian is very simple, as we can essentially just study $\nabla^2h$ and shift by the energy density. \\

\noindent From the definition of $H(\epsilon)$ we see that we need to find the statistics of $\nabla^2 h$ conditioned to $h=\sqrt{2N}\epsilon$ and ${\bf g}={\bf 0}$. To do that, we first start from the original Gaussian distribution of $\nabla^2 h$:
\begin{align}
\mathbb{E}[\nabla^2h]=0,\quad\mathbb{E}[(\nabla^2h)_{\alpha\beta}(\nabla^2h)_{\gamma\delta}]=p(p-1)[\delta_{\alpha\gamma}\delta_{\beta\delta}+\delta_{\alpha\delta}\delta_{\beta\gamma}]
\end{align}
where $\alpha,\beta,\gamma,\delta\leq N-1$ represent indices of elements of the basis of $\tau[{\bm\sigma}]$, which are therefore orthogonal to ${\bm\sigma}$, and are responsible for the huge simplification of the covariance of Hessian elements. We immediately recognize a GOE (Gaussian Orthogonal Ensemble) law for $\nabla^2h|_{N-1}$ (see Chapter~\ref{chapter:rmt_} for an introduction), where $|_{N-1}$ is used to indicate that we restrict the matrix to the $(N-1)\times (N-1)$ inner block. Here we simply state the final result, namely that the empirical spectral distribution of $\nabla^2h|_{N-1}/\sqrt{N}$ converges (as $N$ tends to infinity) almost surely to the Wigner's semicircular law \cite{Tao_book_2012, potters_bouchaud_2020}, given by:
\begin{align}
    \rho_\sigma(\lambda)=\frac{1}{2\pi\sigma^2}\sqrt{4\sigma^2-\lambda^2}\,{\bf 1}_{|\lambda|\leq 2\sigma},\quad\quad\sigma:=\sqrt{p(p-1)}.
\end{align}
The fact that $\nabla^2 h|_{N-1}$ has size size $(N-1)\times (N-1)$ and that we divided by $\sqrt{N}$ above does not influence the convergence to the Wigner's law as $N\to\infty$, since this mismatch contributes only with finite-size $1/N$ contributions, which get flushed away in the large $N$ limit (see Chapter~\ref{chapter:rmt_} or Ref.~\cite{paccoros}). \\

\noindent Now we have to compute the conditional law of $\nabla^2h$ to the fact that ${\bf g}={\bf 0},h=\sqrt{2N}\epsilon$, which we indicate with an upper tilde: 
\begin{align}
\frac{\tilde{\nabla}^2_\perp h}{\sqrt{N}}=\frac{\tilde{\nabla}^2h|_{N-1}}{\sqrt{N}}-p\epsilon\sqrt{2}\mathbb{I}_{N-1}.
\end{align}
In this simple setting, we see directly from Eq.~\ref{eq:pspin_compl_correl_1} that $\mathbb{E}[(\nabla^2h)_{\alpha\beta}\,h]=\mathbb{E}[(\nabla^2h)_{\alpha\beta}\,{\bf g}_\gamma]=0$ for any $\alpha,\beta,\gamma\leq N-1$ representing basis elements in $\tau[{\bm\sigma}]$. This means that the statistics of $\nabla^2h|_{N-1}$ is not changed by the conditioning (see formula for conditioning of Gaussian variables in Appendix.~\ref{app:sec_gausian_conditioning}). Therefore we ultimately have that, in the limit of $N\to\infty$, the determinant term reads:
\begin{align}
    H(\epsilon)=N^{\frac{N-1}{2}}\mathbb{E}\left[\prod_{i=1}^{N-1}\left|\lambda_i-p\epsilon\sqrt{2}\right|\right]\underset{N\to\infty}{\sim} N^{\frac{N-1}{2}}e^{N\int_{-2\sigma}^{2\sigma} d\lambda\,\rho_\sigma(\lambda)\log|\lambda-p\epsilon\sqrt{2}|}
\end{align}
where $\lambda_i$ indicates the eigenvalues of $\nabla^2 h|_{N-1}/\sqrt{N}$. By performing the integration one then finds that \cite{ros2019complex}:
\begin{align}
\int_{-2\sigma}^{2\sigma}d\lambda\,\rho_\sigma(\lambda)\log|\lambda-p\epsilon\sqrt{2}|=\frac{1}{2}\log(2\,p(p-1))+I\left(\epsilon\sqrt{\frac{p}{p-1}}\right)
\end{align}
where the even function $I$ reads:
\begin{align}\label{eq:IDef}
    I(y)=\begin{cases}I_-(y)
        \quad \text{if } y\leq -\sqrt{2}\\
        I_+(y)\quad \text{if }-\sqrt{2}\leq y\leq 0
    \end{cases}
\end{align}
with
\begin{equation}
\begin{split}
&I_-(y)=\frac{y^2-1}{2}+\frac{y}{2}\sqrt{y^2-2}+\log\left(\frac{-y+\sqrt{y^2-2}}{2}\right)\\
& I_+(y)=\frac{1}{2}y^2-\frac{1}{2}(1+\log 2).
 \end{split}
\end{equation}

\begin{figure}[t!]
\centering
\includegraphics[width=1\textwidth, trim= 5 5 5 5, clip]{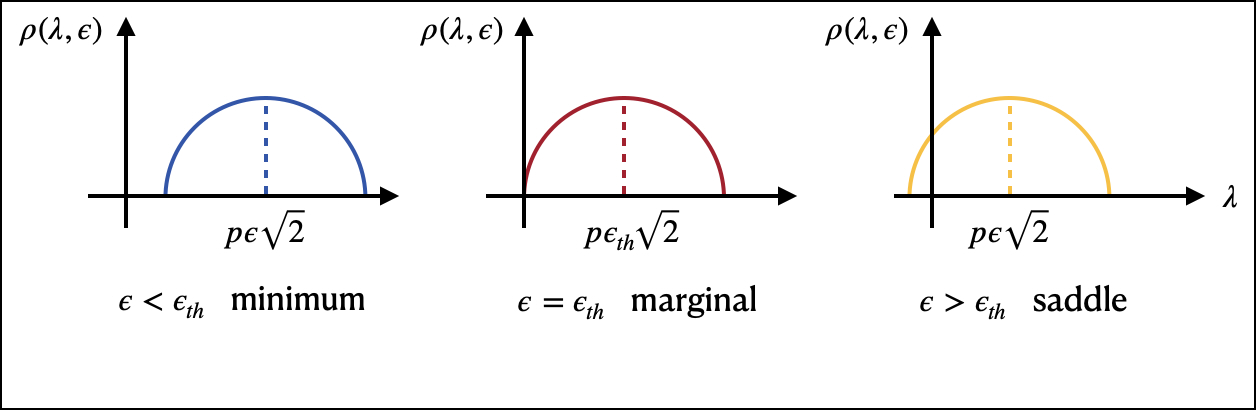}\\
\caption{Idealized representation of the spectral density of the matrix $\nabla^2_\perp h$ for $N\to\infty$. From left to right the energy density $\epsilon$ is increased, showing a transition in the stability of the Hessian, from local minima (positive spectrum) to marginal minima (the left edge touches zero) to saddles (a portion of the spectrum is negative). The spectrum is always centered at $p\epsilon\sqrt{2}$. }
\label{fig:stability_pspin_fp}
\end{figure}

\noindent Let us remark that for $N\to\infty$ the spectrum of the Hessian $\tilde{\nabla}^2_\perp h$ follows a Wigner's semicircle centered at $-p\epsilon\sqrt{2}$ and supported on $\left[-2\sqrt{p(p-1)}-p\epsilon\sqrt{2},\,\,2\sqrt{p(p-1)}-p\epsilon\sqrt{2}\right]$. As we change the energy $\epsilon$, we are shifting this support, thus changing the nature of the stationary points counted. Indeed, we see that the nature (i.e. minima or saddles) of typical stationary points depends on $\epsilon$ only, see Fig.~\ref{fig:stability_pspin_fp}. In particular, we see that for $\epsilon<\epsilon_{th}$ the points are local minima (the spectrum lies in the positive axis) and for $\epsilon>\epsilon_{th}$ the points are saddles with an extensive number of negative directions (the spectrum has a portion on the negative axis)\footnote{this is true considering the extensive instability index (that is, where the support of the semi-circle is located) and does not take into account atypical fixed points that might have isolated eigenvalues.}. The threshold value is found by imposing that the leftmost point of the support touches zero \cite{Crisanti_TAP_pspin_95, cavagna1998stationary, ros2019complex}: 
\begin{align}
\label{eq:pspin_threshold}
    -2\sqrt{p(p-1)}-p\epsilon_{th}\sqrt{2}=0\Rightarrow\epsilon_{th}=-\sqrt{\frac{2(p-1)}{p}}. 
\end{align}
The stationary points at the energy threshold are often referred to as \textit{marginal} or \textit{gapless}.

\subsubsection{Final result and numerical plot}
We have therefore completed the computation of the complexity. Putting everything together we find:
\begin{align}
    \Sigma(\epsilon)=\frac{1}{2}\left[1+\log(2(p-1))\right]-\epsilon^2+I\left(\epsilon\sqrt{\frac{p}{p-1}}\right)
\end{align}
From this expression, we can easily compute the ground state energy. In fact, we can just solve $\Sigma(\epsilon)=0$ for $\epsilon<\epsilon_{th}$. In the case $p=3$ we find $\epsilon_{gs}\approx -1.172$. A plot of this expression is shown in Fig.~\ref{fig:3_spin_complexity_1} for $p=3$.

\begin{figure}[t!]
\centering
\includegraphics[width=0.7\textwidth, trim= 5 5 5 5, clip]{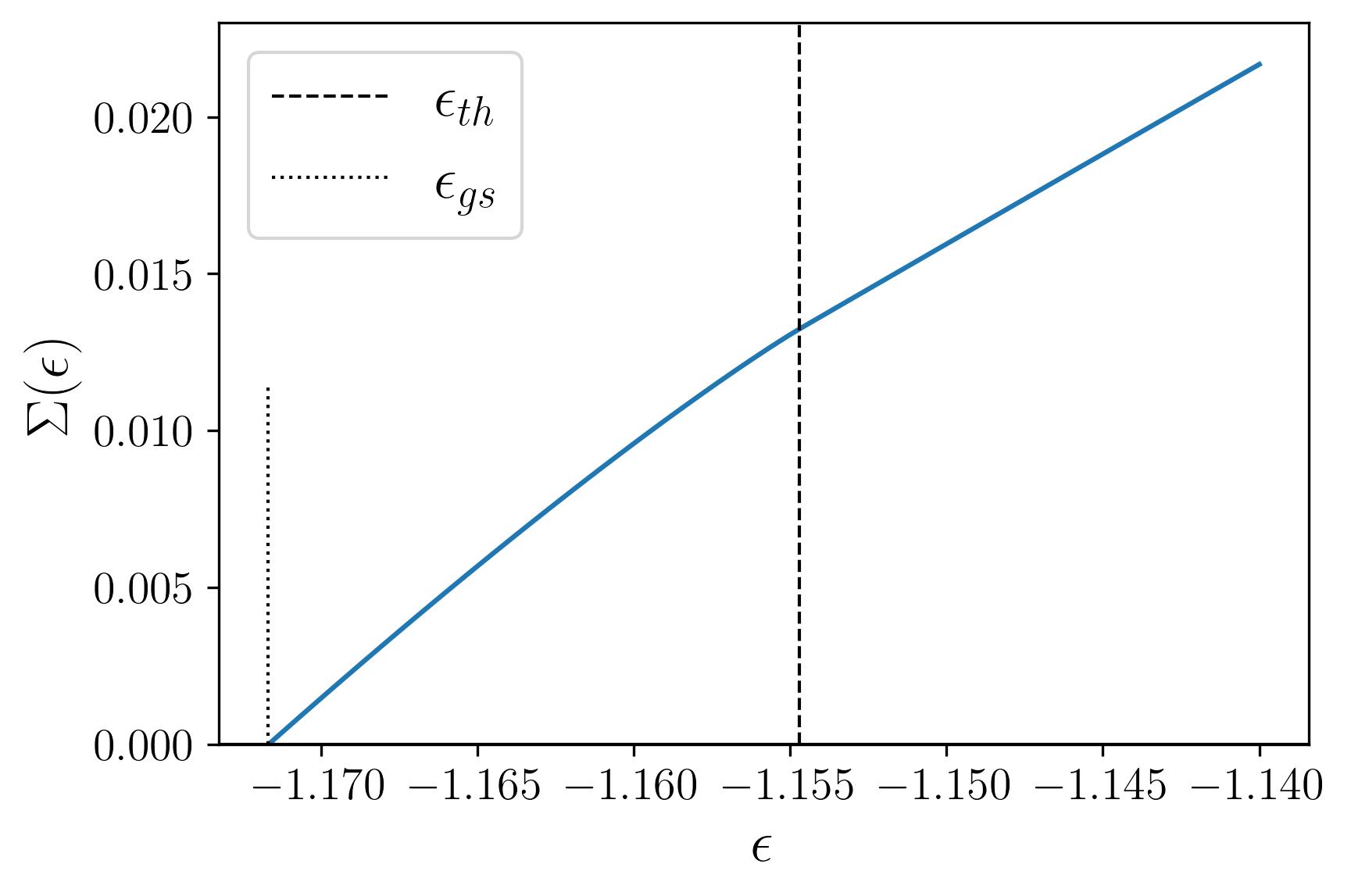}\\
\caption{Plot of the complexity $\Sigma(\epsilon)$ for the case $p=3$. To the right of $\epsilon_{th}$ the complexity refers to saddle points with an extensive number of negative eigenvalues of the Hessian. At $\epsilon_{th}$ the stationary points are marginal minima, see Fig.~\ref{fig:stability_pspin_fp}, whereas between $\epsilon_{gs}$ and $\epsilon_{th}$ they are local minima. At $\epsilon_{gs}$ they are absolute minima.}
\label{fig:3_spin_complexity_1}
\end{figure}

\subsection{Dynamical mean-field theory}
\label{sec:pspin_dynamics_calcs}

We devoted a self-contained section in the Appendix.~\ref{app:dynamical_calculations} to the derivation of the Dynamical Mean-Field Theory (DMFT) equations for a more general family of landscapes, of which the spherical $p$-spin is a special case. Let us briefly resume here the main ideas. To model a physical system in contact with a thermal bath at fixed temperature $T$ (canonical ensemble) a commonly used stochastic differential equation (SDE) is given by the overdamped Langevin equation \cite{kurchan_six_2009, cugliandolo2002dynamics, spin_glass_beyond_86}. This equation is an approximation of Newton's equations of motion, where inertia (i.e., the term proportional to the acceleration of the system) is neglected and where the interaction with the thermal bath is modeled by a Gaussian white noise. Gradient descent is a specific realization of this equation, with $T=0$. In the case of the pure spherical $p$-spin model the Langevin equation reads:
\begin{align}
\label{eq:def_langevin}
    \frac{d{\bf x}}{dt}=-\lambda({\bf x}){\bf x}-\nabla\mathcal{E}({\bf x})+{\bm \eta}(t)
\end{align}
where ${\bm\eta(t)}$ is a zero-mean Gaussian white (or additive) noise with covariance $\langle{\bm\eta}_i(t){\bm\eta}_j(t')\rangle=2T\delta_{ij}\delta(t-t')$ ($\langle\cdot\rangle$ indicates the average over this noise). The reason to write this equation is that one can show, by means of the Fokker-Planck equation, that the equilibrium distribution of this SDE is given by the Boltzmann-Gibbs distribution $\propto e^{-\mathcal{E}({\bf x})/T}$. The goal of DMFT is to obtain an effective SDE for one single unit, representative of the average behavior of all other units, and dates back 40 years \cite{dominicis_dmft_78, sompo_zipp_dyn_phase_1981}. This can be done thanks to three important aspects of the problem under study: the fact that $N$ is large, the fact that interactions are \textit{all-to-all} and the fact that we have a quenched disorder to average over. The core idea is that relevant observables are self-averaging, meaning that as $N\to\infty$, the probability of them diverging from their expected value goes to 0. By averaging over the quenched disorder, we see that the original Langevin equation (in the $N\to\infty$ limit) is equivalent to an SDE for a single unit with a Gaussian noise term, whose covariance is self-consistently determined by the autocorrelation function of the system, and with a memory kernel, encoding correlations with past configurations. This noise term encodes for the average noisy signal that each unit receives from the interaction with all the other units. \\

\noindent A pedagogical explanation of this technique is found in Ref.~\cite{HeliasBook20}.  To obtain the effective SDE (in the Itô convention in our case) one considers the probability of a certain path $p(\{{\bf x}(t')\}_{t'\leq t})$, averaged over the noise ${\bm \eta}$, and writes it as a path integral. To this probability is associated a characteristic function $Z$, which is averaged over the quenched disorder. This averaging creates an effective characteristic function elevated to the power of $N$, which is interpreted as the characteristic function of an effective unit. From this effective characteristic function, we can therefore extract the SDE governing the dynamics of the representative unit. For the $p$-spin we obtain \cite{crisanti1993spherical, cugliandolo1993analytical, franz1995recipes}:
\begin{align}
\begin{split}
&\partial_tx(t)=-\lambda(t)x(t)+\frac{p(p-1)}{2}\int_0^t ds\,R(t,s)C^{p-2}(t,s)x(s)+\eta(t)\\
&\langle\eta(t)\eta(s)\rangle=2T\delta(t-s)+\frac{p}{2}C^{p-1}(t,s),
\end{split}
\end{align}
where $\langle\cdot\rangle$ is the average over the new effective noise $\eta$, which encodes for the original thermal and quenched averages, see Appendix~\ref{app:dynamical_calculations} for details. The observables of interest are the autocorrelation function $C(t,s)=\langle x(t)x(s)\rangle$, the response function $R(t,s)=\langle\delta x(t)/\delta\eta(s)\rangle$ and $\lambda(t)$ which represents the average value of the Lagrange multiplier at time $t$ (and is thus connected to the energy density). The spherical constraint is imposed by setting $C(t,t)=1$ at all times, and moreover one obtains $R(t,t')=0$ for $t< t'$ with $\lim_{t'\to\ t^+}R(t',t)=1$. From this effective equation and using $C(t,t)=1$, one can obtain the DMFT equations for $C,R,\lambda$ \cite{crisanti1993spherical, cugliandolo1993analytical} (see Appendix for a derivation):
\begin{align}
\label{eq:pspin_dmft_equations_general}
&\lambda(t)=T+\frac{p^2}{2}\int_0^t ds\,R(t,s)C^{p-1}(t,s)\\
\begin{split}
&\partial_t C(t,t')=-\lambda(t)C(t,t')+\frac{p(p-1)}{2}\int_0^tds\,R(t,s)C^{p-2}(t,s)C(s,t')\\
&\quad\quad\quad\quad+\frac{p}{2}\int_0^{t'}ds\,R(t',s)C^{p-1}(t,s)+2TR(t',t)\\
\end{split}
\\
&\partial_tR(t,t')=-\lambda(t)R(t,t')+\frac{p(p-1)}{2}\int_{t'}^tds\,R(t,s)R(s,t')C^{p-2}(t,s)+\delta(t-t').
\end{align}
Where we assumed a random initial condition here. These equations are causal, and can be integrated numerically. 

\subsubsection{The energy density}
It is interesting to consider the value of the energy density at time $t$, denoted by $\epsilon(t)=\mathcal{E}(t)/N$. We have seen before in Eq.~\eqref{eq:mu_energy_T=0} that for pure gradient descent with $T=0$ the energy density is essentially the Lagrange multiplier, up to a factor $p$. However when $T>0$ this relation does not hold anymore. Indeed, the SDE \eqref{eq:def_langevin} is formally expressed using a discretization procedure. In our case we are using Itô's prescription (or left point rule). This prescription makes the writing of the SDE simpler (because the vector at a new time is only a function of the previous time), but at the price of modifying the chain rule. The easiest way to express the energy as a function of $\lambda$ is to consider the function $f({\bf x})={\bf x}^2$, which using Itô's lemma follows the SDE:
\begin{align}
\label{eq:intro_ito_f}
\frac{d{f({\bf x})}}{dt}=\left(-\lambda({\bf x}){\bf x}-\nabla\mathcal{E}({\bf x})\right)\cdot 2\,{\bf x} + 2\,T+2\,{\bf x}(t)\cdot{\bm\eta}(t).
\end{align}
Now, using that the model is spherical, we simply have that ${\bf x}^2=N$ which implies $df/dt=0$. It can be proved, see Appendix~\ref{app:dmft_equations}, that by taking the average with respect to the white noise ${\bm\eta}$ we get:
\begin{align}
    \langle {\bm\eta}(t)\cdot{\bf x}(t)\rangle=\sum_i\langle\eta_i(t)x_i(t)\rangle=R(t,t)=0
\end{align}
since in the Itô convention the response function $R$ at equal times is zero. Hence, by taking an average over ${\bm \eta}$ on both sides of Eq.~\eqref{eq:intro_ito_f}, this result together with the fact that $\nabla\mathcal{E}({\bf x})\cdot{\bf x}=N\,p\,\epsilon({\bf x})$ finally implies that 
\begin{align}
\epsilon(t)=\frac{T-\lambda(t)}{p}.
\end{align}

\subsubsection{The TTI regime}

\begin{figure}[t!]
\centering
\includegraphics[width=0.73\textwidth, trim= 5 5 5 5, clip]{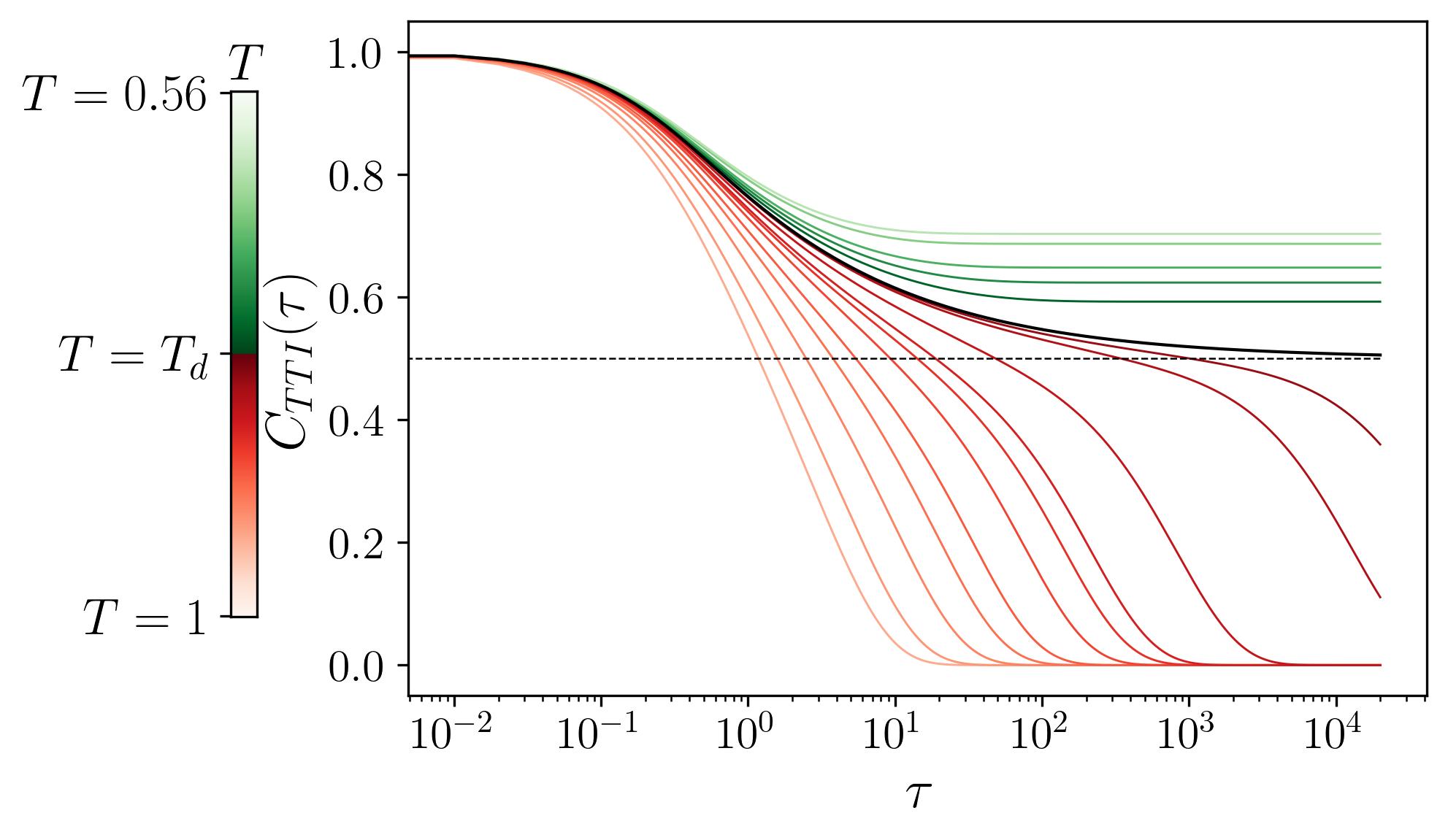}\\
\caption{Numerical solution of the equation for $C_{TTI}(\tau)$ in the TTI regime, for a range of energies above (red) and below (green) $T_d$. The $x$-axis is in log scale. We used Euler discretization with a time step $dt = 0.01$ and $p=3$. 
}
\label{fig:pspin_dmft_tti}
\end{figure}

Let us start by assuming that the system is at equilibrium, i.e. that the probability of being at a certain state follows the Boltzmann-Gibbs distribution. At equilibrium, observables are time translationally invariant (TTI) and the Fluctuation-Dissipation Theorem (FDT) applies. Mathematically, this means: 
\begin{align}
&C(t,t')=C(t-t')\equiv C(\tau)\\
&R(t,t')=R(t-t')\equiv R(\tau)\\
&R(\tau)=-\frac{1}{T}\partial_\tau C(\tau)
\end{align}
\begin{figure}[t!]
\centering
\includegraphics[width=0.6\textwidth]{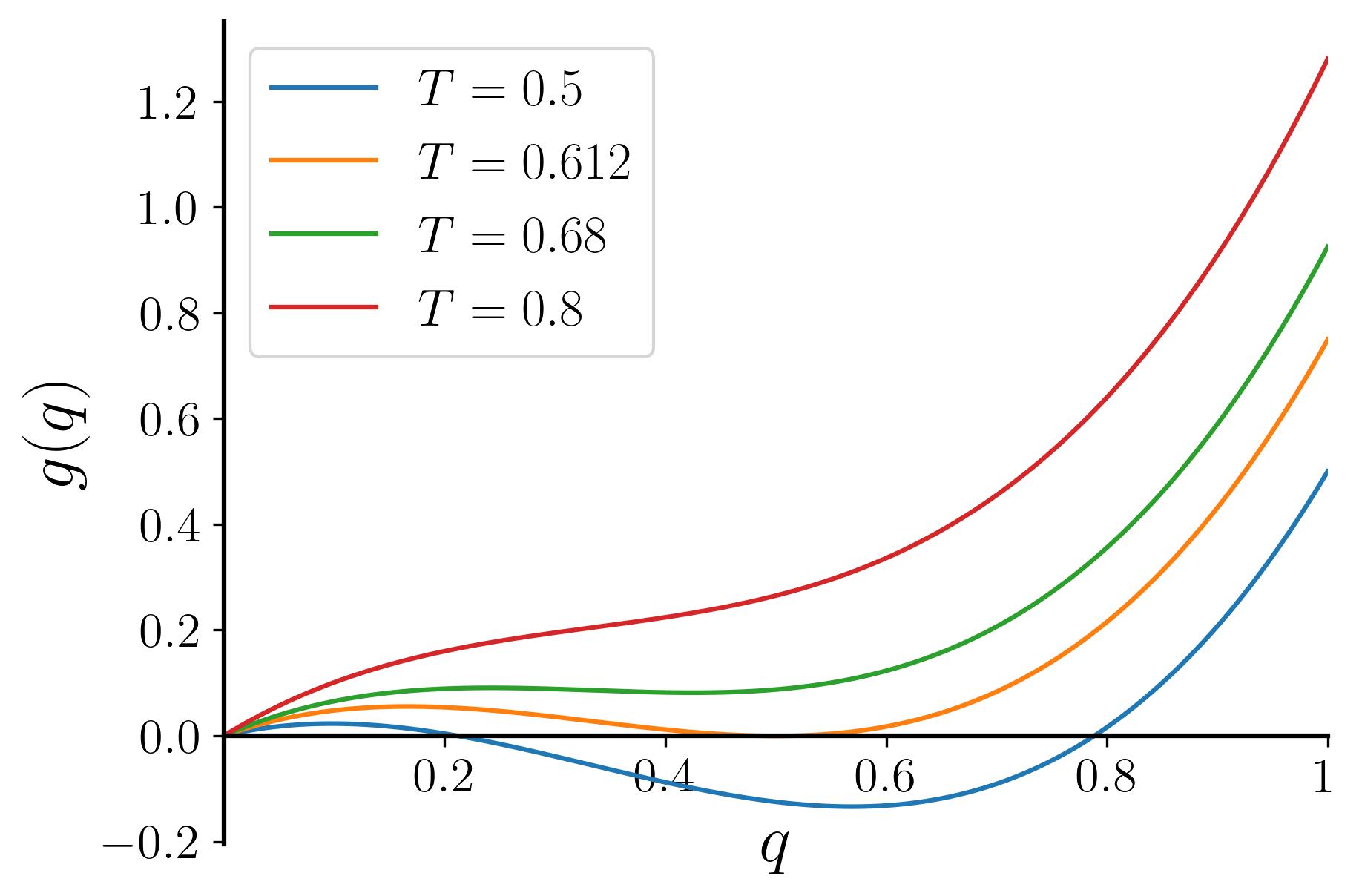}\\
\caption{Plot of $g(q)$ for various temperatures, showing that at $T_d$ (orange curve) the second minimum touches the $q$-axis.}
\label{fig:g(q)}
\end{figure}
One can then show from the DMFT equations above that the autocorrelation function $C(\tau)$ solves the following self-consistent equation:
\begin{align}
\label{eq:C_tti_eqn}
    \dot{C}(\tau)=-TC(\tau)-\frac{p}{2T}\int_0^\tau ds\,C^{p-1}(\tau-s)\dot{C}(s)
\end{align}
where $\dot{C}\equiv\partial_\tau C$. A numerical integration of this equation is shown in Fig.~\ref{fig:pspin_dmft_tti}. The picture clearly shows the appearance of a plateau of $C$ as one lowers the temperature below a critical value $T_d$. We see that for $T>T_d$, $\lim_{\tau\to\infty}C(\tau)=0$, while for $T<T_d$ we have $q_{TTI}(T):=\lim_{\tau\to\infty}C(\tau)|_{T}$, that is a $T$ dependent asymptotic value is reached. In particular, we denote $q_d\equiv q_{TTI}(T_d)$. The subscript "d" stands for dynamical, meaning that this dynamical order parameter, $C$, develops a plateau at $T_d$, with asymptotic value of $C$ given by $q_d$. A non-rigorous argument (in the spirit of \cite{Castellani_2005}) to obtain such values goes as follows:
consider fixing $T$ and, for large $\tau$, consider an expansion $C(\tau)=q_{TTI}(T)+h(\tau)=q+h(\tau)$ (we write just $q$ for simplicity here) with $|h(\tau)|<<1$ and $\lim_{\tau\to\infty}h(\tau)=\lim_{\tau\to\infty}\dot h(\tau)=0$. The goal is to linearize Eq.~\eqref{eq:C_tti_eqn} in $h$, in particular we have that $\dot{C}(\tau)=\dot{h}(\tau)$ and 
\begin{align}
    C^{p-1}(\tau)=(q+h(\tau))^{p-1}=q^{p-1}+(p-1)q^{p-2}h(\tau)+\mathcal{O}(h^2(\tau))
\end{align}
which implies 
\begin{align}
\int_0^\tau ds\,C^{p-1}(\tau-s)\dot{C}(s)=q^{p-1}[h(\tau)-h(0)]+\mathcal{O}(h^2(\tau)).
\end{align}
By plugging these results in Eq.~\eqref{eq:C_tti_eqn} we get
\begin{align}
\begin{split}
\dot{h}(\tau)&=-T\,q-T\,h(\tau)-\frac{p}{2T}q^{p-1}[h(\tau)+q-1]\\
&=-T\,q+\frac{p}{2T}q^{p-1}(1-q)-h(\tau)\left[T+\frac{p}{2T}q^{p-1}\right]
\end{split}
\end{align}
and by taking $\tau\to\infty$ we get the equation relating $T$ and $q_{TTI}(T)$:
\begin{align}
pq^{p-1}(1-q)=2T^2q.
\end{align}
If we define the function $g(q):=2T^2q-pq^{p-1}(1-q)$ we easily see that $g(0)=0$ is always a solution, but that as $T$ decreases, this function develops two fixed points (a local minimum and a local maximum), and that for $T$ low enough the local minimum touches the $q$ axis, that is, a new solution for $q$ appears, see Fig.~\ref{fig:g(q)}.
Hence the critical value is found by imposing that $g'(q)=0$, which gives:
\begin{align}
\begin{cases}
    g(q_d)=0\\
    g'(q_d)=0
\end{cases}
\quad \Rightarrow \quad T_d=\sqrt{\frac{p(p-2)^{p-2}}{2(p-1)^{p-1}}},\quad\quad q_d=\frac{p-2}{p-1}.
\end{align}
This emerging picture is a clear signature of ergodicity breaking, in the sense that below $T_d$ the assumption $\lim_{t\to\infty}C(t)=0$ ceases to be true.

\begin{figure}[t!]
\centering
\includegraphics[width=\textwidth, trim= 5 5 5 5, clip]{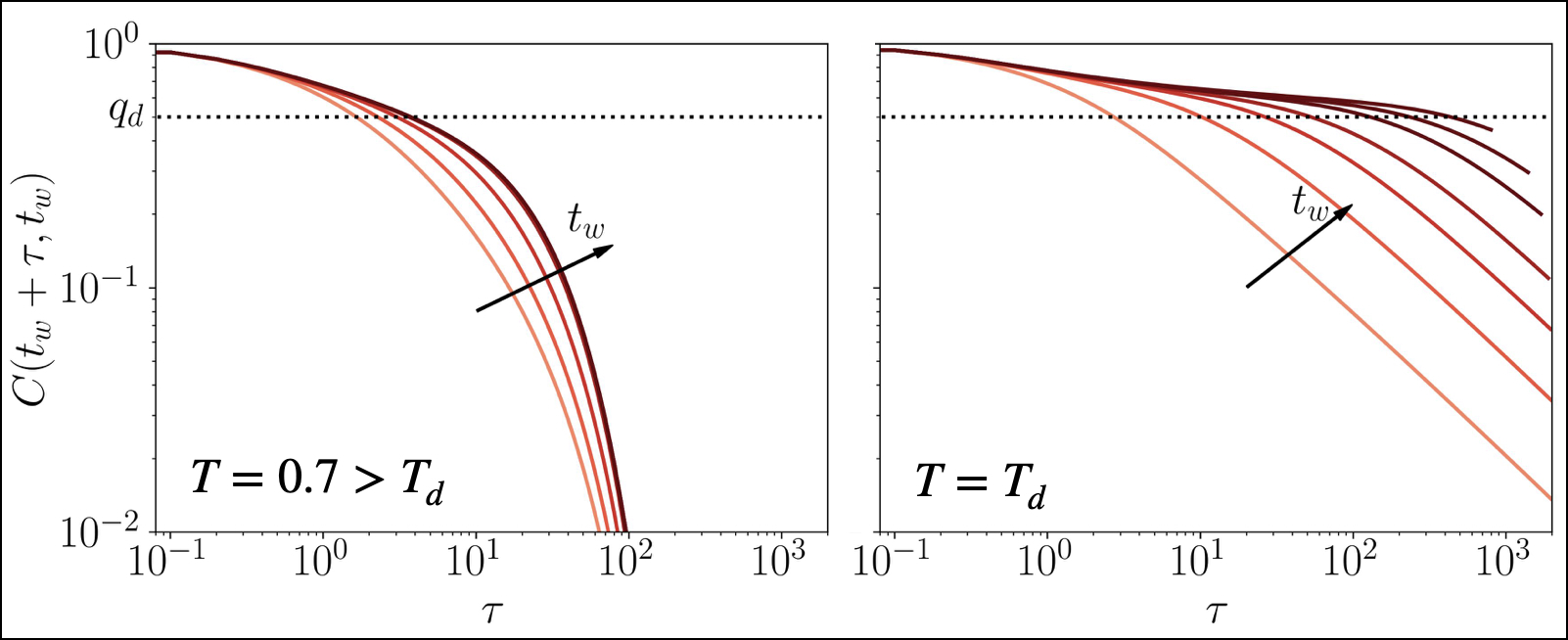}\\
\caption{Numerical integration of the DMFT equations for $p=3$. \textit{Left}. We consider $T>T_d$ and show that the curves $C(t_w+\tau,t_w)$ collapse onto each other as $t_w$ is increased. \textit{Right} We use $T=T_d$, to show that aging occurs, with a plateau at $q_d$. Both plots are in log-log scale.
}
\label{fig:pspin_dmft_2plots}
\end{figure}

\subsubsection{Out-of-equilibrium regime}
By investigating the TTI solution to the DMFT we have seen that ergodicity breaking occurs for temperatures $T\leq T_d$. We can therefore directly integrate numerically the DMFT equations, and plot the autocorrelation function. Here we use a simple Euler discretization scheme with a time step $dt=0.1$ and final time of the order of $10^3$; faster algorithms can be found in the literature \cite{Bongsoo_Kim_num_2001, mannelli_pitfalls_2020, tersenghi_seb_mixed_2025}. In Fig.~\ref{fig:pspin_dmft_2plots} we show the behavior of $C(t_w+\tau,t_w)$ as a function of $\tau$ for increasing $t_w$. We see a different behavior for $T>T_d$ and $T=T_d$: in the former case the various curves collapse onto a unique curve as $t_w$ is increased, while in the latter this does not happen, a phenomenon that takes the name of aging \cite{bouchaud1992weak, bouchaud1995aging, bouchaud1998out, cugliandolo1993analytical, Cugliandolo_1995_weak}. We see that, for fixed $t_w$,  $C(t_w+\tau,t_w)$ eventually goes to 0 as we increase $\tau$. This fact, together with the out-of-equilibrium nature of the problem mentioned above, take the name of \textit{weak ergodicity breaking} \cite{bouchaud1992weak,cugliandolo1993analytical}. The detailed analytical solution of this scenario can be found in \cite{Cugliandolo_1995_weak}, and it relies on a separation of time scales, and on an ansatz on the form of the autocorrelation and response functions. From Fig.~\ref{fig:energy_threshold_dmft} we see that gradient descent initialized at a random initial condition converges to the threshold energy $\epsilon_{th}$ of marginal minima, found in Eq.~\ref{eq:pspin_threshold}. Therefore, the threshold states, being the first local minima to appear as we go down in the landscape, are responsible for trapping the dynamics quenched below $T_d$.

\begin{figure}[t!]
\centering
\includegraphics[width=0.65\textwidth, trim= 5 5 5 5, clip]{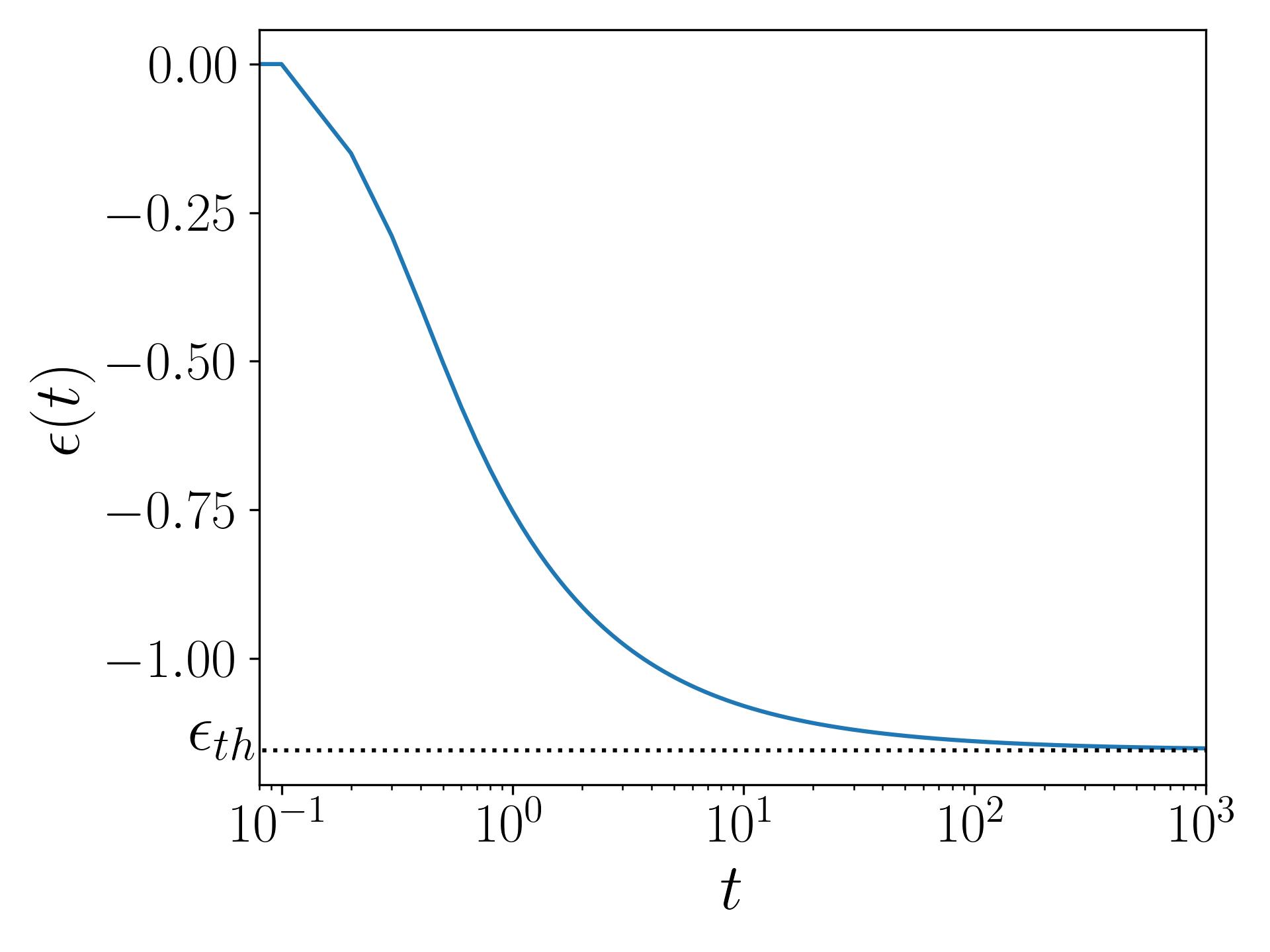}\\
\caption{Plot of the energy density $\epsilon(t)$ obtained by integration of the DMFT equations with $T=0$ (gradient descent) from infinite temperature, with $p=3$. The horizontal axis is in log scale.
}
\label{fig:energy_threshold_dmft}
\end{figure}

\subsection{Free energy}
\label{sec:p_spin_free_energy}
The free energy was computed in \cite{crisanti1992sphericalp} and a detailed computation of the one step replica symmetry breaking is found in the notes by Cavagna and Castellani \cite{Castellani_2005}. Here we shall summarize the main steps of the derivation and of the final result. The partition function at inverse temperature $\beta:=1/T$ and fixed disorder (denoted by $J$ to follow \cite{Castellani_2005}), is given by:
\begin{align}
Z_J(\beta):=\int d{\bf s}\,e^{-\beta\mathcal{E}({\bf s})}
\end{align}
where we avoid to indicate that the integration of ${\bf s}$ is over $\mathcal{S}_N(\sqrt{N})$ assuming it for granted. The problem with taking an average over the partition function is that it is not self-averaging as $N\to\infty$, meaning that a typical sample will not, in general, be close to the average value. The right quantity to work with is the free energy density, defined as:
\begin{align}
F_J(\beta):=-\frac{1}{\beta N}\log Z_J(\beta).
\end{align}
Then, for large $N$, the average of this quantity is representative of a typical sample, meaning that $\lim_{N\to\infty}\sqrt{\text{Var}(F_J(\beta))}/\mathbb{E}[F_J(\beta)]=0$. We will therefore compute a quenched average of the free energy density, defined as:
\begin{align}
F(\beta):=-\lim_{N\to\infty}\frac{1}{\beta N}\mathbb{E}\log Z_J(\beta).
\end{align}
The recipe to compute this quenched quantity using the replica method is the same as given in Sec.~\ref{sec:topo_complexity}, where now we don't have a flat integral over \textit{stationary points}, but over \textit{configurations} sampled with the Boltzmann-Gibbs measure. We thus have to compute $\mathbb{E}[Z_J(\beta)^n]$, which can be done by introducing the matrix of overlaps between the different replicas of the system: $\hat{Q}_{ab}:={\bf s}^a\cdot{\bf s}^b/N$. One then finds \cite{Castellani_2005}:
\begin{align}
\mathbb{E}[Z_J(\beta)^n]\propto \int \prod_{a<b}dQ_{ab}\,\,e^{-\frac{N}{2}S(\hat{Q})}
\end{align}
where $\hat{Q}$ has $n(n-1)/2$ independent entries and we neglect prefactors that are sub-exponential in $N$. The form of this effective action $S$ reads:
\begin{align}
    S(\hat{Q})=-n\log(2\pi e)-\log\det\hat{Q} - \frac{\beta^2}{2}\sum_{a,b}Q_{ab}^p.
\end{align}

\noindent As we already discussed in Sec.~\ref{sec:topo_complexity}, the use of the replica method consists in exchanging the $N\to\infty,\,n\to 0$, making an ansatz on the structure of $\hat{Q}$ and assuming that the results obtained for $n\in\mathbb{N}$ can be analytically continued to $n\to 0$. The great effect of the replica method is that it decouples the sites (i.e. ${\bf s}_i$) but couples different replicas (at the same site). Indeed, the free energy density now takes the form:
\begin{align}
\label{eq:F_from_S}
F(\beta)=\text{extr}_{\hat{Q}}\lim_{n\to 0}\,\frac{S(\hat{Q})}{2n},
\end{align}
where we will henceforth neglect the constant term in $S$, since it just provides a constant shift for any solution of $\hat{Q}$. There are two important issues to discuss before seeing the result. The first one is that the matrix $\hat{Q}$ has $n(n-1)/2$ independent entries, which becomes negative as $n\to 0$. This implies that, when performing a saddle point, the action $S(\hat{Q})$ changes sign as we send $n\to 0$. Indeed, $n(n-1)/2=n^2/2-n/2\sim -n/2$ for $n<<1$. Hence, terms of this type change the sign of the optimization problem, which becomes a maximization of $S$ rather than a minimization. Second, when making an ansatz on the shape of the optimum $\hat{Q}^*$, we must make sure that this solution is, at least, stable with respect to variations of $\hat{Q}^*+\delta \hat{Q}$. If it were not stable, we would be sure that the ansatz is wrong. Hence, one should check that by perturbing the overlap matrix, one maintains the maximum of $S$. In \cite{crisanti1992sphericalp} two solutions to the problem were considered, the RS and 1RSB. One can show that the RS solution is stable at high-temperatures, and unstable at low temperatures. This therefore suggests for a breaking of the replica symmetry.

\subsubsection{RS solution}
The RS ansatz corresponds to choosing $Q^{RS}_{ab}=\delta_{ab}+(1-\delta_{ab})q_0$, and it indicates the presence of a unique state with self-overlap $q_0$. By plugging this above, the free energy then reads \cite{crisanti1992sphericalp, Castellani_2005}:
\begin{align}
    F^{RS}(\beta)=-\frac{1}{2\beta}\left[\frac{\beta^2}{2}(1-q_0^p)+\log(1-q_0)+\frac{q_0}{1-q_0}\right]
\end{align}
where $q_0$ satisfies:
\begin{align}
\beta^2\frac{p\,q_0^{p-1}}{2}=\frac{q_0}{(1-q_0)^2}.
\end{align}
As one can see by doing the 1RSB ansatz, this solution is acceptable only for $T>T_s$ (introduced below), in which case the correct self-overlap is $q_0=0$. This is therefore the paramagnetic state, which can be obtained also by performing an annealed computation, and has free energy density:
\begin{align}
F^{RS}(\beta)=-\frac{\beta}{4}.
\end{align}

\subsubsection{1RSB solution}
The 1RSB solution is obtained with the following ansatz:

\begin{align}
\label{eq:structure_1rsb_matrix}
\hat{Q}^{1RSB} :=\;
\overbrace{
\begin{pmatrix}
\begin{array}{c|c|c}
%
%
\overbrace{\begin{matrix}
  1 & \cdots & q_1\\
  \vdots & \ddots & \vdots\\
  q_1 & \cdots & 1
\end{matrix}}^{m}
&
\begin{matrix}
  q_0 & \cdots & q_0\\
  \vdots & \ddots & \vdots\\
  q_0 & \cdots & q_0
\end{matrix}
&
\mathbf \ldots
\\ \hline
%
%
\begin{matrix}
  q_0 & \cdots & q_0\\
  \vdots & \ddots & \vdots\\
  q_0 & \cdots & q_0
\end{matrix}
& \mathbf \ddots &
\mathbf \vdots
\\ \hline
%
%
\mathbf \vdots
&
\mathbf \ldots
&
\begin{matrix}
  1 & \cdots & q_1\\
  \vdots & \ddots & \vdots\\
  q_1 & \cdots & 1
\end{matrix}
\end{array}
\end{pmatrix}}^{n/m\quad\textit{columns}}
\end{align}

What is important to retain is that replicas act as probing configurations of the structure of the states (i.e. ergodic components of the Boltzmann-Gibbs measure), so that the structure of $\hat{Q}$ encodes for the probability to find configurations in different states. In this case, the overlap $q_0$ indicates the overlap between two different states, while $q_1$ the self-overlap of one state. The variable $m$ also becomes a variational parameter that is connected to the probability of a given overlap. More precisely, the distribution of overlaps obtained by sampling from the Boltzmann-Gibbs measure reads \cite{Castellani_2005}:
\begin{align}
\mathbb{E}[P(q)]=\lim_{n\to 0} \frac{m-1}{n-1}\delta(q-q_1) + \frac{n-m}{n-1}\delta(q-q_0)=(1-m)\,\delta(q-q_1)+m\,\delta(q-q_0),
\end{align}
where we see that $q_0,q_1$ are the two possible values of the overlap (the two peaks) and $m$ the probability to pick $q_0$. This relation is another heuristic part of the replica method; indeed while a priori $1\leq m\leq n$, to maintain a positive probability for the overlap $q_1$, we must promote $m$ to a real number such that $0\leq m\leq 1$ (and clearly $q_0\leq q_1\leq 1$). We can see this as a sampling procedure: if we were to sample configurations from the Boltzmann-Gibbs distribution $\propto e^{-\beta\mathcal{E}({\bm\sigma})}$ for large $N$ and low $T\leq T_s$ ($T_s$ introduced below), and create a histogram for the overlaps between configurations, we would obtain two peaks: one at $q_0$ and one at $q_1$, and with different heights according to $m$. This is a consequence of the structure of the phase space, which is divided  in "blobs" of the same size $q_1$, at distances $q_0$ one from the other (distance meant in the thermodynamic sense, i.e. by averaging over configurations weighted by $e^{-\beta\mathcal{E}({\bm\sigma})}$). \\

\noindent The rest is mainly algebra; one needs to plug the ansatz ~\eqref{eq:structure_1rsb_matrix} into Eq.~\eqref{eq:F_from_S}, and optimize over $m,q_0,q_1$, to obtain \cite{Castellani_2005}:
\begin{align}
\begin{split}
 F^{\mathrm{1RSB}}(\beta)
&=-\frac{1}{2\beta}\Bigg\{\frac{\beta^{2}}{2}\Bigl[\,1 \;+\; (m-1)\,q_{1}^{\,p} \;-\; m\,q_{0}^{\,p}\Bigr]
\;+\;
\frac{m-1}{m}\,\log\bigl(1-q_{1}\bigr)
\;\\&+\;
\frac{1}{m}\,\log\!\Bigl[m\bigl(q_{1}-q_{0}\bigr) + \bigl(1-q_{1}\bigr)\Bigr]
\;+\;
\frac{q_{0}}{\,m\bigl(q_{1}-q_{0}\bigr) + \bigl(1-q_{1}\bigr)}\Bigg\}
\end{split}
\end{align}
where $m,q_0,q_1$ solve the saddle point equations $\partial_{m,q_0,q_1}S(\hat{Q}^{1RSB})=0$:
\begin{align}
&q_0=0\\
&
(1 - m)\,
\left[
\frac{\beta^{2}}{2}\,p\,q_{1}^{\,p-1}
-
\frac{q_{1}}{(1-q_{1})\bigl[(m-1)q_{1}+1\bigr]}
\right]
 \;=\; 0\\
&\frac{\beta^{2}}{2}\,q_{1}^{\,p}
\;+\;
\frac{1}{m^{2}}\,
\log\!\Bigl(\frac{1-q_{1}}{1-(1-m)\,q_{1}}\Bigr)
\;+\;
\frac{q_{1}}{m\,[\,1-(1-m)\,q_{1}\,]}
\;=\;0.
\end{align}

\begin{figure}[t!]
\centering
\includegraphics[width=\textwidth, trim= 5 5 5 5, clip]{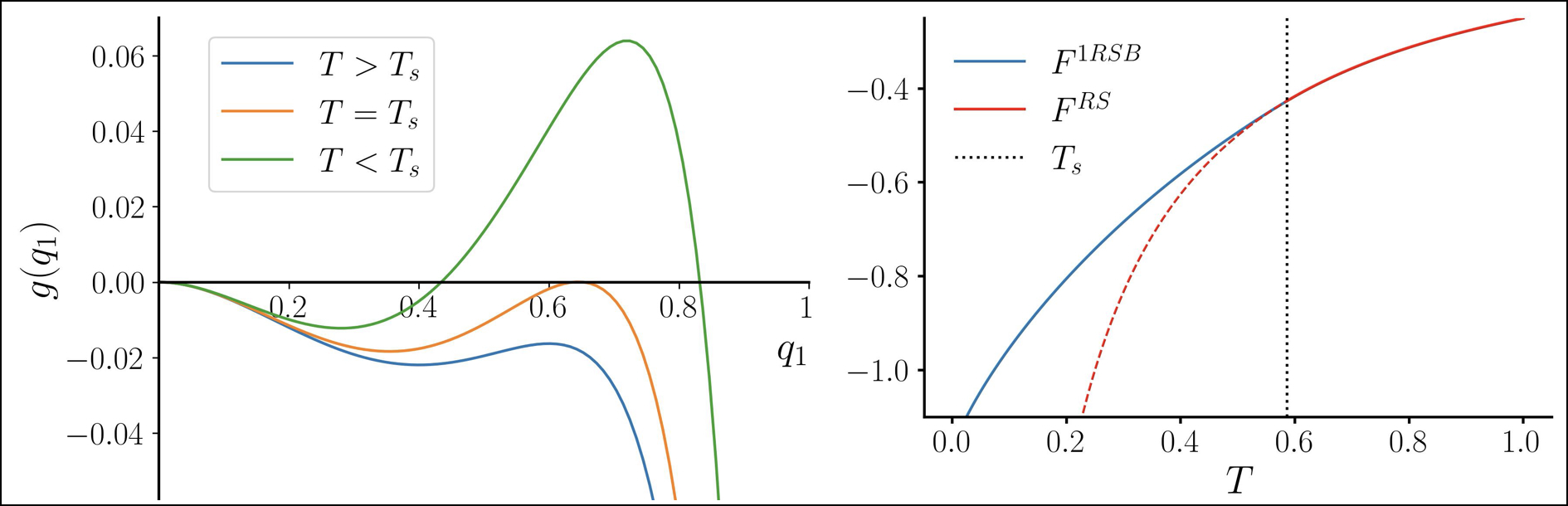}\\
\caption{\textit{Left} Plot of $g(q)$ for three temperatures, the orange one corresponds to $T_s$, and we see that the second minimum touches zero at $q_1=q_s>0$. \textit{Right} Plot of the $RS$ and $1RSB$ free energies for $T\in[0,1]$ (the constant term in both expressions was neglected). }
\label{fig:free_energies_pspin_puro}
\end{figure}

The solution $q_0=0$ says that different states are usually orthogonal to each other. Then, for high $T$ we expect that $q_1=0$: we are in the paramagnetic state and the RS solution is recovered. Now, we look for a transition in $q_1$ as we lower $T$ towards a temperature $T_s$ ("s" for static transition). This transition can be captured by using $m=1$ above, thus solving for $q_1$ in the equation $g(q_1)=0$, defined as $g(q_1):=\frac{\beta^2}{2}q_1^p+\log(1-q_1)+q_1$ (i.e. the last equation of the three above with $m\to 1$). A numerical study of this equation is shown in Fig.~\ref{fig:free_energies_pspin_puro}. We see that the maximum $g(0)=0$ is surpassed by another maximum that touches 0 at a temperature $T_s$ (orange line). At $T_s$ this second maximum has a well defined overlap $q_s>0$; we can solve for the couple $(q_s,T_s)$ by solving $g(q_s)|_{T=T_s}=g'(q_s)|_{T=T_s}=0$:
\begin{align}
\begin{cases}
&\frac{\beta^2}{2} q_1^p + \log(1 - q_1) + q_1=0\\
&\frac{p\,\beta^2}{2} q_1^{p-1} - \frac{1}{1 - q_1} + 1=0
\end{cases}\quad\overset{p=3}{\Longrightarrow}\quad T_s\approx 0.586,\, q_s\approx 0.645.
\end{align}
In Fig.~\ref{fig:free_energies_pspin_puro} we plot the free energy and the overlap $q_1$ as a function of temperature. As $T$ is lowered below $T_s$, the solution takes values $m<1$ and as $T\to 0$, $m\to 0$ and $q_1\to 1$. Let us finally remark that this transition at $T_s$ is of first order in the overlap (discontinuous jump) but is second order in the thermodynamic sense, i.e. discontinuous in the second derivative of the free energy \cite{crisanti1992sphericalp}. \\

\noindent So far, we have seen that the pure spherical $p$-spin model with $p>2$ has a static (thermodynamic transition) at $T_s$, where the correct solution becomes 1RSB. However, we have also seen that for a temperature $T_d>T_s$ we have (weak) ergodicity breaking, meaning that the system is stopped by the first local minima that appear as we decrease in energy in the landscape. We also calculated the complexity of this model, showing that it is positive above a level $\epsilon_{gs}$, and marginal states appear at $\epsilon_{th}>\epsilon_{gs}$. We saw that the energy reached by gradient descent from a random condition is precisely given by $\epsilon_{th}$, and that the system is not able to reach energies below $\epsilon_{th}$ for any $T\leq T_d$. In particular, this transition is not captured by the free energy. However, the fact that the states that appear at $T_s$ are already well formed (since $q_s>0$) hints at the presence of metastability even from this calculation. In the following, we give a summary of the TAP approach from Ref.~\cite{Castellani_2005}, and show how complexity, dynamics, and free energy are connected.

\subsection{The TAP approach}
\label{sec:tap_approach}
Here we summarize the TAP approach applied to the pure spherical $p$-spin, computed in \cite{kurchan_barriers_93, Crisanti_TAP_pspin_95}. Pedagogical calculations are found in the following Refs.~\cite{Castellani_2005, folena_these_2020, zamponi_mean_field_notes_2014, spin_glass_beyond_86, Georges_Yedidia_1991, Biroli_dynTAP_1999}. In mean-field models, we can divide the equilibrium distribution into pure states (or just states), that is, ergodic components that enjoy the clustering property, namely that spins are completely decorrelated within one state \cite{zamponi_mean_field_notes_2014}. A state $\alpha$ is defined by its average magnetizations $m_i:=\langle s_i\rangle_{\alpha}$, with $\langle\cdot\rangle_\alpha:=\int_{\alpha} d{\bf s}\,e^{-\beta\mathcal{E}({\bf s})}$ ($\alpha$ being a subset of $\mathcal{S}_N(\sqrt{N})$); the self-overlap of the pure state is then $q:=\frac{1}{N}\sum_im_i^2$. In the standard free energy approach, we first take an average over the disorder $\mathbb{E}$, and then solve the problem with replicas; the TAP approach instead has a different purpose. One aims at passing from a (standard) description in terms of the intensive variable ${\bf h}$ (an external magnetic field) to the extensive ${\bf m}$ (the vector of magnetizations), by means of a Legendre transform (which is used to change between conjugate variables), before taking any disorder average. The TAP free energy is then a free energy at fixed values of the magnetizations $m_i=\langle s_i\rangle$, and whose local minima (defined as metastable or TAP states) we want to identify with the pure states of the system. The calculation is done by defining a free energy with fixed magnetization vector ${\bf m}$, and take an expansion up to second order $\beta^2$ from the initial point $\beta=0$. In this way one can follow the evolution of the free energy at fixed magnetization ${\bf m}$ as $\beta$ is increased. For the $p-$spin, which is a fully connected model, the calculation up to $\beta^2$ is exact for $N\to \infty$ \cite{Georges_Yedidia_1991}. One starts by introducing Lagrange multipliers (i.e. the magnetic field) to impose that the constraints $\langle s_i-m_i\rangle=0$ are satisfied (the average is taken with respect to the correct Boltzmann weight, i.e. including everything that is in the exponent):
$-\beta Nf_{TAP}({\bf m},\beta):=\log Z_{TAP}({\bf m},\beta)$ with
\begin{align}
&Z_{TAP}({\bf m},\beta):=\int_{S_N(\sqrt{N})} d{\bf s}\,e^{-\beta\mathcal{E}({\bf s})+{\bm \lambda}(\beta)\cdot({\bf m}-{\bf s})};\\
& \langle s_i-m_i\rangle=0\Rightarrow \frac{\partial\log Z_{TAP}({\bf m},\beta)}{\partial m_i}=\lambda_i(\beta).
\end{align}
Notice that at the local minima of $f_{TAP}$ it must be that $\lambda_i(\beta)=0$, and also that for the spherical model an additional Lagrange multiplier must be used to enforce the spherical constraint \cite{zamponi_mean_field_notes_2014}. The goal is then to expand this expression in Taylor around $\beta=0$ \cite{Georges_Yedidia_1991}:
\begin{align}
    -\beta N f_{TAP}({\bf m},\beta)=\sum_{k=0}^\infty \frac{\partial(-\beta f_{TAP})}{\partial\beta}\Bigg|_{\beta=0}\frac{\beta^n}{n!}.
\end{align}
The final result reads \cite{zamponi_mean_field_notes_2014, Castellani_2005, Crisanti_TAP_pspin_95, kurchan_barriers_93}:
\begin{align}
f_{TAP}({\bf m},\beta)=\frac{1}{N}\mathcal{E}({\bf m})\underbrace{-\frac{1}{2\beta}\log(1-q)-\frac{\beta}{4}[(p-1)q^p-pq^{p-1}+1]}_{R(q,\beta)}
\end{align}
where we abusively use $\mathcal{E}$ even if we defined it on $S_N(\sqrt{N})$. Then by using the homogeneity of $\mathcal{E}$, we can use "angular variables" to rescale the problem \cite{kurchan_barriers_93}: $m_i=\sqrt{q}\,s_i$, $\sum_i s_i^2=N$. The result then reads $f_{TAP}({\bf s},\beta)=q^{p/2}\mathcal{E}({\bf s})/N+R(q,\beta)$, and one can directly appreciate the connection between the local minima of $\mathcal{E}$, computed in Sec.~\ref{sec:chap1_topo_compl}, and the local minima of $f_{TAP}$. Since the ${\bf s}$ and $(q,\beta)$ parts of this function are separated, one sees that local minima of $f_{TAP}$ are local minima of $\mathcal{E}$ with a non-trivial self-overlap $q$ that depends on the \textit{bare} energy density $\epsilon({\bf s}):=\mathcal{E}({\bf s})/N$ and on temperature $\beta$. Hence local minima of $\mathcal{E}$ become \textit{dressed} with thermal fluctuations for $T>0$, thus becoming metastable states. We can therefore indicate the free energy density and self-overlap of a state $\alpha$ as functions of its bare energy density $\epsilon_\alpha$: $f(\epsilon_\alpha,\beta), q(\epsilon_\alpha,\beta)$. At fixed $\epsilon_\alpha,\beta$, the self-overlap satisfies the equation $g(q,\epsilon,\beta)=0$ with:
\begin{align}
g(q,\epsilon_\alpha,\beta)=-\frac{p}{2}\, q^{\frac{p}{2}-1} \, \epsilon_\alpha
\;+\;
\frac{1}{2\beta(1-q)}
\;-\;
\frac{\beta}{4}\,\bigl[p(p-1)q^{p-1}-p(p-1)q^{p-2}\bigr].
\end{align}
One can show that this equation has no real solution for $\epsilon_\alpha>\epsilon_{th}$ (the threshold energy where marginal minima appear, see Sec.~\ref{sec:chap1_topo_compl}). We can also show \cite{Castellani_2005} that $q(\epsilon_{th},1/T_d)=q_d$, i.e. that at the threshold energy, the self-overlap at the dynamical temperature is exactly $q_d$ from Sec.~\ref{sec:pspin_dynamics_calcs}.\\

\noindent The \textit{configurational entropy} at temperature $T$ is then given by the number of states at free energy density $f$: $\Sigma_\beta(f):=(1/N)\log\sum_\alpha\delta(f-f_\alpha)$. This function, in general, need not be equal to the \textit{complexity} computed in Sec.~\ref{sec:chap1_topo_compl}, but owing to the one-to-one mapping between metastable states and local minima of $\mathcal{E}$ for this model, we have that $\Sigma_\beta(f)=\Sigma(\epsilon(f,\beta))$, where by $\epsilon(f,\beta)$ we denote the bare energy density of those states that have free energy $f$ at inverse temperature $\beta$. Also notice that $\Sigma$ is positive for any big enough energy (eventually giving the complexity of saddles), while $\Sigma_\beta$ is only positive where $\Sigma$ counts local minima. Then we have that 
\begin{align}
\label{eq:pspin_tap_eq_Z_Phi}
    Z_J(\beta)=\sum_\alpha Z_\alpha(\beta)=\sum_\alpha e^{-N\beta f_\alpha(\beta)}=\int df\sum_\alpha\delta(f-f_\alpha)e^{-N\beta f_\alpha(\beta)}=\int df\,e^{-\beta N \Phi(f,\beta)},
\end{align}
where we defined the \textit{generalized free energy} $\Phi(f,\beta):=f-T\,\Sigma_\beta(f)$. By the saddle point method we can retrieve the total free energy $F(\beta)=\Phi(f^*(\beta),\beta)$, with $f^*(\beta):=\text{argmin}_f\Phi(f,\beta)$. The bare energy density of states of free energy density $f^*(\beta)$ can be denoted as $\epsilon^*(\beta)$, and hence we have $f^*(\beta)=f(\epsilon^*(\beta),\beta)$. Thanks to the one-to-one mapping between local minima and metastable states, we can express $\Phi$ in terms of $\epsilon$ as: $\Phi(\epsilon,\beta)=\Phi(f(\epsilon,\beta),\beta)=f(\epsilon,\beta)-T\Sigma(\epsilon)\mathbb{I}_{\epsilon\in[\epsilon_{gs},\epsilon_{th}]}$, and thus $F(\beta)=\Phi(\epsilon^*(\beta),\beta)$. At fixed $\beta=1/T$ we can solve for $\epsilon^*(\beta)$ (by minimizing $\Phi$) thus recovering the free energy by solving the system:
\begin{align}
\label{eq:system_eps_q}
\begin{cases}\epsilon^*(\beta)=\frac{(p-2) q^{\frac{p}{2}} - \sqrt{p (p \,q^p + 8 T^2)}}{4 T}\\
g(q,\epsilon^*,\beta)=0
\end{cases}.
\end{align}
We can therefore divide in three situations:
\begin{itemize}
    \item $T>T_d$: the system is in the paramagnetic state $q=0,{\bf m}={\bf 0}$, we recover $F(\beta)=-\beta/4$. 
    \item $T_s<T\leq T_d$: there is a positive complexity, and the free energy is given by $F(\beta)=\Phi(\epsilon^*(\beta),\beta)$ where the system \eqref{eq:system_eps_q} must be solved in $(\epsilon,q)$ at fixed $\beta$. One recovers the paramagnetic $F(\beta)=-\beta/4$. Hence, the complexity acts as an entropy (i.e. configurational entropy) that lowers the free energy of single metastable states, thus giving back the correct total one. In particular, at $T_d$ we have the solution $\epsilon^*=\epsilon_{th}, q=q_d$, thus recovering the dynamical results: the marginal minima are those that trap the dynamics, and their self-overlap (at $T_d$) is $q_d$. For $p=3$ moreover we can solve explicitly the system to obtain $q=(1/6) (3 + \sqrt{3} \sqrt{3 - 8 T^2})$.
    \item $T\leq T_s$: we have the solution $\epsilon^*(\beta)=\epsilon_{gs}$ and $q$ given by $g(q,\epsilon_{gs},\beta)=0$. In particular, this solution appears at $T_s$, where the complexity vanishes $\Sigma(\epsilon_{gs})=0$ and the self-overlap at $T_s$ is $q_s$. Therefore, the ground states (which are sub-exponential given their zero-complexity) are exactly the states that belonged to the 1RSB solution, with self-overlap $q$. The 1RSB free energy can be recovered as $F(\beta)=f_{TAP}(\epsilon_{gs},\beta)$ and $q$ solving $g(q,\epsilon_{gs},T)=0$. This scenario takes the name of \textit{entropy crisis} \cite{Castellani_2005}. 
\end{itemize}
We have finally reunited thermodynamical, dynamical and topological calculations. The metastable states contribute to the free energy for $T\leq T_d$, but their exponential number makes it so that we recover the paramagnetic free energy. As soon as we go below $T_s$, the metastable states that contribute the most are given by the ground states of the energy landscape, and their complexity is zero.



\subsection{Franz-Parisi potential}
\label{sec:franz_parisi_analysis}
Another important tool to study mean-field glassy systems is the Franz-Parisi (FP) potential, introduced in \cite{franz1995recipes}. This tool allows us to obtain both the static and dynamic transitions from one single "potential". The FP potential is a constrained free energy to keep a secondary configuration ${\bf s}_2$ at fixed overlap $q_{12}$ from a reference configuration ${\bf s}_1$, both at temperature $T=1/\beta$ (and can be generalized to two temperatures, see e.g. Refs.~\cite{Barrat_bifurcation_95, franz1998effective, Franz_coupled_97, folena_these_2020, cavagna1997structure}). The primary and secondary configuration are also called \textit{real replicas}, and we will consider a similar computation in Chapter~\ref{chapter:energy_landscapes} but between stationary points instead of equilibrium configuration. The idea of the FP potential is to introduce a fixed configuration (the first real replica) in order to study the distribution of overlaps of the second real replica with respect to the first one. In particular, the FP potential can be seen as the large deviation function associated with the probability of finding the two real replicas at an overlap $q_{12}$, and is thus useful to probe the free energy landscape. It is defined as:
\begin{align}
&V(q_{12},\beta):=F(q_{12},\beta)-F(\beta),\\
&F(q,\beta):=-\frac{1}{\beta N}\mathbb{E}\left[\frac{1}{Z_J(\beta)}\int d{\bf s}_1\,e^{-\beta\mathcal{E}({\bf s}_1)}\log Z_J(\beta,{\bf s}_1,q)\right],\\
&Z_J(\beta,{\bf s}_1,q_{12}):=\int d{\bf s}_2\,e^{-\beta\mathcal{E}({\bf s}_2)}\delta({\bf s}_1\cdot{\bf s}_2/N-q_{12}).
\end{align}
The calculation is done by replicating both ${\bf s}_1$ and ${\bf s}_2$, thus introducing two sets of replicas: ${\bf s}_1^{(a)}$ for $a=1,\ldots,n$ and ${\bf s}_2^{(b)}$ for $b=1,\ldots,m$ with two replica indices $m,n$. The most general correct $(n+m)\times (n+m)$ overlap matrix for $T_s<T$ is \cite{franz1998effective, Franz_coupled_97, Barrat_bifurcation_95}:
\begin{align}
&\hat{Q}=\begin{pmatrix}
    \overbrace{\hat{Q}^{11}}^{n} & \overbrace{\hat{Q}^{12}}^{m}\\
    \hat{Q}^{21} & \hat{Q}^{22}
\end{pmatrix},\quad \hat{Q}^{11}_{ab}=\delta_{ab},\quad \hat{Q}^{12}_{ab}=q_{12}\,\delta_{1a}
\end{align}
and $\hat{Q}^{22}$ has a 1RSB structure identical to \eqref{eq:structure_1rsb_matrix} where (quite confusingly) the new parameters are: $m$ for the number of replicas, $x$ for the number of replica blocks, and $q_0,q_1$ as before. Then the FP potential reads \cite{cavagna1997structure}:
\begin{align}
\begin{split}
    V(q_{12},\beta)&=\text{max}_{q_0,q_1,x}-\frac{1}{2\beta}\Bigg\{\beta^2q_{12}^p+\beta^2(x-1)\frac{q_1^p}{2}-\beta^2x\frac{q_0^p}{2}+\log(1-q_1)\\
    &+\frac{1}{x}\log\left(1+x\frac{q_1-q_0}{1-q_1}\right)+\frac{q_0-q_{12}^2}{1-q_1+x(q_1-q_0)}\Bigg\}.
\end{split}
\end{align}

\begin{figure}[t!]
\centering
\includegraphics[width=0.6\textwidth]{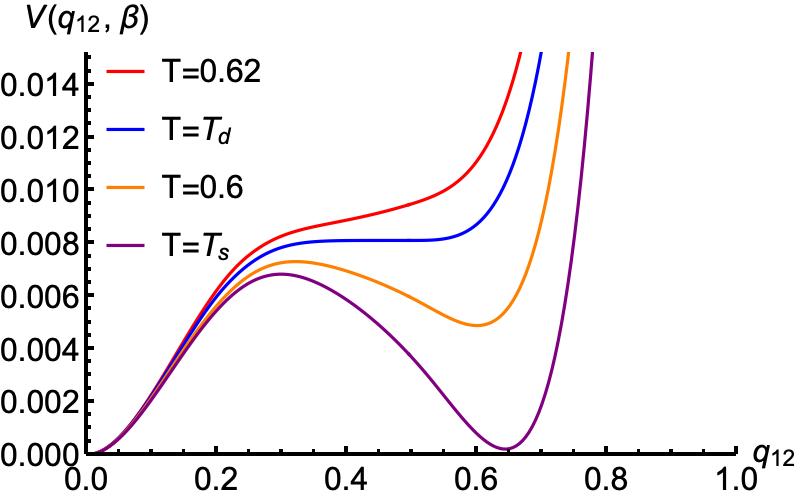}\\
\caption{Picture of the Franz-Parisi potential for the pure spherical $p$-spin model with $p=3$, at various temperatures, including $T_d$ and $T_s$. }
\label{fig:fp_potential}
\end{figure}

By optimizing over $x,q_0,q_1$ one sees that there are three regions as one increases $q_{12}$ from 0 to 1: for $q_{12}$ close to 0 we have the RS solution with $x=1$, in the middle we have $x\neq 0,1$ and $q_1\neq q_0$, and for $q_{12}$ large enough we have the RS solution $q_1=q_0$ \cite{Barrat_bifurcation_95, cavagna1997structure}. This can be obtained numerically by studying the stability of the RS solution \cite{Barrat_bifurcation_95, cavagna1997structure}. Here we limit ourselves to notice that the rightmost minimum of the potential $V$ (see Fig.~\ref{fig:fp_potential}) corresponds exactly to the metastable (TAP) states that contribute maximally to the free energy in Eq.~\eqref{eq:pspin_tap_eq_Z_Phi}. Such minimum is given by $q_0=q_1=q_{12}$ (and thus it lives in the RS regime): indeed, we have
\begin{align}
\partial[ V(q,\beta)|_{q_0=q_1=q}]/\partial q=0\Rightarrow 2 q^2 - p\,q^p \beta^2 + p\,q^{1 + p} \beta^2=0
\end{align}
which, after careful manipulations, is equivalent to \eqref{eq:system_eps_q}: this equation fixes the self-overlap (at inverse temperature $\beta$) of the metastable states that maximize the generalized free energy $\Phi$. Then if we define $r(q,\beta):=2 q^2 - p\,q^p \beta^2 + p\,q^{1 + p} \beta^2$ we recover the dynamical and statical transitions as (see Fig.~\ref{fig:fp_potential}):
\begin{align}
    \begin{cases}r(q_d,1/T_d)=0\\
    \partial_qr(q_d,1/T_d)=0
    \end{cases},\quad\quad
    \begin{cases}r(q_s,1/T_s)=0\\
    V(q_s,1/T_s)=0.
    \end{cases}
\end{align}
Let us finally remark that the height of the potential at this local minimum is related to the configurational entropy, so that the static transition at $T_s$ coincides precisely with the vanishing of the configurational entropy, as in the previous section.

\subsection{Beyond the pure model: mixed models}
\label{sec:mixed_models}

The model we have presented before corresponds to the simplest Gaussian random landscape, which is isotropic and homogeneous. It is isotropic because its law depends only on the distance between two points (i.e. the overlap) and not on their position in space; it is homogeneous because $\mathcal{E}(a{\bf x})=a^p\mathcal{E}({\bf x})$. Then, Euler's homogeneous function theorem implies exactly Eq.~\eqref{eq:intro_homog_pspin} applied to this specific example. However, in general, if this homogeneity ceases to hold, there is a priori no explicit relation between $\mathcal{E}$ and its gradient. In particular, recall that it is this property which implies that the Lagrange multiplier $\lambda$ is proportional to the energy. This in turn has profound consequences for the energy landscape: at a certain energy the local minima all have the same extensive number of unstable directions. Basically, the parameter $\lambda$ is not free, it is automatically fixed by the chosen energy level. It is this peculiar property that gives to the pure model such a nice and clear phenomenology. But how can we create a model where both the energy and the Lagrange multiplier (which controls the instability index) can be tuned ? The simplest thing to do is to break the homogeneity of the pure $p$-spin model by summing two $p$-spins with different values of $p$. In general, we can define an isotropic spherical random Gaussian field as:
\begin{align}
\mathcal{E}:\mathcal{S}_N(\sqrt{N})\to\mathbb{R},\quad \mathbb{E}[\mathcal{E}({\bf x})]=0,\quad \mathbb{E}[\mathcal{E}({\bf x})\mathcal{E}({\bf y})]=Nf\left(\frac{{\bf x}\cdot{\bf y}}{N}\right),\quad f(0)=0.
\end{align}

\begin{figure*}[t!]
  \centering
  \begin{subfigure}[b]{0.49\textwidth}
    \centering
    \includegraphics[width=\linewidth]{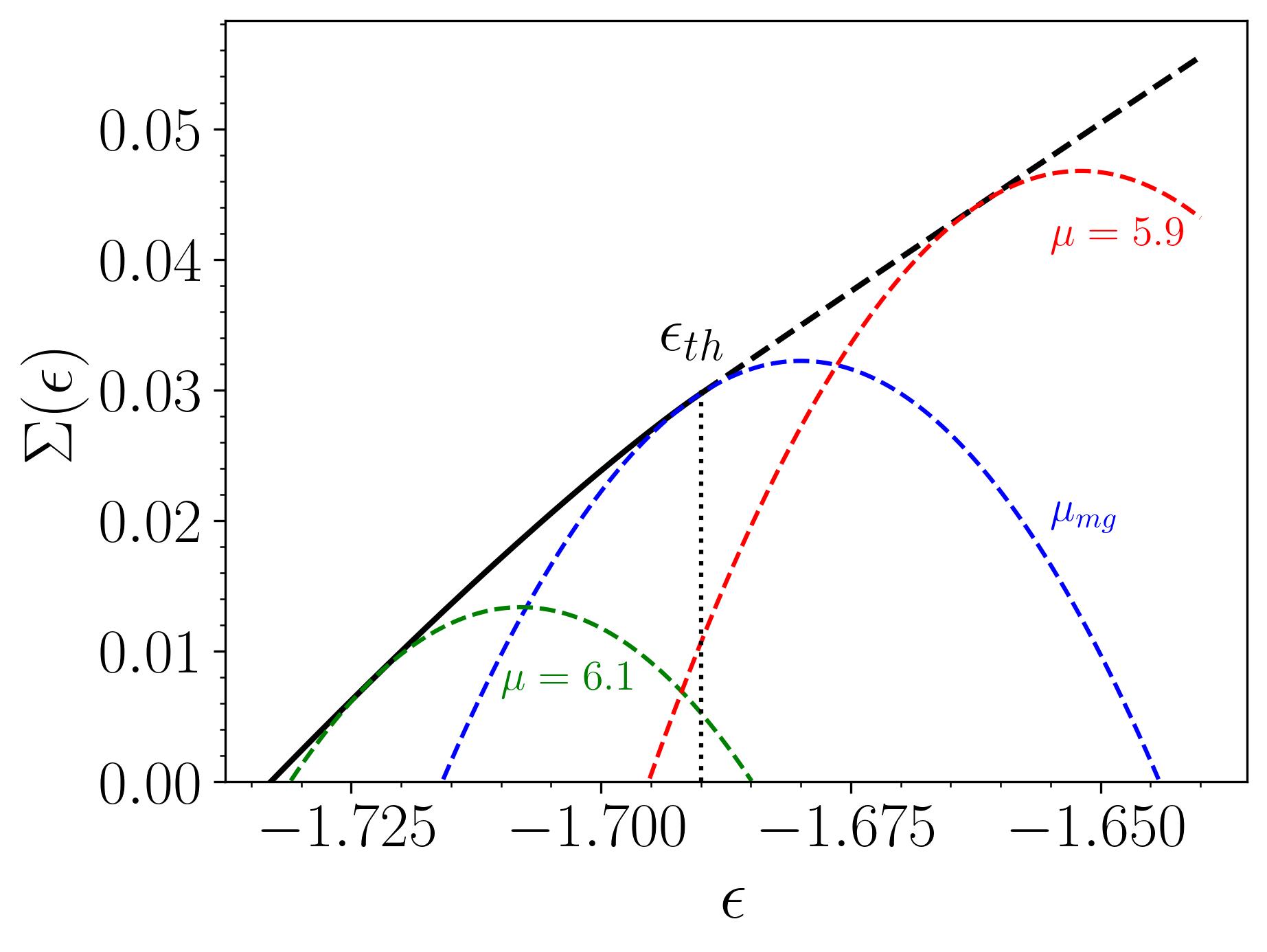}
  \end{subfigure}
  \hfill
  \begin{subfigure}[b]{0.49\textwidth}
    \centering
    \includegraphics[width=\linewidth]{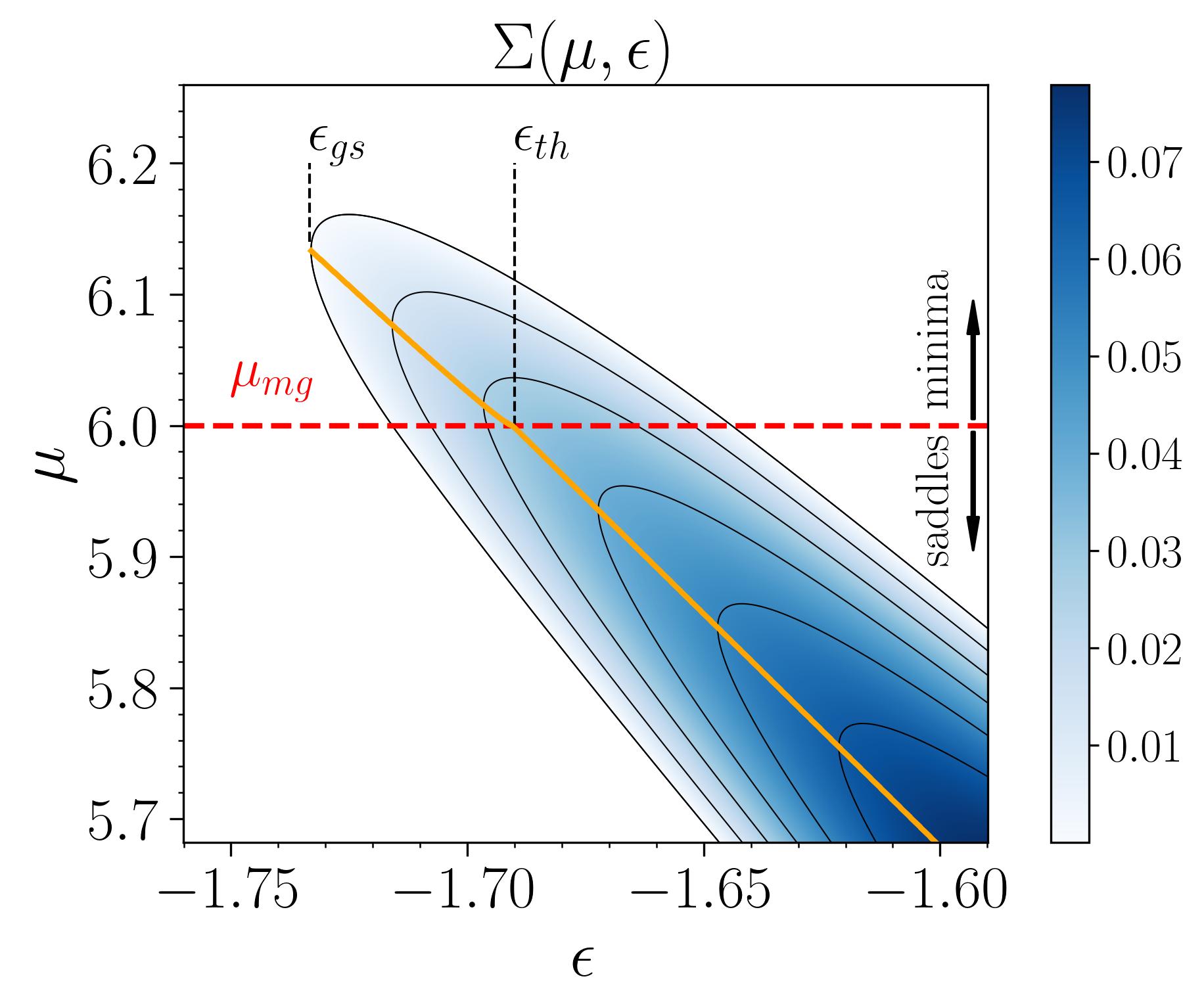}
  \end{subfigure}
  \caption{In the figure we use the $(3+4)$ spherical model, with $f(x)=(x^3+x^4)/2$. \textit{Left}. Plot of $\Sigma(\epsilon):=\text{max}_\mu\Sigma(\epsilon,\mu)$. The black line is the envelope of the complexity as a function of $\epsilon$ with free parameter $\mu$. The continuous black line represents \textit{local minima} while the dashed part \textit{extensive saddles}, that is, saddles with an extensive number of negative eigenvalues of the Hessian. Notice that $\epsilon_{th}$ is defined as the energy where the maximum is achieved at $\mu\to\mu_\text{mg}$, and that there is a broad range of energies that contain marginal minima at $\mu_\text{mg}$. 
  \textit{Right}. Color plot of $\Sigma(\epsilon,\mu)$ (where it is positive). The yellow line is the envelope: for any fixed $\epsilon$, we take the $\mu$ that maximizes the complexity. We independently reproduced the plots by looking at Refs.~\cite{folena2020rethinking, kent2024arrangement} respectively.}
  \label{fig:complexities_mixed_pspin}
\end{figure*}

If the random field is homogeneous, then there must exist a real number $k>0$ such that $f(a)=a^kf(1)$ for any $a$, automatically proving that $f$ is a monomial. Hence a form of the type $f(q)=\sum_{p>0} a_pq^p$ with $\sum_{p>0}a_p$ finite breaks the homogeneity (if at least two $a_p$'s are non-zero). This type of model is called a \textit{mixed} $p$-spin model  \cite{MixedModelNieuwenhuizen95, Barrat_bifurcation_95,Leuzzi_mixed_2004, Leuzzi_mixed_2006, benarous_mixed_2019, Auffinger_Chen_mixed_landscape_2018}, and more recently it has gained a lot of attention from recent work \cite{folena2020rethinking}, and subsequent works \cite{Folena_gd_simul_2021, folena_zamponi_weak_2023, folena_these_2020, tersenghi_seb_mixed_2025}. From a landscape point of view, the interesting thing about mixed models is that we can express the annealed complexity as a function of both the energy density and the Lagrange multiplier:
\begin{align}
\Sigma(\epsilon,\mu)=D(\mu)-\frac{\epsilon^2(f''(1)+f'(1))+2\epsilon\,\mu\,f'(1)+f(1)\mu^2}{2\left[f(1)(f'(1)+f''(1)) -f'(1)^2 \right]}
\end{align}
with 
\begin{align}
\begin{split}
D(\mu)&=\text{Re}\Bigg\{
\frac{\mu}{\sqrt{\mu^2-\mu^2_{\text{mg}}} + \mu}
        + \log\left(\sqrt{\mu^2-\mu^2_{\text{mg}}} + \mu\right)
        + \frac{1}{2}\log\left(\frac{1}{f'(1)}\right)
        - \log(2) 
\Bigg\}
    \end{split}
\end{align}
and $\mu_\text{mg}=2\sqrt{f''(1)}$. The subscript "mg" stands for marginal, as it controls the stability of the fixed points counted: if $\mu>\mu_\text{mg}$ they are overwhelmingly minima, whereas for $\mu<\mu_\text{mg}$ they are saddles or maxima. At the value $\mu=\mu_\text{mg}$ they are marginal equilibria, meaning that the leftmost edge of the support of the spectrum touches 0 in the $N\to\infty$ limit. Like for the pure model, also in this case the Hessian at fixed $\mu$ belongs to the GOE ensemble, which has a spectrum converging to the semi-circular law with eigenvalue density $\rho$ given by:
\begin{align}
    \rho(x)=\frac{2\sqrt{(x-\mu)^2-\mu_\text{mg}^2}}{\pi\mu_\text{mg}}.
\end{align}

As we have already pointed out, the interesting feature of this model is that we can tune the energy and the stability (of the counted equilibria) independently. This leads to a richer landscape with respect to the \textit{pure} model, see Fig.~\ref{fig:complexities_mixed_pspin}. In the figure we show that in particular there is a whole range of energies where there is an exponential abundance of marginal minima with $\mu=\mu_\text{mg}$. In the pure case, the color plot of Fig.~\ref{fig:complexities_mixed_pspin} \textit{right} would collapse to the yellow line. This property has important consequences on the dynamics of this \textit{mixed} model, that differentiates it from its \textit{pure} counterpart, where $f(q)$ is a monomial. Indeed, it is important to point out that many of the features presented above are peculiar to the pure model alone. First of all in the pure model the threshold separates saddles from minima, while we see from Fig.~\ref{fig:complexities_mixed_pspin} that the mixed model has marginal and local minima also above threshold. The most striking difference is probably in the out-of-equilibrium dynamics. Indeed it was initially observed in \cite{folena2020rethinking}, and later confirmed in \cite{Folena_gd_simul_2021, folena_zamponi_weak_2023, tersenghi_seb_mixed_2025} that there is not a unique threshold level that attracts the dynamics. In particular, the pure model presents weak ergodicity breaking \cite{Cugliandolo_1995_weak}: if we prepare the system at any temperature $T>T_d$ and we run a gradient descent (i.e. we switch off the temperature) the system ages towards marginal minima at $\epsilon_{th}$, while forgetting any initial condition (i.e. $\lim_{t\to\infty}C(t_w+t,t_w)=0,\,\,\forall t_w$). For the mixed model it was verified numerically \cite{folena_zamponi_weak_2023, tersenghi_seb_mixed_2025} that $T_d$ plays no special role: the system reaches asymptotically marginal minima that don't correspond to the threshold ones, and the system maintains memory of the initial condition, for any temperature $T>T_{SF}$ ("SF" for state following). In particular, the system can go below $\epsilon_{th}$ (as defined in Fig.~\ref{fig:complexities_mixed_pspin}, that is, the energy where the complexity achieves its maximum at $\mu\to\mu_{mg}$). This phenomena take the name of \textit{strong ergodicity breaking} \cite{Folena_gd_simul_2021, tersenghi_seb_sk_2020}. An interesting open problem is to try and understand what is the relation between the initial temperature (at which the system is prepared) and the energy of marginal minima reached asymptotically by gradient descent. However, the results of \cite{kent2024arrangement} deepen even more the issues. Indeed, in \cite{kent2024arrangement} the distribution of pairs of stationary points at given energies $\epsilon_0,\epsilon_1$, Lagrange multipliers $\mu_0,\mu_1$ and overlap $q$ is studied. The key result is that the neighborhoods of marginal minima are different below, at, and above the threshold: only at $\epsilon_{th}$ the marginal minima are connected by sub-extensive barriers and are found arbitrarily close to each others. Instead, above and below threshold, marginal minima are far apart and separated by extensive barriers (there is a gap in the overlap). Hence, the picture of a marginal manifold \cite{Kurchan_phase_space_1996} of marginal minima with sub-extensive barriers is only true at $\epsilon_{th}$ \cite{kent2024arrangement}. But if gradient descent can converge to marginal minima both above and below threshold (depending on the initial condition), it means that the neighborhoods of these points are not so important \cite{kent2024arrangement}. Let us finally also mention that besides the dynamics and the energy landscape, mixed spherical models can have different replica symmetry-breaking schemes \cite{MixedModelNieuwenhuizen95, Leuzzi_mixed_2004, Leuzzi_mixed_2006} as well as temperature chaos, bifurcations and level crossing \cite{kraza_following_2012, Barrat_bifurcation_95, folena_these_2020}.

\section{Summary of contributions}
\label{sec:sumamry_contributions}
This thesis is organized into three subtopics: Chapters~\ref{chapter:non_reciprocal} and Chapter~\ref{chapter:scs} treat the problem of comparing dynamics and statistics of equilibria in models with non-reciprocal interactions; Chapter~\ref{chapter:energy_landscapes} has the scope of investigating the structure of the energy landscape of the pure spherical $p-$spin by means of two (statical) approaches: curvature driven paths between fixed points, and triplets of fixed points; Chapter~\ref{chapter:rmt_} is motivated by the problem of studying paths between local minima in Chapter~\ref{chapter:energy_landscapes}, and it aims at computing the overlaps between spiked, correlated GOE matrices (which are the Hessians of local minima, in the context of Chapter~\ref{chapter:energy_landscapes}). Below we give a summary of results. 

\subsection{Chapter 2}
\label{summary:non_reciprocal}
\noindent\textbf{The question.} With this Chapter we want to investigate whether the chaotic dynamics of simple high-dimensional toy models of random neural networks with non-reciprocal interactions can be explained from a static approach, namely from the complexity of equilibria. This is an important question because static calculations usually prove easier than solving the DMFT. Dynamical and static approaches are tightly connected for the gradient descent in the model discussed in Sec.~\ref{sec:p_spin_model}, i.e. the pure $p$-spin, but whether such connections exist when the system's interactions are taken to be asymmetric is an open and important problem \cite{CugliandoloNonrelax97, wainrib2013topological, Fyodorov_2016, fyodorov2016nonlinear, ros2023generalized}. This Chapter is based on \cite{us_non_reciprocal_2025}.\\

\noindent\textbf{The model.} We consider a class of Gaussian models, previously studied in Refs.~\cite{CugliandoloNonrelax97, Fyodorov_2016}, that describe the evolution of (large) $N$ units ${\bf x}\in\mathbb{R}^N$ according to the ODE $d{\bf x}/dt=-\lambda({\bf x}){\bf x}+{\bf f}({\bf x})$, with:
    \begin{align*}
        \mathbb{E}[{\bf f}({\bf x})]=\frac{J}{N}\sum_i x_i,\quad \text{Cov}[f_i({\bf x}),f_j({\bf y})]=\delta_{ij}\Phi_1\left(\frac{{\bf x}\cdot{\bf y}}{N}\right)+\frac{x_jy_i}{N}\Phi_2\left(\frac{{\bf x}\cdot{\bf y}}{N}\right).
    \end{align*}
    This class of models is general and can account for gradient and non-gradient dynamics depending on the choice of $\Phi_1,\Phi_2$, with purely non-gradient dynamics when $\Phi_2\equiv 0$, and gradient dynamics when $\Phi_2=\Phi_1'$. We consider two classes of models: spherical models, when $\lambda({\bf x})$ is a Lagrange multiplier enforcing the constraint ${\bf x}^2=N$, or confined models when $\lambda({\bf x})$ is a confining potential depending on ${\bf x}^2$.\\

\noindent\textbf{Calculation 1: the complexity.} We compute both the quenched Replica Symmetric and the annealed complexities of stationary points for any choice of $\Phi_1,\Phi_2,\lambda,J$. The complexity is computed at fixed order parameters: $m=\sum x_i/N,\,q=\sum x_i^2/N,\,\lambda$. For the spherical models one chooses $q=1$ and $\lambda$ controls the stability of the stationary points; for confined models instead $\lambda$ depends on $q$, which then controls the stability of stationary points. \\

\noindent\textbf{Calculation 2: the dynamics.} Following \cite{ChaosSompo88,crisanti_path_2018} we show that when $\Phi_2\equiv 0$, the TTI solution to the DMFT can be characterized exactly, leading to a system of equations in $m,\lambda\,(\text{or}\,\lambda(C_0)),C_\infty$, where $C_0$ is the asymptotic value of the autocorrelation at equal times (which is equal to 1 in the spherical case), and $C_\infty$ at infinite time difference.\\

\noindent\textbf{Result 1: dynamics and complexity for $\Phi_2\equiv 0$.} We focus on the confined model given by $\Phi_1(u)=2g^2u^2$, with $g$ the interaction strength, and $\lambda={\bf x}^2/N-\gamma$, with $\gamma\geq 0$. We give explicit formulas for the dynamical order parameters $m,C_0,C_\infty$ and we integrate the DMFT equations to show that indeed the TTI is the correct solution for large times. We find a rich dynamical phase diagram (see Fig.~\ref{fig:dmft_phase}) with three phases: FFP (ferromagnetic fixed point); FC (ferromagnetic chaos); PC (paramagnetic chaos). In each phase we examine the complexity of equilibria at the values of parameters $m$ and $q=C_0$ selected by the dynamics, to see if anything special happens to the complexity at these values. While the FC-to-FFP transition is correctly predicted by the Kac-Rice calculation, the transition from FC-to-PC is not. In particular, our results are not consistent with the interpretation that the dynamics lingers among the "least unstable" ferromagnetic fixed points (cf. Fig.~\ref{fig:rnn_dmft_phase_kac}). A closer comparison within the PC phase (with $m=J=0$) reveals that the system does not converge to the shell (i.e. to the value of $q$) in configuration space hosting the most abundant fixed points, nor to the one where the complexity vanishes. Nevertheless, we find some interesting connections: the complexity evaluated at the dynamical order parameter $C_0$ in the PC phase is invariant with respect to $g,\gamma$ (in particular, only two families of solutions satisfy this property, one of which is the dynamical one); in the FC phase, the dynamical order parameter $C_0$ is the one where the annealed and quenched complexities coincide. We compared the scalings of the maximal Lyapunov exponent and the complexity across the FC-to-FFP transition, contradicting a previously conjectured link between the two \cite{wainrib2013topological}. \\

\noindent\textbf{Result 2: dynamics and complexity for $\alpha>0$.} 
We further investigate the connection between dynamics and Kac-Rice for the spherical model with $\Phi_1(u)=2g^2u^2$ and $\Phi_2=\alpha\,\Phi_1'$ with $\alpha>0$ \footnote{i.e. we introduce \textit{some} reciprocity among interactions}. Numerical integration of the DMFT equations reveals that the system is TTI (and seemingly chaotic) and does not display aging for any $\alpha<1$, see Fig.~\ref{fig:rnn_two_alphas_dmft}. We are moreover able to conjecture a simple linear behavior for the asymptotic value $\lambda_\infty(\alpha)$\footnote{with $\lambda_\infty(1)$ corresponding to the threshold energy of the $p$-spin model discussed in Sec.~\ref{sec:p_spin_model}.}, see Fig.~\ref{fig:rnn_kac_spm}. Additionally we show that for $1\geq \alpha>\alpha_c$ there are exponentially many stable equilibria (a similar situation as in \cite{fyodorov2016nonlinear}), but that the dynamics lies in a region of unstable ones, for any $\alpha<1$.
\\

\noindent\textbf{Result 3: confined mixed models.} Finally, we consider a class of confined mixed models \cite{fournier2023statistical}, with $J=0$, $\lambda(q)=\mu+q$ (with $\mu$ a tunable parameter and $q={\bf x}^2/N$), and $\Phi_1(u)=2g^2u^2+g^2u$, $\Phi_2\equiv 0$. These models have a transition to chaos at $g=\mu$, from the stable fixed point ${\bf x}={\bf 0}$ to a regime of paramagnetic chaos. We show that the complexity of unstable equilibria can have a very different behavior depending on $\mu$ (see Fig.~\ref{fig:confined_mixed_mus}): for $\mu\leq2/3$ the transition to chaos corresponds to an exponential increase of unstable equilibria (or, equivalently, to a topology trivialization when we transition from chaos to the stable fixed point); instead for $\mu> 2/3$ the transition to chaos still corresponds to an increase in complexity close to $q=0$, but it is not accompanied by a topology trivialization transition. This means that in the non-chaotic phase, there is a resilience gap \cite{Fyodorov_resilient_2021} that develops: the complexity is still positive above a certain value of $q>0$, and the fixed point at ${\bf x}={\bf 0}$ is isolated. Then, for $g$ low enough, the complexity eventually vanishes at a value $q>0$, and the only fixed point that survives is the origin (see Fig.~\ref{fig:confined_mixed_mus}~\textit{left}).

\subsection{Chapter 3}
\label{summary:scs}
\noindent\textbf{The question.}
In this Chapter we continue on the same venue of Chapter 2, but for the model of random neural network introduced by Sompolinsky, Crisanti, Sommers in \cite{ChaosSompo88}, often referred to as SCS model. In particular, we wanted to tackle the problem of computing its complexity, which is a long-standing challenge \cite{wainrib2013topological, HeliasFP2022}. Recent work has achieved a similar task in the context of Generalized Lotka-Volterra equations \cite{ros2023generalized, RosEcoQuenched2023}. This Chapter is based on \cite{pacco_scs_2025}, in preparation.\\

\noindent\textbf{The model.} The model consists of $N$ interacting units encoded in a vector ${\bf x}$ such that $d{\bf x}/dt=-{\bf x}+g\sum_j\phi(x_j)J_{ij}$ where $J$ has Gaussian entries with $\mathbb{E}[J_{ij}]=J_0/N$ and $\text{Cov}[J_{ij},J_{kl}]=(\delta_{ik}\delta_{jl}+\alpha\,\delta_{il}\delta_{jk})/N$. The non-linearity is usually taken to be $\phi=\tanh$ in the literature. Several works have studied its dynamics, both with and without external inputs \cite{crisanti_path_2018, MastrogiuseppeLink2018, AnnibaleDynamics2024, Sompo_transition_2015}. In order to facilitate the calculation of the complexity, in our work we consider $\phi(x)=(\text{sign}(x)-x)H(|x|-1)+x$, with $H$ the Heaviside step function.\\

\noindent\textbf{Calculation 1: the complexity.} With our choice of $\phi$ we overcome the challenges faced when using $\tanh$, allowing us to obtain the quenched complexity for $\alpha=0$ and the annealed complexity for any $\alpha\in[0,1]$. In particular, the random matrix problem associated to the Jacobian of fixed points drastically simplifies. Indeed, for our choice, when $N\to\infty$ the Jacobian spectrum converges to a uniform distribution supported within an ellipse in the complex plane \cite{sommers1988spectrum}. Thanks to this choice, we can calculate the complexity (both annealed and quenched RS) by introducing several order parameters, among which a parameter $D_\phi=\sum_{x_i:|x_i|\leq 1}1/N$ that counts the number of units that are smaller than 1. This parameter controls the extensive fraction of unstable modes of the Jacobian: for $D_\phi>1/g^2$ the fixed points are unstable, and for $D_\phi<1/g^2$ they are stable. \\

\noindent\textbf{Calculation 2: the dynamics.} When $J_0=0=\alpha$, it is known that the system is chaotic for $g>1$, and stable for $g<1$, and that the TTI solution of the DMFT is correct for long times \cite{ChaosSompo88}. In this Chapter, we repeat the analysis of the TTI solution at long times for our choice of $\phi$. We show that, depending on the values of $g,J_0$, the system can have PFP (paramagnetic fixed point), PC, FC, FFP phases, and that the phase diagram is qualitatively the same as for $\phi=\tanh$, see Figs.~\ref{fig:dmft_scs_phase}, \ref{fig:dmft_scs_phase_sign}. \\

\noindent\textbf{Results: dynamics and complexity for $\alpha=0$.} We make thorough comparisons between the annealed complexity and the dynamical order parameters of the TTI solution. As for Chapter 2, we show that, in general, the complexity cannot be used to predict the dynamics, see e.g. Fig.~\ref{fig:complexity_cav_dmft_3}. For this reason, the interpretation that the system is "surfing" between the most numerous fixed points does not hold. The PC-to-PFP transition can be seen as a topology trivialization in the number of unstable fixed points, see Fig.~\ref{fig:para_kac_scs}. However, at variance with Chapter 2, we find no particular invariance of the complexity in the shell chosen by the dynamics. The FC-to-FFP transition is correctly predicted by Kac-Rice, happening at a value $J_0$ strong enough to make two stable ferromagnetic fixed points appear, see Fig.~\ref{fig:complexity_cav_dmft_3}. Similarly to Chapter 2, we show that in the FC phase there is only one point where annealed and quenched complexities coincide, and that point is the only one relevant for a comparison with the dynamics. However, we also show that the dynamics does not dwell in the shell of fixed points corresponding to that point, see Figs.~\ref{fig:two_plots_cav_dmft_scs}, \ref{fig:ferro_scs_comparison_n1}.

\subsection{Chapter 4}
\label{summary:energy_lands}
In this Chapter we consider the pure spherical $p$-spin model, introduced in Sec.~\ref{sec:p_spin_model}. Consistently with Sec.~\ref{sec:p_spin_model}, we use ${\bm\sigma}\in\mathbb{R}^N$ to refer to configurations ${\bm\sigma}^2=1$, and we use ${\bf s}=\sqrt{N}{\bm\sigma}$ when they are not rescaled. By $\epsilon$ we denote the energy density $\epsilon=\mathcal{E}/N$ where $\mathcal{E}$ is the Hamiltonian, defined in Eq.~\eqref{eq:chap1_def_tensor_S}. The goal is to use a static approach to advance our understanding of the typical geometric structure of the energy landscape, building on previous works \cite{Crisanti_TAP_pspin_95, cavagna1998stationary, cavagna1997investigation, cavagna1997structure, ros2019complex, ros2019complexity, ros2021dynamical}. With our static approach we want to better understand the distribution of stationary points (and their stability) below $\epsilon_{th}$, as well as characterizing the energetic barriers between them. We use two approaches, based on \cite{pacco_triplets_2025, pacco_quenched_triplets_2025, pacco2024curvature}.\\

\noindent \textbf{Question 1.} Motivated by works on mildly supercooled particles \cite{xu2010anharmonic,widmer2008irreversible,khomenko2021relationship}, the goal is to assess to which extent, in the pure spherical $3-$spin model, information from
the local Hessian matrices around stationary points can be leveraged to find good paths connecting them. Hence, we consider a class of interpolating paths between stationary points below $\epsilon_{th}$\footnote{the threshold energy of the $p-$spin found in Sec.~\ref{sec:chap1_topo_compl}} and we investigate when local information on the landscapes' curvature allows to find paths with an energy barrier lower than the one of the geodesic path. \\

\noindent\textbf{Calculation 1: curvature-driven pathways.} The starting point is a deep local minimum of energy density $\epsilon_0\in(\epsilon_{gs},\epsilon_{th})$, the final one is a stationary point extracted conditionally to the first one. According to \cite{ros2019complexity} the second one can be either a local minimum or a rank-1 saddle, depending on its energy density and its overlap with the first one. We consider paths of the form:
\begin{align*}
    {\bm\sigma}[\gamma;f]=\gamma{\bm \sigma}_1 + \beta[\gamma; f]\; {\bm \sigma}_0 + f(\gamma){\bf v},\quad \gamma\in[0,1],
\end{align*}
where ${\bf v}$ is a perturbation to the geodesic path (i.e., when $f\equiv 0)$; $f(\gamma)$ is an injective function that controls the strength of the perturbation and $f(0)=f(1)=0$; $\beta$ (not the inverse temperature here) fixes the spherical constraint $({\bm\sigma}[\gamma;f])^2=1$ with boundary conditions ${\bm\sigma}[0;f]={\bm\sigma}_0$ and ${\bm\sigma}[1;f]={\bm\sigma}_1$. We compute the energy density profile along ${\bm\sigma}[\gamma;f]$ for $p=3$~\footnote{in which case the calculation is feasible, both annealed and quenched. In Ref.~\cite{pacco2024curvature} we verified that the two protocols give the same answer.}. We use two ways to leverage local information on the landscape curvature: either ${\bf v}\to{\bf v}^{a}_{\text{soft}}$ which encodes for the eigenvector associated to the softest mode at the starting stationary point ($a=0$) or at the ending one ($a=1$); or ${\bf v}\to{\bf v}_\text{Hess}$ which encodes for the orthogonal component of the gradient along the path. To compute these energy profiles, we had to account for the overlaps between the eigenvectors of the two Hessian matrices (at the initial and final points), and this is done in more generality in the next chapter~\ref{summary:rmt}.\\

\noindent\textbf{Result 1: energy barriers.}
We found that geodesic pathways always lie above threshold, see Fig.~\ref{fig:geodesic_barriers}. Then we compared the energetic barrier along the geodesic with the one found using the perturbed paths specified above. Surprisingly, the direction of softest local curvature at ${\bm\sigma}_0$ is never a reliable predictor of low-energy paths, except in the case in which the direction of softest curvature corresponds to an isolated mode of the Hessian at ${\bm\sigma}_1$, see Fig.~\ref{fig:density_plots}. However, other information encoded in the local
Hessian through ${\bf v}_\text{Hess}$ does allow the identification of pathways associated with lower energy barriers, see Fig.~\ref{fig:energy_barrier}. We observe that all optimized paths lie above threshold, except in the case that the arrival point is a rank-1 saddle, in which case both ${\bf v}^1_\text{soft}$ and ${\bf v}_\text{Hess}$ allow to find below threshold pathways. This analysis suggests that typical deep minima (which are distant in configuration space) are separated by above threshold barriers, a result already pointed out in \cite{rizzo2021path}.  \\

\noindent\textbf{Question 2. }One of the most striking observations made in recent years, both in experiments~\cite{durin2024earthquakelike, korchinski_thermal_2025} and in numerical simulations of elastic interfaces and interacting particles~\cite{Rosso_review_2021, scalliet2022thirty,tahaei2023scaling, de2024dynamical}, are {\em thermal avalanches}, i.e. the occurrence of a cascade of smaller activations following a slow activated nucleation. The question is then: can we find signatures of these phenomena within mean-field models ? To do this we imagine a series of activated jumps between three fixed points, and we look whether the transition rate associated to the second jump is larger than for the first one. If it is larger, then the dynamics is strongly correlated to the first jump. We translate this intuition into a static computation, by considering a series of three fixed points ${\bf s}_0,{\bf s}_1,{\bf s}_2$ extracted conditionally one after the other. Then the question becomes: is the landscape around ${\bf s}_1$ strongly influenced by the presence of ${\bf s}_0$? And more precisely, how does the landscape differ from a \textit{typical} point ${\bf s}_1$ extracted uniformly without conditioning ? In the present case, the landscape is probed by ${\bf s}_2$.\\

\noindent\textbf{Calculation 2: three-point complexity.} We compute the three-point complexity, defined as the complexity of triplets of stationary points ${\bf s}_0,\,{\bf s}_1,\,{\bf s}_2$ at energies $\epsilon_0,\epsilon_1,\epsilon_2$ below $\epsilon_{th}$ and overlaps $q={\bf s}_0\cdot{\bf s}_1/N$, $q_0={\bf s}_0\cdot{\bf s}_2/N$, $q_1={\bf s}_1\cdot{\bf s}_2/N$ \footnote{see Fig.~\ref{fig:landscape_three} for an easier representation}. In formula:
\begin{align}
\Sigma^{(3)}(\epsilon_2, q_0, q_1 | \epsilon_1,\epsilon_0, q)=\lim_{N\to\infty} \frac{1}{N}\mathbb{E}\left[\log\mathcal{N}_{{\bf s}_0 \,{\bf s}_1}(\epsilon_2, q_0, q_1 | \epsilon_1,\epsilon_0, q)\right]_{0,1},
\end{align}
where $\mathcal{N}_{{\bf s}_0,{\bf s}_1}$ is the number of such stationary points ${\bf s}_2$, and $\mathbb{E}[\cdot]_{0,1}$ indicates the order of the averages: first we extract ${\bf s}_0$, then ${\bf s}_1$ and then we average over the randomness; see around Eq.~\eqref{eq:enne} for precise definitions. This averaging over the randomness and over ${\bf s}_0$,${\bf s}_1$ has to be taken with care, in order to obtain typical (and not rare) distributions of the triplets. The replica method, applied to each of ${\bf s}_0,{\bf s}_1,{\bf s}_2$ can be used to handle such problem. In Sec.~\ref{sec:DoublyAnnealed} we explain that along some special lines annealed and quenched complexities coincide, but that in general they don't (although the difference is small) \cite{pacco_quenched_triplets_2025}.  \\

\noindent\textbf{Result 2: clustering of fixed points}
We use the distribution of ${\bf s}_2$ as a probe of the structure of the landscape in the vicinity of ${\bf s}_0$, ${\bf s}_1$. 
We identify two distinct scenarios: when the energy densities are small enough, the maximal three-point complexity reduces to the two-point complexity \cite{ros2019complexity}. Hence, the majority of the configurations ${\bf s}_2$ are arranged in the landscape in a way that is independent of ${\bf s}_0$. Instead, when the energy densities $\epsilon_1, \epsilon_2$ are large enough compared to $\epsilon_0$, the landscape displays different regimes when tuning the overlap $q$ between ${\bf s}_0$ and ${\bf s}_1$: (i) for small overlap $q$ (when ${\bf s}_0$ and  ${\bf s}_1$ are well separated in configuration space), the presence of ${\bf s}_0$ generates a \emph{depletion} of the stationary points ${\bf s}_2$ around ${\bf s}_1$; (ii) for intermediate $q$ we observe an anomalous \emph{accumulation} of the stationary points ${\bf s}_2$ near ${\bf s}_0$; (iii) for large $q$, the points ${\bf s}_2$ \emph{cluster} close to ${\bf s}_1$. For these larger values of $q$, close to a deep minimum, higher-energy stationary points are more densely packed than they are in typical regions of configuration space at that energy, far away from the deep minimum. We see this as a precursor of thermal avalanches, in the sense that a high energetic fixed point has many nearby fixed points at the same energy and large overlaps. A jump among these strongly correlated points is denoted as \textit{avalanche-like}. This is a signature of strong local correlations in the energy landscape. Fig.~\ref{Fig:2Ddensityplots} is a good summary of these landscape transitions. \\

\noindent\textbf{Result 3: dependence on the energies.}
For $p=3$ we analyzed, as a function of $\epsilon_0,\epsilon_1,\epsilon_2$, when there exist values of $q,q_0,q_1$ such that clustering occurs. We found that a particularly important point is $(\epsilon^*(\epsilon_0),q^*(\epsilon_0))$ from the two-point complexity, see the cusp in Fig.~\ref{fig:2D_plot_clustering}. This is the point that corresponds to the smallest energy where rank-1 saddles start to appear around ${\bf s}_0$, and $q^*(\epsilon_0)$ is their overlap with ${\bf s}_0$.  It turns out that we must have $\epsilon_{th}>\epsilon_1\geq\epsilon^*(\epsilon_0)$ and $\epsilon^*(\epsilon_0)\leq\epsilon_2\leq\epsilon^*(\epsilon_1)$ for clustering to occur. Instead for all other values of energies, we have that the landscape is easily interpreted in terms of the two-point complexity, and in that sense it is "memoryless", meaning that the landscape around ${\bf s}_1$ is not strongly influenced by the presence of ${\bf s}_0$, and in particular that, at fixed $q_1$, the maximum of the three-point complexity is achieved for $q_0=q\,q_1$, and it coincides with the two-point complexity. Hence, the dependence on $\epsilon_0$ is lost at the maximum.\\

\noindent\textbf{Result 4: equal energies.}
We found no signatures of clustering at equal energy densities $\epsilon_2=\epsilon_1=\epsilon_0=\epsilon$. This is a consequence of the fact that $\epsilon^*(\epsilon_0)>\epsilon_0$ always. Now, imagine a sequence of fixed points ${\bf s}_0,{\bf s}_1,{\bf s}_2$ at equal energies $\epsilon_{gs}<\epsilon<\epsilon_{th}$. If we denote by $q$ the overlap at which typically we find the closest local minimum ${\bf s}_1$ to ${\bf s}_0$, then we have that the closest local minimum ${\bf s}_2$ to ${\bf s}_1$ is also at overlap $q$, and $q_0=q^2$. We verified that these simple relations hold true even in the case of four consecutive points, and we conjecture that it should be true for any length. More specifically, if we consider the sequence ${\bf s}_0,\ldots,{\bf s}_{n-1}$ and we take ${\bf s}_i$ to be the closest minimum that is typically found near ${\bf s}_{i-1}$, then ${\bf s}_i\cdot{\bf s}_{i-1}/N=q$, and ${\bf s}_a\cdot{\bf s}_b/N=q^{|b-a|}$. Since $0<q<1$, this leads to loss of memory, meaning that we progressively forget of the previous configurations, the next closest one being determined solely by the two-point complexity, and the distance from the previous ones being multiplied by $q$.\\

\noindent\textbf{Result 5: hints on activated dynamics.}
With the results above we can speculate on the implications for the activated dynamics of the $p$-spin. It is reasonable to assume that the optimal energy barrier between minima grows with the distance between them in configuration space, since the larger is the distance, the larger is the amount of local rearrangements to connect them. This assumption is also supported by our results on curvature-driven pathways. Within this assumption, we consider an effective dynamics where the system jumps among closest minima at a given (possibly, the same) energy. 
Using our landscape's results we show
that jumps between local minima at \emph{small enough} energy density are "memoryless": they are characterized by a typical energy barrier that does not depend on configurations previously visited by the system. Based on this, it seems that thermal avalanches do not play a role in the activated dynamics of the $p$-spin at long times, at least when the system visits low energy configurations. However, it may play a role when escaping the minium itself. Indeed we find a precursor of the thermal avalanches: subsequent jumps \emph{from a deep minimum to high-energy minima} are not memoryless but display correlations, and large energy barriers are systematically followed by smaller ones. This clustered zone \textit{may} correspond to the \textit{hub} of \cite{folena_rare_2025}. Such abundance of high energetic saddles may help the system to reach above threshold energies, and decorrelate, providing many directions to descent the landscape into a new deep minimum. \\

\subsection{Chapter 5}
\label{summary:rmt}
\noindent\textbf{The question.}
This Chapter is motivated by Chapter~\ref{chapter:energy_landscapes}, in particular by the problem of determining the energy profile along perturbed geodesic pathways between local minima of the $p$-spin model. As we saw above, to lower the barriers we leverage local information of the Hessians at the two local minima. The two Hessians are correlated GOE matrices with a perturbed column and row (also denoted as spiked, correlated GOE matrices). Such perturbation can generate outliers in the spectrum of the GOE matrices. In order to compute our interpolating paths, we need to know the overlap (i.e. the squared dot product) between eigenvectors associated to any eigenvalues (in particular the isolated ones) of the two matrices, when $N\to\infty$. This Chapter, based on \cite{paccoros}, is about computing such overlaps, extending previous works \cite{bun2016rotational, bun2018overlaps}.\\

\noindent\textbf{The model.}
We consider pairs of random matrices ${\bf M}^{(a)}$, $a\in\{0,1\}$:
\begin{align}\label{eq:summary_matrix_form}
    \mathbf{M}^{(a)}=
    \begin{pmatrix}
     & && & m^{a}_{1\,N}\\
     &  &{\bf B}^{(a)}  & & \vdots\\
      & & & & m^{a}_{N-1\,N}\\
     m^{a}_{1\,N} && \ldots & m^{a}_{N-1\,N} &m_{N\,N}^{a}
    \end{pmatrix}
\end{align}
\noindent where all entries are Gaussians and ${\bf B}^{(a)}$ have zero mean elements with:
\begin{equation}
\mathbb{E}[B_{ij}^{a} \, B_{kl}^b]= \tonde{\delta_{a b} \frac{\sigma^2}{N}+ (1-\delta_{ab})\frac{\sigma^2_H}{N}}(\delta_{ik} \delta_{jl}+ \delta_{il} \delta_{jk}).
\end{equation}
Similarly, the entries $m^a_{i N}$ for $i<N$ have zero mean and correlations given by:
\begin{equation}
    \mathbb{E}[m_{iN}^a \, m_{kN}^b]= \tonde{\delta_{a b} \frac{\Delta^2_a}{N}+ (1-\delta_{ab})\frac{\Delta^2_h}{N}}\delta_{ik},
\end{equation}
and finally
\begin{align}
\mathbb{E}[m_{NN}^a]= \mu_a,\quad\quad\text{Cov}[m_{NN}^a, m_{NN}^b]= \tonde{\delta_{a b} \frac{v^2_a}{N}+ (1-\delta_{ab})\frac{v^2_h}{N}}.
\end{align}
The matrices ${\bf B}^{(a)}$ are decomposed as ${\bf B}^{(a)}={\bf H}+{\bf W}^{(a)}$ where ${\bf H}$ is a GOE matrix in common, and ${\bf W}^{(a)}$ are i.i.d. GOE matrices.\\

\noindent\textbf{Calculation 1: multiresolvent products.} In order to proceed with the computation of the overlaps, we need first to calculate the expected product of arbitrary powers of the resolvents:
\begin{align}
    {\bf \Pi}_{k,m}(z,\xi):=\mathbb{E}[{\bf G}_0^{k+1}(z){\bf G}_1^{m+1}(\xi)],\quad {\bf G}_a(z):=(z-{\bf H}-{\bf W}^{(a)}).
\end{align}
The result is Eq.~\eqref{eq:final_Pi_km}.\\

\noindent\textbf{Calculation 2: eigenvalues and eigenvector overlaps.} We compute the spectral properties of matrices of the type in \eqref{eq:summary_matrix_form}. The eigenvalue density (i.e. the bulk) is the Wigner's semicircle, with two isolated eigenvalues that may exist in the limit $N\to\infty$ depending on $\sigma,\mu_a,\Delta_a$; their expression is Eq.~\eqref{app:solutions}. If $\lambda^a,\,{\bf u}_{\lambda^a}$ denote an eigenvalue/eigenvector of ${\bf M}^{(a)}$, then the overlap is:
\begin{equation}
\Phi(\lambda^{0},\lambda^{1}):=N\,\mathbb{E}[\left( \mathbf{u}_{\lambda^{0}}\cdot \mathbf{u}_{\lambda^{1}}\right)^2],
\end{equation}
which remains of $\mathcal{O}(1)$ in the $N\to\infty$ limit when at least one of the two eigenvalues is in the bulk. If instead both of them are isolated, the correct self-averaging quantity is without $N$ in front. The computation of the overlaps is obtained from perturbative expansions of the function $\psi$ defined in Eq.~\eqref{eq:ExpProRe}.\\

\noindent\textbf{Results: eigenvector overlaps.} We found the overlaps for the following combinations of eigenvalues belonging to the two matrices: bulk-bulk, see Eq.~\eqref{eq:phi_bulk_bulk}; isolated-bulk, see Eq.~\eqref{eq:iso_iso_overlap}; isolated-isolated, see Eq.~\eqref{eq:phi_iso_bulk}. The bulk-bulk overlap was already found in \cite{bun2018overlaps}, but all other results are new. Each formula is verified with several numerical simulations.

\chapter{Non-reciprocal interactions: a class of solvable models}
\label{chapter:non_reciprocal}
In this chapter we consider high-dimensional random
non-gradient autonomous ODE’s, as toy models of random neural networks. We show that these models have a rich phase diagram, and that we can solve exactly the TTI (time translation invariant) solution of the DMFT (dynamical mean-field theory) and the quenched complexity. We compare the typical properties of the configurations belonging to the chaotic attractor, and compare with the statistics of fixed points of the dynamics found via Kac-Rice. This work expands on previous literature \cite{Fyodorov_2016, crisanti1987dynamics, ChaosSompo88, CugliandoloNonrelax97, AnnibaleDynamics2024, berthier2000two, wainrib2013topological, ben2021counting, Fyodorov_resilient_2021, fournier2023statistical}, showing that this model is an ideal playground to make explicit comparisons between the dynamical and statical properties of the system, thus providing quantitative answers to long-standing questions.  \\

\noindent\textit{Road-map}\\
\noindent In Sec.~\ref{sec:rnn_introduction} we introduce the problem. In Sec.~\ref{sec:rnn_general_model} we introduce the family of models and in Sec.~\ref{sec:rnn_topo_general} we show the general expressions of quenched and annealed complexities. In Sec.~\ref{sec:rnn_dynamics} we study a confined model with asymmetric couplings: we compare the exact solution of the DMFT with the complexity found via Kac-Rice. In Sec.~\ref{sec:rnn_alpha>0} we present some results on the case with partially asymmetric couplings for a spherical model. In Sec.~\ref{sec:confined_mixed} we study a confined mixed model, challenging the idea that transition to chaos is associated with a topology trivialization in the complexity. In Sec.~\ref{sec:rnn_future} we point out future directions and unsolved issues.\\

\noindent \textit{Aknowledgements}.\\
This is joint work with Pierfrancesco Urbani, Samantha J. Fournier and Valentina Ros, see Ref.~\cite{us_non_reciprocal_2025}. A special thanks to Samantha and Pierfrancesco, for solving and teaching me the DMFT, and whose code I used as a base to reproduce some of the plots presented here.

\section{Introduction}
\label{sec:rnn_introduction}
In Sec.~\ref{subsec:kac_rice_non_grad} we have already reviewed the importance of systems with many units interacting through random and \emph{non-reciprocal} interactions, and how the Kac-Rice formalism can be useful to count and classify the equilibria. In particular, we have seen that the minimal (that is, simplest) models of complex interacting systems can be written as 
\begin{align*}
    \frac{d{\bf x}}{dt}={\bf F}({\bf x}),\quad {\bf x}\in\mathbb{R}^N,
\end{align*}
where ${\bf F}$ is a high-dimensional random Gaussian field. We could then distinguish among two cases: if ${\bf F}$ comes from the gradient of an energy function, we can talk of an \textit{energy landscape}; if, however, ${\bf F}$ is "non-gradient", this notion of a landscape is missing. Such non-gradient systems are often referred to as having "asymmetric" or "non-reciprocal" interactions, in light of the fact that, when written as units interacting via a random tensor, the interactions are not symmetric under a permutation of the tensor indexes. As is usually done for these systems \cite{CugliandoloNonrelax97, Fyodorov_2016, ben2021counting, ros2019complex}, there is a control parameter, here denoted by $\alpha$ and called the "factor of asymmetry", such that $\alpha=1$ corresponds to the gradient case while $\alpha=0$ to the fully asymmetric case. These systems have often been used as models of random neural networks since the '80s~\cite{crisanti1987dynamics, hertz1986memory, ChaosSompo88, gardner1989phase, Schuster_suppression_92}, where it was already understood how non-reciprocity can lead to chaos~\cite{ChaosSompo88, crisanti_path_2018}.
Such models have gained importance both in the theoretical neuroscience communities~\cite{rajan2010stimulus, marti2018correlations, MastrogiuseppeLink2018, AnnibaleDynamics2024, HeliasMemory18, aguirre2022satisfiability, pereira2023forgetting} and, more recently, in machine learning~\cite{sussillo2009generating, fournier2023statistical, fournier2024generative}, since high-dimensional chaotic systems can be trained, serving as generative models. Systems with non-reciprocal interactions have also gained recent attention in the study of complex interacting ecosystems \cite{ros2023generalized, garnier2024unlearnable, fyodorov2016nonlinear, castedo2024generalised, arnoulx2024many, mahadevan2024continual, RosEcoQuenched2023, rogers2022chaos} and in the physics community in the context of spin models \cite{lorenzana_non_rec_spin_aging_2024, bertin_far_eq_2024, bertin_collective_2024, bertin_discontinuous_oscil_2024, bertin_hidden_2024, vitelli_non_rec_ising_2025, vitelli_when_is_nonrecip_2025}. \\

\noindent In this Chapter we are particularly interested in the "dynamics" versus "statics" debate. Indeed, prototypical models of randomly (and non-reciprocally) interacting neurons show a transition to chaos driven by the strength of random interactions: when the strength is weak the system quickly converges to a unique stable equilibrium; when interactions are strong enough, chaotic motion settles in \cite{ChaosSompo88}. This transition is often concomitant with an increase in complexity, intended as the explosion of the number of unstable equilibria and the loss of stability of the unique equilibrium \cite{wainrib2013topological, Fyodorov_2016,fyodorov2016nonlinear, HeliasFP2022, AnnibaleDynamics2024}. This multiplicity of equilibria naturally
raises the problem of their classification in terms of key properties, such as their linear stability \cite{ben2021counting, lacroix2022counting, Fyodorov_resilient_2021}. In conservative systems (i.e., when ${\bf F}$ is the gradient of an energy function) linking the behavior of gradient descent with the underlying energy landscape is easier. Indeed, in that case, gradient-based algorithms tend to converge to marginally stable equilibria~\cite{Castellani_2005, Crisanti_TAP_pspin_95, cavagna1998stationary, cugliandolo1993analytical} (although \textit{which ones} is still not clear in general \cite{folena2020rethinking, Folena_gd_simul_2021, folena_zamponi_weak_2023, tersenghi_seb_mixed_2025, kent2024arrangement}). However, in the case of non-conservative systems, the system's dynamics cannot be described in terms of optimization of an energy landscape, and linking the system's dynamics with the underlying geometry of phase space is more challenging. Since the first works on random neural networks \cite{ChaosSompo88}, much work has been devoted to understand the dynamics, as well as the distribution of fixed points, for systems with non-reciprocal interactions ~\cite{ CugliandoloNonrelax97, berthier2000two, Fyodorov_2016, berlemont2022glassy, fyodorov2016nonlinear, lacroix2022counting, ben2021counting, Fyodorov_resilient_2021, AnnibaleDynamics2024, ros2023generalized, MastrogiuseppeLink2018}. However, whether the long-time dynamical properties can be inferred from the complexity of equilibria is an open problem~\cite{wainrib2013topological}.\\

\noindent Here we want to tackle these questions quantitatively, by considering a prototypical class of models with non-reciprocal interactions that are sufficiently simple to compute both their complexity as well as the solution to the dynamics. A specific realization of these models was already studied in \cite{fournier2023statistical, fournier2024generative} for $\alpha=0$ in the context of recurrent neural networks for learning. Here, instead, we will compute the quenched and annealed complexities for the most general class of models, and we will solve explicitly the Dynamical Mean-Field Theory (DMFT) for a specific choice of the model, when $N\to\infty$ and $\alpha=0$. Thanks to our results, we can make explicit comparisons between the properties of the chaotic attractor and the statistics of equilibria. In particular, we show that the dynamical order parameters cannot be inferred from any special class of equilibria: neither from the most numerous, nor from the least unstable. Nonetheless, we show that \textit{some} connections do appear. A comparison between the scaling of the complexity and the maximal Lyapunov exponent is also done. We additionally discuss the case  $\alpha>0$, and we extend our results to mixed models \cite{MixedModelNieuwenhuizen95}, that is, when the force is non-homogeneous. With these mixed models, we question the relationship between \textit{transition to chaos} and \textit{topology trivialization} \cite{Fyodorov_2016, Fyodorov_resilient_2021, wainrib2013topological}. The interested reader can refer to Sec.~\ref{summary:non_reciprocal} for a summary of contributions.

\section{A family of models}
\label{sec:rnn_general_model}
\noindent Let us first introduce a very general family of models, which was first studied in ~\cite{CugliandoloNonrelax97, Fyodorov_2016}, and similar to those considered in \cite{crisanti1987dynamics, fyodorov2016nonlinear, ben2021counting, lacroix2022counting, Fyodorov_resilient_2021, Fyodorov_resilient_2021}. The dynamical equations take the general form:
\begin{align}\label{eqapp:mod}
    \partial_tx_i(t)=F_i({\bf x}):=-\lambda({\bf x}) \, x_i(t)+f_i({\bf x}),\quad {\bf x}\in\mathbb{R}^N
\end{align}
where $f_i({\bf x})$ is a Gaussian vector field of mean and covariance given by 
\begin{equation}
\label{app:eq:def_Covariance}
\mathbb{E}\left[f_i({\bf x})\right] = \frac{J }{N}\sum_{i=1}^N{x_i}  = :J\, m({\bf x}), \quad\text{Cov}\left[f_i({\bf x}), f_j({\bf y}) \right] =\delta_{ij} \Phi_1\tonde{\frac{{\bf x} \cdot {\bf y}}{N}}+ \frac{x_j \, y_i}{N} \Phi_2\tonde{\frac{{\bf x} \cdot {\bf y}}{N}}, 
\end{equation}
with $J \geq 0$ and $\Phi_1$ and $\Phi_2$ suitably chosen functions. We will derive the expressions of the complexity and of the dynamics by keeping these functions generic; we implicitly assume that the functions $\Phi_1$ and $\Phi_2$ satisfy all the conditions required for these expressions to be well-defined. This is of course true for the specific choices that we will discuss later.  In particular, we always assume that $\Phi_1(u),{\Phi}'_1(u)$ are positive for $u>0$. There are two variations of this model that can be considered:
\begin{itemize}
    \item {\bf Confined Models (CM). } In this case, ${\bf x} \in \mathbb{R}^N$ and $\lambda({\bf x})$ in \eqref{eqapp:mod} is a confining term that depends only on the norm $||{\bf x}||$:
    \begin{equation}
\lambda({\bf x})= \lambda(||{\bf x}||)
    \end{equation}
This choice defines a generalization of the “confined model" considered in \cite{fournier2023statistical}, which in turn is a simplification of standard recurrent neural network models \cite{ChaosSompo88}, obtained by replacing the neurons' firing rate function with a generic nonlinear function of the neuron's membrane potentials $x_i$.
    \item {\bf Spherical Models (SpM). } In this case, ${\bf x} \in \mathcal{S}_N(\sqrt N)= \left\{ {\bf x}: ||{\bf x}||^2=N \right\}$. To enforce the spherical constraint, $\lambda({\bf x})$ in \eqref{eqapp:mod} is chosen as a Lagrange multiplier, satisfying 
       \begin{equation}
\lambda({\bf x})= \frac{{\bf f} ({\bf x}) \cdot {\bf x}}{N}.
    \end{equation}
    With this choice, $\lambda({\bf x})$ is a Lagrange multiplier providing a strict confinement of the dynamics on the hypersphere in $N$ dimensions.
\end{itemize}
The dynamics described by this model is non-conservative for generic $\Phi_1(u), \Phi_2(u)$: the conservative limit corresponds to choosing $\Phi_2(u) =\Phi_1'(u)$. To control deviations from this limit, following \cite{Fyodorov_2016} we introduce the ratio $\alpha(u):={\Phi_2(u)}/{\Phi'_1(u)}$.
When the dynamics is conservative, meaning that the (random) force field is the derivative of a (random) energy field $\mathcal{E}({\bf x})$ with isotropic covariances, then $\alpha(u)=1$. Indeed, assume that  
\begin{align}
    f_k({\bf x})=-\frac{\partial \mathcal{E}({\bf x})}{\partial x_k}, \quad \quad  \text{Cov}\left[\mathcal{E}({\bf x}),\mathcal{E}({\bf y})\right]=Nh\tonde{\frac{{\bf x}\cdot{\bf y}}{N}},
\end{align}
then by differentiating with respect to $x_k$ and $y_l$, one gets
\begin{align}
    \text{Cov}\left[f_k({\bf x}),f_l({\bf y})\right]=\delta_{kl}h'\left(\frac{{\bf x}\cdot{\bf y}}{N}\right) + h''\left(\frac{{\bf x}\cdot{\bf y}}{N}\right)\frac{y_kx_l}{N}
\end{align}
from which we can immediately identify $\Phi_1(u)=h'(u)$ and $\Phi_2(u)=h''(u)=\Phi'_1(u)$ which clearly implies $\alpha(u)=1$. When $\alpha(u)$ is not constant and equal to unity, the dynamics is, in general, non-gradient. Although the complexity will be computed in the most general case, in our analyses we will focus on the sub-family of models defined by: 
\begin{equation}\label{eq:App:alp}
   \Phi_2(u)=\alpha\, \Phi'_1(u) \quad \quad \alpha \in [0,1]. 
\end{equation}
 In this case, the force ${\bf f}({\bf x})$ can be decomposed in a conservative contribution (the gradient of an energy landscape) and a non-conservative contribution:
 \begin{align}\label{eqapp:decoF}
f_i({\bf x})=\sqrt{1-\alpha}f^d_i({\bf x})-\sqrt{\alpha}\frac{\partial \mathcal{E}({\bf x})}{\partial x_i}+\frac{J}{N}\sum_ix_i,
\end{align}
where both $f_i^d({\bf x})$ and $\mathcal{E}({\bf x})$ are independent Gaussian fields with zero average and covariances satisfying:
\begin{align}\label{eq:CovRel}
 \text{Cov}[f_i^d({\bf x}),f_j^d({\bf y})]=\Phi_1\left(\frac{{\bf x}\cdot{\bf y}}{N}\right)\delta_{ij}, \quad \quad 
\text{Cov}[\mathcal{E}({\bf x}),\mathcal{E}({\bf y})]&= Nh\left(\frac{{\bf x}\cdot{\bf y}}{N}\right)
\end{align}
with $h$ defined from $h'(u)=\Phi_1(u)$. It is simple to check that in this case both \eqref{app:eq:def_Covariance} and \eqref{eq:App:alp}  hold true. As already mentioned, the parameter $\alpha$ measures the strength of the conservative part of the dynamics: for $\alpha=0$ we have pure non-conservative (non-gradient) dynamics, while for $\alpha=1$ we have conservative (gradient) dynamics. \\

\subsubsection{Previous work}
That non-conservative forces destroy glassy behavior, eventually leading to chaos, was already understood several years ago \cite{crisanti1987dynamics, ChaosSompo88, Parisi_asymm_1986, CugliandoloNonrelax97}. A few works have concentrated on computing the complexity of fixed points for these systems \cite{Fyodorov_2016, fyodorov2016nonlinear, ben2021counting, Fyodorov_resilient_2021, RosEcoQuenched2023, ros2023generalized}, but to the best of our knowledge no work has made explicit comparisons between the DMFT and the complexity, as we do here. Regarding the complexity calculations, our work is similar to Ref.~\cite{Fyodorov_2016}, where the annealed complexity of total fixed points for these models was computed. Let us explain the differences from that work. The first difference is in the external field. In~\cite{Fyodorov_2016} a random Gaussian field of variance $\sigma^2$ is added to the force. Because of this, as $\sigma$ is increased, the system undergoes a topology trivialization transition, from an exponential number of fixed points to just two. In our model instead we add a term $J\,m({\bf x})$ which is proportional to the magnetization $m({\bf x})$. Therefore, the topology trivialization looks a bit different: the exponentially abundant fixed points are there for any $J\geq 0$ (in particular, those characterized by a zero magnetization), but as $J$ is increased, a branch of fixed points with non-zero magnetization becomes "less unstable" and forms a small isolated "island" (see Fig.~\ref{fig:rnn_ferro_steps}), until reaching a critical value $J_c$ where this island vanishes, collapsing to only two ferromagnetic (stable) fixed points. Another difference is that we classify points by introducing the following order parameters: the instability index, the magnetization and the self-overlap, which are useful for a comparison with the DMFT. \\

\noindent For completeness let us report one of the results of Ref.~\cite{Fyodorov_2016} which is of main interest for us. The result consists in the asymptotic $N>>1$ behavior of the expectation of the total number of fixed points $\mathcal{N}_{tot}$ and the complexity $\Sigma_{tot}$ of \eqref{eqapp:mod} for the SpM:
\begin{align}
\label{eq:rnn_fyod_express}
    \mathbb{E}[\mathcal{N}_{tot}]\underset{N>>1}{\sim} 2\sqrt{\frac{\tau+1}{1-\tau}}\frac{b}{\sqrt{1-b^2}}e^{N\ln \frac{1}{b}},\quad\Sigma_{tot}=\lim_{N\to\infty}\frac{\log\mathbb{E}[\mathcal{N}_{tot}]}{N}=\ln\frac{1}{b}
\end{align}
with 
\begin{align}
&\tau=\Phi_2(1)/\Phi_1'(1),\quad\quad|\tau|<1\\
&b^2=\frac{\sigma^2+\Phi_1(1)}{\Phi_1'(1)},\quad b<1
\end{align}
and it is moreover assumed that $0<\Phi_1(1)\leq \Phi_1'(1)$ and $-\Phi_1(1)\leq \Phi_2(1)\leq \Phi_1'(1)$. For $b<1$ and $-1<\tau\leq 1$ it is found instead $\lim_{N\to\infty} \mathbb{E}[\mathcal{N}_{tot}]=2$. In this chapter instead we will only be interested in the exponential behavior (i.e. no prefactors) and we will account for the full complexity curve.\\

\noindent Regarding the DMFT side, our work is similar to \cite{CugliandoloNonrelax97}. However, in that work an external white noise is added to the dynamical equations, which we do not have. Moreover, we analyse a much broader class of models, and we are able to obtain the explicit solution to the DMFT (in the TTI regime), which we validate by numerical simulations, and was not previously found. Additionally, one of us was able to compute explicit expressions for the maximal Lyapunov exponent in the various dynamical phases (to appear in the new version of \cite{us_non_reciprocal_2025}).

\section{Topological complexity in the general case}
\label{sec:rnn_topo_general}
\noindent In this section we show the results for the quenched (Replica Symmetric) and annealed complexities for the most general choice of $\lambda,\Phi_1,\Phi_2, J$. As we saw in Sec.~\ref{sec:chap1_topo_compl} the topological complexity is the entropy of equilibria ${\bf x}^*$ satisfying $F_i({\bf x}^*)=0$ for $i=1, \cdots, N$. 
We denote with $\mathcal{N}(\lambda,m,q)$ the number of equilibria having fixed Lagrange multiplier $\lambda$ (or confining potential), magnetization and averaged squared norm (or self-overlap) $m,q$ with 
\begin{align}
m({\bf x}^*)=\frac{1}{N}\sum_i x_i^*,\quad q({\bf x}^*)=\frac{1}{N}\sum_i(x_i^*)^2
\end{align}
at fixed value of $J$. Let us abuse the notation, where we keep both $\lambda,q$ in $\mathcal{N}$, although in the case of CM we will choose $\lambda=\lambda(q)$, and in the case of SpM we set $q=1$. The complexity is defined as
\begin{align}\label{eq:Comp}
    \Sigma(\lambda,m,q):=\lim_{N\to \infty}\frac{\mathbb{E}[\log\mathcal{N}(\lambda,m,q)]}{N}.
\end{align}
As we already discussed in Sec.~\ref{sec:topo_complexity}, this quantity controls the asymptotics of the \emph{typical} value of $\mathcal{N}(\lambda,m,q)$; it is called quenched, as opposed to the annealed one:
\begin{align}
\label{eq:comp_ann_conf}
\Sigma_A(\lambda,m,q):=\lim_{N\to \infty} \frac{\log\mathbb{E}\mathcal{N}(\lambda,m,q)}{N}
\end{align}
which controls the asymptotics of the average number. In general, it holds $\Sigma_A(\lambda,m,q)\geq \Sigma(\lambda,m,q)$. The random variable counting the number of equilibria can be written in terms of the Kac-Rice formula as:
\begin{align}
    \mathcal{N}(\lambda,m,q)=\int_{\mathbb{R}^N} d{\bf x}\,\Delta({\bf x})|\det \mathcal{H}({\bf x})|\delta({\bf f}({\bf x})-\lambda\,{\bf x})
\end{align}
where 
\begin{align*}
&\Delta({\bf x}):=\delta\left(\sum_ix_i-Nm\right)\delta({\bf x}^2-Nq),\quad \quad 
[\mathcal{H}({\bf x})]_{ij}:=\frac{\partial F_i({\bf x})}{\partial x_j}.
\end{align*}

\noindent Let us apply here what we described in Sec.~\ref{sec:chap1_topo_compl}: the quenched complexity can be computed via the replicated Kac-Rice formalism \cite{ros2019complex, ros2022high}, which involves a combination of the replica trick 
\begin{align}
\log \mathcal{N}= \lim_{n \to 0} \frac{\mathcal{N}^n-1}{n}
\end{align}
with the Kac-Rice formula for the moments $\mathbb{E}[\mathcal{N}^n]$. In particular, to compute the quenched complexity one proceeds as follows (we neglect function arguments for brevity):
\begin{align}
\label{eq:replica_complexity}
\Sigma=\lim_{N\to\infty}\frac{\mathbb{E}\log\mathcal{N}}{N}=\lim_{N\to\infty,\,n \to 0}\frac{\mathbb{E}\mathcal{N}^n-1}{Nn}=\lim_{n\to0,N\to\infty}\frac{\log\mathbb{E}[\mathcal{N}^n]}{Nn}
\end{align}
where the second equality implies precisely the use of the replica trick. This trick consists in considering integer $n$ copies (replicas) of the system to determine $\mathcal{N}^n$, and then sending $n\to 0$ at the end of the calculation. There are several aspects of the replica method that are not rigorous and are worth pointing out. Although the second equality in \eqref{eq:replica_complexity} is correct, the third may not be, as we are interchanging the limits over $n$ and $N$. Then, we compute the expected value of the replicated number of equilibria by considering $n$ as an integer, finally assuming an analytic continuation from $n$ integer to $n\to 0$, which is a priori not allowed. Lastly, we are making an assumption on the distribution of the replicas.  Indeed, as usual in the replica formalism, the calculation of the quenched complexity is mapped into a variational problem for an $n \times n$ overlap matrix $\hat Q$, which encodes for the distribution of scalar products ${\bf x}^a \cdot {\bf x}^b$ extracted with uniform measure among the family of equilibria with given $m, q$. Precisely:
\begin{align}
    [\hat{Q}]_{ab}={\bf x}^a\cdot{\bf x}^b/N.
\end{align}
We solve this variational problem within the so called Replica Symmetric (RS) assumption, which corresponds to assuming that the overlaps ${\bf x}^a \cdot {\bf x}^b$ between different equilibria $a \neq b$ have a unique typical value independently of the choice of $a,b$, denoted with $\tilde Q$ below. This implies that the matrix $\hat Q$ is symmetric. To visualize this better, the RS ansatz consists in assuming that the matrix of overlaps between replicas takes the following form:
\begin{align}
    \hat{Q}=\begin{pmatrix}
        q & \tilde{Q} & \ldots & \ldots & \tilde{Q}\\
        \tilde{Q} & q & \tilde{Q} & \ldots & \tilde{Q}\\
        \vdots & \tilde{Q} & \ddots & \tilde{Q} & \vdots\\
        \vdots & \vdots & \tilde{Q} & \ddots  & \tilde{Q}\\
        \tilde{Q} & \tilde{Q} & \ldots & \tilde{Q} & q
    \end{pmatrix}\in\mathbb{R}^{n\times n},
\end{align}
where $|\tilde{Q}|\leq q$. Within the RS ansatz, we write 
\begin{align}\label{eq.compvar}
\begin{split}
&\Sigma(\lambda,m,q)=\text{extr}_{\tilde{Q}}\tilde{\Sigma}(\lambda,m,q,\tilde{Q})
\end{split}
\end{align}
where $\tilde{\Sigma}$ is the expression of the complexity at fixed overlap $\tilde{Q}$. With a unique formula we incorporate both choices of SpM ($q=1$) or CM ($\lambda=\lambda(q)$) models. Let us present right away the final result, as a reference; calculations are done in Appendix.~\ref{app:rnn_quenched} 

\subsubsection{Quenched expression}
The explicit expression of $\tilde{\Sigma}(\lambda,m,q,\tilde{Q})$ depends on the functions  $\Phi_1(u), \Phi_2(u)$ in  \eqref{app:eq:def_Covariance}, evaluated at either $q$ or $\tilde{Q}$. We introduce the following compact notation: 
\begin{align}\label{eq:Notation}
\begin{split}
&\Phi_i^q:=\Phi_i(q), \quad  \tilde \Phi_i:=\Phi_i(\tilde Q), \quad  \dot{\Phi}^q_i:= \Phi_i'(q)\quad   \text{for} \quad  i=1,2\\
&\alpha_q:= \frac{\Phi_2(q)}{\Phi_1'(q)}\quad \quad   \kappa(\lambda, q):=\frac{\lambda}{\sqrt{\Phi_1'(q)}}.
\end{split}
\end{align}
Then 
\begin{align}\label{eq:rnn_CompQgen}
\begin{split}
&\tilde{\Sigma}(m,\lambda,q,\tilde{Q})=\mathcal{V}(m, q, \tilde Q)+\mathcal{P}(\lambda, m,q,\tilde{Q})+\Theta (\kappa(\lambda, q), q)
\end{split}
\end{align}
where 
\begin{align}
\label{app:det_eqn}
    \Theta(\kappa, q)=\begin{cases}
    \frac{\log\dot{\Phi}_1^q}{2}+\frac{1}{2}\left(\frac{\kappa^2}{1+\alpha_q}-1\right)\quad\quad\text{if }\, |\kappa|\leq 1+\alpha_q\\
     \frac{\log\dot{\Phi}_1^q}{2}+\frac{\left(\kappa-\text{sign}(\kappa)\sqrt{\kappa^2-4\alpha_q}\right)^2}{8\alpha_q}+\log\left|\frac{\kappa+\text{sign}(\kappa)\sqrt{\kappa^2-4\alpha_q}}{2}\right|\quad\text{else},
    \end{cases}
\end{align} 
\begin{align}\label{eq:V_m_q_Qt_vol}
&\mathcal{V}(m, q, \tilde Q)= \frac{q-m^2 + (q - \tilde{Q}) \log(2 \pi) + (q - \tilde{Q}) \log\left(q - \tilde{Q}\right)}{2 (q - \tilde{Q})},
\end{align} 
and
\begin{align}\label{eq.Proba}
&\mathcal{P}(\lambda,m,q,\tilde{Q})=-\frac{1}{2}\left\{\log(2\pi)+\log(\Phi_1^q-\tilde{\Phi}_1)+\frac{\tilde{\Phi}_1}{\Phi_1^q-\tilde{\Phi}_1}+\lambda^2U_{11}-\lambda J mU_{12}+J^2m^2U_{22}\right\}
\end{align}
with
\begin{equation}
\begin{split}
&U_{11}=\frac{q \Phi_1^q - 2 q \tilde{\Phi}_1 + 
\tilde{Q} \tilde{\Phi}_1 + (q - \tilde{Q})^2 \tilde{\Phi}_2}{\mathcal{A}}\\
&U_{22}=\frac{1}{\Phi_1^q - \tilde{\Phi}_1}-\frac{m^2(\alpha_q\dot{\Phi}^q_1-\tilde{\Phi}_2)}{\mathcal{A}}\\
&U_{12}=2m\frac{\tilde{\Phi}_2(q-\tilde{Q})+\Phi_1^q-\tilde{\Phi}_1}{\mathcal{A}}\\
&\mathcal{A}=(\Phi_1^q)^2 + (\tilde{\Phi}_1)^2 - 2 q \tilde{\Phi}_1 \alpha_q\dot{\Phi}^q_1 + (q - \tilde{Q})^2 \alpha_q\dot{\Phi}^q_1 \tilde{\Phi}_2 + \tilde{Q} \tilde{\Phi}_1 (\alpha_q\dot{\Phi}^q_1 + \tilde{\Phi}_2)\\& + \Phi_1^q (-2 \tilde{\Phi}_1 - 2 \tilde{Q} \tilde{\Phi}_2 + q (\alpha_q\dot{\Phi}^q_1 + \tilde{\Phi}_2)).\\
\end{split}
\end{equation}
Let us say that $\mathcal{V}$ is the "phase space term", as it arises when calculating the volume satisfying the imposed constraints; $\mathcal{P}$ is the "probability term", as it encodes the probability that a certain point ${\bf x}$ satisfies the constraint to be an equilibrium point; $\Theta$ is the "determinant term", associated to the expected value of the determinant of the Jacobian. 

\subsubsection{Annealed expression}
The annealed complexity reads instead:
\begin{align}\label{eq: AnnCompGen}
\Sigma_A(m,\lambda,q)=\mathcal{P}_A(m,\lambda,q)+\mathcal{V}_A(m,q)+\Theta (\kappa(\lambda, q), q),
\end{align}
where now
\begin{equation}\label{eq.pA}
\begin{split}
&\mathcal{V}_A(m,q)=\frac{1}{2}+\frac{1}{2}\log(2\pi (q-m^2)),\\
&\mathcal{P}_A(\lambda,m,q)=-\frac{1}{2}\left[\log(2\pi)+\log(\Phi_1^q)+\lambda^2U_{11}^A-\lambda J mU_{12}^A+J^2m^2U_{22}^A \right]
\end{split}
\end{equation}
and 
\begin{equation}
\begin{split}
U_{11}^{A}=\frac{q}{\Phi_1^q + q \alpha_q\dot{\Phi}^q_1}, \quad \quad
U_{12}^A =\frac{2m}{\Phi_1^q + q \alpha_q\dot{\Phi}^q_1}, \quad \quad 
U_{22}^A=\frac{\Phi_1^q + (q-m^2) \alpha_q\dot{\Phi}^q_1}{\Phi_1^q (\Phi_1^q + 
q \alpha_q\dot{\Phi}^q_1)}.
\end{split}
\end{equation}
For $m=0$ and $q=1$, this reduces to 
\begin{align}\label{eq: AnnCompGenParaGFyod}
\Sigma_A(\lambda,m=0,q=1)=\frac{1}{2}-\frac{1}{2}\log(\Phi_1^1)-\frac{ \kappa^2}{2[\Phi_1^1 + \alpha_1\dot{\Phi}^1_1]}+\Theta (\kappa, 1), \quad \quad \kappa= \frac{\lambda}{\sqrt{\dot{\Phi}^1_1}}.
\end{align}
Let us stress that this is consistent with Eq. (3.3) in \cite{Fyodorov_2016} by setting $\sigma=0$, imposing a delta function on the Lagrange multiplier $\lambda$ and calculating the annealed complexity for $N\to\infty$. Moreover, for $|\alpha_{q=1}|<1$ and $\sqrt{\Phi_1^1/\dot{\Phi}_1^1}<1$ with $0<\Phi_1^1\leq \dot{\Phi}_1^1$ and $-\Phi_1^1\leq \Phi_2^1\leq \dot\Phi_1^1$, we recover the result in Eq.~\eqref{eq:rnn_fyod_express} by choosing $\lambda=0$, which hence corresponds to the value of most abundance of stationary points. This is quite intuitive: given the symmetry of the problem (for the SpM only, where $\lambda$ is free) the most abundant stationary points are those that have an equal fraction of positive and negative eigenvalues, corresponding to $\lambda=0$. 

\subsubsection{Linear stability of the equilibria}
The linear stability of the equilibria counted by the complexity \eqref{eq:rnn_CompQgen} and its annealed counterpart \eqref{eq: AnnCompGen} is controlled by $\lambda$. As we show in detail in Appendix~ \ref{app:determinant_calc}, the matrix controlling the linear stability of equilibria  (obtained linearizing the dynamical equations around each equilibrium configuration) is an asymmetric random matrix with Gaussian entries; for $N\to\infty$, the  eigenvalues of this matrix are uniformly distributed in a region of compact support on the complex plane. This support has the shape of an ellipse \cite{sommers1988spectrum}, centered on a point that depends on $\lambda$. In particular, for $N\to \infty$ the support of the spectrum of the Jacobian tends to an ellipse in the complex plane, with equation:
\begin{align}
\label{eq:rnn_support_ellipse}
    \frac{(x+\lambda)^2}{(1+\alpha_q)^2}+\frac{y^2}{(1+\alpha_q)^2}=\dot{\Phi}_1^q\,,
\end{align}
and the density of eigenvalues is uniform within this support. A simple computation of the determinant for these types of random matrices has been done in \cite{RosEcoQuenched2023}. \\

\noindent Then, an equilibrium is linearly stable if all the eigenvalues of the Jacobian matrix have negative real part, and unstable otherwise. Moreover, in this work we shall only consider the extensive instability, that is, the position of the edges of the bulk of the spectrum for $N\to\infty$. We will therefore not consider finite size corrections and/or isolated eigenvalues, which are left for future work. \\

\noindent From \eqref{eq:rnn_support_ellipse} we have that in the general case:
\begin{equation}
\begin{cases}
       &\lambda>\sqrt{\dot{\Phi}_1^q}(1+\alpha_q) \quad \longrightarrow \quad   \text{ linearly stable equilibria}\\
    &\lambda<\sqrt{\dot{\Phi}_1^q}(1+\alpha_q) \quad \longrightarrow \quad   \text{ linearly unstable equilibria}.
    \end{cases}
    \end{equation}
and when we have equality we speak of \textit{marginal equilibria}, since the support of the bulk of eigenvalues of the ellipse touches the origin. \\

\noindent For example, in the SpM with $q=1$ marginally stable equilibria correspond to:
 \begin{equation}
 \label{eq:lambda_ms_spm}
     \lambda_{\rm ms}:=\sqrt{\dot{\Phi}^1_1}(1+\alpha_1).
 \end{equation}
Notice that if we consider the SpM with $\Phi_1(u)\propto u^p$ and $\Phi_2=\Phi_1'$, we reduce to the spherical $p$-spin model of Chapter~\ref{chapter:intro}.

\section{Dynamics and complexity: the case $\alpha=0$.}
\label{sec:rnn_dynamics}
 Let us first specify a realization of the CM model that we will use for our analyses here. We will abusively refer to this specific choice of model as just "CM" for simplicity. 
\subsection{Choice of model}
\noindent To be quantitative and make contact with models of random neural networks where the interaction strength is usually denoted by $g$, we will use:
\begin{align}
\label{eq:rnn_choice_phi1}
\Phi_1(u)= 2 g^2 u^2,\quad h(u)= 2 g^2 u^3/3,\quad \lambda({\bf x})=\frac{||{\bf x}||^2}{N}- \gamma, \quad  \gamma \geq 0.
\end{align}
where, as we explain in Sec.~\ref{sec:rnn_general_model}, in this case $\Phi_2(u)=\alpha\Phi_1'(u)$, $\forall u$. This model provides us with a non-linear force whose strength is tuned by $g>0$.  For this choice of $\Phi_1$ the random terms appearing in \eqref{eqapp:decoF} can be parametrized as 
\begin{align}
f^d_i({\bf x}) = g \,  N^{-1} \sum_{jk}^N J_{i}^{jk} x_j x_k,\quad\quad
\mathcal{E}({\bf x}) = - 2 g \, N^{-1} \sum_{i<j<k}^N S_{ijk} x_i x_j x_k,
\end{align}
where $J_{i}^{jk}$ satisfies $J_{i}^{jk}= J_{i}^{kj}$ and 
$\mathbb{E}[J_{i}^{jk}J_a^{bc}]=\delta_{ia}(\delta_{jb}\delta_{kc}+\delta_{jc}\delta_{kb})$. Due to the lack of symmetry in the lower index, the force components $f^d_i({\bf x})$ can not be written as derivatives of a unique function: this term corresponds to the non-conservative part of the dynamics. The second tensor satisfies $\mathbb{E}[S_{ijk}^2]=1$. It corresponds to the conservative part of the force: a gradient term driving the system towards minima of the energy $\mathcal{E}({\bf x})$. The parameters $\alpha, g $  control the strength of the conservative contribution and of the randomness, while the term proportional to $J\geq 0$ favors configurations with a non-zero magnetization.\\

\noindent We insist that the first, deterministic term in~\eqref{eqapp:mod} provides a confinement to the dynamics, through the function $\lambda({\bf x})$. The model \eqref{eqapp:mod} exhibits a rich dynamical phase diagram, with single fixed points and chaotic phases, as we will see below. This model is particularly interesting, since it is rich enough to present a complex phenomenology, but easy enough to be solvable both in terms of dynamics (DMFT) and complexity (Kac-Rice). Thus, it represents the ideal playground to study non-gradient (eventually chaotic) dynamics in relation to the statistics of fixed points of its dynamical equations \eqref{eqapp:mod}. \\

The linear stability of the equilibria is controlled by the parameter $q$:
 \begin{equation}
\begin{cases}
       &q>\gamma+\sqrt{\dot{\Phi}_1^q}(1+\alpha) \quad \longrightarrow \quad   \text{ linearly stable equilibria CM}\\
    &q<\gamma +\sqrt{\dot{\Phi}_1^q}(1+\alpha) \quad \longrightarrow \quad   \text{ linearly unstable equilibria CM}.
    \end{cases}
    \end{equation}
  Marginally stable equilibria correspond to saturation of the inequality, which occurs at 
 \begin{equation}
    q_{\rm ms}:=\gamma +\sqrt{\dot{\Phi}_1^q}(1+\alpha).
 \end{equation}

In this section we study the case $\alpha=0$, where we will see that the DMFT is exactly solvable in the TTI (time translation invariant) regime. \\

\begin{figure}[t!]
\centering
\includegraphics[width=1
\textwidth]{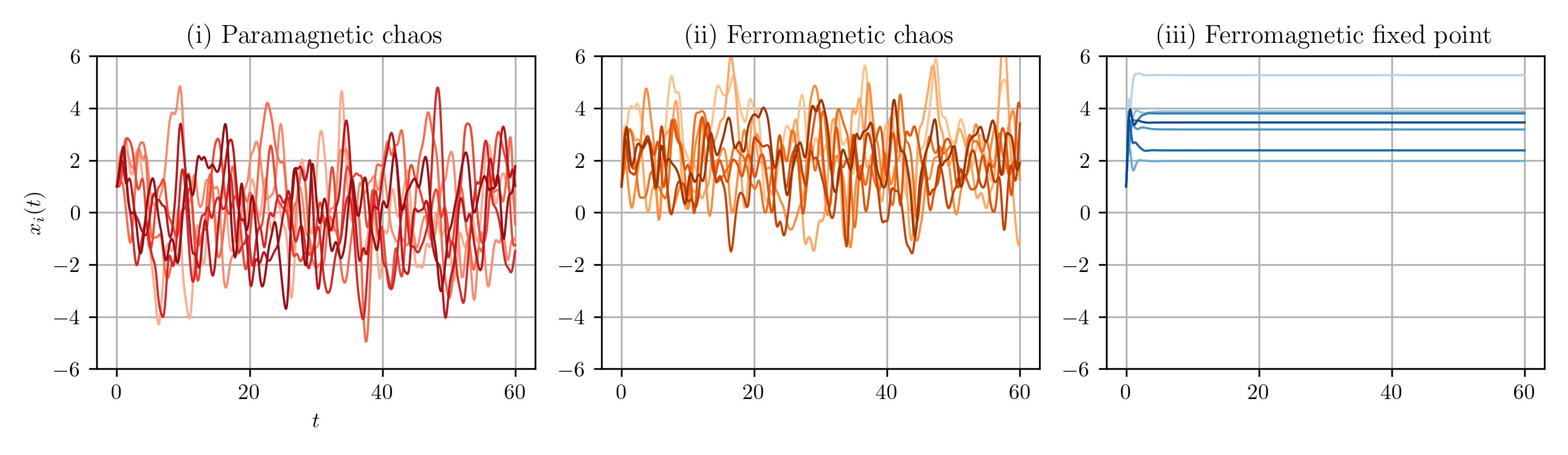}\\
\caption{Numerical integration of Eqs.~\eqref{eqapp:mod} using Euler discretization with $N=180$, $dt=0.1$ for the CM model with $\Phi_1(u)=2g^2u^2,\,\Phi_2(u)=0$ and $g=1$, $\gamma=0.5$. From left to right $J=1,\,3.1,\,6$. We track the first seven entries of the vector ${\bf x}$. The points $(i),(ii),(iii)$ are also referenced in the $(J,g)$ dynamical phase diagram in Fig.~\ref{fig:dmft_phase}.}
\label{fig:chaos_num}
\end{figure}

\noindent A numerical integration of Eqs.~\eqref{eqapp:mod} shows three qualitatively distinct regimes, shown in Fig.~\ref{fig:chaos_num}. At fixed $g, \gamma$, for large $J$ the variables $x_i(t)$ quickly converge to a ferromagnetic fixed point, and temporal fluctuations are suppressed, see Fig.~\ref{fig:chaos_num}.(iii). As $J$ decreases, the system transitions to a chaotic regime characterized by persistent oscillations of the $x_i(t)$, around values that are non-zero for intermediate $J$ (ferromagnetic chaos), and zero for smaller $J$ (paramagnetic chaos), see Fig.~\ref{fig:chaos_num}.(ii) and (i) respectively. These phases can be characterized quantitatively in the limit $N\to\infty$ by means of a DMFT approach, similar to what has been adopted in  Refs.~\cite{ChaosSompo88, crisanti_path_2018}, see also Sec.~\ref{sec:pspin_dynamics_calcs} and Appendix.~\ref{app:dynamical_calculations}. The first step of the DMFT calculation is to write a Path Integral, taking an average over paths and over the disorder, of Eqs.~\eqref{eqapp:mod}, with random initial conditions. In the limit $N\to \infty$, one obtains an effective stochastic differential equation of a single degree of freedom $x(t)$, which reads:
\begin{align}
\begin{split}
\label{eq:effective_dynamical_system}
 & \partial_t{x}(t) =  -\lambda(t) x(t) + J m(t) + \eta(t)\equiv F(t), \quad \quad  \lambda(t)=C(t,t)-\gamma.
\end{split}
\end{align}
where $\eta(t)$ is a zero-mean Gaussian process with variance $\langle \eta(t)\eta(t') \rangle = 2g^2C^2(t,t')$, and where $\langle\cdot\rangle$ denotes the average over this effective Gaussian random noise. Using this effective dynamical equation, we can obtain the self-consistent DMFT equations for 
\begin{align}
&m(t) := \langle x(t)\rangle,\quad C(t,t') := \langle x(t)x(t')\rangle,\quad R(t,t') := \langle \delta x(t)/\delta \eta(t') \rangle
\end{align}
through averages with respect to $\eta(t)$. In the Itô convention that we are using, we need to remember that $R(t,t')=0$ for $t\leq t'$. In Appendix~\ref{app:dynamical_calculations} we propose a short derivation of the effective dynamical equation for any Gaussian field of the type in Eq.~\eqref{eqapp:mod}. A long and detailed derivation is found in Chapters 7 and 10 of Ref.~\cite{HeliasBook20}.

\subsection{Derivation of the DMFT equations}
\label{sec:derivation_dmft_rnn}
Let us go through some details of the derivation of the DMFT equations here. The equation for $m$ is obtained by taking an average on both sides of \eqref{eq:effective_dynamical_system}, and it gives:
\begin{align}
    \partial_t m(t)=-\lambda(t) m(t)+Jm(t)
\end{align}
The equation for $C$ is obtained by multiplying Eq.~\eqref{eq:effective_dynamical_system} at time $t$ by $x(t')$, and then taking an average over $\eta$. This gives
\begin{align}
\partial_tC(t,t')=-\lambda(t)C(t,t')+Jm(t)m(t')+\langle \eta(t)x(t')\rangle
\end{align}
We see that we are left with computing the value $\langle\eta(t)x(t')\rangle$. The computation is found in Appendix~\ref{app:comp_eta_x_t_tp}, where it is shown that 
\begin{align*}
    \int_0^{t'}ds\,R(t',s)\Sigma(t,s)=\langle\eta(t)x(t')\rangle,\quad\quad\Sigma(t,t'):=\langle\eta(t)\eta(t')\rangle.
\end{align*}
By substituting $\Sigma(t,t')=2g^2C^2(t,t')$, we are left with the final equation for $C$:
\begin{align}
    \partial_tC(t,t')=-\lambda(t) C(t,t')+2g^2\int_0^{t'}ds\,R(t',s)C^2(t,s)+Jm(t)m(t').
\end{align}
The equation for $R$ is more immediate, we just have to derive with respect to $\eta(t')$ in Eq.~\eqref{eq:effective_dynamical_system}. It follows immediately that
\begin{align}
\partial_t R(t,t')=-\lambda(t)R(t,t')+\delta(t-t').
\end{align}
We have therefore found the three DMFT equations, which in summary read:
\begin{align}
\begin{split}
\label{eq:dmft_equations_alpha=0}
&\partial_t m(t)=-\lambda(t) m(t)+Jm(t)\\
&\partial_tC(t,t')=-\lambda(t)C(t,t')+Jm(t)m(t')+2g^2\int_0^{t'}ds\,R(t',s)C^2(t,s)\\
&\partial_t R(t,t')=-\lambda(t)R(t,t')+\delta(t-t').
\end{split}
\end{align}

\subsection{DMFT in the TTI regime.}
\label{sec:rnn_dmft_tti}
The first solution of the DMFT equations that we can try to look for is the Time Translationally Invariant (TTI) one. This is well motivated for $\alpha=0$, since we do not have conservative forces that slow down the dynamics of $x(t)$. We will show later that this ansatz is also justified numerically, by both integrating the DMFT numerically and comparing with real simulations of the dynamical system. The result of TTI is that the dynamics quickly evolves to a stationary regime where the dynamical order parameters become invariant under time translations. Hence, in this regime, we assume that our dynamics is described by:
\begin{align*}
&\lambda_\infty:=\lim_{t\to\infty} \lambda(t), \quad \quad \quad\quad\quad\quad\quad\,\,\,\,
  m_\infty:=\lim_{t\to\infty} m(t), \\
  &  
  C_{TTI}(\tau):=\lim_{t'\to\infty} C(t'+\tau,t'), \quad  \quad
  R_{TTI}(\tau):=\lim_{t'\to\infty} R(t'+\tau,t')
\end{align*}
where $\tau=t-t'$ is fixed. This means that for long times, $m,\lambda$ become time independent, and $C,R$ only depend on the time difference. Let us further introduce the two order parameters $C_0$ and $C_\infty$:
\begin{align*}
  C_0 := \lim_{\tau\to0} C_{TTI}(\tau) \quad \quad  \mathrm{and} \quad \quad  C_\infty := \lim_{\tau\to\infty} C_{TTI}(\tau).
\end{align*}
As we will see in the following, in the TTI regime, the attractor manifold of the dynamics at long times can be described as a function of $C_0, C_\infty, m_\infty$ (remark that for our choice of CM, $\lambda_\infty=C_0-\gamma$ by definition). Now, the DMFT equation for $R$ can be solved and reads:
\begin{align}
    R_{TTI}(\tau)=H(\tau)e^{-\lambda_\infty \tau},\quad H(\tau) \text{ the Heaviside step function}.
\end{align}
Notice that $\lim_{\tau\to 0^+}R_{TTI}(\tau)=1$ but $\lim_{\tau\to 0^-}R_{TTI}(\tau)=0$. The DMFT equations for $m,C$ then read:
\begin{align}
  \label{eq: equation for m infinity}
  &0 = (-\lambda_{\infty} + J) m_{\infty},\\
  \label{eq: equation for dC of tau}
  &\partial_{\tau} C_{TTI}(\tau) = -\lambda_{\infty} C_{TTI}(\tau) + 2g^2 \int_0^{\infty}ds\, C_{TTI}^2(\tau+s) e^{-\lambda_{\infty}s} + J m_{\infty}^2.
\end{align}
In the rest of this derivation we drop, for simplicity, the subscript $TTI$, as it is now clear which functions we are working with. Because $C(\tau)$ is TTI, we know that $C(\tau)=C(-\tau)$ which implies $\partial_{\tau} C(\tau)|_{\tau=0}=0$. 
While solving for the full function $C(\tau)$ is more complicated, we can find closed formulas for $C_0,C_\infty,m_\infty$, which give us the long-times autocorrelation at equal times, and at infinitely separated times, and the value of the magnetization. Therefore we need to obtain three equations to solve for these unknowns. The first of them is Eq.~\eqref{eq: equation for m infinity}. The second one can be obtained by considering the long time limit of Eq.~\eqref{eq: equation for dC of tau}. Indeed given that $\lim_{\tau\to \infty}\partial_\tau C(\tau)=0$ (which holds since $C(\tau)$ is bounded), we get:
\begin{equation}
  \label{eq: equation for C infinity}
  0 = -\lambda_{\infty} C_{\infty} + 2g^2 \frac{C^2_{\infty}}{\lambda_{\infty}} + J m_{\infty}^2\:.
\end{equation}
The last equation that we find below can be obtained following the same steps as in \cite{crisanti_path_2018, ChaosSompo88}. Multiplying Eq.~\eqref{eq:effective_dynamical_system} by itself and averaging over $\eta$ we get:
\begin{align}
    (\partial_t+\lambda(t))(\partial_{t'}+\lambda(t'))C(t,t')=2g^2C^2(t,t')+J^2m(t)m(t').
\end{align}
If now we take the TTI limit $t,t'\to\infty$ with $\tau=t-t'$, we obtain:
\begin{equation}
  (\lambda_{\infty}^2 - \partial^2_{\tau}) C(\tau) = 2g^2 C^2(\tau) + (J m_{\infty})^2\:.
\end{equation}
This equation can be conveniently re-written in the form of a particle moving under the influence of a potential $V(C)$, as:
\begin{equation}\label{Newton}
  \partial^2_{\tau} C(\tau) = - \frac{\partial V}{\partial C},
\end{equation}
where $V$ is simply found to be:
\begin{equation}
  V(C) := \frac23 g^2 C^3 -\frac12 \lambda_{\infty}^2 C^2 + \left( J m_{\infty}\right)^2 C.
\end{equation}
Then, consider the energy function $E$ defined as:
\begin{equation}
 E = \frac12 \left( \partial_{\tau} C(\tau) \right)^2 + V(C).
\end{equation}
Using \eqref{Newton}, it is not hard to see that $dE/d\tau=0$, so that this fictitious energy is conserved in time. Acceptable solutions $C(\tau)$ must be bounded $|C(\tau)| \leq C_0$ and must conserve $E$. In particular, we have seen that $\partial_{\tau} C(\tau)|_{\tau\to\infty}=0$ and $\partial_\tau C(\tau)|_{\tau=0}=0$, which implies, together with the conservation of $E$, that:
\begin{equation}
V(C_0) = V(C_\infty).
\end{equation}
This equation can be recast in the following form:
\begin{equation}
  \label{eq: equation for C0 or lambda infinity}
  (C_\infty-C_0)^2\left[ \frac43g^2C_\infty - \frac12\lambda_\infty^2 + \frac23g^2C_0 \right] = 0.
\end{equation}
Hence, finally, we have this set of closed equations that have to be solved:
\begin{equation}\label{final_eq}
\begin{cases}
    \begin{split}
        &\lambda_{\infty}=C_0-\gamma\\
        &0=(-\lambda_{\infty} + J) m_{\infty},\\
        &0 = -\lambda_{\infty} C_{\infty} + 2g^2 \frac{C^2_{\infty}}{\lambda_{\infty}} + J m_{\infty}^2,\\
        &0=(C_\infty-C_0)^2\left[ \frac43g^2C_\infty - \frac12\lambda_\infty^2 + \frac23g^2C_0 \right].
    \end{split}
    \end{cases}
\end{equation}
Notice that the last equation allows both for stable fixed point solutions with $C_0=C_\infty$, and for out-of-equilibrium solutions with $C_\infty<C_0$.

\subsection{Dynamical phase diagram}
The set of equations in \eqref{final_eq} can be solved exactly. As a result, we can obtain the dynamical phase diagram of the CM model with $\alpha =0$ (a similar derivation holds for the SpM model, see our work \cite{us_non_reciprocal_2025}). Below, we will derive the order parameters of each phase, and the transition line between phases. In the following we will often refer to paramgnetic or ferromagneitc solutions, indicating that $m_\infty=0$ or $m_\infty\neq 0$ respectively. 

\begin{figure*}[t!]
    \centering
    \includegraphics[width=0.68\textwidth]{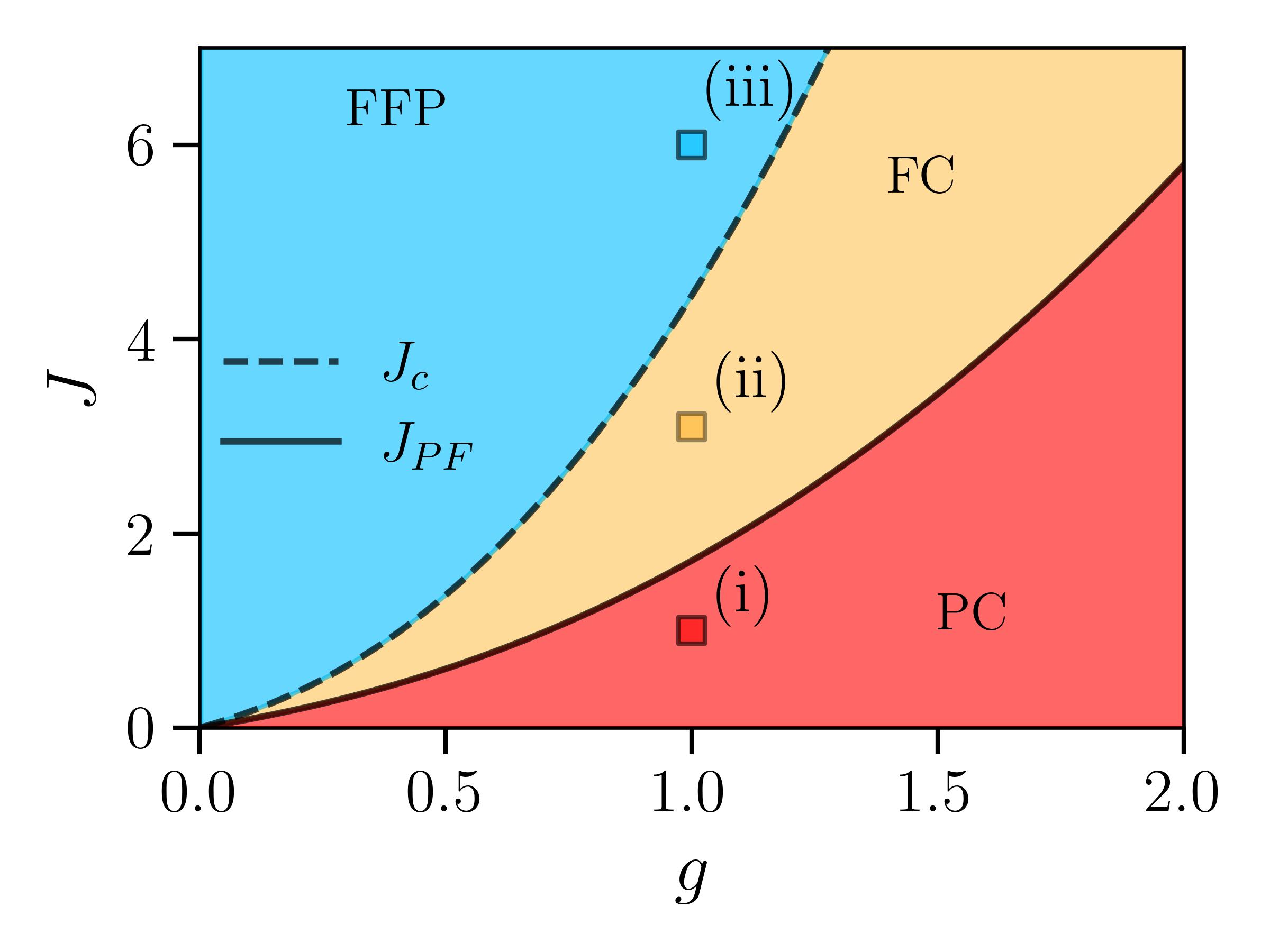}
    \caption{Dynamical phase diagram of the confined model with $\alpha=0$, $\Phi_1(u)=2g^2u^2$. At fixed $g$, increasing $J$ the system is in a (i) PC - Paramagnetic Chaotic  phase, (ii) FC - Ferromagnetic Chaotic phase, and (iii) FFP - Ferromagnetic Fixed Point phase. The colored squares in proximity to (i),(ii),(iii) represent the points at which the simulations of Fig.~\ref{fig:chaos_num} have been done (with a finite $N$, see the figure for details).}
    \label{fig:dmft_phase}
\end{figure*}

\subsubsection{Paramagnetic Fixed Point phase (PFP)}
If we choose the $C_0=C_\infty$ and $m_\infty=0$ solution, we find $C_0=C_\infty=0$, that is, the dynamics  converges to the fixed point ${\bf x}={\bf 0}$. We see from Eq.~\eqref{eqapp:mod} that this solution is only marginally stable. This means that the system will remain in this solution only if it is exactly initialized with it.  Another solution with $m_\infty=0$ can be found, but it is unstable. In Sec.~\ref{sec:confined_mixed} we propose a model with a stable PFP phase. 

\subsubsection{Ferromagnetic Fixed Point phase (FFP)}
If we choose $C_0=C_\infty$, and we look for ferromagnetic solutions, we find:
\begin{equation}\label{DMFT_fixFFCM}
C_{\infty}= C_0 = J+\gamma,\quad \quad m_{\infty} = \pm \sqrt{J+\gamma - 2g^2\frac{(J+\gamma)^2}{J^2}}.
\end{equation}
In this case the system converges to a ferromagnetic fixed point with non-zero magnetization. Since the dynamical system \eqref{eq:effective_dynamical_system} is invariant under change of sign (because $\eta$ is centered and Gaussian), the system has two (opposite) solutions for $m_\infty$.

\subsubsection{Paramagnetic Chaotic phase (PC)}
If instead we look for paramagnetic solutions such that $C_\infty<C_0$, by solving the equations in \eqref{final_eq}, we obtain the following order parameters:
\begin{equation}
\label{eq:dmft_para_orders}
m_{\infty}=0,\quad \quad C_{\infty} = 0,\quad \quad C_0 = \gamma + \frac23g\left(g+\sqrt{g^2+3\gamma}\right).
\end{equation} 
In this case, at long times, the system is found in an endogenously fluctuating stationary state with zero magnetization.

\subsubsection{Ferromagnetic Chaotic phase (FC)}
We can now look for solutions with $C_\infty<C_0$ and $m_\infty\neq 0$. The corresponding solution of the order parameters found from \eqref{final_eq} reads:
\begin{equation}
\label{eq:FC_phase}
    C_0 = J+\gamma, \quad  C_\infty=\frac12\left( \frac{3 J^2}{4 g^2} - J -\gamma \right), \quad  m_\infty = \pm \frac{1}{2}\sqrt{\frac{3J^2}{8 g^2} + J+\gamma - \frac{2g^2}{J^2}(J+\gamma)^2}.
\end{equation}
Like in the PC phase, in this case the dynamics reaches a stationary state with fluctuations that are persistent in time, in a region of phase space having a fixed magnetization.

\subsubsection{FFP to FC transition}
Such transition can be identified by analyzing the stability of the FFP solution, with $C(\tau)=C_0=C_\infty=J+\gamma$ for all $\tau$. In order to do this, we can look at the dynamical equation \eqref{Newton}. In particular, the solution $C_0=C_\infty$ will become unstable as soon as there is a change in convexity of the function $V(C)$ at $C=C_0$. This happens when
\begin{align}
\label{eq:Jc}
    V''(C)|_{C=C_0}=0\Rightarrow J_c=2g\left(g+\sqrt{\gamma+g^2}\right),
\end{align}
where the subscript $c$ stands for "critical".

\subsubsection{FC to PC transition}
A transition from ferromagnetic chaos to paramagnetic chaos appears as one lowers the value of $J$. Such transition is easily found by looking at the value of $J$ at which $m_\infty$ becomes equal to 0 from the FC phase. Thus, solving for $m_\infty=0$ in Eq.~\eqref{eq:FC_phase}, we immediately find that: 
\begin{align}
    J_{PF}=\frac{2}{3}g\left(g + \sqrt{3 \gamma + g^2}\right)
\end{align}
where PF stands for "Para-to-Ferro". We can further verify that at the same value of $J_{PF}$, also $C_\infty$ becomes equal to zero.

\subsubsection{Phase diagram and simulations}

\begin{figure*}[t!]
    \centering
    \includegraphics[width=0.7\textwidth]{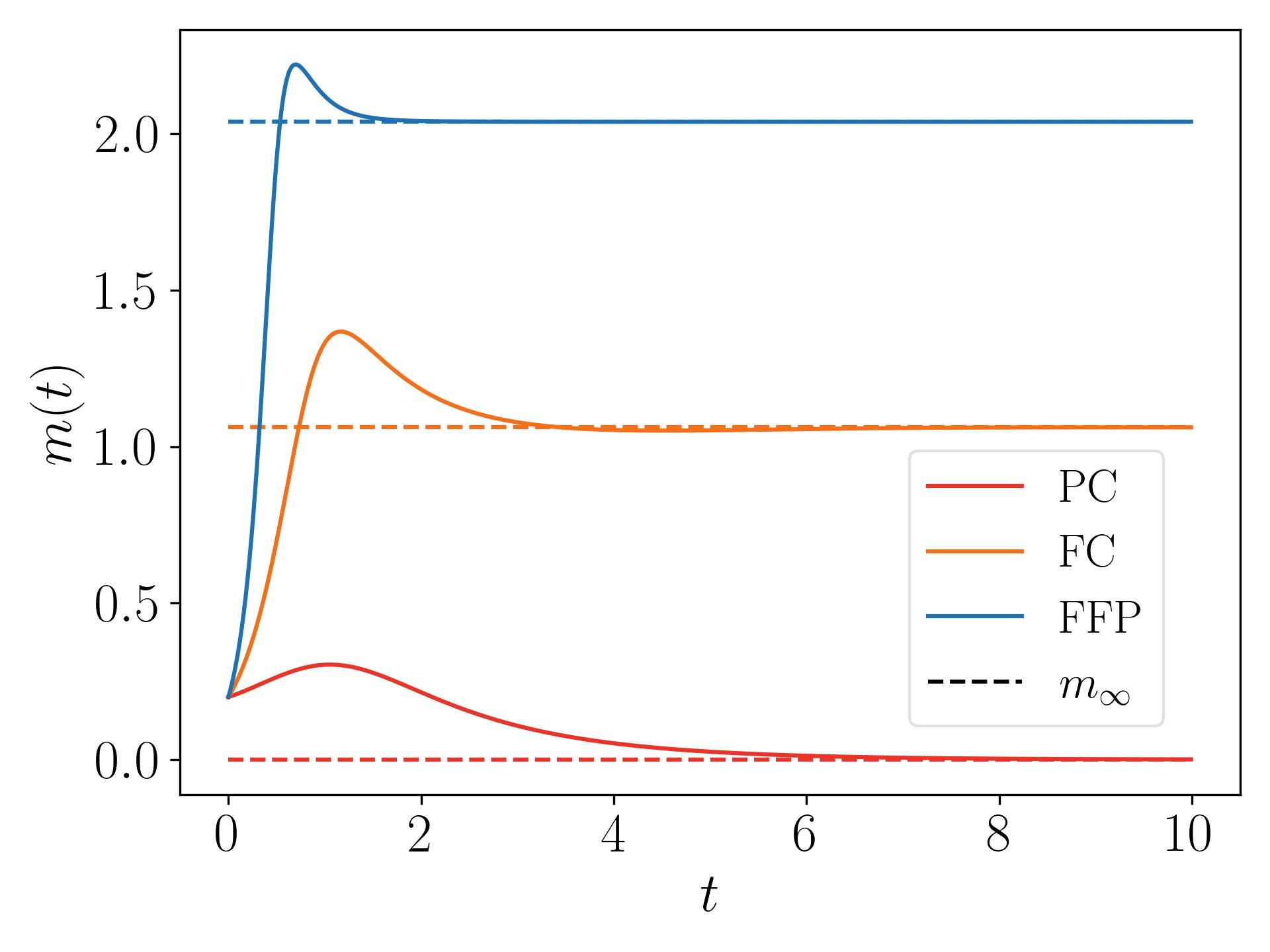}
    \caption{Numerical solution of the DMFT equations using an Euler discretization with $g=1$, $
\gamma=0.5$, $dt=0.01$, and $J=1, 3.1, 6$ in the PC,FC,FFP phases respectively (see legend). We take $C(0,0)= 1$ and $m(0)= 0.2$ and plot $m(t)$. We see that the dynamical value $m(t)$ quickly converges to the TTI value $m_\infty$ in all phases.}
    \label{fig:dmft_m_confined}
\end{figure*}

\noindent Some plots, figures, and comparisons with numerical simulations are now due. Let us start with a representation of the phase diagram in a $(g,J)$ plot, see Fig.~\ref{fig:dmft_phase}. From the figure, we can identify the three regimes: in \textit{blue}, for high values of $J$, the system converges to a ferromagnetic stable fixed point; as $J$ is lowered, we find a transition to chaos (see below for details), identified by the curve $J_c(g)$, at which the system converges at long times to a chaotic attractor manifold characterized by the order parameters in \eqref{eq:FC_phase}. The average magnetization remains non-zero, meaning that we are in a ferromagnetic chaotic regime (\textit{orange} regime). By lowering $J$ below the transition line $J_{PF}(g)$, the chaos becomes paramagnetic (\textit{red} regime), meaning that $m_\infty=0$. We remark that in the chaotic phases (\textit{orange} and \textit{red} in Fig.~\ref{fig:dmft_phase}) it holds $C_0 > C_\infty$, since the system decorrelates during the chaotic trajectories, either completely (paramagnet) or partially (ferromagnet). The fact that the motion is indeed chaotic in the orange and red zones of Fig.~\ref{fig:dmft_phase} can also be checked by computing the maximal Lyapunov exponent and was done by one of our collaborators \footnote{private communication with S.J. Fournier}(see also the upcoming version of \cite{us_non_reciprocal_2025}). \\

\noindent The numerical integration of the DMFT equations \eqref{eq:dmft_equations_alpha=0} in the case $\alpha=0$ is straightforward with an Euler discretization scheme, and it returns the magnetization $m(t)$ and the correlation function $C(t,t')$. A plot of the magnetization $m(t)$ is shown in Fig.~\ref{fig:dmft_m_confined}, in the three different regimes PC, FP, FFP. We can see that in all cases the value of $m(t)$ obtained via numerical integration quickly converges to the infinite time limit $m_\infty$ predicted within the TTI assumption. \\

 \noindent Similarly, in Fig.~\ref{fig:dmft_Cinf_confined}, we plot the autocorrelation function $C(t,t+\tau)$ for (different) values of $t$, while varying $\tau$. We see from the figure that for $t$ large enough all curves for different values of $t$ collapse onto each other. This is true in all the three different dynamical regimes, as seen in Fig.~\ref{fig:dmft_Cinf_confined}. The variations in $t$ are represented by variations in color, where darker colors correspond to larger values of $t$. In all cases, as $\tau$ grows, the curves reach asymptotically the value $C_\infty$ found within the TTI ansatz. We can therefore conclude that in this scenario ($\alpha=0$) the right solution to the DMFT equations is given by the TTI ansatz, given that the curves of the autocorrelation function only depend on the time difference $\tau$, and therefore do not display aging.

\begin{figure*}[t!]
    \centering
    \includegraphics[width=0.7\textwidth]{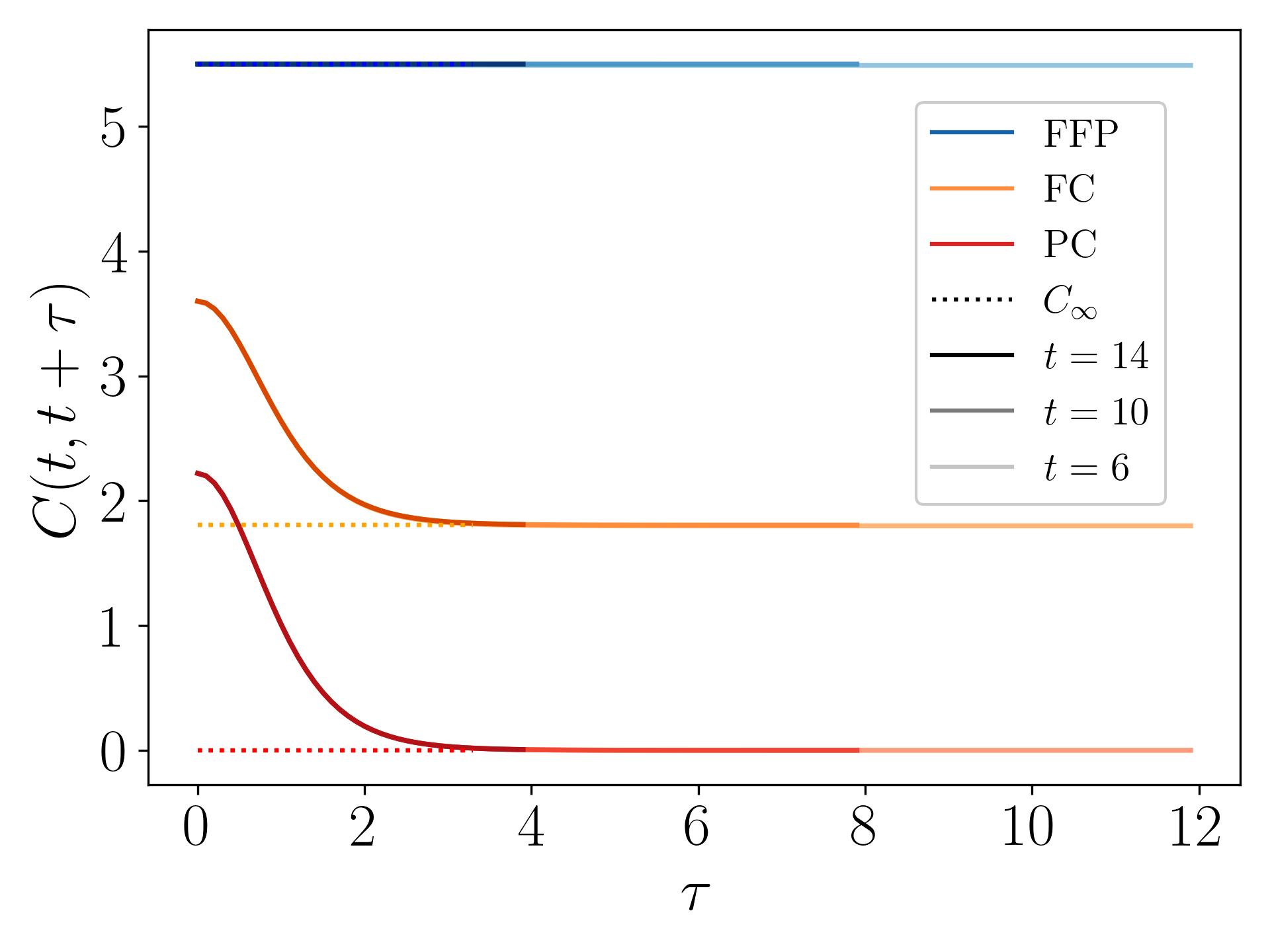}
    \caption{Numerical solution of the DMFT equations using an Euler discretization with $g=1$, $\gamma=0.5$, $dt=0.01$, and $J=1, 3.1, 5$ in the PC,FC,FFP phases respectively. We take $C(0,0)= 1$ and $m(0)= 0.2$ and plot $C(t,t+\tau)$ as a function of $\tau$ for three times: $t=14,10,6$ in each phase. Darker lines correspond to older times. 
    We see that the dynamics does not display any sign of aging, and for any $t$, $C(t,t+\tau)$ quickly converges to $C_\infty$ as $\tau$ increases (in each phase).}
    \label{fig:dmft_Cinf_confined}
\end{figure*}

\noindent An important check to do is to verify that by simulating the actual dynamical system \eqref{eqapp:mod} with $\alpha=0$ and $\Phi_1(u)=2g^2u^2$ and averaging over many realization of the disorder, we obtain the same curves as those obtained by integrating numerically the original DMFT equations for the order parameters. In order to do this we need to make sure that we initialize, both the DMFT equations and the real system, with the same initial condition. The results are presented in Fig.~\ref{fig:rnn_num_analytical_comparison}, where we use an initial condition of $m(0)=0.2$ and $C(0,0)=1$. This can be easily implemented in the DMFT, where $m$ and $C$ are functions to be integrated numerically. Instead, when simulating the dynamical system \eqref{eqapp:mod}, we can average over the disorder and, for each realization of the disorder, over paths starting at a vector picked randomly with the mean and self-overlap indicated above. A good matching shows that the DMFT equations are indeed correctly describing the real dynamics of the system for $N>>1$ units.

\begin{figure*}[t!]
    \centering
    \includegraphics[width=0.71\textwidth]{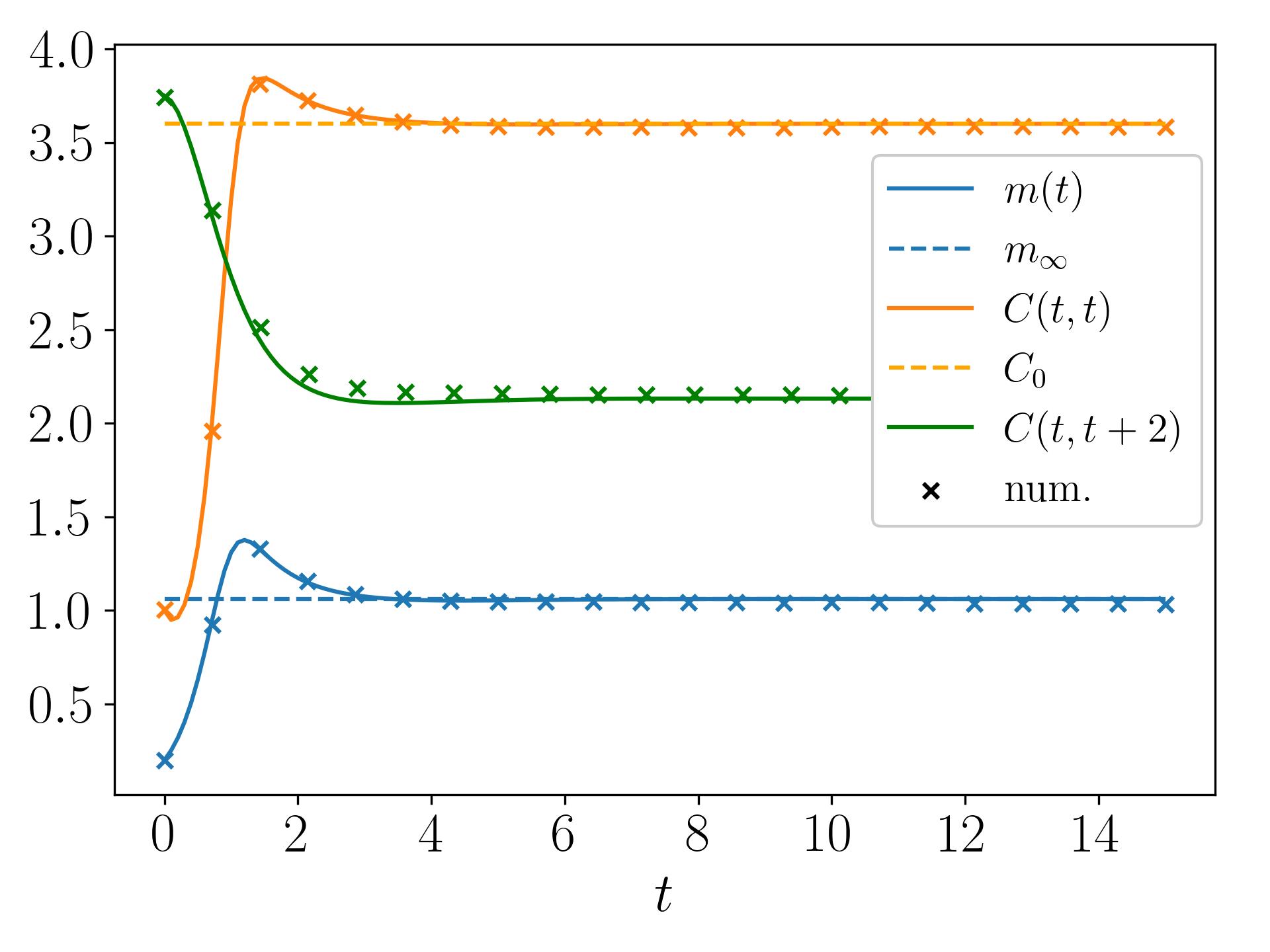}
    \caption{Comparison of the results found via numerical integration of the DMFT equations and a real simulation of $N$ interacting units, with $g=1$, $\gamma=0.5$, $dt=0.1$, $J=3.1$, $N=220$, $C(0,0)=1$, $m(0)=0.2$ and 20000 iterations over the disorder for the real simulation. We achieve good matching.}
    \label{fig:rnn_num_analytical_comparison}
\end{figure*}

\subsubsection{The Force}
It is interesting to note that with the dynamical formalism presented above we can directly compute the asymptotic average value of the force in Eq.~\eqref{eq:effective_dynamical_system}, i.e. $\lim_{t\to\infty} \langle F^2(t)\rangle$. We can show that, in general:
\begin{align}
    \lim_{t\to\infty}\langle F(t)^2\rangle=2g^2 C_0^2-\lambda_\infty^2C_0+(Jm_\infty)^2.
\end{align}
By replacing $\lambda_\infty=C_0-\gamma$ one can plot this quantity to see that it is non-zero in the chaotic phases (cf Fig. S-5 in \cite{us_non_reciprocal_2025}). This already points to the fact that a direct comparison between the DMFT and the Kac-Rice complexity of equilibria (where the force is zero) can be difficult. Nevertheless, given the simplicity of the model, we think it is worth exploring the comparison of dynamical and static approaches as a starting point.

\subsection{Complexity of the CM}
We now present the result of the complexity for the CM \footnote{In the RS ansatz. We have not checked the stability of this ansatz, which is left for future works.}, by applying the choice $\Phi_1(u)=2g^2u^2$, $\alpha=0$ and $\lambda=q-\gamma$ to the results of Sec.~\ref{sec:rnn_topo_general}. We further assume $q>0$. 

\subsubsection{Quenched complexity}
The quenched complexity is found from Eq.~\eqref{eq.compvar} by considering
\begin{align}\label{eq:cm_CompQgen}
\begin{split}
&\tilde{\Sigma}(m,q,\tilde{Q})=\mathcal{V}(m, q, \tilde Q)+\mathcal{P}(m,q,\tilde{Q})+\Theta (q),
\end{split}
\end{align}
where here we remove $\lambda$ as a function argument, and keep only $q$; we also keep the same notations as for the general case in Sec.~\ref{sec:rnn_topo_general}, with the expressions of $\mathcal{V},\mathcal{P},\Theta$ given below:
\begin{align}\label{eq.vol}
&\mathcal{V}(m, q, \tilde Q)= \frac{q-m^2 + (q - \tilde{Q}) \log(2 \pi) + (q - \tilde{Q}) \log\left(q - \tilde{Q}\right)}{2 (q - \tilde{Q})}.
\end{align} 
 and
\begin{align}
\begin{split}
\mathcal{P}(m,q,\tilde{Q})=&\frac{1}{2} \Bigg[
  -\frac{J^2 m^2}{2g^2\left(q^2-\tilde{Q}^2\right)}
  +\frac{J m^2\left(q-\gamma\right)}{g^2\left(q^2-\tilde{Q}^2\right)}
  -\frac{(\gamma-q)^2\left(q^2+q \tilde{Q}-\tilde{Q}^2\right)}{2g^2\left(q-\tilde{Q}\right)\left(q+\tilde{Q}\right)^2}\\&
  -\frac{\tilde{Q}^2}{q^2-\tilde{Q}^2}
  -\log\left(4\pi g^2\left(q^2-\tilde{Q}^2\right)\right)\Bigg]
\end{split}
\end{align}
and finally:
\begin{align}
    \Theta(q)= \begin{cases}
        \frac{1}{2} \left(-1 + \frac{(\gamma - q)^2}{4 g^2 q} + \log(4 g^2 q)\right)\quad\quad\text{if }\quad |q - \gamma|< 2 g \sqrt{q}\\
        \log(q-\gamma)\quad\quad \text{if }\quad|q-\gamma|\geq 2 g \sqrt{q}
    \end{cases}.
\end{align}

\subsubsection{Annealed complexity}
In the annealed case instead, we obtain:
\begin{align}
    \Sigma_A(m,q)=\mathcal{V}_A(m,q)+\mathcal{P}_A(m,q)+\Theta(q)
\end{align}
where $\Theta$ is the same as above, and 
\begin{equation}\label{eq:cm_pA}
\begin{split}
&\mathcal{V}_A(m,q)=\frac{1}{2}+\frac{1}{2}\log(2\pi (q-m^2)),\\
&\mathcal{P}_A(m,q)=-\frac{J \left(2\gamma + J\right)m^2 + \left(\gamma^2 - 2J m^2\right) q - 2\gamma q^2 + q^3}{4g^2 q^2} - \frac{1}{2} \log\left(4g^2 \pi q^2\right)
\end{split}
\end{equation}
Notice that, in general: the determinant terms of quenched and annealed computations coincide \cite{ros2019complex, RosEcoQuenched2023}; the probability term in the annealed case can be obtained by setting $\tilde{Q}\to 0$; however the volume terms differ and one cannot simply obtain the annealed one by sending $\tilde{Q}\to 0$. This happens because, while the probability term can be factored when the replicas are orthogonal, the phase space term counts the portion of phase space where all constraints (including the constraint on all replicas being orthogonal) are met. Therefore, in general, sending $\tilde{Q}\to0$ in the quenched computation will not factor the $n$ phase space terms associated to each replica.


\subsection{Analysis of the complexity and comparison with the dynamics}
Let us dive into the analysis of the quenched and annealed expressions presented above. As we analyze the complexity, we make comparisons with the results of the DMFT found above. In order to compare the two results, we need to compare the corresponding order parameters, according to the following associations:
\begin{align}
    C_0\longleftrightarrow q,\quad C_\infty\longleftrightarrow \tilde{Q},\quad m_\infty\longleftrightarrow m.
\end{align}

\noindent We start by looking for stable fixed points, and then we will consider the unstable fixed points. Plots of the complexity will be given as well.

\paragraph*{Isolated stable equilibrium and the FFP phase. }
Let us consider the formula for the annealed complexity $\Sigma_A$, and study the stable solutions to $\partial_m\Sigma_A(m,q)=0$ and $\partial_q\Sigma_A(m,q)=0$. One can verify that there is only one solution (that represents an isolated stable equilibrium), with order parameters:
\begin{equation}\label{eq:DMFT_fix_CM}
  q_{\text{ffp}}:=J+\gamma, \quad \quad \quad m_{\text{ffp}}:=\pm \sqrt{J+\gamma - 2g^2\frac{(J+\gamma)^2}{J^2}},
\end{equation}
satisfying $\Sigma_A(m_\text{ffp},q_\text{ffp})=0$. This point is stable provided that $q_{\text{ffp}}>\gamma+2g\sqrt{q_{\text{ffp}}}$, which implies that $J>J_c=2 g (g + \sqrt{\gamma + g^2})$: this is precisely the regime of parameters which corresponds to the FFP phase identified by the DMFT solution, see Eq.~\eqref{eq:Jc}. Moreover, the parameters in \eqref{eq:DMFT_fix_CM} coincide with the values predicted by the asymptotic solution of the DMFT for the CM, see Eq.~\eqref{DMFT_fixFFCM}. As we shall see below, in the region $J>J_c$ unstable equilibria also exist and have a positive complexity; however, only the stable isolated equilibrium matters for the long time dynamics. In the FFP phase there is therefore a direct correspondence between the result found with the DMFT and with the Kac-Rice: there exists only two stable isolated fixed points with zero complexity, and they attract the dynamics at long times, see Fig.~\ref{fig:dmft_Cinf_confined} and Fig.~\ref{fig:kac_conf_ffp}.   \\

\begin{figure*}[t!]
    \centering
    \includegraphics[width=0.75\textwidth]{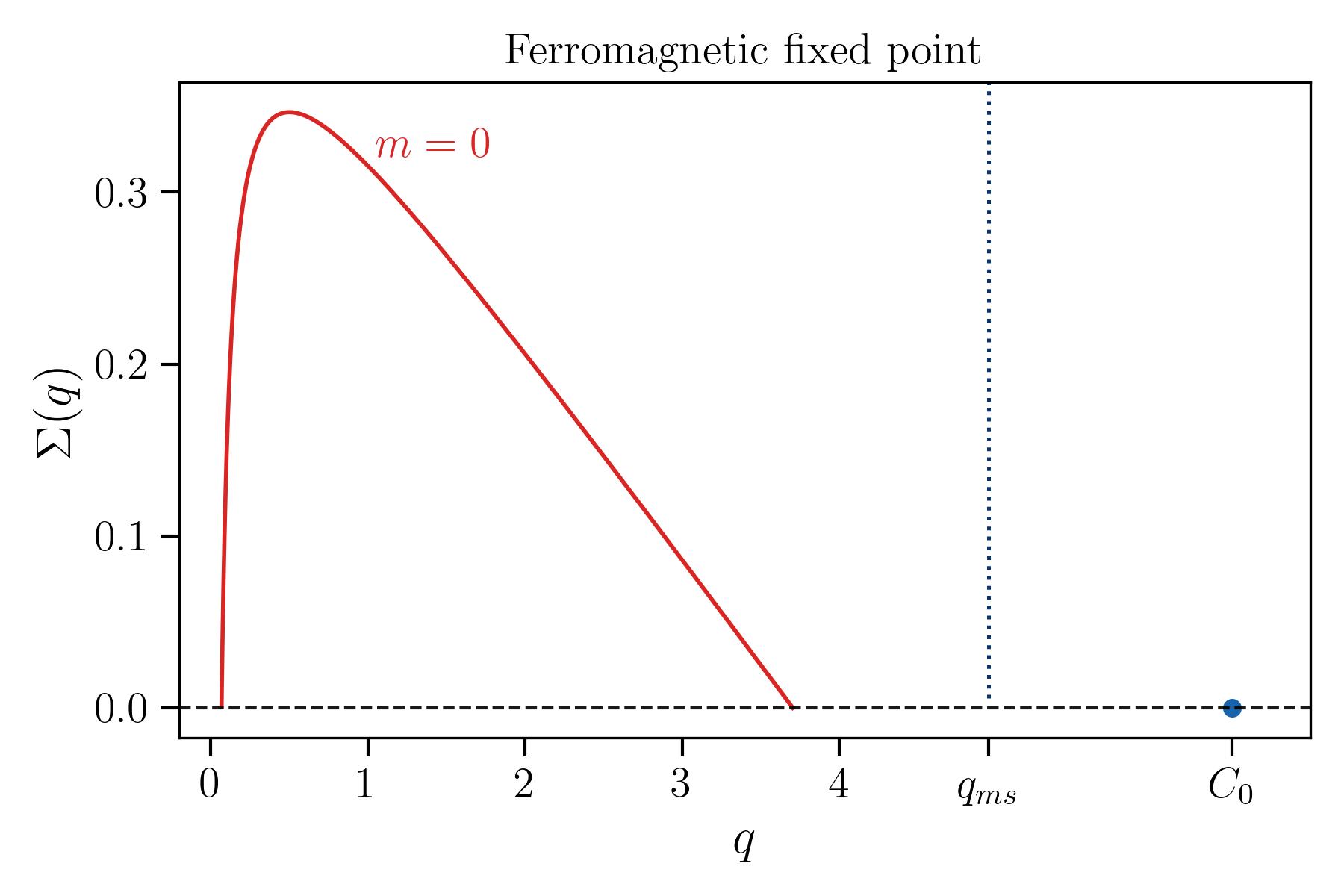}
    \caption{Plot of the quenched complexity for the CM with $g=1,J=6,\gamma=0.5$. The complexity is optimized over all order parameters except $q$, which is used to control the instability index. The curve is red to indicate that the points are paramagnetic, that is, the optimal $m$ is 0. The blue dotted line represents the critical $q=q_{ms}$ ("ms" for marginally stable) that separates unstable from stable equilibria. The blue point is the unique stable fixed point, located at a value of $q$ corresponding to the dynamical order parameter $C_0$.}
    \label{fig:kac_conf_ffp}
\end{figure*}

\noindent Let us remark that in this regime the complexity does not scale with $N$, but instead gives a finite number of equilibria (in this case two). The annealed complexity is exact in this regime, and one can check that, by starting from the chaotic region, the quenched complexity (optimized over $\tilde{Q}$ and $m$) converges to 0 as we approach the transition value $J\to J_c$. More precisely, it is not hard to check that 
\begin{align}
    \lim_{J\to J_c^-}\text{extr}_{m,\tilde{Q},q}\tilde{\Sigma}(m,q,\tilde{Q})=0.
\end{align}
We show a figure of this in Fig.~\ref{fig:rnn_ferro_steps}, when we discuss more carefully the scaling of the complexity at the transition from the FC to FFP phases. \\

\noindent When we lower $J<J_{c}$ the system's dynamics enters the chaotic phase, and the complexity calculation shows that the dynamical order parameters admit exponentially many equilibria. More precisely we observe that, in general, $C_0<q_{ms}$ and $\Sigma(m,C_0)>0$, meaning that there are exponentially many unstable fixed points in the shell chosen by the dynamics. Below we shall make more quantitative comparisons in both the PC and FC phases. \\

\subsubsection{Unstable equilibria in the PC phase. } 
\noindent When we optimize the complexity, we consider fixed $q$, and optimize $\tilde\Sigma$ over $\tilde{Q}$ and $m$, denoting  with $\tilde{Q}_{\text{typ}}$ and $m_{\text{typ}}$ the values where the optimum is attained. In this way we can plot the complexity solely as a function of the extensive instability index of the equilibria, controlled by the self-overlap $q$.  \\

\noindent The paramagnetic solution with $m_{\text{typ},P}=0$ ("P" for paramagnetic) is a solution of the optimization problem, with $\tilde{Q}_{\text{typ},P}=0$. In this case we can verify that the corresponding quenched complexity of unstable equilibria coincides with the annealed one, and they both read:
\begin{align}
 \Sigma(m=0,q)=\frac{1}{8} \left(-\frac{(\gamma - q)^2}{g^2 q} + 4\log2\right),\quad q<q_{ms}.
\end{align}
We can easily see that this curve has a maximum in $q=\gamma$ and goes to zero at 
\begin{align} q_\pm=\gamma + g \left(\pm2 \sqrt{\log2} \sqrt{\gamma + g^2 \log2} + g \log4\right).
\end{align}
A plot is shown in Fig.~\ref{fig:kac_conf_pc}. In particular, we see that the order parameters of the DMFT in the PC region (see Eq.\eqref{eq:dmft_para_orders}) do not correspond to the maximum (nor the minimum) of the complexity curve of paramagnetic fixed points. In order to link dynamical and statical approaches, one might hypothesize that the dynamics  wonders in particular regions of phase space, such as where the equilibria are most abundant, or where they are less unstable. Thus one might try to make a comparison $C_0\leftrightarrow q$, where $q$ is either the maximum of the complexity curve $(q=\gamma)$ or the point where equilibria have the least fraction of unstable modes, that is $q_+$. However we see (red dashed line in Fig.~\ref{fig:kac_conf_pc}) that none of these happen, and the dynamics chooses a value of $C_0$ that lies within the complexity curve. \\


\begin{figure*}[t!]
    \centering
    \includegraphics[width=0.75\textwidth]{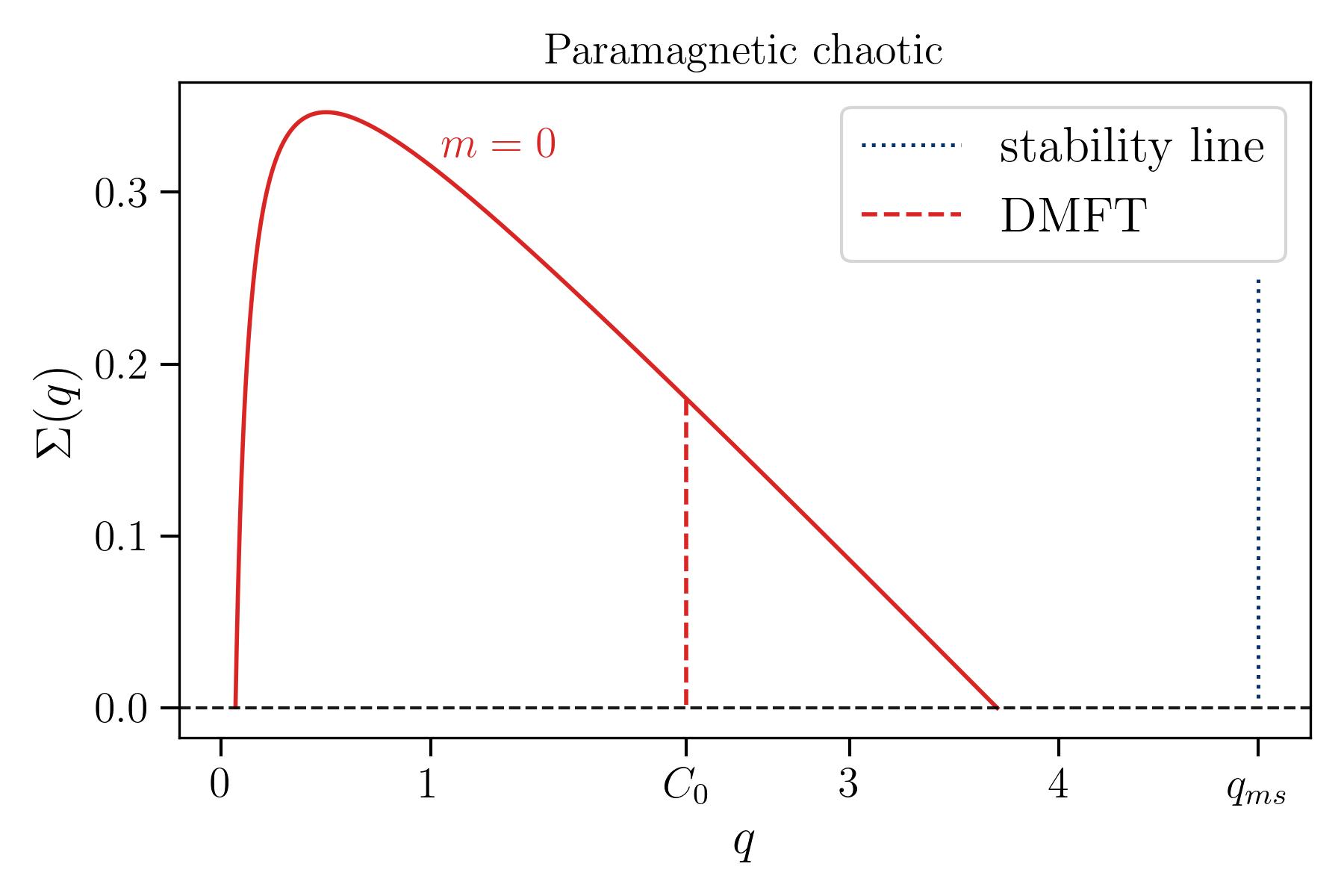}
    \caption{Plot of the paramagnetic complexity of unstable equilibria ($q<q_{ms}$) in the PC phase, with $g=1$, $\gamma=0.5$ ($J$ being irrelevant for the paramagnetic solution). We show the value $C_0$ of the self-overlap achieved by the dynamics in this phase, observing that it does not correspond to any particular point (maximum or minimum) of the complexity.}
    \label{fig:kac_conf_pc}
\end{figure*}

\noindent Quite interestingly, we noticed that at the DMFT value $q\to C_0$ in the PC phase, the complexity is a constant (i.e. does not depend on $g,\gamma$):
\begin{align}
    \Sigma\left(m=0,q=C_0\right)=\frac{1}{6} (3\log2-1).
\end{align}
More precisely, what we observe is that both differential equations
\begin{align} d\Sigma(m=0,q(g,\gamma))/dg=d\Sigma(m=0,q(g,\gamma))/d\gamma=0
\end{align}
admit two families of solutions: one corresponds to the maximum of the complexity, $q(g,\gamma)=\gamma$, and the other one has the same shape of the DMFT solution $C_0(g,\gamma)$, up to a constant. However, we do not have boundary conditions to fix the constants of the solution, which must be imposed by knowledge of at least one point of the DMFT solution (e.g. by knowing expressions for $g=1$ or $\gamma=0$). Additionally, this seems to be a feature of the fact that this model does not present a dynamical PFP (paramagnetic fixed point) phase (perhaps because $\Phi_1$ is homogeneous), which implies that there are always exponentially many fixed points for any $g,J,\gamma>0$ (i.e. no topology trivialization of the red curve in Fig.~\ref{fig:kac_conf_pc} by varying the control parameters). As we show in Chapter~\ref{chapter:scs}, in the model of randomly interacting neurons proposed in \cite{ChaosSompo88} the complexity goes to 0 at the PFP-PC transition happening there, so that this invariance cannot be a general feature. See also Sec.~\ref{sec:confined_mixed} for a choice of $\Phi_1$ that does not possess this property.\\

\noindent We think that this interesting fact is, nonetheless, evidence of a connection between the complexity found via Kac-Rice and the dynamical solution found via DMFT, although at this point we are not able to see how.

\begin{figure*}[t!]
    \centering
    \includegraphics[width=0.75\textwidth]{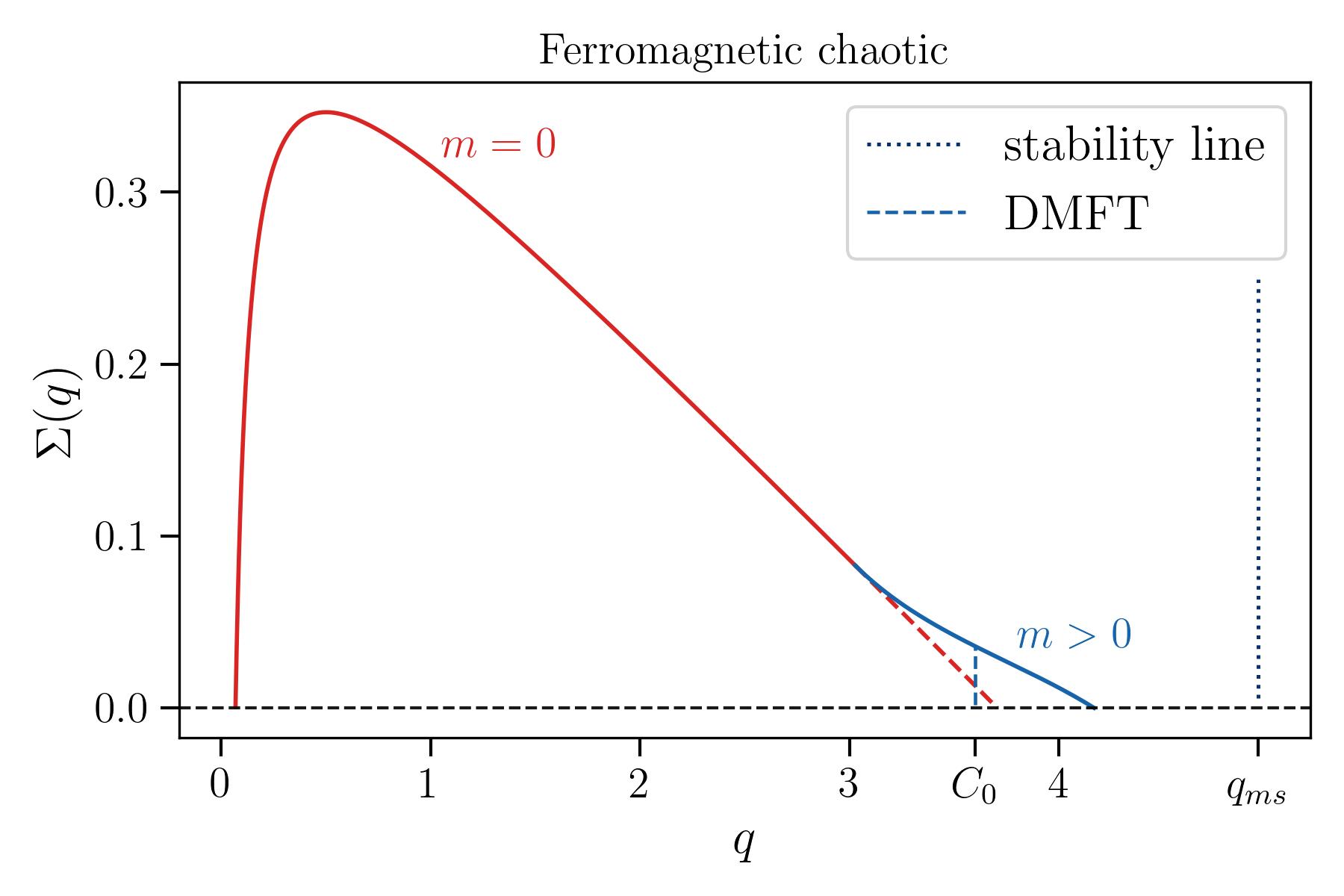}
    \caption{Plot of the quenched complexity in the FC phase, with $g=1$, $\gamma=0.5$, $J=3.1$. The red (blue) curve represents the paramagnetic (ferromagnetic) complexity. The dynamical value $C_0$ obtained for this choice of $g,\gamma,J$ is shown on the $q$-axis. We see that it lies within the curve of the ferromagnetic complexity of unstable equilibria.}
    \label{fig:kac_conf_fc}
\end{figure*}

\subsubsection{Unstable equilibria in the FC phase. }
Another solution of the complexity that is found by optimizing over $m$ and $\tilde{Q}$ is given by ferromagnetic unstable fixed points. This solution, denoted by $m_\text{typ,F}$ and $\tilde{Q}_\text{typ,F}$, appears only in the tail of the complexity curve (see Fig.~\ref{fig:kac_conf_fc}). The values that optimize the complexity are given by:
\begin{align}
\begin{split}
&\tilde{Q}_\text{typ,F}=-\frac{J (2 \gamma + J - 2 q)}{2 g^2} - q,\\
&m_\text{typ,F}
=\pm\sqrt{\frac{\tilde{Q}_\text{typ,F}\,\,\beta}
{2(q + \tilde{Q}_\text{typ,F}) \left[J (2\gamma + J - 2q) \tilde{Q}_\text{typ,F} + g^2 (q + \tilde{Q}_\text{typ,F})^2\right]}}\\
&\beta:=  2g^2 \left(q^3 + 3q^2 \tilde{Q}_\text{typ,F} + q (\tilde{Q}_\text{typ,F})^2 - (\tilde{Q}_\text{typ,F})^3\right)-(\gamma - q)^2 (3q - \tilde{Q}_\text{typ,F}) \tilde{Q}_\text{typ,F}.
\end{split}
\end{align}
An interesting value of $q$ to consider is the one where ferromagnetic equilibria become paramagnetic (see the meeting of the blue and red lines in Fig.~\ref{fig:kac_conf_fc}). By solving for $\tilde{Q}_\text{typ,F}=0$ we find:
\begin{align}
\label{eq:rnn_qpf_def}
q\to q_{\text{PF}}=\frac{2 \gamma J + J^2}{2 (J-g^2)},\quad \quad \text{"PF = para-to-ferro"}.
\end{align}
which also implies $m_{\text{typ, F}}=0$, meaning that at this value of $q$ the optimization of the complexity over $m$ gives only the paramagnetic solution. Notice that for $q<q_{PF}$ ferromagnetic equilibria are still present, but the paramagnetic ones exponentially outnumber them. \\

\noindent At difference with the paramagnetic one, the ferromagnetic quenched complexity does not coincide with the annealed one. Indeed, the annealed complexity is optimized at 
\begin{align}m_{\text{typ, F}}^A=\pm\sqrt{\frac{q (
 2 \gamma J + J^2 + 2 g^2 q - 2 J q)}{J (2 \gamma + J - 2 q)}},
 \end{align}
 and the difference between the annealed and quenched complexities after optimization over $\tilde{Q}$ and $m$ reads:
\begin{align}
D(q):= \Sigma_A(m_\text{typ,F}^A,q)-\Sigma(m_\text{typ,F},q)=\frac{-(\gamma + J - q)^2 (J (2 \gamma + J - 2 q) + 2 g^2 q)^3)}{
 4 g^2 J^2 (2 \gamma + J - 2 q)^2 q (J (2 \gamma + J - 2 q) + 4 g^2 q)}.
\end{align}
The difference $D(q)$ vanishes at $q=J+\gamma$ and $q=q_{PF}$. Notice that the first value corresponds to the asymptotic value of $C_0$ selected by the dynamics in the FC phase, see \eqref{eq:FC_phase}. Hence, we can try to find the magnetization and the overlap $\tilde{Q}$ that maximize the complexity in the shell $q=C_0=J+\gamma$ and compare them with the dynamical values of $m_\infty,C_\infty$. The optimization of the complexity at $q=J+\gamma$ gives:
\begin{align}
\begin{split}
\label{eq:confined_m_Q_ferrp_dmft_comp}
&m_\text{typ,F}(q=J+\gamma)=\pm\sqrt{(\gamma + J) (-2 \gamma g^2 - 2 g^2 J + J^2)}/J\\
&\tilde{Q}_\text{typ,F}(q=J+\gamma)=-\gamma + \frac{1}{2} J (J/g^2-2).
\end{split}
\end{align}
If the system's dynamics, at long times, were to surf between ferromagnetic fixed points in a certain shell of size $q=C_0$, one could expect that $C_\infty$ corresponds to the typical overlap between two ferromagnetic fixed points within that shell, and $m_\infty$ is their typical magnetization. Comparing the values in \eqref{eq:confined_m_Q_ferrp_dmft_comp} with $m_\infty,C_\infty$ as found within the FC phase of the dynamics in Eq.~\eqref{eq:FC_phase}, we see that they do not match. \\

\noindent One can further try to impose that $\tilde{Q}(q=C_0=J+\gamma)=C_\infty$ and then solve for the $m_\text{typ,F}$ such that $\partial_{\tilde{Q}}\tilde{\Sigma}_Q(m_\text{typ},q=C_0,\tilde{Q})|_{\tilde{Q}=C_\infty}=0$: however, also the $m_\text{typ}$ found in this way does not coincide with the DMFT value $m_\infty$. Moreover, at difference with the paramagnetic chaotic case, the complexity evaluated at the dynamical order parameters, that is $\tilde{\Sigma}(m_\infty, C_0,C_\infty)$, is not $g,\gamma$ independent.\\

\noindent We have therefore carried out a precise comparison between the complexity and the DMFT in the FC case. Quite interestingly we have seen that at the dynamical value ($C_0=J+\gamma$) quenched and annealed complexities coincide (after optimization over the other order parameters). However, the optimal values of the other order parameters of the complexity do not match with those found via DMFT, at the saddle point. Additionally, we used an RS ansatz, but we did not check for its stability, which is left for future works.\\ 

\noindent Thus, besides this matching of quenched and annealed complexities at $C_0$ in the FC phase, and the strange invariance of the complexity in the PC phase, we are not able to draw other solid links between the Kac-Rice complexity and the DMFT in the chaotic regions. Moreover, these two coincidental features are probably due to the fact that $\Phi_1$ was chosen to be homogeneous; indeed in the SCS model studied in Chapter~\ref{chapter:scs} these links do not hold anymore.

\subsubsection{Ferromagnetic to Paramagnetic transition.} 
We have seen that for $J$ big enough the complexity has a ferromagnetic branch, see blue curve in Fig.~\ref{fig:kac_conf_fc}. In particular, the red (paramagnetic) and blue (ferromagnetic) curves of the complexity meet at a point $q_{PF}$ found in Eq.~\ref{eq:rnn_qpf_def}. In particular, there must be a range $J\in[J_-,J_+]$ where this ferromagnetic (blue) branch is present. This range is found by imposing that $q_{PF}=q_+$ and solving for $J$:
\begin{align*}
&q_{PF}\overset{!}{=}q_+\Rightarrow \frac{2 \gamma J + J^2}{2 (J-g^2)}=\gamma + g \left(2 \sqrt{\log2} \sqrt{\gamma + g^2 \log2} + g \log4\right)\\
&\Rightarrow  J_{\pm}=g \left(\sqrt{\ln 16} \pm \sqrt{\ln16-2}\right) \left(\sqrt{\gamma + g^2\ln2} + g \sqrt{\ln2} \right).
\end{align*}
In particular, we denote by $J_\Sigma:=J_-$. This line is particularly important, as it is very close to the line of $J_{PF}$ found via DMFT, where a dynamical transition from paramagnetic to ferromagnetic chaos appears, see Fig.~\ref{fig:rnn_dmft_phase_kac}. A possible conjecture is that the dynamical transition from paramagnetic to ferromagnetic chaos appears at the same value of $J=J_\Sigma$ where the most abundant equilibria at large $q$ become ferromagnetic. Although we have seen that in the FC phase $C_0$ does not correspond to the point where the ferromagnetic complexity goes to zero (see Fig.~\ref{fig:kac_conf_fc}), one could still imagine that the two points collapse at the transition. However, as we can see, the equations for $J_\Sigma$ and $J_{PF}$ are different. Despite this, the fact that they are surprisingly close (see Fig.~\ref{fig:rnn_dmft_phase_kac}) suggests that there might be a stronger underlying connection. 

\begin{figure*}[t!]
    \centering
    \includegraphics[width=0.68\textwidth]{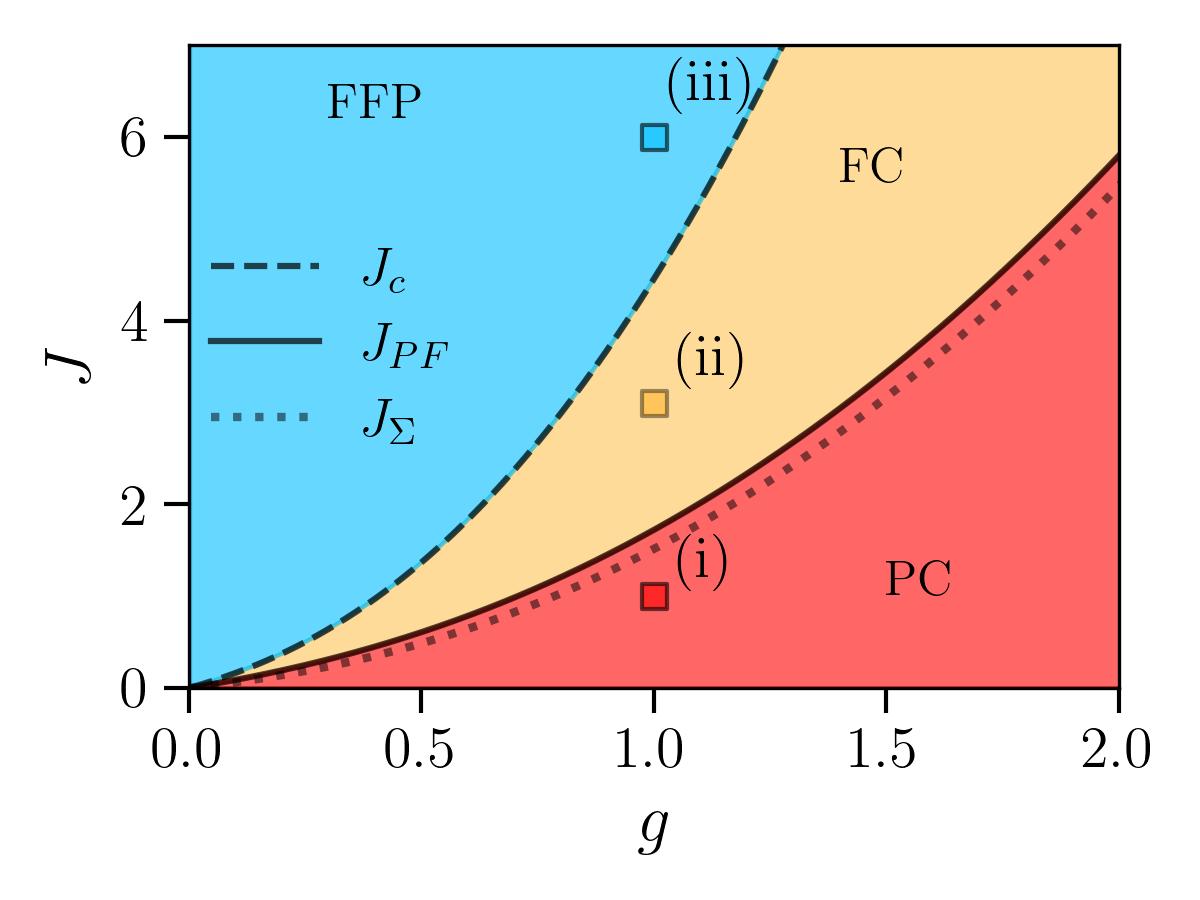}
    \caption{Dynamical phase diagram where now we have added the transition line $J_\Sigma$, which indicates the point where a branch of ferromagnetic complexity appears, as we increase $J$ from the PC phase.}
    \label{fig:rnn_dmft_phase_kac}
\end{figure*}

\noindent In particular, $J_\Sigma$ lies below the critical line $J_\text{PF}$. This means that in the PC phase (sufficiently close to the transition line $J_{PF}$), the fixed points with less unstable directions are ferromagnetic, despite the chaotic attractor being paramagnetic. The region where these ferromagnetic unstable equilibria exist extends up to $J_+$, which lies above $J_c$, meaning that in the FFP phase, there are also exponentially many ferromagnetic unstable equilibria for $J$ sufficiently close to $J_+$ (not shown in the picture).\\

\noindent One last attempt that we can make is to find the transition when the complexity at the point $q=C_0$ in the FC phase becomes paramagnetic as we lower $J$. This happens for:
\begin{align}
    q_{PF}\overset{!}{=}C_0\Rightarrow\frac{2\gamma J+J^2}{2(J-g^2)}=J+\gamma\Rightarrow J=g \left(g + \sqrt{2 \gamma + g^2}\right),
\end{align}
which again does not correspond to the dynamical ferromagnetic-to-paramagnetic chaotic transition line $J_\text{PF}$.\\

\begin{figure*}[t!]
    \centering
    \includegraphics[width=1\textwidth, trim={4 4 4 4},clip]{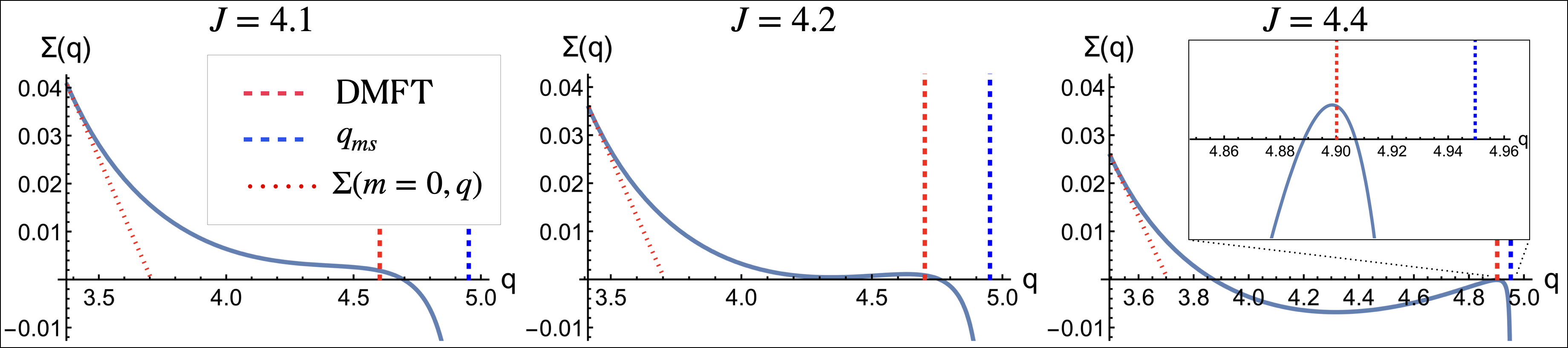}
    \caption{Complexity curves $\Sigma(q)\equiv \Sigma(m_\text{typ}(q),q)$ of the confined model for $g=1$ and $\gamma=0.5$, in the ferromagnetic chaotic phase (we zoom on the ferromagnetic branch). Blue lines correspond to ferromagnetic equilibria, red dotted lines to paramagnetic ones. The blue dashed line marks the marginal stability line $q_{\text{ms}}$, red dashed lines indicate the position of the dynamical order parameter $C_0$. From left to right, as $J$  increases, the ferromagnetic complexity develops a local maximum. As $J\to J_c$, the local maximum converges to $q_\text{ms}$ and the (local) complexity vanishes.}
    \label{fig:rnn_ferro_steps}
\end{figure*}

\subsection{Additional results and comparison with the Lyapunov}
In this paragraph we wish to investigate the behavior of the complexity for values of $J$ that are close to the transition line $J_c$ (i.e. from FC to FFP). We see from Fig.~\ref{fig:rnn_ferro_steps} that the dynamical order parameter $C_0$ (red dashed line) always lies within a region of positive complexity (blue line) as $J\to J_c$. In particular, as we can see from the figure, the complexity develops a small \textit{island} with a local maximum, close to which $C_0$ resides. However, $C_0$ does not correspond to the maximum of this island. It was suggested in \cite{wainrib2013topological} that the complexity takes the same shape of the maximal Lyapunov exponent as we cross a topology trivialization transition, from one stable equilibrium to an exponential abundance of unstable ones. In this case we therefore have to study the scaling of the maximum of this \textit{island} of complexity at the transition. If we fix $g$ (and use $\gamma=0$ for simplicity) and we write $J=J_c-4g^2\delta=4g^2(1-\delta)$ we find that for $\delta<<1$:
\begin{align}
    \Sigma_{lm}(\delta)=\frac{1}{4}\delta^2+\mathcal{O}(\delta^3)
\end{align}
where $\Sigma_{lm}$ is the complexity at this \textit{local maximum} described above. The quadratic behavior of $\Sigma_{lm}$ is robust as the same behavior holds if we consider the annealed complexity (although the prefactor will be different). The scaling of the maximal Lyapunov exponent at the transition instead reads\footnote{Private conversation with S. J. Fournier}:
\begin{align}
    \lambda_{max}(\delta)=\frac{5}{4}g^2\delta+\mathcal{O}(\delta^2)
\end{align}
We see that not only do they not match, but they also have different critical exponents. \\

\noindent We note that this comparison can be made for multiple choices of the model (both CM and SpM) and varying both $g$ or $J$ at the FC-FFP transition. For all cases, we verified that the scalings of maximal Lyapunov and complexity do not match (see the new upcoming version of \cite{us_non_reciprocal_2025}). This proves that the Kac-Rice is not sufficient to explain the dynamics of the system, and we hope that this will motivate further research to bridge the gap.

\section{Dynamics and complexity: $\alpha>0$}
\label{sec:rnn_alpha>0}
In this section we analyze what happens when we switch on correlations between different units, that is, we study the case $\Phi_2(u)=\alpha\Phi'_1(u)$ with $\alpha\in(0,1)$. Moreover, we now consider the SpM with $J=0$ and $q=1$, since calculations are a bit easier in that setting (and we refer to our work \cite{us_non_reciprocal_2025} for much broader considerations). This essentially consists in the spherical $p$-spin model with asymmetry $\alpha$ and where we allow for the interaction strength between the units to be tuned.\\

 \begin{figure*}[t!]
    \centering
    \includegraphics[width=\textwidth, trim={5 5 5 5},clip]{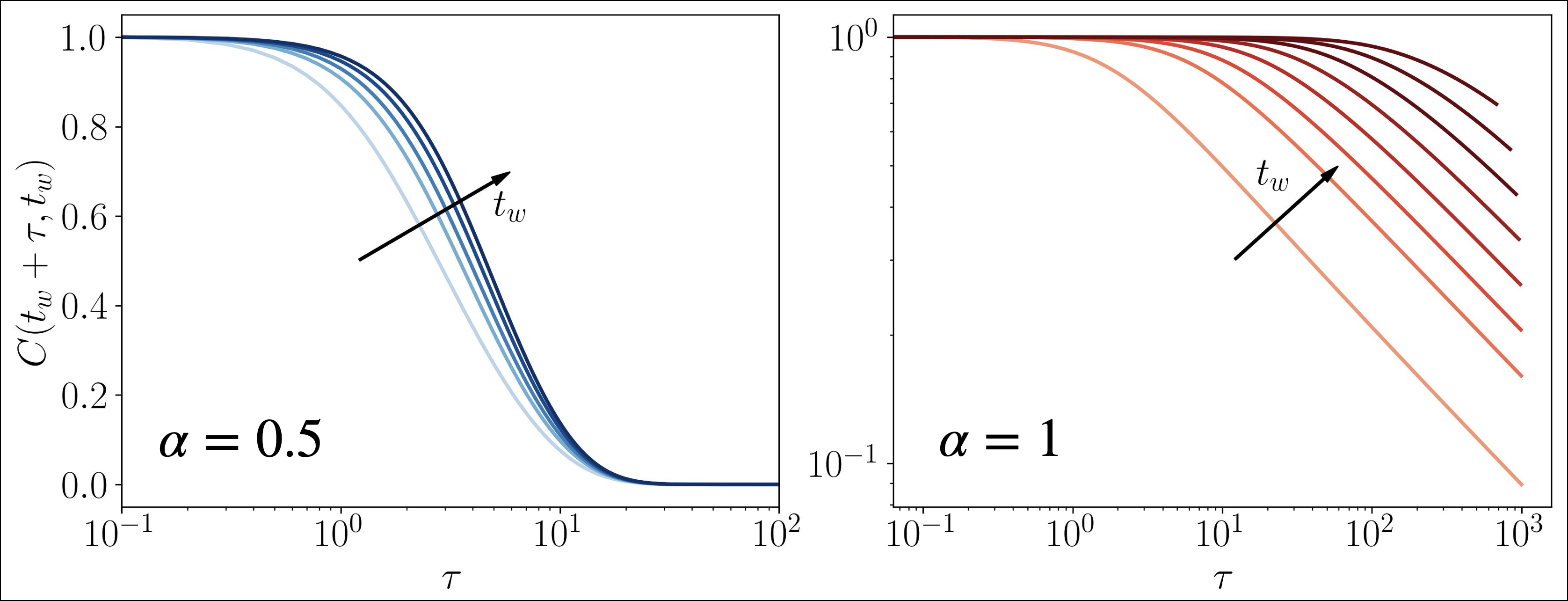}
    \caption{Plot of the correlation function $C(t_w+\tau,t_w)$ of the zeo-temperature SpM with $g=1$, $J=0$ and $\Phi_1(u)=2g^2u^2$, $\Phi_2(u)=\alpha\Phi_1'(u)$, as a function of $\tau$ for $\alpha=0.5$ (\textit{left}) and $\alpha=1$ (\textit{right}). We see that for $\alpha<1$ the correlation function becomes time translationally invariant as we increase $t_w$ \cite{CugliandoloNonrelax97}, whereas the same does not hold for $\alpha=1$, which displays aging \cite{cugliandolo1993analytical}.
    The plots are obtained by Euler discretization of the DMFT equations for $C$ and $R$ with a time-step $dt=0.1$. This method is more than sufficient for $\alpha<1$, while for $\alpha=1$ better methods exist \cite{Bongsoo_Kim_num_2001, tersenghi_seb_mixed_2025}. Both plots have the $x$-axis in log scale, while the \textit{right} plot has the $y$-axis in log scale as well. }
    \label{fig:rnn_two_alphas_dmft}
\end{figure*}

\noindent The threshold value that separates unstable from stable equilibria in the SpM is given by $\lambda_{ms}=\sqrt{\dot{\Phi}_1^1}(1+\alpha)$, see Eq.~\ref{eq:lambda_ms_spm}. The Replica Symmetric (RS) computation of the quenched complexity in Sec.~\ref{sec:rnn_topo_general} suggests that the quenched paramagnetic complexity (i.e. $m=0$) coincides with the annealed one if $\Phi_1(0)=\Phi_1'(0)=0$. Indeed, in that case, imposing $\tilde{Q}=0$ solves the Saddle Point equation on $\tilde{Q}$ and reduces the RS quenched complexity to the annealed one \footnote{for a more rigorous proof we could use our formulas for the $n-$th moment used in the replica trick and plug $n=2$, thus computing the second moment. This is left for future work.}. Assuming that $0\leq \Phi_1^1\leq \dot{\Phi}_1^1$ and that $\alpha\dot{\Phi}_1^1\neq \Phi_1^1$, the annealed complexity in the unstable region reads
\begin{align}
\label{app:eq:unstable_compl}
\begin{split}
\Sigma_A(m,\lambda)&=\frac{1}{2} \Bigg\{
\frac{ 
\lambda^2  ( \Phi_{1}^{1}-\dot{\Phi}_1^1)}{
(1 + \alpha) \dot{\Phi}_1^1 ( \alpha \dot{\Phi}_1^1 + \Phi_{1}^{1})}+
\log{\left[\frac{(1-m^2) \dot{\Phi}_1^1}{\Phi_{1}^{1}}\right]}
\Bigg\},
\end{split}
\end{align}
and we see that the maximum of $\Sigma_A(m,\lambda)$ is always attained for $m=0$ at any values of $\lambda$. Then the paramagnetic complexity of unstable equilibria reads:
\begin{align}
    \Sigma(m=0,\lambda)=\frac{1}{2} \Bigg\{
\frac{ 
\lambda^2  \left( \Phi_{1}^{1}-\dot{\Phi}_1^1\right)}{
(1 + \alpha) \dot{\Phi}_1^1 \left( \alpha \dot{\Phi}_1^1 + \Phi_{1}^{1}\right)}+
\log{\left[\frac{ \dot{\Phi}_1^1}{\Phi_{1}^{1}}\right]}
\Bigg\}.
\end{align}
The absolute maximum is obtained for $\lambda=0$ and reads $\Sigma(m=0,\lambda=0)=\log(\dot{\Phi}_1^1/\Phi_1^1)$, as already obtained in \cite{Fyodorov_2016}. By solving $\Sigma(m=0,\lambda_{ms})=0$ we can find the critical value $\alpha_c$ such that marginally stable equilibria start to appear. A simple computation gives:
\begin{align}
\alpha_c=
\frac{\dot{\Phi}_1^1 - \Phi_{1}^{1}-\Phi_{1}^{1} \log{\left(\frac{\dot{\Phi}_1^1}{\Phi_{1}^{1}}\right)}}{ \Phi_{1}^{1}- \dot{\Phi}_1^1 +  \dot{\Phi}_1^1 \log{\left(\frac{ \dot{\Phi}_1^1}{\Phi_{1}^{1}}\right)}},
\end{align}
which can be easily seen to be bounded, in modulus, by 1. \\

\begin{figure*}[t!]
  \centering
  \begin{subfigure}[b]{0.495\textwidth}
    \centering
    \includegraphics[width=\linewidth,  trim={5 5 5 5}, clip]{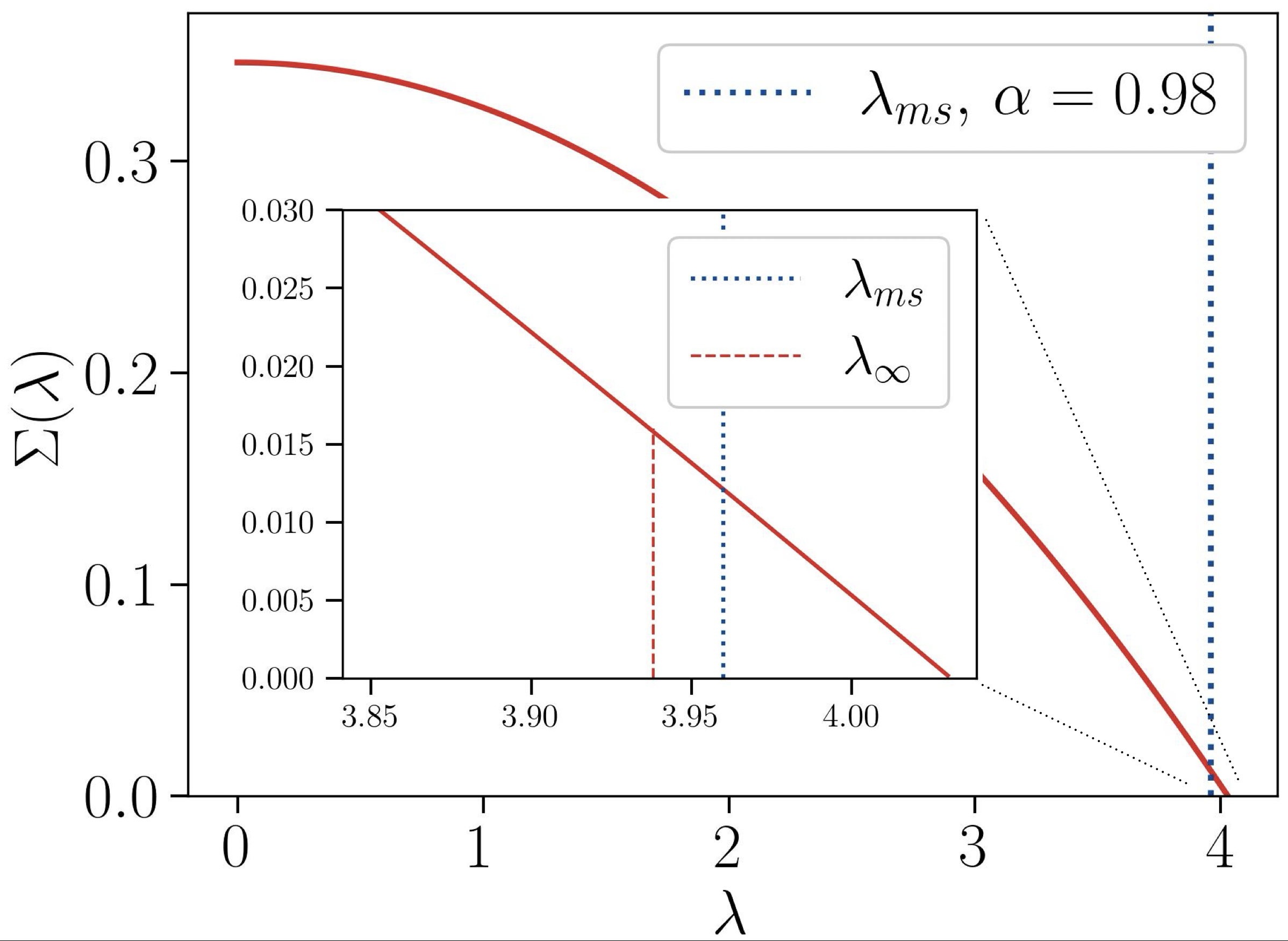}
  \end{subfigure}
  \hfill
  \begin{subfigure}[b]{0.495\textwidth}
    \centering
    \includegraphics[width=\linewidth]{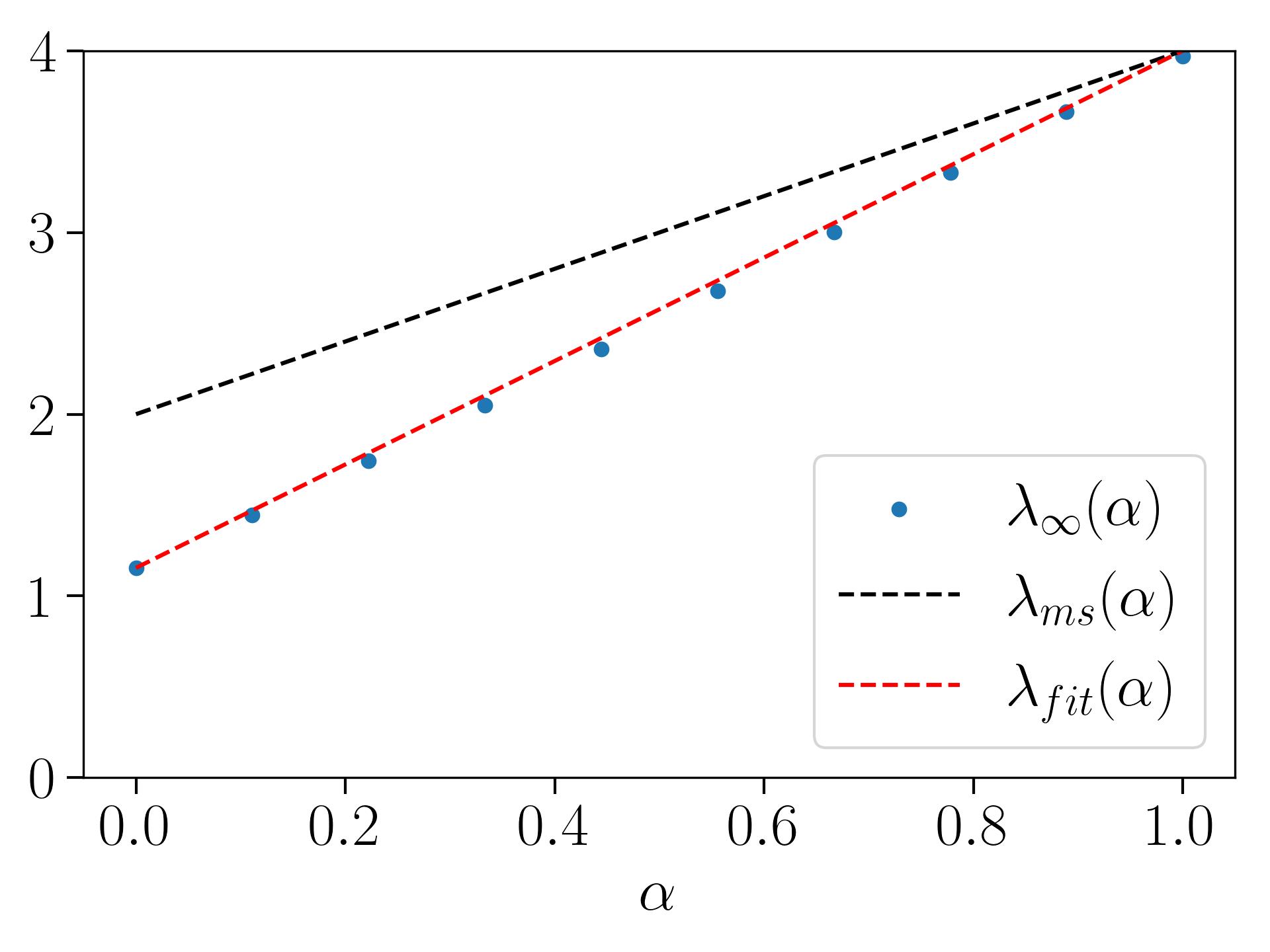}
  \end{subfigure}
  \caption{\textit{Left}. Plot of the paramagnetic complexity $\Sigma(\lambda)=\Sigma(m=0,\lambda)$ for the SpM with $\Phi_1(u)=2g^2u^2$, $\Phi_2(u)=\alpha\Phi_1'(u)$, $g=1$, $J=0$. We show the (small) difference between $\lambda_{ms}$ and $\lambda_\infty$. \textit{Right}. Plot of $\lambda_{ms}(\alpha)$ and $\lambda_\infty(\alpha)$ (found by integrating the DMFT equations with Euler discretization) for $g=1$. We see that $\lambda_\infty(\alpha)$ lies on a line $\lambda_{fit}(\alpha)$. The choice $g=1$ is just for simplicity, the plots do not show qualitative difference by changing $g$.}
  \label{fig:rnn_kac_spm}
\end{figure*}

\subsubsection{Choice $\Phi_1(u)=2g^2u^2$}
\noindent For the particular choice of $\Phi_1(u)=2g^2u^2$, we find for the  SpM that $\alpha_c=(1-\log(2))/(2\log(2)-1)\approx 0.8$, which is $g$ invariant, and $\lambda_{mg}=2g(1+\alpha)$. While finishing our work \cite{us_non_reciprocal_2025}, we realized that this value of $\alpha_c$ was already computed in Ref.~\cite{Kivimae2024}, but our results are more general and can be applied to a wide range of models, with or without spherical constraints, and with arbitrary $\Phi_1,\Phi_2$. Let us also mention that a similar coexistence of stable and unstable equilibria was also observed in \cite{ben2021counting}.\\

\noindent Providing an analytical solution to the DMFT equations with $\alpha<1$ remains unsolved (for $\alpha=0$ they were solved by us in \cite{us_non_reciprocal_2025}, and for $\alpha=1$ in \cite{cugliandolo1993analytical}). A numerical integration of the DMFT equations at zero-temperature for the model under consideration (see Appendix.~\ref{app:dmft_equations} for the equations) reveals that the correlation function $C(t,t')$ becomes time translation invariant for any $\alpha<1$ (this was already noted in Ref.~\cite{CugliandoloNonrelax97, berthier2000two}), while for $\alpha=1$ we have the well-known phenomenon of aging \cite{cugliandolo1993analytical, cugliandolo2002dynamics, Cugliandolo_1995_weak, bouchaud1998out}. In fact, we see from Fig.~\ref{fig:rnn_two_alphas_dmft} \textit{left} that for $\alpha=0.5$ the correlation functions $C(t_w,t_w+\tau)$ converge to a unique curve as $t_w$ increases. The same is not true for $\alpha=1$, as we see in the \textit{right} plot in the same figure. \\

\noindent To the best of our knowledge, it was not previously observed that the asymptotic values of the Lagrange multiplier $\lambda_\infty$ form a simple line as a function of $\alpha$, see Fig.~\ref{fig:rnn_kac_spm} \textit{right}. Our conjecture is that $\lambda_\infty(\alpha)$ takes the form of $\lambda_{fit}$ in the figure, defined as
\begin{align}
    \lambda_{fit}(\alpha):=4g\alpha + \frac{2}{\sqrt{3}}g(1-\alpha)
\end{align}
where we used knowledge that $\lambda_\infty(\alpha=1)=\lambda_{ms}(\alpha=1)=4g$ and that $\lambda_\infty(\alpha=0)=\frac{2}{\sqrt{3}}g$ as we found in Ref.~\cite{us_non_reciprocal_2025} by applying the same techniques as for the CM in Sec.~\ref{sec:rnn_dmft_tti}. A proof that $\lambda_{fit}$ is the right solution is left for future work. From Fig.~\ref{fig:rnn_kac_spm}~\textit{left} we see that for $\alpha>\alpha_c$ there are stable equilibria, although the system is chaotic \footnote{when starting from a random initial condition.} and $\lambda_\infty<\lambda_{ms}$; then, as $\alpha\to 1$, we have that $\lambda_\infty$ smoothly converges to $\lambda_{ms}$ (as we see from Fig.~\ref{fig:rnn_kac_spm}~\textit{right}). It was observed in Ref.~\cite{CugliandoloNonrelax97} that if the system is prepared at low energies at a certain temperature, it can get stuck in stable equilibria, even at non-zero temperature for $\alpha$ large enough. With the present analysis we can precisely characterize the regions and the values of $\alpha$ where there is a positive complexity of stable equilibria, cf. Fig.~\ref{fig:rnn_kac_spm}~\textit{left}. The complexity is still $g$ invariant when evaluated at $\lambda_{fit}(\alpha)$, but it is not $\alpha$ invariant, meaning that $\Sigma(m=0,\lambda_{fit}(\alpha))$ is a function of $\alpha$ alone. However, with the complexity alone we are not able at the moment to predict $\lambda_\infty$ from purely static arguments, but we hope that this motivates research to look for a way to bridge the static and dynamic properties. In particular, we should be able to predict $\lambda_\infty(\alpha)$ and $\mathcal{E}_\infty(\alpha)$ 
\footnote{the asymptotic value of the conservative part of the force ${\bf f}$.} both by solving the DMFT equations and by performing some static calculation. Ideally, we should build an exact mapping between static and dynamic approaches for any $\alpha\in[0,1)$, just as we did for the pure spherical $p$-spin (i.e. for $\alpha=1$) in Chapter~\ref{chapter:intro}. We tried several routes, e.g. by considering the volume of phase space that satisfies the basic steady state constraints that $d\lambda/dt=0$ and $d\,{\bf F}^2/dt=0$, but at the saddle point we find values that differ from the dynamical ones. This research is left for future work.

\section{Dynamics and complexity: confined mixed models}
\label{sec:confined_mixed}
We have found that a class of confined \textit{mixed} models (CMM), first studied in \cite{fournier2023statistical}, has a very interesting phenomenology. Let us consider a confined model where $\lambda=\lambda(q)$, $\Phi_2(u)=0$ and consider at first $\Phi_1(u)$ generic with $\Phi'_1(u)>0$ for $u>0$, and $J=0$ for simplicity. Having the dependence only in $q$ has the advantage that the extensive instability index of fixed points is only a function of $q$. Hence, in each shell $q$ of phase space, we typically have only one type of fixed points (for $N\to\infty$). The \textit{paramagnetic} (i.e. $m=0$) complexity of unstable fixed points reads:
\begin{align}
\Sigma(q)=\frac12\!\left(
\frac{\lambda^{2}(q)}{\dot{\Phi}_{1}(q)}
-\frac{q\,\lambda^{2}(q)}{\Phi_{1}(q)}
+\log\left(\frac{q\,\dot{\Phi}_{1}(q)}{\Phi_{1}(q)}\right)
\right),\quad \lambda(q)<\sqrt{\dot{\Phi}_1(q)}.
\end{align}
Moreover, we can use the same methods of Sec.~\ref{sec:rnn_dmft_tti} to derive general equations for the TTI regime of the DMFT for a system starting in random initial conditions \footnote{a priori, a system might not have a TTI regime, which should therefore be checked a posteriori (e.g. by simulations and direct integration of the DMFT equations)}:
\begin{align}
&\Phi_1(C_\infty)-\lambda^2(C_0)\,C_\infty+=0\\
&V(C_0)-V(C_\infty)=0\\
&V(C):=-\frac{\lambda^{2}(C_\infty)\,C^{2}}{2}+
   \int_0^Cdx\,\Phi_1(x).
\end{align}
Now, we want to model a system that has a dynamical PFP-PC phase transition (i.e. the stable fixed point ${\bf x}={\bf 0}$ becomes unstable as a control parameter is varied). The interest stems from studying the relation between the dynamical transition and the topology trivialization of the number of equilibria (vanishing of complexity). To achieve this we need a function $\Phi_1$ that allows for the fixed point ${\bf x}={\bf 0}$ to be stable. A possible choice is a \textit{mixed} model, in the same spirit of the spherical $(3+4)-$spin model analyzed in Sec.~\ref{sec:mixed_models}. We therefore choose $\Phi_1(u)=2g^2u^2+g^2u$,
\begin{figure*}[t!]
    \centering
    \includegraphics[width=\textwidth, trim={5 5 5 5},clip]{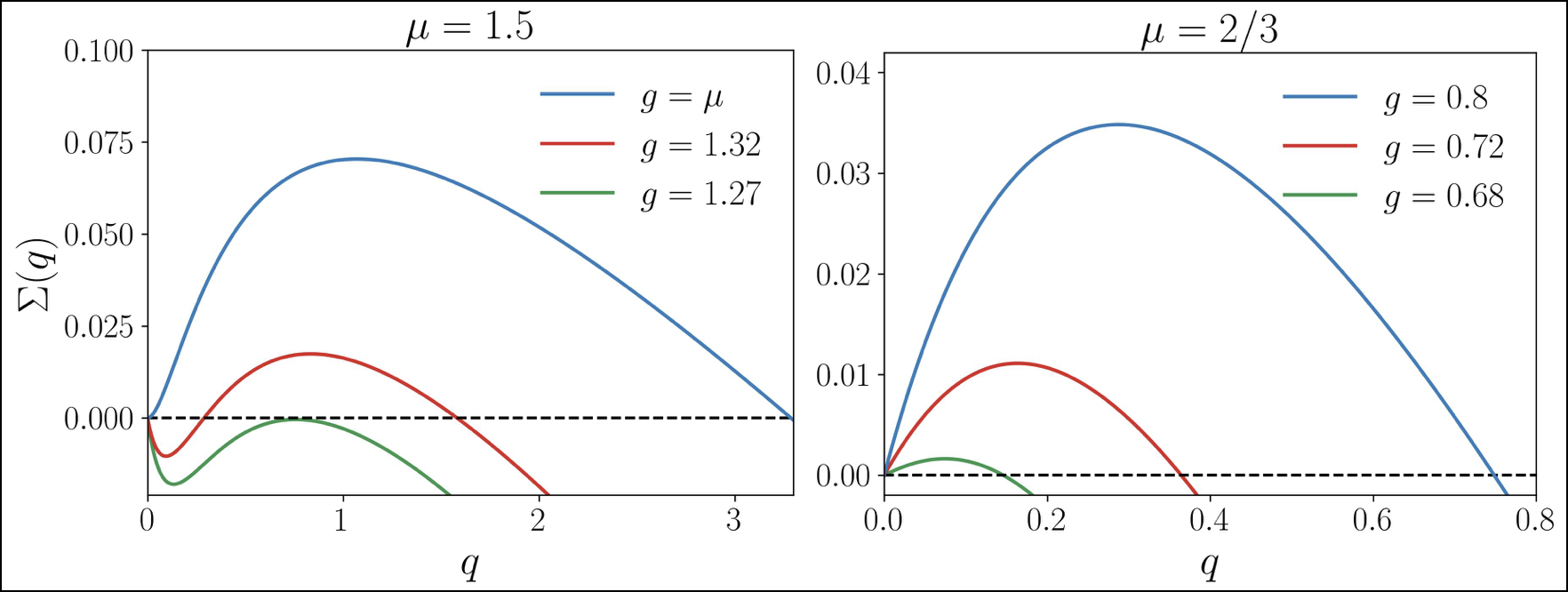}
    \caption{Picture of the complexity for the confined mixed model with $\Phi_1(u)=2g^2u^2+g^2u$, $\lambda(q)=q+\mu$, $\alpha=0$. We abusively plot the complexity also when it is negative, to show its behavior. \textit{Left}. Choice $\mu=1.5>2/3$, we see that the PFP-PC transition at $g=\mu$ corresponds to a vanishing of complexity in the neighborhood of $q=0$ (blue line) , but for $g\lesssim \mu$ the complexity remains positive after a certain gap from the origin (see red line). For $g\approx 1.27$ the complexity vanishes (green line). \textit{Right}. For $\mu=2/3$, the PFP-PC transition correponds to a topology trivialization: the complexity vanishes monotonously as $g\to\mu^+$.}
    \label{fig:confined_mixed_mus}
\end{figure*}
where $g$ is the interaction strength that we can tune, just like before. We can then choose a confining potential of the form $\lambda(q)=\mu+q$, with $\mu>0$ a tunable number. Then we have that, dynamically, the system is stable in ${\bf x}={\bf 0}$ for $g<\mu$, and transitions to chaos for $g>\mu$ \footnote{this is easily seen numerically, but a more convincing computation of the Lyapunov exponent will appear in the new version of  \cite{us_non_reciprocal_2025}.}. We can then analyze the behavior of the complexity. It is easy to see that $\lim_{q\to 0}\Sigma(q)=\lim_{q\to 0}\Sigma'(q)=0$. Moreover, by looking at the second derivative $\Sigma''(q)|_{q=0}$ and expanding $g=\mu+\epsilon$, we can find a criterion for a change in convexity of the complexity at $q=0$ for $g$ close to $\mu$. In particular, we find that for $0<\mu\leq 2/3$ the convexity of $\Sigma(q)$ at $q=0$ is negative, see Fig.~\ref{fig:confined_mixed_mus}~\textit{right}: the transition to chaos then corresponds to a topology trivialization, since for $g<\mu$ no more fixed points are present (except at ${\bf x}={\bf 0}$ of course). Instead for $\mu>2/3$ we have that the complexity has a positive convexity for $g$ close to $\mu$: for $g<\mu$ a \textit{resilience gap} develops between the stable fixed point (at ${\bf x}={\bf 0}$) and the point where the complexity is positive; for $g>\mu$ instead a branch of positive complexity starts to develop from $q=0$; as $g$ is lowered enough, the complexity eventually vanishes, see Fig.~\ref{fig:confined_mixed_mus}~\textit{left}. In Fig.~\ref{fig:confined_mixed_mus} we show the two scenarios (left for $\mu=1.5$ and right for $\mu=2/3$). In particular, for $\mu=1.5$ an "island" of positive complexity is still present even before the transition to chaos, i.e for $g<1.5$ (red line), and the topology trivialization happens for $g\approx 1.27<\mu$ (green line). A rather similar phenomenon was noticed in \cite{Fyodorov_resilient_2021}, where a \textit{resilience gap} was observed, between the stable fixed point ${\bf x}={\bf 0}$ and the emergence of exponentially many unstable equilibria beyond a certain critical radius. Here, for the choice in Fig.~\ref{fig:confined_mixed_mus}~\textit{left}, we have both a gap and a subsequent topology trivialization. It is interesting that the transition to chaos does not correspond to the vanishing of the complexity in this case, but nonetheless it corresponds to an exponential increase in unstable stationary points close to the origin. Let us finally remark that this phenomenon arises thanks both to the choice of $\Phi_1$ being \textit{mixed} and to the choice $\lambda=\mu+q$. Further research is needed to understand what are the general conditions under which the gap develops. We also leave for further work the detailed analysis of the DMFT equations in this case.

\section{Conclusions and Perspectives}
\label{sec:rnn_future}
We have carefully compared the dynamics and the statistics of equilibria for systems with non-reciprocal interactions and an external field, considering models with a confining potential as well as spherically constrained models. Our results indicate that, within the chaotic phase of these models, a simple interpretation of the dynamics as "surfing" among the most abundant equilibria is unlikely to be true. Indeed, we have found that the dynamical order parameters cannot be inferred from the Kac-Rice complexity of unstable equilibria. Nonetheless, we have observed some links between the two calculations, that point to deeper connections worth exploring. Finally, we have considered an interesting class of confined models, namely confined mixed models, that challenge our understanding of transition to chaos in terms of an explosion in the number of unstable equilibria. Indeed, we have seen that, depending on some tunable parameter, such transition to chaos \textit{can} or \textit{cannot} be concomitant with a topology trivialization transition. In the latter case, a resilience gap develops between the unique stable fixed point and the shell where exponentially many unstable fixed points are present. In particular, a topology trivialization still occurs, but well within the non-chaotic phase. With this research, our hope is to lay the foundations, and motivate further research, to better understand the relation between statics and dynamics within the chaotic phase of high-dimensional models with non-reciprocal interactions.\\

\noindent Here are some future directions:
\begin{itemize}
    \item It remains an open and important problem to match the static and dynamic approaches as soon as some non-reciprocity is introduced in the system. In particular, for $\alpha=0$ we have the explicit DMFT solution, but yet no static approach can predict the dynamical overlaps. 
    \item Despite lack of evident connections, the behavior of the complexity and of the TTI solution of the DMFT are correlated (e.g. see $J_c$ and $J_\Sigma$ in Fig.~\ref{fig:rnn_dmft_phase_kac}). It would be interesting to explore if there is a deeper link between the two methods.
    \item I believe that a solution to the DMFT for $1>\alpha>0$ in the PC phase can be obtained (i.e. find $\lambda_{\infty}$ in the SpM, and $C_0$ in the CM, as functions of $\alpha,g$). This is a very interesting open problem that needs further research. In addition, more research is needed to better explore the relation between the dynamics and the complexity in this case.

    \item The confined mixed models (CMM) deserve more attention in future works. In particular, we saw that a transition to chaos is accompanied by an increase in complexity of unstable equilibria close to the origin ${\bf x}={\bf 0}$; but depending on the choice of $\mu$, this can be accompanied by a topology trivialization or not (in which case a resilience gap develops). Dynamically, what distinguishes the two cases? In the former case, can the presence of unstable equilibria even in the stable phase influence the dynamics of the system, if properly initialized ? A closer comparison between dynamics and complexity for these models is an interesting future direction.
\end{itemize}

\chapter{Complexity and dynamics of a random neural network}
\label{chapter:scs}

\noindent This chapter is based on article \cite{pacco_scs_2025} in preparation. The goal is to study the topological complexity of the SCS (Sompolinsky, Crisanti, Sommers) random neural network with an external input, and make a comparison with the results found via dynamical mean-field theory (DMFT), which are already known in the literature. \\

\noindent \textit{Road-map}\\
In Sec.~\ref{sec:scs_introduction} we introduce the problem. In Sec.~\ref{sec:scs_the_model} we present the definition of the model and its effective dynamical equation in the large $N$ (number of interacting units) limit. In Sec.~\ref{sec:scs_dmft_zero_gamma} we study the case of  totally asymmetric interactions and derive the DMFT equations in the TTI regime; we then explore the dynamical phase diagram. In Sec.~\ref{sec:scs_topo_compl} we compute the topological complexity in the annealed setting via Kac-Rice formalism, and in Sec.~\ref{sec:scs_ann_topo_compl} we give a detailed analysis of its solution. Then in Sec.~\ref{sec:comparison_scs_kac_dmft} we compare the results found via DMFT for $\alpha=0$ with those found via Kac-Rice in the various dynamical phases. Finally in Sec.~\ref{sec:scs_perspectives} we discuss open questions and future directions. Details on the computations are presented in Appendix.~\ref{app:scs_ann_compl}.\\

\noindent\textit{Acknowledgments}\\
I acknowledge collaboration and discussions with Samantha J. Fournier, Pierfrancesco Urbani, Valentina Ros, Alessia Annibale. A special thanks goes to Alessia for drawing our attention to this problem, to Valentina for her idea on the choice of non-linearity, and to Pierfrancesco for providing his solution to the DMFT.

\section{Introduction}
\label{sec:scs_introduction}
This chapter is a continuation of Chapter~\ref{chapter:non_reciprocal}, and therefore the motivations are similar. However, the model under study is different. We consider a prototypical model of randomly interacting neurons introduced by Sompolinsky, Crisanti and Sommers in 1988 \cite{ChaosSompo88}. In the original work the dynamics of these neurons is modeled through the following system of $N>>1$ coupled ODEs:
\begin{align}
\frac{dx_i}{dt}=-x_i+\sum_j \phi(g\,x_j)J_{ij}
\end{align}
where $\phi(u)=\tanh(u)$ is a non-linear gain function, $g$ measures the degree of non-linearity, and $J_{ij}\sim\mathcal{N}(0,J^2/N)$ are Gaussian i.i.d. random variables that model the interaction between neurons. In a biological context, these are the Kirchoff's laws of the nerve cells: $x_i$ represents the membrane potential; $\phi(g\,x_i)$ its firing rate and $J_{ij}$ the synaptic efficacy which couples the output
of the (presynaptic) $j$-th neuron to the input of the (postsynaptic) $i$-th neuron. In the original work and subsequent works \cite{crisanti_path_2018}, the authors use a Path Integral formalism to obtain the $N\to\infty$ DMFT equation for the autocorrelation function $\langle x(t)x(t+\tau)\rangle$ of a single representative unit. They are then able to obtain the maximal Lyapunov exponent as (related to) the ground state of a one-dimensional Shrödinger's equation. They show that in the limit $N\to\infty$ the system exhibits a sharp transition to chaos as the value $gJ$ crosses $1$ from below. \\

\noindent While the dynamics and eigenvalue/ Lyapunov spectra of these models have received considerable attention in recent years~\cite{crisanti_path_2018, Abbott_lyapunov_2023, MastrogiuseppeLink2018, AnnibaleDynamics2024, Sompo_transition_2015, HeliasMemory18, HeliasBook20, Rajan_spectra_06}, the topological complexity for this model is still considered an uncharted territory. To the best of our knowledge, the only two works that treat the complexity of equilibria are Refs.~\cite{wainrib2013topological, HeliasFP2022}. The first work, by Wainrib and Touboul \cite{wainrib2013topological} gives an estimate of the mean number of fixed points of the dynamics close to the transition, arguing that the topological complexity and the maximal Lyapunov exponent have the same behavior, both at the edge of chaos and far from it (the second one being more of a conjecture). The second paper \cite{HeliasFP2022} instead extends this analysis, leveraging tools from random matrix theory~\cite{sommers1988spectrum, Wei_spectra_2012, Miller_rmt_2015} to compute the mean number of total fixed points.\\

\noindent In this chapter we considerably extend upon Refs.~\cite{wainrib2013topological, HeliasFP2022} by both finding a way to get the full complexity curve for any $\alpha$, thus classifying fixed points in terms of their instability index, and extending the analysis in terms of low-rank perturbations~\cite{MastrogiuseppeLink2018, AnnibaleDynamics2024} and in terms of the asymmetry factor~\cite{AnnibaleDynamics2024} (henceforth $\alpha$), which tunes the degree of correlations between $J_{ij}$ and $J_{ji}$, as we will see below. A summary of contributions is also found in Sec.~\ref{summary:scs}.

\section{The model}
\label{sec:scs_the_model}
We consider a slightly different realization of the equations from the original paper~\cite{ChaosSompo88}, and we follow instead more recent works~\cite{HeliasBook20}, placing $g$ outside of $\phi$ for convenience. We thus have ${\bf x}\in\mathbb{R}^N$ interacting units, with \textit{all-to-all} interactions, defined by:
\begin{align}
\label{eq:scs_random_net}
    \frac{dx_i}{dt}=-x_i+g\sum_{j=1}^N\phi(x_j)J_{ij},\quad i=1,\ldots,N
\end{align}
with $g>0$, and where $\phi$ is a non-linear function such that $\phi(0)=0,\phi'(0)=1$ and $J$ is the interaction matrix, which contains the interaction between all units. The usual choice for the distribution of $J$ is given by Gaussian entries:
\begin{align}
    \mathbb{E}[J_{ij}]=\frac{J_0}{N},\quad\quad\text{Cov}[J_{ij},J_{kl}]=\frac{1}{N}\left(\delta_{ik}\delta_{jl}+\alpha\,\delta_{il}\delta_{jk}\right)
\end{align}
where we added the external perturbation $J_0$, and where $\alpha\in [0,1]$ tunes the degree of asymmetry between $J_{ij}$ and $J_{ji}$. If $\alpha=1$ the interactions are symmetric (i.e. $J_{ij}=J_{ji}$), and the system effectively represents a gradient descent in a landscape. However, if $\alpha<1$, the interactions are non-reciprocal, and in particular for $\alpha=0$ they are totally asymmetric (that is, independent). One may rewrite the system of equations in a more compact form:
\begin{align}
\frac{d{\bf x}}{dt}={\bf F}({\bf x})\equiv-{\bf x}+{\bf f}({\bf x})
\end{align}
where now ${\bf f}$ is a Gaussian field of statistics:
\begin{align}
    \mathbb{E}[{\bf f}({\bf x})]=gJ_0M_\phi({\bf x}),\quad\text{Cov}[f_i({\bf x}),f_j({\bf y})]=g^2\left[\delta_{ij}\left(\frac{\phi({\bf x})\cdot\phi({\bf y})}{N}\right)+\frac{\alpha}{N}\phi(x_i)\phi(y_j) \right]
\end{align}
with $M_\phi({\bf x}):=\frac{1}{N}\sum_{i=1}^N \phi(x_i)$.

\subsection{The effective dynamical equation}
\label{sec:scs_dmft_effective_equations}
In analogy to what has been done in Chapter~\ref{chapter:intro} and Chapter~\ref{chapter:non_reciprocal}, we can study the dynamics of this model in the limit of $N$ large. This has been done first in \cite{ChaosSompo88}, without the external term $J_0$. The addition of a generic low-rank perturbation has been considered in \cite{MastrogiuseppeLink2018} with $\alpha=0$, and in \cite{AnnibaleDynamics2024} with $J_0$ and arbitrary $\alpha$. In the following we will first reproduce the dynamical analysis of \eqref{eq:scs_random_net} with $\alpha=0$ and $J_0$, showing that the phase diagram is richer than what was initially obtained in \cite{AnnibaleDynamics2024} (in particular, there is also a ferromagnetic chaotic phase); then we will introduce a specific choice of $\phi$ that allows the computation of the complexity of fixed points of \eqref{eq:scs_random_net} via Kac-Rice, thus providing an explicit way of comparing dynamics and complexity for this model. \\\\
Following the steps in Appendix~\ref{app:dynamical_calculations}, with the form for ${\bf f}$ written above, we obtain an effective SDE (stochastic differential equation) that describes the evolution of a representative unit. In the general case this SDE reads:
\begin{align}
\partial_t x(t)=-x(t)+gJ_0M_\phi(t)+\alpha g^2\int_0^t ds\,R_\phi(t,s)\phi(x(s)) +\eta(t)
\end{align}
with $\eta$ zero-mean Gaussian noise with $\langle\eta(t)\eta(s)\rangle=g^2C_\phi(t,s)$, $\langle\cdot\rangle$ representing the average over the randomness of $\eta$, and where we defined
\begin{align}
M_\phi(t):=\langle \phi(x(t))\rangle,\quad C_\phi(t,s):=\langle\phi(x(t))\phi(x(s))\rangle,\quad R_\phi(t,s):=\left\langle\frac{\delta \phi(x(t))}{\delta\eta(s)}\right\rangle.
\end{align}

\section{Dynamical phase diagram for $\alpha=0$.}
\label{sec:scs_dmft_zero_gamma}
Let us derive here the phase diagram from the DMFT equations with random initial condition in the case $\alpha=0$; the derivation is similar to the one in Sec.~\ref{sec:rnn_dynamics}, but we repeat it since the equations are quite different. The effective single unit SDE in this case reads:
\begin{align}
\label{eq:def_effective_scs}
    \partial_t x(t)=-x(t)+gJ_0M_\phi(t)+\eta(t).
\end{align}
Notice that one can conveniently write this SDE as 
\begin{align*}
&e^t(\partial_t x(t)+x(t))=e^t[gJ_0M_\phi(t)+\eta(t)]\\
\Rightarrow&\partial_t(e^tx(t))=e^t[gJ_0M_\phi(t)+\eta(t)]\\
\Rightarrow&x(t)=x(0)e^{-t}+\int_0^tds\, e^{s-t}[gJ_0M_\phi(s)+\eta(s)]
\end{align*}
so we see that $x(t)$ is a Gaussian random variable with mean $m(t):=\langle x(t)\rangle$ and variance $\Delta(t,s):=C_x(t,s)-m(t)m(s)$, where $C_x(t,s):=\langle x(t)x(s)\rangle$.
Therefore, $x(t)$ can be written as $x(t)=m(t)+\sqrt{\Delta(t,t)}\,z$ with $z\sim\mathcal{N}(0,1)$. Notice that we can write an equation for $M_\phi$ from its definition:
\begin{align}
    M_\phi(t)=\langle\phi(x(t))\rangle=\int dz \,p(z)\,\phi\left(m(t)+z\sqrt{\Delta(t,t)}\right),\quad p(z)=\frac{e^{-z^2/2}}{\sqrt{2\pi}}.
\end{align}
Moreover, by averaging the SDE on both the LHS (left hand side) and RHS, we obtain:
\begin{align}
    \partial_tm(t)=-m(t)+gJ_0M_\phi(t).
\end{align}
The response function is defined as $R(t,t')=\langle \delta x(t)/\delta\eta(t')\rangle$, and therefore its equation can be derived by taking a functional derivative of the SDE:
\begin{align}
    \partial_t R(t,t')=-R(t,t')+\delta(t-t').
\end{align}
Lastly, the equation for $C_x$ reads (by multiplying the SDE by $x(t')$ and averaging):
\begin{align}
\label{eq:scs_eq_Cx_tt'}
\partial_t C_x(t,t')=-C_x(t,t')+gJ_0M_\phi(t)m(t')+g^2\int_0^{t'}ds\, R(t',s)C_\phi(t,s),
\end{align}
where $\langle \eta(t)x(t')\rangle$ is computed in Appendix~\ref{app:comp_eta_x_t_tp}.
Therefore, we can summarize the DMFT equations for the variables $m,M_\phi,C_x,C_\phi,\Delta,R$ 
as:
\begin{align}
\begin{cases}
\partial_t C_x(t,t')=-C_x(t,t')+gJ_0M_\phi(t)m(t')+g^2\int_0^{t'}ds\, R(t',s)C_\phi(t,s)\\
\partial_t R(t,t')=-R(t,t')+\delta(t-t')\\
\partial_tm(t)=-m(t)+gJ_0M_\phi(t)\\
M_\phi(t)=\langle\phi(x(t))\rangle\\
C_\phi(t,t')=\langle\phi(x(t))\phi(x(t'))\rangle\\
\Delta(t,t')=C_x(t,t')-m(t)m(t')
\end{cases}
\end{align}
where the averages are intended over $x$ (which is a Gaussian variable that depends on $m(t),m(t')$ and $C_x(t,t')$ at all times). These are therefore coupled integro-differential equations that can be solved self-consistently.
As we will see below, by multiplying Eq.~\eqref{eq:def_effective_scs} by itself at two different times, and then taking an average, we get an additional equation that will be useful in the TTI regime:
\begin{align}
\label{eq:2_deg_Cx}
    (1+\partial_t)(1+\partial_{t'})C_x(t,t')=g^2J_0^2M_\phi(t)M_\phi(t')+g^2C_\phi(t,t').
\end{align}

\subsection{TTI regime}
For the present analysis, we are interested in the time translation invariance (TTI) regime, as in the original work by Sompolinsky et al. \cite{ChaosSompo88, crisanti_path_2018}. As we will see shortly, the present system can exhibit two types of behavior in the $N\to\infty $ limit: either chaotic motion or a stable fixed point. While for $N$ finite periodic oscillations are possible, it has been shown in \cite{AnnibaleDynamics2024} that these do not happen for large $N$. In the TTI regime we have that $t=\tau+t'$ ($\tau$ fixed) with $t'\to\infty$, and:
\begin{align}
&m:=\lim_{t'\to\infty} m(t'),\quad M_\phi:=\lim_{t'\to\infty}M_\phi(t'),\quad C_x(\tau):=\lim_{t'\to\infty} C_x(t'+\tau,t') ,\\
&R(\tau):=\lim_{t'\to\infty} R(t'+\tau,t'),\quad C_\phi(\tau):=\lim_{t'\to\infty} C_\phi(t'+\tau,t'),\quad \Delta(\tau):=\lim_{t'\to\infty}\Delta(t'+\tau,t').
\end{align}
Notice that we do not write a subscript to indicate the TTI regime, so to avoid heavy notation. The TTI hypothesis directly implies that
\begin{align}
    m=gJ_0M_\phi, \quad R(\tau)=e^{-\tau} H(\tau)
\end{align}
with $H$ the Heaviside step function. This gives us from Eq.~\eqref{eq:scs_eq_Cx_tt'}:
\begin{align}
\partial_\tau C_x(\tau)=-C_x(\tau)+g^2J_0^2M_\phi^2+g^2\int_0^{\infty}ds\,e^{-s}C_\phi(\tau+s).
\end{align}
Let us conveniently denote $C_x^0:=C_x(0),\, C_x^\infty:=\lim_{\tau\to\infty}C_x(\tau)$ (and similarly for $C_\phi$ and $\Delta$), then by sending $\tau\to\infty$ in the previous equation we get:
\begin{align}
C_x^\infty=g^2J_0^2M_\phi^2+g^2C_\phi^\infty
\end{align}
since $\partial_\tau C(\tau)|_{\tau\to\infty}=0$ ($C(\tau)$ is a decreasing function bounded from below by 0). This can be written also as $\Delta_\infty=g^2C_\phi^\infty$. Additionally, by taking the TTI limit of Eq.~\eqref{eq:2_deg_Cx} we get
\begin{align}
\label{eq:scs_second_der_Cx}
\quad (1-\partial^2_\tau)C_x(\tau)=g^2J_0^2M_\phi^2+g^2C_\phi(\tau).
\end{align}
or in terms of $\Delta$:
\begin{align}
\label{eq:ddelta_=_C}
    (1-\partial_\tau^2)\Delta(\tau)=g^2C_\phi(\tau),
\end{align}
which corresponds to the equation originally obtained in  \cite{ChaosSompo88, crisanti_path_2018}. Let us abbreviate $\Delta=\Delta(\tau)$. We can moreover write the equation for $M_\phi$ as 
\begin{align}
    M_\phi=\int dz\,p(z)\,\phi\left(gJ_0M_\phi +z\sqrt{\Delta_0}\right).
\end{align}
We should now express the quantity $C_\phi(\tau)=\langle\phi(x(t))\phi(x(t+\tau))\rangle$ as an average over $x$; we abusively keep $t$ even if it should be sent to infinity, a more rigorous analysis should consider everything $t$ dependent, then taking the $t$ limit in the end. Recall that $x(t+\tau),x(t)$ are two correlated Gaussian random variables such that:
\begin{align}
    (x(t),x(t+\tau))\sim\mathcal{N}\left(\begin{pmatrix}
        m\\ m
    \end{pmatrix},\begin{pmatrix}
        \Delta_0 & \Delta\\
        \Delta & \Delta_0
    \end{pmatrix}\right).
\end{align}
Then these two Gaussian variables can be compactly written as
\begin{align}
&x(t)=m+z\sqrt{\Delta_0},\quad\quad z\sim\mathcal{N}(0,1)\\
&x(t+\tau)=m+z\frac{\Delta}{\sqrt{\Delta_0}}+\rho\sqrt{\Delta_0-\frac{\Delta^2}{\Delta_0}} ,\quad\quad\rho\sim\mathcal{N}(0,1).
\end{align}
which implies that $C_\phi(\tau)=h_\phi(\Delta,\Delta_0)$ for $h$ a function given by
\begin{align}
\label{eq:scs_h_phi}
h_\phi(\Delta,\Delta_0)=\int d\rho\,dz\,p(z)\,p(\rho)\,\phi\left(m+\sqrt{\Delta_0}z\right)\phi\left(m+z\frac{\Delta}{\sqrt{\Delta_0}}+\rho\sqrt{\Delta_0-\frac{\Delta^2}{\Delta_0}}\right).
\end{align}
Now, the main idea is that we can use Price's theorem (see \cite{HeliasBook20}, Sec. 10.5) which says that:
\begin{align}
    \frac{\partial h_\Phi(\Delta,\Delta_0)}{\partial \Delta}=h_{\Phi'}(\Delta,\Delta_0)
\end{align}
where $\Phi$ can be any function and in our particular case it is the integral of $\phi$: $\Phi(x)=\int_0^x \phi(y)dy$. This means that we can write
\begin{align}
C_\phi(\tau)= \frac{\partial h_\Phi(\Delta,\Delta_0)}{\partial \Delta}
\end{align}
and ultimately, using Eq.~\eqref{eq:ddelta_=_C} we obtain
\begin{align}
(1-\partial^2_\tau)\Delta=g^2\frac{\partial h_\Phi(\Delta,\Delta_0)}{\partial \Delta}\Rightarrow \partial^2_\tau\Delta=-\partial_\Delta V(\Delta,\Delta_0)
\end{align}
where we defined $V$ by 
\begin{align}
V(\Delta,\Delta_0)=-\frac{1}{2}\Delta^2+g^2h_\Phi(\Delta,\Delta_0)-g^2h_\Phi(0,\Delta_0)
\end{align}
the last term being subtracted arbitrarily to ensure that $V(0,\Delta_0)=0$. Like in Chapter~\ref{chapter:non_reciprocal}, we have that this is the equation of motion for a particle in a potential, and it is easy to verify that the following quantity ("energy") is conserved:
\begin{align}
E=\frac{1}{2}(\partial_\tau \Delta(\tau))^2+V(\Delta(\tau),\Delta_0).
\end{align}
The key observation at this point is that, since $\Delta(\tau)$ is symmetric (because of time translation invariance), we have that $\partial_\tau\Delta(\tau)|_{\tau=0}=\partial_\tau\Delta(\tau)|_{\tau\to\infty}=0$, thus leading (through the conservation of $E$) to the equality $V(\Delta_0,\Delta_0)=V(\Delta_\infty,\Delta_0)$. Please note that here the potential is implicitly defined through \eqref{eq:scs_h_phi}, at variance with Chapter~\ref{chapter:non_reciprocal}, where the potential was a simple polynomial of its variables. For this reason, the analysis in Chapter~\ref{chapter:non_reciprocal} was much easier. Now, we can summarize the DMFT equations in the TTI regime as
\begin{align}
\label{eq:dmft_system_eqns}
    \begin{cases}
    m=gJ_0M_\phi\\
    M_\phi=\int dz\,p(z)\,\phi\left(gJ_0M_\phi+z\sqrt{\Delta_0}\right)\\
    \Delta_{0/\infty}=C_x^{0/\infty}-g^2J_0^2M_\phi^2\\
     V(\Delta_0,\Delta_0)=V(\Delta_\infty,\Delta_0)\\
     \Delta_\infty=g^2C_\phi^\infty=g^2h_\phi(\Delta_\infty,\Delta_0)\\
     C_\phi^0=\int dz\,p(z)\phi^2\left(gJ_0M_\phi+z\sqrt{\Delta_0}\right)
    \end{cases}.
\end{align}
These are effectively only three equations on $\Delta_\infty,\Delta_0,M_\phi$ (the 2nd, 4th and 5th), but we included also $C_\phi,m,C_x$ to show how they are related to each other, since their role will be important when comparing the dynamical results with the Kac-Rice complexity. Below we will discuss the different solutions of this system of equations, and derive the corresponding phase diagram.

\subsubsection{Paramagnetic Fixed Point Phase (PFP)}
This solution corresponds to ${\bf x}={\bf0}$, and therefore $m=M_\phi=0$, and also to $\Delta_0=\Delta_\infty=0$. Since all parameters are equal to zero, and since $\phi(0)=0$, the DMFT equations are automatically verified. This solution is stable for $g<1$ and unstable when $g>1$. If $J_0=0$, the system is inevitably chaotic when $g>1$ \cite{ChaosSompo88}, however for $J_0$ big enough, we can have a stable FFP (ferromagnetic fixed point) even for $g>1$. Let us remark that to be precise about the transition to chaos, one should compute the maximal Lyapunov exponent and show that it becomes positive at the transition; an explicit computation of this fact is done in Ref.~\cite{crisanti_path_2018, HeliasBook20}

\subsubsection{Paramagnetic Chaotic Phase (PC)}
This regime corresponds to $m=M_\phi=0$ but $\Delta_\infty=0<\Delta_0$, that is, the system decorellates and is therefore not in a fixed point configuration. In this regime all equations are automatically satisfied except for $V(\Delta_0,\Delta_0)=V(0,\Delta_0)=0$, which is used to derive $\Delta_0$, which solves the following equation:
\begin{align}
\label{eq:delta0_para}
\begin{split}
\Delta_0^2&=2g^2h_\Phi(\Delta_0,\Delta_0)-2g^2h_\Phi(0,\Delta_0)\\
&=2g^2\int dz\,p(z)\Phi^2\left(z\sqrt{\Delta_0}\right)-2g^2\left[\int dz\,p(z)\Phi\left(z\sqrt{\Delta_0}\right) \right]^2.
\end{split}
\end{align}
In this regime we find therefore that the other variables read:
\begin{align}
C_x^0=\Delta_0,\quad \,C_x^\infty=C_\phi^\infty=0,\quad C_\phi^0=\int dz\,p(z)\phi^2(z\sqrt{\Delta_0}).
\end{align}
It is interesting to compute the limiting value of $\Delta_0$ and $C_\phi^0$ as $g\to\infty$. Indeed, an asymptotic analysis of the equation above with $\phi(x)=\tanh(x)$, $\Phi(x)=\log\cosh(x)$ reveals that 
\begin{align}
\lim_{g\to\infty}\frac{\Delta_0}{g^2}=\frac{2\pi-4}{\pi},\quad \lim_{g\to\infty}C_\phi^0=1.
\end{align}

\subsubsection{Ferromagnetic Fixed Point Phase (FFP)}
This phase corresponds to a stable fixed point (therefore $\Delta_0=\Delta_\infty$) but with ferromagnetic properties (that is, $m\neq 0$). By plugging $\Delta_0=\Delta_\infty$ in the DMFT equations we find the self-consistent set of equations:
\begin{align}
\label{eq:M_D_ffp}
\begin{cases}
M_\phi=\int dz\, p(z)\phi\left(gJ_0M_\phi+z\sqrt{\Delta_0}\right)\\
\Delta_0=g^2\int dz\, p(z)\phi^2\left(gJ_0M_\phi+z\sqrt{\Delta_0}\right).
\end{cases}
\end{align}
Once these equations are solved numerically, we can find the remaining quantities of interest as:
\begin{align}
m=gJ_0M_\phi,\quad C_x^0=C_x^\infty=\Delta_0+g^2J_0^2M_\phi^2,\quad C_\phi^\infty=C_\phi^0=\Delta_0/g^2.
\end{align}

\subsubsection{Ferromagnetic Chaotic Phase (FC)}
In this phase we have an unstable (chaotic) motion with $\Delta_0>\Delta_\infty>0$ and $M_\phi\neq 0$. The order parameters in this case are given by the following self-consistent equations:
\begin{align}
\label{eq:fc_scs_dmft_eqns}
\begin{cases}
M_\phi=\int dz\,p(z)\,\phi\left(gJ_0M_\phi+z\sqrt{\Delta_0}\right)\\
V(\Delta_0,\Delta_0)=V(\Delta_\infty,\Delta_0)\\
\Delta_\infty=g^2h_\phi(\Delta_\infty,\Delta_0)
\end{cases}
\end{align}
the remaining order parameters are found by means of the equations in ~\eqref{eq:dmft_system_eqns}. To be more rigorous, we should extend here the computation of the maximal Lyapunov, to show that the motion is chaotic. This question will be tackled in a future work. 

\subsubsection{The force}
It is interesting to compute the value of the quantity $(d{\bf x}/dt)^2$ using the effective single site equation \eqref{eq:def_effective_scs} and the TTI analysis presented above. In particular if the value of this quantity at large times is non-zero, it means that the dynamics is well separated from the stationary points. Using the effective equation \eqref{eq:def_effective_scs} we define:
\begin{align}
\begin{split}
\Gamma(t,t'):=\langle \partial_tx(t)\partial_{t'}x(t')\rangle
\end{split}
\end{align}
and by taking the TTI limit we find, from Eq.~\eqref{eq:scs_second_der_Cx}:
\begin{align}
\hat{\Gamma}(\tau):=\lim_{t'\to\infty}\Gamma(t'+\tau,t')=-\partial^2_\tau C_x(\tau)=-C_x(\tau)+g^2J_0^2M_\phi^2+g^2C_\phi(\tau).
\end{align}
Then the average force at equal times takes the value
\begin{align}
    \hat{\Gamma}_0:=\hat{\Gamma}(0)=-C_x^0+g^2J_0^2M_\phi^2+g^2C_\phi^0=g^2C_\phi^0-\Delta_0
\end{align}
and as $\tau\to\infty$ we get instead:
\begin{align}
    \hat{\Gamma}^\infty=0
\end{align}
which implies that in the long time difference limit there is no correlation between the driving forces. We will see below that the force at equal times in the long time limit is non-zero in all phases, except the FFP and PFP phases, where we know the system settles in an equilibrium fixed point. \\\\

\noindent We will now derive the dynamical phase diagram and plot it for the choice $\phi(x)=\tanh(x)$. The specific analysis and resolution of the DMFT equations found in the various phases is deferred to Sec.~\ref{sec:comparison_scs_kac_dmft}, where we make an explicit choice for $\phi$ that is more convenient for comparison with the Kac-Rice complexity. 

\subsection{Dynamical transitions and phase diagram}
In this section we find the lines that separate the various phases, and numerically solve them to obtain the dynamical phase diagram.

\begin{figure*}[t!]
    \centering
    \includegraphics[width=0.7\textwidth]{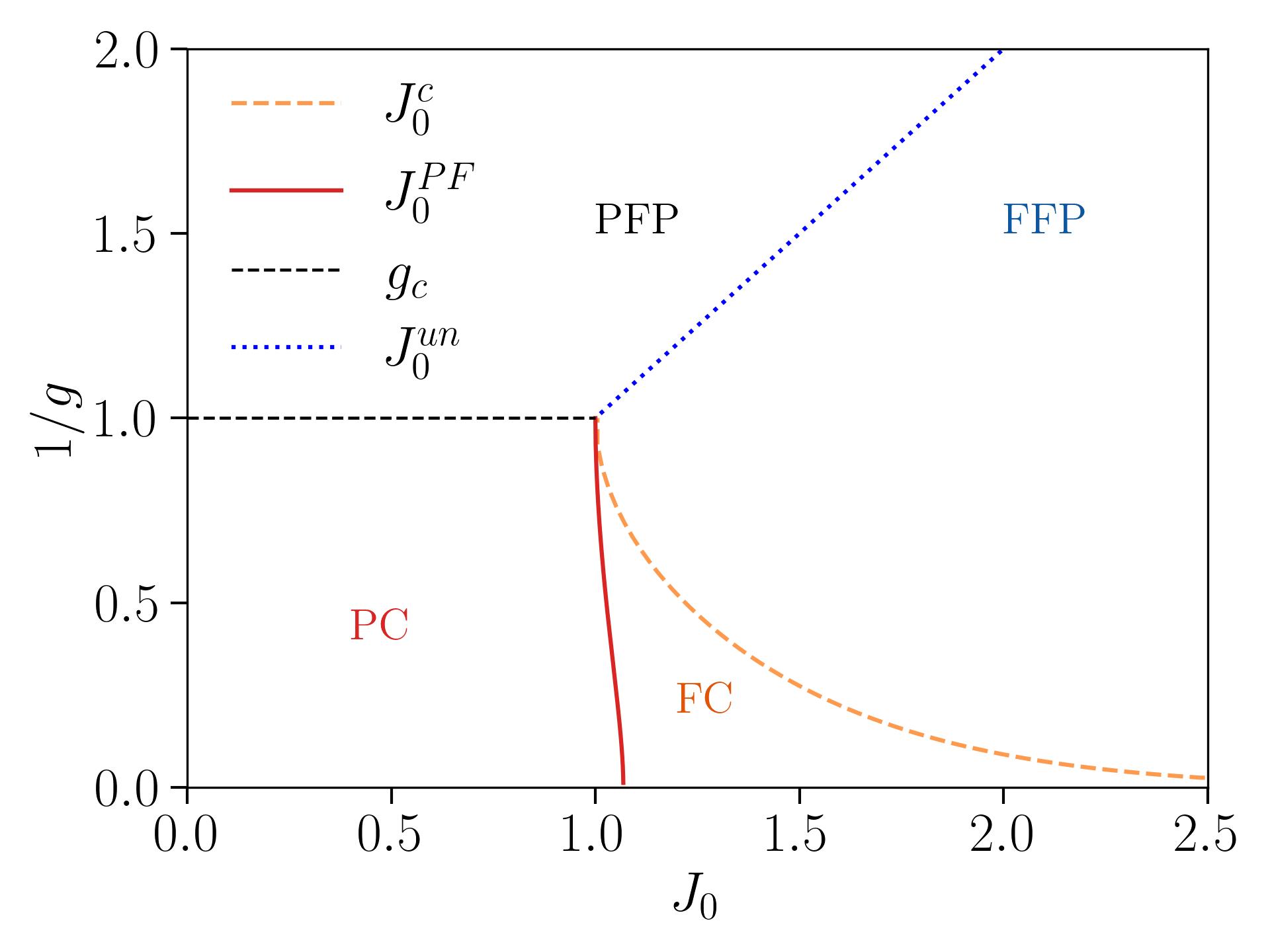}
    \caption{Dynamical phase diagram of the SCS model with $\phi(x)=\tanh(x)$ and where the connectivity (or interaction) matrix has "excitatory" columns, that is, every element is Gaussian with mean $J_0\geq 0$.  The four different phases are divided by transition lines that are marked in the legend: $g_c$ stands for the "critical" $g$ above which the PFP phase becomes PC; $J_0^{un}$ stands for "unstable", that is, the value that renders the PFP unstable by increasing $J_0$, thus leading to the FFP phase; $J_0^{PF}$ stands for "para-to-ferro", meaning that the PC phase develops a bifurcation in $M_\phi$, thus becoming ferromagnetic. }
    \label{fig:dmft_scs_phase}
\end{figure*}

\subsubsection{PFP to PC transition}
We have seen that this transition occurs when $g=1$ (above 1 we are PC and below 1 we are PFP). This transition is called type-II \cite{AnnibaleDynamics2024, Galla_random_LV_2023}, as it happens when the edge the Jacobian touches the origin (and thus an extensive fraction of unstable modes starts to appear), see below.

\subsubsection{PFP to FFP transition}
This line is found by linear stability analysis of the equation of motion itself. We start from the solution ${\bf x}={\bf 0}$ in the PFP phase, and we perturb the system by an amount $\Delta{\bf x}$. The equation of motion then becomes
\begin{align}
    d\Delta{\bf x}/dt=\Delta{\bf x}\left(-\mathbb{I}+\partial{\bf f}({\bf 0})\right)
\end{align}
where $[\partial{\bf f}({\bf 0})]_{ij}=gJ_{ij}$ since $\phi'(0)=1$ by definition. For this equation to be stable, all eigenvalues of the matrix $-\delta_{ij}+gJ_{ij}=-\delta_{ij}+g\left(J_0\frac{1}{N}+\tilde{J}_{ij}\right)$ have to be negative, with $\mathbb{E}\tilde{J}_{ij}=0$ and $\mathbb{E}[\tilde{J}_{ij}\tilde{J}_{kl}]=\frac{1}{N}\delta_{ik}\delta_{jl}$. The random matrix $\tilde{J}$ belongs to the real Ginibre ensemble \cite{Ginibre_1965}, and for $N\to\infty$ its spectrum satisfies Girko's circular law, that is, it converges to a uniform disk of unit radius \cite{Ginibre_1965, Girko_1985, mehta_rmt} in the complex plane. If $J_0>1$ a real outlier appears at $J_0$ \cite{tao_circularlaw}. Here we limit ourselves to observe that the largest eigenvalue will lie at the rightmost edge of the disk as long as $0\leq J_0\leq 1$, and it will lie at $J_0$ as soon as $J_0>1$. For the PFP solution to be stable we get in the first case that $g\leq 1$ and in the second case that $J_0\leq 1/g$. The second transition (from the PFP to FFP phase) is also referred to as type-I \cite{AnnibaleDynamics2024}, since only one isolated eigenvalue becomes unstable.

\subsubsection{FFP to FC transition}
This transition line is obtained starting from the FFP phase and then looking at when the solution becomes unstable, as $J_0$ is lowered. Hence the transition happens when $\partial^2_\Delta V(\Delta,\Delta_0)|_{\Delta=\Delta_0}=0$ from the FFP phase. This implies that the transition line in the $(g,J_0)$ plane is given by the equation
\begin{align}
1=g^2h_{\phi'}(\Delta_0,\Delta_0)=g^2\int dz\,p(z)\left[\phi'\left(gJ_0M_\phi+z\sqrt{\Delta_0}\right)\right]^2
\end{align}
where $M_\phi$ and $\Delta_0$ are given by Eq.~\eqref{eq:M_D_ffp}. 
This transition is also referred to as type-II transition \cite{AnnibaleDynamics2024}, since it is given by the fact that the maximum eigenvalue of the bulk of the spectrum of the fixed point crosses the origin, giving rise to an extensive number of unstable directions.

\subsubsection{FC to PC transition}
This transition is better studied from the PC region: as we increase $J_0$ the magnetization $M_\phi$ develops a bifurcation, thus rendering the system ferromagnetic chaotic with $M_\phi\neq 0$. The point of bifurcation (in the $g,J_0$ phase diagram) is easily found as the point where the solution $M_\phi=0$ is unstable, which is obtained when $J_0$ (at fixed $g$) is such that 
\begin{align}
    1=gJ_0\int dz\,p(z)\phi'\left(z\sqrt{\Delta_0}\right)\Rightarrow J_0=\frac{1}{g\int dz\,p(z)\,\phi'(z\sqrt{\Delta_0})},
\end{align}
obtained from the equation of $M_\phi$ in the FC phase by differentiating and then setting $M_\phi=0$. In this equation, $\Delta_0$ is found from Eq.~\eqref{eq:delta0_para}. \\\\

\noindent The dynamical phase diagram is presented in Fig.~\ref{fig:dmft_scs_phase} for the choice $\phi(x)=\tanh(x)$.

\section{Topological complexity}
\label{sec:scs_topo_compl}
To compute the topological complexity we keep $\alpha$ and we make a particular choice of $\phi$, i.e. we choose
\begin{align}
\label{eq:def_phi_sign}
    \phi(x)=\begin{cases}
        x\quad\quad\quad\,\,\,\text{if}\quad{|x|\leq 1}\\
        \text{sign}(x)\quad \text{if}\quad{|x|> 1}\\
    \end{cases}
\end{align}
which approximates (see Fig.~\ref{fig:tanh_vs_phi}) the function $\tanh$ (normally used in these problems \cite{HeliasBook20, ChaosSompo88,HeliasMemory18, Sompo_transition_2015}), while maintaining the same main characteristics, that is, $\phi(x)=-\phi(-x)$, $\phi(0)=0,\phi'(0)=1$ and $\lim_{x\to\pm\infty}\phi(x)=\pm 1$. Indeed, as we will see below, the dynamical phase diagram looks very similar to the one for $\tanh$. The reason why we prefer this choice for the function $\phi$ is that it makes possible an explicit computation of the complexity. This is essentially because the stability of the stationary points that we count is given by the spectral radius of the Jacobian matrix:
\begin{align}\label{eq:Hessian}
\partial{\bf F}({\bf x}):=\partial {\bf F}/\partial {\bf x}=-\mathbb{I}+\frac{g}{N}J_0D({\bf x})+g\tilde{\mathbf{J}}D({\bf x})
\end{align}
with $D({\bf x})=\text{diag}(\phi'(x_1),\ldots,\phi'(x_N))$ and $\tilde{\bf J}$ a Gaussian random matrix with statistics given by 
\begin{align}
\mathbb{E}[\tilde J_{ij}]=0,\quad 
\mathbb{E}[\tilde{J}_{ij}\tilde{J}_{kl}]=\frac{1}{N}(\delta_{ik}\delta_{jl}+\alpha\,\delta_{il}\delta_{jk}).
\end{align}
The spectrum of the matrix $\tilde{\bf J}$ for large $N$ belongs to the elliptic law \cite{sommers1988spectrum}, that is, its spectrum is uniform within an ellipse in the complex plane; however, an explicit formula for the spectrum of $\partial{\bf F}$ cannot be found, if not perturbatively in $\alpha$, see Refs.~\cite{baron_fine_structure_2024, baron_path_2022}. For $\alpha=0$, in particular, the radius $r({\bf x})$ of $g\,\tilde{{\bf J}}D({\bf x})$ in the limit of $N$ large has been computed in Refs.~\cite{Rajan_spectra_06, Wei_spectra_2012} and reads:
\begin{align}
r^2({\bf x})=g^2 \,\frac{1}{N}\text{tr}\,D^2({\bf x}).
\end{align}
\begin{figure*}[t!]
    \centering
    \includegraphics[width=0.6\textwidth]{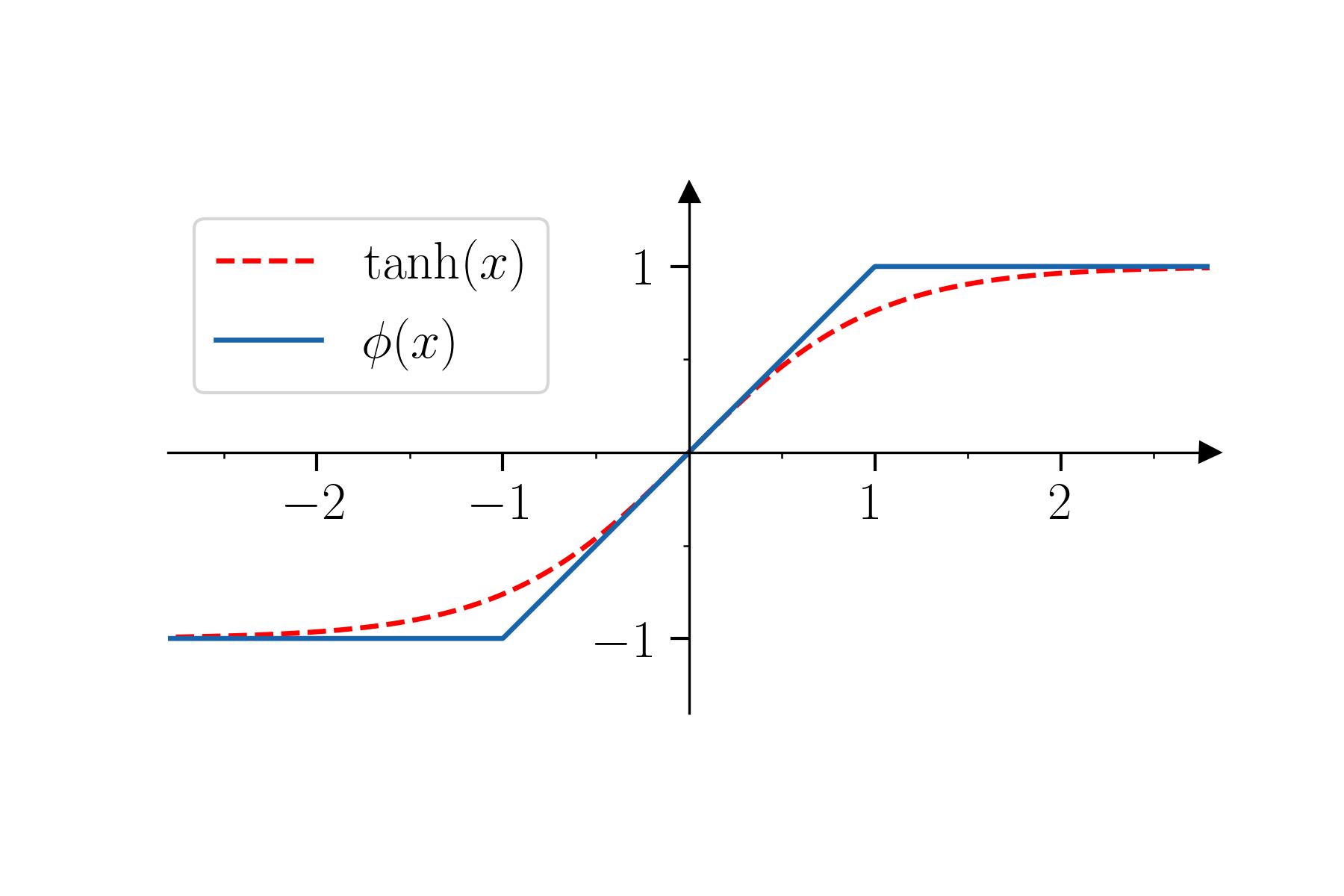}
    \caption{Visualization of the function $\tanh(x)$ versus our choice for $\phi(x)$.}
    \label{fig:tanh_vs_phi}
\end{figure*}

\noindent We see that this is not practical for analytical computations; moreover, the spectrum for $\alpha=0$ does not present an explicit form in general \cite{Wei_spectra_2012}. However, for our choice of $\phi$ this expression drastically simplifies, since 
\begin{align}
    \phi'(x)=\begin{cases}
    1 \quad\quad\text{if}\quad|x|\leq1\\
    0 \quad\quad\text{else}
    \end{cases}
\end{align}
so that $[\phi'(x)]^2=\phi'(x)$ and we can introduce an order parameter, $D_\phi$, that counts the number of ones in $\phi'({\bf x})$:
\begin{equation}
    D_\phi({\bf x}):=\frac{1}{N}\sum_i \phi'(x_i)= \frac{1}{N}\sum_i \mathbb{I}_{|x_i| \leq 1}
\end{equation}
and such that
\begin{align}
    r^2({\bf x})=g^2D_\phi({\bf x}).
\end{align}
 For this choice of $\phi$ the spectrum of ${\bf M}:=g\,\tilde{\bf J} D({\bf x})$ has a well-known form (elliptical law) as $N\to\infty$ \cite{sommers1988spectrum} (see also Ref.~\cite{Burda_non_herm_2011} for a simpler derivation):
\begin{align*}
\rho(x,y)=
\begin{cases}
\frac{1}{\pi (1-\alpha^2) g^2 D_\phi}\quad\quad\quad\text{if}\quad\quad \frac{(x+1)^2}{(1+\alpha)^2}+\frac{y^2}{(1-\alpha)^2}=g^2D_\phi\\
0\quad\quad\text{else}
\end{cases}.
\end{align*}
A small remark: formally what $D({\bf x})$ does is that it deletes those columns for which $\phi'(x_i)=0$. By shuffling the columns this results in a block diagonal matrix that has an elliptical law (with $N$ reduced by $D_\phi$) in the first block and zeros in the block at the bottom. Since ultimately the matrix is shifted by the identity, this does not cause troubles with the determinant. \\

\noindent Like in Chapter~\ref{chapter:non_reciprocal} we compute the complexity at some fixed order parameters. The reason for introducing the order parameters is to extract the minimal number of expressions that, once fixed, leave us with only Gaussian integrals to compute. At the same time, we do not want to introduce unnecessary order parameters. For our present analysis, it turns out that we need the following six order parameters:
\begin{align*}
 &m= \frac{1}{N} \sum_i x_i \quad&q &= \frac{1}{N} \sum_i x^2_i\\
&M_\phi =\frac{1}{N} \sum_i \phi(x_i) \quad&Q_\phi &=\frac{1}{N} \sum_i \phi^2(x_i)\\
&Z= \frac{1}{N} \sum_i x_i\, \phi(x_i)\quad 
&D_\phi&=\frac{1}{N}\sum_i \phi'(x_i)
\end{align*}
where, for notational simplicity, we avoid the argument ${\bf x}$.
\noindent Therefore the (quenched) topological complexity is defined as:
\begin{align}
\Sigma(m,q,M_\phi,Q_\phi,D_\phi,Z):=\lim_{N\to\infty}\frac{1}{N}\mathbb{E}\left[\log \mathcal{N}(m,q,M_\phi,Q_\phi,D_\phi,Z)\right]
\end{align}
with $\mathcal{N}$ the number of fixed points of the dynamical equation that satisfy the order parameters, given via Kac-Rice by:
\begin{align}
\mathcal{N}(m,q,M_\phi,Q_\phi,D_\phi,Z):=\int_{\mathbb{R}^N}d{\bf x}\,\Omega({\bf x})\delta({\bf F}({\bf x}))\,|\det\partial{\bf F}({\bf x})|
\end{align}
where we defined
\begin{align}
\label{eq:scs_constraints}
\begin{split}
    \Omega({\bf x})&:=
    \delta\left(m-\frac{1}{N}\sum_ix_i\right)
    \delta\left(q-\frac{1}{N}\sum_ix_i^2\right)\delta\left(D_\phi-\frac{1}{N}\sum_i \phi'(x_i)\right)\times\\
    &
    \times\delta\left(M_\phi-\frac{1}{N}\sum_i\phi(x_i)\right)
    \delta\left(Q_\phi-\frac{1}{N}\sum_i\phi^2(x_i)\right)
    \delta\left(Z-\frac{1}{N}\sum_ix_i\phi(x_i)\right).
\end{split}
\end{align}
From now on, for simplicity, we avoid to write all the order parameters as function arguments, and instead we use $\Omega$ and  write $\mathcal{N}(\Omega),\Sigma(\Omega)$.

\section{Annealed complexity}
\label{sec:scs_ann_topo_compl}
The annealed complexity is easier to compute, since we interchange the logarithm and the expectation:
\begin{align}
    \Sigma_A(\Omega):=\lim_{N\to\infty}\frac{1}{N}\log\mathbb{E}\,\left[\mathcal{N}(\Omega)\right].
\end{align}
The expected value of the number of fixed points is then easier to compute with respect to the quenched computation:
\begin{equation}\label{eq:FirstMomentmq}
   \mathbb{E}[\mathcal{N}(\Omega)]= \int_{\mathbb{R}^N} d{\bf x} \,
    \Omega({\bf x}) \mathbb{E}[\delta({\bf F}({\bf x}))] \, \mathbb{E} \left[ |\mathrm{det} \,  \partial {\bf F}({\bf x})| \, \Big| {\bf F}({\bf x})=0\right]
\end{equation}
Like before in Chapter~\ref{chapter:non_reciprocal}, this expected value can be decomposed in three parts: the phase space term $v$, the probability term $P_\alpha$, and the determinant term $d_\alpha$. The details of the computation are done in Appendix~\ref{app:scs_ann_compl}, and here we simply present the main result:
\begin{align}
    \Sigma_A(\Omega)=p_\alpha(M_\phi,Q_\phi,Z) + d_\alpha(D_\phi)+\text{extr}_{\hat{\lambda},\hat{\omega},\hat{\xi},\hat{\theta},\hat{\eta},\hat{t}}\left\{ v(\Omega,\Gamma)\right\}
\end{align}
where $\Gamma:=\{\hat\lambda,\hat\omega,\hat\xi,\hat\theta,\hat\eta,\hat t\}$ is the set of conjugate parameters that appears when opening up in Fourier the Dirac deltas that impose the various constraints. The terms appearing above read:
\begin{align}
\label{eq:scs_p_alpha}
\begin{split}
    p_\alpha(M_\phi, Q_\phi, Z)=&-\frac{1}{2} \log (2 \pi g^2 Q_\phi)-\frac{1}{2}\frac{1}{g^2 Q_\phi} \left( q-2 g J_0 M_\phi \, m+ g^2 J_0^2 M_\phi^2 \right)\\
    &+\frac{1}{2}\frac{\alpha}{g^2 Q_\phi^2 (1+\alpha)}\left( Z^2 - 2 g J_0 \, M^2_\phi Z + g^2 J_0^2 M_\phi^4 \right)
\end{split}
\end{align}
and 
\begin{align}
\label{eq:scs_d_alpha}
 &d_\alpha(D_\phi)=\\
 &=
    \begin{cases}
    \frac{1}{2 g^2}\frac{1}{1+\alpha}-\frac{D_\phi}{2}+\frac{D_\phi}{2}\log(g^2D_\phi)\quad\quad\text{if}\quad g\sqrt{D_\phi}(1+\alpha)>1\\
        \frac{1}{4\alpha g^2}(1-\sqrt{1-4\alpha g^2D_\phi})+D_\phi\log(1+\sqrt{1-4\alpha g^2D_\phi})-D_\phi(\frac{1}{2}+\log(2))\quad\text{else}
    \end{cases}
\end{align}
and 
\begin{align}
\label{eq:scs_v}
v(\Omega,\Gamma)=\hat{\lambda} Q_\phi+\hat{\omega} q + \hat{\eta} m + \hat{\xi} M_\phi+\hat{t}(D_\phi-1) + \hat{\theta}Z+\log(R) 
\end{align}
with $R=\Theta_1+\Theta_2+\Theta_3$ and:
\begin{align}
\label{eq:scs_thetas}
\begin{split}
&\Theta_1=\frac{\sqrt{\pi}e^{\frac{b'^2}{4a'}}}{2\sqrt{a'}}\left[\text{erf}\left(\frac{b'+2a'}{2\sqrt{a'}}\right)+\text{erf}\left(\frac{2a'-b'}{2\sqrt{a'}}\right)\right]\\
&\Theta_2=\frac{\sqrt{\pi}e^{\frac{b^2}{4a}+c+t}}{2\sqrt{a}}\left[-\text{erf}\left(\frac{2a-b}{2\sqrt{a}}\right)+1  \right]\\
&\Theta_3=\frac{\sqrt{\pi}e^{\frac{b''^2}{4a}+c''+t}}{2\sqrt{a}}\left[-\text{erf}\left(\frac{2a+b''}{2\sqrt{a}}\right)+1  \right]
\end{split}
\end{align}
and the additional factors taking the following values:

\begin{align}
\begin{alignedat}{4}
a   &= \hat\omega                       &\qquad
b   &= \hat\eta-\hat\theta \\[4pt]
c   &= \hat\xi-\hat\lambda                  &\qquad
a'  &= \hat\lambda+\hat\omega+\hat\theta \\[4pt]
b'  &= \hat\eta+\hat\xi                     &\qquad
c'' &= -(\hat\xi+\hat\lambda) \quad\quad b''=\hat{\eta}+\hat{\theta}.
\end{alignedat}
\end{align}

\noindent 
Since ultimately we want to be able and plot the complexity as a function of $D_\phi$ (the parameter that controls the fraction of unstable modes of the stationary points), we will optimize over the rest of the order parameters, keeping track of the values of these at the optimum. The goal is then to compare such values with the corresponding values given by the DMFT. We can therefore define the complexity just as a function of $D_\phi$, while optimizing for the other order parameters:
\begin{align}
\Sigma_A(D_\phi):=\text{extr}_{\tilde{\Omega},\Gamma}\Sigma_A(\Omega,\Gamma)=\text{extr}_{\tilde{\Omega},\Gamma}\left\{p_\alpha(M_\phi,Q_\phi,Z)+d_\alpha(D_\phi)+v(\Omega,\Gamma)\right\}
\end{align}
where we use the notation $\tilde{\Omega}:=\Omega/D_\phi=\{m,q,M_\phi,Q_\phi,Z\}$, and also for $\Sigma_A$ we keep the same notation, just changing the arguments depending on what we are optimizing over. Therefore, the object that we really are interested in studying is $\Sigma_A(\Omega,\Gamma)$, which is obtained by performing a saddle point of the action resulting from the computation of $\mathbb{E}\mathcal{N}$ after the conjugate parameters have been introduced. See Appendix~\ref{app:scs_ann_compl} for details on the calculation.

\section{Analysis of the annealed complexity}
\label{sec:scs_ann_analysis_topo_compl}
In order to compute $\Sigma_A(D_\phi)$ we have to solve a set of 11 self-consistent equations:
\begin{align}
    \nabla_{\tilde{\Omega}}\Sigma_A(\Omega,\Gamma)={\bf 0},\quad\nabla_{\Gamma}\Sigma_A(\Omega,\Gamma)={\bf 0}.
\end{align}
The optimization over the $\tilde{\Omega}$ variables is easy and gives the following equations:
\begin{align}
&\hat\eta+\frac{J_0 M_\phi}{g Q_\phi}=0\\
&\hat\omega-\frac{1}{2 g^2 Q_\phi}=0\\
& \hat\xi+\frac{J_0 \left( -2 M_{\phi} Z \alpha + m Q_{\phi} (1 + \alpha) - g J_0 M_{\phi} (Q_{\phi} - 2 M_{\phi}^2 \alpha + Q_{\phi} \alpha) \right)}{g Q_{\phi}^2 (1 + \alpha)} =0\\
&\hat\lambda+\frac{-2 g J_0 m M_{\phi} + g^2 J_0^2 M_{\phi}^2 + q}{2 g^2 Q_{\phi}^2} - \frac{1}{2 Q_{\phi}} - \frac{(-g J_0 M_{\phi}^2 + Z)^2 \alpha}{g^2 Q_{\phi}^3 (1 + \alpha)}=0\\
&\hat\theta+\frac{(-g J_0 M_{\phi}^2 + Z) \alpha}{g^2 Q_{\phi}^2 (1 + \alpha)} =0.
\end{align}
The remaining equations are obtained by taking derivatives with respect to the $\Gamma$ variables, and therefore involve the non-trivial term $R$ (which depends on the $\Gamma$ variables but not on the $\tilde{\Omega}$ ones):
\begin{align}
&m+\frac{\partial_{\hat\eta} R}{R}=0\\
&q+\frac{\partial_{\hat\omega} R}{R}=0\\
&M_\phi+\frac{\partial_{\hat\xi} R}{R}=0\\
&Q_\phi+\frac{\partial_{\hat\lambda} R}{R}=0\\
&Z+\frac{\partial_{\hat\theta} R}{R}=0\\
&D_\phi-1+\frac{\partial_{\hat t} R}{R}=0.
\end{align}
We avoid presenting here the explicit expression for these last equations, since it is rather cumbersome, and we directly pass to the numerical solution.

\subsection{Case $\alpha=0$}
Let us consider the case $\alpha=0$, which is the one we were able to solve with the DMFT. In this case we lose the dependence on $Z$, and $\hat{\theta}=0$. We are therefore left with 9 equations. Remark that the determinant term in this case simplifies to:
\begin{align}
\label{eq:d_0(Dphi)}
d_0(D_\phi)=\begin{cases}
0 \quad\text{if}\quad D_\phi<1/g^2\\
 \frac{1}{2 g^2} - \frac{D_\phi}{2} + \frac{D_\phi}{2}\log(g^2 D_\phi)\quad\text{else}.
\end{cases}
\end{align}
By solving the 9 equations numerically, we see that, similarly to the problem presented in Chapter~\ref{chapter:non_reciprocal}, we have three types of complexity: a paramagnetic and a ferromagnetic complexity of unstable fixed points, and a ferromagnetic stable fixed point. Let us remark that the transition from the PFP phase to the FFP phase is solely based on a linear stability analysis, and does not need to rely on the Kac-Rice formalism. In particular, in the PFP phase the complexity diverges because $q=Q_\phi=0$. However, as we see below, we can predict the transition from the chaotic PC phase to the PFP phase by studying the paramagnetic complexity and sending $g\to 1^+$. Let us start indeed from the paramagnetic complexity.

\subsubsection{Paramagnetic complexity}
\begin{figure*}[t!]
    \centering
    \includegraphics[width=0.72\textwidth]{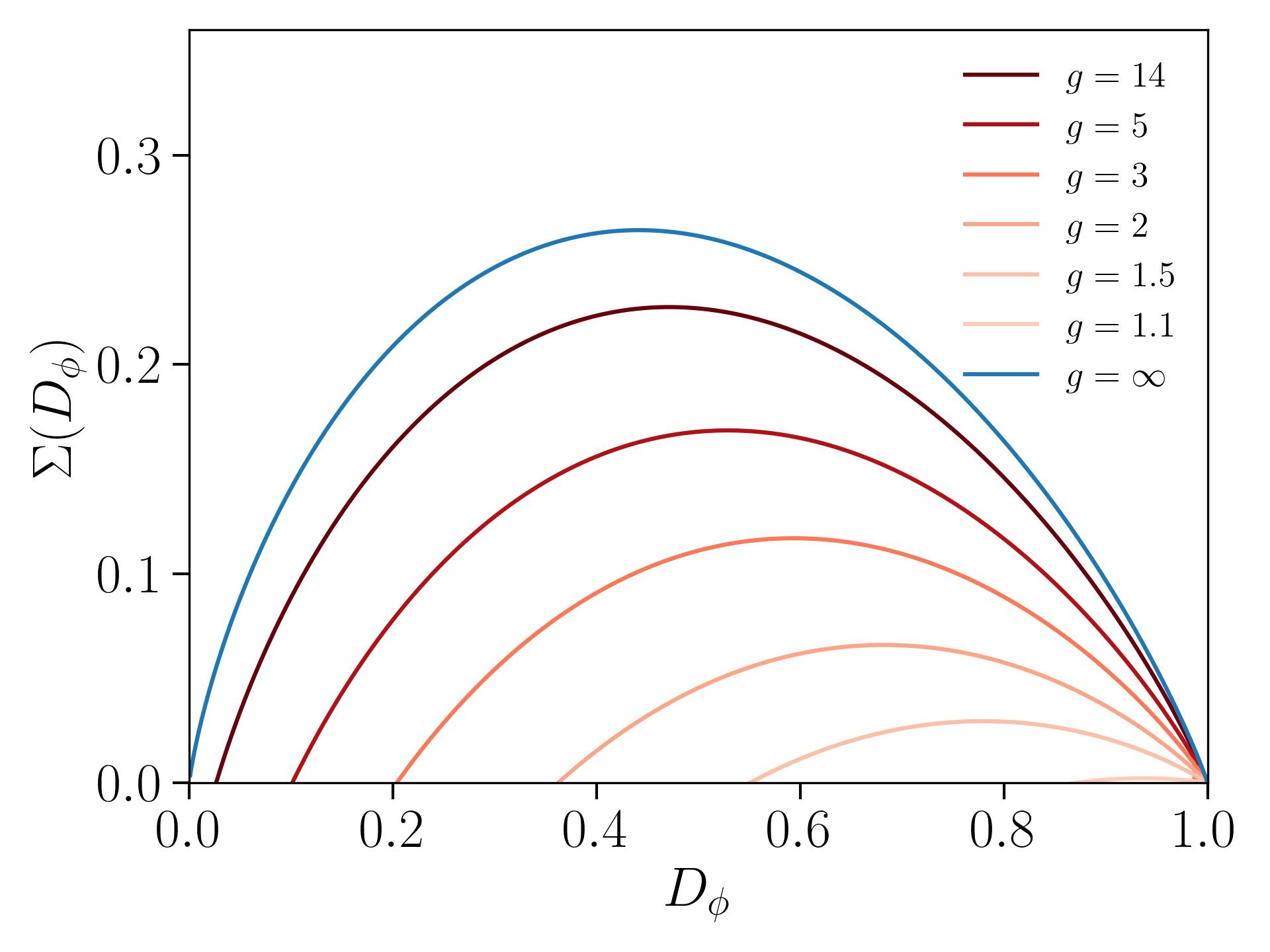}
    \caption{Plot of the paramagnetic (i.e. $m=M_\phi=0$) complexity of unstable stationary points as a function of $D_\phi$. As we decrease $g$, we see that the curve collapses to a point at $D_\phi=1$, indicating a transition to one single paramagnetic stable fixed point. The blue curve represents the complexity in the limit of $g\to\infty$}
    \label{fig:para_kac_scs}
\end{figure*}
The paramagnetic complexity corresponds to studying solutions to the equations above with $m=M_\phi=0$ (in which case the $J_0$ dependence is lost). Using the equations for $\hat{\eta}$ and $\hat{\xi}$ we see that we must have $\hat{\eta}=\hat\xi=0$. In this case the equations greatly simplify. The equation for $\hat{t}$ reads:
\begin{align}
e^{\hat t}=
\frac{(D_\phi - 1)\,e^{\hat\lambda}\,\sqrt{\hat\omega}\,
      \text{erf}\left(\sqrt{\hat\omega+\hat\lambda}\right)}
     {D_\phi\,\sqrt{\hat\omega+\hat\lambda\,}\,
      \bigl(\text{erf}\bigl(\sqrt{\hat\omega}\bigr)-1\bigr)}
\end{align}
and the remaining equations read:
\begin{align}
\begin{split}
&\hat{w}=\frac{1}{2 g^2 Q_\phi},\quad \hat\lambda=\frac{1}{2 Q_\phi}  -\frac{q}{2 g^2 Q_\phi^2} \\
&q=
\frac{\hat\omega+\hat\lambda-D_{\phi}\,\hat\lambda}
       {2\,\hat\omega\,(\hat\omega+\hat\lambda)}
-\frac{e^{-\hat\omega}}{\sqrt{\pi}}
  \left[
    \frac{D_{\phi}\,e^{-\hat\lambda}}
         {\sqrt{\hat\omega+\hat\lambda}\,
    \operatorname{erf}\bigl(\sqrt{\hat\omega+\hat\lambda}\bigr)}
    +\frac{D_{\phi}-1}
         {\sqrt{\hat\omega}\,
          \operatorname{erfc}\bigl(\sqrt{\hat\omega}\bigr)}
  \right]\\
&Q_{\phi}=1
- D_{\phi}
+\frac{D_{\phi}}{2\bigl(\hat\omega+\hat\lambda\bigr)}
-\frac{D_{\phi}\,e^{-(\hat\omega+\hat\lambda)}}
{\sqrt{\pi}\,\sqrt{\hat\omega+\hat\lambda}\,
\operatorname{erf}\left(\sqrt{\hat\omega+\hat\lambda}\right)}.
\end{split}
\end{align}
We will denote with a superscript $*$ the solution of the order parameters at the saddle point (that is, for these equations). Let us now consider the unstable branch of the paramagnetic complexity (that is, from \eqref{eq:d_0(Dphi)}, $D_\phi>1/g^2$), then the expression for the complexity reads:
\begin{align}
\begin{split}
\Sigma(D_\phi)&=
q^*\hat\omega^*
+Q_{\phi}^*\hat\lambda^*
-\frac{q^*}{2\,g^{2}\,Q^*{\phi}}
-(1-D_{\phi})\log(1-D_{\phi})
-D_{\phi}\log D_{\phi}
-\frac12\ln\bigl(2\,g^{2}\,\pi\,Q_{\phi}^*\bigr)\\&
+(1-D_{\phi})
\ln\left(
\frac{e^{-\hat\lambda^*}\sqrt{\pi}\,\bigl[1-\operatorname{erf}(\sqrt{\hat\omega^*})\bigr]}{\sqrt{\hat\omega^*}}\right)+D_{\phi}\,\ln\left(\frac{\sqrt{\pi}\operatorname{erf}\bigl(\sqrt{\hat\omega^*+\hat\lambda^*}\bigr)}{\sqrt{\hat\omega^*+\hat\lambda^*}}\right)\\
&+\frac{1}{2 g^2} - \frac{D_\phi}{2} + \frac{D_\phi}{2} \log(g^2 D_\phi).
\end{split}
\end{align}
We will drop the subscript "A" when referring to the paramagnetic complexity, since we can argue from the computation of the RS quenched complexity (see Appendix.~\ref{app:scs_ann_compl}) that it is equal to the annealed one. Let us look at Fig.~\ref{fig:para_kac_scs}, which shows the plot of the paramagnetic complexity of unstable equilibria as a function of $g$. It is interesting to see that the complexity decreases as $g\to 1^+$, while finally collapsing to a point at $D_\phi=1$ at $g=1$. This is in perfect agreement with the DMFT analysis described above: as $g$ increases above $1$, the system crosses from a stable fixed point to chaotic motion. The interpretation in terms of complexity is clear, for $g<1$ we have $\Sigma(D_\phi)=0$ and the (unique) fixed point is stable, while for $g>1$ this point becomes unstable and there is an explosion in the number of unstable stationary points. We can refer to this phenomenon as a \textit{topology trivialization}, similarly to Refs.~\cite{Fyodorov_2016, fyodorov2016nonlinear, ben2021counting}. Although we do not report it in Fig.~\ref{fig:para_kac_scs}, it can be seen that the paramagnetic complexity is always unstable for $g>1$, that is:
\begin{align}
    \forall D_\phi\in[0,1]\quad s.t.\quad\Sigma(D_\phi)\geq0,\quad D_\phi>1/g^2
\end{align}
meaning that there is an exponential abundance of unstable fixed points, and no stable ones. In the following, we examine more in detail the limits $g\to\infty$ and $g\to 1^+$.\\\\

\noindent\textit{The $g\to\infty$ limit}

\begin{figure*}[t!]
    \centering
    \includegraphics[width=0.75\textwidth]{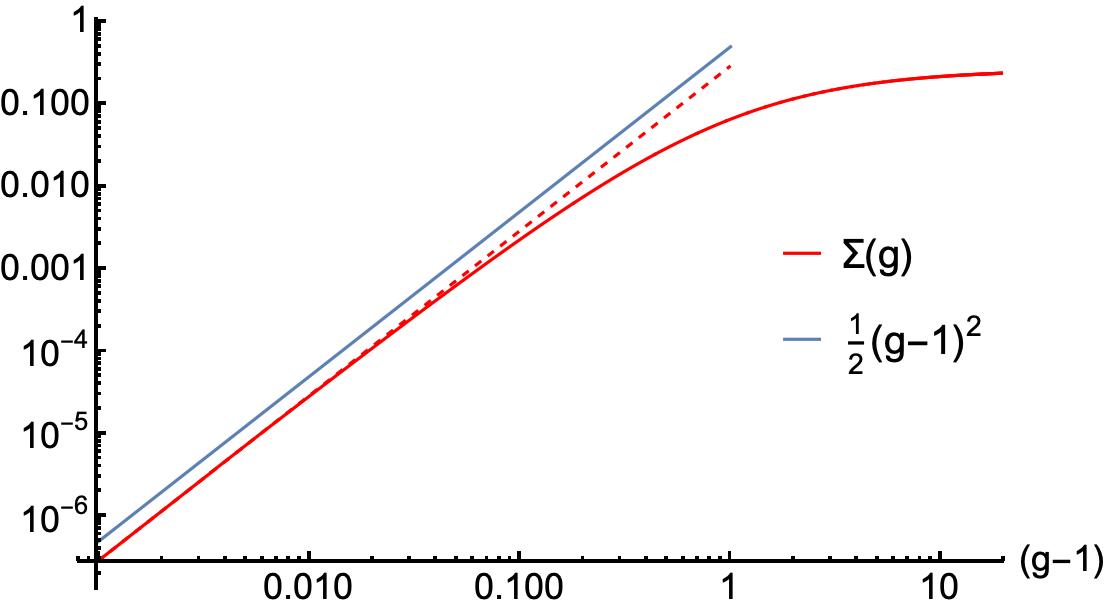}
    \caption{This log-log plot shows the maximum value of the complexity in the paramagnetic region (i.e. $J_0=0$) as a function of $g$ (denoted by $\Sigma(g)$ for simplicity, see red curve). The scaling of the red curve as $g\to 1^+$ is quadratic in $(g-1)$, while as $g\to\infty$ it saturates. The blue line represents the behavior of the maximal Lyapunov exponent with $\phi=\tanh$ for $g$ close to 1: $\lambda_{max}(g)=\frac{1}{2}(1-g)^2+\mathcal{O}((1-g)^3)$ \cite{crisanti_path_2018}. It remains to be verified whether $\lambda_{max}(g)$ behaves in the same way for our choice of $\phi$, see \cite{pacco_scs_2025}.}
    \label{fig:lyap_vs_compl}
\end{figure*}

\noindent It is interesting to note that as $g\to\infty$ the equation for the complexity reduces to a simple limit that can be computed exactly. Indeed, by considering the $g\to \infty$ limit in the equations describing the order parameters and the complexity (denoted $\Sigma^\infty$ in short) we find that 
\begin{align}
    \begin{split}
    \Sigma^\infty(D_\phi)&=(-1 + D_\phi)\log(1 - D_\phi)\\& + 
 \frac{D_\phi}{2Q_\phi^*} \left[D_\phi-1 - 2 Q_\phi^* \log(D_\phi) + 
    2 Q_\phi^* \log\left(\text{erf}\left(\sqrt{\frac{D_\phi}{2Q_\phi^*}}\right)\right)\right]
\end{split}
\end{align}
where $Q_\phi^*$ is the solution of 
\begin{align}
\label{eq:scs_Qphi_eqn}
    -1 + D_\phi + \frac{e^{-\frac{D_\phi}{2 Q_\phi^*}} \sqrt{\frac{2}{\pi}} \sqrt{D_\phi Q_\phi^*}}{
 \text{erf}\left(\sqrt{\frac{D_\phi}{2 Q_\phi^*}}\right)}=0.
\end{align}

\noindent This complexity is plotted in Fig.~\ref{fig:para_kac_scs} in blue as a function of $D_\phi$. We can also find the optimum of the complexity by taking a derivative with respect to $D_\phi$. In the limit of $g\to\infty$ the limiting value of the optimum solves the following equation:
\begin{align}
\label{eq:scs_Dphi_eqn}
\frac{D_\phi^*}{2\, Q_\phi^*} + \log\left[\frac{(1 - D_\phi^*) \cdot \text{erf}\left(\sqrt{\frac{D_\phi^*}{2Q_\phi^*}}\right)}{D_\phi^*}\right]=0.
\end{align}
Whereas Eq.~\eqref{eq:scs_Qphi_eqn} is solved for any $D_\phi$, Eq.~\eqref{eq:scs_Dphi_eqn} must be solved together with Eq.~\eqref{eq:scs_Qphi_eqn} to find the optimum of the blue curve in Fig.~\ref{fig:para_kac_scs}. The solution to these two equations is given by:
\begin{align}
&D_\phi^*=\frac{\operatorname{erf}\left(\frac{1}{\sqrt{\pi}}\right)}{e^{-1/\pi} + \operatorname{erf}\left(\frac{1}{\sqrt{\pi}}\right)}\approx0.44\\
&Q_\phi^*=\frac{\pi}{2}D_\phi^*.
\end{align}
The final expression of the complexity at the optimum is then:
\begin{align}
\Sigma^\infty(D_\phi^*)=-\frac{1}{\pi} + \log\left(1 + e^{1/\pi} \, \text{Erf}\left(\frac{1}{\sqrt{\pi}}\right)\right)\approx0.264.
\end{align}
Notice that this is the same result found in Ref.~\cite{HeliasFP2022} for the different choice $\phi(x)=\tanh(x)$. Our result hints at the fact that this result might be invariant on the specific choice of $\phi$, as long as $\phi$ is a sigmoidal function with $\phi'(0)=1$ and $\lim_{x\to\infty}\phi(\pm x)=\pm 1$. Moreover, with our approach we can obtain the full curve, as well as the value of the instability index at the maximum, which is $D_\phi^*\approx 0.44$. \\\\

\noindent Our result contradicts a conjecture in Ref.~\cite{wainrib2013topological}, where the complexity is thought to scale as the maximum Lyapunov exponent for large $g$ ($\log(g)$ according to \cite{crisanti_path_2018}). Since the (maximum) of the complexity converges to a fixed value as $g\to\infty$, this conjectured link is non-existent. \\

\noindent\textit{The $g\to1^+$ limit}

\noindent Obtaining an analytic scaling for the complexity as $g\to1^+$ does not seem to be an easy task. However, we can solve the equations numerically to see that $\text{max}_{D_\phi}\Sigma(D_\phi)\sim (1-g)^2$ as $g\to 1^+$, meaning that the complexity scales quadratically in $g-1$ for $g$ close to 1. The prefactor can be inferred numerically, and its exact value is not of particular importance, see Fig.~\ref{fig:lyap_vs_compl}. In Ref.~\cite{crisanti_path_2018} it is found that the maximum Lyapunov exponent behaves as $\lambda_{max}(g)= \frac{1}{2}(g-1)^2+\mathcal{O}((g-1)^3)$ when $g\to 1^+$ with $\phi=\tanh$. It remains to be verified whether this scaling is unchanged with our choice of $\phi$. However, we expect that the prefactors of the complexity and the maximal Lyapunov will differ, at variance with previously thought conjectures \cite{wainrib2013topological}. The article \cite{pacco_scs_2025} will present the updated plot.

\subsubsection{Ferromagnetic complexity}
The paramagnetic solution to the complexity is not, in general, the only solution to the 9 saddle point equations. Indeed, there may be ferromagnetic solutions, that is, with $m\neq 0$ and $M_\phi\neq 0$. We first concentrate on the complexity of unstable stationary points. The equations are rather cumbersome to be written explicitly; we will therefore just present their numerical solution. As in Chapter~\ref{chapter:non_reciprocal}, the ferromagnetic complexity develops from the paramagnetic complexity by increasing $J_0$, up to a point where a bifurcation in $m,M_\phi$ appears. At that point, the ferromagnetic complexity appears as a (higher value) branch detaching from the paramagnetic curve, see Fig.~\ref{fig:complexity_cav_dmft_3}.

\begin{figure*}[t!]
    \centering
    \includegraphics[width=\textwidth]{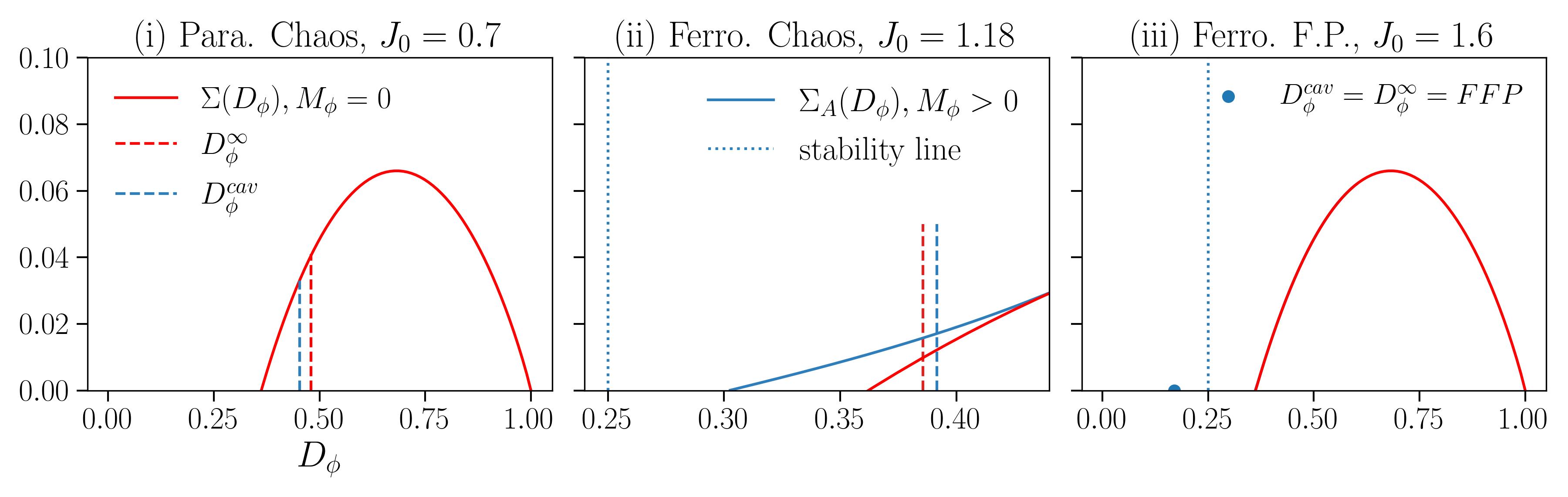}
    \caption{This figure shows how the complexity changes as $J_0$ is increased, for $g=2$. The values (i),(ii), (iii) of $J_0$ are also indicated in Fig.~\ref{fig:dmft_scs_phase_sign}. From left to right we represent the complexity in the three different dynamical phases: PC, FC and FFP. In all plots we show with dashed lines the values of $D_\phi$ reached by either the DMFT ($D_\phi^\infty$) or computed by extending the cavity solution within the chaotic regime ($D_\phi^{cav}$). In the left plot we see that these two values do not match, and the dynamical value lies in the middle of the complexity curve. In the middle plot we zoom to show better where the complexity bifurcates, from paramagnetic (red) to  ferromagnetic (blue). Also here the two values of $D_\phi$ do not match. In the rightmost plot there is only one (stable) ferromagnetic fixed point that has zero complexity (blue dot) and it corresponds to the dynamical attractor.}
    \label{fig:complexity_cav_dmft_3}
\end{figure*}


\subsubsection{Ferromagnetic (stable) fixed point }
As one increases $J_0$ the ferromagnetic branch of the complexity expands, up to a point in which it develops a single fixed point in the stable region (i.e. for $D_\phi<1/g^2$), see Fig.~\ref{fig:complexity_cav_dmft_3}. We identify the position of this point by $D_\phi^{cav}$, because it can be found also by means of a cavity approach, omitted here as it gives the same result of the annealed complexity. The transition line in the plane $(g,J_0)$ is found by studying the saddle point solution to the complexity in the regime where $D_\phi^{cav}\leq 1/g^2$, thus identifying the value of $J_0$ such that the stable fixed point reaches $D_\phi^{cav}=1/g^2$. As for the DMFT, we call such value $J_0^{c}$ ("c" for "critical"), and we will see below that they coincide. The solutions to the saddle-point that give us the Ferromagnetic Stable Fixed Point are the following 10 equations (we added the equation $\partial_{D_\phi}\Sigma_A(D_\phi)=0$):
\begin{align}
\label{eq:ffp_scs_params}
\begin{split}
&\hat\xi=0,\quad \hat\lambda=0,\quad \hat\eta = -\frac{J_0 M_\phi}{g Q_\phi},\quad \hat w = \frac{1}{2 g^2 Q_\phi},\quad m = g J_0 M_\phi,\quad \hat t=0\\
&q = g^2 J_0^2 M_\phi^2 + g^2Q_\phi,\quad D_\phi^{cav}=\frac{1}{2} \text{erf}\left(\frac{1 - g J_0 M_\phi}{ g \sqrt{2Q_\phi}}\right) + \frac{1}{2} \text{erf}\left(\frac{1 + g J_0 M_\phi}{ g \sqrt{2 Q_\phi}}\right)
\end{split}
\end{align}
and by already injecting these into the equations for $M_\phi$ and $Q_\phi$ we get directly two self-consistent equations:
\begin{align}
\label{eq:ffp_scs_mphi}
\begin{split}
&M_{\phi}
+\frac{g\sqrt{Q_{\phi}}
\Bigl(
    e^{-\frac{(gJ_{0}M_{\phi}-1)^{2}}{2g^{2}Q_{\phi}}}
  - e^{-\frac{( \;1+gJ_{0}M_{\phi})^{2}}{2g^{2}Q_{\phi}}}
\Bigr)}{\sqrt{2\pi}}
+\frac{1}{2}\,\bigl(1-gJ_{0}M_{\phi}\bigr)\,
   \operatorname{erf}\left(
      \frac{1-gJ_{0}M_{\phi}}{\sqrt{2}\,g\sqrt{Q_{\phi}}}
   \right)\\
   &
-\frac{1}{2}\,( \,1+gJ_{0}M_{\phi}\bigr)\,
   \operatorname{erf}\left(
      \frac{1+gJ_{0}M_{\phi}}{\sqrt{2}\,g\sqrt{Q_{\phi}}}
   \right)=0
\end{split}
\end{align}
and 
\begin{align}
\label{eq:ffp_scs_qphi}
\begin{split}
& D_{\phi}^{cav}-1
+ Q_{\phi}
- D_{\phi}^{cav}\,g^{2}\!\bigl(J_{0}^{2}M_{\phi}^{2}+Q_{\phi}\bigr)
\\&+
    e^{-\frac{\bigl(1+gJ_{0}M_{\phi}\bigr)^{2}}{2g^{2}Q_{\phi}}}\,
    g\sqrt{\dfrac{Q_{\phi}}{2\pi}}\,
    \Bigl(
        1-gJ_{0}M_{\phi}
        +e^{\frac{2J_{0}M_{\phi}}{gQ_{\phi}}}\bigl(1+gJ_{0}M_{\phi}\bigr)
    \Bigr)=0
\end{split}
\end{align}
where in the last equation we kept $D_\phi^{cav}$ but it is intended that it has to be replaced with its expression found above. The critical line is found by solving these two equations for $M_\phi$ and $Q_\phi$ at fixed $g$ and varying $J_0$: the value of $J_0^c$ then is the one satisfying 
\begin{align}
\label{eq:J0^c_kac}
D_\phi^{cav}=\frac{1}{g^2}\Longleftrightarrow \frac{g^2}{2} \left[\text{erf}\left(\frac{1 - g J_0^c M_\phi}{ g \sqrt{2Q_\phi}}\right) + \text{erf}\left(\frac{1 + g J_0^c M_\phi}{ g \sqrt{2 Q_\phi}}\right)\right]=1.
\end{align}
It can be verified that the complexity for this choice of variables is indeed zero, indicating a sub-exponential number of fixed points. In the present scenario there are actually two stable fixed points, each of opposite magnetization.

\section{Comparing complexity and dynamics: the case $\alpha=0$}
\label{sec:comparison_scs_kac_dmft}
The comparison between the values of the order parameters at the saddle-point obtained via DMFT (left) and those obtained via Kac-Rice (right) is done as follows:
\begin{align}
\begin{split}
    &m^\infty\leftrightarrow m,\quad M_\phi^\infty \leftrightarrow M_\phi,\quad C_x^0\leftrightarrow q, \quad C_\phi^0\leftrightarrow Q_\phi,\quad C_x^\infty\leftrightarrow \tilde{q},\quad C_\phi^\infty\leftrightarrow\tilde{Q}_\phi,\\
    & D_\phi^\infty:=\int dz\,p(z)\phi'\left(gJ_0M_\phi^\infty+z\sqrt{\Delta_0}\right)\leftrightarrow  D_\phi
\end{split}
\end{align}
where we use a superscript $\infty$ to distinguish the values $m,M_\phi,D_\phi$ of the DMFT from those of the Kac-Rice, and 
where the last two parameters $\tilde{q},\tilde{Q}_\phi$ are replica parameters that are introduced when computing the quenched complexity. However, as we explain in Appendix~\ref{sec:quenched_scs}, they are not necessary to carry out our comparison. In particular, in the PC phase, in the FFP phase and at $D_\phi^{cav}$ (in the FC phase) they are not necessary, since in the first case they are 0, in the second they converge to $q,Q_\phi$ (respectively) at the transition line, and in the last case they combine to give back the annealed result.

\begin{figure*}[t!]
  \centering
  \begin{subfigure}[b]{0.49\textwidth}
    \centering
    \includegraphics[width=\linewidth]{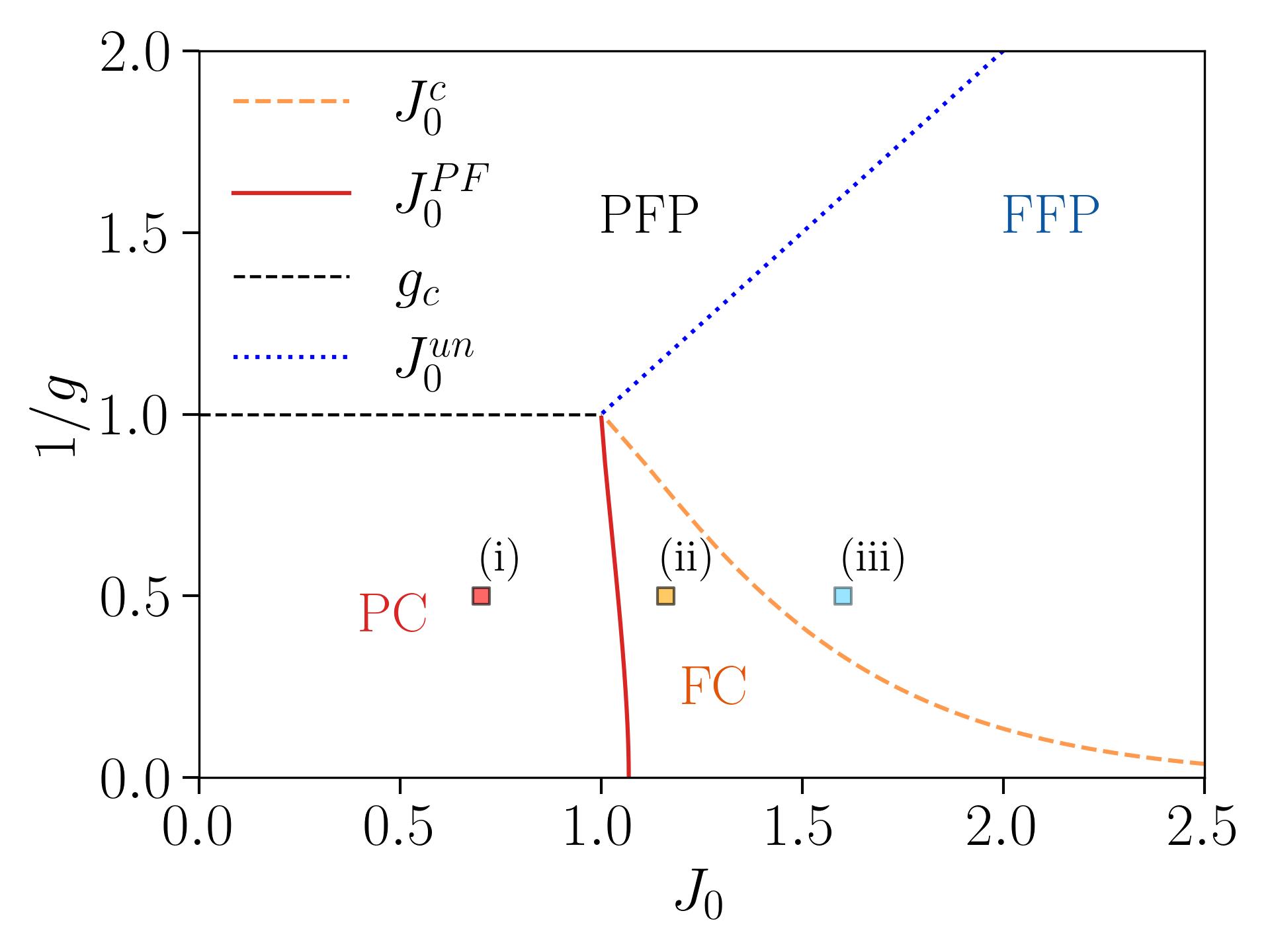}
  \end{subfigure}
  \hfill
  \begin{subfigure}[b]{0.49\textwidth}
    \centering
    \includegraphics[width=\linewidth]{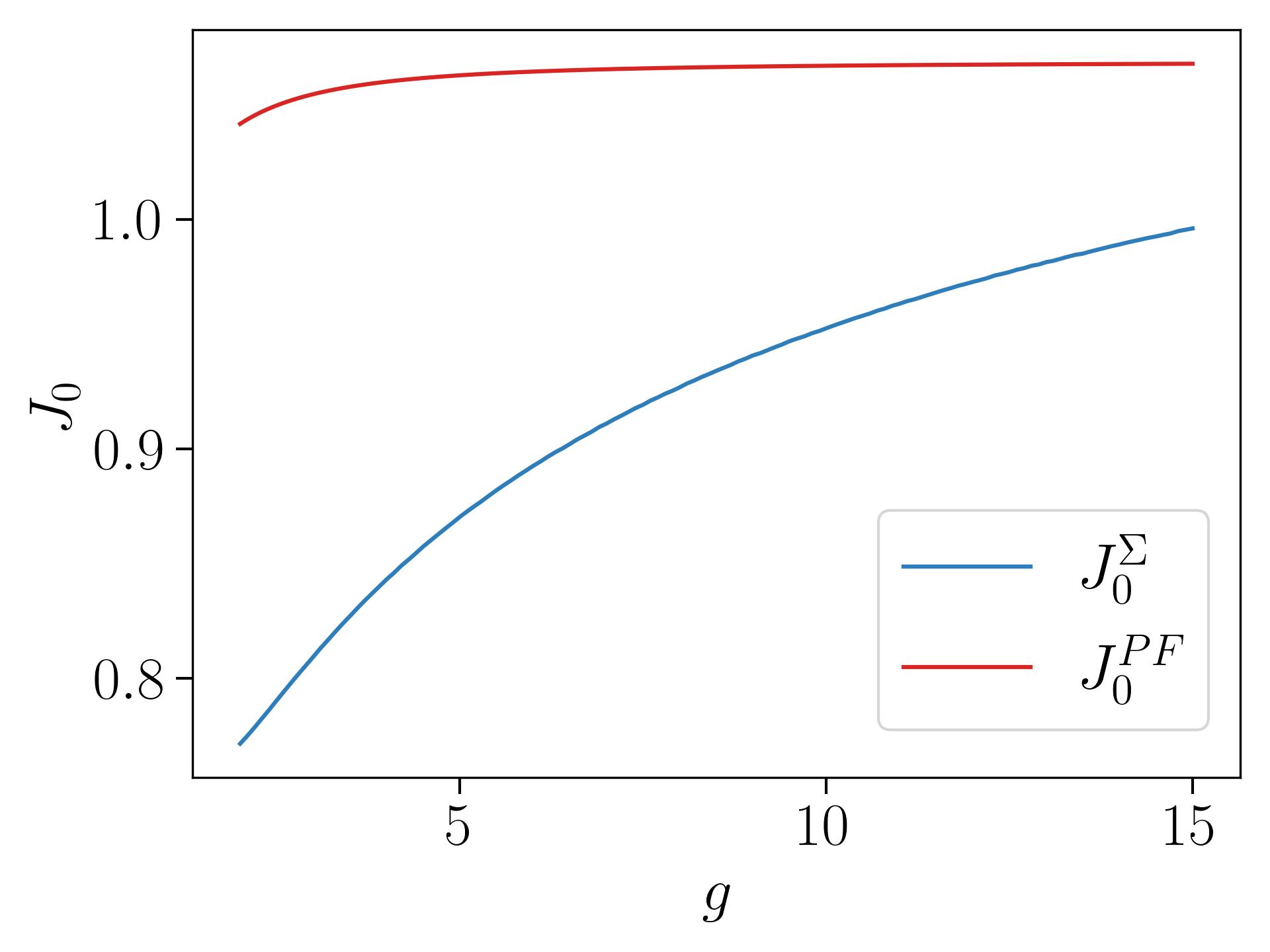}
  \end{subfigure}
  \caption{\textit{Left}. Dynamical phase diagram for $\phi$ as chosen in Eq.~\eqref{eq:def_phi_sign}. The three colored squares are for reference to the choice of values of Fig.~\ref{fig:complexity_cav_dmft_3}. \textit{Right}. For each $g>1$ we plot the curves $J_0^{PF}(g)$ (indicating the point where the dynamics passes from the PC phase to the FC phase), and $J_0^\Sigma(g)$, the curve where the complexity bifurcates to have a branch of ferromagnetic (unstable) stationary points. These two curves are clearly separated. Moreover, they do not have the same shape, in contrast to the model studied in Chapter~\ref{chapter:non_reciprocal}. }
  \label{fig:dmft_scs_phase_sign}
\end{figure*}

\subsection{FFP phase}
Let us start from the simplest phase, the Ferromagnetic (stable) Fixed Point phase. In this case we have already obtained above the explicit equations for $M_\phi$ and $Q_\phi$ of the Kac-Rice case, as well as the explicit equation for the transition line to the FC phase. The other Kac-Rice order parameters at the saddle-point are then related (from Eq.~\eqref{eq:ffp_scs_params}) by:
\begin{align}
    m=gJ_0M_\phi,\quad q=m^2+g^2Q_\phi.
\end{align}
These last relations also hold for the DMFT saddle-point equations, see \eqref{eq:dmft_system_eqns}, where one has $m^\infty=gJ_0M_\phi^\infty$ and $C_x^0=(m^\infty)^2+\Delta_0=(m^\infty)^2+g^2C_\phi^0$. The DMFT equations for $M_\phi^\infty$ and $C_\phi^0$ instead read (see Eq.~\eqref{eq:M_D_ffp}):
\begin{align*}
\begin{cases}
    M_\phi^\infty=\int dz\, p(z)\phi\left(gJ_0M_\phi^\infty+zg\sqrt{C_\phi^0}\right)\\
    C^0_\phi=\int dz\, p(z)\phi^2\left(gJ_0M_\phi^\infty+zg\sqrt{C_\phi^0}\right)\\
\end{cases}
\end{align*}
By plugging inside these equations the choice for $\phi(x)$ in \eqref{eq:def_phi_sign} and after carrying out a few calculations, it is not hard to show that we get back the same equations of the order parameters $M_\phi,Q_\phi$ found via Kac-Rice. The link is thus clear: the unique stable fixed point (or rather the two unique points with opposite magnetization) is the only attractor of the dynamics provided that $g>1$ and $J_0>J_0^c$. This is found by imposing that $J_0^c$ must satisfy 
\begin{align*}
    1=g^2\int dz\,p(z)\left[\phi'\left(gJ_0^cM_\phi^\infty+zg\sqrt{C_\phi^0} \right)\right]^2
\end{align*}
which for the choice $\phi'(x)=\mathbb{I}_{|x|<1}$ can be turned exactly into Eq.~\eqref{eq:J0^c_kac}. Hence also the transition line from the FFP to the FC regime coincides using either the Kac-Rice or DMFT formalisms. The transition line $J_0^c$ is shown in Fig.~\ref{fig:dmft_scs_phase_sign} \textit{left} (dashed, orange), while a plot of the complexity and the fixed point is shown in Fig.~\ref{fig:complexity_cav_dmft_3}.(iii). \\

\begin{figure*}[t!]
  \centering
  \begin{subfigure}[b]{0.49\textwidth}
    \centering
    \includegraphics[width=\linewidth]{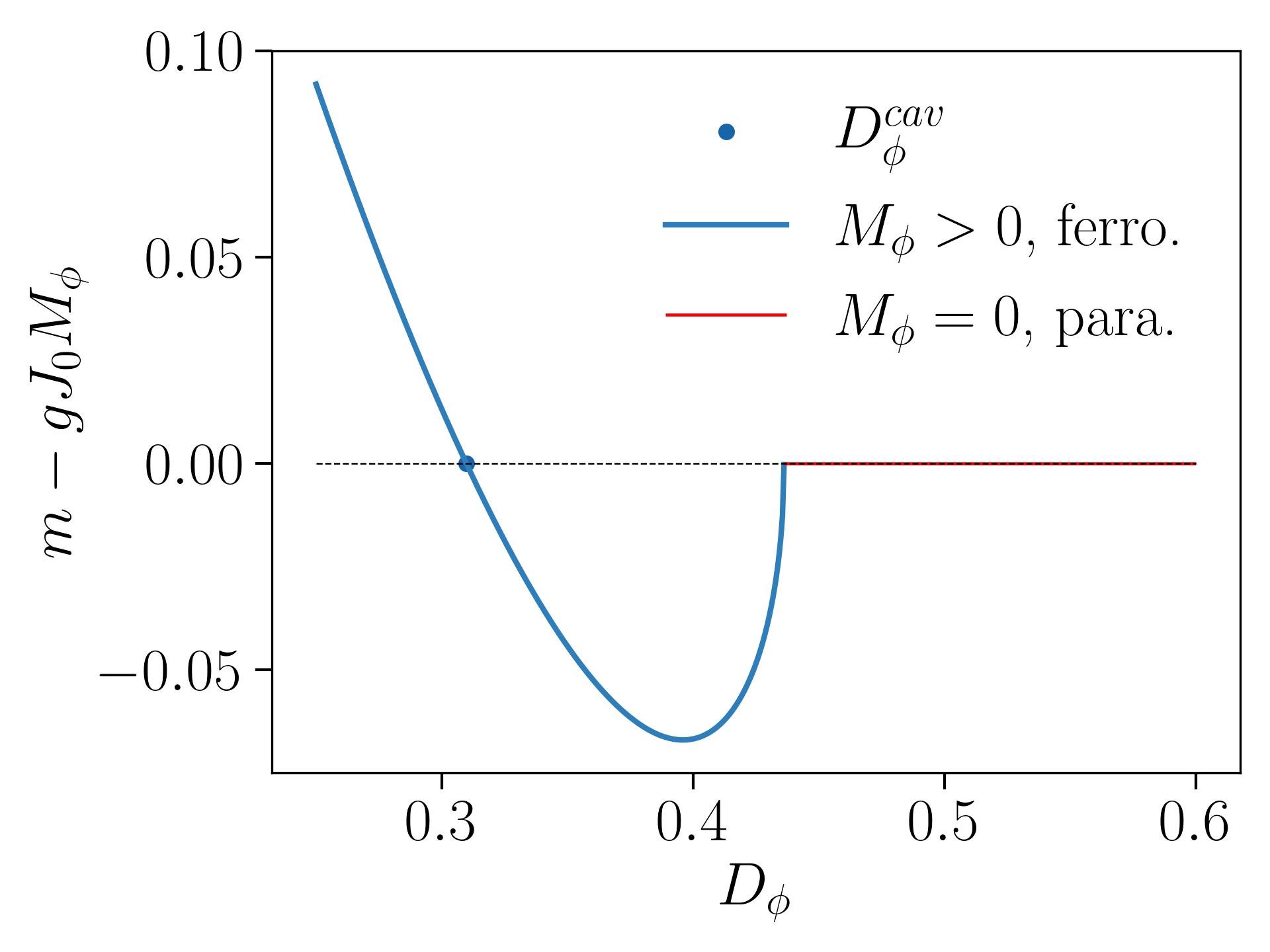}
    \label{fig:panel-a}
  \end{subfigure}
  \hfill
  \begin{subfigure}[b]{0.49\textwidth}
    \centering
    \includegraphics[width=\linewidth]{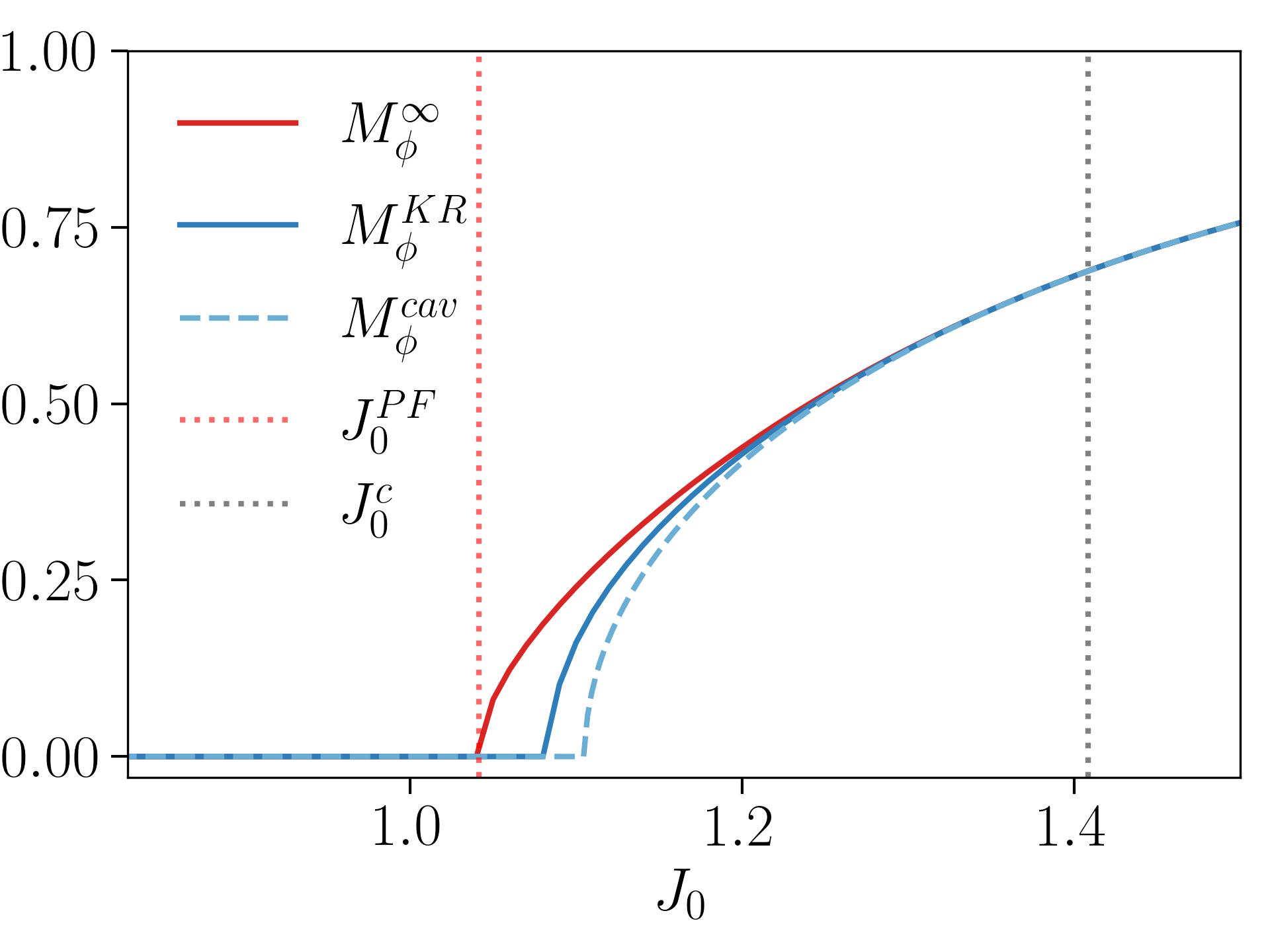}
    \label{fig:panel-b}
  \end{subfigure}
  \caption{\textit{Left}. Plot of the difference between $m$ and $gJ_0M_\phi$ (found by optimizing the complexity) for $g=2, J_0=1.3$. We see that there is only one point, at $D_\phi^{cav}$, where $m=gJ_0M_\phi$. It is therefore the only good candidate for a comparison with the DMFT. \textit{Right}. For $g=2$ and as a function of $J_0$ (comprising the three phases), comparison between the dynamical order parameters $M_\phi^{\infty}$, the "cavity" solution $M_\phi^{cav}$ and the optimal solution of the Kac-Rice, denoted $M_\phi^{KR}$, at $D_\phi^{\infty}$. The three solutions are identical in the PC and FFP phases (as we have shown to be the case), but they differ in the FC phase. In particular in the FC phase, the order parameter $M_\phi^{KR}$ obtained by optimizing the Kac-Rice at the dynamical $D_\phi^{\infty}$ does not correspond to the dynamical $M_\phi^{\infty}$.}
  \label{fig:two_plots_cav_dmft_scs}
\end{figure*}

\noindent Let us make an important comment; the equation for $D_\phi^{cav}$ can be extended to the chaotic region, although it will clearly not correspond to a unique fixed point. However, it can be checked that there is a unique point such that equations \eqref{eq:ffp_scs_params}, \eqref{eq:ffp_scs_mphi}, \eqref{eq:ffp_scs_qphi} are all satisfied simultaneously. This point, which we abusively denote as the "cavity point", was shown in our previous work, Chapter~\ref{chapter:non_reciprocal} (Ref.~\cite{us_non_reciprocal_2025}), to be the one where the stability parameter (here $D_\phi$) of the DMFT lied. Hence, in the present case, we want to check whether that claim in Chapter~\ref{chapter:non_reciprocal} remains true for the SCS model.  Moreover, this "cavity point" is crucial for the comparison between the DMFT and the Kac-Rice, as it is the only point where $m=gJ_0M_\phi$, an equality that is always satisfied by the TTI solution of the DMFT. 
From Fig.~\ref{fig:two_plots_cav_dmft_scs} \textit{left} we see that this is indeed the case, only one point can be the good candidate for a matching between DMFT and Kac-Rice equations, and it is at $D_\phi^{cav}$. As we will see below, however, $D_\phi^{cav}\neq D_\phi^{\infty}$ in the chaotic phases for the SCS model. \\

\noindent The last remark is that, differently with other systems (see, e.g., \cite{Fyodorov_2016}) choosing $J_0$ as a perturbation that multiplies $M_\phi$ has the consequence that even in the FFP phase there is still a branch of the complexity that contains exponentially many unstable fixed points (see for instance Fig.~\ref{fig:complexity_cav_dmft_3}.(iii)). This means in particular that despite the exponential abundance of (unstable) stationary points, the system is dynamically attracted to the unique stable fixed point. This is easily seen numerically (see for example Fig.~\ref{fig:chaos_num} in Chapter~\ref{chapter:non_reciprocal} for an analogous situation).

\begin{figure*}[t!]
  \centering
  \begin{subfigure}[b]{0.49\textwidth}
    \centering
    \includegraphics[width=\linewidth]{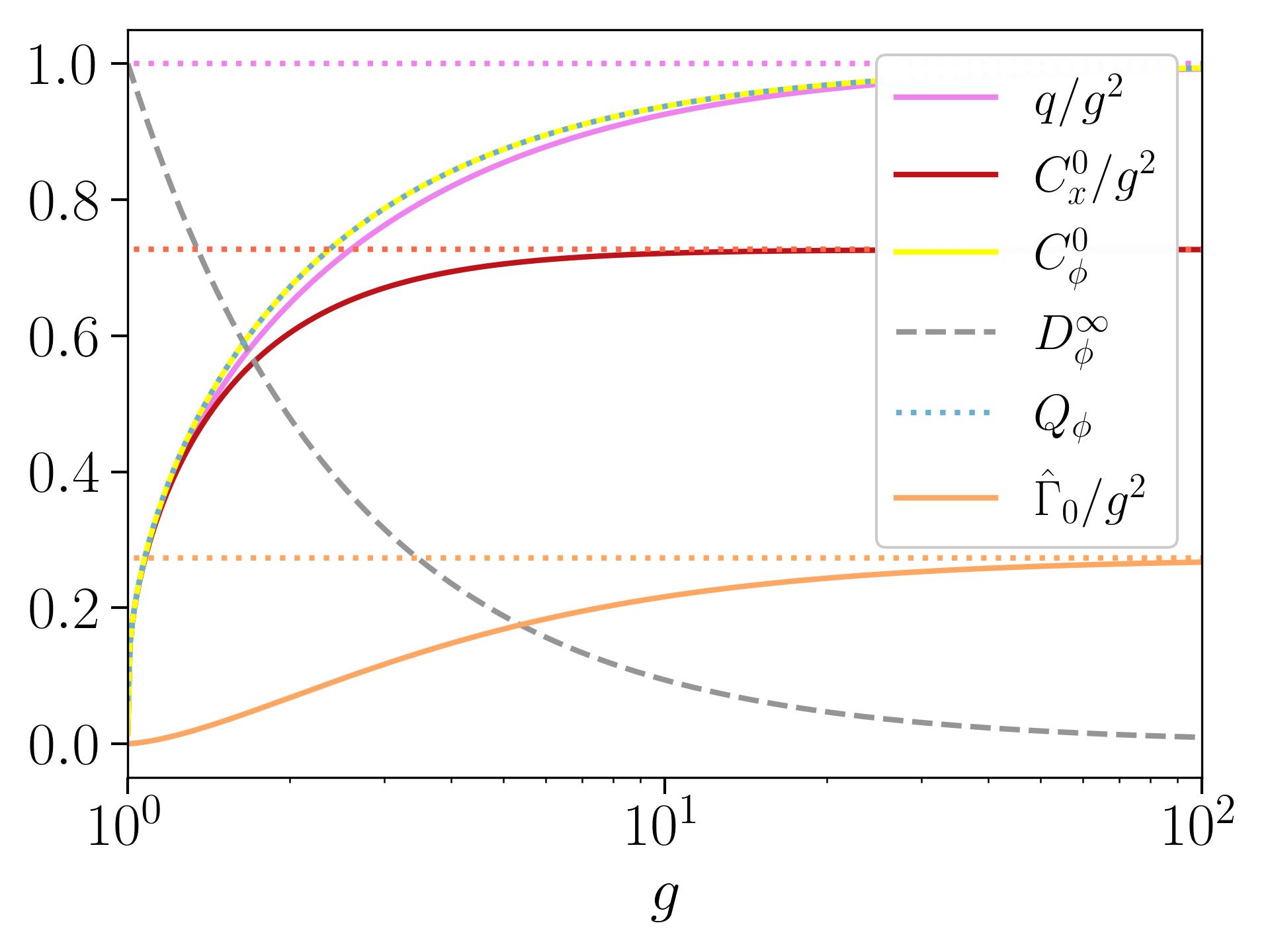}
  \end{subfigure}
  \hfill
  \begin{subfigure}[b]{0.49\textwidth}
    \centering
    \includegraphics[width=\linewidth]{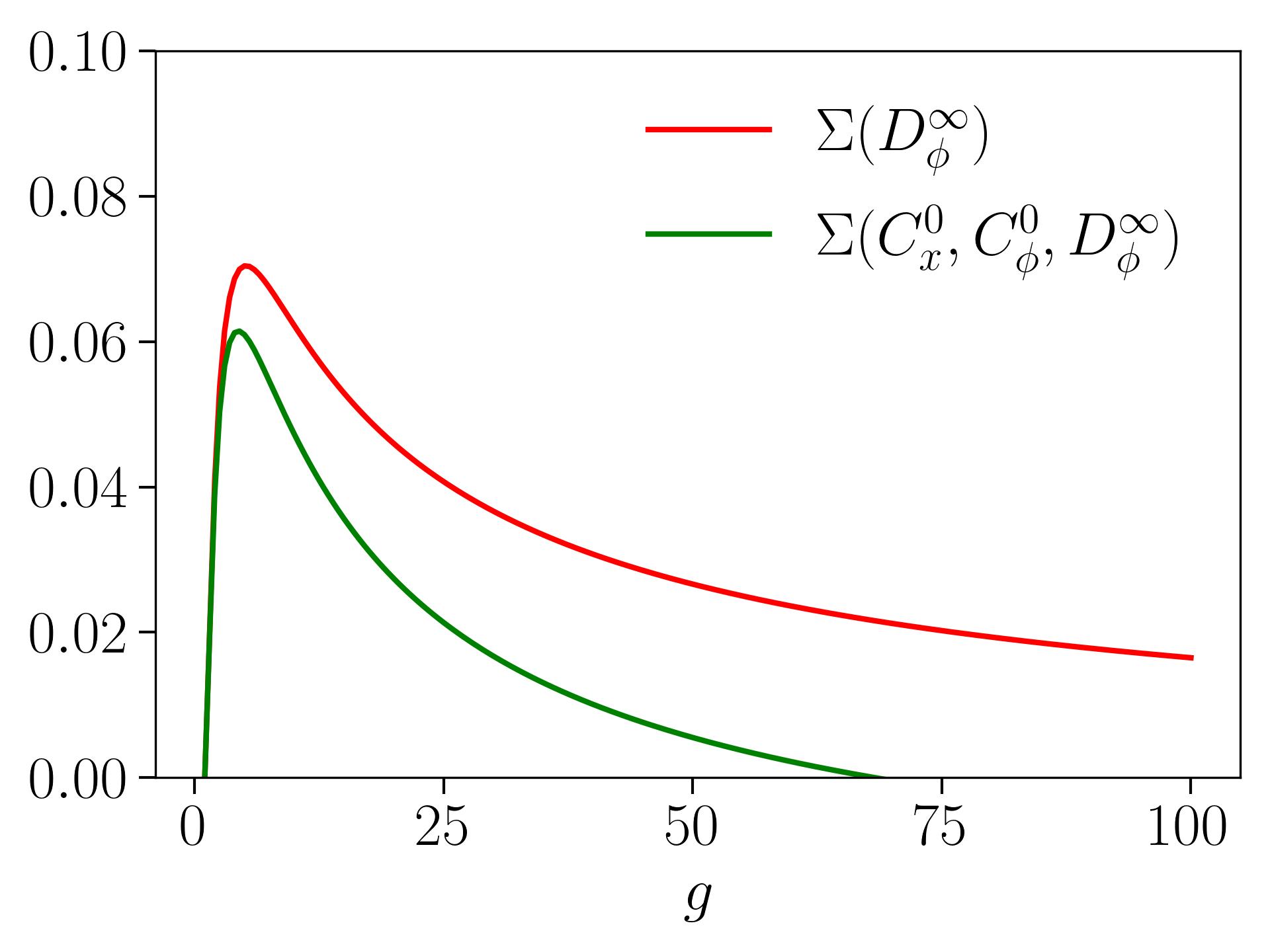}
  \end{subfigure}
  \caption{\textit{Left}. In the paramagnetic case $J_0=0$ we plot the dynamical values $D_\phi^\infty, C_x^0/g^2,C_\phi^0,\hat{\Gamma}_0/g^2$ as a function of $g$. The force $\hat{\Gamma}_0$ is shown to be positive and increasing. The other values also converge. The Kac-Rice order parameters $q,Q_\phi$ are evaluated at $g,D_\phi^\infty(g)$ for any $g$. We see that the curves of $C_\phi^0$ and $Q_\phi$ coincide graphically, (although numerically they differ at the 4-th decimal digit), wheres $q$ and $C_x^0$ are clearly separated. \textit{Right}. Plot of the complexity as a function of $g$, evaluated at $D_\phi^\infty(g)$ and either optimized over the order parameters (red) or evaluated at the DMFT values $q\to C_x^0,\,Q_\phi\to C_\phi^0$ (green). The green curve eventually goes to zero, indicating that the shell where the dynamics lives does not contain fixed points.}
  \label{fig:para_scs_two_plots}
\end{figure*}

\subsection{PC phase}
The paramagnetic chaotic phase is present when choosing $g>1$ and $J_0<J_0^{PF}(g)$. In this phase the dynamics is characterized by $m^\infty=M_\phi^\infty=0$, $C_x^0$ given by Eq.~\eqref{eq:delta0_para} and $C_\phi^0=\int dz\, p(z)\phi^2(z\sqrt{C_x^0})$. The explicit equation of $C_x^0$ for our choice of $\phi$ reads:
\begin{align}
\begin{split}
&\bigl(C_x^{0}\bigr)^{2}
+\frac{g^{2}}{2}\Biggl[
   e^{-\frac{1}{2C_x^{0}}}\,
     \sqrt{\frac{2}{\pi}}\,
     \sqrt{C_x^{0}}\,
     \bigl(1+3C_x^{0}\bigr)\\&
   +3\bigl(C_x^{0}\bigr)^{2}
     \operatorname{erf}\!\Bigl(1/\sqrt{2C_x^{0}}\Bigr)-(1+4C_x^{0})\,
     \operatorname{erfc}\!\Bigl(1/\sqrt{2C_x^{0}}\Bigr)
   \\&+\frac{
       \Bigl[
           e^{-\frac{1}{C_x^{0}}}
           \Bigl(
               \sqrt{2C_x^{0}}
               +e^{\frac{1}{2C_x^{0}}}\sqrt{\pi}\,
                \bigl(
                   -1
                   +(1+C_x^{0})
                    \operatorname{erf}\!\bigl(1/\sqrt{2C_x^{0}}\bigr)
                 \bigr)
           \Bigr)
       \Bigr]^{2}
     }{\pi}
\Biggr]=0
\end{split}
\end{align}
Moreover $C_x^0$ and $C_\phi^0$ are related  by:
\begin{align}
C_\phi^0=1 - e^{-\frac{1}{2C_x^0}} \sqrt{\frac{2}{\pi}} \sqrt{
  C_x^0}  + (C_x^0-1)\,\text{erf}\left(\frac{1}{\sqrt{2C_x^0}}\right).
\end{align}
The asymptotic value $D_\phi^\infty$ reached by the dynamics is given by the equation:
\begin{align}
D_\phi^\infty=\text{erf}\left(\frac{1}{\sqrt{2 C_x^0}}\right).
\end{align}

\begin{figure*}[t!]
  \centering
  \begin{subfigure}[b]{0.49\textwidth}
    \centering
    \includegraphics[width=\linewidth]{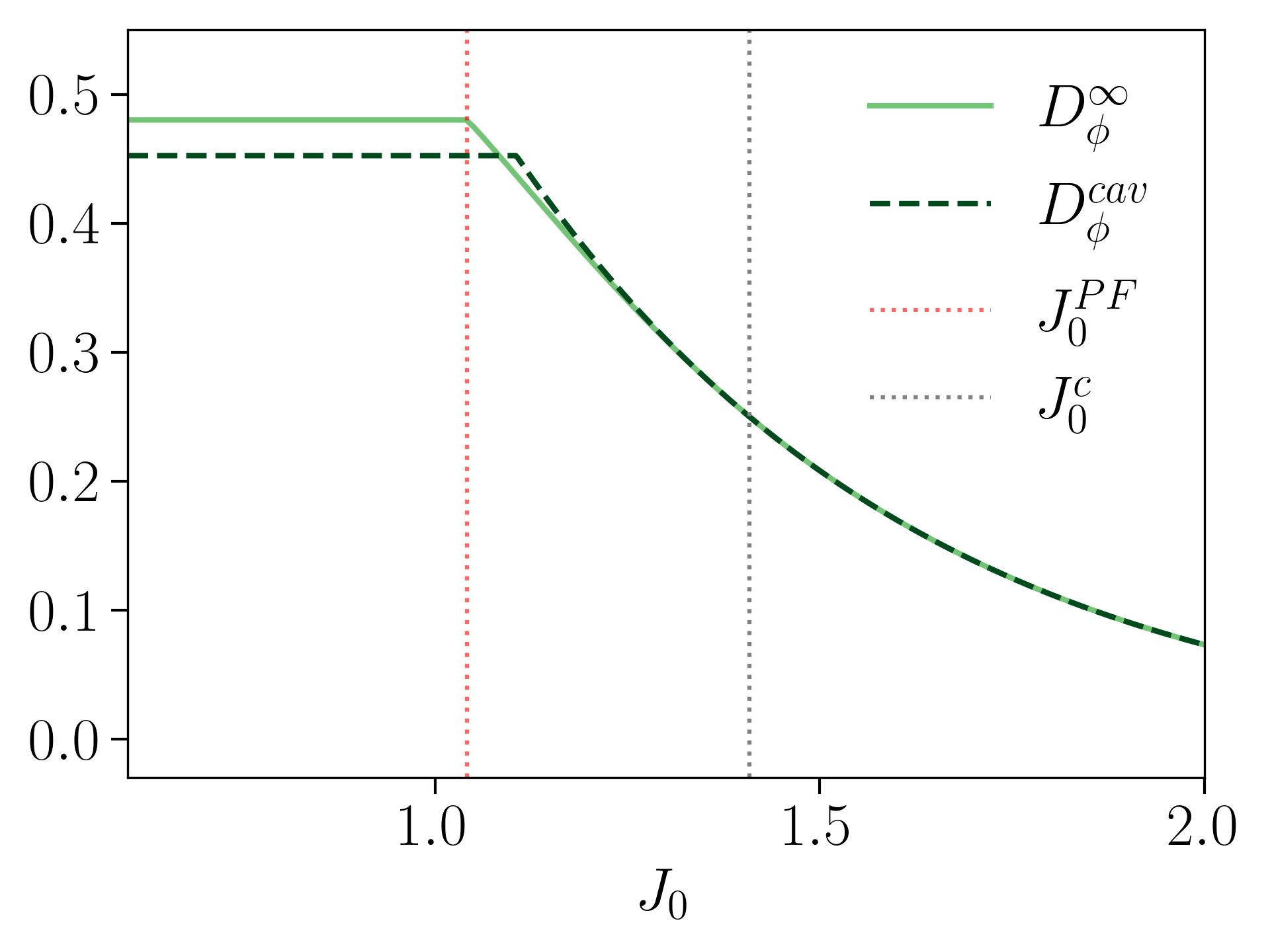}
  \end{subfigure}
  \hfill
  \begin{subfigure}[b]{0.49\textwidth}
    \centering
    \includegraphics[width=\linewidth]{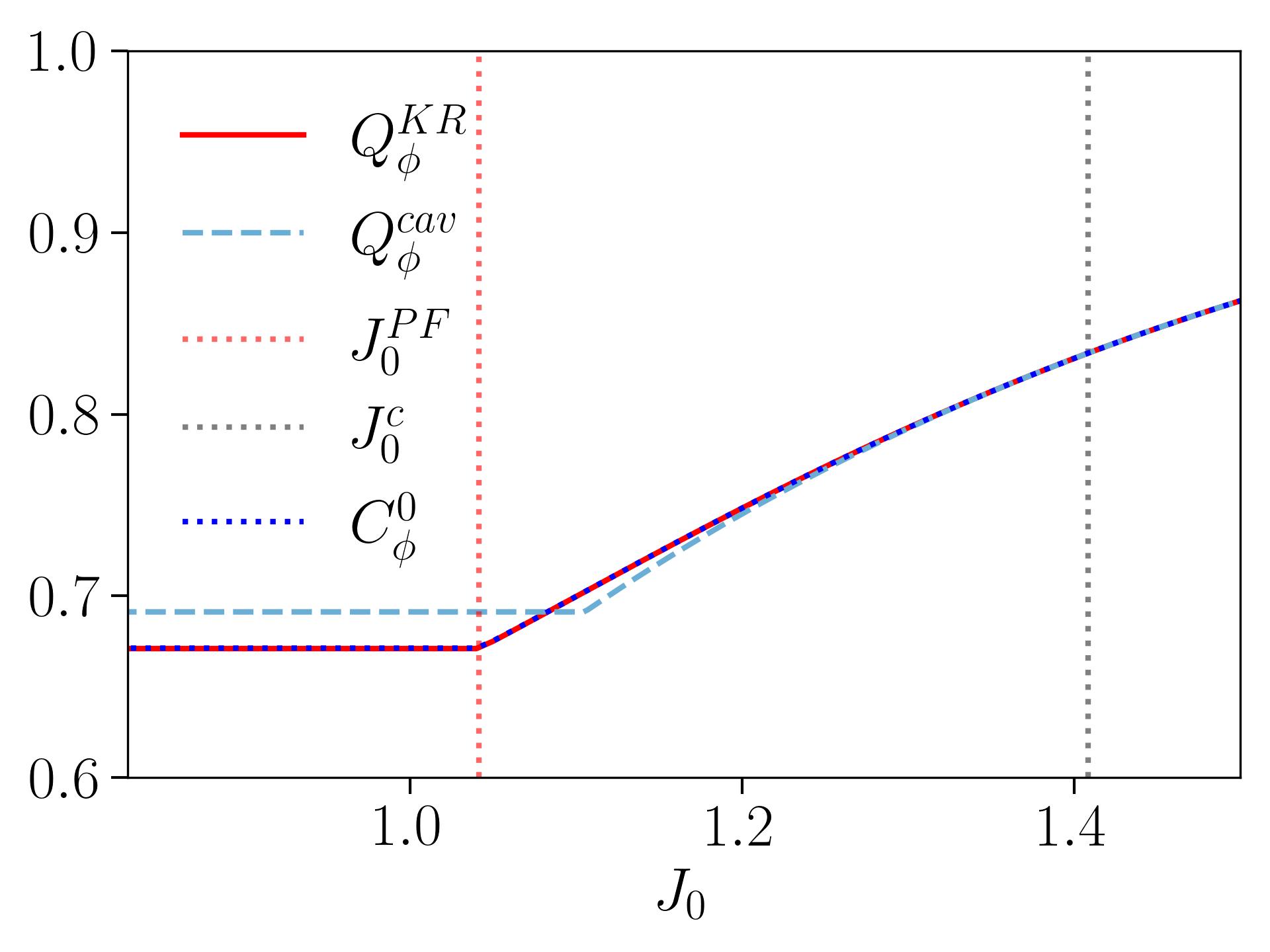}
  \end{subfigure}
  \caption{\textit{Left}.  Plot of $D_\phi^\infty$ versus $D_\phi^{cav}$ for $g=2$ and varying $J_0$ in the three dynamical phases PC,FP,FFP. The two are close but different in the PC and FC phases, and coincide in the FFP phase. In particular, the two curves transition to ferromagnetic values at two different points. \textit{Right}. Still for $g=2$ and varying $J_0$, comparison of $C_\phi^0$, $Q_\phi^{cav}$ and $Q_\phi^{KR}$ (the last one being the value of the Kac-Rice order parameter at the corresponding $D_\phi^\infty$). We see that $C_\phi^0$ and $Q_\phi^{cav}$ are slightly mismatched, whereas $C_\phi^0$ and $Q_\phi^{KR}$ coincide up to the 4-th numerical digit (not seen graphically). Whether they are actually identical requires an analytical solution of the corresponding equations.}
  \label{fig:ferro_scs_comparison_n1}
\end{figure*}

It is interesting to plot these quantities as a function of $g$ and compare with the Kac-Rice. Already in Fig.~\ref{fig:complexity_cav_dmft_3}.(i) we see that the value $D_\phi^\infty$ lies in the middle of the paramagnetic complexity curve, and it does not coincide with $D_\phi^{cav}$. A further plot of the force $\hat{\Gamma}_0$ confirms that the force is non-zero for any $g>1$, showing that the system is being constantly driven out of equilibrium; see Fig.~\ref{fig:para_scs_two_plots} \textit{left} (orange line). In the same plot, we compare the dynamical values of $C_\phi^0(g), C_x^0(g)$ with $Q_\phi(g),q(g)$ optimized at $D_\phi=D_\phi^\infty(g)$. While we see that $C_x^0/g^2$ and $q/g^2$ converge to different values, $C_\phi^0$ and $Q_\phi$ lie on the same curve. A careful numerical comparison reveals that they differ at the 4-th decimal digit. Whether they are actually identical (or just numerically close) at $D_\phi^\infty$ remains a challenging open problem given the difficulty of the equations involved. In Fig.~\ref{fig:para_scs_two_plots} \textit{right} we also investigate the value of the complexity at $D_\phi^\infty(g)$ as $g$ grows. The red curve shows the complexity obtained by optimizing over the order parameters, whereas the green curve shows the complexity evaluated at the order parameters of the DMFT. The red curve remains positive and converges to zero as $g\to\infty$, but the green curve touches zero for $g\approx 70$. Therefore, for $g$ big enough, there are no stationary points in the subspace chosen by the dynamics.

\begin{figure*}[t!]
  \centering
  \begin{subfigure}[b]{0.49\textwidth}
    \centering
    \includegraphics[width=\linewidth]{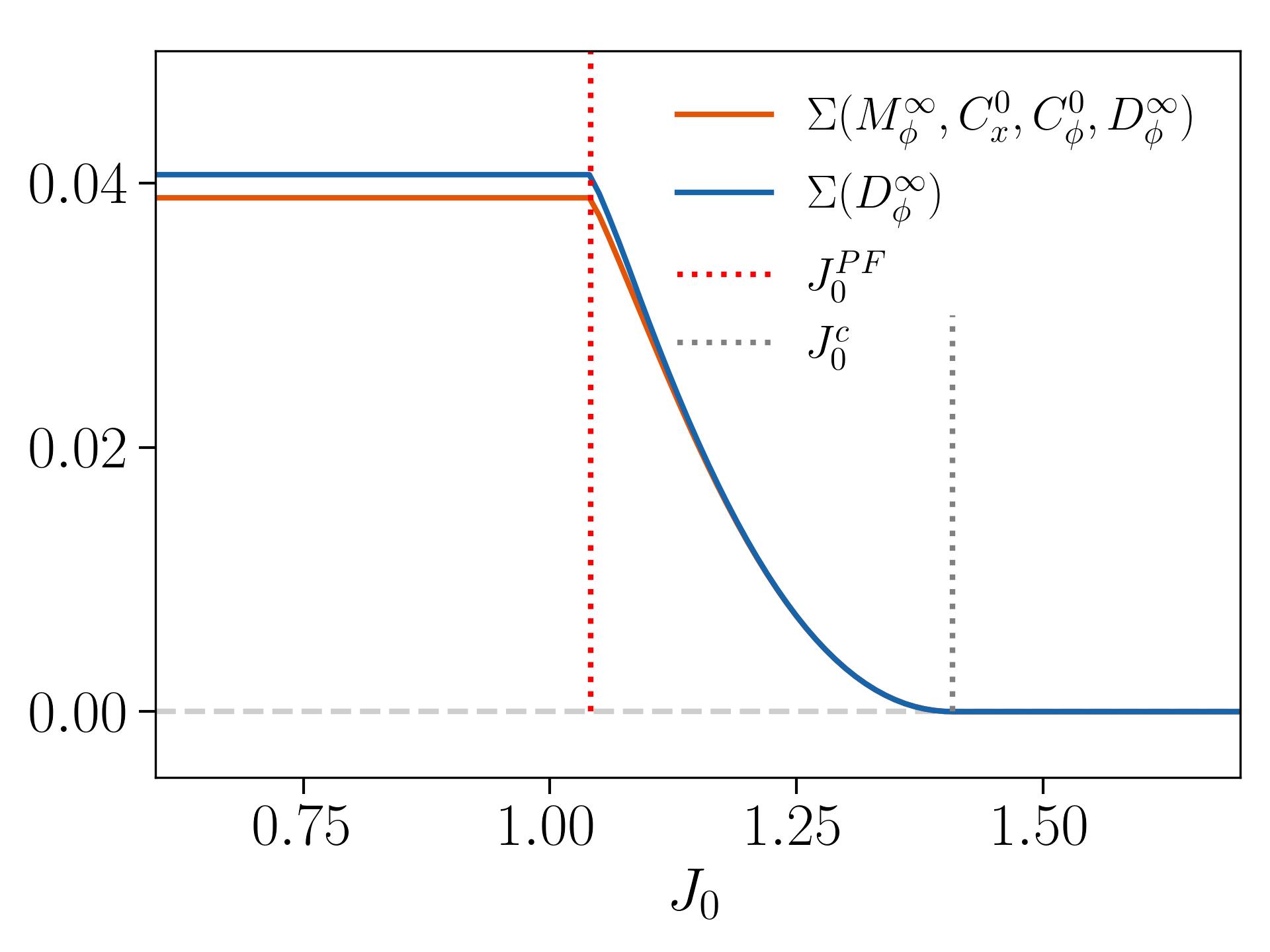}
  \end{subfigure}
  \hfill
  \begin{subfigure}[b]{0.49\textwidth}
    \centering
    \includegraphics[width=\linewidth]{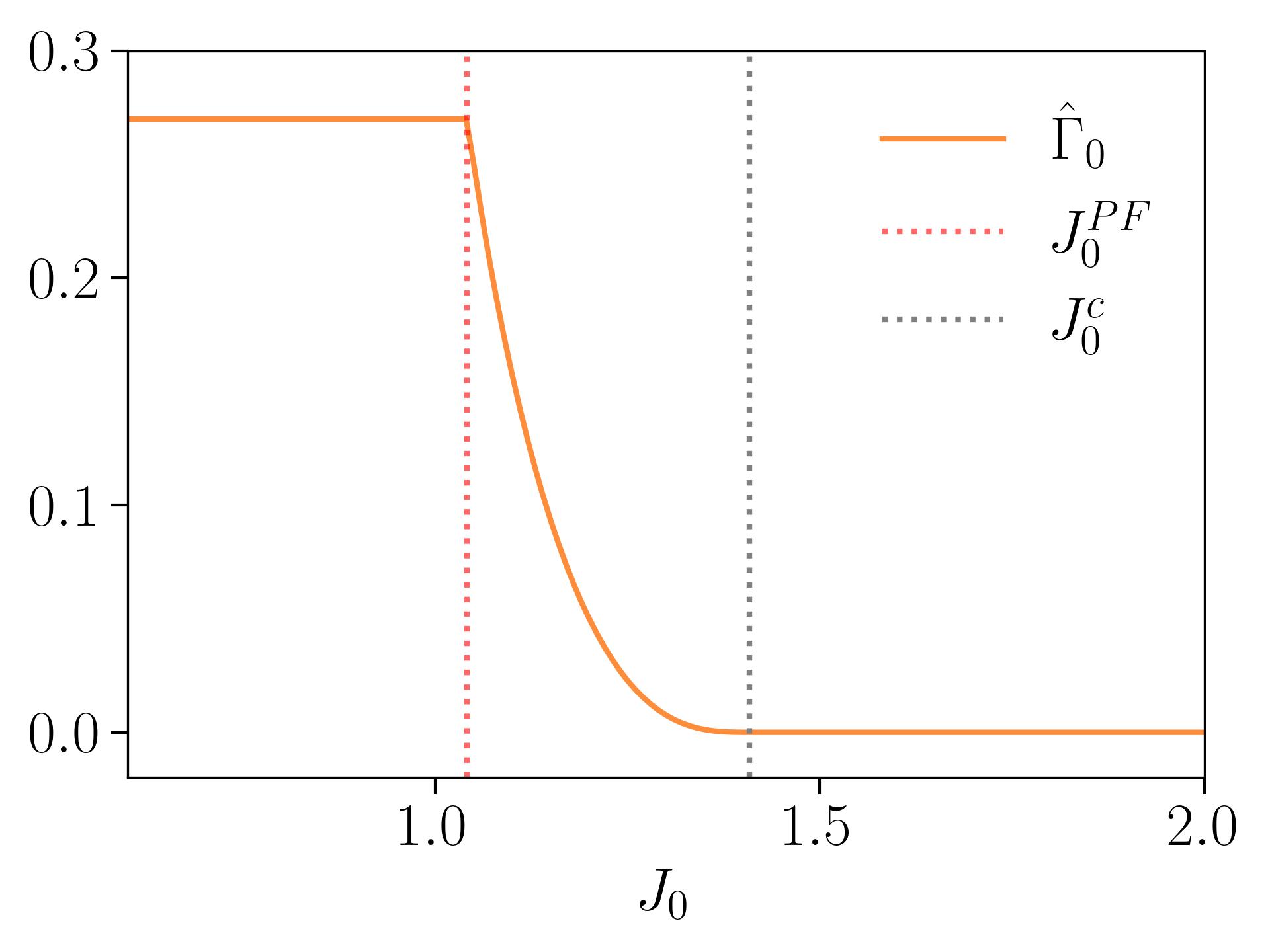}
  \end{subfigure}
  \caption{\textit{Left}. Plot of the complexity at $D_\phi^\infty$. The blue curve is the complexity optimized over its order parameters, while the orange one is evaluated at the dynamical order parameters. \textit{Right}. Plot of the force $\hat{\Gamma}_0$ as a function of $J_0$ for $g=2$, spanning all the three dynamical phases: PC, FC, FFP. }
  \label{fig:ferro_scs_comparison_n2}
\end{figure*}

\subsection{FC phase}
The equations for the order parameters in this phase are cumbersome and we avoid to write them explicitly. However, we can solve numerically the DMFT equations in \eqref{eq:fc_scs_dmft_eqns} and the Kac-Rice equations. In Fig.~\ref{fig:complexity_cav_dmft_3}.(ii) we can see how the ferromagnetic branch (blue) of the complexity bifurcates from the paramagnetic one (red). These ferromagnetic stationary points are the "least unstable" ones, meaning that their $D_\phi$ is the smallest, and thus their number of unstable modes is the lowest among the population of typical equilibria.
 A plausible conjecture would be that at fixed $g$, by decreasing $J_0$, the point where this blue branch of the complexity disappears, call it $J_0^\Sigma(g)$, coincides with $J_0^{PF}$ (the point where the dynamics transitions from PC to FC). That is, one could imagine that the dynamics becomes ferromagnetic chaotic when there start to appear ferromagnetic fixed points. However, Fig.~\ref{fig:dmft_scs_phase_sign} \textit{right} contradicts this hypothesis. The asymptotic value of $D_\phi^\infty$ reached by the dynamics reads:
\begin{align}
    D_\phi^\infty=\frac{1}{2} \left[\text{erf}\left(\frac{1 - g J_0 M_\phi^\infty}{\sqrt{2C_x^0}}\right) + 
   \text{erf}\left(\frac{1 + g J_0 M_\phi^\infty}{\sqrt{2C_x^0}}\right)\right]
\end{align}
and it is plotted in Fig.~\ref{fig:ferro_scs_comparison_n1} \textit{left} as a function of $J_0$ for $g=2$. We also compare it with the $D_\phi^{cav}$ found by extending the solution \eqref{eq:ffp_scs_params} of the FFP point into the chaotic region. Indeed, as we can see from Fig.~\ref{fig:two_plots_cav_dmft_scs} this "cavity point" is the only point where $m=gJ_0M_\phi$ (a critical equality for comparison with the DMFT). However, by looking at Fig.~\ref{fig:complexity_cav_dmft_3}.(ii) and Fig.~\ref{fig:ferro_scs_comparison_n1}\textit{left} we see that in general $D_\phi^{cav}\neq D_\phi^\infty$ (the first being obtained via Kac-Rice, the second one via DMFT). The two curves are close but different, and merge as soon as $J_0>J_0^c$, as we have shown before. Importantly, the driving force $\hat{\Gamma}_0$ is shown in Fig.~\ref{fig:ferro_scs_comparison_n2} \textit{right} for $g=2$ as a function of $J_0$. We see that the force is continuous and monotonically decreasing in $J_0$, until it reaches the value of 0 at $J_0^c$. This fact implies that the force gradually decreases until, for $J_0$ large enough, the system is at an equilibrium. Even if the force being positive implies that the system is out of equilibrium, the system must know of the existence of fixed points that are less and less unstable, as $J_0$ is increased. Indeed, the transitions for all the order parameters from the FC to the FFP phase are continuous, implying that the system gradually converges to the fixed point as $J_0$ increases. However, this analysis shows that a clear link between the dynamics and the distribution of stationary points is not obvious, and further investigation is needed to show how the properties of the phase space shape the dynamical order parameters. Nonetheless, the fact that the DMFT and Kac-Rice parameters are very close within this FC phase (see Figs.~\ref{fig:ferro_scs_comparison_n1} and ~\ref{fig:ferro_scs_comparison_n2}) suggests that these quantities are strongly correlated, although we are not able, with the present analysis, to understand how.

\section{Conclusions and Perspectives}
\label{sec:scs_perspectives}
In this Chapter we have studied the SCS model with an external field (i.e. an excitation of the neurons). We have chosen a specific non-linearity $\phi$ that approximates the usual $\tanh$ used for these models, thus showing that it allows for a computation of the annealed complexity with any asymmetry of the interactions $\alpha\in[0,1]$. Our approach allows to classify the fixed points in terms of several order parameters, including the instability index. We have been able to make careful comparisons between the complexity and the TTI solution of the DMFT for $\alpha=0$, thus revealing that, as for Chapter~\ref{chapter:non_reciprocal}, the DMFT cannot be inferred from the Kac-Rice in the chaotic phases. We have also shown that those connections that appeared in Chapter~\ref{chapter:non_reciprocal} do not hold here anymore. We have explicitly plotted the complexity as $g\to 1^+$, showing that the transition to chaos is concomitant with the appearance of exponentially many unstable fixed points. Like for Chapter~\ref{chapter:non_reciprocal}, we hope that this work with our choice of $\phi$ motivates further research to better understand the connection between the dynamics and the statics for the SCS model. \\

\noindent Several follow-ups arise from this research.
\begin{itemize}
    \item Although we already have the equations of the annealed complexity for any $\alpha$, we used them only for $\alpha=0$, hence it would be interesting to analyze them for any $\alpha$. Is there a critical $\alpha_c$ such that stable fixed points start to appear? If so, what is the relation to the dynamics for $1\geq \alpha>\alpha_c$? To which stable fixed points does the dynamics converge for $\alpha=1$? Does our choice of $\phi$ simplify the numerical integration of the DMFT equations for $\alpha\neq 0$?

    \item We have verified numerically that the complexity scales quadratically at the PFP-PC transition (that is, as $g\to 1^+$), but it would be interesting to find this result analytically. 

    \item We were able to obtain the equations for the quenched complexity within the Replica Symmetric hypothesis, but we didn't solve them numerically, which is an interesting future problem.

    \item In order to compute the quenched complexity (see Appendix) we had to first compute the $n-$th moment of the number of fixed points; $\mathbb{E}[\mathcal{N}^n]$. By considering the case $n=2$ we forsee the possibility to use a second moment approach to show that quenched and annealed results coincide within the paramagnetic region for $\alpha=0$.

    \item It would be interesting to study the scaling of the complexity and the maximal Lyapunov exponent across the FFP-FC transition. 

    \item For the present analysis, we have only classified fixed points as a function of their extensive instability index. Moreover, we have not analyzed the conditioning on the Jacobian at the fixed points, which might induce external eigenvalues in the spectrum. Extending the present analysis to account for a full analysis of the Jacobian remains an open challenge. 
\end{itemize}

\chapter{Energy landscapes}
\label{chapter:energy_landscapes}
In this Chapter we consider the pure spherical $p$-spin model, introduced in detail in Chapter~\ref{chapter:intro}. In particular, we concentrate on Kac-Rice-based methods to probe the typical properties of the landscape, including distribution of stationary points and barriers between them. We propose two approaches. The first approach (see Ref.~\cite{pacco2024curvature}) consists in determining the typical energy profile along specific paths between local minima. This problem also requires the study of overlaps between eigenvectors of spiked, correlated random matrices; treated in Chapter~\ref{chapter:rmt_}.  The second one (see Refs.~\cite{pacco_quenched_triplets_2025, pacco_triplets_2025}) consists in computing the "three-point complexity", namely the distribution of triplets of fixed points, extracted conditionally one after the other in the energy landscape. Many of the calculations are lengthy and not reported here; the interested reader can refer to the Appendices of Refs.~\cite{paccoros, pacco_triplets_2025, pacco_quenched_triplets_2025} for details (some notations might differ). \\

\noindent \textit{Road-map}\\
In Sec.~\ref{sec:ener_land_intro} we introduce and motivate our research. In Sec.~\ref{sec:TwoPoint} we explain the main ingredients of the previously computed two-point complexity. In Sec.~\ref{sec:curvature_driven_paths} we show how we compute energy barriers by selecting specific pathways. In Sec.~\ref{sec:en_land_three_point} we introduce and explain the three-point complexity, by mainly concentrating on some types of "landscape transitions" that we observe. In Sec.~\ref{sec:en_land_actv_dyn} we show how our results can give insights of the activated dynamics, and we conclude in Sec.~\ref{sec:en_land_perspectives} with perspectives.\\

\noindent \textit{Acknowledgments}\\
The two works presented in this Chapter are the result of collaborations with Valentina Ros, Alberto Rosso and Giulio Biroli. I thank them very much for stimulating discussions. I especially thank Valentina and Alberto for their ideas and constant support during these works.

\section{Introduction}
\label{sec:ener_land_intro}
In this Chapter we further analyze the energy landscape of the pure spherical $p$-spin model described in Sec.~\ref{sec:p_spin_model}. We have seen that it represents a prototypical glassy model, with a rough energy landscape below a threshold energy density $\epsilon_{th}$, reviewed also below. This models has been extensively studied in terms of its thermodynamics \cite{kurchan_barriers_93,Crisanti_TAP_pspin_95, crisanti1992sphericalp, franz1995recipes, Barrat_bifurcation_95, cavagna1997structure, cavagna1997investigation}, its energy landscape \cite{cavagna1998stationary, leuzzi2003complexity, ros2019complex, ros2019complexity, auffinger_complexity_2013, Cavagna_saddles_2001} and its dynamics \cite{sompo_zipp_dyn_phase_1981, cugliandolo2011effective, franz2013quasi, cugliandolo1993analytical, Cugliandolo_1995_weak, barrat_dynmeta_96, cugliandolo2002dynamics, ros2021dynamical}. However, the dynamical equations solved within the framework of DMFT (see Appendix~\ref{app:dynamical_calculations} for the equations) assume that $N\to\infty$ before $t,t'\to\infty$. This, however, rules out crossing over barriers that scale with $N$, associated with times that scale exponentially with $N$. Indeed, the picture of the out-of-equilibrium dynamics of this model is that quenching the system from infinite temperature (that is, a random initial condition) to any $T\leq T_d$ leads to a weak ergodicity breaking scenario \cite{bouchaud1992weak, Cugliandolo_1995_weak, cugliandolo1993analytical, bouchaud1998out}, where the system ages forever, never going below the energy threshold. How to characterize the dynamics for large but finite $N$ below $T_d$ is an open problem. In particular, what are the barriers between local minima (that is, for energies in $(\epsilon_{gs}, \epsilon_{th})$) ? If we initialize (or \textit{plant} \cite{krzakala_planting_2009}) the system in a deep minimum, how will the system restore ergodicity ? To which minima will it go ? A detailed analysis of the landscape in the vicinity of deep local minima \cite{ros2019complexity, cavagna1997investigation} reveals that there is an overlap gap between the reference minimum and the closest fixed points, which are rank-1 saddles, meaning that, typically, no fixed points are found at closer distances. We review this calculation in Sec.~\ref{sec:TwoPoint}. In particular, thanks to this result, one can find the smallest energy barrier to the closest rank-1 saddle. The minima that are reached by activated barrier crossing over these rank-1 saddles have been analyzed in \cite{ros2021dynamical}, by combining Kac-Rice and dynamical methods. It was found that minima have much higher energy densities than the reference one, and are correlated with it. By drawing analogies to finite-size simulations in Ising $p$-spin models \cite{stariolo2019activated, stariolo2020barriers}, the authors argue that these saddles would matter in the earlier times of the dynamics, and that escaping through them the system would undergo a back and forth motion with frequent returns to the original minimum. However, this dynamics only accounts for jumps between minima that are connected by the typical rank-1 saddles (see Fig.~\ref{fig:2_point_phase_diag}). It thus remains an open question to understand how the system escapes from these high-energetic nearby saddles and decorrelates from the reference minimum. Overcoming previous works \cite{lopatin1999instantons, lopatin_barriers_2000}, a recent article by Rizzo \cite{rizzo2021path} considers a dynamical theory to compute the exponentially small probability of a jump from one metastable state to another, by fixing the initial and final conditions. By focusing on the exponentially small probability that the system jumps to another equilibrium state
in a finite time, the author shows that an intermediate jump
to one of the exponentially many metastable states provides a more efficient path for restoring ergodicity than a direct jump to another equilibrium state. Moreover a plot from the same article indicates that the dynamical path from two states at equal energy goes well above threshold. An even newer strategy was adopted in the most recent work \cite{folena_rare_2025}, where the authors consider a dynamical potential reminiscent of the Franz-Parisi potential in the statics, see Sec.~\ref{sec:franz_parisi_analysis}. This is done by randomly selecting an equilibrium reference configuration, and then constraining its dynamical evolution so that at time $t$ it has overlap $q$ from the original one. They still send $N\to\infty$ first, and then extract information from the large deviation function (i.e. the potential). Quite importantly, they validate, via dynamical analysis, static predictions regarding the basin of attraction of metastable states. The basin is convex for large (enough) overlaps, and then fibered as the overlap decreases, up to a point where the system exits the basin of attraction and the dynamics becomes irreversible. They suggest that the dynamics for $T_s<T<T_d$ (see Sec.~\ref{sec:tap_approach}) starting from a deep minimum would manage to escape through the closest rank-1 saddles
and temporarily get trapped in some high energy states (denoted as \textit{hubs}) before relaxing again to other equilibrium
states, and likely exploring above threshold regions before. \\

\noindent Here, we use purely static approaches to tackle questions related to the activated dynamics and barrier crossing in energy landscapes. These approaches are justified by the fact that metastable states and local minima of the pure spherical $p$-spin are in one-to-one correspondence (see Sec.~\ref{sec:tap_approach}). Since we look for geometric properties of the landscape, no timescales are involved in our analyses. We use two approaches introduced below.

\subsubsection{Approach 1: curvature-driven pathways.} Characterizing the profile of the landscape along pathways connecting different minima is an ubiquitous problem in many complex systems. These can include energy or fitness functions \cite{wales1998archetypal,jonsson1998nudged, mauri2022mutational, tian2020exploring}, cost functions optimized by algorithms \cite{draxler2018essentially,freeman2016topology,annesi2023star,garipov2018loss} (where the landscape might be characterized by flat minima separated by low barriers \cite{baity2018comparing}) and glasses, where the height of the typical effective barriers crossed during the dynamics increases at low temperatures, thus leading to super-Arrhenius behaviors \cite{debenedetti2001supercooled, biroli_berthier_review_2010}. In Sec.~\ref{sec:curvature_driven_paths} we will consider pairs of stationary points below $\epsilon_{th}$, and compute the typical energy density profile along a geodesic, determining how the energy barriers depend on the energies and the overlap of the two points. We will then compare with "perturbed" geodesics, which follow directions correlated to the landscape curvature around one of the two stationary points. Our goal is to see when the information encoded in the local Hessian allows to lower the energetic barrier associated to the unperturbed path. Such question is motivated by studies of finite-dimensional systems of jammed and mildly supercooled particles, where the softest Hessian mode at a given configuration is associated to low energy barriers for
particle rearrangements \cite{xu2010anharmonic, widmer2008irreversible}. 
For the $p-$spin model, instead, we show here that the smallest Hessian eigenvalue at the starting minimum is not a
predictor of paths with lower energy barriers, except in the case in which the local Hessian has an isolated mode at the arrival point.
However, we show that having access to the whole local Hessian in general allows to identify pathways associated to
lower energy barriers.

\subsubsection{Approach 2: three-point complexity.} This method extends upon the two-point complexity of Ref.~\cite{ros2019complexity}, by considering the distribution of triplets of stationary points, the first one being a deep minimum, and the second and third one either local minima or rank-1 saddles. This study is motivated by works on disordered elastic interfaces and interacting particles, for which recent work has provided significant insight into their low temperature dynamics~\cite{ninarello2017models, ferrero2017spatiotemporal,liu2018creep}. One of the most striking observations made in recent years, both in experiments~\cite{durin2024earthquakelike, korchinski_thermal_2025} and in numerical simulations~\cite{Rosso_review_2021, scalliet2022thirty,tahaei2023scaling, de2024dynamical}, are {\em thermal avalanches}, i.e. the occurrence of a cascade of smaller activations following a slow activated nucleation. We wanted to look for precursors of thermal avalanches in the context of mean-field models, of which the prototypical model is the pure $p$-spin. In this work, we \textit{interpret} thermal avalanches as activated jumps between nearby fixed points: after a first jump from a deep local minimum, the system makes smaller jumps at higher energies, associated to smaller rearrangements. Concerning the energy landscape, we see this in terms of \textit{clustering} of fixed points: the first jump leads to a region of the landscape where fixed points cluster close to each other, and making jumps to nearby fixed points should be easier. In particular, this clustering is a consequence of coming from the reference minimum, and thus a signature of strong correlations within the landscape around it. In the pure $p$-spin model we see that signatures of clustering are present only at higher energies, meaning that, if the reference minium has energy $\epsilon_{gs}<\epsilon_0<\epsilon_{th}$, then clustering can happen  when the first jump is at energy density higher than a certain $\epsilon_0$ dependent value (still below threshold). In particular, we found no clustering for sequences of minima that are at equal energies. In this sense, in the pure $p$-spin a sequence of jumps at equal energy densities below threshold is "memoryless", as we shall see in Sec.~\ref{subsec:defs_clustering} and Sec.~\ref{sec:LandscapeEvolution}.


\subsubsection{The p-spin model: recap on the complexity}
We have already shown the derivation of the complexity of the pure spherical $p$-spin model in Chapter~\ref{chapter:intro}. Here let us just recall that this model is defined by an energy landscape (i.e. Hamiltonian) which is a random Gaussian field $\mathcal{E}$ defined on the hypersphere of radius $\sqrt{N}$ in $\mathbb{R}^N$. It is assumed that $\mathcal{E}$ has zero-mean and variance given by 
\begin{align}
\mathbb{E}[\mathcal{E}({\bf s}_0)\mathcal{E}({\bf s}_1)]=\frac{N}{2}\left(\frac{{\bf s}_0\cdot{\bf s}_1}{N}\right)^p.
\end{align}
The complexity is defined as
\begin{align}
    \Sigma(\epsilon_0)=\lim_{N\to\infty}\frac{\log\mathcal{N}(\epsilon_0)}{N}
\end{align}
with 
\begin{equation}\label{eq:Norma1}
    \begin{split}
&\mathcal{N}(\epsilon_0)= \int_{\mathcal{S}_N} d {\bf s}_0  \,\omega_{\epsilon_0}({\bf s}_0)
    \end{split}
\end{equation}
and 
\begin{equation}\label{eq:Measure0}
 \omega_{\epsilon_0}({\bf s}_0)=|\det \nabla^2_\perp \mathcal{E}({\bf s}_0)|\delta(\nabla_\perp \mathcal{E}({\bf s}_0))\delta(\mathcal{E}({\bf s}_0)-N\epsilon_0).
\end{equation}
The subscript $\perp$ indicates that the gradient and Hessian are Riemannian, that is, restricted on the hypersphere; see Chapter~\ref{chapter:intro}. The quantity $\Sigma(\epsilon_0)$ can be referred to as a "one-point" complexity since it counts the number of stationary points extracted from the landscape, without any further constraints.  In addition, we are using the subscript $0$ for quantities that represent this "one-point" complexity. Indeed, as we will see below, one can compute a "two-point" and even a "three-point" complexity, that is, the complexity of stationary points restricted to be at a certain distance (overlap) from previously extracted stationary points. In the context of spin glasses, these are generally referred to as "real replicas" \cite{franz1995recipes, franz1998effective, cavagna1997structure}, or \textit{primary} and \textit{secondary} (and eventually, \textit{tertiary}) \cite{Barrat_bifurcation_95}.

\noindent We recall, moreover, that the most typical stationary points, counted by $\Sigma$, are local minima for $\epsilon<\epsilon_{th}$, marginal minima for $\epsilon=\epsilon_{th}$, and saddles for $\epsilon>\epsilon_{th}$; see Sec.~\ref{sec:p_spin_model} for details.

\section{Previous results: the two-point complexity}
\label{sec:TwoPoint}
Here we recap the findings of Ref.~\cite{ros2019complexity}, please skip to Sec.~\ref{sec:curvature_driven_paths} for direct results of this Chapter.\\

\noindent The geometry of the landscape in the vicinity of a local minimum can be described~\cite{cavagna1997investigation,ros2019complexity,ros2020distribution} by computing a two-point complexity $\Sigma^{(2)}(\epsilon_1, q |\epsilon_0)$. This is done by extracting first a typical minimum of the landscape at energy $\epsilon_0$, and then by analyzing the statistics (in terms of overlap $q$ and energy $\epsilon_1$) of stationary points extracted upon conditioning to the first one, see Fig.~\ref{fig:landscape_two_point} for a visualization. In formula, this constrained complexity is defined as:
\begin{align}
\label{eq:quenchedcomp2}
\Sigma^{(2)}(\epsilon_1,q|\epsilon_0):=\lim_{N\to\infty}\frac{1}{N}\mathbb{E}\left[\log\mathcal{N}_{{\bf s}_0}(\epsilon_1,q|\epsilon_0)\right]_0.
\end{align}
Here $\mathcal{N}_{{\bf s}_0}(\epsilon_1,q|\epsilon_0)$ 
is the number of stationary points  
 ${\bf s}_1$ of energy $\epsilon_1$, that are at overlap $q$ with another stationary point ${\bf s}_0$, and reads:
 \begin{align}
 \label{eq:en_land_N_s0}
    \mathcal{N}_{{\bf s}_0}(\epsilon_1,q|\epsilon_0):=\int_{\mathcal{S}_N(\sqrt{N})}d{\bf s}_1\,\omega_{\epsilon_1,q}({\bf s}_1|{\bf s}_0)
 \end{align}
 with 
\begin{equation}\label{eq:Measure1}
  \begin{split}
&\omega_{\epsilon_1,q}({\bf s}_1|{\bf s}_0):=|\det \nabla^2_\perp \mathcal{E}({\bf s}_1)|\delta(\nabla_\perp \mathcal{E}({\bf s}_1))\,\delta(\mathcal{E}({\bf s}_1)-N\epsilon_1)\delta({\bf s}_1 \cdot {\bf s}_0-N q).
\end{split}
\end{equation}

 \begin{figure*}[t!]
    \centering
    \includegraphics[width=0.42\textwidth, trim={2 2 2 2},clip]{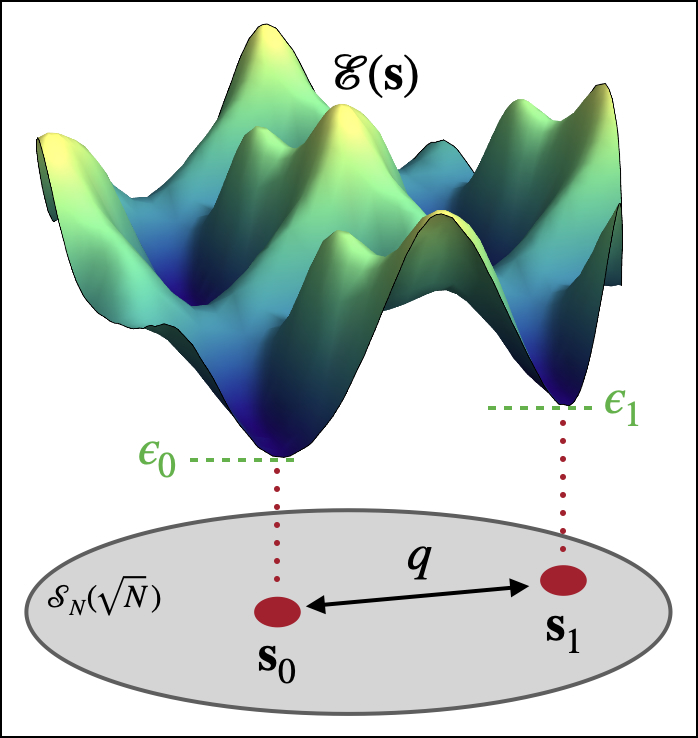}
    \caption{Artistic representation of a landscape, to show visually the primary and the secondary configurations in the calculation of the two-point complexity. }
    \label{fig:landscape_two_point}
\end{figure*}

 \noindent The average $\mathbb{E} \left[\cdot \right]_0$ denotes both a flat average over all stationary points ${\bf s}_0$ with energy $\epsilon_0$ at fixed realization of the landscape, and over the realizations of the landscape. More precisely:
 \begin{equation}
     \mathbb{E} \left[\cdot \right]_0:= \mathbb{E}\left[ \frac{1}{\mathcal{N}(\epsilon_0)} \int_{\mathcal{S}_N(\sqrt{N})} d{\bf s}_0 \, \omega_{\epsilon_0}({\bf s}_0) \, \cdot \; \right],
 \end{equation}
 where $ \omega_{\epsilon_0}({\bf s}_0)$, defined in Eq.~\eqref{eq:Measure0}, is the measure selecting configurations ${\bf s}_0$ that are stationary points of energy density $\epsilon_0$. The explicit expression of the two-point complexity can be found in \cite{ros2019complexity}, and it reads: 
\begin{equation}\label{eq:FinalComplexityPostSaddle}
\Sigma^{(2)}(\epsilon_1,q|\epsilon_0) =   \frac{Q(q)}{2}-\xi(\epsilon_0,\epsilon_1, q)+ I\tonde{\epsilon_1 \sqrt{\frac{p}{p-1}}} 
 \end{equation}
 where: 
\begin{equation*}
 \begin{split}
           &Q(q)=  1+\log \tonde{\frac{2(p-1)(1-q^2)}{1-q^{2p-2}}},\\
           &\xi(\epsilon_0,\epsilon_1, q)= \epsilon_0^2 U_0(q)+ \epsilon_0 \epsilon_1 U(q)+ \epsilon^2_1 U_1(q)
            \end{split}
\end{equation*}
 with 
\begin{equation}\label{eq:UC}
 \begin{split}
  U_0(q)&=\frac{q^{2 p} [p q^2-q^4(p-1)-q^{2 p}]}{\mathcal{A}(q)},\\
  U(q)&=\frac{2 q^{3 p} \left(p \left(q^2-1\right)+1\right)-2 q^{p+4}}{\mathcal{A}(q)},\\
  U_1(q)&=\frac{q^4-q^{2 p} - p q^{2 p}  [(p-1) q^4+(3-2 p) q^2+p-2]}{\mathcal{A}(q)},\\
  \mathcal{A}(q)&=q^{4 p}-q^{2 p}[(p-1)^2 (1+q^4)-2 (p-2) p q^2] +q^4,
 \end{split}
\end{equation}
and where 
the symmetric function $I(y)$ was already defined in Chapter~\ref{chapter:intro}, see Eq.~\eqref{eq:IDef}. \\

 \begin{figure*}[t!]
    \centering
    \includegraphics[width=0.74\textwidth, trim={2 2 2 2},clip]{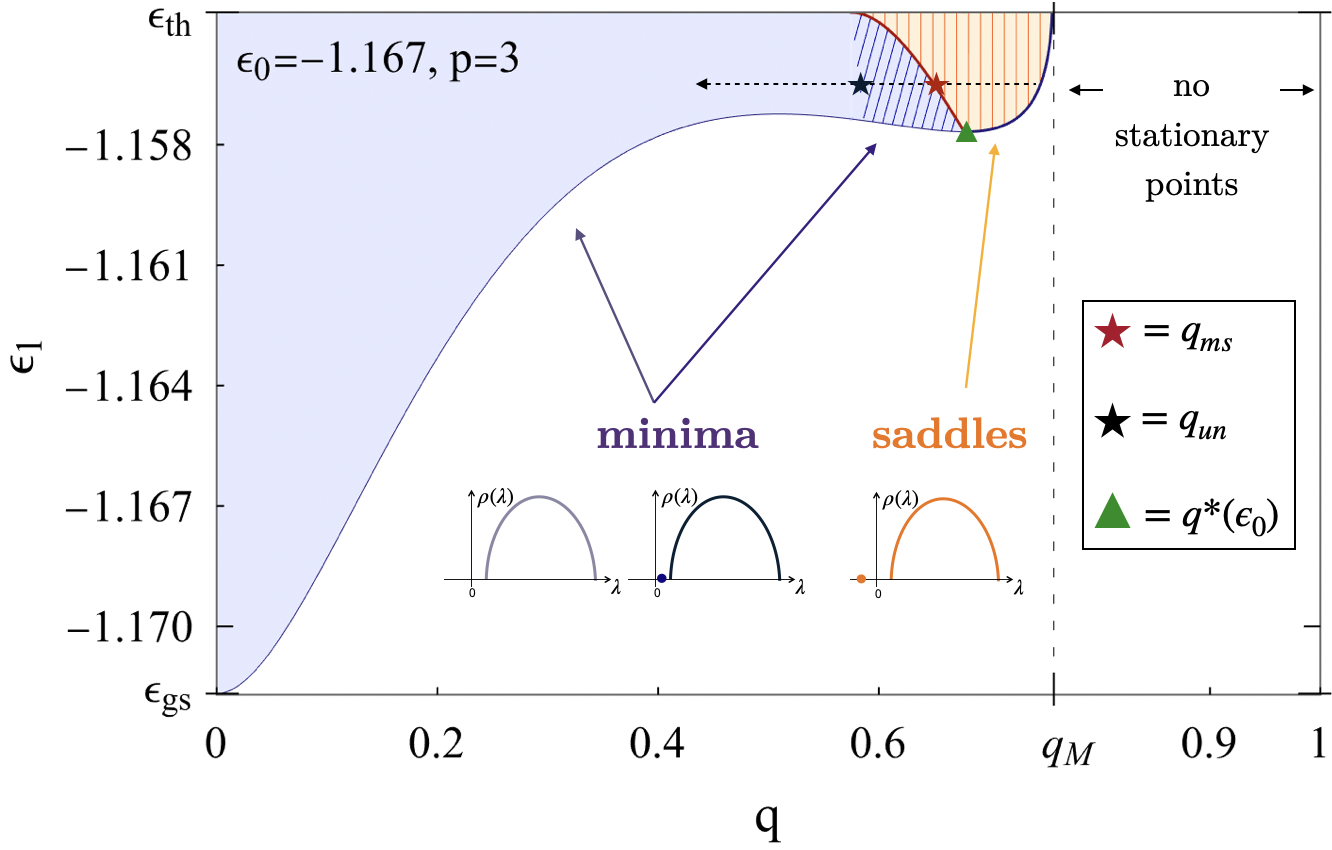}
    \caption{The plot shows in color the region in the plane $(q, \epsilon_1)$ where the two-point complexity is positive, for $\epsilon_0=-1.167$ with $p=3$. The blue leftmost area corresponds to minima, the dashed blue zone to minima with an isolated eigenvalue in the Hessian spectrum, the rightmost (dashed) yellow zone to rank-1 saddles. The stars mark the transitions in the properties of the Hessians of the stationary points, for fixed $\epsilon_1$. The red star indicates $q_{ms}$, i.e. when minima become saddles. The black star indicates $q_{un}$, i.e. when minima become correlated. In the white area typically no fixed points are found. The green triangle indicates the value $q^*(\epsilon_0)$ at energy $\epsilon^*(\epsilon_0)$, which is the first energy for which spikes are present in the spectrum of the Hessian at ${\bf s}_1$.}
    \label{fig:2_point_phase_diag}
\end{figure*}

\noindent The results of this calculation are summarized in Fig.~\ref{fig:2_point_phase_diag} for a representative value of $\epsilon_0 < \epsilon_{\rm th}$ and $p=3$. The colored region in the figure identifies the values of $q, \epsilon_1$ for which the function \eqref{eq:FinalComplexityPostSaddle} is positive (the plot is cutoff at $\epsilon_1=\epsilon_{\rm th}$, since at $\epsilon_1>\epsilon_{\rm th}$ the landscape is dominated by saddles with large index; this portion of the landscape is easily explored by relaxational dynamics, and it is therefore not of interest for our analysis). For $q=0$, the range of energy density is maximal, and extends down to the ground state energy $\epsilon_{\rm gs}$: at $q=0$ one has the largest two-points complexity, meaning that most of the stationary points of the landscape are at zero overlap with the reference one at energy $\epsilon_0$. In fact,  one has  
\begin{equation}
    \lim_{q \to 0}\Sigma^{(2)}(\epsilon_1,q|\epsilon_0)= \Sigma(\epsilon_1),
\end{equation}
meaning that one recovers the expression of the unconstrained complexity counting the number of minima irrespective of their location in configuration space. This is also intuitive, since by isotropy two vectors uniformly extracted from the hypersphere have a typical overlap of $0$ for $N>>1$. When $q$ increases, the range of energies at which one finds a positive complexity first decreases, and then increases again at the larger values of $q$, reaching a local maximum at a given $q^*(\epsilon_0)$, see Fig.~\ref{fig:2_point_phase_diag}. The maximal $q$ at which one finds stationary points at energy below the threshold one is $q=q_M(\epsilon_0)$. For each value of $\epsilon_1$, one can define the maximal overlap at which stationary points of that energy density are found: this is denoted with 
\begin{equation}\label{eq:qMAx}
q_M(\epsilon_1|\epsilon_0) := \text{max  } q \text{  such that } \Sigma^{(2)}(\epsilon_1,q|\epsilon_0) \geq 0.
\end{equation}

The different regions in Fig.~\ref{fig:2_point_phase_diag} are related to the linear stability of the stationary points found at those values of $q, \epsilon_1$; the linear stability is described by the spectrum of the Hessian matrices $\nabla^2_\perp \mathcal{E}$ at the stationary points, whose statistical properties are recalled in Sec.~\ref{app:AnnealedDistribution} in the annealed setting. A peculiar property found in~\cite{ros2019complexity} is that the \textit{quenched} and \textit{annealed} two-point complexities coincide, where the Replica Symmetric (RS) overlap between two replicas takes the value $q^2$ at the saddle point. An important consequence that they find is that the spectral properties of the Hessians in the annealed setting match those of the quenched.\\

\noindent The blue region in Fig.~\ref{fig:2_point_phase_diag} corresponds to stationary points ${\bf s}_1$ whose Hessian has all eigenvalues positive: these points are thus local minima of the landscape. The blue hatched area corresponds to minima whose Hessian has a single mode that is detached from the rest of the eigenvalues distribution (it is an isolated eigenvalue), that is smaller and whose eigenvector is partially aligned in the direction of the reference minimum of energy $\epsilon_0$; these minima thus display a softest curvature in the direction of the reference minimum, and we call them {\it correlated minima} to emphasize that their Hessians displays correlations with ${\bf s}_0$. Finally, the yellow area corresponds to rank-1 saddles, with one single Hessian mode that is negative and correlated with the direction of the reference minimum. These saddles are geometrically connected to the minimum, but also dynamically, meaning that the dynamics starting from the saddle relaxes to the local minimum~\cite{ros2021dynamical}. For each $\epsilon_1$, we denote with $q_{\rm un}(\epsilon_1|\epsilon_0)$ (where the subscript  “un" stands for  “uncorrelated") and  $q_{\rm ms}(\epsilon_1|\epsilon_0)$ (where the subscript  “ms" stands for  “minima-to-saddles") the overlaps at which the corresponding transitions occur, see Fig.~\ref{fig:2_point_phase_diag}. Remark also the overlap $q^*(\epsilon_0)$ in the figure, which is associated to an energy density $\epsilon^*(\epsilon_0)$, that plays a crucial role in our following discussions. Fig.~\ref{fig:2_point_phase_diag} shows that this is also the critical energy above which rank-1 saddles and correlated minima appear in the landscape in the vicinity of the reference minimum ${\bf s}_0$. How this critical energy depends on $\epsilon_0$ is illustrated in Fig~\ref{fig:enland_epsilon_star} for $p=3$. \\

\begin{figure}[t!]
\centering
\includegraphics[width=0.62
\textwidth, trim=5 5 5 5,clip]{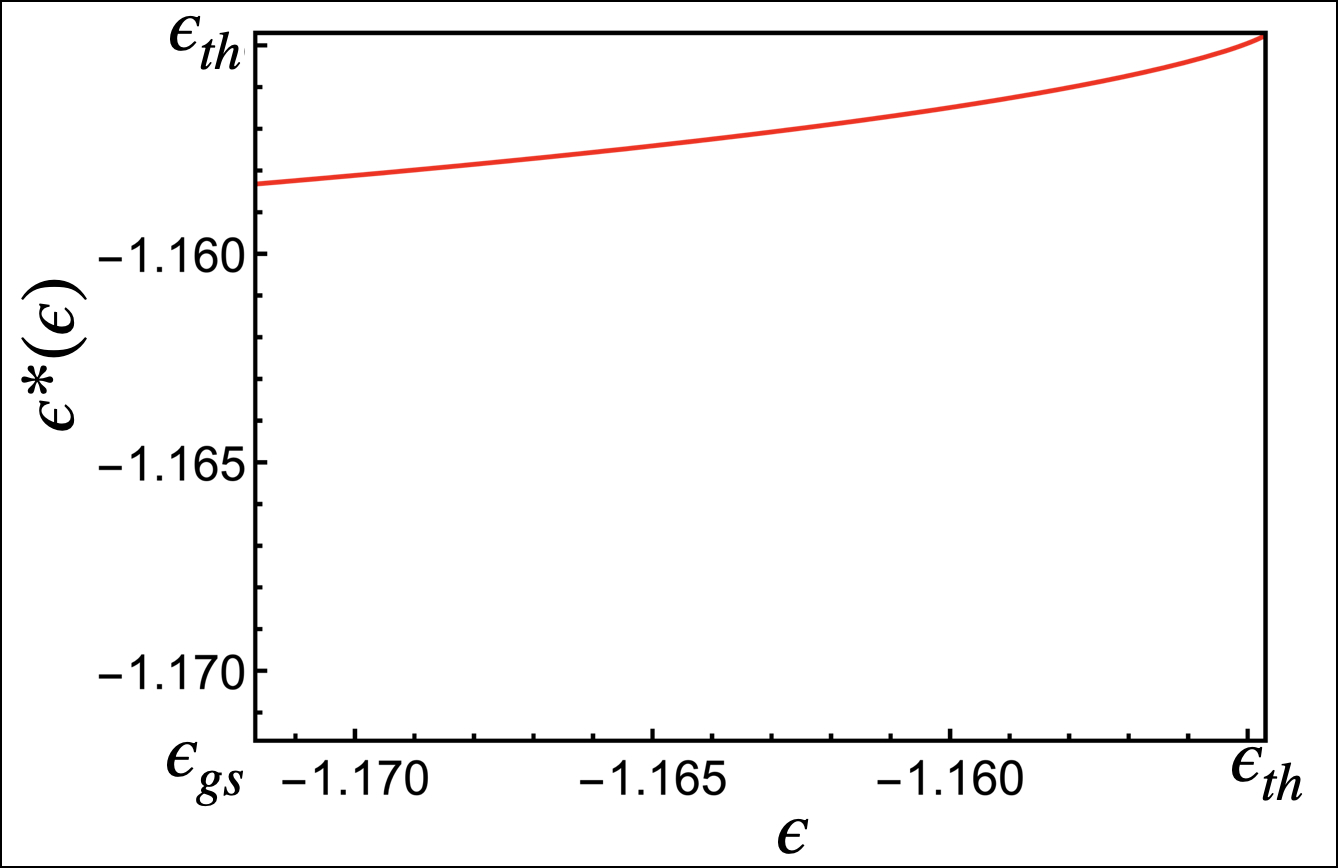}
\caption{Plot of the dependence of the energy $\epsilon^*(\epsilon)$ (energy corresponding to the green triangle in Fig.~\ref{fig:2_point_phase_diag}) for $p=3$.}
\label{fig:enland_epsilon_star}
\end{figure}

\subsection{Setting of the two-point complexity}
\label{sec:en_land_settin_two_point}
Let us give a brief summary of the setting in which the  two-point complexity is computed. Indeed, for our next discussion in Sec.~\ref{sec:curvature_driven_paths} on interpolating paths between local minima we need to have this setting under control. 

\subsubsection{Bases and tangent planes}
Like in Chapter~\ref{chapter:intro}, we work with the rescaled field on $\mathcal{S}_N(1)$ (the unit radius hypersphere in $\mathbb{R}^N$), defined as $h({\bm\sigma})=\sqrt{2/N}\,\,\mathcal{E}(\sqrt{N}{\bm\sigma})$. The general notation that we are trying to keep is that ${\bf s}_a$ refers to a configuration in $\mathcal{S}_N(\sqrt{N})$, whereas its corresponding ${\bm\sigma}_a={\bf s}_a/\sqrt{N}$ lives in $\mathcal{S}_N(1)$. Now, consider also the tangent plane $\tau[{\bm\sigma}]$ and local orthonormal basis $\mathcal{B}[{\bm\sigma}]$ as defined in Sec.~\ref{sec:chap1_topo_compl} for a generic configuration ${\bm\sigma}$. \\

\noindent The tangent plane $\tau[{\bm\sigma}]$, and therefore its basis vectors $\mathbf{e}_{\alpha}({\bm\sigma})$, depend on the particular point in configuration space that one is looking at. For two different configurations, ${\bm \sigma}_0$ and  ${\bm \sigma}_1$ at overlap $q={\bm \sigma}_0\cdot {\bm \sigma}_1$, it is convenient to choose the bases on the tangent planes $\tau[{\bm \sigma}_a]$, with $a=0,1$, as follows: first, one can always choose an orthonormal basis ${\bf x}_i$
in the $N$-dimensional space $\mathbb{R}^N$ in such a way that the first $N-2$ vectors ${\bf x}_1, \cdots, {\bf x}_{N-2}$ are orthogonal to both ${\bm \sigma}_0$ and ${\bm \sigma}_1$. These vectors belong to both tangent planes (because they are orthogonal to both the configuration vectors), and therefore one can choose ${\bf e}_i({\bm \sigma}_a):= {\bf x}_i$ for all $i=1, \cdots, N-2$. Concretely, in the ${\bf x}_i$ basis one can set ${\bm \sigma}_0=(0, 0, \cdots, 0,1)$ and ${\bm \sigma}_1=(0,0, \cdots, -\sqrt{1-q^2},q)$. Then, in the tangent plane $\tau[{\bm \sigma}_0]$ there remains a single  basis vector to be chosen, which will have a non-zero projection with ${\bm \sigma}_1$:
\begin{align}
\label{eq:def_basis_0}
 {\bf e}_{N-1}({\bm \sigma}_0):=\frac{q {\bm \sigma}_0-{\bm \sigma}_1}{\sqrt{1-q^2}}.
\end{align}
It has unit norm and it is orthogonal to all the others, since ${\bf  x}_i \perp {\bm \sigma}_a$ for any $i$ and $a=0,1$. The vector ${\bm \sigma}_0$ then completes the local basis: $\mathcal{B}[{\bm\sigma}_0]=\tau[{\bm\sigma}_0]\oplus\text{Span}({\bm\sigma}_0)$. With the choice above, $ {\bf e}_{N-1}({\bm \sigma}_0)=(0,0, \cdots, 1, 0)$ in the ${\bf x}_i$ basis. Similarly, $\tau[{\bm \sigma}_1]$ is spanned by ${\bf x}_1, \cdots, {\bf x}_{N-2}$ plus the vector ${\bf e}_{N-1}({\bm \sigma}_1)$ defined similarly as
\begin{align*}
 {\bf e}_{N-1}({\bm \sigma}_1):=\frac{q {\bm \sigma}_1-{\bm \sigma}_0}{\sqrt{1-q^2}}.
\end{align*}

Again, with the choice above we get $ {\bf e}_{N-1}({\bm \sigma}_1)=(0,0, \cdots, -q, -\sqrt{1-q^2})$.
In summary, we can choose ${\bf e}_\alpha ({\bm \sigma}_0)={\bf e}_\alpha ({\bm \sigma}_1)$  with  $\alpha=1, \cdots, N-2$ as a basis of the subspace  $\text{Span}\{{\bf x}_\alpha\}$, while the remaining basis vectors of the tangent planes are ${\bf e}_{N-1} ({\bm \sigma}_0)$ and ${\bf e}_{N-1} ({\bm \sigma}_1)$ respectively. Moreover, for practical purposes one can consider the $i-$th basis element of $\mathcal{B}[{\bm \sigma}_0]$ to be the vector with a 1 in the $i$-th slot and 0 elsewhere, and one can then express the elements of $\mathcal{B}[{\bm \sigma}_1]$ in that basis, as shown above. Let us refer to Sec.~\ref{sec:chap1_topo_compl} for the definitions of ${\bf g}$ and $\nabla_\perp h$, the gradient and Hessian restricted on the hypersphere. \\

\noindent We will often use the notation ${\bf e}_\alpha ({\bm \sigma}_a)={\bf e}_\alpha ^a$, as well as $h({\bm \sigma}_a)=h^a$,  ${\bf g}_\alpha({\bm\sigma}_a)={\bf g}_\alpha^a$ for simplicity. 

\subsubsection{Quick summary of the two-point complexity calculation}
For completeness, we provide a quick summary of the computation of Ref.~\cite{ros2019complexity}, upon which part of the present work is built. In Ref.~\cite{ros2019complexity} the authors compute the quenched complexity of stationary points ${\bm\sigma}_1$ of $h({\bm \sigma})$ at energy density $\epsilon_1$ (meaning that ${\bf g}({\bm\sigma}_1)=0$ and $h({\bm\sigma}_1)=\sqrt{2 N} \epsilon_1$), that are conditioned to be at fixed overlap ${\bm \sigma}_0 \cdot {\bm \sigma}_1= q$ with a given \emph{reference} (or \emph{primary}) minimum ${\bm\sigma}_0$ at energy density $\epsilon_0$. We refer to ${\bm \sigma}_1$ as the \emph{secondary} configuration. From Eq.~\eqref{eq:quenchedcomp2} we can compute the expected value of the logarithm by means of the replica trick; one needs to replicate the secondary configuration ${\bm\sigma}_1$:
\begin{align}\label{eq:moments}
\begin{split}
&\Sigma^{(2)}(\epsilon_1,q|\epsilon_0)=\lim_{N\to\infty}\lim_{n\to 0} \frac{M_n(\epsilon_1,q|\epsilon_0)-1}{Nn}\\
&M_n(\epsilon_1,q|\epsilon_0):=\mathbb{E}\left[
\frac{1}{\mathcal{N}(\epsilon_0)}\int_{\mathcal{S}_N(1)}d{\bm\sigma}_0\,\omega_{\epsilon_0}({\bm\sigma}_0)\int_{\mathcal{S}_N(1)}\prod_{a=1}^n d{\bm\sigma}_a\,\omega_{\epsilon_1,q}({\bm\sigma_a}|{\bm\sigma}_0)
\right]
\end{split}
\end{align}
where all symbols have been defined above, with the slight difference that here we use the rescaled variables $h$ and ${\bm\sigma}$ instead of $\mathcal{E}$ and ${\bf s}$. In order to carry out such computation, in principle one has to replicate the reference configuration ${\bm\sigma}_0$ as well by raising the denominator to the numerator, thus writing:
\begin{align}
M_n(\epsilon_1,q|\epsilon_0):=\lim_{k\to 0}\mathbb{E}\left[
\int_{\mathcal{S}_N(1)}\prod_{\beta=1}^k d{\bm\sigma}_{0,\beta}\,\omega({\bm\sigma}_{0,\beta})\int_{\mathcal{S}_N(1)}\prod_{a=1}^n d{\bm\sigma}_a\omega_{\epsilon_1,q}({\bm\sigma_a}|{\bm\sigma}_0),
\right]
\end{align}
where ${\bm\sigma}_{0,\beta=1}= {\bm\sigma}_{0}$.
Due to the isotropy of the correlation function of the random field $h({\bm \sigma})$, it turns out that this expectation value depends on the configurations ${\bm \sigma}_a, {\bm \sigma}_{0,\beta}$ only through the overlaps $q_{\alpha \beta}^0={\bm \sigma}_{0,\alpha} \cdot {\bm \sigma}_{0,\beta}$, $ q_{a \beta}={\bm \sigma}_{a} \cdot {\bm \sigma}_{0,\beta}$  and $q^1_{a b}={\bm \sigma}_{a} \cdot {\bm \sigma}_{b}$, where by construction $q_{\alpha \alpha}^0= q^1_{aa}=1$ and $ q_{a 1}={\bm \sigma}_{a} \cdot {\bm \sigma}_{0,\beta=1}=q$. This implies that the integral over the configurations can be replaced by an integral over these order parameters, which can be computed with the saddle point method. In \cite{ros2019complexity} it is shown that the saddle point equations for the parameters $q_{\alpha \beta}^0$ enforce $q_{\alpha \beta}^0=0$ for all $\alpha \neq \beta=1, \cdots, k$; this reflects the fact that the overlap between metastable states in the pure spherical $p$-spin model is vanishing \cite{crisanti1992sphericalp, cavagna1998stationary}, i.e. the fixed points are typically orthogonal to each others (cf. Sec.~\ref{sec:p_spin_model}). This also implies that $ q_{a \beta}=0$ for all $a=1, \cdots, n$ and $\beta \neq 1$. As a consequence, at the saddle point solution the secondary replicas ${\bm \sigma}_1$ are coupled only to the original reference configuration  ${\bm \sigma}_0$ and not to its replicas. Moreover this implies that to leading order in $N$, the expectation value \eqref{eq:moments} is identical to its annealed version, which is obtained averaging separately the numerator and the denominator. 
 This means that we can write
\begin{align}
\label{eq:2_point_ann_quench_calc}
\begin{split}
&M_n(\epsilon_1,q|\epsilon_0)=\frac{1}{\mathbb{E}[\mathcal{N}(\epsilon_0)]}\mathbb{E}\left[
\int_{\mathcal{S}_N(1)}d{\bm\sigma}_{0}\,\omega_\epsilon({\bm\sigma}_{0})\int_{\mathcal{S}_N(1)}\prod_{a=1}^n d{\bm\sigma}_a\omega_{\epsilon_1,q}({\bm\sigma_a}|{\bm\sigma}_0)
\right]\\
&=\mathbb{E}\left[
\int_{\mathcal{S}_N(1)}\prod_{a=1}^n d{\bm\sigma}_a\omega_{\epsilon_1,q}({\bm\sigma_a}|{\bm\sigma}_0)\bigg| \;  \begin{subarray}{l}
 h({\bm \sigma}_0) = \sqrt{2 N}\epsilon_0\\ 
 {\bf g}({\bm \sigma}_0)={\bf 0} \end{subarray}\right]\\
&=\int_{\mathcal{S}_N(1)}\prod_{a=1}^n d{\bm\sigma}_a\,\mathbb{E}\left[\prod_{a=1}^n \,   |\text{det} \nabla^2_\perp h({\bm \sigma}_a)|  \; \Big| \;   \grafe{
 \begin{subarray}{l}
 h^a = \sqrt{2 N}\epsilon_1, \,\,h^0=\sqrt{2 N} \epsilon_0\\ 
 {\bf g}^a={\bf 0}, \,\, {\bf g}^0={\bf 0} \; \forall  a=1, ...,n \end{subarray}}\right]  P_{\vec{{\bm \sigma}}|{\bm \sigma}_0}({\bf 0}, \epsilon_1)
\end{split}
\end{align}
where $\vec{\bm \sigma}:=({\bm \sigma}_1, \cdots, {\bm \sigma}_n)$ and 
\begin{align}
\label{eq:app:jpdf}
P_{\vec{{\bm \sigma}}|{\bm \sigma}_0}({\bf 0}, \epsilon_1):=\mathbb{E}\quadre{
\prod_{a=1}^n\delta(h^a-\sqrt{2N}\epsilon_1) \delta( {\bf g}^a) \; \Big| \;  \begin{subarray}{l}
 h^0=\sqrt{2 N} \epsilon_0\\ 
 {\bf g}^0={\bf 0} \end{subarray}}.
\end{align}
In the second line of \eqref{eq:2_point_ann_quench_calc} we have canceled the contribution of ${\bm\sigma}_0$ from denominator and numerator, while abusing the notation (we keep ${\bm\sigma}_0$ inside the integral). This can be done because the dependence on ${\bm\sigma}_0$ inside the integral remains only through the overlap $q$, and after carrying out the computations, one can see that the contributions from the integral on ${\bm\sigma}_0$ cancel exactly the denominator $\mathbb{E}[\mathcal{N}(\epsilon_0)]$. The expectation value is now a function of the overlaps $q^1_{ab}$, and its leading order term can be determined again with a saddle point calculation. Searching for a saddle point solution where all replicas have the same overlap (RS ansatz), $q^1_{ab}= q_1$ for all $a,b=1, \cdots n$ and $a \neq b$, one finds the solution $q_1=q^2$ \cite{ros2019complexity}, which is the smallest possible overlap between replicas that are all subject to the constraint of having overlap $q$ with the reference configuration. It can be shown explicitly that this implies that the complexity $\Sigma^{(2)}(\epsilon_1,q|\epsilon_0)$ computed within the quenched formalism is the same as that obtained within the annealed framework, i.e., setting $n=1$ in the formulas above. This means that:
\begin{align}\label{eq:QuenchedEqualAnnealed}
\Sigma^{(2)}(\epsilon_1,q|\epsilon_0)=\lim_{N\to\infty}\frac{1}{N}\mathbb{E}\left[ \log\mathcal{N}_{{\bm\sigma}_0}(\epsilon_1,q|\epsilon_0) \right]_0 = \lim_{N\to\infty}\frac{1}{N}\log \mathbb{E}\left[ \mathcal{N}_{{\bm\sigma}_0}(\epsilon_1,q|\epsilon_0)\right]_0.
\end{align}
Moreover, given our discussion on the fact that also the average over ${\bm\sigma}_0$ is annealed (that is, we can factor out the denominator), the final result is actually \textit{doubly annealed}, meaning that:
\begin{align}
\Sigma^{(2)}(\epsilon_1,q|\epsilon_0) = \lim_{N\to\infty}\frac{1}{N}\log \mathbb{E}_{2A}\left[ \mathcal{N}_{{\bm\sigma}_0}(\epsilon_1,q|\epsilon_0)\right]
\end{align}
where 
 \begin{equation}
  \label{eq:doubly_annealed_2point}
 \mathbb{E}[ \cdot ]_{2A}:= \frac{\mathbb{E}\left[
\int d{\bm \sigma}_0  \,  \omega_{\epsilon_0}({\bm\sigma}_0)\,\,\cdot\right]}{\mathbb{E}[\mathcal{N}(\epsilon_0)] }.
\end{equation}

\noindent Given this result, and given that in our work in Ref.~\cite{pacco2024curvature} we have seen that calculating paths in the annealed and quenched settings gives the same result, we will concentrate from now on only on the annealed setup for the two-point complexity. However, as we showed in Refs.~\cite{pacco_quenched_triplets_2025, pacco_triplets_2025}, the same simplification does not hold, in general, for the three-point complexity.

\subsection{Statistics of the Hessians: the annealed setup}\label{app:AnnealedDistribution}
Let us briefly discuss the statistical distribution of the unconstrained Hessian matrices  $\nabla^2 h({\bm \sigma}_a)$ with $a=0,1$, subject to the conditioning on the energies and gradients of the two configurations. Let us remark that if we stick to the \textit{doubly annealed} setting of Eq.~\eqref{eq:doubly_annealed_2point}, then ${\bm\sigma}_0$ and ${\bm\sigma}_1$ are on the same footing, and we need to take into account their mutual conditioning. As we have seen previously in Chapter~\ref{chapter:intro}, the unconstrained Hessian $\nabla^2 h$ and the Riemannian Hessian $\nabla^2_\perp h$ are easily related by a projection and shift. The main goal of presenting this section is to motivate the computations in Chapter~\ref{chapter:rmt_} about overlaps between eigenvectors of spiked, correlated GOE random matrices.

\subsubsection{Matrix distribution after conditioning}
In this case, one has to determine the (Gaussian) joint distributions of the two Hessian matrices $\nabla^2 h({\bm \sigma}_a)$, each one conditioned to  ${\bf g}({\bm \sigma}_a)={\bf 0}$ and $h({\bm \sigma}_a)= \sqrt{2 N} \epsilon_a$ for $a=0,1$. This conditional, joint distribution has been determined in \cite{ros2019complexity,subag2017complexity}. Here the conditioned Hessians are indicated by adding an upper tilde "$\,\,\tilde{}\,\,$". Each matrix $\tilde{\nabla}^2 h({\bm \sigma}_a)\in\mathbb{R}^{N\times N}$, expressed in its own local basis $\mathcal{B}[{\bm\sigma}_a]$ has the following block structure:   
\begin{equation}
\label{eq:app:annealed_matrix_0}
\frac{\tilde{\nabla}^2 h({\bm \sigma}_a)}{\sqrt{N-1}} = \begin{pmatrix}
 & & &m^a_{1 N-1}&0\\
 \\
& {\bf B}^a&&\vdots &0\\
\\
 & & &m^a_{N-2 \,N-1}&0\\
m^a_{1 N-1}& \cdots &m^a_{N-2 \,N-1}&m^a_{N-1 \,N-1}+\mu_a&0\\
0&0&0&0&  \sqrt{\frac{2 N}{N-1}} \, p(p-1)\epsilon_a
\end{pmatrix}.
\end{equation}
This expression can be found from the covariance structure between gradients and Hessians found in Sec.~\ref{sec:chap1_topo_compl} of Chapter~\ref{chapter:intro}, and by considering the formulas for Gaussian conditioning (see Appendix~\ref{app:sec_gausian_conditioning}).\\

\noindent The entries in the $(N-2) \times (N-2)$ blocks ${\bf B}^a$ are independent of the entries $m^a_{i N-1}$ in the "special" row and column. They are all zero-mean with correlations given by:
\begin{equation}\label{eq:b_correlations2}
\begin{split}
   & \mathbb{E}[B_{ij}^a \, B_{kl}^b]= \tonde{\delta_{a b} \frac{\sigma^2}{N-1}+ (1-\delta_{ab})\frac{\sigma^2_H}{N-1}}(\delta_{ik} \delta_{jl}+ \delta_{il} \delta_{jk}), \quad \quad \\
    & \mathbb{E}[m_{iM}^a \, m_{jM}^b]= \tonde{\delta_{a b} \frac{\Delta^2}{N-1}+ (1-\delta_{ab})\frac{\Delta^2_h}{N-1}}\delta_{ik} \quad \quad i,j <N-1, \\
    & \mathbb{E}[m_{N-1N-1}^a m_{N-1N-1}^b]= \frac{v_{ab}}{N-1}, \quad \quad a, b \in \{ 0,1 \}.
    \end{split}
\end{equation}
 
Therefore, the blocks ${\bf B}^a$ are $(N-2) \times (N-2)$ matrices with GOE (Gaussian Orthogonal Ensemble) statistics \footnote{In Chapter~\ref{chapter:rmt_}, Sec.~\ref{sec:Goe_intro}, we give a short introduction about GOE matrices, before computing overlaps between eigenvectors of matrices of the form \eqref{eq:app:annealed_matrix_0}.}, with rescaled variance $\sigma^2 (N-2)/(N-1)$. The two blocks ${\bf B}^0,{\bf B}^1$ are coupled component-wise. The parameters $\sigma, \sigma_H, \Delta, \Delta_h$ and $\mu_a$ are functions of $p$ and of the parameters $\epsilon_a, q$, and read \cite{ros2019complexity}: 

\begin{equation}\label{eq:app:FormPar}
\begin{split}
&\sigma^2=p(p-1)\\
&\sigma_H^2=p(p-1)q^{p-2}\\
&\sigma_W^2=\sigma^2-\sigma_H^2=p(p-1)(1-q^{p-2})\\
&\Delta^2=p(p-1) \quadre{1-\frac{(p-1)(1-q^2) q^{2p-4}}{1-q^{2p-2}}}\\
&\Delta_h^2=p(p-1) q^{p-3} (-1)\quadre{1-\frac{(p-1) (1-q^2)}{1-q^{2p-2}}}\\
&\mu_1= \mu(\epsilon_0, \epsilon_1), \quad \mu_0= \mu(\epsilon_1, \epsilon_0) 
\end{split}
\end{equation}
where: 
\begin{equation}\label{eq:app:Mu}
 \begin{split}
&\mu(x, y) := \sqrt{2} p(p-1) \left(1-q^2\right)\times \\
&\times\frac{ [q^4- (p-1)q^{2 p}+ (p-2)q^{2 p+2}] x-[q^{3 p}+(p-2) q^{p+2}-(p-1) q^{p+4}] y }{q^{6-p}+q^{3 p+2}- q^{p+2}[(p-1)^2 (1+q^4)-2 p (p-2) q^2]}.\\
 \end{split}
\end{equation}
The conditional distribution of the Riemannian Hessians then reads as follows in the $\tau[{\bm\sigma}_a]$ basis:
\begin{align}
\begin{split}
\label{eq:app:annealed_matrix_1a}
\frac{\tilde{\nabla}^2_\perp h({\bm \sigma}_a)}{\sqrt{N-1}} &= \frac{\tilde{\mathcal{H}}({\bm\sigma}_a)}{\sqrt{N-1}} -\sqrt{\frac{2 N}{N-1}} \, p\epsilon_a\mathbb{I}_{N-1}
\end{split}
\end{align}
the last term being the shift, and where we defined the following matrix in the $\tau[{\bm\sigma}_a]$ basis:
\begin{equation}
\label{eq:app:TildeMatDef} 
\frac{\tilde{\mathcal{H}}({\bm \sigma}_a)}{\sqrt{N-1}} := \begin{pmatrix}
 & & &m^a_{1 N-1}\\
 \\
& {\bf B}^a&&\vdots \\
\\
 & & &m^a_{N-2 \,N-1}\\
m^a_{1 N-1}& \cdots &m^a_{N-2 \,N-1}&m^a_{N-1 \,N-1}
\end{pmatrix} +\mu_a\,{\bf e}_{N-1}^a [{\bf e}_{N-1}^a]^T
\end{equation}
These conditioned matrices can be though of as matrices with a GOE statistics, perturbed with both an additive and multiplicative finite rank perturbation along the direction corresponding to the basis vector ${\bf e}_{N-1}({\bm \sigma}_a)$, and shifted by a diagonal matrix. Notice that in the (doubly) annealed formalism the random matrix problem is symmetric, in the sense that the two conditioned Hessians have the same structure at both the configurations ${\bm \sigma}_0$ and ${\bm\sigma}_1$. \\

\noindent Let us now consider specifically the case $p=3$, which in practice is the one we will analyse for the pathways. One sees from \eqref{eq:app:FormPar} that for $p=3$ it holds
$\Delta^2=\Delta_h^2= 6 (1-q^2)(1+q^2)^{-1}$,
which implies that $m^0_{iN-1}=m^1_{iN-1}\equiv m_{iN-1}$. Moreover, in this case $\mu_1= 6 \sqrt{2} q \tonde{\epsilon_0-q \epsilon_1}(1-q^2)^{-1}$ and $\mu_0= 6 \sqrt{2} q \tonde{\epsilon_1-q \epsilon_0}(1-q^2)^{-1}$. \\

\noindent The statistics of two jointly distributed Gaussian random matrices of the form $\tilde{\mathcal{H}}^0,\tilde{\mathcal{H}}^1$ (the upper index being used for simplicity instead of ${\bm\sigma}_a$) is the object of study of Chapter~\ref{chapter:rmt_}, and of our work \cite{paccoros}, where in that Chapter the covariances between elements take up the most general form. There, we will derive both the spectral properties (in particular, the isolated eigenvalues) and the overlaps between eigenvectors associated both to bulk and isolated eigenvalues of the pair of matrices. Below we state the main results for completeness, as they will be useful in Sec.~\ref{sec:curvature_driven_paths} when we discuss curvature driven pathways between local minima.

\noindent A brief but important detail from Chapter~\ref{chapter:rmt_} is that the particular normalizations (by $N$, or $N-1$ or $N-2$) do not matter when we take the limit of $N\to\infty$, as they only give sub-leading corrections. In the same fashion, also the covariance $v_{ab}$ of the last elements in the $(N-1)\times (N-1)$ blocks does not matter to leading order. Hence, the blocks ${\bf B}^a$ can effectively be considered as GOE matrices of variance $\sigma^2$.

\subsubsection{Spectral statistics and isolated eigenvalue}
Here we give the quick summary of the spectral properties of $\tilde{\mathcal{H}}^a$. To leading order in the size of the matrix, the eigenvalue density $\rho_N^a(\lambda)$ of both matrices is not affected  by the presence of the special line and column, and for $N\to\infty$ it just coincides with the eigenvalue density of the GOE block, i.e., it is given by the Wigner's semicircular law (see also Sec.~\ref{sec:chap1_topo_compl}):
\begin{align*}
\rho_\sigma(\lambda):=\frac{1}{2 \pi \sigma^2} \sqrt{4 \sigma^2 - \lambda^2}\,\mathds{1}_{|\lambda|\leq2\sigma}, \quad \quad \sigma^2=p(p-1).
\end{align*}

 The presence of the special row and column can give rise to subleading contributions to the eigenvalue density: these contributions correspond to eigenvalues that do not belong to the support of the semicircular law (and are said to be ``isolated"), and whose typical value depends on the parameters $\Delta, \mu_a$ governing the statistics of the entries of the special row and column. As explained in Sec.~\ref{sec:rmt_theoretical_res}, the fact that $\Delta \leq \sigma$ (as can be easily verified to be the case here), implies that only one isolated eigenvalue can exist for these matrices. Such eigenvalue exists whenever 
\begin{equation}
\label{eq:app:ConditionEx}
|\mu_a|>  \sigma \tonde{1+ \frac{\sigma^2-\Delta^2}{\sigma^2}}.
\end{equation}
Explicitly, the eigenvalue is given by:
\begin{equation}
\label{eq:app:IsoExplicit}
\lambda_{\rm iso}^a=\frac{{2 \mu_a \sigma^2- \Delta^2 \mu_a- \text{sign}(\mu_a) \Delta^2 \sqrt{\mu_a^2-4 (\sigma^2- \Delta^2)}}}{2 (\sigma^2-\Delta^2)}.
\end{equation}
Notice that for $\mu_a <0$, this eigenvalue is negative and it coincides with the smallest one of the matrix, i.e., $\lambda_{\rm \min}^a= \lambda_{\rm iso}^a$.  
Whenever the isolated eigenvalue exists, its eigenvector ${\bf e}^a_{\rm iso}$ has a projection on the vector ${\bf e}_{N-1}({\bm\sigma}_a)$, corresponding to the special line and column of the matrix, which remains of $\mathcal{O}(1)$ when $N$ is large. The typical value of this projection has been computed in Refs.~\cite{paccoros, ros2020distribution} and reads:
\begin{equation}
\label{eq:app:ProjVectors}
\begin{split}
  &u^a:=({\bf e}^a_{\rm iso} \cdot  {\bf e}_{N-1}^a)^2= \mathfrak{q}_{\sigma,\Delta}(\lambda_{\rm iso}^a,\mu_a)
   \end{split}
\end{equation}
where we introduced the function:
\begin{equation}\label{eq:Qfunc}
\begin{split}
   & \mathfrak{q}_{\sigma,\Delta}(\lambda,\mu):=\text{sign}(\mu)\,\frac{\text{sign}(\lambda)\Delta^2\sqrt{\lambda^2-4\sigma^2}-\lambda(2\sigma^2-\Delta^2)+2\mu\sigma^2}{2\Delta^2\sqrt{\mu^2-4(\sigma^2-\Delta^2)}}.
    \end{split}
\end{equation}
For more details on these results, and for some derivations, please refer to Chapter~\ref{chapter:rmt_}. \\

\noindent It can be shown that Eq.~\eqref{eq:app:ProjVectors} is indeed positive on the Right Hand Side, as it should be. Moreover, whenever the squared overlap is non-zero, the eigenvector is aligned with the direction of the finite-rank perturbation, meaning that ${\bf e}^a_{\rm iso} \cdot  {\bf e}_{N-1}^a>0$. For the $p$-spin model with $p=3$, it can be easily shown that $\mu_a<0$ and thus the isolated eigenvalue is the smallest one: $\lambda_{\rm \min}^a= \lambda_{\rm iso}^a$.

\section{Curvature-driven pathways}
\label{sec:curvature_driven_paths}
In this section we pursue the goal of estimating energy barriers between local minima of the pure spherical $3-$spin model by studying the typical energetic profile along curvature driven pathways that interpolate between two local minima ${\bm\sigma}_0, {\bm\sigma}_1 \in \mathcal{S}_N(1)$.

\subsection{Interpolating paths and energy profiles.}

\begin{figure*}[t!]
\centering
\includegraphics[width=0.57
\textwidth, trim=7 7 7 7,clip]{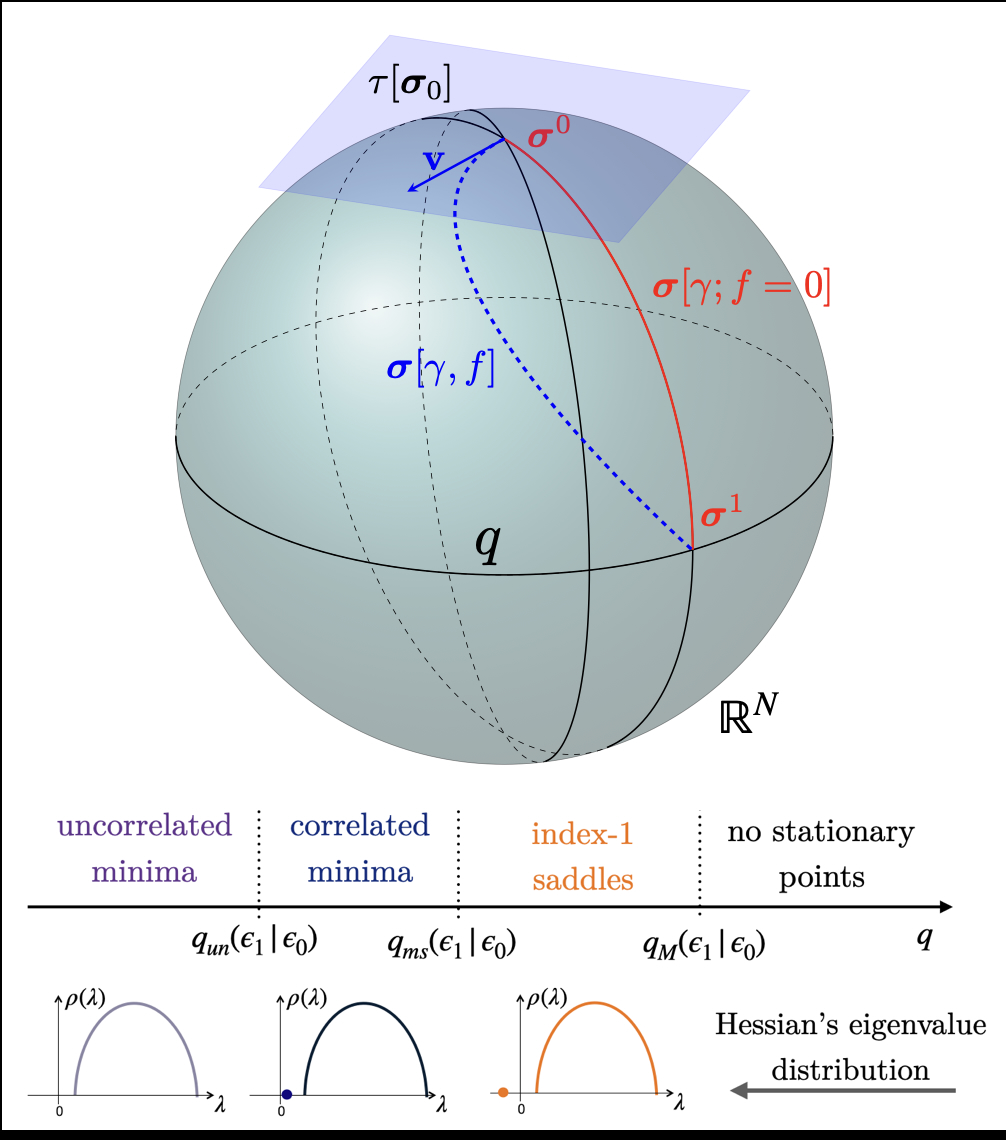}
\caption{ {\it Top.} Sketch of the configuration space (sphere), with a geodesic path (red thick line) and a perturbed path (blue dashed line) interpolating between two configurations ${\bm \sigma}_a$ with $a=0,1$.  The tangent plane $\tau[{\bm \sigma}_0]$ to which ${\bf v}$ belongs is also sketched. {\it Bottom.} Eigenvalue distribution of the Hessian at ${\bm \sigma}_1$ as a function of the overlap $q$, for fixed $\epsilon_0$ and $\epsilon_1>\epsilon^*(\epsilon_0)$ (see also Fig.~\ref{fig:2_point_phase_diag}).}
\label{fig:SketchPath}
\end{figure*}

\noindent For a fixed realization of the landscape, we consider two configurations ${\bm \sigma}_0$  and ${\bm \sigma}_1$ drawn at random from the population of stationary points of $h({\bm \sigma})$ such that ${\bf g}({\bm \sigma}_a)=0$ for $a=0,1$. We extract the stationary points in such a way that each ${\bm \sigma}_a$ has energy density $\epsilon_a \in [\epsilon_{\rm gs}, \epsilon_{\rm th}]$ for $a=0,1$, and their overlap ${\bm \sigma}_0 \cdot {\bm \sigma}_1$ equals to some  $q\in[0,1]$. We are interested in the energy density profile along paths lying on the surface of $\mathcal{S}_N(1)$, which interpolate between the two stationary points. We parametrize the paths as follows:
\begin{align}
\label{eq:path}
    {\bm\sigma}[\gamma;f]=\gamma{\bm \sigma}_1 + \beta[\gamma; f]\; {\bm \sigma}_0 + f(\gamma){\bf v},\quad \gamma\in[0,1],
\end{align}
where $f$ is a continuous function such that $f(0)=f(1)=0$, and $\bf v$ is a norm-1 vector orthogonal to both ${\bm \sigma}_0$ and ${\bm \sigma}_1$. 
The condition that ${\bm \sigma}[\gamma;f]$ lies on $\mathcal{S}_N(1)$ enforces:
\begin{equation}
\label{eq:beta}
({\bm\sigma}[\gamma,f])^2=1\Rightarrow\beta[\gamma; f] = - \gamma q+ \sqrt{1- \gamma^2(1-q^2)- f^2(\gamma)},
\end{equation}
where we chose the sign $+$ to ensure that ${\bm\sigma}[0;f]={\bm\sigma}_0$ and ${\bm\sigma}[1;f]={\bm\sigma}_1$. Notice that this quantity has to be real, and this imposes some constraints on $f(\gamma)$, namely:
\begin{align}
f^2(\gamma)\leq 1-\gamma^2(1-q^2).
\end{align}
When $f \equiv 0$, Eq.~\eqref{eq:path} gives the geodesic path connecting the two stationary points. The function  $f$ acts as a perturbation of the geodesic, along the direction identified by the vector ${\bf v}$, see Fig.~\ref{fig:SketchPath}. Therefore, we are restricting to interpolating paths belonging to the low-dimensional subspace spanned by the vectors ${\bf v}, {\bm \sigma}_0$ and ${\bm\sigma}_1$, intersected with the surface of the hypersphere. 
We aim at computing the typical energy density profile 
\begin{equation}\label{eq:Profile}
\epsilon_{\bf v}[\gamma;f]:= \lim_{N \to \infty}  \mathbb{E}\quadre{\frac{h({\bm\sigma}[\gamma;f])}{\sqrt{2N}}}_{0,1},
\end{equation}
where the average is over the distribution of stationary points ${\bm \sigma}_0, {\bm \sigma}_1$ with energy densities $\epsilon_0$, $\epsilon_1$ and overlap $q$ (denoted with a subscript "$0,1$"), and over the realizations of the landscape. Below in Sec.~\ref{sec:en_land_quench_vs_ann_consid} we shall give a more detailed explanation of how these averages are taken. For now, let us specify the choice of ${\bf v}$.\\

\noindent To specify our choice of ${\bf v}$ in \eqref{eq:path}, we recall that we introduced the vectors:
\begin{align*}
    {\bf e}_{N-1}({\bm \sigma}_0)= \frac{q \; {\bm \sigma}_0-{\bm \sigma}_1}{\sqrt{1-q^2}},\quad
    {\bf e}_{N-1}({\bm \sigma}_1)= \frac{q \; {\bm \sigma}_1-{\bm \sigma}_0}{\sqrt{1-q^2}}.
\end{align*}
Each $ {\bf e}_{N-1}({\bm \sigma}_a)$ belongs to the tangent plane $\tau[{\bm \sigma}_a]$ (which means that ${\bf e}_{N-1}({\bm \sigma}^a) \perp {\bm \sigma}^a$), and identifies the direction in the tangent plane pointing towards the other configuration ${\bm \sigma}_{b \neq a}$. We also recall that  ${\bf x}_i$ with $i=1, \cdots, N-2$ is an arbitrary orthonormal basis of the subspace orthogonal to both ${\bm \sigma}_0,{\bm\sigma}_1$. Finally, we denote with $\lambda_{\rm min}^a$ the minimal eigenvalue of $\nabla^2_\perp h({\bm \sigma}_a)$ and with  
${\bf e}_{\rm min}({\bm \sigma}_a)$ the associated eigenvector. We also recall that in Eq.~\ref{eq:app:ProjVectors}, we have defined:
\begin{align*}
    u^a=({\bf e}_{\rm min}({\bm \sigma}_a)\cdot {\bf e}_{N-1}({\bm \sigma}_a))^2.
\end{align*}
We consider two possible choices for ${\bf v}$. The first corresponds to
\begin{equation}
\label{eq:v_start}
{\bf v}\to{\bf v}_{\rm soft}^a:= \frac{{\bf e}_{\rm min}({\bm \sigma}_a)- \sqrt{u^a}\,  {\bf e}_{N-1}({\bm \sigma}_a)}{\sqrt{1-u^a}},
\end{equation}
i.e., the path is deformed in the directions of the softest curvature of the energy landscape at ${\bm \sigma}_0$ or ${\bm \sigma}_1$, and $({\bf v}_\text{soft}^a)^2=1$. The second corresponds to 
\begin{equation}
\label{eq:v_start2}
{\bf v}\to{\bf v}_{\rm Hess}:= Z \sum_{i=1}^{N-2} [{\bf x}_i \cdot \tilde{\mathcal{H}}({\bm \sigma}_0) \cdot {\bf e}_{N-1}({\bm \sigma}_0)] \, {\bf x}_i
\end{equation}
where  $\tilde{\mathcal{H}}({\bm \sigma}_a)$ was defined in Eq.~\eqref{eq:app:TildeMatDef} (and it essentially denotes the unshifted Hessian). This choice of ${\bf v}_{\rm Hess}$ is motivated by the study of the gradient vector ${\bf g}$ at each configuration ${\bm \sigma}(\gamma)$ along the geodesic path (i.e. when $f\equiv 0$). As we show in Appendix~\ref{app:gradient_max_barr}, at each point the gradient has a tangent component ${\bf g}^\parallel(\gamma)$ to the path, and an orthogonal component ${\bf g}^\perp(\gamma)$ that is proportional to \eqref{eq:v_start2}. 
While ${\bf g}^\parallel$ obviously vanishes at the value of $\gamma$ that corresponds to the local maximum of the geodesic energy profile, ${\bf g}^\perp$ does not, meaning that the maximum of the geodesic profile is not a stationary point of $h({\bm \sigma})$. 
This suggests that interpolating paths associated to lower barriers can be found by deforming the geodesic path in the direction of ${\bf g}^\perp$, since this is likely the direction that the path would follow if it was allowed to relax in configuration space by gradient descent \cite{jonsson1998nudged, bolhuis2002transition, freeman2016topology}.

\subsection{Quenched vs Annealed: general considerations}
\label{sec:en_land_quench_vs_ann_consid}
The meaning of the terms \textit{quenched}, \textit{annealed}, \textit{doubly-annealed} can be quite different depending on the context. In this paragraph we wish to make things more clear so that we will not need to enter into too much detail later on. In Sec.~\ref{sec:TwoPoint} the term \textit{quenched} meant that things were done properly, meaning that the $\log$ stays inside the expectation, and the expectation is done in order: first we average over the distribution of stationary points ${\bm\sigma}_0$ at fixed disorder, and then we average over the disorder. By \textit{annealed} we meant to exchange the $\log$ and the expectation; by \textit{doubly-annealed} we meant that the expectation over ${\bm\sigma}_0$ could be factored (i.e. no need to replicate it). \\

\noindent In the present context of curvature-driven pathways, these meanings are a bit different. Indeed, \textit{quenched} means that we average the path in order over: points ${\bm\sigma}_1$, points ${\bm\sigma}_0$, the disorder. \textit{Annealed} instead means that we factor out the first denominator from configuration ${\bm\sigma}_0$; \textit{doubly-annealed} means that we factor out also the second one, over ${\bm\sigma}_1$. Moreover, in general, some of these simplifications might be correct, in other cases not. Let us therefore put on some firmer grounds our definition of \textit{quenched}, \textit{annealed}, \textit{doubly-annealed}:
\begin{itemize}
    \item \textit{quenched}: with this term we mean no simplification whatsoever in the way the average should be done, so in the case of complexities we keep the $\log$ outside and we average over \textit{primary}, \textit{secondary} (and eventually \textit{tertiary}) configurations in the correct order and without factoring out the denominators; if no $\log$ is present, such as in the case of the interpolating paths, we keep the correct order of averaging the configurations. 

    \item \textit{annealed}: one step of simplification is done. By this we mean that if there is a $\log$ then we exchange the $\log$ with all averages; if there is no $\log$ then we just factor out the first denominator.

    \item \textit{doubly-annealed}: we could clearly define a $k-annealed$ calculation when $k$ simplifications are done. However, to make things simple, depending on the context, the \textit{doubly-annealed} calculation means that all possible simplifications are done: $\log$'s and averages are exchanged, and all configurations are on the same footing (meaning that we factor out all denominators and no-one gets replicated).
\end{itemize}
With these definitions we will not need to specify at each time what the various averages mean, and we will assume it from the context under study. \\

\noindent For the problem at hand, the quenched computation reads:
\begin{align}
\label{eq:en_land_E_0_1_h}
\begin{split}
    \mathbb{E}[h({\bm\sigma}[\gamma;f])]_{0,1}&:=\mathbb{E}\Bigg[\frac{1}{\mathcal{N}(\epsilon_0)}\int_{\mathcal{S}_N(1)}d{\bm\sigma}_0\,\omega_{\epsilon_0}({\bm\sigma}_0)\frac{1}{\mathcal{N}_{{\bm\sigma}_0}(\epsilon_1,q|\epsilon_0)}\\
    &\times \int_{\mathcal{S}_N(1)}d{\bm\sigma}_1\,\omega_{\epsilon_1,q}({\bm\sigma}_1|{\bm\sigma}_0)\,\,h({\bm\sigma}[\gamma;f])\Bigg].
\end{split}
\end{align}
The \textit{doubly-annealed} computation then corresponds to:
\begin{align}
\begin{split}
    \mathbb{E}_{2A}[h({\bm\sigma}[\gamma;f])]&:=\frac{\mathbb{E}\left[\int_{\mathcal{S}_N(1)}d{\bm\sigma}_1\,\omega_{\epsilon_1,q}({\bm\sigma}_1|{\bm\sigma}_0)\,\,h({\bm\sigma}[\gamma;f])\bigg|\begin{subarray}{l}
 h^0=\sqrt{2 N} \epsilon_0\\ 
 {\bf g}^0={\bf 0} \end{subarray}\right]}{\mathbb{E}[\mathcal{N}_{{\bm\sigma}_0}(\epsilon_1,q|\epsilon_0)]}
\end{split}
\end{align}
where the terms involving ${\bm\sigma}_0$ at the numerator and denominator cancel, but we leave a conditioning inside the expectation. This simplification is a consequence of the isotropy of the model, namely on the fact that ${\bm\sigma}_0$ only enters the expressions through the overlap $q$, which is enforced with a Dirac's delta. \\

\noindent The true simplification only happens in the $N\to\infty$ limit; but we shall treat finite $N$ expressions as if $N\to\infty$, keeping in mind that we will ultimately take the limit. In particular, in Ref.~\cite{pacco2024curvature} we have shown that one actually has
\begin{align}
   \lim_{N\to\infty} \frac{\mathbb{E}[h({\bm\sigma}[\gamma;f])]_{0,1}}{\sqrt{2N}}= \lim_{N\to\infty} \frac{\mathbb{E}_{\text{cond}}[h({\bm\sigma}[\gamma;f])]}{\sqrt{2N}}
\end{align}
where
\begin{align}
    \mathbb{E}_\text{cond}[\cdot]:=\mathbb{E}\left[\cdot\bigg|\begin{subarray}{l}
 h^a=\sqrt{2 N} \epsilon_a\\ 
 {\bf g}^a={\bf 0},\quad a=0,1 \end{subarray}\right].
\end{align}
Basically, for $N\to\infty$ at the saddle point the denominators are canceled with the integrations at the numerator, and we are left with a disorder average conditioned over the statistics of energies and gradients of two representative points ${\bm\sigma}_0,{\bm\sigma}_1$ at overlap $q$. This is essentially a consequence of the fact that for this model both the one-point and two-point complexities give the same result at the quenched and annealed levels.

\noindent In the following part of this section on interpolating paths we will use $\mathbb{E}\equiv \mathbb{E}^\text{cond}$, and implicitly assume that we take $N\to\infty$.

\subsection{The case $p=3$ and the eigenvector overlaps.} In this section we specifically set $p=3$ for simplicity. This restriction is motivated by the fact that in this case the energy profile \eqref{eq:Profile} can be expressed as a function of the local properties of the landscape at ${\bm \sigma}_a$ only, i.e. of the local gradients ${\bf g}({\bm \sigma}_a)$ and Hessian matrices $\nabla^2_\perp h({\bm \sigma}_a)$. By implementing the constraints $h({\bm \sigma}_a)= \sqrt{2 N} \epsilon_a$, ${\bf g}({\bm \sigma}_a)=0$ and by using the fact that typically $h({\bf v})=0$, one obtains:
\begin{equation}\label{eq:PrePAth}
\begin{split}
\epsilon_{\bf v}[\gamma,f]&=
\tonde{\gamma^3 +3\gamma^2 \beta \,q} \epsilon_1+ \tonde{\beta^3+3 \beta^2 \gamma \,q} \epsilon_0 - \sqrt{1-q^2} \gamma \beta \frac{f}{\sqrt{2}}  \; \mathbb{E}\quadre{{\bf v}\cdot \frac{ \tilde{\nabla}^2 h({\bm \sigma}_0) }{\sqrt{N}}\cdot {\bf e}_{N-1}({\bm \sigma}_0)}\\
& +   \frac{f^2}{2\sqrt{2}} \mathbb{E}\quadre{  \gamma \, {\bf v} \cdot \frac{ \tilde{\nabla}^2 h({\bm \sigma}_1) }{\sqrt{N}}\cdot {\bf v} + \beta\, {\bf v} \cdot \frac{ \tilde{\nabla}^2 h({\bm \sigma}_0) }{\sqrt{N}}\cdot {\bf v}},
\end{split}
\end{equation}
where we omitted the function arguments to make notation easier. It follows that for $p=3$ and with our choices of ${\bf v}$,  the energy profile depends only on correlations between the entries of the Hessian matrices, whose statistics has been explained in the previous section. We do not reproduce the derivation of this result here, details can be found in the Appendix of Ref.~\cite{paccoros}. \\

\noindent With the choice \eqref{eq:v_start}, in particular, the profile depends on the matrix elements and on the minimal eigenvector of the Hessians.  For any $q>0$, the two matrices are correlated (due to the fact that the landscape at the two points is correlated) and thus the matrix element is non-trivial. As we argue below, to determine its typical behavior for large $N$ one needs to compute the typical value of the overlap between arbitrary eigenvectors of the two correlated matrices. \\

\subsubsection{Eigenvector overlaps, the case of ${\bf v}_\text{soft}$}
After carrying out the algebra of Eq.~\eqref{eq:PrePAth} for ${\bf v}_\text{soft}^a$ one sees that there are essentially two terms that cannot be solved with "conventional" tools and results described in the previous sections. These are the eigenvector overlaps between the two Hessian matrices $\tilde{\nabla}^2_\perp h({\bm\sigma}^a)$ for $a=0,1$.\\

\noindent If we consider ${\bf v}_\text{soft}^0$, the term ${\bf v}_\text{soft}^0\cdot \tilde{\nabla}^2h({\bf\sigma}_1)\cdot {\bf v}_\text{soft}^0$ generates easily solvable terms plus the following element:
\begin{align}
\chi_0:={\bf e}_\text{min}({\bm\sigma}_0)\cdot \frac{\tilde{\mathcal{H}}({\bm\sigma}_1)}{\sqrt{N}} \cdot{\bf e}_\text{min}({\bm\sigma}_0).
\end{align}

\noindent Instead if we consider ${\bf v}_\text{soft}^1$ we can exchange 0s and 1s above; the difficult term becomes: 
\begin{align}
\chi_1:={\bf e}_\text{min}({\bm\sigma}_1)\cdot \frac{\tilde{\mathcal{H}}({\bm\sigma}_0)}{\sqrt{N}} \cdot{\bf e}_\text{min}({\bm\sigma}_1).
\end{align}
In these expressions we have already applied the change of basis (into the local basis of the outer vectors), and as it turns out the (inner) matrices are now expressed in the local basis of the outer vectors, while retaining the same matrix form as in Eq.~\eqref{eq:app:TildeMatDef}. This can be achieved with some simple manipulations of a change of basis from $\mathcal{B}[{\bm\sigma}_0]\to\mathcal{B}[{\bm\sigma}_1]$ and vice-versa. Moreover, at this point, we are not making any difference between division by $\sqrt{N}$ or $\sqrt{N-1}$, since ultimately for $N\to \infty$ it makes no difference inside Eq.~\eqref{eq:Profile}.\\

\noindent We are now ready to expand the matrices using $\{{\lambda}_\alpha^a,{\bf u}_\alpha^a\}_{\alpha\leq N-1}$, that is, their eigenvalue/ eigenvector decomposition :
\begin{align}
\chi_0=\sum_{\alpha=1}^{N-1}{\lambda}_\alpha^1[{\bf e}_\text{min}^0\cdot {\bf u}_\alpha^1]^2, \quad\quad\chi_1=\sum_{\alpha=1}^{N-1}{\lambda}_\alpha^0[{\bf e}_\text{min}^1\cdot {\bf u}_\alpha^0]^2.
\end{align}
For $N\to\infty$ then we have:
\begin{align}
\begin{split}
&\mathbb{E}[\chi_0]\to \int d\lambda\left[\lambda \,\rho_\sigma(\lambda)\,\Phi(\lambda_\text{min}^0,\lambda)\right]\\
&\mathbb{E}[\chi_1]\to \int d\lambda \left[\lambda\,\rho_\sigma(\lambda)\,\Phi(\lambda_\text{min}^1,\lambda)\right]
\end{split}
\end{align}
where we defined the eigenvector overlaps as:
\begin{align}
    \Phi(\lambda^0_\alpha,\lambda^1_\beta):=N\mathbb{E}[({\bf u}_\alpha^0\cdot{\bf u}_\beta^1)^2].
\end{align}
The various simplifications that appear here and the fact that we can use the doubly annealed average are likely to be a consequence of the choice $p=3$. Indeed, in this case one can verify that $\tilde{\mathcal{H}}^0/\sqrt{N}$ has never isolated eigenvalues, so that $\lambda_\text{min}^0=-2\sigma$. The same is not true for $\tilde{\mathcal{H}}^1/\sqrt{N}$, which can have a minimal eigenvalue isolated from the bulk of the semicircle law. \\

\noindent We also need to remark that although we have treated ${\bm\sigma}_0$ and ${\bm\sigma}_1$ on the same footing within the \textit{doubly-annealed} computation, we must still choose $\epsilon_0,\epsilon_1,q$ in such a way to respect the order in which they have to be extracted. This means that for any choice of $\epsilon_0$, we must choose $\epsilon_1,q$ in such a way that we have a positive complexity $\Sigma^{(2)}(\epsilon_1,q|\epsilon_1)>0$. This also greatly simplifies the problem: for acceptable values of these control parameters, the \textit{doubly-annealed} computation gives the right result. If we were to consider atypical choices of $\epsilon_0,\epsilon_1,q$ we would not benefit of this simplification, and we would have to recur to the \textit{quenched} calculation. \\

\noindent In light of this we can write:
\begin{align}
\begin{split}
\label{eq:curvature_chi0_chi1}
&\mathbb{E}[\chi_0]\to \int d\lambda\left[\lambda \,\rho_\sigma(\lambda)\,\Phi(-2\sigma,\lambda)\right]\\
&\mathbb{E}[\chi_1]\to \int d\lambda \,\,\lambda\,\rho_\sigma(\lambda)\Bigg\{\,\Phi(-2\sigma,\lambda)[1-H(|\mu_1|- 2\sigma+\Delta^2/\sigma)]\\&
+\Phi(\lambda_\text{iso}^1,\lambda)H(|\mu_1|- 2\sigma+\Delta^2/\sigma)\Bigg\}
\end{split}
\end{align}
with $H$ the Heaviside step function. The properties of these overlaps are the core topic of Chapter~\ref{chapter:rmt_} and Ref.~\cite{paccoros}. In Sec.~\ref{sec:theoretical_res2} we compute their expressions in the most general case of spiked, correlated GOE random matrices. Let us just state that for the present case we obtain the simple result $\chi_0=-2\sqrt{6}\,q$, whereas the integral for $\chi_1$ must be solved numerically by using the expression for $\Phi(\lambda_\text{iso}^1,\lambda)$ with $\lambda\in[-2\sigma,2\sigma]$ found in Sec.~\ref{sec:rmt_simulations_overlaps}.

\begin{figure}[t!]
\centering
\includegraphics[width=1\textwidth]{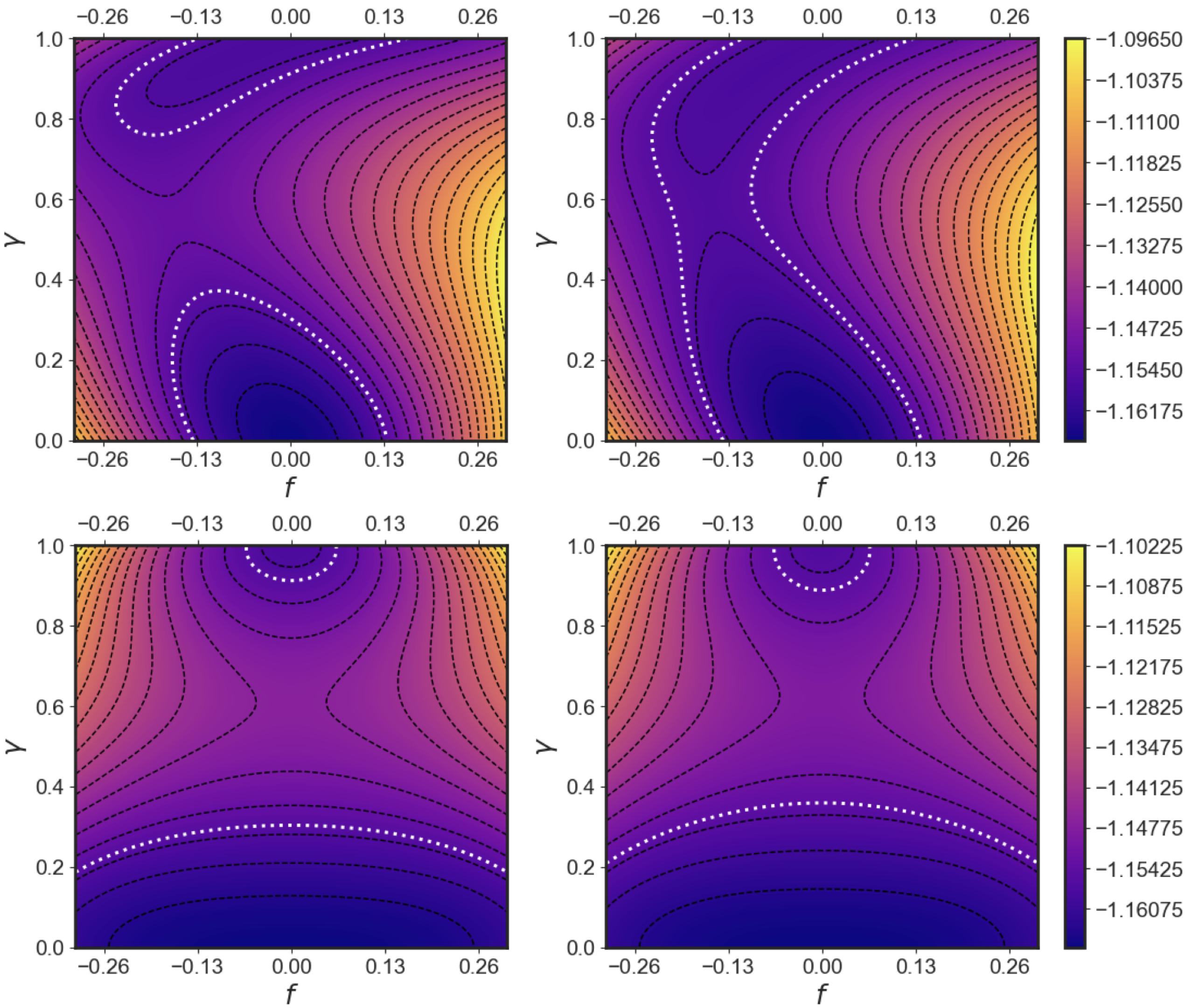}
\caption{Color plot of $\epsilon_{\bf v_{\rm soft}}[\gamma;f]$ as a function of $\gamma$ and $f$ for $\epsilon_0=-1.167$ and $\epsilon_1=-1.157$. Acceptable paths are those that start at $f(0)=0$ and end at $f(1)=0$ in a continuous and injective way. \textit{Left} corresponds to $q=0.66 \in [q_{un},q_{ms}]$, meaning that the secondary configuration is a correlated minimum with an isolated eigenvalue. \textit{Right} corresponds to $q=0.7>q_{ms}$ and the secondary configuration is a rank-1 saddle. \textit{Bottom} corresponds to ${\bf v}_{\rm soft} \to {\bf v}_{\rm soft}^0$ and \textit{Top} to ${\bf v}_{\rm soft} \to {\bf v}_{\rm soft}^1$.}
\label{fig:density_plots}
\end{figure}

\subsection{Results: the case ${\bf v}\to{\bf v}_{\rm soft}^0$}
For ${\bf v}\to{\bf v}_{\rm soft}^0$ a stylized representation of the interpolating path is drawn in Fig.~\ref{fig:SketchPath}. When $\epsilon_0 < \epsilon_{\rm th}$, we have seen that ${\bm\sigma}_0$ is typically a local minimum, and the spectrum of $\tilde{\mathcal{H}}^0/\sqrt{N}$ for large $N$ is the Wigner's semicircle, with variance $\sigma^2=6$. The minimal eigenvalue therefore tends to  $-2\sqrt{6}$. The associated eigenvector will show no condensation in the direction of ${\bf e}_{N-1}({\bm \sigma}_0)$, meaning that  typically $u^0$ takes the value $u^0_{\rm typ}=0$. Then, the path along \eqref{eq:path} reads (we avoid function arguments but recall that $\beta$ depends on $f$ as well):
\begin{align}
\label{eq:energy_profile_quench_0}
\begin{split}
&\epsilon_{{{\bf v}_{\rm soft}^0}}[\gamma;f]=\tonde{\gamma^3 +3\gamma^2 \beta \,q}  \epsilon_1+  (\beta^3+3 \beta^2 \gamma \,q)\epsilon_0-\sqrt{3}\,f^2(\gamma)\left( \beta+\gamma\,q\right).
\end{split}
\end{align}

\begin{figure}[t!]
\centering
\includegraphics[width=0.7\textwidth]{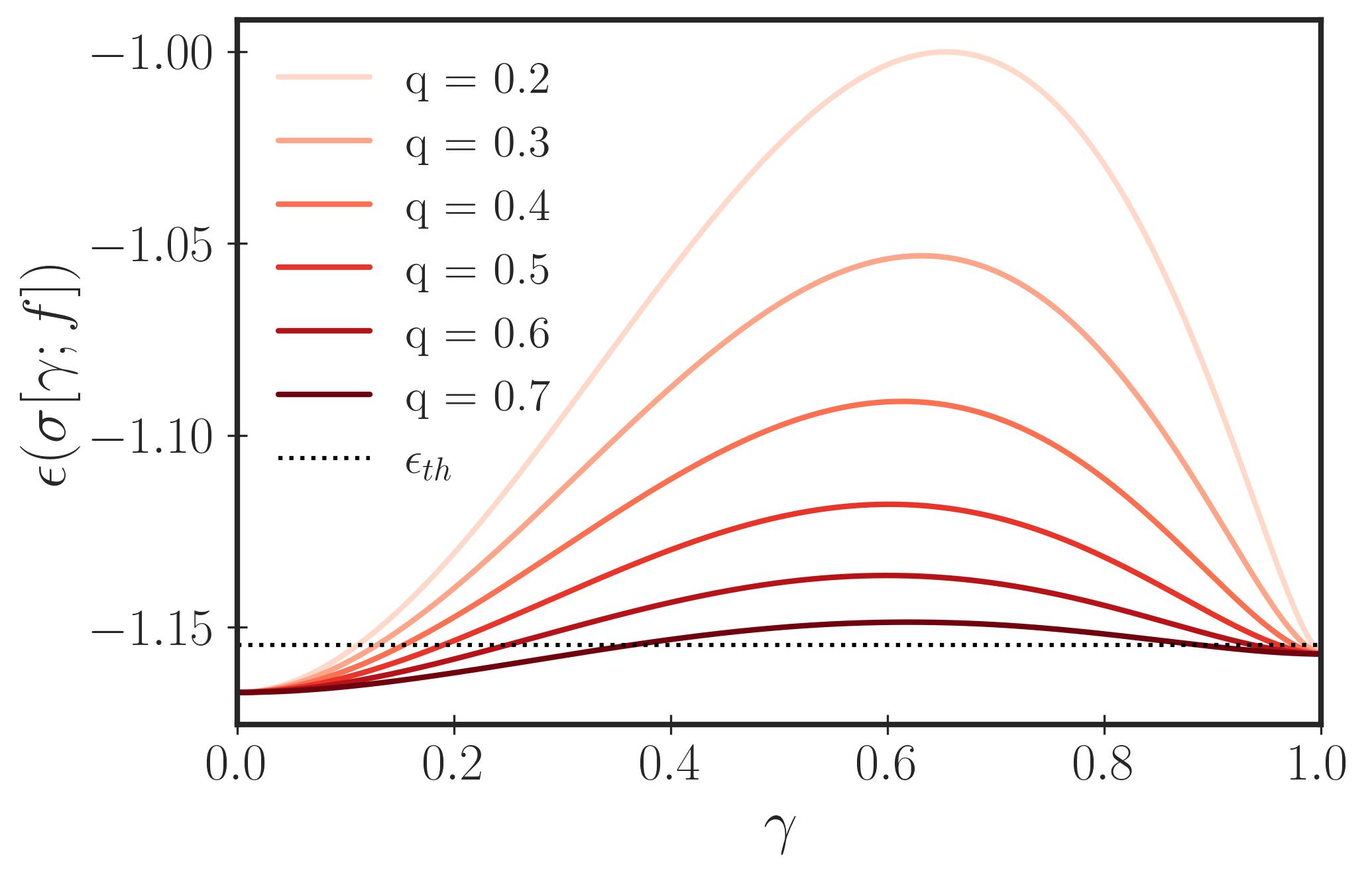}
\caption{Energy profile along Geodesic paths for various values of $q$ with $\epsilon_0=-1.167$, $\epsilon_1=-1.157$. We see that increasing $q$ decreases the barrier.}
\label{fig:geodesic_barriers}
\end{figure}

\noindent One finds that $\delta \epsilon_{{{\bf v}_{\rm soft}^0}}/\delta f=0$ is satisfied by $f \equiv 0$, and that for $\epsilon_a <\epsilon_{\rm th}$ it holds $\delta^2 \epsilon_{{{\bf v}_{\rm soft}^0}}/\delta f^2>0$ at $f \equiv 0$, meaning that the geodesic path is a minimum of the functional~\eqref{eq:energy_profile_quench_0} for any value of $\gamma\in[0,1]$. This means that arbitrary  deformations of the interpolating path in the direction of softest curvature at ${\bm \sigma}_0$ go through regions of the landscape of higher energy density, on average. This observation is confirmed by Fig.~\ref{fig:density_plots} \textit{bottom}, which shows two density plots of Eq.~\eqref{eq:energy_profile_quench_0} as a function of $\gamma$ and $f$, where $f$ is allowed to take any value within its range of validity that keeps $\beta$ in \eqref{eq:beta} well defined. Arbitrary paths are obtained drawing curves connecting the points $f(0)=0$ and $f(1)=0$ in a continuous and {injective} fashion. One finds that the energy profile along these curves is non-monotonic, with a local maximum whose energy we refer to as the energy barrier. In Fig.~\ref{fig:density_plots} the parameters $\epsilon_0$ and $\epsilon_1$ are fixed, while $q$ is tuned in such a way that ${\bm \sigma}_1$ is either a minimum with an Hessian with a single isolated eigenvalue (\textit{left}) or an rank-1 saddle (\textit{right}). The white dotted lines represent the level curves of value $\epsilon_\text{th}$. 
The \textit{bottom} figures of the plot confirm that the lowest energy barrier of $\epsilon_{{{\bf v}_{\rm soft}^0}}[\gamma;f]$ is obtained for the geodesic path $f\equiv 0$, and shows that such barrier is well above $\epsilon_{\rm th}$.
The same results could be obtained for different values of $\epsilon_a$ and $q$.\\

\noindent Moreover, one finds that the barrier associated to the geodesic path increases when the overlap $q$ decreases (see Fig.~\ref{fig:geodesic_barriers}) and when $\epsilon_1$ increases towards $\epsilon_{\rm th}$ at fixed $q$ (not shown here). \\

\noindent A final remark: the profile \eqref{eq:energy_profile_quench_0} is not the same that one would obtain by choosing ${\bf v}$ in \eqref{eq:PrePAth} as a purely random Gaussian vector ${\bf v}_{\rm rand}$, uncorrelated to the local Hessian: with that choice, all the terms depending on ${\bf v}$ in Eq. \eqref{eq:PrePAth} vanish on average, and one obtains:
$\epsilon_{{{\bf v}_{\rm rand}}}[\gamma;f]=\tonde{\gamma^3 +3\gamma^2 \beta\,q}  \epsilon_1+  (\beta^3 +3 \beta^2\gamma \,q)\epsilon_0.$
For fixed $f(\gamma)$ this energy profile is systematically higher than \eqref{eq:energy_profile_quench_0}. However, the functional is again minimized by $f \equiv 0$.\\

\subsection{Results: the case ${\bf v}\to{\bf v}_{\rm soft}^1$}
\noindent The case ${\bf v}\to{\bf v}_{\rm soft}^1$ is richer. In this case, depending on the values of $\epsilon_a$ and $q$, the stationary point ${\bm \sigma}_1$ is either a rank-1 saddle ($q_{sm} < q \leq q_M$), a minimum with one isolated mode in the Hessian ($q_{un} < q  \leq q_{sm}$), or an uncorrelated minimum ($q<q_{un}$), see Fig. \ref{fig:2_point_phase_diag}. Whenever the Hessian at ${\bm \sigma}_1$ has an isolated eigenvalue ($q>q_{un}$), it is the smallest eigenvalue and its typical value and the typical value $u^1_{\rm typ}$ are given by
\begin{equation}\label{eq:Typicals}
\begin{split}
&\lambda_{\rm typ}^1=
\frac{3 \, \delta \epsilon_q}{\sqrt{2}q }
\quadre{\frac{1+3 q^2}{
   1-q^2}
      - \sqrt{1-  \frac{2(1-q^2)^2}{3(1+q^2)\delta \epsilon_q^2}}}\\
      &u_{\rm typ}^1=\frac{1+q^2}{1+3q^2 -\frac{\sqrt{2} (q-q^3) \lambda_{\rm typ}^1 }{3 \, \delta \epsilon_q}}\quadre{1- \frac{(1-q^2) \, g( \lambda_{\rm typ}^1)}{12 \sqrt{2} (q+q^3) \delta \epsilon_q} }
      \end{split}
\end{equation}
where $\delta \epsilon_q=\epsilon_0-q \epsilon_1$ and $g(\lambda)=1+3 q^2 \lambda- \text{sign}(\lambda) (1-q^2)\sqrt{\lambda^2- 24}$, see Eq.~\eqref{eq:app:IsoExplicit} and Eq.~\ref{eq:app:ProjVectors}. In this case, a careful computation of the energy profile gives: 
\begin{align}
\label{eq:energy_quench_1}
\begin{split}
\epsilon_{{{\bf v}_{\rm soft}^1}}[\gamma;f]&=(\gamma^3+3\gamma^2\beta\,q)\epsilon_1+
(\beta^3+3\gamma\beta^2\,q)\epsilon_0\\& + \gamma\,\beta\, f
\sqrt{\frac{u^1_{\rm typ}(1-q^2)}{2(1-u^1_{\rm typ})}}
\left(\frac{6 \sqrt{2} \, q \, (\epsilon_0-q \epsilon_1)}{1-q^2}-\lambda_{\rm typ}^1\right)\\
&+\frac{f^2\,(\gamma + 2 \beta)\,  u^1_{\rm typ} }{2\sqrt{2}(1-u^1_{\rm typ})}\left(\frac{6 \sqrt{2} \, q \, (\epsilon_0-q \epsilon_1)}{1-q^2}-\lambda_{\rm typ}^1\right)\\&+\frac{f^2 \,\beta\,u^1_{\rm typ} }{2\sqrt{2}(1-u^1_{\rm typ})}\frac{6 \sqrt{2} q (q \epsilon_0-\epsilon_1)}{1-q^2} +\frac{f^2\,\gamma}{2\sqrt{2}}\lambda_{\rm typ}^1\\
& +\frac{f^2 \,\beta\, }{2\sqrt{2}(1-u^1_{\rm typ})}\int_{-2\sqrt{6}}^{2\sqrt{6}} d\lambda \,\frac{\sqrt{36 - \lambda^2}}{18 \, \pi } \,\Phi(\lambda_{\rm typ}^1,\lambda)\,\lambda.
\end{split}
\end{align}
 The function $\Phi$ in the last term is precisely the eigenvector overlap introduced above, which gives the typical value of the overlap between the eigenvector associated to $\lambda^1_{\rm min}$ and any arbitrary eigenvector of $\mathcal{\tilde{H}}({\bm \sigma}_0)$ with eigenvalue $\lambda$ in the bulk of the Wigner's semicircle. The computation of this overlap is found in Chapter~\ref{chapter:rmt_}, Sec.~\ref{sec:theoretical_res2}; the precise expression is rather involved and not of particular interest, we therefore defer to that chapter for the calculations. In Fig.~\ref{fig:density_plots} (\textit{top}) we show two density plots associated to \eqref{eq:energy_quench_1}: clearly in this case the geodesic path is no longer optimal, and thus the energy barrier is lowered by deforming the path in the direction of softest curvature at ${\bm \sigma}_1$. The optimal path is obtained numerically, by selecting the lowest energy point for each $\gamma$ and by verifying a posteriori that this leads to a well defined continuous path. To show better the relation between the optimal and geodesic paths, in Fig.~\ref{fig:barriers} we compare the energy profile \eqref{eq:energy_quench_1} evaluated along the geodesic and optimal paths, for  ${\bm \sigma}_1$ being a correlated minimum  ($q_{un} < q \leq q_{ms}$) or a saddle ($q_{ms} < q < q_M$). In the latter case, the optimal path lies entirely below $\epsilon_{\rm th}$ (dashed horizontal line). When ${\bm \sigma}_1$ is an uncorrelated minimum, since $u_{\rm typ}^1=0$ the behavior is analogous as in Eq.~\eqref{eq:energy_profile_quench_0} and the optimal path is again the geodesic one. Quite interestingly, the optimal paths seen in Fig.~\ref{fig:barriers} are significantly lower than their geodesic counterparts. In particular, we see that the optimal path to the rank-1 saddle (yellow dotted line) is below threshold.
 \begin{figure}[t!]
\centering
\includegraphics[width=0.65\textwidth]{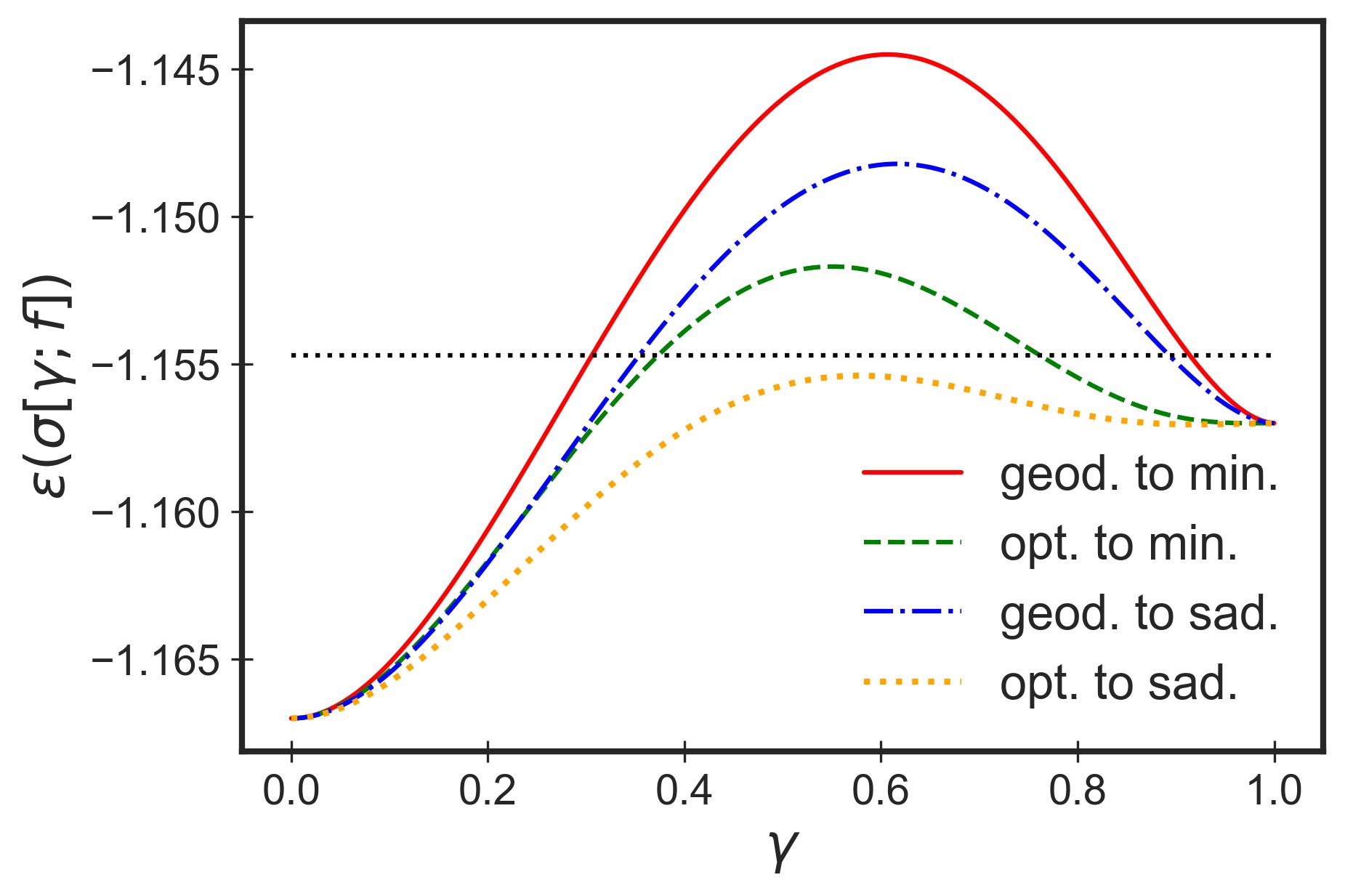}
\caption{Comparison between the energy profile along the geodesic (geod.) and optimal (opt.) paths, for the same parameters as in Fig.~\ref{fig:density_plots}({\it top}) which correspond to ${\bm \sigma}_1$ being a minimum (min.) or rank-1 saddle (sad.). The saddle is obtained by keeping the same energy and by increasing $q$. We see that optimal paths are significantly lower than the geodesic ones, with the optimal path to the saddle being below threshold. }
\label{fig:barriers}
\end{figure} 
Although we know that the rank-1 saddle must be dynamically connected to the reference minimum \cite{ros2021dynamical}, it is not obvious a priori that geodesic paths perturbed in the direction of the softest mode of the rank-1 saddle should be representative of the actual trajectory followed by the system. Our analysis shows that these artificial paths can \textit{at least} capture the fact that the softest mode provides a lower path to the reference minimum.


\subsection{Results: the case ${\bf v}\to{\bf v}_\text{Hess}$}
Similarly as above, plugging ${\bf v}_{\rm Hess}$ into \eqref{eq:PrePAth} and computing the averages one finds:
\begin{align}
\label{eq:energy_profile_Hess}
\begin{split}
\epsilon_{{{\bf v}_{\rm Hess}}}[\gamma;f]&=\tonde{\gamma^3 +3\gamma^2 \beta \,q}  \epsilon_1+  (\beta^3+3 \beta^2 \gamma \,q)\epsilon_0- \gamma \beta f(\gamma)\;\sqrt{\frac{3(1-q^2)^2}{1+q^2}}.
\end{split}
\end{align}
This case is interesting because the optimal barrier along the perturbed path is lower than the geodesic one even for $q<q_{un}$  (although in that case it is still above the threshold for all choices of $\epsilon_a < \epsilon_{\rm th}$).
In Fig.~\ref{fig:energy_barrier} we plot a comparison between the energy barrier of the geodesic path and that of the optimal paths obtained with the various prescriptions described above, for given $\epsilon_0=-1.167,\epsilon_1=-1.157$ and varying $q$. All barriers decrease as $q$ increases; when ${\bf v}={\bf v}_{\rm Hess}$ the deformed path is always associated to a lower energy barrier with respect to the geodesic, while for ${\bf v}={\bf v}^1_{\rm soft}$ this is true only for $q>q_{un}$. For large values of $q$ the barrier along the perturbed paths lies below $\epsilon_{\rm th}$, but this is true only within the range $q_{ms}<q<q_M$, when the arrival point is a rank-1 saddle. We find that this remains true for arbitrary values of  $\epsilon_0,\epsilon_1< \epsilon_{\rm th}$. For the largest $q \lesssim q_M$, the curve associated to ${\bf v}^1_{\rm soft}$ is flat (blue line in Fig.~\ref{fig:energy_barrier}), indicating that the energy profile becomes monotonically increasing in the interval $\gamma \in [0,1]$ with a maximum at $\gamma=1$, equal to $\epsilon_1$. In view of this comparison, one also understands why the geodesic path is no longer optimal when ${\bf v} \to {\bf v}_{\rm soft}^1$ and $u^1_{\rm typ} \neq 0$: in this case, indeed, ${\bf v}_{\rm soft}^1$ has an $\mathcal{O}(1)$ projection on the vector ${\bf v}_{\rm Hess}$, and thus allowing the path to deviate in the direction of ${\bf v}_{\rm soft}^1$ leads to a lower energy barrier.

\begin{figure}[t!]
\centering
\includegraphics[width=0.65\textwidth]{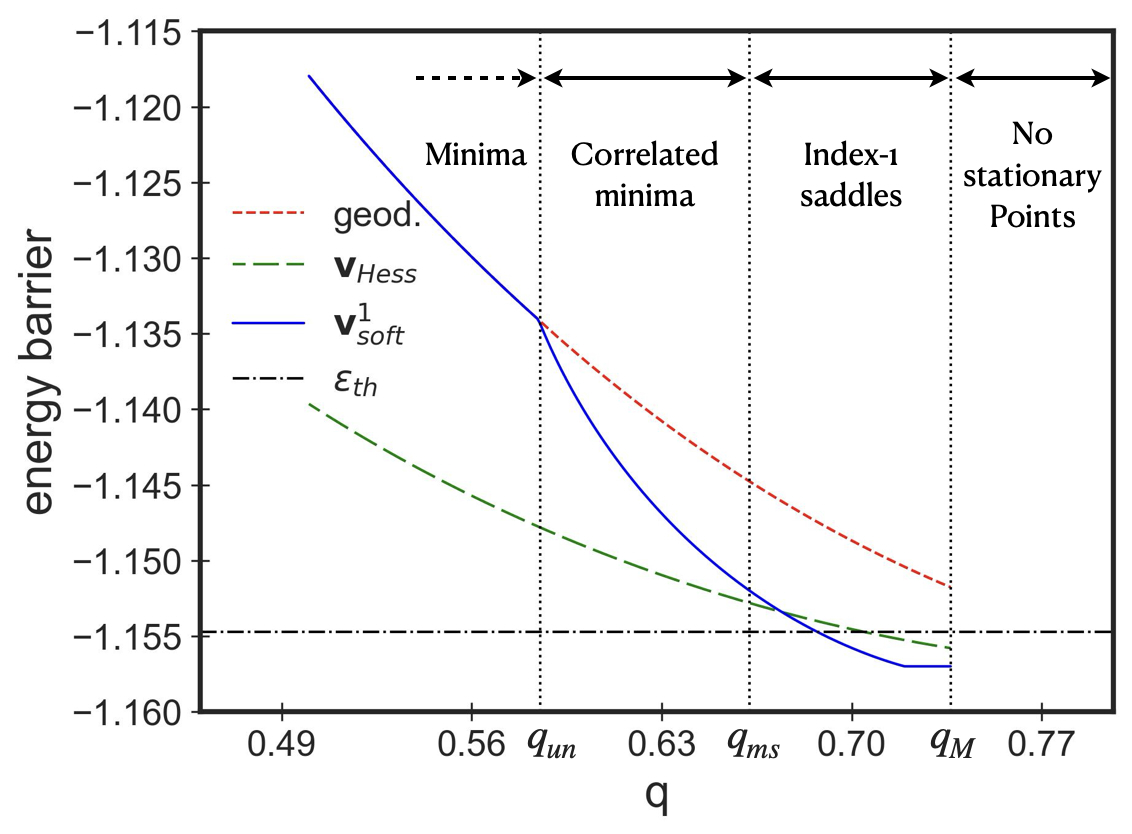}
\caption{Energy barrier along the geodesic and optimal (perturbed) paths as a function of $q$, for $\epsilon_0=-1.167$, $\epsilon_1=-1.157$ and ${\bf v}$ being equal to either $ {\bf v}^1_{\rm soft}$ or ${\bf v}_\text{Hess}$.}
\label{fig:energy_barrier}
\end{figure}

\subsection{Summary of results}
Let us briefly resume the findings of this first approach. We have characterized the energy density profile along perturbed geodesic paths between a deep minimum of energy density $\epsilon_0$ and a local minimum/ rank-1 saddle of energy density $\epsilon_1$, both below threshold. We have found that
\begin{itemize}
    \item Geodesic paths are always above threshold
    \item Perturbed paths that follow the softest curvature of the landscape at ${\bm \sigma}_0$, given by ${\bf v}^0_\text{soft}$ are on average associated to higher energy barriers, so that they are likely to be uncorrelated to good transition paths, at variance with simulations of jammed and mildly super-cooled particles \cite{xu2010anharmonic,widmer2008irreversible, khomenko2021relationship}. A possible discrepancy is that the Hessian plays a different role, since our regime would correspond to a deep super-cooled state.
    \item However, we can leverage information on the Hessian at ${\bm\sigma}_0$ to lower the barrier, by following ${\bf v}_\text{Hess}$, which encodes for the orthogonal component of the gradient along the path. It would be interesting to test the role of the Hessian by studying energy paths and barriers in small systems \cite{heuer2008exploring,baity2021revisiting}.
    \item When $(\epsilon_1,q)$ belongs to the hatched zone in Fig.~\ref{fig:2_point_phase_diag}, then the fixed point at the end has an outlier in the spectrum, and is either a correlated minimum or a rank-1 saddle. In that case, we can significantly lower the barrier by encoding in ${\bf v}_\text{soft}^1$ information on the eigenvector associated to this spike. This case of correlated minima is likely to be the one more closely related to the results found in finite dimensional systems \cite{coslovich2019localization}.
    \item The barriers of the perturbed paths can lie below the threshold energy, but only when ${\bm \sigma}_1$ is a rank-1 saddle. We cannot exclude below threshold pathways also when ${\bm\sigma}_1$ is a deep minimum, but our methods are not enough to identify them, and further research is required.
    \item All in all our analysis seems to suggest that deep minima are separated by above threshold barriers, similarly as \cite{rizzo2021path}. Whether typical pathways below threshold exist at all is an open question. In Sec.~\ref{sec:en_land_perspectives} we will make some hints on the activated dynamics arising from both this and the next calculation.
\end{itemize}

\section{The three-point complexity}
\label{sec:en_land_three_point}
In this section we concentrate on probing the arrangement of triplets of local minima in the energy landscape of the pure spherical $p-$spin model, by means of a Kac-Rice computation. We remind that here there is no temperature, and we are solely interested in the properties of the random landscape $\mathcal{E}$.

\subsection{Definition}
\label{subsec:3_pt_def}
In this work we go beyond the two-point complexity, and consider a three-point complexity, thus extending the analysis of~\cite{ros2019complexity,ros2020distribution} by computing the asymptotic behavior (for large $N$) of the typical value of the random variable $\mathcal{N}_{{\bf s}_0\, {\bf s}_1}(\epsilon_2, q_0, q_1|q,\epsilon_0,\epsilon_1)$. This is the number of stationary points ${\bf s}_2$ of energy density $\epsilon_2$, that are found at overlap $q_1$ with a stationary point ${\bf s}_1$ and at overlap $q_0$  from another stationary point ${\bf s}_0$, see Fig.~\ref{fig:landscape_three}. In turn, the stationary points ${\bf s}_1$ and ${\bf s}_0$ are at a given overlap $q$ with each others.
We are interested in the three-point complexity $\Sigma^{(3)}(\epsilon_2, q_0, q_1 | \epsilon_1,\epsilon_0, q)$, defined as 
\begin{align}
\label{eq:quenchedcomp3}
\Sigma^{(3)}=\lim_{N\to\infty} \frac{1}{N}\mathbb{E}\left[\log\mathcal{N}_{{\bf s}_0 \,{\bf s}_1}(\epsilon_2, q_0, q_1)\right]_{0,1},
\end{align}
where we neglect some function arguments (in $\mathcal{N}_{{\bf s}_0,{\bf s}_1}$ and $\Sigma^{(3)}$) to make notations a bit less heavy.  Much like for the two-point complexity, in this case the average $\mathbb{E} [ \cdot ]_{0,1}$, which we already encountered before in Eq.~\ref{eq:en_land_E_0_1_h}, denotes the flat average over the stationary points ${\bf s}_1$ with energy density $\epsilon_1$, \emph{constrained} to be at overlap $q$ with another stationary point ${\bf s}_0$ with energy density $\epsilon_0$ extracted with a flat measure; additionally, the landscape is averaged over:
\begin{equation}\label{eq:averages}
\mathbb{E} [ \cdot ]_{0,1}=\mathbb{E}\left[
\int_{\mathcal{S}_N(\sqrt{N})} d{\bf s}_0\frac{\omega_{\epsilon_0}({\bf s}_0)}{\mathcal{N}(\epsilon_0)}\int_{\mathcal{S}_N} d{\bf s}_1 \frac{\omega_{\epsilon_1,q}({\bf s}_1|{\bf s}_0)}{\mathcal{N}_{{\bf s}_0}(\epsilon_1,q|\epsilon_0)} \;\, \cdot\right],
\end{equation}
\begin{figure*}[t!]
    \centering
    \includegraphics[width=0.42\textwidth, trim={2 2 2 2},clip]{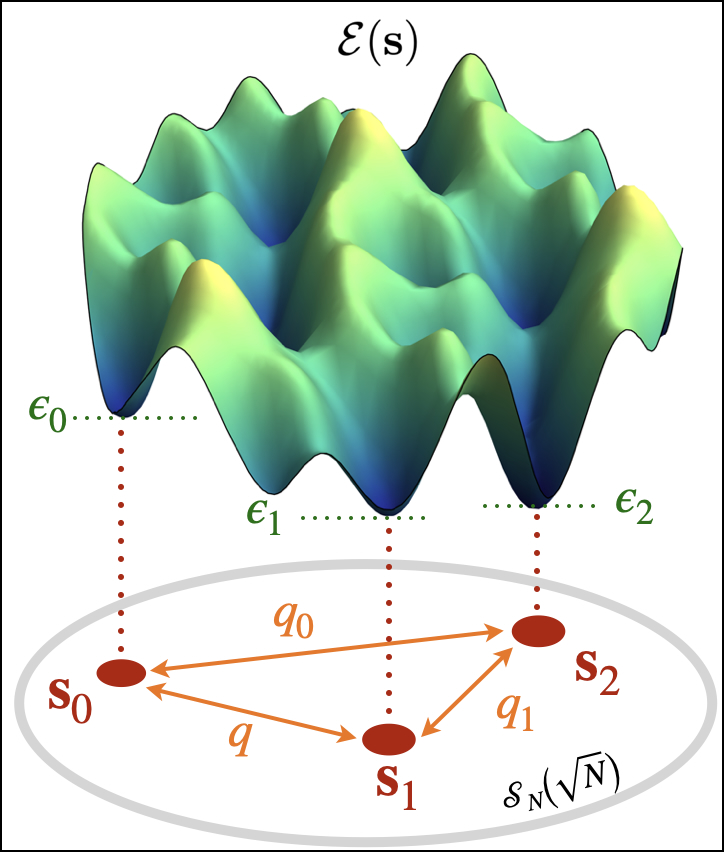}
    \caption{Pictorial representation of a landscape with the three configurations ${\bf s}_0,{\bf s}_1,{\bf s}_2$, and their overlaps.}
    \label{fig:landscape_three}
\end{figure*}
where $\omega_{\epsilon_1,q}$ is the measure that selects configurations ${\bf s}_1$ that are stationary points with given parameters $\epsilon_1, q$, defined in Eq.~\eqref{eq:Measure1}, and similarly $\omega_{\epsilon_0}$ defined in Eq.~\eqref{eq:Measure0}. The random variables $\mathcal{N}_{{\bf s}_0}(\epsilon_1,q|\epsilon_0)$ was also defined in Eq.~\ref{eq:en_land_N_s0}. The parameters $\epsilon_1, q$ are chosen in such a way that typically stationary points are found at those values, i.e., $\Sigma^{(2)}(\epsilon_1,q|\epsilon_0) \geq 0$ (the colored region in Fig. \ref{fig:2_point_phase_diag}). The number $\mathcal{N}_{{\bf s}_0\, {\bf s}_1}(\epsilon_2, q_0, q_1)$ is obtained as
\begin{equation}\label{eq:enne}
   \mathcal{N}_{{\bf s}_0\, {\bf s}_1}(\epsilon_2, q_0, q_1)= \int_{\mathcal{S}_N(\sqrt{N})} d {\bf s}_2\,  \omega_{\epsilon_2,q_0, q_1}({\bf s}_2 | {\bf s}_1,{\bf s}_0)
\end{equation}
where now
\begin{equation}\label{eq:Measure2}
  \begin{split}
\omega_{\epsilon_2,q_0, q_1}({\bf s}_2|{\bf s}_1,{\bf s}_0)=&|\det \nabla^2_\perp \mathcal{E}({\bf s}_2)|\delta(\nabla_\perp \mathcal{E}({\bf s}_2))\times\\
&\times \delta(\mathcal{E}({\bf s}_2)-N\epsilon_2)\delta({\bf s}_2 \cdot {\bf s}_0-N q_0)\delta({\bf s}_2 \cdot {\bf s}_1-N q_1).
\end{split}
\end{equation}
Notice that in the definition of $\Sigma^{(3)}$, the roles of ${\bf s}_1$ and ${\bf s}_0$ are not interchangeable a priori: while ${\bf s}_0$  is selected with no other constraints than its energy, ${\bf s}_1$ is selected with a measure that is \emph{conditioned} to the overlap with ${\bf s}_0$. We also remark that the three-point complexity differs from the zero-temperature limit of the three-replica potential introduced in \cite{cavagna1997structure}, in that the configurations are extracted with a different measure which enforces them to be stationary points of the landscape. In particular, our computation is particularly useful for the pure $p$-spin model given the correspondence between TAP states and local minima (cf. Sec.~\ref{sec:tap_approach}). \\

\subsection{On quenched vs annealed}
\label{subsec:quenched_annealed}
We already made some general considerations on quenched vs annealed in Sec.~\ref{sec:en_land_quench_vs_ann_consid}. Let us apply those considerations to the present case. The complexity~\eqref{eq:quenchedcomp3} as we defined it above is a \emph{quenched} quantity, in that it determines the asymptotic scaling of the \emph{typical} number of stationary points ${\bf s}_2$. The \textit{annealed} version then reads:
\begin{align}
\label{eq:annealedcomp3}
\Sigma^{(3)}_{A}=\lim_{N\to\infty}\frac{\log \mathbb{E}[\mathcal{N}_{{\bf s}_0\,{\bf s}_1}(\epsilon_2, q_0, q_1)]_{0,1}}{N}
\end{align}
with the same average as in \eqref{eq:averages}. Of course, this quantity controls the asymptotic scaling of the \emph{average} number of stationary points ${\bf s}_2$, which in general is an upper bound to the quenched value. In the setting we are considering, the annealed complexity \eqref{eq:annealedcomp3} still requires some form of replica trick to be calculated,  due to the presence of the denominators $\mathcal{N}(\epsilon_0)$ and $\mathcal{N}_{{\bf s}_0}(\epsilon_1,q)$ in \eqref{eq:averages}. To bypass the use of replicas, one may consider an approximation in which the expectation value of the ratio in \eqref{eq:averages} is factorized into the ratio of expectation values of the numerator and denominator, meaning that the average $ \mathbb{E}[ \cdot ]_{0,1}$ is replaced with:
\begin{equation}\label{eq:averagesAnn}
 \mathbb{E}[ \cdot ]_{2A}:= \frac{\mathbb{E}\left[
\int_{\mathcal{S}_N(\sqrt{N})^{\otimes 2}} d{\bf s}_0 \,d{\bf s}_1 \,  \omega_{\epsilon_0}({\bf s}_0)\,\omega_{\epsilon_1,q}({\bf s}_1|{\bf s}_0) \;\, \cdot\right]}{\mathbb{E}[\mathcal{N}(\epsilon_0)\,  \mathcal{N}_{{\bf s}_0}(\epsilon_1,q)] }.
\end{equation}
This corresponds to the \emph{doubly-annealed} average for the present case, where all possible approximations are made: $\log$ and averages are exchanged, denominators and numerators are factorized, see Sec.~\ref{sec:en_land_quench_vs_ann_consid}. The \textit{doubly-annealed} complexity that one obtains then reads\footnote{again omitting function arguments for notational simplicity}:
\begin{equation}\label{eq:annealedcomp3doubly}
\Sigma^{(3)}_{2A}=\lim_{N\to\infty} \frac{1}{N}\log \mathbb{E}_{2A}[ \mathcal{N}_{{\bf s}_0\,{\bf s}_1}].
\end{equation}
In general, as we have shown in Ref.~\cite{pacco_triplets_2025}, annealed and quenched averages differ for the three-point complexity. Nonetheless, they match in some special important points (as explained in the following). Moreover, they are very close numerically, so that the physical picture is qualitatively unchanged. The calculations for the quenched complexity using a Replica Symmetric ansatz are the topic of \cite{pacco_quenched_triplets_2025}, although the plots here are done with it.

\section{Landscape's geometry: accumulation and clustering}
\label{subsec:defs_clustering}
We now introduce some notions and terminology relevant to the subsequent discussion. The goal of this work is to use the three-point complexity \eqref{eq:quenchedcomp3} to determine to what extent the landscape in the vicinity of a stationary point ${\bf s}_0$ (e.g., a deep local minimum) differs from the landscape in typical regions of configuration space that are not conditioned to be near ${\bf s}_0$. In other words, the knowledge of both the two- and three-point complexity allows us to compare the local structure of the landscape (probed by ${\bf s}_2$) in the vicinity of:\\

\begin{itemize}
    \item[(i)]   \emph{typical} stationary points ${\bf s}_1$ with energy density $\epsilon_1$, i.e., stationary points extracted with the uniform measure over all stationary points with that energy density and no additional constraint;
    \item[(ii)]  \emph{conditioned} stationary points ${\bf s}_1$ with  energy density $\epsilon_1$, at overlap $q$ with another stationary point ${\bf s}_0$ with energy density $\epsilon_0$.
\end{itemize}
The first information is encoded in the two point complexity $\Sigma^{(2)}(\epsilon_2, q_1 |\epsilon_1)$, the second one in the three-point complexity $\Sigma^{(3)}(\epsilon_2, q_0, q_1 | \epsilon_1,\epsilon_0, q)$. Put differently, we want to see what is the effect that the presence of ${\bf s}_0$ has on the distribution of ${\bf s}_2$, extracted conditionally to ${\bf s}_1$.\\

\begin{figure*}[h!]
\centering
\includegraphics[width=0.8
\textwidth, trim=5 5 5 5,clip]{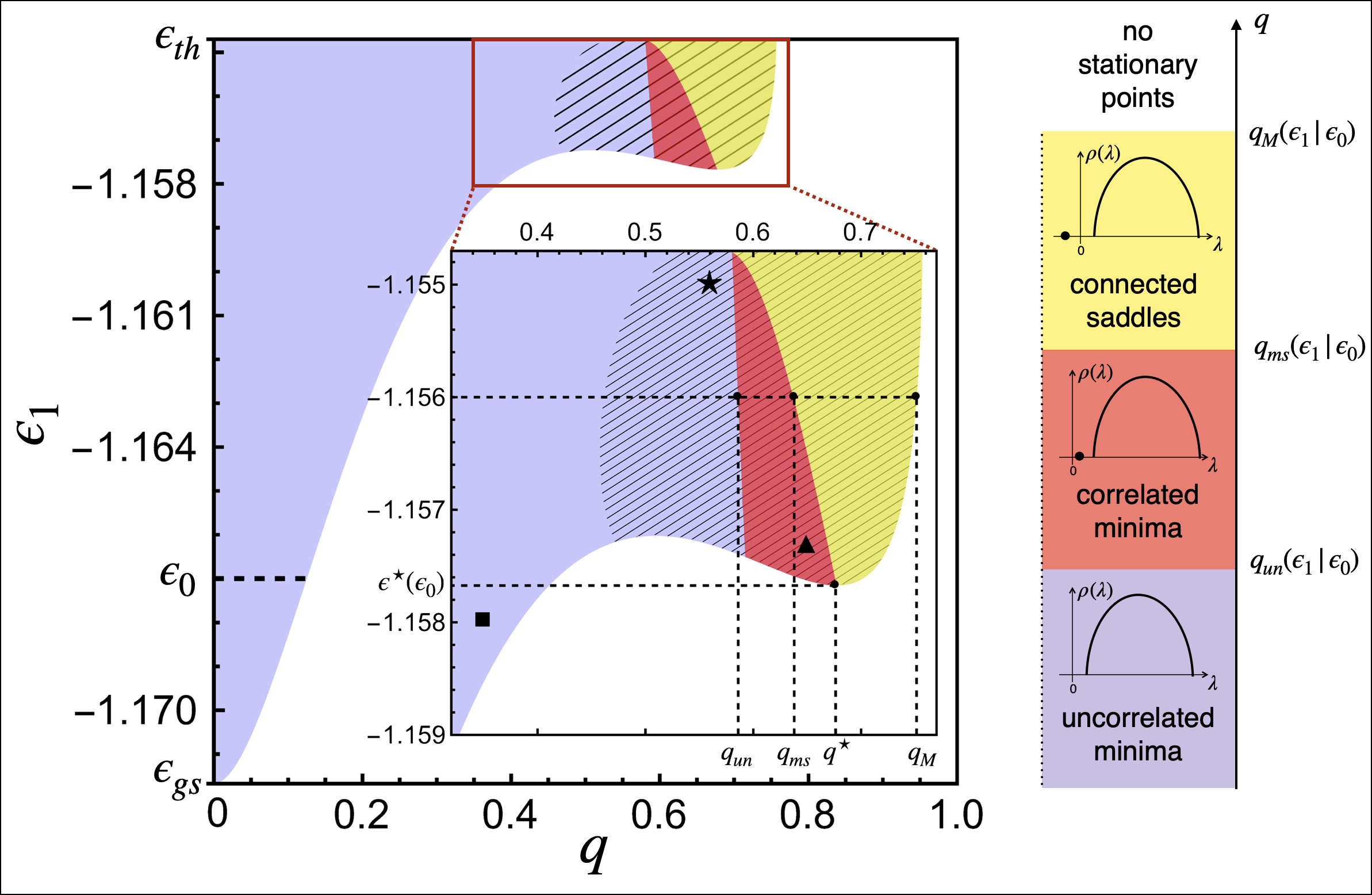}
\caption{Signatures of clustering for $p=3$. \textit{Left.} The colored area shows the range of $\epsilon_1$ and $q$ where $\Sigma^{(2)}(\epsilon_1,q|\epsilon_0) \geq 0$ for $\epsilon_0 = -1.167$. Blue indicates minima, red indicates correlated minima, and yellow indicates rank-1 saddles (see \textit{Right} picture). The lowest energy of rank-1 saddles is $\epsilon^*(\epsilon_0)$, with overlap $q^*(\epsilon_0)$. \textit{Inset}. The black hatched region marks the values of $q$ for which clustering occurs with $\epsilon_2 = \epsilon_1$, and $\epsilon_0 = -1.167$. We refer to Sec.~\ref{subsec:defs_clustering} for a definition of clustering. The symbols $\blacksquare, \blacktriangle, \bigstar$ mark specific values of parameters listed in Table \ref{fig:table} and considered in  Fig.~\ref{Fig:Plots3D}. \textit{Right.} Eigenvalue distribution of the Hessian at the stationary points ${\bf s}_1$ for $\epsilon_1 > \epsilon^*(\epsilon_0)$. For $q < q_{\rm un}(\epsilon_1|\epsilon_0)$, the eigenvalues are distributed according to a semicircle law and ${\bf s}_1$ is an uncorrelated minimum. For $q_{\rm un}(\epsilon_1|\epsilon_0) < q < q_{\rm ms}(\epsilon_1|\epsilon_0)$, the Hessian has a positive isolated eigenvalue, and  ${\bf s}_1$ is a correlated minimum. For $q_{\rm ms}(\epsilon_1|\epsilon_0) < q < q_M(\epsilon_1|\epsilon_0)$, ${\bf s}_1$ is a rank-1 saddle. No stationary points exist for $q > q_M(\epsilon_1|\epsilon_0)$.}
\label{fig:2D_plot_clustering}
\end{figure*}
\noindent When comparing the three-point and two-point complexities, it is straightforward to argue that the following inequality must hold for all values of parameters:
\begin{equation}\label{eq:BoundTrivial}
\Sigma^{(3)}(\epsilon_2, q_0, q_1 | \epsilon_1,\epsilon_0, q) \leq \Sigma^{(2)}(\epsilon_2, q_0 |\epsilon_0).
    \end{equation}
Indeed, the number of stationary points ${\bf s}_2$ conditioned on the properties of ${\bf s}_1$ is smaller than the number of stationary points of the same energy density counted without such conditioning. In both the quantities in \eqref{eq:BoundTrivial}, the stationary points ${\bf s}_2$ are enforced to be at given overlap with  ${\bf s}_0$, and only a fraction of these points also satisfy the constraint on the overlap with ${\bf s}_1$. As we justify below, the bound is saturated  for $q_1=q\cdot q_0$.\\

\noindent On the other hand, one can not assume an analogous bound exchanging the role of ${\bf s}_0$ and ${\bf s}_1$, i.e. comparing $\Sigma^{(2)}(\epsilon_2, q_1 |\epsilon_1)$ with $\Sigma^{(3)}(\epsilon_2, q_0, q_1 | \epsilon_1,\epsilon_0, q)$. Indeed, the properties of ${\bf s}_1$ are not the same in these two quantities: in one case ${\bf s}_1$ is a typical stationary point at that energy density, while in the other case it is conditioned. In fact, we shall show that there are values of parameters for which
\begin{equation}\label{eq:Bound2}
\Sigma^{(3)}(\epsilon_2, q_0, q_1 | \epsilon_1,\epsilon_0, q) > \Sigma^{(2)}(\epsilon_2, q_1 |\epsilon_1).
\end{equation}
This leads us to define the notions of \textit{local accumulation} and \textit{clustering} of stationary points, as specific instances of this phenomenon.\\

\noindent \textbf{Local accumulation.} We say that local accumulation occurs whenever for some fixed values of $\epsilon_0,\epsilon_1,q,q_0$, there exists a region of values of $q_1,\epsilon_2$ such that Eq.~\eqref{eq:Bound2} holds. In other words, there are regions in the vicinity of ${\bf s}_0$ in which the number of stationary points ${\bf s}_2$ (conditioned to the properties of ${\bf s}_1$) is higher than the value predicted by the two-point complexity, which measures the complexity in absence of the conditioning to ${\bf s}_0$. \\

 \noindent\textbf{Clustering.} We use the word clustering to designate a special instance of local accumulation, occurring when $\Sigma^{(3)}(\epsilon_2, q_0, q_1 | \epsilon_1,\epsilon_0, q)>0$ but $\Sigma^{(2)}(\epsilon_2, q_1 |\epsilon_1) = - \infty$: there are exponentially many stationary points ${\bf s}_2$ at an overlap $q_1$ (with ${\bf s}_1$) which is large enough that, typically—i.e., in the absence of ${\bf s}_0$—there would be none. In other words, clustering occurs whenever $\Sigma^{(3)}(\epsilon_2, q_0, q_1 | \epsilon_1,\epsilon_0, q)$ is positive for values of $q_1>q_M(\epsilon_2|\epsilon_1)$, where $q_M(\epsilon_2|\epsilon_1)$ was defined in Eq.~\eqref{eq:qMAx} and it identifies the maximal overlap of  stationary points of energy density $\epsilon_2$ with a minimum of energy density $\epsilon_1$. In Fig.~\ref{fig:2D_plot_clustering} we plot again the two-point complexity of ${\bf s}_1$ conditioned to ${\bf s}_0$ for $\epsilon_0=-1.167$ ($p=3$), and we show in the inset the region where clustering occurs (hatched region) for the particular choice $\epsilon_2=\epsilon_1$. This means that within that region there exist values of $q_0$ and $q_1>q_M(\epsilon_2|\epsilon_1)$ such that $\Sigma^{(3)}(\epsilon_2, q_0, q_1 | \epsilon_1,\epsilon_0, q)>0$, where it is intended $\epsilon_2=\epsilon_1$\footnote{otherwise we could not show the region in a 2D plot}. \\

It is clear that local accumulation (as well as clustering, which is a special case of it) indicates the fact that the landscape in the vicinity of ${\bf s}_0$ is strongly correlated to ${\bf s}_0$ itself: the distribution of the other stationary points in that region is not the typical one. In Sec.~\ref{eq:SpecialLines} and Sec.~\ref{sec:LandscapeEvolution} we show that such correlations are indeed present in the landscape, through a quantitative analysis of the results of the three-point complexity calculation. These landscape correlations also have an interpretation in the context of activated dynamics: as we elaborate in Sec.~\ref{sec:en_land_actv_dyn}, they can be seen as a signature of avalanche-like behavior in the dynamics.

\begingroup
\renewcommand{\arraystretch}{0.9}
\begin{table}[h!]
    \centering
    \begin{tabular}{|c|c|c|c|c|c|}
        \hline
        \multicolumn{6}{|c|}{$\epsilon_{gs}\approx -1.17167,\quad\epsilon_{th}\approx -1.1547$} \\
        \hline
        \multicolumn{6}{|c|}{$\epsilon_0=-1.167,\quad \epsilon^*(\epsilon_0)\approx-1.15767$} \\
        \hline
        \text{icon} &$\bm{\epsilon}$ & $\bf{q_{un}}(\epsilon|\epsilon_0)$ & $\bf{q_{ms}}(\epsilon|\epsilon_0)$ & $\bf{q_{M}}(\epsilon|\epsilon_0)$ & $\bf{q}$ \\
        \hline
        $\bigstar$ & -1.155 & 0.5763 & 0.6028 & 0.7564 & 0.56\\
        \hline
        $\blacktriangle$ & -1.1573 & 0.586 & 0.669 & 0.7266 & 0.65\\
        \hline
        $\blacksquare$ & -1.158 &  &  & 0.3829 & 0.35\\
        \hline
    \end{tabular}
    \caption{Values of the parameters of the points marked in Fig.~\ref{fig:2D_plot_clustering}, $p=3$.}
    \label{fig:table}
\end{table}
\endgroup

\section{The three-point complexity: results}
\label{sec:3point_results}
In the following, we discuss the results of the calculation of both the doubly annealed and quenched complexity. Calculations of the doubly annealed complexity are found in \cite{pacco_triplets_2025}, whereas the quenched complexity will be presented in \cite{pacco_quenched_triplets_2025}. \\

\noindent We begin by noting that annealed and doubly annealed complexities coincide:
\begin{equation}\label{eq:EqaAnn}
   \Sigma^{(3)}_{2A}(\epsilon_2, q_0, q_1 | \epsilon_1,\epsilon_0, q)= \Sigma^{(3)}_{A}(\epsilon_2, q_0, q_1 | \epsilon_1,\epsilon_0, q),
\end{equation}
this being a consequence of the fact that the quenched one-point and two-point complexities equal their annealed counterparts. Therefore, the factorization of the expectation values \eqref{eq:averagesAnn} is justified. Instead, the quenched complexity is strictly smaller than the annealed one for generic choices of the parameters \cite{pacco_triplets_2025}. For the values of parameters that we explore, however, the numerical values of quenched and annealed complexity happen to be quite close. Moreover, for special choices of the parameters $q_0, q_1$, the two functions coincide, as we discuss in Sec.~\ref{eq:SpecialLines}. 
\begin{figure}[t!]
\centering
\includegraphics[width=0.95
\textwidth, trim=5 5 5 5,clip]{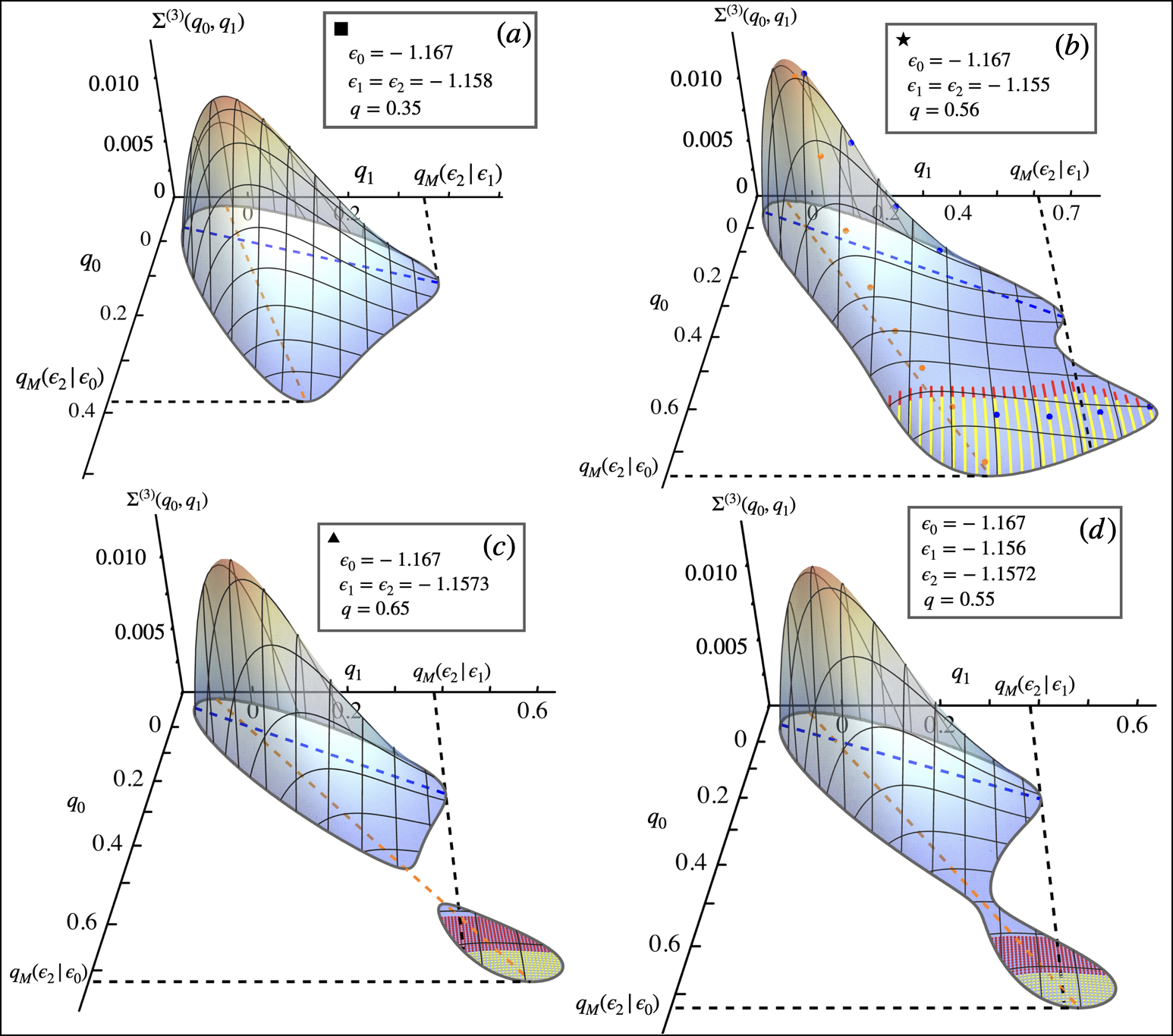}
\caption{Three-point complexity $\Sigma^{(3)}(\epsilon_2,q_0,q_1|\epsilon_1,\epsilon_0,q)$ as a function of overlaps $q_0$ and $q_1$, for specific choices of parameters referring to Fig.~\ref{fig:2D_plot_clustering}, with symbols $\bigstar, \blacktriangle,\blacksquare$. In all plots, the \textit{dashed orange line} in the $(q_0, q_1)$-plane follows $q_1=q q_0$, and the \textit{dashed blue line} follows $q_0=q q_1$. In (b), the orange points indicate the maxima of the complexity at fixed $q_0$, occurring at $q_1=q q_0$. The blue points indicate the maximum at fixed $q_1$, following $q_0=q q_1$ for small $q_1$, but jumping to higher $q_0$ values when $q_1 \approx 0.4$, thus showing the \textit{local accumulation} effect. Red and yellow zones indicate the regions where the stationary points counted, that is ${\bf s}_2$, are typically correlated minima and rank-1 saddles, respectively. In all pictures, ${\bf s}_1$ is taken to be a minimum, except (c), where it is a correlated one (see Fig~\ref{fig:2D_plot_clustering}). Except for (a), where the maximum value $q_M(\epsilon_2|\epsilon_1)$ is not exceeded by the three-point complexity, all the other pictures present \textit{clustering} (meaning that such value is exceeded). In (d), we show an example of clustering where $\epsilon_1\neq \epsilon_2$.}
\label{Fig:Plots3D}
\end{figure}

\subsection{The doubly-annealed complexity, and its reduction to the two-point complexity}\label{sec:DoublyAnnealed}
The calculation of the doubly-annealed complexity gives:
\begin{align}
\begin{split}\label{eq:ann_formulaComp2A}
&\Sigma^{(3)}_{2A}=\frac{Q_{2A}({\bf q})}{2}- f_{2A}({\bm \epsilon}, {\bf q})+I\left(\epsilon_2\sqrt{\frac{p}{p-1}}\right)
\end{split}
\end{align}
where now 
\begin{equation}
\begin{split}
    &Q_{2A}({\bf q})=1+\log\left(\frac{2 (p-1)(1-q^{2p-2})}{1-q^2}\right)+\\
    &\log\left|\frac{1-q^2-q_0^2-q_1^2+2q\,q_0\,q_1}{1-q^{2p-2}-q_0^{2p-2}-q_1^{2p-2}+2(q\, q_0 \,q_1)^{p-1}}\right|
    \end{split}
\end{equation}
and 
\begin{equation}\label{eq:FuncFDA}
\begin{split}
&f_{2A}({\bm \epsilon}, {\bf q})=  \epsilon_2^2Y^{(p)}_{2}({\bf q})+\epsilon_1\epsilon_2 \,Y_{12}^{(p)}({\bf q})+\epsilon_0\epsilon_2 Y_{02}^{(p)}({\bf q})+\\
&\epsilon_0\epsilon_1 [Y_{01}^{(p)}({\bf q})-U(q)]+\epsilon_1^2[Y_{1}^{(p)}({\bf q})-U_1(q)]+\\
&\epsilon_0^2[Y_{0}^{(p)}({\bf q})-U_0(q)-1].
\end{split}
\end{equation}
The functions $I$ and $U_0(q), U(q)$ and $U_1(q)$ are the same functions appearing in the two-point complexity, see Eq.~\eqref{eq:UC}. The remaining $Y^{(p)}({\bf q})$ are 
 functions of the overlaps ${\bf q}=(q, q_0, q_1)$. Their implicit definition is given in Appendix~\ref{app:Y_vars}. Since their explicit expression for general $p$ is rather cumbersome \footnote{we were not able to find an efficient way of simplifying the corresponding expressions, which are better used as implicit functions in a mathematical software.}, we report it only for the case $p=3$, which is the value of $p$ we consider for all plots in this Chapter.\\
We now discuss some special limits of this function. It can be checked explicitly that for $q_1$ fixed, when $q_0= q q_1$, the three-point complexity \eqref{eq:ann_formulaComp2A} reduces to a two point complexity, meaning that:
\begin{equation}\label{eq:Id1}
    \begin{split}
 &\Sigma^{(3)}_{2A}(\epsilon_2, q_0=q q_1, q_1 | \epsilon_1,\epsilon_0, q)\equiv \Sigma^{(2)}(\epsilon_2, q_1|\epsilon_1).
    \end{split}
\end{equation}
Indeed, when $q_0=q q_1$ one finds that all the coefficients multiplying $\epsilon_0$ in \eqref{eq:FuncFDA} vanish exactly, while $Y_1^{(p)}({\bf q})-U_1(q) \to U_0(q_1)$, $Y_2^{(p)}({\bf q})\to U_1(q_1)$ and $Y_{12}^{(p)}({\bf q})\to U(q_1)$. As a result, the three-point complexity becomes independent of the configuration ${\bf s}_0$.
In an analogous manner, for given $q_0$ and $q_1= q q_0$, it holds
\begin{equation}\label{eq:Id2}
    \begin{split}
    \Sigma^{(3)}_{2A}(\epsilon_2, q_0, q_1=q q_0 | \epsilon_1,\epsilon_0, q)\equiv \Sigma^{(2)}(\epsilon_2, q_0|\epsilon_0).
    \end{split}
\end{equation}
In this case, the three-point complexity becomes independent of the configuration ${\bf s}_1$. These two lines cross at $q_0=q_1=0$, where the complexity is maximal, independent of ${\bf q}=(q, q_0, q_1)$ and equal to  $\Sigma(\epsilon_2)$, see Fig.~\ref{Fig:Plots3D} for an example.\\
These two special lines in the $(q_0, q_1)$-plane have a simple entropic interpretation: in fact, if one asks what is the value of the overlap $q_0$ that maximizes the volume of the configuration space associated to ${\bf s}_2$, given the constraints that ${\bf s}_2 \cdot{\bf s}_1= N q_1$
 and ${\bf s}_1 \cdot{\bf s}_0= N q$, one finds that $q_0= q q_1$. 
Analogously, $q_1= q q_0$ maximizes the configuration space associated to ${\bf s}_2$, given the constraints that ${\bf s}_2 \cdot{\bf s}_0= N q_0$ and ${\bf s}_1 \cdot{\bf s}_0= N q$.

\subsection{The quenched complexity}\label{eq:SpecialLines}
We now discuss the results of the calculation of the quenched complexity. We fix again the reference local minimum at $\epsilon_0=-1.167$ (with $p=3$), for the plots in this section. Some relevant values of parameters descending from the two-point complexity for this particular value of $\epsilon_0$ are recalled in Table~\ref{fig:table}. \\

\noindent {\bf When $\Sigma^{(3)}$ reduces to  $\Sigma^{(2)}$.} Representative plots of the quenched three-point complexity $\Sigma^{(3)}(\epsilon_2,q_0,q_1|\epsilon_1,\epsilon_0,q)$ (denoted with $\Sigma^{(3)}(q_0,q_1)$ in the plots for brevity) as a function of the overlaps $q_0, q_1$ are given in Fig.~\ref{Fig:Plots3D} for different values of parameters. One sees that the quenched complexity is always maximal for $q_0=0=q_1$, where it reduces to the one-point complexity $\Sigma(\epsilon_2)$. 
The orange and blue dashed lines in the $(q_0, q_1)$-planes correspond to $q_1= q q_0$ (for fixed $q_0$), and at $q_0= q q_1$ (for fixed  $q_1$), respectively. These are the lines along which the doubly-annealed complexity~\eqref{eq:ann_formulaComp2A} reduces to the two-point complexity, see Eqs.~\eqref{eq:Id1} and \eqref{eq:Id2}. We find that an analogous statement remains true in the quenched calculation: along these lines, the three-point \emph{quenched} complexity coincides with the three-point \emph{doubly-annealed} complexity, 
\begin{equation}
    \begin{split}
 &\Sigma^{(3)}(\epsilon_2, q_0, q_1=q q_0 )= \Sigma^{(3)}_{2A}(\epsilon_2, q_0, q_1=q q_0 ),\\       &\Sigma^{(3)}(\epsilon_2, q_0=q q_1, q_1 )= \Sigma^{(3)}_{2A}(\epsilon_2, q_0=q q_1, q_1),
    \end{split}
\end{equation}
implying:
\begin{equation}
    \begin{split}
 &\Sigma^{(3)}(\epsilon_2, q_0, q_1=q q_0 | \epsilon_1,\epsilon_0, q)= \Sigma^{(2)}(\epsilon_2, q_0|\epsilon_0),\\       &\Sigma^{(3)}(\epsilon_2, q_0=q q_1, q_1 | \epsilon_1,\epsilon_0, q)= \Sigma^{(2)}(\epsilon_2, q_1|\epsilon_1).
    \end{split}
\end{equation}
As explained in Sec.~\ref{sec:DoublyAnnealed}, these lines can be interpreted in terms of maximization of the volume of configuration space accessible to the third configuration ${\bf s}_2$. The value of the complexity along these lines is therefore exact. \\

\noindent {\bf The maxima of the complexity.} We find that for any choice of the parameters ${\bm \epsilon}=(\epsilon_0, \epsilon_1, \epsilon_2)$ and $q$, the quenched three-point complexity at fixed $q_0$ is in fact always maximal for $q_1=q q_0$: the highest number of stationary points ${\bf s}_2$ fulfilling the constraint on $q_0$ is found in the region of configuration space corresponding to the maximal entropy.  Fig.~\ref{Fig:Plots3D}(a) gives an example for which the analogous statement holds true at fixed $q_1$ (the  complexity being maximal at $q_0=q q_1$). In this case, the maxima of the complexity are thus attained at values of the parameters $q_0$ and $q_1$ where the corresponding quenched complexity reduces to the annealed one, and it also coincides with a two-point complexity. Fig.~\ref{Fig:Plots3D}(b) represents instead a different scenario, in which one observes a jump in the value of $q_0$ maximizing the complexity at fixed $q_1$ (see blue points on the 3D plot). This is what happens: for small $q_1$ the maximum is attained at $q_0 = q q_1$; when $q_1$ exceeds a critical value (around $q_1\approx 0.4$ in Fig.~\ref{Fig:Plots3D}(b)) the point where the complexity is maximal jumps to a much larger $q_0$. This means that when ${\bf s}_2$ approaches ${\bf s}_1$, typically one finds a higher number of stationary points in the region where ${\bf s}_2$ is also close to ${\bf s}_0$. This corresponds to the \emph{local accumulation} defined in Sec.~\ref{subsec:defs_clustering}.\\

\noindent {\bf The boundaries of the domain.} The dashed black lines in Fig.~\ref{Fig:Plots3D} indicate the overlaps $q_M(\epsilon_2 | \epsilon_0)$ and $q_M(\epsilon_2 | \epsilon_1)$; these are two quantities related to the two-point complexity, see \eqref{eq:qMAx}. The first one, $q_M(\epsilon_2 | \epsilon_0)$, is the overlap at which one finds the stationary points at energy density $\epsilon_2$ that are  \emph{closest} (at highest overlap) to a stationary point at energy density $\epsilon_0$. One sees from the plots that this is always also the maximal overlap $q_0$ for which the three-point complexity is non-zero. Again, the analogous statement in $q_1$ does not necessarily hold true: except for Fig.~\ref{Fig:Plots3D}(a), in all the other cases the domain where $\Sigma^{(3)}>0$ exceeds $q_M(\epsilon_2 | \epsilon_1)$. 
This corresponds to the  \textit{clustering} defined in Sec.~\ref{subsec:defs_clustering}.\\

Away from the special lines $q_0=q q_1$ and $q_1=q q_0$, the three-point complexity can not be written in terms of the two-point one, meaning that genuine three-point correlations exist between the stationary points. These correlations give rise to several transitions in the structure of the landscape as we change the parameters ${\bm \epsilon}=(\epsilon_0, \epsilon_1, \epsilon_2)$ and ${\bf q}=(q,q_0, q_1)$ describing the properties of the stationary points. We discuss these transitions in Sec.~\ref{sec:LandscapeEvolution}.

\begin{figure}[t!]
\centering
\includegraphics[width=
\textwidth, trim=5 5 5 5,clip]{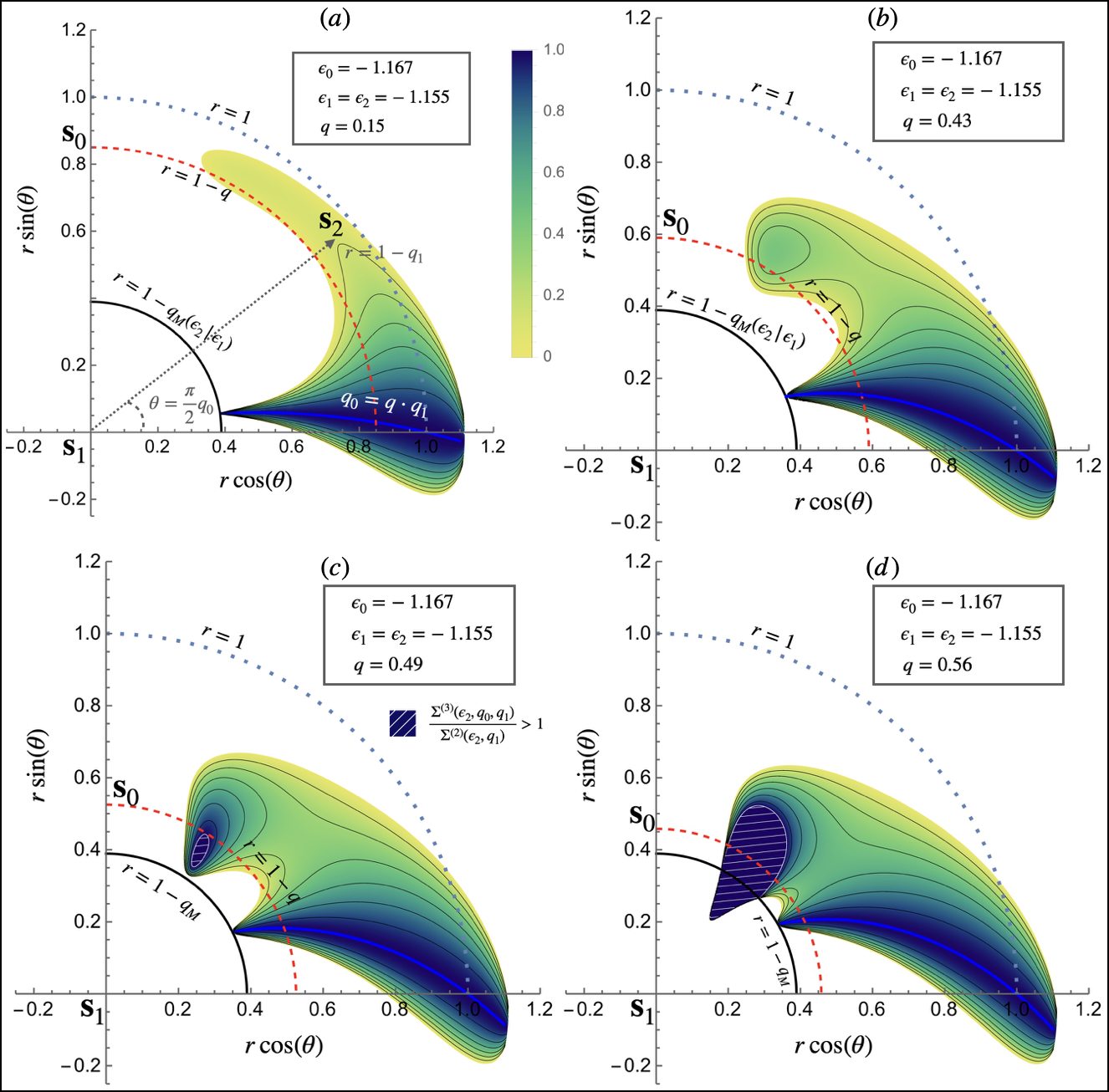}
\caption{ The figures show a density plot of the ratio $\frac{\Sigma^{(3)}(\epsilon_2,q_0,q_1)}{\Sigma^{(2)}(\epsilon_2,q_1)}$. Each point $(\theta, r)$ corresponds to ${\bf s}_2$ with overlaps $q_1=1-r$ and $q_0=2\theta/\pi$. The origin represents ${\bf s}_1$, and the point $(\theta=\pi/2, r=1-q)$ represents ${\bf s}_0$. The blue line follows $q_0=q\,q_1$, where the ratio reaches the value 1. The black arc marks $q_1=q_M(\epsilon_2|\epsilon_1)$, which is exceeded in the case of clustering (figure (d)). Regions with white hatching indicate density values greater than 1, implying \textit{accumulation}. From left to right, and top to bottom, $q$ increases while the other parameters are fixed. In (b), a second maximum emerges along constant $q_1$ curves; in (c), the corresponding area shows accumulation; and in (d), it surpasses $q_1=q_M(\epsilon_2|\epsilon_1)$, indicating \textit{clustering}.}
 \label{Fig:2Ddensityplots}
\end{figure}

\subsection{Dependence on $q$: landscape's transitions}\label{sec:LandscapeEvolution}

To describe the local arrangement of stationary points encoded in the three-point complexity and unveil the presence of landscape's transitions, we consider first the case of fixed energy densities $\epsilon_0=-1.167$, $\epsilon_1= \epsilon_2=-1.155$, and study the landscape's evolution varying the parameter $q$, which measures the overlap between ${\bf s}_0$ and ${\bf s}_1$, in Fig.~\ref{Fig:2Ddensityplots}.  The figure represents a projection of configuration space, where the stationary point ${\bf s}_1$ is placed at the center of the reference frame, while ${\bf s}_0$ is aligned along the $y$ axis at a distance $r=1-q$ (red dashed quarter of a circle). The third stationary point has radial coordinate $r=1-q_1$ (measuring its distance to ${\bf s}_1$) and angular coordinate $\theta= \frac{\pi}{2} q_0$ (measuring its vicinity to ${\bf s}_0$). The colored area of the plot identifies the region of configuration space where we find an exponentially large population of stationary points ${\bf s}_2$ at energy density $\epsilon_2$, meaning that $\Sigma^{(3)}(\epsilon_2, q_0, q_1 | \epsilon_1,\epsilon_0, q)>0$. The intensity of the color plot corresponds to the value of the ratio $\Sigma^{(3)}(\epsilon_2, q_0, q_1 | \epsilon_1,\epsilon_0, q)/\Sigma^{(2)}(\epsilon_2, q_1 | \epsilon_1)$. The different plots in Fig.~\ref{Fig:2Ddensityplots} represent the evolution of the landscape as $q$ increases, meaning that ${\bf s}_0$ and ${\bf s}_1$ are chosen to be progressively closer to each others in configuration space. \\
With increasing $q$, we see the following different regimes:

\begin{itemize}
    \item[(a)] {\bf Depletion regime. } This corresponds to Fig.~\ref{Fig:2Ddensityplots}(a), i.e., to small values of $q$. In this regime the complexity of ${\bf s}_2$ is always maximal along the curve $q_0= q q_1$, which identifies the region of configuration space maximizing the entropy of configurations at fixed $q_1$; along this line (blue in the figures), $\Sigma^{(3)}(\epsilon_2, q_0, q_1 | \epsilon_1,\epsilon_0, q)=\Sigma^{(2)}(\epsilon_2, q_1 | \epsilon_1)$. The complexity of stationary points decreases when moving away from this line, and it is always smaller than $\Sigma^{(2)}(\epsilon_2, q_1 | \epsilon_1)$. In particular, it is progressively smaller as one looks at regions closer and closer to ${\bf s}_0$ (i.e., increasing $\theta$). In this regime, the population of stationary points with energy $\epsilon_2$ in the vicinity of ${\bf s}_0$ is \emph{depleted}. Moreover, no clustering occurs, as indicated by the fact that the colored area never exceeds $r=1-q_M(\epsilon_2|\epsilon_1)$ (black continuous quarter of a circle). An example of the complexity surface in this regime (for a different choice of parameters) is given also in Fig.~\ref{Fig:Plots3D}(a).

    \item[(b)] {\bf Non-monotonic regime. } This corresponds to Fig.~\ref{Fig:2Ddensityplots}(b). In this case, the highest concentration of stationary points ${\bf s}_2$ is again at $q_0= q q_1$. However, at fixed distance to ${\bf s}_1$ (i.e., for fixed $r$), the distribution of stationary points is non monotonic in $\theta$: in the vicinity of ${\bf s}_0$, the number of points ${\bf s}_2$ increases again, and the complexity $\Sigma^{(3)}(\epsilon_2, q_0, q_1 | \epsilon_1,\epsilon_0, q)$ has a local maximum. This is a sign of  correlations in the landscape, i.e., of the fact that the presence of ${\bf s}_0$ affects the distribution of ${\bf s}_2$.  

    \item[(c)] {\bf Local accumulation transition/regime. } This corresponds to Fig.~\ref{Fig:2Ddensityplots}(c). In this case, the highest number of stationary configurations ${\bf s}_2$ is found in the vicinity of ${\bf s}_0$, where the three-point complexity is \emph{larger} than the two-point one and Eq.~\eqref{eq:Bound2} is satisfied. The line $q_0=q q_1$ (continuous blue line) is now a line of maxima of the quenched complexity, that are only local and no longer global. \\
    In this regime, the correlations in the energy landscape become strong enough, to generate \emph{local accumulation} of stationary points, whose complexity exceeds $\Sigma^{(2)}(\epsilon_2, q_1| \epsilon_1)$, see the blue zone with white hatching; the majority of the stationary points ${\bf s}_2$ is no longer found in the region where the volume of configuration space accessible to them is maximal, but it is found closer to ${\bf s}_0$, meaning that energy correlations prevail on entropy.  
     The non-monotonic behavior in $q_0$ is particularly evident for large enough $q_1$ (small $r$), where by increasing $\theta= \frac{\pi}{2} q_0$ one goes through regions of configurations space at intermediate $q_0$ where there is no stationary point of energy $\epsilon_2$. 
    For fixed $q_1$, the landscape undergoes a transition at a critical value of $q$, where the maximum of $\Sigma^{(3)}(\epsilon_2, q_0, q_1 | \epsilon_1,\epsilon_0, q)$ as a function of $q_0$ jumps from $q_0=q q_1$ to a higher value, that depends on the energy densities ${\bm \epsilon}$.
     We refer to this transition as the “local accumulation transition".

    \item[(d)] {\bf Clustering transition/regime. } This corresponds to Fig.~\ref{Fig:2Ddensityplots}(d), i.e. to large values of $q$. This regime is characterized by clustering: the maximal value of $q_1$ (smallest value of $r$) for which the three-point complexity is non-negative is \emph{larger} than the value predicted by the two-point complexity in absence of ${\bf s}_0$, meaning that 
     \begin{equation}\label{eq:q1max}
   q_1^{\rm max}({\bm \epsilon}, q) :=\text{argmax}_{q_1} \left[\max_{q_0} \Sigma^{(3)}(\epsilon_2, q_0, q_1)\right]
    \end{equation}
   satisfies
    \begin{equation}
   q_1^{\rm max}({\bm \epsilon}, q) > q_M(\epsilon_2|\epsilon_1).
    \end{equation}
 This corresponds to the fact that the blue zone with white hatching extends within the black quarter of a circle in the picture: the three-point complexity in that zone is positive, while the two-point complexity $\Sigma^{(2)}(\epsilon_2,q_1|\epsilon_1)$ predicts that there are typically no stationary points at those energies and at those overlaps $q_1$ from a ${\bf s}_1$ extracted without conditioning on ${\bf s}_0$. 
 Examples of the complexity surfaces in this clustering regime are given in Fig.~\ref{Fig:Plots3D}(b-d) as well. We call the associated transition a “clustering transition".
\end{itemize}

In summary, Fig.~\ref{Fig:2Ddensityplots} describes how the distribution of stationary points in configuration space evolves as one tunes the overlap $q$, for a fixed choice of $\epsilon_1= \epsilon_2 > \epsilon_0$. 
 Recall that in the inset of  Fig.~\ref{fig:2D_plot_clustering} we show the region where clustering exists (black hatched region), for fixed $\epsilon_0=-1.167$, and as a function of $\epsilon_1=\epsilon_2$ and $q$. For these values of parameters, clustering always occurs whenever ${\bf s}_1$ is a correlated minimum or an unstable saddle, but it also occurs when ${\bf s}_1$ is a local minimum with an Hessian that shows no correlations to ${\bf s}_0$. The parameters of Fig.~\ref{Fig:Plots3D}(b) are precisely chosen in such a way that  ${\bf s}_1$ is a stable minimum (see $\star$ in Fig.~\ref{fig:2D_plot_clustering}) and clustering is present. In the following section, we describe how much this scenario and the landscape's transitions depend on the choices of the energies $\epsilon_1, \epsilon_2.$

\subsection{Dependence of clustering on the energies}\label{sec:EnergyDependence}
For $\epsilon_1= \epsilon_2 > \epsilon_0$, as shown in Fig.~\ref{Fig:2Ddensityplots} 
a rather rich phenomenology occurs, with two distinct transitions (local accumulation and clustering) happening at different critical values of $q$. We now consider, for $\epsilon_0$ fixed, arbitrary values of $\epsilon_1, \epsilon_2$ in the range $[\epsilon_{\rm gs}, \epsilon_{\rm th}]$, and ask for which choices of the energies local accumulation and clustering occur for at least some values of the parameters ${\bf q}$. We focus in particular on clustering, which is a special case of local accumulations, see \eqref{eq:Bound2}.\\

\noindent Since correlations in the landscape become relevant for large values of the overlaps, to simplify the discussion in this section we set $q=q_M(\epsilon_1|\epsilon_0)$ for each choice  of $\epsilon_1$. We then ask for which values of $\epsilon_2$  clustering occurs in the energy landscape, meaning that there is at least one choice of $q_1, q_0$ for which 
$\Sigma^{(3)}(\epsilon_2, q_0, q_1 | \epsilon_1,\epsilon_0, q) >0$ but $\Sigma^{(2)}(\epsilon_2, q_1 |\epsilon_1)=-\infty$. We find that a crucial role is played by the energy density $\epsilon^*(\epsilon)$. We recall that this critical energy is the one above which rank-1 saddles or correlated minima appear in the landscape in the vicinity of a local minimum ${\bf s}$ of energy density $\epsilon$, in the two-point complexity, see Fig.~\ref{fig:2D_plot_clustering}. This critical energy acquires also another role in the context of our analysis:
we find indeed that clustering can happen only whenever two conditions on the energy densities are met: (i) $\epsilon_1, \epsilon_2 \geq \epsilon^*(\epsilon_0)$, and (ii) $\epsilon_2 \leq \epsilon^*(\epsilon_1)$ \footnote{we checked this statement only for $p=3$, and leave for future investigation to verify it for any $p$}. Therefore, the critical energies identified by the two-point complexity play also a crucial role in determining which stationary points are strongly correlated in the landscape. In Fig.~\ref{fig:2D_plot_clustering_e1e2} we show, for $\epsilon_1, \epsilon_2> \epsilon^*(\epsilon_0)$, the maximal values of $\epsilon_2$  for which clustering occurs (red points), derived by inspecting the three-point complexity: these values match perfectly with $\epsilon^*(\epsilon_1)$.

\begin{figure}[t]
\centering
\includegraphics[width=0.62
\textwidth, trim=5 5 5 5,clip]{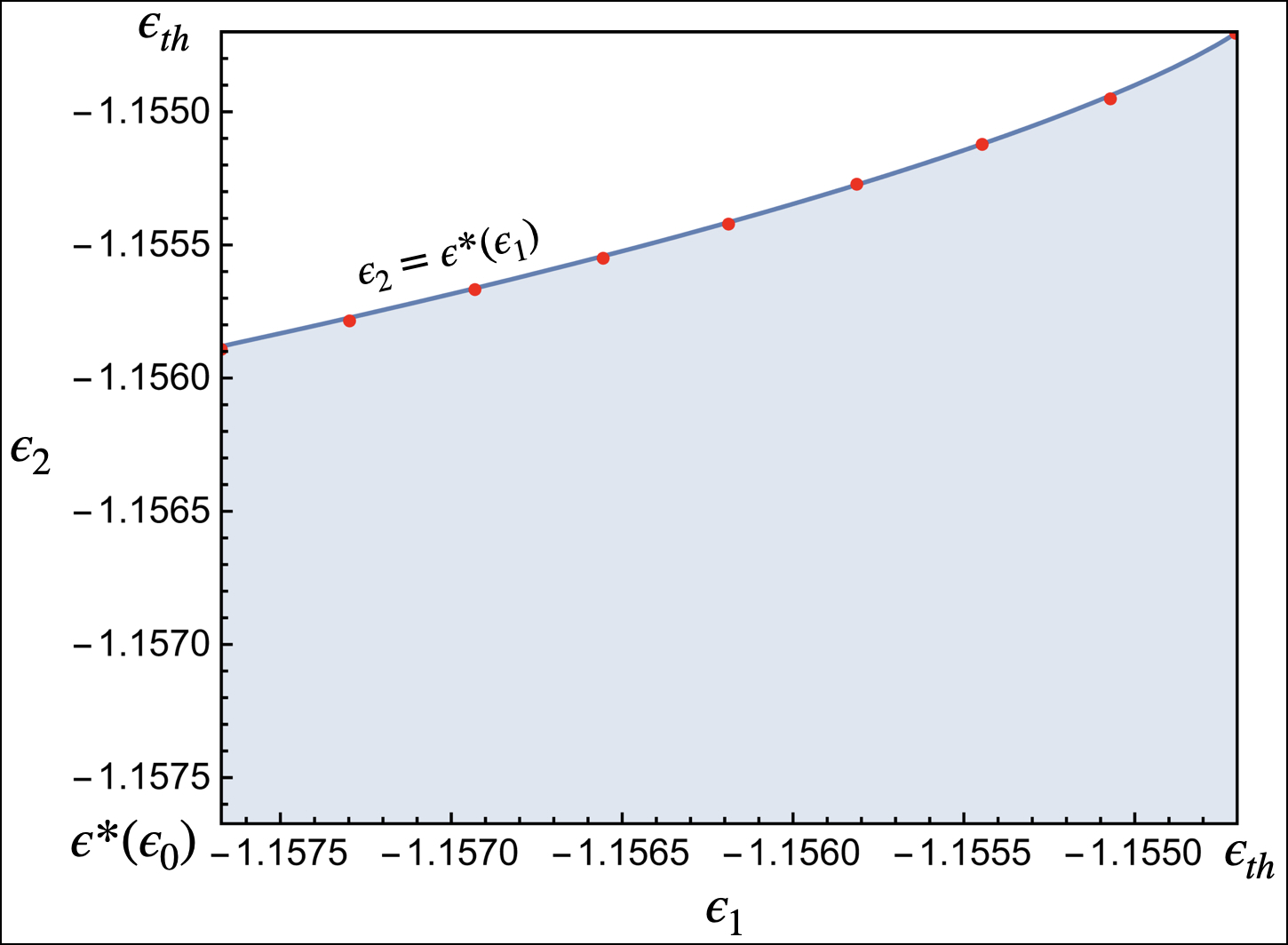}
\caption{The colored area identifies the range of energies $\epsilon_1, \epsilon_2 \geq \epsilon^*(\epsilon_0)$ for which clustering occurs, for $\epsilon_0=-1.167$ and $p=3$. Red points are values extracted from the  analysis of the three-point complexity; the continuous line corresponds to $\epsilon_2=\epsilon^*(\epsilon_1).$  Clustering is present whenever $\epsilon_1>\epsilon^*(\epsilon_0)$ and $\epsilon_2\in[\epsilon^*(\epsilon_0),\epsilon^*(\epsilon_1)]$.}
\label{fig:2D_plot_clustering_e1e2}
\end{figure}

\noindent Notice that clustering appears discontinuously as a function of $\epsilon_1$: for $\epsilon_1 < \epsilon^*(\epsilon_0)$ there is no clustering (no matter what is the value of $\epsilon_2$), whereas as soon as $\epsilon_1\geq  \epsilon^*(\epsilon_0)$ there is a whole range of energies $\epsilon_2$ such that the closest stationary points at those energies are in the clustering region. When clustering occurs for $q=q_M(\epsilon_1|\epsilon_0)$, it is also in general present in the landscape for smaller values of $q$:  Fig.~\ref{Fig:Plots3D}(d) gives an example of clustering occurring when $ \epsilon_1> \epsilon_2> \epsilon^*(\epsilon_0)$, for a value of $q<q_M(\epsilon_1|\epsilon_0)$; see also Fig.~\ref{fig:2D_plot_clustering}, where a large portion of the hatched zone is in the blue region. An important remark is that the conditions on the energies for the existence of clustering are the same with the quenched and doubly-annealed computations. This is a consequence of the simplifications that occur on the line $q_1=q\,q_0$. In fact it is precisely $q_M(\epsilon_1|\epsilon_0)$ that sets the onset of clustering when the energies are increased above $\epsilon^*(\epsilon_0)$.\\

\noindent Let us finally comment on the connections between clustering and isolated modes in the spectrum of the Hessian of the stationary points. The necessary conditions on the energy densities that we have identified for clustering to occur involve the special energy $\epsilon^*(\epsilon)$,
which is related to the appearance of isolated eigenvalues in the Hessian of stationary points counted by the two-point complexity: for such isolated eigenvalue to be present, the energy density of the counted stationary points must be larger than this value. It is interesting that this energy density, that can be determined solely from the two-point complexity, plays this critical role for the appearance of clustering. However, as remarked above, the occurrence of clustering \textbf{is not} in one-to-one correspondence with the presence of isolated eigenvalues in the Hessian spectra of either ${\bf s}_1$ or ${\bf s}_2$: Fig.~\ref{Fig:Plots3D}(b) shows on one hand that stationary points cluster in the vicinity of ${\bf s}_1$ even when the latter is a local minimum with an Hessian that shows no isolated eigenvalues and no signatures of instability; on the other hand, the stationary points ${\bf s}_2$ that cluster are not necessarily unstable saddles or correlated minima, but can be stable minima with an Hessian that shows no signatures of correlations to ${\bf s}_0, {\bf s}_1$ (at least, within the annealed study of the Hessian statistics done in  Ref.~\cite{pacco_triplets_2025} \footnote{given that annealed and quenched complexities are close, we do not expect the quenched statistics of the Hessians to vary significantly. Its analysis proves much harder than the annealed case studied in \cite{pacco_triplets_2025}, and is left for future work.}). Put differently, we want to say that while it is \textit{necessary} that $\epsilon_1,\epsilon_2>\epsilon^*(\epsilon_0)$ to have clustering, once we are above that energy, we can choose values of the overlaps such that at ${\bf s}_1,{\bf s}_2$ the Hessians have no isolated eigenvalues (and are thus solely represented by the semicircle law). 

\subsection{Hints on activated dynamics}
\label{sec:en_land_actv_dyn}
\begin{figure}[t]
\centering
\includegraphics[width=0.7
\textwidth, trim=3 3 3 3,clip]{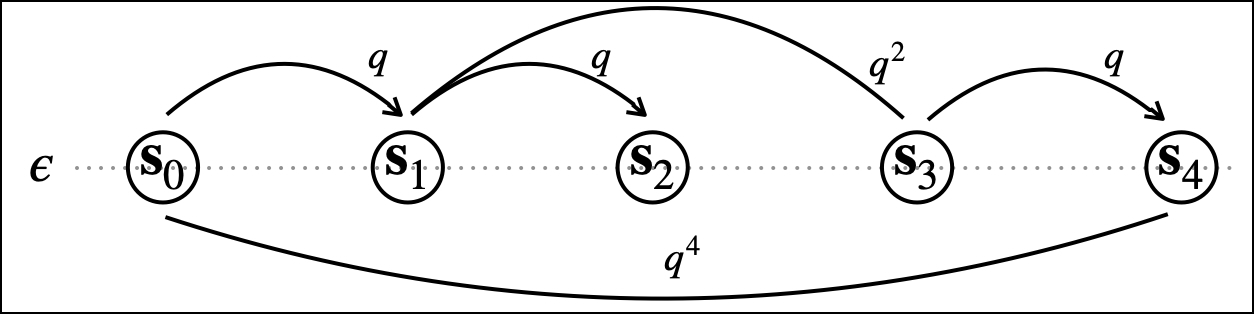}
\caption{Pictorial representation of the memoryless jump process at equal energies, with $q=q_M(\epsilon|\epsilon)$.}
\label{fig:activ_dyn_equal}
\end{figure}
With our results so far, we can try to give some \textit{speculative} hints on the activated dynamics.

\noindent {\bf Typical vs atypical regions of the landscape. } As we saw in Sec.~\ref{subsec:defs_clustering}, while with $\Sigma^{(2)}$ we can analyze \emph{typical} (unconstrained) minima ${\bf s}_1$ with energy density $\epsilon_1$; with $\Sigma^{(3)}$ we get information on \emph{conditioned} stationary points ${\bf s}_1$ with  energy density $\epsilon_1$, at overlap $q$ with another \textit{typical} stationary point ${\bf s}_0$ with energy density $\epsilon_0$. The absence of local accumulation (and so of clustering ) implies that the energy landscape close to ${\bf s}_1$, probed with ${\bf s}_2$, is not strongly affected by the conditioning to ${\bf s}_0$, but it is similar to the landscape in the vicinity of a typical stationary point ${\bf s}_1$ with the same energy density $\epsilon_1$. 
In particular, given a  sequence of three stationary points ${\bf s}_0$, ${\bf s}_1$ and ${\bf s}_2$ with fixed overlaps $q, q_1$ between the consecutive pairs, when there is no local accumulation then the third stationary point ${\bf s}_2$ is typically at overlap $q_0=q \cdot q_1$ with the first one; the corresponding complexity shows no dependence on ${\bf s}_0$, as it coincides precisely with the two-point complexity that one would get neglecting the conditioning to ${\bf s}_0$, see Sec.~\ref{sec:DoublyAnnealed}. Notice that this does not mean that the landscape is totally unaffected by the presence of ${\bf s}_0$: in fact, for all values of $q_0 \neq q q_1$ the three-point complexity is sensitive to the conditioning to ${\bf s}_0$, and it is smaller than the two-point complexity. However, optimizing over $q_0$ this dependence disappears. On the other hand, local accumulation and clustering occur when the landscape in the vicinity of ${\bf s}_1$ is strongly affected by ${\bf s}_0$, and it is characterized by a higher concentration of stationary points with respect to the concentration one finds around a typical point of the same energy density $\epsilon_1$. We can then interpret local accumulation and clustering transitions as \textit{decorrelation-correlation transitions} in the energy landscape. \\

\noindent{\bf Memoryless jumps and avalanche precursors.} 
Low-temperature activated dynamics in glassy landscapes is characterized by a separation of timescales: the system spends long periods fluctuating near a local minimum, until a rare noise-induced fluctuation pushes it over an energy barrier, allowing it to settle into a new local minimum. When modeling this dynamics, it is natural to neglect short-time fluctuations and introduce effective models in which the dynamics is described as a stochastic jump process on a network of states (representing the local minima). Models of this type have been studied quite extensively, from the exactly solvable trap model~\cite{bouchaud1992weak, dyre1987master, monthus1996models, bouchaud1995aging}, up to more recent generalizations~\cite{margiotta2018spectral, margiotta2019glassy, bertin2003cross, tapias2020entropic}. These models assume a Markovian effective dynamics, in which the system jumps with transition rates that do not depend on the configurations (minima) previously visited by the system. In view of this we can now give an interpretation of the "memoryless" and "avalanche-like" properties of the landscape.\\

\noindent\textit{Mermoryless jumps. } In particular, we have seen that if $\epsilon_1,\epsilon_2<\epsilon^*(\epsilon_0)$ then the closest fixed point to ${\bf s}_0$ is at overlap $q=q_M(\epsilon_1|\epsilon_0)$ and, irrespective of this, the closest to ${\bf s}_1$ is at overlap $q_1=q_M(\epsilon_2|\epsilon_1)$, corresponding to an overlap $q_0=q\,q_1$ from ${\bf s}_0$. In this sense, if we assume that activated dynamics proceeds by reaching closest configurations (which corresponds to less particles rearrangements), then the jump from ${\bf s}_1$ to ${\bf s}_2$ is independent on the previous history, that is, on ${\bf s}_0$. Moreover, this jump corresponds to a loss of memory of the initial minimum ${\bf s}_0$, since $q_0=q\,q_1<q$. We have verified that, within a doubly-annealed four-point complexity, this remains true: jumps at low enough energies (and, in particular, at equal energies) are memoryless, and the overlaps are related by simple products of each other, as in the three-point case. Hence, we can conjecture that this property remains true for an arbitrary number of steps: the long time activated dynamics in the pure $p$-spin could follow a memoryless jump process ${\bf s}_0,{\bf s}_1,\ldots,{\bf s}_n,\ldots $ at equal energies $\epsilon$, in which each jump is at overlap $q=q_M(\epsilon|\epsilon)$ and we progressively lose memory of the initial condition, that is, ${\bf s}_a\cdot{\bf s}_b/N=q^{|b-a|}$. We give a pictorial representation of this process in Fig.~\ref{fig:activ_dyn_equal}.\\

\noindent\textit{Avalanche-like jumps. } We have seen that if $\epsilon_{th}>\epsilon_1>\epsilon^*(\epsilon_0)$ and $\epsilon_2<\epsilon^*(\epsilon_1)<\epsilon_{th}$ then the clustering property holds. Namely if, as above, we assume that the first jump takes the largest overlap $q=q_M(\epsilon_1|\epsilon_0)$, then the largest overlap for the second jump is $q_1^{\text{max}}({\bm\epsilon},q)>q_M(\epsilon_2|\epsilon_1)$. This is a consequence of strong correlations to ${\bf s}_0$ at high energies in the landscape. In this sense, the jump is \textit{avalanche-like}, because jumping to a high energetic minimum $\epsilon_1$ (or rank-1 saddle) gives access to exponentially many fixed points at energies $\epsilon_2$ with overlaps bigger than $q_M(\epsilon_2|\epsilon_1)$, which is the maximum overlap to find fixed points for a typical minimum at that energy density $\epsilon_1$. These avalanche-like jumps still occur on times that are exponentially suppressed in $N$ (since the barriers are still extensive), but on times that are nonetheless exponentially larger with respect to memoryless jumps, which are expected to have huge barriers separating them (by our analysis on curvature-driven pathways).\\

\noindent Despite finding such signatures of clustering, we must remark that these are only happening at higher energies than the starting one, thus presenting differences from the low-dimensional systems mentioned in the introduction \cite{durin2024earthquakelike, korchinski_thermal_2025, scalliet2022thirty,tahaei2023scaling, de2024dynamical}, where the configurations reached by thermal avalanches are at equal or lower energies. \\

\begin{figure}[t]
\centering
\includegraphics[width=0.6
\textwidth, trim=3 3 3 7,clip]{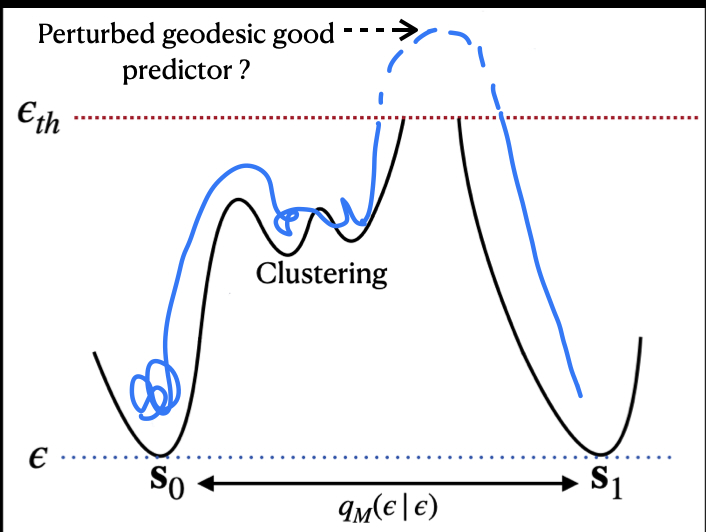}
\caption{Speculative representation of the activated dynamics.}
\label{fig:activ_dyn_pic}
\end{figure}

\section{Conclusion and perspectives.}
\label{sec:en_land_perspectives}
In this Chapter we have used static approaches to probe the geometry of the energy landscape of the pure spherical $p$-spin model. On one side we have used curvature-driven pathways to study energetic barriers between a deep energy minimum and its nearby fixed points, both at low and higher energies, but still below threshold. We have shown that, at difference with models of interacting particles \cite{xu2010anharmonic, widmer2008irreversible}, the softest mode of the Hessian at the start is not a good predictor of efficient energy paths. However, we have seen that having access to the Hessian along the path allows us to find lower energetic pathways, either by leveraging isolated eigenvalues at the arrival point, or by lowering the whole path by following the direction of steepest descent of the gradient. On the other side, we have used the three-point complexity to probe the landscape around a fixed point ${\bf s}_1$ reached after a jump from a reference deep minium ${\bf s}_0$. We have seen that at high energies (but still below threshold) there is a clustering of stationary points, most of which are rank-1 saddles, while at lower energies the landscape around ${\bf s}_1$ is almost independent on the presence of ${\bf s}_0$. We interpreted clustering as a precursor of \textit{avalanche-like} jumps, while still being different from thermal-avalanches in low-dimensional systems \cite{ferrero2017spatiotemporal, tahaei2023scaling}. \\

\noindent\textbf{A picture of activated dynamics. } From the analyses of the three-point complexity and curvature-driven pathways we can try to speculate on the emerging picture of activated dynamics when introducing a small noise into the system, partially supported in view of recent results \cite{rizzo2021path,folena_rare_2025}. We show a pictorial representation in Fig.~\ref{fig:activ_dyn_pic}: starting from a deep minimum at energy density $\epsilon$, it is natural to imagine that the system will initially jump to one of closest rank-1 saddles \cite{ros2021dynamical}, which corresponds to a high-energetic region with notable \textit{clustering} effects. Thanks to the vicinity of exponentially many nearby rank-1 saddles and correlated minima, the system may spend some time in this high energetic region (which \textit{could} correspond to the \textit{hub} found in \cite{folena_rare_2025}), and where the energy barriers (thanks to our analysis on curvature-driven pathways) are much smaller due to the presence of isolated modes in the spectra of nearby fixed points. From here the system would have easier access to above threshold regions. Indeed, according to our picture on curvature-driven pathways, and on recent work \cite{rizzo2021path}, the system is likely to reach above threshold regions before reaching other deep minima. The minima reached after overcoming the threshold would then correspond to the "memoryless" minima described in the previous section, that is, minima at equal (or lower energy) \cite{anderson1964hard,ioffe1987dynamics} and overlap $q_M(\epsilon|\epsilon)$. This process would then repeat, until eventually losing memory of the initial condition, since after each macroscopic jump to a new deep minimum, the distance from every other previous minimum decreases as a power of $q_M(\epsilon|\epsilon)$. It would be interesting to test whether the maximal energy barrier reached by the geodesic perturbed with ${\bf v}_\text{Hess}$ is a proxy to the real energy barrier overcome by the system to reach a new deep minimum.  \\

\noindent This picture remains speculative, but two takeaways should be retained: no below-threshold obvious paths between deep minima exist, and while local information can be used to lower the barrier, it still lies above threshold; strong correlations in the landscape are present among the high-energetic fixed points close to a deep minimum. \\

\noindent\textbf{Perspectives. } It would be interesting to extend our analyses to mixed models, see Sec.~\ref{sec:mixed_models}. Their two-point complexity was computed in \cite{kent2024arrangement}, where the author observes that the neighborhood of typical marginal minima is rather different above or below threshold. The threshold value is the only one where typical marginal minima are found arbitrarily close to each other, and are separated by sub-extensive barriers. Instead, above or below threshold typical marginal minima are separated by extensive energy barriers, and there is an overlap gap between them (similarly as $q_M(\epsilon_1|\epsilon_0)$ in Fig.~\ref{fig:2_point_phase_diag}). Moreover, marginal states above and below the threshold have different neighborhoods, rendering the relaxational dynamics of such systems, which are believed to converge to marginal minima, an even more interesting mystery \cite{folena2020rethinking}. We expect that the calculation of the perturbed geodesic pathways, as well as the doubly-annealed three-point complexity, should be within reach for mixed models. Since in mixed models we can freely choose the extensive stability of the stationary points (see Sec.~\ref{sec:mixed_models}), we expect that the results will differ considerably. In particular, we expect clustering (in the sense of the three-point complexity) to be a much more noticeable effect for mixed models. Indeed, consider a series of three stationary points with same energies $\epsilon$: first a local minimum at stability given by the Lagrange multiplier $\mu_0$; then another local (or marginal) minimum that satisfies
\begin{align}
    \mu_1^{max}=\mu\,\,\text{associated to the largest $q_{01}^{max}$ such that }\,\Sigma^{(2)}(\epsilon,\mu,q_{01}^{max}|\epsilon,\mu_1)=0,
\end{align}
and finally a third stationary point chosen with $\mu_2^{max}$ and $q_{12}^{max}$. Then we expect that, in general, $q_{12}^{max}>q_{01}^{max}$ and that if we imagine a series of more points, the overlaps saturate to a certain value. In the pure model, given that we cannot choose the $\mu$s, the maximum overlap given by the three-point complexity is always the same at every "jump", and thus equivalent to just considering the maximum overlap from the two-point complexity. If the same remains true for the mixed models, we could just use the two-point complexity in \cite{kent2024arrangement} to check how maximum overlaps at equal energies evolve after adding new points. Of course, it would be interesting to study the distribution of triplets of stationary points at any energies, and see which new features the mixed model has. Whether such features could be interesting for the problem of understanding gradient descent in these landscapes remains open.

\chapter{Spiked, correlated random matrices}
\label{chapter:rmt_}
In this chapter we consider a variant of the well-known spiked matrix problem. More precisely we consider pairs of spiked, correlated GOE (Gaussian Orthogonal Ensemble) random matrices. Our aim is twofold: on one side we derive the $N>>1$ spectral properties of these matrices; on the other we compute the expected overlap between their eigenvectors (i.e. the squared scalar product). This problem is of interest both for the spiked matrix problem, as well as for problems of high-dimensional random landscapes considered in Chapter~\ref{chapter:energy_landscapes}, which motivated us for this work. \\

\noindent\textit{Road-map}\\
In Sec.~\ref{sec:rmt_introduction} we review the basic properties of GOE and spiked GOE matrices. In Sec.~\ref{sec:rmt_theoretical_res} we present the specific spiked model under study and compute its spectral properties. In Sec.~\ref{sec:theoretical_res2} we introduce the overlaps between eigenvectors of the spiked, correlated random matrices, and we show how they are computed. In doing so, we also show new results on multiresolvent products. In Sec.~\ref{sec:rmt_simulations_overlaps} we give the final expressions of the overlaps and compare with numerical simulations; finally we show how our results are useful for problems of signal recovery.\\

\noindent\textit{Acknowledgments}\\
This was joint work with Valentina Ros, and the results are found in Ref.~\cite{paccoros}. I thank her very much for our stimulating discussions that led to this work.

\section{Introduction}
\label{sec:rmt_introduction}
In this section we review the basic properties of GOE and of GOE spiked random matrices, and we motivate our work. All results are explained in a physicists' informal style, as will be our calculations.  Every important calculation is nonetheless confirmed with many numerical simulations.

\subsection{GOE cookbook (informal)}
\label{sec:Goe_intro}
In this section we will briefly introduce the GOE random matrices, summing up the main quantities and results that we will need along the way. A comprehensive introduction to random matrix theory for physicists is found in the book by Potters and Bouchaud \cite{potters_bouchaud_2020}, but there are also many other resources \cite{vivo_book_rmt_2018, ros_lecture_2025, mehta_rmt, Tao_book_2012}. \\

\noindent A GOE matrix ${\bf X}$ of size $N\times N$ and variance ${\bf\sigma}^2$ is a symmetric matrix with Gaussian entries with the following statistics:
\begin{align}
\label{eq:goe_def_X}
    \mathbb{E}[X_{ij}]=0,\quad\quad \mathbb{E}[X_{ij}X_{kl}]=\frac{\sigma^2}{N}(\delta_{il}\delta_{jk}+\delta_{ik}\delta_{jl}).
\end{align}
The joint probability density of the random variable ${\bf X}$ can be written as
\begin{align}
    P({\bf X})=Z_N^{-1}e^{-\frac{N}{4\sigma^2}\text{Tr}{\bf X}^2},
\end{align}
where integration is done over the $N(N+1)/2$ elements on or above the diagonal (the matrix being symmetric). From this probability density, it is clear why the ensemble is called "Orthogonal", namely because probabilities are left invariant under orthogonal transformations. In fact, consider an orthogonal matrix ${\bf O}$ such that ${\bf O}{\bf O}^\top=1$, and the rotated matrix $\tilde{{\bf X}}={\bf O}{\bf X}{\bf O}^\top$, then from the cyclicity of the trace and $|\det{\bf O}|=1$ it follows that:
\begin{align}
    P({\bf X})d{\bf X}=P(\tilde{\bf X})d\tilde{{\bf X}}.
\end{align}
This statistical equivalence between pairs of rotated matrices implies in particular that the distribution of (normalized) eigenvectors is the one of random vectors on the hypersphere of unit radius in $\mathbb{R}^N$. \\

\noindent Now, to each realization of ${\bf X}$ one can associate an empirical spectral distribution (ESD):
\begin{align}
    \rho_N(\lambda):=\frac{1}{N}\sum_i\delta(\lambda_i-\lambda),\quad\quad\;\lambda_i\quad\text{eigenvalues of {\bf X}}.
\end{align}
\noindent When $N\to\infty$ we can see by doing a simple simulation that the spectrum of these matrices takes a very precise form, see Fig.~\ref{fig:goe_vs_spike} \textit{left}, which resembles a semicircle (a semiellipse to be more precise). The properties of this spectrum were first studied by Wigner in the 50's \cite{Wigner_1,Wigner_2}. It is now well known that if we take a sequence $\{{\bf X}_N\}_{N\geq 1}$ of GOE matrices of the form \eqref{eq:goe_def_X} then the sequence of ESDs $\{\rho_N\}_{N\geq 1}$ converges almost surely to the deterministic probability measure $\rho_\sigma$ denoted as \textit{Wigner's semicircle law}:
\begin{equation}\label{eq:DoSgoe}
\rho_\sigma(\lambda)=\frac{1}{2 \pi \sigma^2} \sqrt{4 \sigma^2 - \lambda^2}\,\mathds{1}_{|\lambda|\leq 2\sigma}
\end{equation}
see Tao \cite{Tao_book_2012} Sec.~2.4 or Potters and Bouchaud \cite{potters_bouchaud_2020} Sec.~2.2. \\

\noindent Two fundamental objects that we will need all along this Chapter are the \textit{resolvent} and the \textit{Stieltjes transform}, that for a generic $N\times N$ random matrix ${\bf M}$ are defined as
\begin{align}
\label{eq:rmt_def_res_stiel}
{\bf G}(z)=(z-{\bf M})^{-1},\quad \mathfrak{g}(z)=\frac{1}{N}\text{Tr}\,{\bf G}(z),\quad z\in\mathbb{C}\setminus \text{Sp}({\bf M}) 
\end{align}
with  $\text{Sp}({\bf M})$ the set of eigenvalues. Notice that the averaged Stieltjes transform is the moment generating function of ${\bf M}$. All normalized quantities of interest for our discussion (mainly Stieltjes transform and spectral distribution) are \textit{self-averaging} \cite{potters_bouchaud_2020} when $N\to \infty$, meaning that they converge to their expected values. We will therefore use the same notation for both the averaged and non-averaged quantities for notational simplicity.  \\

\noindent The importance of the Stieltjes transform is (at least) twofold: when $N\to\infty$ it satisfies a very particular relation with the cumulant generating function \cite{Burda_non_herm_2011} (the $R$ transform), which turns out to be useful when considering sums of independent large random matrices; the knowledge of $\mathfrak{g}$ allows one to obtain the spectral distribution as well as other important quantities such as eigenvector overlaps, that is, their squared scalar product. \\

\noindent The relation between the Stieltjes transform and the spectral distribution is obtained via the Sokhotski–Plemelj theorem:
\begin{align}
{\displaystyle \lim _{\eta \to 0^{+}}{\frac {1}{x\pm i\eta }}=\mp i\pi \delta (x)+{\mathcal {P}}{{\Big (}{\frac {1}{x}}{\Big )}}}
\end{align}
which also implies:
\begin{align}
{\displaystyle \lim _{\eta \to 0^{+}}\left[{\frac {1}{x-i\eta }}-{\frac {1}{x+i\eta }}\right]=2\pi i\delta (x),}
\end{align}
with $\mathcal{P}$ the principal value. By using this identity one can show that the spectral distribution of ${\bf M}$ is obtained from the Stieltjes transform as:
\begin{align}\label{eq:InvStj}
    \rho(\lambda)=\frac{1}{\pi}\lim_{\eta\to0^+}\text{Im}\left[\mathfrak{g}(\lambda-i\eta)\right].
\end{align}
In the case of the GOE random matrices, one can show that in the limit of $N\to\infty$, the Stieltjes transform converges to the following complex valued function:
\begin{align}
\label{eq.Stjlt}
\mathfrak{g}_\sigma(z):=\frac{z-\text{sign}(\text{Re}(z))\sqrt{z^2-4\sigma^2}}{2\sigma^2},\quad z\in\mathbb{C}\setminus [-2\sigma,2\sigma],
\end{align}
where we excluded the real interval $[-2\sigma,2\sigma]$ since the function presents a branch cut on this interval. Indeed, by considering the behavior of the square root near the real axis (see below), we can easily see that $\mathfrak{g}_\sigma$ is not continuous across the transition from positive to negative imaginary axis in the interval $[-2\sigma,2\sigma]$. This was expected if we look at the definition of $\mathfrak{g}$ in \eqref{eq:rmt_def_res_stiel}, which presents a pole at each eigenvalue of the matrix. As $N\to\infty$ the distribution of GOE eigenvalues expands over the interval $[-2\sigma,2\sigma]$, which leads to a branch cut in $\mathfrak{g}_\sigma$. Along this Chapter we will consider the principal value of the square root, defined by:
\begin{align}
    z=re^{i\theta}\longmapsto \sqrt{z}:=\sqrt{r}e^{i\theta/2},\quad\quad -\pi<\theta\leq \pi,
\end{align}
which means that we use the non-positive real axis as the branch cut. Then we can show that:
\begin{align}
    \lim_{\eta\to 0^+}\sqrt{(x\pm i\eta)^2-4\sigma^2}=
    \begin{cases}
        \sqrt{x^2-4\sigma^2}\quad |x|> 2\sigma\\
        \pm\text{sign}(x)i\sqrt{4\sigma^2-x^2}\quad |x|\leq 2\sigma
    \end{cases}
\end{align}
and by applying this to $\mathfrak{g}_\sigma$, defined in Eq.~\eqref{eq.Stjlt}, we obtain:
\begin{equation}
\lim_{\eta \to 0^+}\mathfrak{g}_\sigma(x \mp i \eta)=\begin{cases}
  \frac{1}{2 \sigma^2} \tonde{x- \text{sign} (x) \sqrt{x^2-4 \sigma^2}} \quad  |x|> 2 \sigma\\
  \frac{1}{2 \sigma^2}\tonde{x \pm i \sqrt{4 \sigma^2-x^2}} \quad  |x|\leq 2 \sigma
  \end{cases} \equiv
\mathfrak{g}_R(x) \pm i \mathfrak{g}_I(x).
\end{equation}
From this we see that the function has a branch cut in $[-2\sigma,2\sigma]$, and moreover this formula will be very important for our computations in the rest of the Chapter. Let us also comment on the choice of $-\text{sign}$ in Eq.~\eqref{eq.Stjlt}: first we know that for large $|z|$, the Stieltjes transform must behave as $1/z$ (as can be seen from its definition), and the choice of $-\text{sign}$ gives the correct scaling ; second, the term $\sqrt{z^2-4\sigma^2}$ has an additional branch cut on the imaginary axis, while the Stieltjes transform should not: the choice of $-\text{sign}$ also solves this issue.\\

\noindent The last important ingredient that we will need for the upcoming calculations is Wick's theorem.\\

\noindent \textit{Wick's theorem. }
If ${\bf X}$ denotes a zero-mean multivariate Gaussian distribution, then for any $n\in\mathbb{N}$ we have the following:
\begin{align*}
    \mathbb{E}[X_1X_2\ldots X_{2n+1}]=0
\end{align*}
and 
\begin{align*}
    \mathbb{E}\left[X_1\ldots X_{2n}\right]=\sum_{\pi\in P_{2n}}\prod_{(i,j)\in \pi}\mathbb{E}[X_iX_j],
\end{align*}
where $P_{2n}$ is the set of all ways to make distinct pairs of $\{1,\ldots,2n\}$, which has cardinality $n!/(2^{n/2}(n/2)!)$.

\begin{figure}[t!]
\centering
\includegraphics[width=\textwidth]{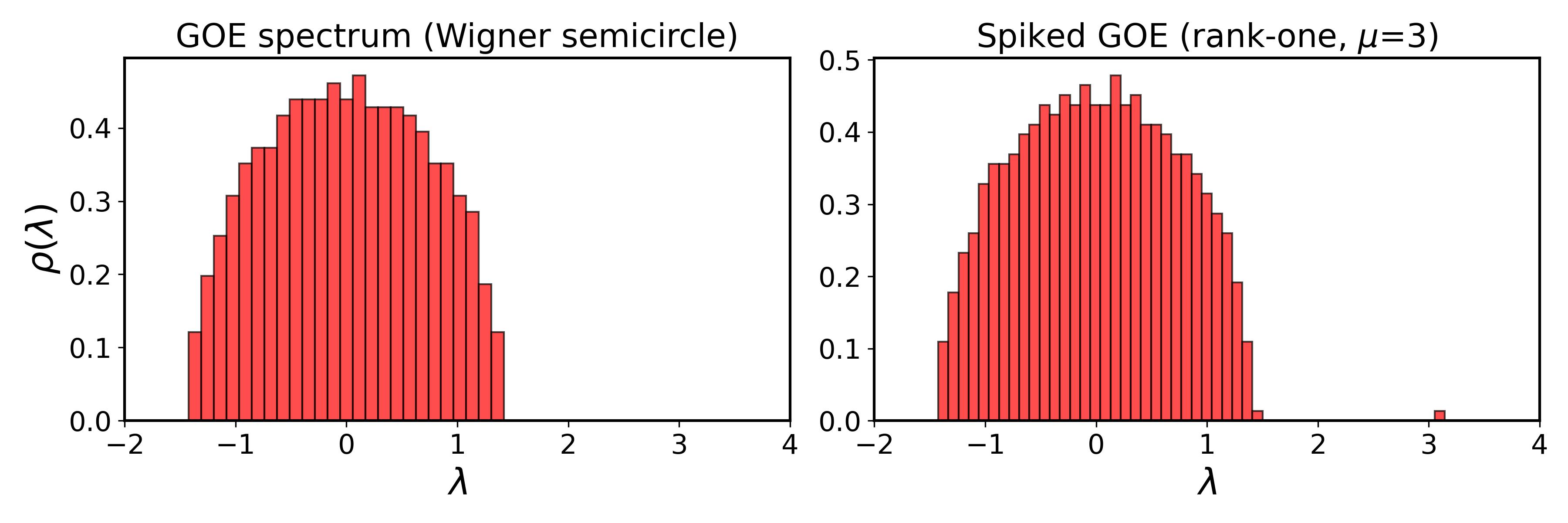}
\caption{Numerical simulation of the eigenvalue spectrum of a GOE matrix (\textit{left}) and a spiked GOE (\textit{right}). Both matrices have variance $\sigma^2=1$, and on the \textit{right} we have $\mu=3$. We see that an outlier eigenvalue pops out of the bulk of the spiked matrix. }
\label{fig:goe_vs_spike}
\end{figure}

\subsection{Spiked random matrices}
\label{sec:rmt_intro_spiked}
Spiked random matrices are random matrices deformed by additive and/or multiplicative perturbations (see Eq.~\ref{eq:rmt_mle} for an additive example) and their study goes back almost 40 years \cite{jones1978eigenvalue, edwards1976eigenvalue,furedi1981eigenvalues}. When the size of these matrices is large and these perturbations are low-rank, there can be isolated eigenvalues that pop out of the bulk of the spectrum. Therefore, these perturbations are often referred to as "spikes" \cite{johnstone2001distribution}, see e.g. Fig.~\ref{fig:goe_vs_spike}~\textit{right}. The appearance of an isolated eigenvalue takes the name of a "BBP transition", with reference to the seminal work by Baik, Ben Arous and Péché \cite{baik2005phase}, where they analyse spectral transitions in perturbed Wishart matrices. Spiked matrix problems and their BBP transitions are still the subject of current research in the mathematical community ~\cite{peche2006largest, benaych2011eigenvalues, benaych2012singular, capitaine2011free, knowles2014outliers, tao2013outliers, bordenave2016outlier, rochet2016isolated, capitaine2016spectrum}. These spiked random matrices and their isolated eigenvalues also play a relevant role in several applications, including: finance~\cite{bun2017cleaning, bun2018overlaps, potters_bouchaud_2020}, inference and detection problems~\cite{montanari2015limitation, lu2020phase}, constraint satisfaction problems~\cite{hwang2020force}, quantum chaos~\cite{fyodorov2022extreme}, localization of polymers by defects~\cite{ikeda2022bose}, theoretical ecology~\cite{fraboul2021artificial, baron2022non}, spin glasses~\cite{ros2021dynamical, ros2020distribution, paccoros} and random neural networks, cf. Chapter~\ref{chapter:scs}. Of particular importance, as we shall see in this Chapter, are the eigenvectors associated to these outliers: their projection on the subspace spanned by the low-rank perturbations remains finite in the limit of large matrix size~\cite{nadler2008finite, benaych2011eigenvalues}. 

\noindent As we show below, spiked matrix problems can be interpreted as \emph{signal versus noise} problems, the low-rank perturbations representing the signal. Then, when the signal is too weak, recovering it is impossible, because the spectral properties of the matrix are the same as if there were no perturbations. There exists a critical value of the strength of the signal beyond which recovery is possible, at least partially. This analysis is done by studying the extremal eigenvalues of the matrix and their associated eigenvectors - that is, by means of principal component analysis (PCA).\\

\noindent
\textbf{Basics of the problem.}\\
Let us give an introduction to the main properties of these spiked matrix problems, reviewed extensively in Ref.~\cite{ros_lecture_2025}. In the simplest setting, the goal is to infer a signal ${\bf e}_N$ corrupted by noise:
\begin{align}
    {\bf M}={\bf X}+\mu\,{\bf e}_N\,{\bf e}_N^\top,
\end{align}
where the first term represents the noise, modeled here with a GOE matrix of variance $\sigma^2$ and size $N\times N$, and the second is a rank-one perturbation with ${\bf e}_N^2=1$ and $\mu>0$ (without loss of generality). For large $N$ this problem is high-dimensional, and it could therefore be challenging to recover the signal. Assuming the knowledge of the form of the noise and its variance ${\bf\sigma}$, we want to know whether we can say something about the signal in the large $N$ limit. The ratio $\mu/\sigma$ is commonly known as \textit{signal-to-noise ratio}: if $\sigma=0$ we can perfectly recover the signal, and if $\sigma>>\mu$ it will be undetectable. Then the question is whether one can retrieve information on $\mu$ and ${\bf e}_N$ by observing ${\bf M}$. To achieve this we will consider the maximum likelihood estimator, which in this case is given by \cite{ros_lecture_2025}:
\begin{align}
\label{eq:rmt_mle}
    {\bf s}_{\text{MLE}}:=\underset{{\bf s}:\,{\bf s}^2=1}{\text{argmax}}\,\,{\bf s}^\top\,{\bf M}\,{\bf s}.
\end{align}
Notice that ${\bf s}_\text{MLE}={\bf u}_\text{N}$, the eigenvector associated to the maximal eigenvalue $\lambda_N$ of ${\bf M}$. Therefore, we see that the problem of signal recovery is associated to the problem of identifying the largest eigenvalue and its eigenvector. Indeed, it is evident that when $\mu>>\sigma$ the maximal eigenvalue of ${\bf M}$ will be close to $\mu$, and ${\bf e}_N$ will be very close to the eigenvector ${\bf u}_N$. In particular when $\sigma=0$ this equivalence becomes exact. It is very interesting to remark that ${\bf s}_\text{MLE}$ also coincides with the ground state of the following energy landscape:
\begin{align}
    \mathcal{E}_\mu({\bf s})=-\sum_{i,j}X_{i,j}s_is_j-\mu\,q_{{\bf s},{\bf e}_N}^2
\end{align}
where $q_{{\bf s},{\bf w}}={\bf s}\cdot{\bf w}$ is the overlap between two vectors on the hypersphere in $\mathbb{R}^N$. When $\mu=0$ this landscape precisely coincides with the $p=2$ spherical spin glass, introduced in \cite{kosterlitz1976spherical}. One can show that, in general, the stationary points of $\mathcal{E}_\mu$ correspond to the eigenvectors of ${\bf M}$, which is rather intuitive if one looks at \ref{eq:rmt_mle}. More precisely, for each eigenvector there are two stationary points of $\mathcal{E}_\mu$, making a total of $2N$ of them. The number of stationary points therefore does not scale exponentially in $N$, and the complexity is zero. If we order the eigenvalues as $\lambda_1\leq \ldots \leq \lambda_N$, then the stationary point associated to the eigenvalue $\lambda_\alpha$ is the eigenvector ${\bf u}_\alpha$ which has a Hessian with an instability index (number of negative modes) given by:
\begin{align}
    I({\bf u}_\alpha)=\alpha-N.
\end{align}
From this we see that all stationary points are saddles, except the two ground states. \\

\noindent Let us now come back to the signal recovery problem. We say that the signal can be detected if 
\begin{align}
    \lim_{N\to\infty} \mathbb{E}[q_{{\bf u}_N,{\bf e}_N}]>0,
\end{align}
meaning that the estimator must have a non-zero overlap with the signal. If $\mu=0$ then it is clear that ${\bf u}_N$ is a random vector with respect to the signal, since the expected overlap is $1/N$ by isotropy, which goes to zero for large $N$. In the limit of $N\to\infty$, this criterion actually corresponds to the appearance of an isolated eigenvalue in the spectrum of ${\bf M}$, which is otherwise given by the Wigner's semicircle. This eigenvalue was first computed in Refs.~\cite{jones1978eigenvalue, edwards1976eigenvalue, furedi1981eigenvalues}. In particular, one has (see also Sec.~\ref{sec:pca} where we recover these results as special cases of our work):
\begin{align}
    \lim_{N\to\infty}\mathbb{E}[\lambda_N]=\begin{cases}
        2\sigma\quad\quad\text{if}\quad\quad \mu\leq \sigma\\
        \frac{\sigma^2}{\mu}+\mu\quad\quad\text{if}\quad\quad \mu> \sigma\\
    \end{cases}
\end{align}
and 
\begin{align}
    \lim_{N\to\infty}\mathbb{E}[q_{{\bf u}_N,{\bf e}_N}]=\begin{cases}
        0\quad\quad\text{if}\quad\quad \mu\leq \sigma\\
        1-\frac{{\sigma}^2}{\mu^2}\quad\quad\text{if}\quad\quad \mu> \sigma\\
    \end{cases}.
\end{align}
The picture that emerges is therefore clear: when $\mu$ is bigger than $\sigma$ the perturbation is strong enough to generate a spike in the otherwise Wigner semicircular spectrum, and this also coincides with the moment where the overlap between the maximal eigenvector and the signal becomes of order 1. When instead $\mu$ is smaller than $\sigma$ the spectrum is indistinguishable from the GOE spectrum, and ${\bf u}_N$ is orthogonal to the signal. Moreover the detection threshold $\mu=\sigma$ has been shown to be optimal, meaning that no other estimator can distinguish between the GOE and the spiked matrix \cite{perry2018optimality}. Let us also remark that the quantities described above (eigenvalues and overlaps) are self-averaging \cite{ros_lecture_2025}, meaning that they converge to their average values as $N\to\infty$.

\subsection{Motivation: curvature driven pathways}

In this work, we are interested in characterizing the squared overlaps between the eigenvectors of \emph{pairs} of correlated, GOE random matrices which are deformed by rank-1 additive \emph{and} multiplicative perturbations, building on the work in Refs.~\cite{bun2016rotational,bun2018overlaps}. In certain parameter regimes, these perturbations generate outliers in the spectra of the pair of matrices: we determine the overlap between the eigenvectors of the outliers, as well as between the eigenvector of the outlier of one matrix and any other eigenvector associated to eigenvalues in the bulk of the other matrix. This analysis is motivated by the study of curvature-driven pathways in energy landscapes of Chapter~\ref{chapter:energy_landscapes}, and more specifically of Sec.~\ref{sec:curvature_driven_paths}. Indeed we have seen around Eq.~\ref{eq:curvature_chi0_chi1} that in order to compute the energy profile along paths between two minima, we need to compute the expressions $\mathbb{E}[\chi_0],\mathbb{E}[\chi_1]$ (cf. Eq.~\ref{eq:curvature_chi0_chi1}). In particular, such expressions were necessary to obtain the energy along perturbed paths (between two fixed points) that try to leverage the local structure of the Hessians. We have seen that such Hessians are correlated GOE matrices with rank-1 additive and multiplicative perturbations, just like those in this Chapter. When the Hessian at the arrival point has an isolated eigenvalue, we saw in Eq.~\ref{eq:curvature_chi0_chi1} that we had to compute the overlap between the eigenvector associated to that eigenvalue, and any eigenvector of the starting Hessian. Hence, here we present a much more general calculation, which has as a special case the one we need for Sec.~\ref{sec:curvature_driven_paths}.\\

\noindent The study of eigenvector overlaps for statistical purposes is a rather recent topic of research ~\cite{potters_bouchaud_2020, bun2016rotational, ledoit2011eigenvectors, allez_free_2014, bun2017cleaning, bun2018overlaps, bun2018optimal, allez_cov_overlaps_2025}. More precisely, our work extends some results of Refs.~\cite{bun2018overlaps, potters_bouchaud_2020}. The work in \cite{bun2018overlaps} is motivated by problems of estimating a (fixed) matrix ${\bf C}$ corrupted by noise, via the so called Rotational Invariant Estimator. The authors consider two cases: multiplicative noise where ${\bf S}=\sqrt{{\bf C}}{\bf \mathcal{W}}\sqrt{{\bf C}}$ and  ${\bf \tilde{S}}=\sqrt{{\bf C}}{\bf \tilde{\mathcal{W}}}\sqrt{{\bf C}}$ with ${\bf \tilde{\mathcal{W}}},{\bf \mathcal{W}}$ independent Wishart matrices; additive noise with 
${\bf S} = {\bf C} +{\bf W}$ and ${\bf \tilde{S}} = {\bf C} + {\bf \tilde{W}}$ with ${\bf W}$ and ${\bf \tilde{W}}$ two independent GOE matrices. In both cases they are able to obtain the expected overlap between eigenvectors of ${\bf S},{\bf \tilde{S}}$ in the large $N$ limit, observing that it does not depend on ${\bf C}$ explicitly, but only on observable quantities and on the statistics of the random noise. As we will see below, we consider a slightly different problem, with correlated GOE matrices perturbed by both additive and multiplicative perturbations. In particular, we do not have a fixed matrix ${\bf C}$, but rather a common GOE matrix ${\bf H}$ to both matrices.

\section{The model and its spectral properties}\label{sec:rmt_theoretical_res}
Let us now become more quantitative, and define the problem that we are interested in studying. We will make our derivations in the most general model of coupled GOE matrices, introduced below in Sec.~\ref{subsec:matrix}. Since these coupled GOE matrices have the same structure, in Sec.\ref{main:isoalted_eigenvalue} and \ref{sec:rmt_outlier_evecs} we will derive their spectral properties in the most general case.

\subsection{The matrix ensembles}
\label{subsec:matrix}
\subsubsection{Definition}
We consider pairs of correlated random matrices with a perturbed GOE statistics.  In our model of interest, the perturbation is given by a special row and column in each matrix of the pair, whose entries are correlated to each others in a different way. More precisely, let $\mathbf{M}^{(a)}$ with $a\in\{0,1\}$ be a pair of $N \times N$ ($N>>1$) matrices with the following block structure:  
\begin{align}\label{eq:MatrixForm}
    \mathbf{M}^{(a)}=
    \begin{pmatrix}
     & && & m^{a}_{1\,N}\\
     &  &{\bf B}^{(a)}  & & \vdots\\
      & & & & m^{a}_{N-1\,N}\\
     m^{a}_{1\,N} && \ldots & m^{a}_{N-1\,N} &m_{N\,N}^{a}
    \end{pmatrix}
\end{align}
\noindent where the ${\bf B}^{(a)}$ are two $N-1 \times N-1$ correlated GOE matrices with components $B_{ij}^a$ having zero mean, and correlations given by:
\begin{equation}
\label{eq:b_correlations}
    \mathbb{E}[B_{ij}^a \, B_{kl}^b]= \tonde{\delta_{a b} \frac{\sigma^2}{N}+ (1-\delta_{ab})\frac{\sigma^2_H}{N}}(\delta_{ik} \delta_{jl}+ \delta_{il} \delta_{jk})
\end{equation}
for $a, b \in \grafe{0,1}$.  
The two GOE matrices ${\bf B}^{(a)} $ have equal variance $ N^{-1}\sigma^2$, and for all $i \leq j$ the component $B_{ij}^0$ is correlated only with $B_{ij}^1$. Similarly, the entries $m^a_{i N}$ for $i<N$ have zero mean and correlations given by:
\begin{equation}
\label{eq:m_correlations}
    \mathbb{E}[m_{iN}^a \, m_{kN}^b]= \tonde{\delta_{a b} \frac{\Delta^2_a}{N}+ (1-\delta_{ab})\frac{\Delta^2_h}{N}}\delta_{ik}
\end{equation}
for $a,b\in\{0,1\}$. Finally, the diagonal entries $m^a_{NN}$ have a non-zero average:
\begin{equation}
\mathbb{E}[m_{NN}^a]= \mu_a,\,\,\quad a\in\{0,1\},
\end{equation}
 and covariances given by:
\begin{equation}
\label{eq:m_NN}
   \text{Cov}(m_{NN}^a, m_{NN}^b)=\mathbb{E}[m_{NN}^a \, m_{NN}^b]-\mu_a \mu_b= \tonde{\delta_{a b} \frac{v^2_a}{N}+ (1-\delta_{ab})\frac{v^2_h}{N}}
\end{equation}
for $a,b\in\{0,1\}$.
The choice of correlations in \eqref{eq:b_correlations} implies that the matrices ${\bf B}^{(0)}, {\bf B}^{(1)}$ can be written as the sum of two GOE matrices:
\begin{equation}
    {\bf B}^{(a)}={\bf H}+ {\bf W}^{(a)},\quad a\in\{0,1\}
\end{equation}
where  ${\bf H}$ is an $N-1\times N-1$ GOE matrix with 
\begin{equation}
\begin{split}
    &\mathbb{E}[H_{ij} H_{kl}]=\frac{\sigma_H^2}{N} (\delta_{ik} \delta_{jl}+ \delta_{il} \delta_{jk}),
\end{split}
\end{equation}
that is in common to both elements of the pair, while  ${\bf W}^{(0)}, {\bf W}^{(1)}$ are $N-1\times N-1$ independent and identically distributed GOE matrices satisfying
\begin{equation}
\begin{split}
\mathbb{E}[W^{a}_{ij} W^{a}_{kl}]=\frac{\sigma_W^2}{N} (\delta_{ik} \delta_{jl}+ \delta_{il} \delta_{jk}),\quad a\in\{0,1\}
\end{split}
\end{equation}
and clearly $\sigma^2= \sigma^2_H + \sigma^2_W$. Thanks to \eqref{eq:m_correlations} and \eqref{eq:m_NN}, the entries belonging to the last row and column admit a similar decomposition in terms of independent random variables, 
\begin{equation}
m^{a}_{iN}= h_{iN} + w_{iN}^a,\quad a\in\{0,1\}
\end{equation}
with $h_{iN}\sim\mathcal{N}(0,N^{-1}\Delta^2_h)$ and $w^a_{iN}\sim\mathcal{N}(0,N^{-1}\Delta^2_{w,a})$ for $i<N$, while  $h_{NN}\sim\mathcal{N}(0,N^{-1} v^2_h)$ and $w^a_{NN}\sim\mathcal{N}(\mu_a,N^{-1} v^2_{w,a})$. Of course, $\Delta^2_a= \Delta^2_h + \Delta^2_{w,a}$ and $v^2_a= v^2_h+ v^2_{w,a}$, for $a=0$ and $a=1$.

\subsubsection{Recasting the matrices}
Each matrix of the form \eqref{eq:MatrixForm} can be re-written as a GOE matrix perturbed with both additive and multiplicative rank-one perturbations along one fixed direction identified by the normalized basis vector ${\bf e}_N$ (corresponding to the last row and column). We can indeed write:
\begin{equation}\label{eq:MatRiscritte}
    {\bf M}^{(a)}= \left[{\bf F}^{(a)}\right]^\top  {\bf X}^{(a)} \,  {\bf F}^{(a)} +  \tonde{\mu_a + \zeta_a\frac{\xi^a}{\sqrt{N}} } \, {\bf e}_N {\bf e}_N^T
\end{equation}
where ${\bf X}^{(a)}$ is now an $N \times N$ GOE with variance $N^{-1}\sigma^2$, $ \xi^a \sim \mathcal{N}(0,1)$ is an independent standard Gaussian variable, and the terms ${\bf F}^{(a)}$ and $\zeta_a$ are introduced to reproduce the correct variance of the entries belonging to the special row and column of the matrices \eqref{eq:MatrixForm}; more precisely, 
\begin{equation}
    {\bf F}^{(a)}= \mathbb{I}_N - \tonde{1-\frac{\Delta_a}{\sigma}}{\bf e}_N {\bf e}_N^T,
\end{equation}
while $\zeta_a= \tonde{v^2_a - \Delta^4_a/\sigma^{2}}^{\frac{1}{2}}$ is chosen in such a way that $m^a_{NN} \sim \mathcal{N}(\mu_a, N^{-1} v^2_a)$ is recovered. The matrix ${\bf F}^{(a)}$ represents a deterministic, multiplicative perturbation to the GOE, while the second term in \eqref{eq:MatRiscritte} is the additive one. \\

\noindent We also remark that each of the two matrices ${\bf M}^{(a)}$  has a statistics that \emph{is not} rotational invariant, since there is a basis vector ${\bf e}_N$ that identifies a special direction along which the statistics of the entries is special. Nonetheless, rotational invariance is preserved in the subspace orthogonal to ${\bf e}_N$, given that the corresponding blocks ${\bf B}^{(a)}={\bf H}+{\bf W}^{(a)}$ have a statistics which is invariant with respect to changes of basis. \\

\noindent We introduce the notation $\mathbf{u}_1^{a},\ldots,\mathbf{u}_N^{a}$ for the normalized eigenvectors of the matrix $\mathbf{M}^{(a)}$, and $\lambda_{1}^{a}, \ldots,\lambda_N^{a}$ for the associated real eigenvalues.
In the rest of the Chapter, we give for granted that the index $a$ can be either 0 or 1, and every time it appears it is understood that that property holds for both $a=0$ and $a=1$.

\subsection{Spectral properties and outliers}
\label{main:isoalted_eigenvalue}
\noindent 
In this section, we will show how to obtain the spectral properties of matrices of the form in Eq.~\eqref{eq:MatrixForm}. Notice that the matrices $\mathbf{M}^{(0)}$ and $\mathbf{M}^{(1)}$ have the same structure; each one has a statistics fully described by the parameters $\sigma, \Delta_a, v_a$ and $\mu_a$ for $a=0,1$. Since the spectral properties discussed in this section involve only eigenvalues and eigenvectors of one single element of the pair of matrices, they are independent of the parameters $\Delta_h, \sigma_H$ and $v_h$ describing the correlations between the entries of the two matrices in the pair. We therefore drop the superscript $a$ and denote the single-matrix parameters simply with $\sigma, \Delta, v$ and $\mu$ in this section. \\

\noindent It has already been shown in Refs.~\cite{ros2019complexity, ros2020distribution} that the perturbation given by the special row and column of ${\bf M}$ can generate a transition in the eigenvalue density in the large-$N$ limit, occurring at a critical value of the parameters $\Delta, \mu, \sigma$. This transition separates a regime in which the eigenvalue density is independent of $\Delta, \mu$ and simply coincides with the eigenvalue density of the GOE matrix $ {\bf X}$ in \eqref{eq:MatRiscritte}, from a regime in which one or two isolated eigenvalues are present, see Fig. \ref{fig:double_evals}. These isolated eigenvalues, which will henceforth be referred to as $\lambda_{\rm iso}$, are detached from the bulk of eigenvalues forming a continuum density in the limit $N \to \infty$. 
As we have already mentioned in the Introduction, these types of spectral transitions belong to the BBP-like transition family \cite{baik2005phase}. For GOE matrices, the BBP transition has been widely investigated in the case of an additive finite-rank perturbation \cite{potters_bouchaud_2020, edwards1976eigenvalue, benaych2011eigenvalues}, corresponding in our setting to $\Delta=\sigma$. We now discuss in full generality the results that hold in the case $\Delta \neq \sigma$ and $\mu\neq 0$, up to $1/N$ corrections, extending the previous works \cite{ros2019complexity, ros2020distribution}, where a special (physically relevant) case of this problem was considered.

\begin{figure}[t!]
\centering
\includegraphics[width=0.55\textwidth]{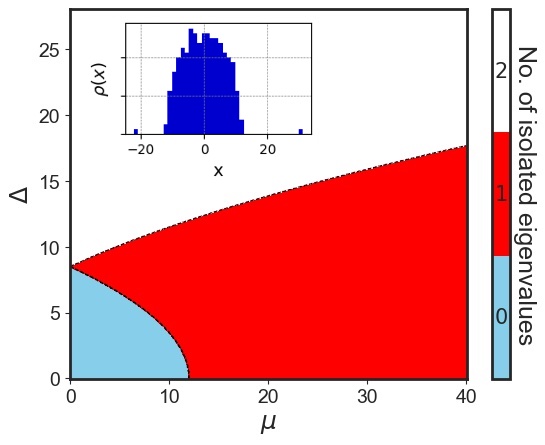}
\caption{The figure represents the regions in the plane $(\Delta,\mu)$ where either 2 (white zone), 1 (red zone) or 0 (lower left light blue zone) isolated eigenvalues emerge out of the bulk of the spectral density of $\mathbf{M}$. The plot is given for $\sigma=6$. As discussed in the main text, the existence conditions and the typical value of the isolated eigenvalue(s) are independent on $v$. \emph{Inset. } A particular realization of the spectral density for a random matrix $\mathbf{M}$ of size $N=300$, with $\sigma=6, \Delta=25$, $\mu=10$ belonging to the "white" zone, thus presenting two outliers.}
\label{fig:double_evals}
\end{figure}

\subsubsection{Computation of the Stieltjes transform}
\noindent The goal of this section is to show the main steps in deriving the Stieltjes transform of ${\bf M}$ up to order $1/N$. The techniques presented are rather general and will be used throughout this chapter. Therefore, here we will give some of the details of these calculations, without the goal of being rigorous, nor exhaustive. We assume here the definitions of the Stieltjes transform, resolvent, and the Sokhotski–Plemelj identity discussed in Sec.~\ref{sec:Goe_intro}, particularly in the case of GOE matrices. The Stieltjes transform of ${\bf M}$, here denoted by $\mathfrak{g}(z)$, can be written as:
\begin{align}
\label{app:sokholski}
    \mathfrak{g}(z)=\frac{1}{N}\mathbb{E}\left[\Tr(z-\mathbf{M})^{-1}\right]=\frac{1}{N}\mathbb{E}\left[\sum_{i=1}^{N-1}(z-\mathbf{M})^{-1}_{ii}+(z-\mathbf{M})^{-1}_{NN}\right].
\end{align}
As in Sec.~\ref{sec:Goe_intro}, we denote by $\mathfrak{g}_\sigma(z)$ the Stieltjes transform for the GOE ensemble with parameter $\sigma$. Given the block structure of the matrix ${\bf M}$, the components of the resolvent matrix can be obtained making use of Schur's matrix inversion lemma, which states:
\begin{align}
\label{eq:Schur_lemma}
\begin{bmatrix}
    \mathbf{A} & \mathbf{B}\\
    \mathbf{C} & \mathbf{D}
    \end{bmatrix}^{-1}=
    \begin{bmatrix}
    (\mathbf{A}-\mathbf{B}\mathbf{D}^{-1}\mathbf{C}) ^{-1} & -(\mathbf{A}-\mathbf{B}\mathbf{D}^{-1}\mathbf{C})^{-1}\mathbf{B}\mathbf{D}^{-1}\\
    - \mathbf{D}^{-1}\mathbf{C}(\mathbf{A}-\mathbf{B}\mathbf{D}^{-1}\mathbf{C})^{-1}\quad\quad & \mathbf{D}^{-1} + \mathbf{D}^{-1}\mathbf{C}(\mathbf{A}-\mathbf{B}\mathbf{D}^{-1}\mathbf{C})^{-1}\mathbf{B}\mathbf{D}^{-1}
    \end{bmatrix}.
\end{align} 

\noindent In the following we will set $M=N-1$, and introduce the $M \times M$ matrices
\begin{equation}
\label{eq:DefA}
   {\bf A}(z)=\frac{{\bf m} \; [{\bf m}]^T }{z- m_{NN}}, \quad {\bf m}=(m_{1N},\, m_{2N},\, \cdots ,\, m_{M \,  N})^T.
\end{equation}
Then, by Eq.~\eqref{eq:Schur_lemma}, for $i,j\leq M$ one has
\begin{equation}\label{eq:Block1}
    (z  -\mathbf{M})^{-1}_{ij}=
    (z -{\bf H}-{\bf W}-{\bf A}(z))^{-1}_{ij}
\end{equation}
and 
\begin{equation}
    (z -\mathbf{M})^{-1}_{iN}=
       -\sum_{k=1}^{N-1} \frac{m_{kN}}{z-m_{NN}}   (z  -\mathbf{M})^{-1}_{ik},
\end{equation}
while the remaining component reads
\begin{equation}\label{eq:Block2}
\begin{split}
    &(z-\mathbf{M})^{-1}_{N\, N}=(z-m_{NN})^{-1}\left\{1+\sum_{k,l=1}^{N-1}\frac{m_{kN}\,m_{lN}}{z-m_{NN}}(z  -\mathbf{M})^{-1}_{kl}\right\}.
    \end{split}
\end{equation}
We will also we make use of the following Dyson expansion:
\begin{align}
\label{eq:dyson}
  \frac{1}{z-{\bf H}-{\bf W}-{\bf A}(z)}=\mathbf{G}(z) \sum_{u=0}^\infty [{\bf A}(z) {\bf G}(z)]^u
\end{align}
with the resolvent defined as 
\begin{equation}
\label{eq:resolvent}
    {\bf G}(z):=(z-{\bf H}-{\bf W})^{-1}.
\end{equation}

\noindent By using these formulas, we obtain:
\begin{align}
\begin{split}
\mathbb{E} \quadre{\sum_{i=1}^{N-1}\tonde{z-\mathbf{M}}^{-1}_{ii}}&=\mathbb{E}\Tr[ (z-{\bf H}-{\bf W}-{\bf A}(z))^{-1}]=\mathbb{E}\Tr\left[\mathbf{G}(z) \sum_{u=0}^\infty [{\bf A}(z) {\bf G}(z)]^u \right]\\
&=\mathbb{E}\left[\Tr {\bf G}(z)\right] + \sum_{u=1}^\infty\mathbb{E}\Tr {\bf G}(z)[{\bf A}(z){\bf G}(z)]^u.
\end{split}
\end{align}
We wish to evaluate this expression to the second leading order. Let us analyse a single element of the form $\mathbb{E}\Tr {\bf G}(z)[{\bf A}(z){\bf G}(z)]^u$. By Gaussian integration, it is not hard to show that for $N>>1$ (see Appendix of Ref.~\cite{paccoros}):
\begin{align}
    \mathbb{E}\left[\frac{1}{(z-m_{NN})^u}\right]=\frac{1}{(z-\mu)^u}+\mathcal{O}\left(\frac{1}{N}\right).
\end{align}
Then, by plugging the expression for ${\bf A}(z)$ and using that $m_{i<N},m_N$ are independent (we remove the subscript $N$ in $m_{iN}$ for simplicity) we obtain:
\begin{align}
\begin{split}
&\mathbb{E}\Tr {\bf G}(z)[{\bf A}(z){\bf G}(z)]^u=\frac{1}{(z-\mu)^u}\mathbb{E}\Tr {\bf G}(z)[{\bf m}\,{\bf m}^\top {\bf G}(z)]^u\\
&=\frac{1}{(z-\mu)^u}\sum_{i_1,\ldots,i_{2u+1}}\mathbb{E}\left[\,{\bf G}(z)_{i_1i_2}\,m_{i_2}\,m_{i_3}\,{\bf G}(z)_{i_3i_4}\ldots m_{i_{2u}}\,m_{i_{2u+1}}\, {\bf G}(z)_{i_{2u+1}i_1}\right]\\
&=\frac{1}{(z-\mu)^u}\sum_{i_1,\ldots,i_{2u+1}}\mathbb{E}\left[ {\bf G}(z)_{i_1i_2}{\bf G}(z)_{i_3i_4}\ldots {\bf G}(z)_{i_{2u+1}i_1}\right]\mathbb{E}\left[m_{i_2}m_{i_3}\ldots m_{i_{2u+1}}\right]
\end{split}
\end{align}
where we used independence of the elements in ${\bf G}(z)$ and ${\bf A}(z)$. Now, to leading order, we must maximize the number of traces involving ${\bf G}$ in the sum above, since each ${\bf G}^k(z)$ converges to the identity multiplied by an $\mathcal{O}(1)$ factor as $N\to\infty$, see Sec.~\ref{sec:Goe_intro}. Hence, each trace will bring an additional $N$ factor. By Wick's theorem, the expected value $\mathbb{E}\left[m_{i_2}m_{i_3}\ldots m_{i_{2u+1}}\right]$ is given by the sum, over all possible pairings, of the product of expected values of each pair. The pairing that maximizes the number of traces is the following:
\begin{align}
    \mathbb{E}[m_{i_2}m_{i_{2u+1}}]\mathbb{E}[m_{i_3}m_{i_{4}}]\mathbb{E}[m_{i_5}m_{i_{6}}]\ldots \mathbb{E}[m_{i_{2u-1}}m_{i_{2u}}]=\frac{\Delta^{2u}}{N^u}.
\end{align}
Indeed, we isolate as many ${\bf G}$'s as possible (i.e. contract their own indices), except the first and the last ones, which must be necessarily paired. We then obtain:
\begin{align}
\begin{split}
    \mathbb{E}\Tr {\bf G}(z)&[{\bf A}(z){\bf G}(z)]^u=\frac{1}{N^u}\frac{\Delta^{2u}}{(z-\mu)^u}\mathbb{E}\left[\Tr{\bf G}(z)^2 [\Tr{\bf G}(z)]^{u-1} \right]+\mathcal{O}\left(\frac{1}{N}\right)\\
    &=\frac{\Delta^{2u}}{(z-\mu)^u}\left(\frac{1}{N}\Tr\mathbb{E}[{\bf G}(z)^2]\right)\left(\frac{1}{N}\Tr\mathbb{E}{\bf G}(z)\right)^{u-1}+\mathcal{O}\left(\frac{1}{N}\right).
\end{split}
\end{align}
Where in the last line we used that, since traces of GOE resolvents are self-averaging for large $N$, we can split the expected value to leading order. Combining everything for all powers of $u$, we finally obtain:
\begin{align}
\label{eq:WeC}
\mathbb{E} \quadre{\sum_{i=1}^{N-1}\tonde{z-\mathbf{M}}^{-1}_{ii}}=\mathbb{E}\Tr{\bf G}(z)+\frac{1}{N}\Tr\mathbb{E}[{\bf G}(z)^2]f(z;\Delta,\mu)+\mathcal{O}\tonde{\frac{1}{N}},
\end{align}
where we defined
\begin{align}
\label{def:f}
    f(z;\Delta,\mu):=\frac{\Delta^2}{z-\mu-\Delta^2\frac{1}{N}\Tr\mathbb{E}{\bf G}(z)}.
\end{align}
Now, notice that from the definition of ${\bf G}$ in Eq.~\eqref{eq:resolvent} we get:
\begin{align}
\frac{1}{N}\Tr\mathbb{E}[{\bf G}(z)^2]=-\frac{1}{N}\partial_z\Tr\mathbb{E}{\bf G}(z)=-\partial_z\mathfrak{g}_\sigma(z)+\mathcal{O}\tonde{\frac{1}{N}}
\end{align}
where the corrections come from the fact that ${\bf G}$ has size $(N-1)\times (N-1)$ while the Gaussian entries of ${\bf H}+{\bf W}$ have variance $\sigma^2/N$. Therefore, to leading order, we get:
\begin{align}
\frac{1}{N}\Tr\mathbb{E}[{\bf G}(z)^2]f(z;\Delta,\mu)=-\frac{\Delta^2}{z-\mu-\Delta^2\mathfrak{g}_\sigma(z)}\partial_z\mathfrak{g}_\sigma(z)+\mathcal{O}\tonde{\frac{1}{N}}.
\end{align}

\noindent We now focus on $\frac{1}{N}\mathbb{E}\Tr{\bf G}(z)$, which is the first term in the right-hand side of \eqref{eq:WeC}, normalized by $N$. To leading order, this term is $\mathfrak{g}_\sigma(z)$. However, three types of $1/N$ corrections can be computed: one coming from the trace, which sums only $N-1$ (and not $N$) matrix elements; one coming from the variances being $\sigma^2/N$ in a matrix of size $(N-1)\times (N-1)$; and one coming from the $1/N$ corrections to the GOE resolvent, already determined in \cite{VERBAARSCHOT1984367}. In order to distinguish between these terms, we multiply the first set of corrections by a factor $u$ and eventually take $u \to 1$ at the end of the calculation. Adapting the derivation of \cite{VERBAARSCHOT1984367} to the perturbed case we obtain:
\begin{equation}\label{eq:SingleRes}
\frac{1}{N}\mathbb{E} \Tr {\bf G}(z)= \mathfrak{g}_\sigma(z)+ \frac{1}{N} \quadre{\frac{z- \sqrt{z^2-4 \sigma^2}}{2[z^2- 4 \sigma^2]} - \frac{u \sigma^2 \mathfrak{g}_\sigma^3(z) }{1- \sigma^2 \mathfrak{g}^2_\sigma(z)}   - u \mathfrak{g}_\sigma(z)}+\mathcal{O}\tonde{\frac{1}{N^2}}.
\end{equation}
In here, the first contribution to the $1/N$ corrections is the one determined in \cite{VERBAARSCHOT1984367}; the first term proportional to $u$ arises from the fact that we are considering matrices of size $N-1$ with variances normalized by a factor $N$, while the second term proportional to $u$ is due to the fact that we normalize by $N$ the sum over $N-1$ components. Proceeding as above we also find:
\begin{equation}
\frac{1}{N}\mathbb{E}\left[(z-\mathbf{M})^{-1}_{NN}\right]=\frac{1}{N}\frac{\Delta^2}{z-\mu-\Delta^2\mathfrak{g}_{\sigma}(z)}\frac{\mathfrak{g}_\sigma(z)}{z-\mu}+  \mathcal{O}\tonde{\frac{1}{N^2}}.
\end{equation}
Combining everything, we finally obtain the Stieltjes transform of ${\bf M}$ up to $1/N$ corrections for large $N$:
\begin{align}\label{eq:FullG}
\begin{split}
    \mathfrak{g}(z)=\mathfrak{g}_\sigma(z)&+\frac{1}{N} \quadre{\frac{z- \sqrt{z^2-4 \sigma^2}}{2[z^2- 4 \sigma^2]}- \frac{u}{\sqrt{z^2- 4 \sigma^2}}}\\&+\frac{1}{N}\frac{\Delta^2}{z-\mu-\Delta^2\mathfrak{g}_{\sigma}(z)}\left[\frac{\mathfrak{g}_\sigma(z)}{z-\mu}-\partial_z\mathfrak{g}_{\sigma}(z)\right]+ \mathcal{O} \tonde{\frac{1}{N^2}}.
\end{split}
\end{align}

\subsubsection{Computation of the spectrum}
The spectral measure can then be obtained using \eqref{eq:InvStj}. To leading order, one recovers the GOE density \eqref{eq:DoSgoe}. The first contribution to the $1/N$ correction, denoted with $\rho^{(1)}_\sigma(x)$, is obtained from:
\begin{equation}
\label{eq:rho1_express}
\rho^{(1)}_{\sigma}(x)=\frac{1}{\pi}\lim_{\eta \to 0^+}\text{Im } \left[\frac{z- \sqrt{z^2-4 \sigma^2}}{2[z^2- 4 \sigma^2]}- \frac{u}{\sqrt{z^2- 4 \sigma^2}}\right]_{z=x-i \eta}
\end{equation}
with $u=1$ (setting $u=0$ would amount to considering matrices of size $N\times N$). The term in brackets exhibits a branch cut in the region $z \in [-2 \sigma, 2 \sigma]$, and two poles at the boundaries of the interval.  Thus, this term gives rise to $1/N$ corrections to the continuous eigenvalue density plus two delta peaks at the boundaries.  The second set of $1/N$ corrections arises from the second $1/N$ term in \eqref{eq:FullG}, which exhibits poles at the solutions of the equation 
\begin{align}
    \label{egval_eqn}
    z-\mu-\Delta^2\mathfrak{g}_{\sigma}(z)=0.
\end{align}
The real solutions $\lambda_{\rm iso, \pm}$ of this equation, whenever they exist, are the isolated eigenvalues of the matrix. We discuss extensively the conditions for their existence below. \\

\noindent In an expansion in $N^{-1}$ the average spectral measure of the matrices ${\bf M}$ reads:
\begin{equation}\label{eq:Dos}
d\nu_N(\lambda)= \rho_N(\lambda) d\lambda + \frac{1}{N}\sum_{*=\pm} \alpha_*\delta(\lambda-\lambda_{\rm iso,*}) + \mathcal{O}\tonde{\frac{1}{N^2}},
\end{equation}
where $\rho_N(\lambda)$ is defined for $|\lambda| \leq 2 \sigma$ and it 
admits the expansion 
\begin{equation}\label{eq:ED}
\rho_N(\lambda)= \rho_\sigma(\lambda)+ \frac{1}{N} \rho^{(1)}_\sigma(\lambda) + \mathcal{O}\tonde{\frac{1}{N^2}},
\end{equation}
and $\alpha_\pm$ are to be determined later. From Eq.~\ref{eq:rho1_express} we can compute the  sub-leading correction:
\begin{equation}\label{eq:Rho1}
\rho^{(1)}_{\sigma}(\lambda)= \frac{\sqrt{4 \sigma^2-\lambda^2}}{2 \pi (\lambda^2-4 \sigma^2)}- \frac{\text{sign}(\lambda)}{\sqrt{4 \sigma^2-\lambda^2}}+ \sum_{x=\pm 1}\frac{1}{4} \, \delta(\lambda +2 x \sigma),
\end{equation}
The delta peaks in the measure \eqref{eq:Dos} correspond to the isolated eigenvalues $\lambda_{\rm iso, \pm}$.
Isolated eigenvalues are real solutions $z \to \lambda$ of \eqref{egval_eqn}. By using the expression for $\mathfrak{g}_\sigma(\lambda)$ we can easily rewrite Eq.~\eqref{egval_eqn} as:
\begin{align}
\label{app:sign_condition}
\lambda\left(1-\frac{\Delta^2}{2\sigma^2}\right)-\mu=-\text{sign}(\lambda)\frac{\Delta^2}{2\sigma^2}\sqrt{\lambda^2-4\sigma^2}
\end{align}
from which we can take the square on both sides, keeping in mind that the equation to be satisfied by our final solution is \eqref{app:sign_condition}, with the proper sign on the right-hand side. Taking the square we get:
\begin{align*}
\lambda^2\left(1-\frac{\Delta^2}{\sigma^2}\right)-2\mu\left(1-\frac{\Delta^2}{2\sigma^2}\right)\lambda+\mu^2+\frac{\Delta^4}{\sigma^2}=0
\end{align*}
which gives us the two solutions
\begin{align}
\label{app:solutions}
\lambda_{\rm iso,\pm}(\mu, \Delta, \sigma)= \frac{{2 \mu \sigma^2- \Delta^2 \mu\pm \text{sign}(\mu) \Delta^2 \sqrt{\mu^2-4 (\sigma^2- \Delta^2)}}}{2 (\sigma^2-\Delta^2)}.
\end{align}
From this expression we get that the condition 
\begin{align}
\label{app:determinant_condition}
\mu^2-4(\sigma^2-\Delta^2)\geq 0
\end{align}
must be satisfied for the solutions to exist on the real line. We now discuss several cases for the parameters.\\

\noindent \textbf{Case $\Delta= \sigma$.}\\
As a consistency check, let us verify that we get back the result of \cite{edwards1976eigenvalue} when $\Delta= \sigma$. This case corresponds to a GOE matrix perturbed by a rank-1 additive term. The solution $\lambda_{\rm iso, +}$ diverges in this limit, while the solution $\lambda_{\rm iso, -}$ converges to
\begin{align}
\label{eq:Delta=Sigma_eval}
    \lambda_{\rm iso}= \lambda_{\rm iso,-}=\mu+\frac{\sigma^2}{\mu}.
\end{align}
 The condition \eqref{app:determinant_condition} is automatically satisfied, whereas \eqref{app:sign_condition} is satisfied provided that $|\mu|\geq \sigma$. The (only) isolated eigenvalue thus exists (i.e. it is bigger than $2\sigma$ in absolute value) for any $|\mu|>\sigma$.\\

\noindent \textbf{Case $\Delta< \sigma$}
This case has been already discussed in \cite{ros2019complexity,ros2020distribution}, and here we re-derive those results.
The isolated eigenvalues exist whenever at least one among $\lambda_{\rm iso, \pm}$ is bigger than $2\sigma$ in absolute value and the conditions \eqref{app:sign_condition}, \eqref{app:determinant_condition} are both satisfied. We notice that in this setting, if they exist, the eigenvalues satisfy $\text{sign}(\lambda_{\rm iso, \pm})=\text{sign}(\mu)$. In order to study these existence conditions, we plug the expressions for \eqref{app:solutions} inside equation \eqref{app:sign_condition} and find:
\begin{align*}
    \text{sign}(\mu)\text{sign}\left[\left(2\sigma^2-\Delta^2\pm\Delta^2\sqrt{1-\frac{4(\sigma^2-\Delta^2)}{\mu^2}}\right)\frac{2\sigma^2-\Delta^2}{4\sigma^2(\sigma^2-\Delta^2)}-1
    \right]=-\text{sign}(\mu)
\end{align*}
where we are assuming that the eigenvalues are indeed isolated (to be verified a posteriori). This expression is equivalent to 
\begin{align*}
&\left(2\sigma^2-\Delta^2\pm\Delta^2\sqrt{1-\frac{4(\sigma^2-\Delta^2)}{\mu^2}}\right)\frac{2\sigma^2-\Delta^2}{4\sigma^2(\sigma^2-\Delta^2)}-1\leq 0 \quad\\
&\Leftrightarrow \quad\Delta^2\pm(2\sigma^2-\Delta^2)\sqrt{1-\frac{4(\sigma^2-\Delta^2)}{\mu^2}}\leq 0
\end{align*}
from which it is clear that the only acceptable isolated eigenvalue in this setting is $\lambda_{\rm iso, -}$, since otherwise we would have that the sum of two positive quantities is smaller or equal than 0. By studying the second degree equation $\lambda_{\rm iso, -}(\mu, \Delta,\sigma )>2\sigma$ we find that it is verified provided that 
\begin{equation}\label{eq:ConditionEx}
|\mu|>  \sigma \tonde{1+ \frac{\sigma^2-\Delta^2}{\sigma^2}}.
\end{equation}
Under this condition it is straightforward to see that \eqref{app:determinant_condition} is automatically verified.  Hence in this setting there exists only one isolated eigenvalue, whose explicit expression is precisely $\lambda_{\text{iso},-}$ in \eqref{app:solutions}. This eigenvalue appears as soon as Eq.~\eqref{eq:ConditionEx} is satisfied. For $\mu>0$, this eigenvalue is the maximal eigenvalue of the random matrix, while for $\mu<0$ it is the minimal.\\

\noindent \textbf{Case  $\sigma<\Delta$.}
This case is richer, and to the best of our knowledge was not discussed in previous literature. We notice that in this case whenever they exist, then the isolated eigenvalues satisfy $\text{sign}(\lambda_{\rm iso, \pm})=\mp\text{sign}(\mu)$. Moreover the condition \eqref{app:determinant_condition} is always verified in this setting. By plugging \eqref{app:solutions} into \eqref{app:sign_condition} we obtain 
\begin{align*}
&\text{sign}(\mu)\text{sign}\left(\frac{\pm\Delta^2\sqrt{1+\frac{4(\Delta^2-\sigma^2)}{\mu^2}}-\Delta^2+2\sigma^2}{2(\Delta^2-\sigma^2)}\frac{2\sigma^2-\Delta^2}{2\sigma^2}+1\right)\\
&=-\text{sign}(\mu)\text{sign}\left(\frac{\pm\Delta^2\sqrt{1+\frac{4(\Delta^2-\sigma^2)}{\mu^2}}-\Delta^2+2\sigma^2}{2(\Delta^2-\sigma^2)}\right)
\end{align*}
which gives us the condition
\begin{align*}
    \pm\left(\frac{\pm\Delta^2\sqrt{1+\frac{4(\Delta^2-\sigma^2)}{\mu^2}}-\Delta^2+2\sigma^2}{2(\Delta^2-\sigma^2)}\frac{2\sigma^2-\Delta^2}.{2\sigma^2}+1\right)\leq 0
\end{align*}
In the case in which we choose the sign $-$, this inequality becomes 
\begin{align*}
    \Delta^2\geq (2\sigma^2-\Delta^2)\sqrt{1+\frac{4(\Delta^2-\sigma^2)}{\mu^2}}
\end{align*}
from which we deduce that it is always verified when $\Delta\geq\sqrt{2}\sigma$ and it is verified only for $|\mu|\geq 2\sigma-\Delta^2/\sigma$ when $\Delta<\sqrt{2}\sigma$. The combination of these conditions leads to $|\mu|\geq 2\sigma-\Delta^2/\sigma$ (in the first case $2\sigma-\Delta^2/\sigma$ becomes negative and therefore any $\mu$ will satisfy the condition).
In the case in which we choose the sign $+$, the inequality becomes 
\begin{align*}
    0\geq \Delta^2+ (2\sigma^2-\Delta^2)\sqrt{1+\frac{4(\Delta^2-\sigma^2)}{\mu^2}}
\end{align*}
which is never true for $\Delta\leq\sqrt{2}\sigma$ and becomes true for $\Delta>\sqrt{2}\sigma$ as long as $|\mu|\leq \Delta^2/\sigma-2\sigma$. It remains to verify that both of these isolated eigenvalues are bigger than $2\sigma$ in their domain of existence. By plugging the expression for $\lambda_{\rm iso, \pm}$ as in \eqref{app:solutions}, it is straightforward to verify that both of these isolated eigenvalues are outside of the bulk if we take strict inequalities in the existence conditions that we have just found. \\
\noindent Henceforth, we can resume our results as follows: $\lambda_{\rm iso, -}$ is an isolated eigenvalue provided that $|\mu|>2\sigma-\Delta^2/\sigma$; $\lambda_{\rm iso, +}$ is also an isolated eigenvalue provided that $\Delta>\sqrt{2}\sigma$ and $|\mu|<\Delta^2/\sigma-2\sigma$. In particular, notice that whenever $\lambda_{\text{iso},+}$ exists, then also $\lambda_{\text{iso},-}$ exists. \\\\
\noindent Therefore, we can incorporate such conditions within the spectral measure by defining that in Eq.~\eqref{eq:Dos} it holds:
\begin{equation}
\begin{split}
&\alpha_-\propto H\tonde{|\mu|-2\sigma+\frac{\Delta^2}{\sigma}}\\
&\alpha_+\propto H(\Delta-\sqrt{2}\sigma) H\tonde{-|\mu|-2\sigma+\frac{\Delta^2}{\sigma}}.
\end{split}
\end{equation}
with $H$ the Heaviside step function. We refer the reader to Fig.~\ref{fig:double_evals} for an example of how the $(\mu,\Delta)$ phase diagram is partitioned in terms of number isolated eigenvalues present. The proportionality factors are not important for the present discussion, but they can be found in Ref.\cite{paccoros}. However, let us just remark that the following holds:
\begin{equation}
\label{eq:rmt_delta_q}
\lim_{z \to \lambda_{\rm iso,\pm}}\frac{1}{\pi} \text{Im} \frac{1}{z-\mu-\Delta^2\mathfrak{g}_{\sigma}(z)}= \delta(z- \lambda_{\rm iso, \pm}) \mathfrak{q}_{\sigma, \Delta}(z, \mu).
\end{equation}
with 
\begin{equation}\label{eq:Defq}
\begin{split}
   & \mathfrak{q}_{\sigma,\Delta}(\lambda,\mu):=\text{sign}(\mu) \frac{\text{sign}(\lambda)\Delta^2\sqrt{\lambda^2-4\sigma^2}-\lambda(2\sigma^2-\Delta^2)+2\mu\sigma^2}{2\Delta^2\sqrt{\mu^2-4(\sigma^2-\Delta^2)}}.
    \end{split}
\end{equation}
This computation can be done by finding the Imaginary part using $z=\lambda_{\text{iso},\pm}-i\eta$ with $\eta\to 0^+$, and using the results in Sec.~\ref{sec:Goe_intro}.

\subsection{Outlier eigenvectors}
\label{sec:rmt_outlier_evecs}
The eigenvector ${\bf u}_{\rm iso,-}$ associated to the isolated eigenvalue $\lambda^0_{\text{iso},-}$~\eqref{app:solutions} has a 
 projection on the basis vector ${\bf e}_N$ corresponding to the special line and column of the matrix, which remains of $O(1)$ when $N$ is large; the typical value of this projection has been computed in~\cite{ros2020distribution} and reads:
\begin{equation}\label{eq:ProjVectors}
\begin{split}
  &\mathbb{E}[({\bf u}_{\rm iso,-} \cdot  {\bf e}_N)^2]= \mathfrak{q}_{\sigma,\Delta}(\lambda_{\rm iso,-},\mu)+\mathcal{O}(1/N)
   \end{split}
\end{equation}
It can be shown rather easily that whenever Eq.\eqref{eq:ConditionEx} is satisfied, then \eqref{eq:ProjVectors} is positive, as it should be. This can be done by considering separately the cases $\mu>0$ and $\mu<0$ and by using the expression of $\lambda_{\rm iso,-}$. In particular, since $\lambda_{\text{iso},-}$ and $\mu$ have the same sign, one can show that the condition \eqref{eq:ConditionEx} is equivalent to 
\begin{align*}
    -|\lambda_{\text{iso},-}|(2\sigma^2-\Delta^2)+2|\mu|\sigma^2 > 0
\end{align*}
which immediately implies the positivity of \eqref{eq:ProjVectors}. In particular, Eq.\eqref{eq:ProjVectors} is zero if and only if $|\lambda_{\text{iso}, -}|=2\sigma$, which is equivalent to $|\mu|=\sigma^{-1}(2\sigma^2-\Delta^2)$, i.e. the isolated eigenvalue is at the edge of the bulk. \\
\noindent The explicit dependence of \eqref{eq:ProjVectors} on the parameters $\sigma, \Delta,\mu$ reads:
 \begin{equation}
     \mathbb{E}[({\bf u}_{\rm iso, -} \cdot  {\bf e}_N)^2]= \frac{\quadre{\Delta^2 (|\mu|+ \Omega)-2 \sigma^2 \Omega+ \text{sign}(\Delta^2-\sigma^2)\sqrt{\kappa}}}{4 \Omega (\Delta^2-\sigma^2)}
 \end{equation}
 where $\Omega=(\mu^2- 4 \sigma^2+ 4 \Delta^2)^{\frac{1}{2}}$
and $\kappa=(\Delta^2 |\mu|-2 \sigma^2 |\mu|+ \Delta^2 \Omega)^2-16 \sigma^2(\sigma^2-\Delta^2)^2$.  In the case of a purely additive perturbation ($\Delta=\sigma$), using \eqref{eq:Delta=Sigma_eval} one sees that this expression reduces to:
\begin{equation}\label{eq:OvPCA}
({\bf u}_{\rm iso,-} \cdot  {\bf e}_N)^2 \stackrel{\Delta \to \sigma}{\longrightarrow}  1-\frac{\sigma^2}{\mu^2}
\end{equation}
consistently with previous results~\cite{benaych2011eigenvalues, potters_bouchaud_2020}.  
We remark that for the matrices \eqref{eq:MatrixForm} the joint isolated eigenvalue-eigenvector large deviation function has been determined as well~\cite{ros2020distribution}, generalizing the case of a purely additive perturbation \cite{BiroliGuionnet}.

\section{Eigenvectors overlaps}\label{sec:theoretical_res2}
\noindent In this section we aim at characterizing the correlations between eigenvectors of pairs of correlated matrices with the distribution  \eqref{eq:MatrixForm}, similarly to what is discussed in \cite{bun2018overlaps} for unperturbed GOE matrices. In particular, our objects of interest are the {averaged squared overlaps} between eigenvectors associated to different eigenvalues of the two matrices:
\begin{equation}
\label{phi}
\Phi(\lambda^{0},\lambda^{1}):=N\mathbb{E}[\left( \mathbf{u}_{\lambda^{0}}\cdot \mathbf{u}_{\lambda^{1}}\right)^2],
\end{equation}
where $\lambda^a$ are eigenvalues of ${\bf M}^{(a)}$, ${\bf u}_{\lambda^a}$ the corresponding normalized eigenvectors, and the expectation value $\mathbb{E}$ represents the average over the distribution of all the entries of the two matrices. In the limit of large $N$  this quantity remains of $\mathcal{O}(1)$ for values of $\lambda^a$ belonging to the continuous part (i.e. the \emph{bulk}) of the eigenvalue density of the two matrices. Indeed, in that case the eigenvectors are random vectors on the high-dimensional sphere and it can be shown that their typical squared overlap is $1/N$ (see paragraph B2 in Ref.~\cite{ros_lecture_2025}).\\

\noindent In this Chapter, we are interested in computing both the overlap between eigenvectors associated to eigenvalues in the bulk, as well as the overlaps involving the eigenvectors associated to the isolated eigenvalues of the matrices, whenever they exist. In the first case, the average $\mathbb{E}$ over different realizations of the random matrices can be replaced by an average, for fixed
randomness, over eigenvectors associated to eigenvalues within windows of width $d\lambda \gg N^{-1}$ centered around $\lambda^{0}, \lambda^{1}$: as a matter of fact, the quantity \eqref{phi} is self-averaging in the large $N$ limit \cite{bun2018overlaps}. 

\noindent Consider now the overlaps involving the eigenvectors associated to the isolated eigenvalues. As we have discussed in the previous section, any element of the pair of matrices in \eqref{eq:MatrixForm} can present zero, one or two isolated eigenvalues (see Fig.~\ref{fig:double_evals}). Such eigenvalues pop out of the bulk of the spectral density, which for $N\to\infty$ is given by the Wigner's semicircle law. Two isolated eigenvalues, denoted by $\lambda^a_{\text{iso},\pm}$, exist for each matrix $a\in\{0,1\}$ only when the noise from the special row and column is considerably bigger than the variance of the main GOE blocks, i.e. $\Delta_a >\sqrt{2} \sigma$. 
In the following, we restrict to the case in which only one isolated eigenvalue exists, equal to $\lambda^a_{\text{iso},-}$. We recall that this happens whenever Eq.~\eqref{eq:ConditionEx} is satisfied. To simplify the notation, henceforth we set
\begin{equation}
\label{eval_iso}
\lambda_{\rm iso}^a := \lambda_{\rm iso, -}^a,
\end{equation}
meaning that 
 $\lambda^0_{\text{iso}}$ is the isolated eigenvalue of $\mathbf{M}^{(0)}$, and $\lambda^1_{\text{iso}}$ of $\mathbf{M}^{(1)}$ . All of the results presented in the following can be easily generalized to the other isolated eigenvalues of the random matrices, whenever they exist. \\
 
 \noindent We also remark that in the case in which both the eigenvalues in \eqref{phi} are isolated, the relevant quantity to compute is the rescaled function:
\begin{align}
\label{phi_rescaled}
\tilde \Phi(\lambda_{\rm iso}^0,\lambda_{\rm iso}^1)   :=\frac{\Phi(\lambda_{\rm iso}^0,\lambda_{\rm iso}^1)}{N}=\mathbb{E}\left[\left( \mathbf{u}_{\lambda_{\text{iso}}^0}\cdot \mathbf{u}_{\lambda_{\text{iso}}^1} \right)^2\right].
\end{align}
This is because both eigenvectors have an $\mathcal{O}(1)$ projection on the special direction ${\bf e}_N$ given by \eqref{eq:ProjVectors}, so that their overlap is also of order 1, as can be shown by expanding these vectors in a component parallel to ${\bf e}_N$ and one orthogonal to it. This clearly indicates that the quantity remaining of $\mathcal{O}(1)$ in the limit of large $N$ is the
rescaled quantity \eqref{phi_rescaled}.


\noindent The overlap \eqref{phi} takes a different form depending on whether the considered eigenvalues (either both, or one of them) belong to the bulk of the eigenvalues density of their respective matrices, or are isolated. Our contribution consists in explicit formulas for the eigenvector overlaps in all the different cases, as a function of the parameters defining the statistics of the matrices \ref{eq:MatrixForm} . For simplicity, we will refer to these cases as \textit{bulk-bulk}, \textit{bulk-iso} and \textit{iso-iso}. Moreover, we will extend the case of bulk-bulk overlaps in \cite{bun2018overlaps} by computing its $1/N$ finite size corrections.\\

\noindent In the following sections, we will dive into the main steps of the calculations of such overlaps, before presenting the final results.

\subsection{A formula to extract the overlaps}
In this section, we aim at giving an overview of how the computation of the overlaps \eqref{phi} is carried out in the three cases mentioned above. The derivation is similar to that discussed in Ref.~\cite{bun2018overlaps}. We begin by introducing the auxiliary function
\begin{align}\label{eq:ExpProRe}
\psi(z,\xi):&=\frac{1}{N}\mathbb{E}\left[\Tr\left[\left(z-\mathbf{M}^{(0)}\right)^{-1}\left(\xi-\mathbf{M}^{(1)}\right)^{-1}\right]\right].
\end{align}
For finite $N$, the spectral decomposition of the matrices yields:
\begin{align}\label{eq:DoubleSum1}
    \psi(x-i\eta,y\pm i\eta)&=\frac{1}{N}\sum_{\alpha,\beta}\mathbb{E}\bigg[\frac{1}{x-i\eta-\lambda^{0}_{\alpha}}\frac{1}{y\pm i\eta-\lambda^{1}_{\beta}}\left( \mathbf{u}^{0}_{\alpha}\cdot\mathbf{u}^{1}_{\beta}\right)^2
    \bigg]\\
    &=\frac{1}{N^2}\sum_{\alpha,\beta}\mathbb{E}\left[
  R_{x,y,\eta}^{\pm}(\lambda^{0}_{\alpha},\lambda^{1}_{\beta}) \; 
    N\left(\mathbf{u}^{0}_{\alpha}\cdot\mathbf{u}^{1}_{\beta}\right)^2
    \right].
\end{align}
where we defined:
\begin{equation}
R_{x,y,\eta}^{\pm}(\lambda, \chi):=\frac{1}{x-\lambda-i\eta}\frac{1}{y-\chi\pm i\eta}.
\end{equation}

\noindent In the large $N$ limit, the sums over the eigenvalues can be turned into integrals over the spectral measure of the matrices, taking care of the presence of the sub-leading terms due to the isolated eigenvalues. 

\noindent The above expression hence becomes equivalent to:
{\medmuskip=0mu
\thinmuskip=0mu
\thickmuskip=0mu
\begin{align}
\begin{split}
\psi(x-i\eta,y\pm i\eta)=&  \int d\lambda\,d\chi\,\rho_\sigma(\lambda)\rho_\sigma(\chi)R^{\pm}_{x,y,\eta}(\lambda, \chi)  \Phi(\lambda,\chi)\\
    &+\frac{1}{N}\int  d\lambda\rho_\sigma(\lambda)R^{\pm}_{x,y,\eta}(\lambda, \lambda^{1}_{\rm iso}) \Phi(\lambda,\lambda^{1}_{\rm iso})\\
    &+\frac{1}{N}\int d\chi\rho_\sigma(\chi)R^{\pm}_{x,y,\eta}(\lambda^{0}_{\rm iso}, \chi) \Phi(\lambda^{0}_{\rm iso},\chi)\\
    &+\frac{1}{N}R^{\pm}_{x,y,\eta}(\lambda^{0}_{\rm iso}, \lambda^{1}_{\rm iso}) \frac{\Phi(\lambda^{0}_{\rm iso}, \lambda^{1}_{\rm iso})}{N}
\end{split}
\end{align}}

\noindent where $\rho_\sigma$ denotes the continuous part of the eigenvalue densities of the matrices ${\bf M}^{(a)}$, for $a\in\{0,1\}$, defined in \eqref{eq:DoSgoe}. We set  $\psi_0=\lim_{\eta\to0^+}\psi$. We recall Sokhotsky's formula:
\begin{align}
    \text{Im}\left[\lim_{\eta\to 0^+}\frac{1}{x\pm i\eta}\right]=\mp \pi\delta(x),
\end{align}
which applied to $R^\pm$ gives:
\begin{align}
\lim_{\eta\to 0^+}\text{Re}[R^+_{x,y,\eta}(\lambda,\chi)-R^-_{x,y,\eta}(\lambda,\chi)]=2\pi^2\delta(x-\lambda)\delta(y-\chi).
\end{align}
Applying this to $\psi_0$ gives the following:
\begin{align}
\label{re_psi}
    \begin{split}
    &\text{Re}\left[\psi_0(x-i\eta,y+i\eta)-\psi_0(x-i\eta,y-i\eta)\right]\\&=2\pi^2\Phi(x,y)\rho_\sigma(x)\rho_\sigma(y)\\
    &+\frac{2\pi^2}{N}\Phi(\lambda^{0}_{\rm iso},y)\rho_\sigma(y)\delta(x-\lambda^{0}_{\rm iso})\\
    &+\frac{2\pi^2}{N}\Phi(x,\lambda^{1}_{\rm iso})\rho_\sigma(x)\delta(y-\lambda^{1}_{\rm iso})\\
    &+\frac{2\pi^2}{N} \tilde{\Phi}(\lambda^{0}_{\rm iso},\lambda^{1}_{\rm iso})\delta(x-\lambda^{0}_{\rm iso})\delta(y-\lambda^{1}_{\rm iso}).
    \end{split}
\end{align}
We therefore see that in order to get the expression for $\Phi(\lambda^{0}_{\rm iso},y)$ we have to compute $\psi(z, \xi)$ and isolate the $1/N$ correction proportional to $\delta(x-\lambda^{0}_{\rm iso})$  appearing in the formula above. Instead, the term proportional to two delta peaks will give information on the overlap $\tilde \Phi(\lambda^{0}_{\rm iso},\lambda^{1}_{\rm iso})$. Notice that even though we are focusing on the case in which one single isolated eigenvalue $\lambda_{\rm iso} \equiv \lambda_{\rm iso, -}$ exists, all calculations can be extended straightforwardly to the second isolated eigenvalue, whenever it exists. \\
\noindent The above computations show that the expressions for the various overlaps can be obtained provided that the auxiliary function $\psi$ is computed up to order $1/N$. In the following sections we give an overview of such computation.

\subsection{Computation of $\psi(z,\xi)$}
It is convenient to decompose $\psi(z, \xi)$ as $\psi=\psi_{00}+\psi_{0N}+\psi_{NN}$ with:
\begin{align}
\label{psi_decomposition}
\begin{split}
    &\psi_{00}(z, \xi)=\mathbb{E}\left[\frac{1}{N}\sum_{i,j=1}^{N-1}(z-\mathbf{M}^{(0)})^{-1}_{ij}(\xi-\mathbf{M}^{(1)})^{-1}_{ij}\right]\\
    &\psi_{0N}(z, \xi)=\mathbb{E}\left[\frac{2}{N}\sum_{i=1}^{N-1}(z-\mathbf{M}^{(0)})^{-1}_{iN}(\xi-\mathbf{M}^{(1)})^{-1}_{iN}\right]\\
    &\psi_{NN}(z, \xi)=\mathbb{E}\left[\frac{1}{N}(z-\mathbf{M}^{(0)})^{-1}_{NN}(\xi-\mathbf{M}^{(1)})^{-1}_{NN}\right].
\end{split}
\end{align}
Here we will make use of the expansions from Eq.~\eqref{eq:DefA} to Eq.~\eqref{eq:resolvent}, where however now we keep track of the (sub/super)script $a$. So, for example, we have ${\bf G}_a(z)=(z-{\bf H}-{\bf W}^{(a)})^{-1}$. For simplicity, we first perform the average over the entries  $m^a_{i \, N}$ for $i<N$, with $a=0,1$, since they don't appear in the resolvents. \\\\

\noindent We will give a detailed derivation of $\psi_{00}$, while computing $\psi_{0N}$ and $\psi_{NN}$ is very similar, and therefore we won't do it here (see Ref.~\cite{paccoros} for additional details). By using \eqref{eq:Block1} and  \eqref{eq:dyson} we can easily rewrite
\begin{align}
\label{eq:Sum}
\begin{split}
&\psi_{00}(z,\xi)= \sum_{k,m=0}^{+\infty}\frac{1}{N}\Tr\mathbb{E}\left[{\bf S}_{k,m} \right],\\
& {\bf S}_{k,m}:={\bf G}_0(z)[{\bf A}^{(0)}(z){\bf G}_0(z)]^k{\bf G}_1(\xi)[{\bf A}^{(1)}(\xi){\bf G}_1(\xi)]^m \quad 
\end{split}
\end{align}
where we recall that the trace is over a subspace of dimension $N-1$. The definition for $\mathbf{A}^{(a)}$ is expressed in \eqref{eq:DefA}. We now compute the partial averages of the strings ${\bf S}_{k,m}$ over the entries $m^0_{iN}, m^1_{iN}$, to order $1/N$. Since the term with $k=0=m$ is independent of the entries $m_{iN}^a$, we focus on the remaining terms. \\

\noindent For either $k$ or $m$ different from 0, we need to evaluate
\begin{align*}
    \frac{1}{N}&\Tr\mathbb{E}\left[{\bf S}_{k,m}\right] =\frac{1}{N}\sum_{\substack{i_1,\ldots,i_{2k+1}=1\\j_1,\ldots j_{2m+1}=1}}^{N-1}\mathbb{E}\bigg[\mathbf{G}_0(z)_{i_1i_2}\frac{m^0_{i_2N}m^0_{i_3N}}{z-m^{0}_{NN}}\mathbf{G}_0(z)_{i_3i_4}\cdots\frac{m^0_{i_{2k}N}m^0_{i_{2k+1}N}}{z-m^{0}_{NN}}\times\\&\times\mathbf{G}_0(z)_{i_{2k+1}j_1} 
    \mathbf{G}_1(\xi)_{j_1j_2}\frac{m^1_{j_2N}m^1_{j_3N}}{\xi-m^{1}_{NN}}\mathbf{G}_1(\xi)_{j_3j_4}\cdots \frac{m^1_{j_{2m}N}m^1_{j_{2m+1}N}}{\xi-m^{1}_{NN}}\mathbf{G}_1(\xi)_{j_{2m+1}i_1}
    \bigg].
\end{align*}
We will first take the average over $m_{NN}^a$ for $a\in\{0,1\}$, since they do not appear in the resolvents. As for the computation of the spectrum, we can use here that (see Appendix in Ref.~\cite{paccoros} for the explicit proof):
\begin{align}
    \mathbb{E} \quadre{\frac{1}{(z-m^{0}_{NN})^{k}} \frac{1}{(\xi-m^{1}_{NN})^{m}}}=\frac{1}{(z-\mu_0)^{k}}\frac{1}{(\xi-\mu_1)^{m}}+\mathcal{O}\tonde{\frac{1}{N}}.
\end{align}

\noindent Therefore:
\begin{align}\label{eq:Nina}
\begin{split}
\frac{1}{N}\Tr \mathbb{E}&\left[{\bf S}_{k,m}\right]= \quadre{\frac{1}{(z-\mu_0)^k}\frac{1}{(\xi-\mu_1)^m} + \mathcal{O}\tonde{\frac{1}{N}}} \\&\times\sum_{\substack{i_1,\ldots,i_{2k+1}=1\\j_1,\ldots j_{2m+1}=1}}^{N-1}\mathbb{E}\bigg[\mathbf{G}_0(z)_{i_1i_2}\cdots \mathbf{G}_0(z)_{i_{2k+1}j_1}
    \mathbf{G}_1(\xi)_{j_1j_2}\cdots \mathbf{G}_1(\xi)_{j_{2m+1}i_1}
    \bigg]\\
    &\times\mathbb{E}
    [m^0_{i_2N}m^0_{i_3N}\cdots m^0_{i_{2k}N}m^0_{i_{2k+1}N}m^1_{j_2N}m^1_{j_3N}\cdots m^1_{j_{2m}N}m^1_{j_{2m+1}N}].
    \end{split}
\end{align}
The second average can be evaluated using Wick's theorem, paying attention on whether the contractions involve matrix elements with the same or with different $a=0,1$. Below, we determine the subset of contractions that contribute to leading order. We begin by discussing some special cases.

\vspace{.3 cm}
\noindent Let us focus first on the case $k=0$. By Wick's theorem, the average $\mathbb{E}
[m^1_{j_2N}\ldots m^1_{j_{2m+1}N}]$ appearing in \eqref{eq:Nina} will be contributed by all possible pairwise contractions of the variables $m^1_{jN}$, each one contributing with a factor of $\Delta^2_1/N$. To each contraction, there is a contraction of the indices in the term $ \mathbf{G}_1(\xi)_{j_1j_2}\mathbf{G}_1(\xi)_{j_3j_4}\cdots \mathbf{G}_1(\xi)_{j_{2m+1}i_1}$ also appearing in \eqref{eq:Nina}. 
We now argue that there is a unique Wick contraction that contributes to \eqref{eq:Nina} to leading order, which is the contraction corresponding to $\delta_{j_3 j_4}\delta_{j_5 j_6}\cdots \delta_{j_{2m+1} j_2}$. In fact, the products of resolvent operators converge in the large-$N$ limit to a deterministic matrix proportional to the identity. Therefore, each trace of such products is of order $N$. For this reason, to get the largest contribution from the term  $\mathbb{E}\bigg[\mathbf{G}_0(z)_{i_{1}j_1}
\mathbf{G}_1(\xi)_{j_1j_2}\mathbf{G}_1(\xi)_{j_3j_4}\cdots \mathbf{G}_1(\xi)_{j_{2m+1}i_1}
    \bigg]$ in \eqref{eq:Nina}, one has to select the contraction of indices that corresponds to maximizing the number of resulting traces, while recalling that some matrices have common indices and cannot therefore be decoupled into separate traces. For $k=0, m\geq 1$ we see that $\mathbf{G}_0(z)_{i_1j_1}\mathbf{G}_1(\xi)_{j_1j_2}\mathbf{G}_1(\xi)_{j_{2m+1}i_1}$ is the only block which cannot be decoupled. Let us define the normalized trace operator as $\TR:=\frac{1}{N}\Tr$ to optimize space. We will use either $\TR$ or $\frac{1}{N}\Tr$ throughout the rest of the chapter. The leading contribution to \eqref{eq:Nina} reads:
\begin{align*}
\begin{split}
\TR \mathbb{E}\left[{\bf S}_{0,m}(z, \xi)\right]&=    \frac{\Delta^{2m}_1}{N(\xi-\mu_1)^m}\left(\TR\mathbb{E}[\mathbf{G}_0(z)\mathbf{G}_1(\xi)^2]\right)\left(\TR\mathbb{E}\mathbf{G}_1(\xi)\right)^{m-1}+ \mathcal{O} \tonde{\frac{1}{N^2}},
\end{split}
\end{align*}
where we recall that $\Delta^2_a= \Delta^2_h + \Delta^2_{w,a}$. The case $k\geq 1, m=0$ is analogous, and we get the leading contribution:
\begin{align*}
 \TR  \mathbb{E}\left[{\bf S}_{k,0}(z, \xi)\right]=   \frac{1}{N}\frac{\Delta^{2k}_0}{(z-\mu_0)^k}\left(\TR \mathbb{E}[\mathbf{G}_0(z)^2\mathbf{G}_1(\xi)]\right)\left(\TR \mathbb{E}\mathbf{G}_0(z)\right)^{k-1}+ \mathcal{O} \tonde{\frac{1}{N^2}}.
\end{align*}

\vspace{.3 cm}

\noindent In the case $k\geq 1, m\geq 1$, the only coupled matrices are the two pairs $\mathbf{G}_0(z)_{i_{2k+1}j_1}\mathbf{G}_1(\xi)_{j_1j_2}$ and $\mathbf{G}_0(z)_{i_1i_2}\mathbf{G}_1(\xi)_{j_{2m+1}i_1}$. A reasoning analogous to the one above shows that the leading term in the $1/N$ expansion is given by:

\begin{align*}
\begin{split}
\TR  \mathbb{E}\left[{\bf S}_{k,m}\right]&=    \frac{\Delta_h^4}{N}\frac{\Delta^{2k-2}_0}{(z-\mu_0)^k}\frac{\Delta^{2m-2}_1}{(\xi-\mu_1)^m}\left(\TR \mathbb{E}[\mathbf{G}_0(z)\mathbf{G}_1(\xi)]\right)^2\\
&\times \left(\TR \mathbb{E}\mathbf{G}_0(z)\right)^{k-1}\left(\TR \mathbb{E}\mathbf{G}_1(\xi)\right)^{m-1} + \mathcal{O}\tonde{\frac{1}{N^2}}
\end{split}
\end{align*}
The dependence on  $\Delta_h$ appears due to the fact that the contractions corresponding to $\delta_{i_{2k+1} j_2}$ and $\delta_{i_2 j_{2m+1}}$ involve elements $m^a_{i N}$ corresponding to two different indices $a\in\{0,1\}$.\\
It is not hard to check that the re-summation of the $1/N$  contributions for arbitrary $k, m$ leads to the following expression:
\begin{equation}
\label{eq:psi00Implicit}
\begin{split}
\psi_{00}(z,\xi)&= \TR\mathbb{E}[\mathbf{G}_0(z)\mathbf{G}_1(\xi)]\\
&+\frac{1}{N}\TR\mathbb{E}[\mathbf{G}^2_0(z)\mathbf{G}_1(\xi)] f(z; \Delta_0, \mu_0)\\
  &+\frac{1}{N}\TR\mathbb{E}[\mathbf{G}_0(z)\mathbf{G}_1^2(\xi)]   f(\xi; \Delta_1, \mu_1)\\
 &+ \frac{1}{N}\tonde{\frac{\Delta_h^4}{\Delta^2_0 \,  \Delta^2_1}} \left(\TR\mathbb{E}[\mathbf{G}_0(z)\mathbf{G}_1(\xi)]\right)^2   f(\xi; \Delta_1, \mu_1) f(z; \Delta_0, \mu_0)\\
 &+ \mathcal{O}\tonde{\frac{1}{N^2}}.
     \end{split}
\end{equation}
where $f$ was defined in \eqref{def:f}, and it is now intended that $f(z; \Delta_a, \mu_a)$ contains the resolvent of ${\bf M}^{(a)}$. With similar computations, one can obtain $\psi_{0N}$ and $\psi_{NN}$:
\begin{equation}
\label{eq:psi0NImplicit}
\begin{split}
    \psi_{0N}(z,\xi)&=\frac{2}{N} \frac{\Delta^2_h}{\Delta^2_0 \,  \Delta^2_1} \; f(\xi; \Delta_1, \mu_1)  f(z; \Delta_0, \mu_0)\times\\
    & \times \frac{1}{N}\Tr\mathbb{E}[\mathbf{G}_0(z)\mathbf{G}_1(\xi)]+\mathcal{O}\left(\frac{1}{N^2}\right)
    \end{split}
\end{equation}
and
\begin{align}
\label{eq:psiNNImplicit}
  \psi_{NN}(z,\xi)=\frac{1}{N}\frac{L_{NN}(z, \xi)}{(z-\mu_0)(\xi-\mu_1)} +\mathcal{O}\left(\frac{1}{N^2}\right)
\end{align}
with 
\begin{equation}
\begin{split}
\label{eq:L}
    &L_{NN}(z,\xi)= 1 + 
    f(z; \Delta_0, \mu_0) \frac{1}{N}\Tr\mathbb{E}[\mathbf{G}_0(z)]\\
    &+f(\xi; \Delta_1, \mu_1) \frac{1}{N}\Tr\mathbb{E}[\mathbf{G}_1(\xi)]\\
    &+ f(z; \Delta_0, \mu_0) \, f(\xi; \Delta_1, \mu_1) \frac{1}{N^2} \Tr\mathbb{E}[\mathbf{G}_0(z)] \,\Tr\mathbb{E}[\mathbf{G}_1(\xi)].
    \end{split}
\end{equation}
From these expressions, one can get the final equation for $\psi$. Let us notice that the following quantity is particularly important to get an explicit expression for $\psi$:
\begin{equation}
\label{main:multi_resolv}
  {\bf \Pi}_{k,m}(z,\xi):= \mathbb{E} \left[  \mathbf{G}_0(z)^{k+1} \, \mathbf{G}_1(\xi)^{m+1} \right],\quad m,k\geq 0.
\end{equation}
In particular, it appears from the expressions above that one needs to compute the leading order contributions to the quantities ${\bf \Pi}_{0,0}$, ${\bf \Pi}_{0,1}$ and ${\bf \Pi}_{1,0}$. 

\subsection{Multiresolvent products}
An important result of our work is to find an explicit expression for ${\bf \Pi}_{k,m}$ up to $1/N$ corrections. For the present discussion, we will only concentrate on the leading contribution as $N\to\infty$; the interested reader can refer to Ref.~\cite{paccoros} for details on the $1/N$ contributions to ${\bf \Pi}_{k,m}$, which are useful when computing the finite-size corrections to the eigenvector overlaps, through $\psi$. Let us now give the main steps of the derivation of the leading contribution to ${\bf \Pi}_{k,m}$ for large $N$.\\

\noindent To leading order in $N$, the matrix ${\bf \Pi}_{0,0}(z,\xi)=\mathbb{E} \left[  \mathbf{G}_0(z) \, \mathbf{G}_1(\xi)\right]$ converges to a diagonal one with components given by \cite{cipolloni2022thermalisation, potters_bouchaud_2020}:
\begin{equation}
\label{BigPsi}
\Psi(z,\xi):=\frac{\mathfrak{g}_\sigma(z)-\mathfrak{g}_\sigma(\xi)}{\xi -z-\sigma_W^2(\mathfrak{g}_\sigma(\xi)-\mathfrak{g}_\sigma(z))}.
\end{equation} 
Hence, for large $N$ we have ${\bf \Pi}_{0,0}=\Psi(z,\xi)\mathbb{I}_{N-1}+\mathcal{O}(1/N)$. Let us begin by giving a short motivation for this fact. Since we are only interested in the leading order behavior here, we will consider that ${\bf G}_0(z)$ and ${\bf G}_1(z)$ are simply of size $N\times N$. Let us start by defining ${\bf R}_H(z):=(z-{\bf H})^{-1}$ and averaging one resolvent:
\begin{align*}
\mathbb{E}_{\bf W^{(a)}}[{\bf G}_a(z)]=\mathbb{E}_{\bf W^{(a)}}\left[\frac{1}{z-{\bf H}-{\bf W}^{(a)}}\right]=\mathbb{E}_{\bf W^{(a)}}\left[\mathbf{R}_H(z)\sum_{u=0}^\infty [{\bf W}^{(a)}\mathbf{R}_H(z)]^u\right]
\end{align*}
Making use of the same techniques as in the paragraphs above, namely Wick's theorem for Gaussian variables on ${\bf W}^{(a)}$ and keeping only leading orders, we obtain \cite{VERBAARSCHOT1984367}:
\begin{align}
\mathbb{E}_{\bf W^{(a)}}[{\bf G}_a(z)]&={\bf R}_H\left(z-\sigma_W^2\frac{1}{N}\Tr\mathbb{E}_{\bf W^{(a)}}[{\bf G}_a(z)]\right)+\mathcal{O}\tonde{\frac{1}{N}}\\
&={\bf R}_H\left(z-\sigma_W^2\mathfrak{g}_\sigma(z)\right)+\mathcal{O}\tonde{\frac{1}{N}}
\end{align}
where we used that $\frac{1}{N}\Tr\mathbb{E}_{\bf W^{(a)}}[{\bf G}_a(z)]=\mathfrak{g}_\sigma(z)+\mathcal{O}\tonde{\frac{1}{N}}$ as $N\to \infty$. We shall now make use of this resolvent identity \cite{bun2018overlaps}:
\begin{align}
\mathbb{E}[{\bf G}_0(z){\bf G}_1(\xi)]&=\mathbb{E}_{\bf H}[{\bf R}_H(z-\sigma_W^2\mathfrak{g}_\sigma(z)){\bf R}_H(\xi-\sigma_W^2\mathfrak{g}_\sigma(\xi))]\\
&=\frac{\mathbb{E}_\mathbf{H}[{\bf R}_H(z-\sigma_W^2\mathfrak{g}_\sigma(z))]-  \mathbb{E}_\mathbf{H}[{\bf R}_H(\xi-\sigma_W^2\mathfrak{g}_\sigma(\xi))]}{\xi -z-\sigma^2_W\tonde{\mathfrak{g}_\sigma(\xi)- \mathfrak{g}_\sigma(z)}},
\end{align}
which finally implies to leading order :
\begin{align}
&\mathbb{E}[{\bf G}_0(z){\bf G}_1(\xi)]=
\Psi(z,\xi)\mathbb{I},\\
&\Psi(z,\xi):=\frac{\mathfrak{g}_\sigma(z)-\mathfrak{g}_\sigma(\xi)}{\xi -z-\sigma^2_W\tonde{\mathfrak{g}_\sigma(\xi)- \mathfrak{g}_\sigma(z)}},
\end{align}
where we used the (not too hard to prove) identity $\mathfrak{g}_{\sigma_H}(z-\sigma_W^2\mathfrak{g}(z))=\mathfrak{g}_\sigma(z)$.

\noindent  Next, we have to show how the quantity ${\bf \Pi}_{k,m}$ can be reduced to ${\bf \Pi}_{0,0}$. We first introduce two infinitesimal parameters $\epsilon, \gamma$ and write:
\begin{align*}
\mathbb{E}\left[ \mathbf{G}_0(z)^{k+1}\mathbf{G}_1(\xi)^{m+1}\right] = \lim_{\epsilon,\gamma\to 0}\mathbb{E}\left[ \mathbf{G}_0(z)\cdots \mathbf{G}_0(z+k\gamma)\mathbf{G}_1(\xi)\cdots \mathbf{G}_1(\xi+m\epsilon)\right].
\end{align*}
We aim at re-writing this product as a sum of single resolvent matrices. To do this, we make use of the following Lemma. 
\begin{theorem}
\label{lemma_prod}
If ${\bf M}$ is a real symmetric matrix and we denote ${\bf A}_j:=(j\epsilon+{\bf M})^{-1}$ for $j\in\mathbb{Z}, \epsilon\in\mathbb{R}$, then for any  $k\in\mathbb{N}_{\geq 1}$:
\begin{align*}
   {\bf  A}_0\cdots {\bf A}_k = \frac{1}{\epsilon^kk!}\sum_{j=0}^k(-1)^j\binom{k}{j}{\bf A}_j.
\end{align*}
\end{theorem}
\begin{proof}

We proceed by induction. Indeed notice that for $k=1$ we have ${\bf A}_0{\bf A}_1=(\mathbf{M})^{-1}(\epsilon+\mathbf{M})^{-1}=\frac{1}{\epsilon}((\mathbf{M})^{-1}-(\epsilon+\mathbf{M})^{-1})=\frac{1}{\epsilon}\mathbf{A}_0-\frac{1}{\epsilon}\mathbf{A}_1$. Now suppose that our Lemma is true for a certain $k$, we will prove that it works also for $k+1$. Let us write:
\begin{align*}
    &\mathbf{A}_0\cdots \mathbf{A}_k\mathbf{A}_{k+1}= \frac{1}{\epsilon^kk!}\sum_{j=0}^k(-1)^j\binom{k}{j}\mathbf{A}_j\mathbf{A}
_{k+1}\\&=\frac{1}{\epsilon^{k+1}k!}\sum_{j=0}^k(-1)^j\binom{k}{j}\frac{1}{(k+1-j)}(\mathbf{A}_j-\mathbf{A}_{k+1})\\
    &=\frac{1}{\epsilon^{k+1}(k+1)!}\left[\sum_{j=0}^k(-1)^j\binom{k+1}{j}\mathbf{A}_j-\sum_{j=0}^k(-1)^{j}\binom{k+1}{j}\mathbf{A}_{k+1}\right]\\&=\frac{1}{\epsilon^{k+1}(k+1)!}\sum_{j=0}^{k+1}(-1)^j\binom{k+1}{j}\mathbf{A}_j
\end{align*}
where in the last equality we used that the identity $0=(1-1)^{k+1}=\sum_{j=0}^{k+1}(-1)^j\binom{k+1}{j}$ implies that $\sum_{j=0}^k(-1)^j\binom{k+1}{j}=-(-1)^{k+1}\binom{k+1}{k+1}=-(-1)^{k+1}$. Hence the induction hypothesis is proved.
\end{proof}

\noindent Applying this Lemma, we see that the expectation $\mathbb{E}\left[ \mathbf{G}_0(z)^{k+1}\mathbf{G}_1(\xi)^{m+1}\right]$ can be written as a linear combination of terms of the form $\mathbb{E}\left[ \mathbf{G}_0(z + i \gamma)\mathbf{G}_1(\xi+ j \epsilon)\right]$ for integers $i,j$. For instance, for $k=0$ it holds:
\begin{align*}
    \lim_{\epsilon\to 0}\mathbb{E}\left[\mathbf{G}_0(z)\mathbf{G}_1(\xi)\cdots \mathbf{G}_1(\xi+m\epsilon)\right]&=\lim_{\epsilon\to 0}\frac{(-1)^m}{\epsilon^mm!}\sum_{j=0}^m(-1)^{m-j}\binom{m}{j}\mathbb{E}\left[\mathbf{G}_0(z)\mathbf{G}_1(\xi+j\epsilon)\right]\\
    &=\frac{(-1)^m}{m!}\frac{\partial^m}{\partial\xi^m}\mathbb{E}[\mathbf{G}_0(z)\mathbf{G}_1(\xi)]
\end{align*}
where in the last line we recognized the expression of the $m$-th derivative in the limit $\epsilon\to 0$. The same holds for $k>0$. Then it is not hard to see that we finally obtain
\begin{align}
\label{eq:multi_resolv_formula0}
    {\bf \Pi}_{k,m}(z,\xi)=\frac{(-1)^{k+m}}{k!m!}\frac{\partial^k}{\partial z^k} \frac{\partial^m}{\partial \xi^m } 
\mathbb{E} \left[  \mathbf{G}_0(z) \, \mathbf{G}_1(\xi)\right].
\end{align}
At first order, when $N$ is large, we then get 
\begin{align}
\label{eq:final_Pi_km}
    {\bf \Pi}_{k,m}(z,\xi)=\frac{(-1)^{k+m}}{k!m!}\frac{\partial^k}{\partial z^k} \frac{\partial^m}{\partial \xi^m } 
\Psi(z,\xi)+\mathcal{O}\tonde{\frac{1}{N}}.
\end{align}

\section{Expression of the overlaps and simulations}
\label{sec:rmt_simulations_overlaps}
By combining the expressions \eqref{eq:psi00Implicit},\eqref{eq:psi0NImplicit} and \eqref{eq:psiNNImplicit} with Eq.~\eqref{eq:final_Pi_km}, we can obtain an explicit expression for $\psi$:
\begin{align}
\label{eq:explicit_final_psi}
\begin{split}
    \psi(z,\xi)&=\Psi(z,\xi)+\frac{1}{N}\bigg\{\Psi^{(1)}(z,\xi)+\frac{1}{(z-\mu_0)(\xi-\mu_1)}\\
    &+\frac{\Delta_0^2}{z-\mu_0-\Delta_0^2\mathfrak{g}_\sigma(z)}\left[-\partial_z\Psi(z,\xi)+\frac{\mathfrak{g}_\sigma(z)}{(z-\mu_0)(\xi-\mu_1)}\right]\\
    &+\frac{\Delta_1^2}{\xi-\mu_1-\Delta_1^2\mathfrak{g}_\sigma(\xi)}\left[-\partial_\xi\Psi(z,\xi)+\frac{\mathfrak{g}_\sigma(\xi)}{(z-\mu_0)(\xi-\mu_1)}\right]\\
    &+\frac{1}{[z-\mu_0-\Delta_0^2\mathfrak{g}_\sigma(z)][\xi-\mu_1-\Delta_1^2\mathfrak{g}_\sigma(\xi)]}\bigg[\Delta_h^4\Psi^2(z,\xi)\\
&+2\Delta_h^2\Psi(z,\xi)+\frac{\Delta_0^2\Delta_1^2\mathfrak{g}_\sigma(z)\mathfrak{g}_\sigma(\xi)}{(z-\mu_0)(\xi-\mu_1)}\bigg]\bigg\}+ \mathcal{O}\tonde{\frac{1}{N^2}}.
\end{split}
\end{align}
For completeness, we introduced $\Psi^{(1)}$, which represents the $1/N$ finite size correction to $\frac{1}{N}\Tr\mathbb{E}[{\bf G}_0(z){\bf G}_1(\xi)]$. However, one can show that such contribution is not important when using $\psi$ to extract the eigenvector overlaps with Eq.~\eqref{re_psi}. The interested reader can refer to Ref.~\cite{paccoros} for a derivation of $\Psi^{(1)}$.

With the expression for $\psi$, we can now use equation \eqref{re_psi} to compute the three types of eigenvector overlaps. Moreover, we make numerical simulations to verify their correctness.

\subsection{Bulk-bulk overlaps}
Each element of the pair of random matrices defined in Eq. \eqref{eq:MatrixForm} has a GOE block ${\bf B}^{(a)}$ having the same statistics (only the matrix elements in the special row and column have a statistics that depends on $a$). The bulk spectral densities  $\rho_\sigma(\lambda)$ of both matrices in the large $N$ limit are determined by these blocks, and thus are exactly the same for both matrices, given by \eqref{eq:DoSgoe}. The spectral densities are supported on the interval $[-2\sigma,2\sigma]$; when the respective eigenvalues are in the bulk, i.e. $\lambda^0,\lambda^1\in [-2\sigma,2\sigma]$, recall that the overlap between the two corresponding eigenvectors reads:
\begin{align*}
\Phi(\lambda^{0},\lambda^{1}):=N\mathbb{E}[\left( \mathbf{u}_{\lambda^{0}}\cdot \mathbf{u}_{\lambda^{1}}\right)^2].
\end{align*}
To recover the leading-order term in the overlap between bulk eigenvectors we have to neglect all $1/N$ corrections in Eq.~\eqref{eq:explicit_final_psi}. In general, we see that when applying Eq.~\eqref{re_psi}, the square roots in $\psi$ will give rise to branch cuts, contained in $\mathfrak{g}_\sigma$ and its derivatives. In order to face this issue we have to carefully take the limit of $\mathfrak{g}_\sigma(x\pm i\eta)$ when $\eta\to0^+$. Since the branch cuts come from all the terms of the form $\sqrt{x^2-4\sigma^2}$, we have to carefully analyse $\sqrt{(x\pm i\eta)^2-4\sigma^2}$ as $\eta\to 0^+$. The square root function in the complex plane presents a branch cut, which recall we fix here to be toward the negative real axis (i.e. we define angles between $(-\pi,\pi]$). With such convention, we simply have that the square root behaves as follows, see also Sec.~\ref{sec:Goe_intro}:
\begin{align*}
    \lim_{\eta\to 0^+}\sqrt{(x\pm i\eta)^2-4\sigma^2}=
    \begin{cases}
        \sqrt{x^2-4\sigma^2}\quad |x|\geq 2\sigma\\
        \pm\text{sign}(x)i\sqrt{4\sigma^2-x^2}\quad |x|< 2\sigma
    \end{cases}
\end{align*}
and by applying this to $\mathfrak{g}_\sigma$, defined in Eq.~\eqref{eq.Stjlt}, we obtain
\begin{equation}
\label{br_cut}
\lim_{\eta \to 0^+}\mathfrak{g}_\sigma(x \mp i \eta)=\begin{cases}
  \frac{1}{2 \sigma^2} \tonde{x- \text{sign} (x) \sqrt{x^2-4 \sigma^2}} \quad  |x|> 2 \sigma\\
  \frac{1}{2 \sigma^2}\tonde{x \pm i \sqrt{4 \sigma^2-x^2}} \quad  |x|< 2 \sigma
  \end{cases} \equiv
\mathfrak{g}_R(x) \pm i \mathfrak{g}_I(x),
  \end{equation}
where $\text{Im}\,\mathfrak{g}(x) \neq 0$ only if $|x|< 2 \sigma$. For later convenience, let us also define the following function:
\begin{equation}
    \zeta(z):=z-\sigma_W^2\mathfrak{g}_\sigma(z)
\end{equation}
and notice that
\begin{equation}
   \lim_{\eta \to 0^+} \zeta(x \mp i \eta)= x - \sigma_W^2
\mathfrak{g}_R(x) \mp i \sigma^2_W \mathfrak{g}_I(x)=: \zeta_R(x)\pm i \zeta_I(x).
\end{equation}
With these definitions, from Eq.~\eqref{re_psi} we can extract the bulk-bulk overlap by noting that:
\begin{align}
\Phi(x, y)&=  \frac{\lim_{\eta\to 0^+}\text{Re}\left[\Psi(x-i\eta,y+i\eta)-\Psi(x-i\eta,y-i\eta)\right]}{2\pi^2\rho_\sigma(x)\rho_\sigma(y)}+ \mathcal{O}\tonde{\frac{1}{N}}.
\end{align} 
Now, we have that we can write
\begin{align}
    \Psi(z,\xi)=\frac{\mathfrak{g}_\sigma(z)-\mathfrak{g}_\sigma(\xi)}{\zeta(\xi)-\zeta(z)},
\end{align}
which implies:
\begin{align*}
&\lim_{\eta\to 0^+}\text{Re}\left[\Psi(x-i\eta,y+i\eta)-\Psi(x-i\eta,y-i\eta)\right]\\
&=\frac{\mathfrak{g}_R(x)-\mathfrak{g}_R(y)+i\mathfrak{g}_I(x)+i\mathfrak{g}_I(y)}{\zeta_R(y)-\zeta_R(x)-i\zeta_I(y)-i\zeta_I(x)}-\frac{\mathfrak{g}_R(x)-\mathfrak{g}_R(y)+i\mathfrak{g}_I(x)-i\mathfrak{g}_I(y)}{\zeta_R(y)-\zeta_R(x)+i\zeta_I(y)-i\zeta_I(x)}\\
&=\frac{ -4\,\mathfrak{g}_I(x)\,\mathfrak{g}_I(y)\, \sigma_W^2\, 
((\mathfrak{g}_R(x) - \mathfrak{g}_R(y))\sigma_w^2 - x + y) \, (2 (\mathfrak{g}_I(x) - \mathfrak{g}_I(y))\sigma_w^2 - x + y)}{\quadre{\tonde{\zeta_R(x)-\zeta_R(y)}^2+ \tonde{\zeta_I(x)+\zeta_I(y)}^2}\quadre{\tonde{\zeta_R(x)-\zeta_R(y)}^2+ \tonde{\zeta_I(x)-\zeta_I(y)}^2}}.
\end{align*}
By using the definitions for $\mathfrak{g}$ and $\zeta$, we can simplify to get:
\begin{align}
\label{eq:phi_bulk_bulk}
\begin{split}
    \Phi(x,y)&=\frac{ 2 \sigma_W^2 \quadre{\zeta_R(x)-\zeta_R(y)} (x-y) }{\quadre{\tonde{\zeta_R(x)-\zeta_R(y)}^2+ \tonde{\zeta_I(x)+\zeta_I(y)}^2}\quadre{\tonde{\zeta_R(x)-\zeta_R(y)}^2+ \tonde{\zeta_I(x)-\zeta_I(y)}^2}}\\
    &=\frac{2 \sigma_W^2 \, \tonde{1-\frac{\sigma_W^2}{2 \sigma^2}} (x-y)^2}{\displaystyle \prod_{k=\pm }A_k} 
\end{split}
\end{align}

where we have defined 
\begin{align}
\begin{split}
    A_k:=\frac{\sigma_W^4}{4 \sigma^4}\tonde{\sqrt{4 \sigma^2-x^2}+k\sqrt{4 \sigma^2-y^2}}^2+\tonde{1-\frac{\sigma_W^2}{2 \sigma^2}}^2 (x-y)^2.
\end{split}
\end{align}

This expression depends only on the parameters $\sigma, \sigma_W$ defining the statistics of the GOE blocks ${\bf B}^{(a)}$, and it is consistent with the results of Ref.~\cite{bun2018overlaps}. Indeed, Ref.~\cite{bun2018overlaps} presents the calculation of the overlap between bulk eigenvectors of matrices of the form ${\bf C}+ {\bf A}+ {\bf D}^{(a)}$, where ${\bf C}$ is a deterministic (population) matrix, while ${\bf A}$ and ${\bf D}^{(a)}$ are $N\times N$ GOE matrices with variances $\sigma_{A}^2$ and $\sigma^2_a-\sigma^2_{A}$, respectively. The overlap is shown to be independent of the matrix ${\bf C}$, and to coincide with \eqref{eq:phi_bulk_bulk} with $\sigma_H^2 \to \sigma_A^2$ and $ \sigma_W^2 \to \sigma_a^2-\sigma_A^2$, as expected.  
Notice that the case considered in Ref.~\cite{bun2018overlaps} corresponds to vanishing finite-rank perturbations ($\Delta_a=v_a=\sigma$, $\mu_a=0$); therefore, no isolated eigenvalue(s) are present in that case. Eq.~\eqref{eq:phi_bulk_bulk}  shows that the finite rank perturbations do not affect the overlap between bulk eigenvectors, to leading order.  
Let us stress that the $1/N$ contribution to \eqref{eq:phi_bulk_bulk} can also be determined explicitly, see Fig.5 in Ref.~\cite{paccoros}. \\

\begin{figure}[t!]
\centering
\includegraphics[width=0.65\textwidth]{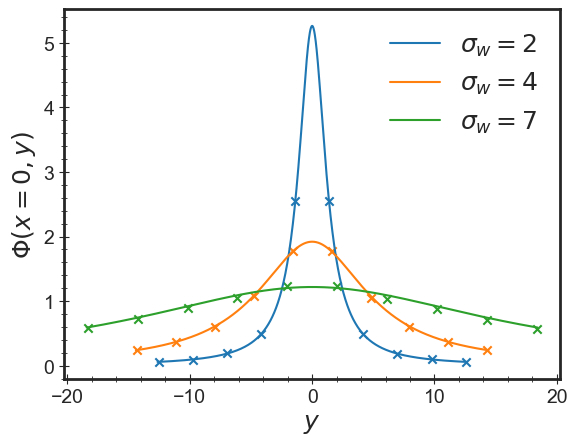}
\caption{Plot representing the theoretical curves of the bulk-bulk overlap \eqref{eq:phi_bulk_bulk} with their respective numerical simulations (colored crosses). We used $\sigma_H=6, \Delta_h=2.5, \Delta_{w,0}=2, \Delta_{w,1}=1.5, \mu_0=\mu_1=0$, and we plot the overlap for $x=\lambda^0=0$ and $y=\lambda^1\in[-2\sigma,2\sigma]$, for several choices of $\sigma_W$: the central peak height is inversely proportional to the value of $\sigma_W$.
The numerical simulations were carried out by generating 500 times pairs of random matrices of size $N=200$, with $d\lambda=0.1$.}
\label{fig:phi_bulk_bulk}
\end{figure}

\noindent A numerical check of \eqref{eq:phi_bulk_bulk} is given in Fig.~\ref{fig:phi_bulk_bulk}. Let us briefly comment on how the numerical simulations are performed. In order to obtain the eigenvector overlaps numerically, we generate the three GOE random matrices $\mathbf{H}$, $\mathbf{W}^{(0)}$ and $\mathbf{W}^{(1)}$; similarly, we generate the Gaussian variables $h_{iN}$, $w_{iN}^0, w_{iN}^1$. The elements $m_{NN}^0$ and $m_{NN}^1$ are simply set equal to $\mu_0$ and $\mu_1$ respectively, i.e. we set $v_0=v_1=0$; this is motivated by the fact that, as we have shown above, to leading order all of our analytical results are independent on the variances $v_0,v_1$. After having generated such entries, we sum them up to get the two matrices $\mathbf{M}^{(0)}$ and $\mathbf{M}^{(1)}$, according to Sec.\ref{subsec:matrix}. We then diagonalize them and consider eigenvectors associated to eigenvalues in the intervals $[x-d\lambda,x+d\lambda]$ and $[y-d\lambda,y+d\lambda]$ respectively, with $d\lambda\gg N^{-1}$. Then for each pair of such eigenvectors of the two matrices, we compute their squared dot product, and average them together. We repeat this procedure over many realizations: the numerical points in the Figures correspond to averages over the realizations. All the Figures reported in the following are generated in this way, with the slight difference that when isolated eigenvalues are considered, there is no window $d\lambda$ on which to perform the first average, and thus the number of realizations has to be increased significantly.

\subsection{Iso-bulk overlaps}
Consider the case in which at least one of the two matrices $\mathbf{M}^{(0)}, \mathbf{M}^{(1)}$ has the isolated eigenvalue. Without loss of generality, we take such matrix to be $\mathbf{M}^{(0)}$, meaning that \eqref{eq:ConditionEx} is satisfied and $\lambda^{0}_{\text{iso}}$ exists (it is clear that all results will hold if we exchange the two matrices).  We impose no condition on $\mathbf{M}^{(1)}$, and we pick a bulk eigenvalue $y:=\lambda^1\in[-2\sigma,2\sigma]$. \\

\noindent In order to compute $\Phi(\lambda_{\rm iso}^0,y)$ one has to make use of  Eq.~\eqref{re_psi}, and consider only the terms of $\psi$ in Eq.~\eqref{eq:explicit_final_psi} which present a singularity when evaluated at $x=\lambda_{\rm iso}^0:=\lambda_{\rm iso, -}^0$, see \eqref{app:solutions}. We are therefore focusing on  $|x|=|\lambda_{\rm iso}^0|>2\sigma$ and $|y|\leq 2\sigma$: the first argument of $\psi(x,y)$ does not belong to the bulk of the eigenvalue density, while the second does. From the expression of $\psi$ in Eq.~\eqref{eq:explicit_final_psi} we see that the term that can generate a singularity is $1/(x-\mu_0\Delta_0^2\mathfrak{g}_\sigma(x))$. It is simple to check that, given the expressions for $\lambda^0_{\text{iso},\pm}$ in \eqref{app:solutions} (notice that $\lambda^0_{\text{iso},+}$ might not be isolated by its expression is well defined when $\lambda^0_{\text{iso},-}$ exists), then for $x$ real one has:
\begin{align}\label{eq:res}
    \frac{1}{x-\mu_0-\Delta_0^2\mathfrak{g}_\sigma(x)}= \frac{\tonde{1- \frac{\Delta_0^2}{2 \sigma^2}}x-\mu_0-\text{sign}(x) \frac{\Delta_0^2}{2 \sigma^2} \sqrt{x^2-4\sigma^2}}{\tonde{1-\frac{\Delta_0^2}{\sigma^2}} [x-\lambda^0_{\text{iso},+}(\mu_0, \Delta_0, \sigma)][x-\lambda^0_{\text{iso},-}(\mu_0, \Delta_0, \sigma)]}.
\end{align}
This term is therefore singular for $x \to \lambda_{\rm iso}^0$.  In particular, as we already pointed out in Eq.~\ref{eq:rmt_delta_q}, using Eq.~\eqref{eq:res} we obtain:
\begin{align}\label{eq:Residuo}
   \lim_{\eta \to 0^+} \frac{1}{x\pm i \eta-\mu_0-\Delta_0^2\mathfrak{g}_\sigma(x\pm i \eta)}=  \mp i \pi \delta(x-\lambda_{\rm iso}^0) \mathfrak{q}_{\sigma,\Delta_0}(\lambda_{\rm iso}^0,\mu_0) + \text{   regular terms},
   \end{align}
where $\mathfrak{q}_{\sigma,\Delta}(\lambda,\mu)$ is given in Eq. \eqref{eq:Defq} while the regular terms are not proportional to the delta. To select the relevant contributions to  $\Phi(\lambda_{\rm iso}^0,y)$ we single out the terms in \eqref{eq:explicit_final_psi} which produce a delta function when $x \to \lambda_{\rm iso}^0$. These include all terms multiplied by $(z-\mu_0-\Delta_0^2\mathfrak{g}(z))^{-1}$ that have a branch cut in $y\pm i\eta$ when $\eta\to 0^+$ (because, due to \eqref{eq:Residuo}, we need factors that provide a non-zero imaginary part). From \eqref{eq:explicit_final_psi} we can group these relevant terms in a new function $\hat{\psi}$:
\begin{align}\label{eq:Aux}
\begin{split}
    \hat{\psi}(z,\xi):=\frac{1}{z-\mu_0-\Delta_0^2\mathfrak{g}(z)}\Bigg[&-\Delta_0^2\partial_z\Psi(z,\xi)+\frac{1}{\xi-\mu_1-\Delta_1^2\mathfrak{g}_\sigma(\xi)}\bigg[\Delta_h^4\Psi^2(z,\xi)\\
&+2\Delta_h^2\Psi(z,\xi)+\frac{\Delta_0^2\Delta_1^2\mathfrak{g}_\sigma(z)\mathfrak{g}_\sigma(\xi)}{(z-\mu_0)(\xi-\mu_1)}\bigg]
\Bigg].
\end{split}
\end{align}
Then, by using \eqref{re_psi}, one recovers the iso-bulk overlap from the following expression:
\begin{align}
\label{eq:psihat_to_overlap}
\begin{split}
       \text{Re}\lim_{\eta\to 0^+}\left[\hat{\psi}(x-i\eta,y+i\eta)-\hat{\psi}(x-i\eta,y-i\eta)\right]=2\pi^2\rho_\sigma(y)\delta(x-\lambda_\text{iso}^0)\Phi(\lambda_\text{iso}^0,y)
\end{split}
\end{align}
\noindent Consider the first term, one has:
\begin{equation}
    \partial_z \Psi(z,\xi)= \frac{\mathfrak{g}_\sigma(z)-\mathfrak{g}_\sigma(\xi)+ \mathfrak{g}'_\sigma(z)(\xi-z)}{[\xi-z -\sigma_W^2 (\mathfrak{g}_\sigma(\xi)-\mathfrak{g}_\sigma(z))]^2}.
\end{equation}
Now, for real $|x|> 2 \sigma$ , given that $\mathfrak{g}_I(x)= \zeta_I(x)=0$, we obtain:
\begin{align*}
&\text{Im}\lim_{\eta \to 0^+} \; \lim_{z \to x-i \eta} \quadre{\partial_z \Psi(z, y + i \eta)- \partial_z \Psi(z, y - i \eta)}\\&=
\text{Im}\left[\frac{\mathfrak{g}_R(x)-\mathfrak{g}_R(y)+i\mathfrak{g}_I(y) + \mathfrak{g}'_R(x)(y-x)}{[\zeta_R(y)-i\zeta_I(y)-\zeta_R(x)]^2}-\frac{\mathfrak{g}_R(x)-\mathfrak{g}_R(y)-i\mathfrak{g}_I(y) + \mathfrak{g}'_R(x)(y-x)}{[\zeta_R(y)+i\zeta_I(y)-\zeta_R(x)]^2}\right]\\
&=\mathfrak{g}_I(y)\frac{2 [\zeta_R(y)- \zeta_R(x)]^2 -2\zeta^2_I(y) +4 \sigma_W^2 \, \quadre{ \mathfrak{g}_R(x)-\mathfrak{g}_R(y)-(x-y)\mathfrak{g}'_R(x)} \,[\zeta_R(y)-\zeta_R(x)]}{\tonde{\quadre{\zeta_R(y)-\zeta_R(x)}^2 + \zeta_I^2(y)}^2}\\
&\equiv \mathfrak{g}_I(y)A(x,y),
\end{align*}
notice that $\mathfrak{g}_I(y)$ will simplify with $\rho_\sigma(y)$ in \eqref{eq:psihat_to_overlap} (up to a factor $\pi$). We can proceed in a similar fashion for all the other terms in Eq.~\eqref{eq:Aux}. The final expression for the iso-bulk overlap then reads:
\begin{align}
\begin{split}
\label{eq:phi_iso_bulk}
\Phi(\lambda_\text{iso}^0,y)&=\frac{\mathfrak{q}_{\sigma,\Delta_0}(\lambda_{\rm iso}^0,\mu_0) }{2}\\&\times \quadre{\Delta_0^2 A(\lambda_{\rm iso}^0,y)-\Delta_h^4 B(\lambda_{\rm iso}^0,y) - 2 \Delta_h^2 C(\lambda_{\rm iso}^0,y)- \Delta_0^2\Delta_1^2 D(\lambda_\text{iso}^0,y)}
\end{split}
\end{align}
where $\mathfrak{q}_{\sigma,\Delta}$ has been defined in Eq.~\eqref{eq:Defq}, and with analogous analytical computations as above, we  have computed
\begin{align}
\begin{split}
&B(x,y)=\frac{1}{[y -\mu_1-\Delta^2 \mathfrak{g}_R(y)]^2+ \Delta^4 \mathfrak{g}^2_I(y)} \frac{1}{\quadre{(\zeta_R(x)- \zeta_R(y))^2 +\zeta^2_I(y)}^2} \, b(x,y)\\
&C(x,y)=\frac{
2  (y-x)\,  [y -\mu_1-\Delta^2 \mathfrak{g}_R(y)]- 2 \Delta^2\quadre{\sigma_W^2 \mathfrak{g}^2_I(y) + (\mathfrak{g}_R(x)-\mathfrak{g}_R(y)) (\zeta_R(x)-\zeta_R(y))}}{\quadre{\tonde{\zeta_R(x)-\zeta_R(y)}^2+ \zeta^2_I(y)} \; \quadre{\tonde{y -\mu_1-\Delta^2 \mathfrak{g}_R(y)}^2+ \Delta^4 \mathfrak{g}^2_I(y)}},\\
&D(x,y)=-\frac{\mathfrak{g}_R(x)}{(x - \mu_0) \left[y^2 + \mathfrak{g}_I(y)^2 \Delta_1^4 - 
2 y (\mathfrak{g}_R(y) \Delta_1^2 + \mu_1) + (\mathfrak{g}_R(y)\Delta_1^2 + \mu_1)^2\right]}
\end{split}
\end{align}
with
\begin{equation}\label{eq:CostantiMisto2}
\begin{split}
 b(x,y)&= \quadre{y -\mu_1-\Delta^2 \mathfrak{g}_R(y)} 4 (x-y)
\quadre{(\mathfrak{g}_R(x)- \mathfrak{g}_R(y)) \, (\zeta_R(x)- \zeta_R(y))- \sigma_W^2 \mathfrak{g}_I^2(y)}\\
&- 2 \Delta^2 
\tonde{  \quadre{(\mathfrak{g}_R(x)- \mathfrak{g}_R(y)) \, (\zeta_R(x)- \zeta_R(y))- \sigma_W^2 \mathfrak{g}_I^2(y)}^2- (x-y)^2 \mathfrak{g}^2_I(y)}.
 \end{split}
\end{equation}
Let us remark that we have slightly changed notation with respect to the original work in Ref.~\cite{paccoros}, and that one can get a more explicit expression in terms of the parameters of the model. However, we have found that the above method is more straightforward to understand the computation. \\

\noindent In Fig.~\ref{fig:phi_iso_y} we show that the complicated parameter dependencies of \eqref{eq:phi_iso_bulk} are exact, and numerical simulations perfectly agree with our theoretical results.\\

The formula for $\Phi(\lambda_\text{iso}^0,y)$ can be directly applied to Eq.~\eqref{eq:energy_quench_1} in Chapter~\ref{chapter:energy_landscapes}, in order to obtain the perturbed geodesic pathways. This is done numerically by plugging inside Eq.~\eqref{eq:phi_iso_bulk} the specific parameters of that chapter, found in Eq.~\eqref{eq:app:FormPar}.

\begin{figure}[h!]
\centering
\includegraphics[width=0.65\textwidth]{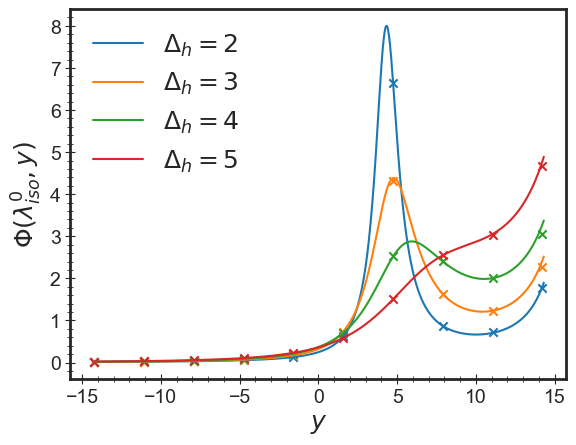}
\caption{Plot representing the theoretical curves of the bulk-isolated overlap \eqref{eq:phi_iso_bulk} with their respective numerical simulations (colored crosses). We used $\sigma_H=6.5, \sigma_W=3, \Delta_{w,0}=2, \Delta_{w,1}=1.5, \mu_0=15, \mu_1=4$, and we plot the overlap for $x=\lambda^0_{\text{iso}}$ and $y=\lambda^1\in[-2\sigma,2\sigma]$ (were clearly $\sigma=\sqrt{\sigma_H^2+\sigma_W^2}$) for several choices of $\Delta_h$: the function value at $y=2\sigma$ is proportional to the strength of $\Delta_h$.  The numerical simulations were carried out by generating 1000 times pairs of random matrices of size $N=500$. As for Fig.\ref{fig:double_evals} we set $v_0=v_1=0$ given that the final results do not depend on them, to leading orders.}
\label{fig:phi_iso_y}
\end{figure}

\subsection{Iso-iso overlaps}
\noindent We now consider the case in which both $\lambda_{\text{iso}}^0$ and $\lambda_{\text{iso}}^1$ exist, and we give the expression for the rescaled overlap \eqref{phi_rescaled} of the corresponding eigenvectors. 
In order to compute $\Phi(\lambda_{\rm iso}^0,\lambda_{\rm iso}^1)$ we have to make use of  Eq.~\eqref{re_psi}, and consider only the part of $\psi$ in Eq.~\eqref{eq:explicit_final_psi} which presents a singularity when evaluated at both of the isolated eigenvalues $\lambda_{\rm iso}^a:=\lambda_{\rm iso, -}^a$ for $a\in\{0,1\}$. The relevant term is the one proportional to the product  
\begin{align}
\gamma(z,\xi):=\frac{1}{z-\mu_0-\Delta_0^2\mathfrak{g}_\sigma(z)}\frac{1}{\xi-\mu_1-\Delta_1^2\mathfrak{g}_\sigma(\xi)}
\end{align}
in \eqref{eq:explicit_final_psi}. We single out such a term, defining:
\begin{align}\label{eq:PsiTilde}
    \begin{split}
&\tilde{\psi}(z,\xi):=
\gamma(z,\xi)\beta(z,\xi)
    \end{split}
\end{align}
with 
\begin{align}
\label{eq:rmt_beta}
\beta(z,\xi):=\bigg(\Delta_h^4\Psi^2(z,\xi)+2\Delta_h^2\Psi(z,\xi)+\frac{\Delta_0^2\Delta_1^2\mathfrak{g}_\sigma(z)\mathfrak{g}_\sigma(\xi)}{(z-\mu_0)(\xi-\mu_1)}\bigg)
\end{align}
Then, by \eqref{re_psi}, we extract the iso-iso overlap with the following identity:
\begin{align}
\label{eq:psi_tilde_Re}
    \frac{\text{Re}\lim_{\eta\to0^+}[\tilde{\psi}(x-i\eta,y+i\eta)-\tilde{\psi}(x-i\eta,y-i\eta)]}{2\pi^2}=\delta\left(x-\lambda^0_\text{iso}\right)\delta\left(x-\lambda^1_\text{iso}\right)\tilde{\Phi}(\lambda^0_\text{iso},\lambda^1_\text{iso})
\end{align}

\noindent Using Eq.~\eqref{br_cut} we can obtain:
\begin{equation}
\label{eq:gamma_sokot}
\begin{split}
     \lim_{\eta \to 0^+} \gamma(x-i\eta,y\pm i\eta)=&\mp \pi^2  \delta(x-\lambda_{\rm iso}^0)  \delta(y-\lambda_{\rm iso}^1) \mathfrak{q}_{\sigma,\Delta_0}(\lambda_{\rm iso}^0,\mu_0)\mathfrak{q}_{\sigma,\Delta_1}(\lambda_{\rm iso}^1,\mu_1)\\&+ \text{  regular terms  },
    \end{split}
\end{equation}
where we neglect all terms that are not proportional to the product of delta functions, since they won't contribute to \eqref{eq:psi_tilde_Re}. 
The terms within brackets in \eqref{eq:rmt_beta} are real when computed at $x,y \to \lambda_{\rm iso}^a$ due to the fact that $|\lambda_{\rm iso}^a|> 2 \sigma$, for $a\in\{0,1\}$. Therefore, by combining \eqref{eq:gamma_sokot} and \eqref{eq:psi_tilde_Re}, we find:
\begin{align}
\label{eq:iso_iso_overlap}
\begin{split}
\tilde{\Phi}(\lambda_\text{iso}^0,\lambda_\text{iso}^1)&=
\mathfrak{q}_{\sigma,\Delta_0}(\lambda_{\rm iso}^0,\mu_0)\mathfrak{q}_{\sigma,\Delta_1}(\lambda_{\rm iso}^1,\mu_1) \beta(\lambda_\text{iso}^0,\lambda_\text{iso}^1)\\
&=\mathfrak{q}_{\sigma,\Delta_0}(\lambda_{\rm iso}^0,\mu_0)\mathfrak{q}_{\sigma,\Delta_1}(\lambda_{\rm iso}^1,\mu_1)\quadre{\Delta_h^2\Psi(\lambda_{\rm iso}^0,\lambda_{\rm iso}^1)
+1}^2
\end{split}
\end{align}
where in the last line we used that the equation satisfied by the isolated eigenvalues, Eq. \eqref{egval_eqn}, implies that the last term within brackets in \eqref{eq:rmt_beta} is equal to $1$. In Fig.~\ref{fig:phi_iso_iso_simulations} we show the validity of our formula, comparing with some numerical simulations. \\

\begin{figure}[t!]
\centering
\includegraphics[width=0.495\textwidth]
{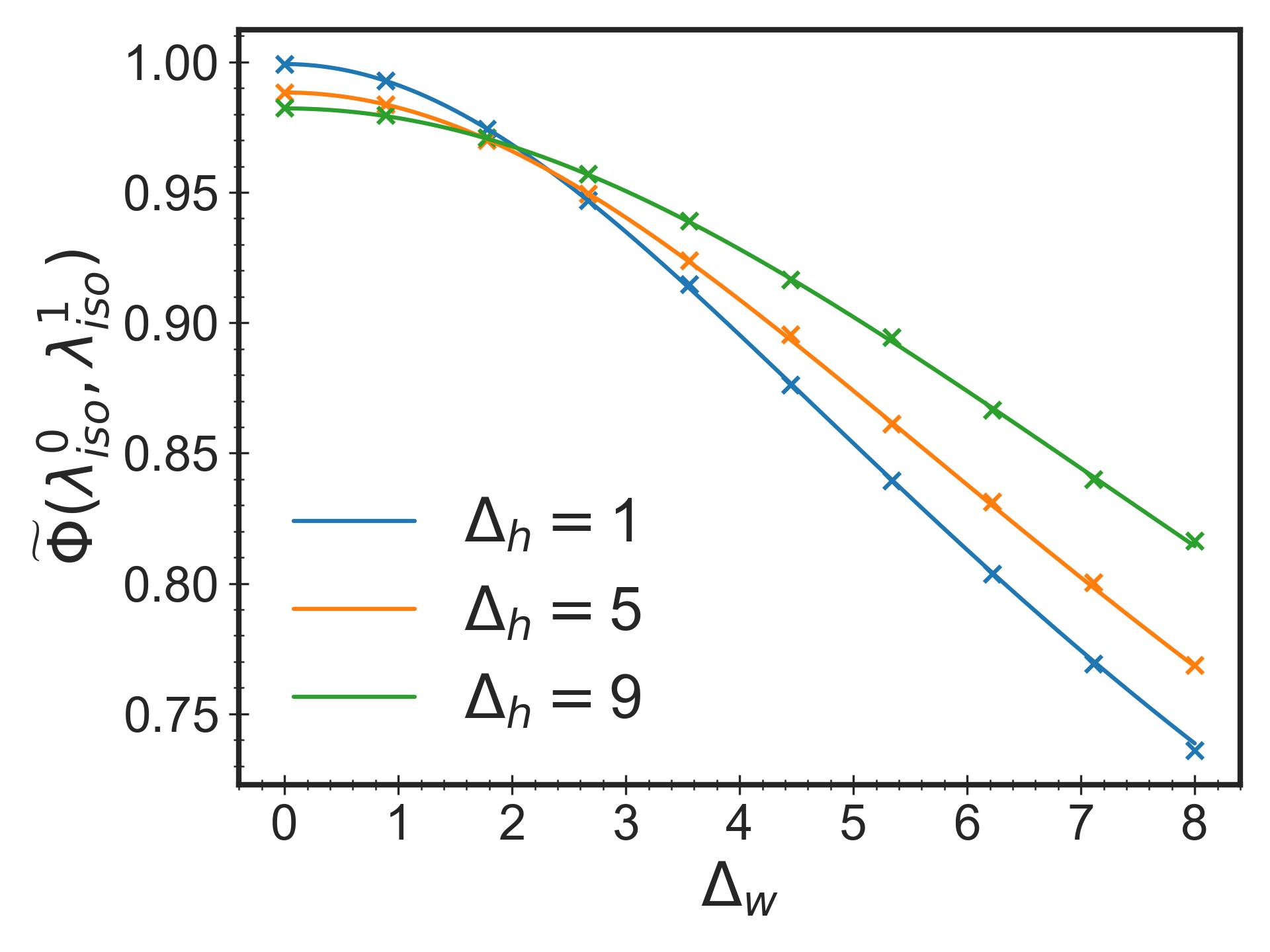}
\includegraphics[width=0.495\textwidth]{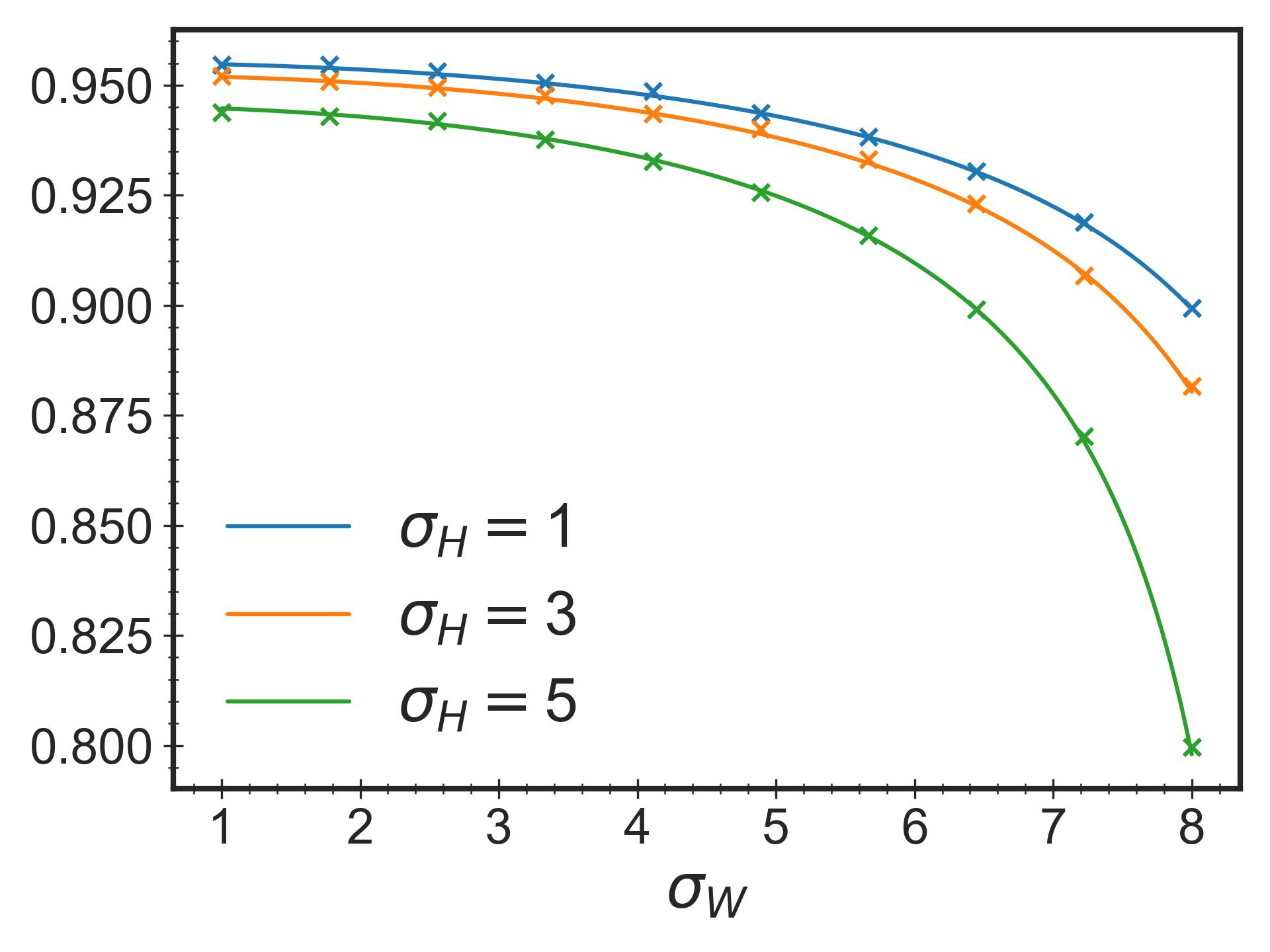}
\caption{Plots of the iso-iso overlap in Eq.~\ref{eq:iso_iso_overlap} and comparison with numerical simulations. \textit{Left}: we used $\sigma_H=8$, $\sigma_W=5$, $\mu_0=\mu_1=23$, $v_0=v_1=0$, and we plot the overlap as a function of $\Delta_w=\Delta_{w,0}=\Delta_{w,1}$ for various choices of $\Delta_h$ (one per color, three in total). \textit{Right}: we used $\Delta_h$, $\Delta_w=\Delta_{w,0}=\Delta_{w,1}=2$, $\mu_0=\mu_1=19$ and we plot as a function of $\sigma_W$ for various choices of $\sigma_H$ (one per color, three in total). \\
In both plots we made simulations with 500 iterations of matrices with $N=300$.}
\label{fig:phi_iso_iso_simulations}
\end{figure}

\noindent Let us comment on some limiting values of this expression. In the case in which the two matrices ${\bf M}^{(a)}$ have uncorrelated entries in the special line and column (meaning that $\Delta_h=0$) then the overlap reduces to $\mathfrak{q}_{\sigma,\Delta_0}(\lambda_{\rm iso}^0,\mu_0)\mathfrak{q}_{\sigma,\Delta_1}(\lambda_{\rm iso}^1,\mu_1)$ and thus it coincides with the product of two terms like \eqref{eq:ProjVectors}, one for each matrix. In fact, if we decompose the eigenvectors associated to the isolated eigenvalues into a component along the special direction ${\bf e}_N$ and one orthogonal to it, we obtain: $\mathbf{u}_{\lambda_{\text{iso}}^a}= (\mathbf{u}_{\lambda_{\text{iso}}^a} \cdot {\bf e}_N) {\bf e}_N + {\bf v}^a$ with ${\bf v}^a\cdot{\bf e}_N=0$.
In particular, the overlap reads $\mathbf{u}_{\lambda_{\text{iso}}^1}\cdot \mathbf{u}_{\lambda_{\text{iso}}^0}=(\mathbf{u}_{\lambda_{\text{iso}}^0} \cdot {\bf e}_N)(\mathbf{u}_{\lambda_{\text{iso}}^1} \cdot {\bf e}_N)+{\bf v}^0\cdot{\bf v}^1$. Then it is natural to expect ${\bf v}^0\cdot{\bf v}^1=0$ if the two special lines are not correlated, implying that $(\mathbf{u}_{\lambda_{\text{iso}}^0} \cdot \mathbf{u}_{\lambda_{\text{iso}}^1})^2= (\mathbf{u}_{\lambda_{\text{iso}}^0} \cdot {\bf e}_N)^2 (\mathbf{u}_{\lambda_{\text{iso}}^1} \cdot {\bf e}_N)^2$, which using \eqref{eq:ProjVectors} is precisely \eqref{eq:iso_iso_overlap} for $\Delta_h=0$.\\

\noindent On the other hand, when the two matrices are fully correlated ($\sigma_W=0=\Delta_{w,0}=\Delta_{w,1}$) the overlap is maximal and equal to one, as can be checked. The dependence of the overlap $\tilde{\Phi}(\lambda_{\rm iso}^0,\lambda_{\rm iso}^1)$ on the variances $\sigma_W, \sigma_H$ and $\Delta_h,\Delta_w$ is shown in Fig. \ref{fig:phi_iso_iso_simulations}.

\subsection{Application: matrix denoising}
\label{sec:pca}

In this section, we consider the case of purely additive rank-1 perturbations to the GOE matrices. In our setting, this corresponds to choosing $\Delta_h=v_h=\sigma_H$ and $\Delta_{w,a}=v_{w,a}=\sigma_W$ for both $a=0,1$. Moreover we shall assume $\mu_0=\mu_1=\mu$. This setting has a clear interpretation as a denoising problem: the perturbed matrices \eqref{eq:MatRiscritte} can in fact be written in this case as:
\begin{equation}\label{eq:addPert}
{\bf M}^{(a)}= {\bf X}^{(a)} + \mu\, {\bf e}_N {\bf e}_N^T,
\end{equation}
where the second term is a rank-one perturbation, also known as the \textit{spike}, ${\bf e}_N$ is the \textit{signal}, and ${\bf X}^{(a)}$ are $N \times N$  GOE matrices with variance $\sigma^2$ identified as the  \textit{noise}. The quantity $\mu/\sigma$ is known as \emph{signal-to-noise ratio}. The single matrix properties (i.e. by dropping the superscript $a$) of this class of problems have been introduced in Sec.~\ref{sec:rmt_intro_spiked}. \\

\begin{figure}[t!]
\centering
\includegraphics[width=0.49\textwidth]
{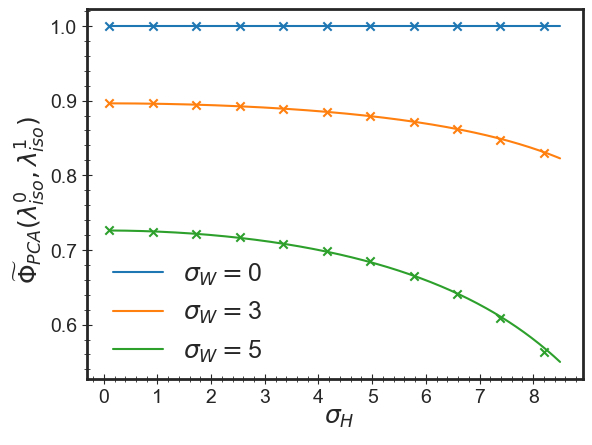}
\includegraphics[width=0.49\textwidth]{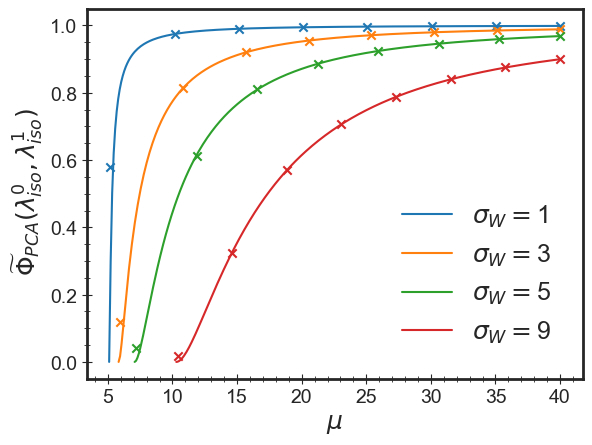}
\caption{
Overlap of the eigenvectors associated to the isolated eigenvalues of matrices subject to purely additive perturbations ($\Delta_h=\sigma_H$, $\Delta_{w,0}=\Delta_{w,1}=\sigma_W$) with $\mu=\mu_0=\mu_1$. The points are obtained from direct diagonalization of matrices of size $N=600$ averaged over $1500$ realization, while the continuous curves correspond to Eq.~\eqref{eq:PhiPCA}. As for Fig.\ref{fig:double_evals} we set $v_0=v_1=0$ given that the final results do not depend on them, to leading orders. {\it Left.}  Overlap as a function of $\sigma_H$,  for various $\sigma_W$ and $\mu_0=\mu_1=13$. Higher curves correspond to lower $\sigma_W$. {\it Right.} Overlap as a function of $\mu$, for  various $\sigma_W$ (from left to right the curves have lower $\sigma_W$) and $\sigma_H=5$. }
\label{fig:phi_iso_iso_pca_vs_sigmah_sigmaw}
\end{figure}

\noindent Consider now the case in which pairs of  spiked matrices ${\bf M}^{(a)}$ of the form \eqref{eq:addPert} are given, which differ from each others only by the  fluctuations in the noisy component ${\bf X}^{(a)}$ (thus $\mu_0=\mu=\mu_1$), the noise being correlated as described in Sec.~\ref{subsec:matrix}. Such pairs may correspond to measurements performed at different times, between which the noise has changed partially, without decorrelating completely with the previous configuration. At both times the estimator of the signal is given by the eigenvector ${\bf u}^{a}_{\rm iso}$ associated to the isolated eigenvalue of the spiked matrix. The correlation in the noisy components of the matrices implies that estimators ${\bf u}^a_{\rm iso}$ will have a non-trivial overlap with each other, which corresponds to $\tilde \Phi(\lambda^0_{\rm iso},\lambda^1_{\rm iso})$. This function then quantifies the typical similarity between the estimators ${\bf u}^{a}_{\rm iso}$ of the signal, obtained from different measurements of the signal corrupted by correlated noise. For the purely additive rank-1 perturbation considered here we have the following simplifications: 
\begin{align}
&\lambda_{\rm iso}= \mu + \frac{\sigma^2}{\mu}\quad\quad \mathfrak{q}_{\sigma,\Delta}(\lambda_{\rm iso},\mu)= 1- \frac{\sigma^2}{\mu^2}.
\end{align}
Then the overlap \eqref{eq:iso_iso_overlap} in this particular limit reduces to:
\begin{equation}\label{eq:PhiPCA}
\begin{split}
&\tilde \Phi_{\rm PCA}(\lambda_{\rm iso},\lambda_{\rm iso}):=\tonde{1-\frac{\sigma^2}{\mu^2}}^2 \, \bigg[\sigma_H^2\omega(\lambda_{\rm iso})
    +1\bigg]^2,
 \end{split}
\end{equation}
where
\begin{align}\label{eq:omega}
\omega(z):=\lim_{z\to\xi}\Psi(z,\xi)=\frac{\mathfrak{g}_\sigma(z)}{z+(\sigma_W^2-2\sigma^2)\mathfrak{g}_\sigma(z)}.
\end{align}

\noindent In Fig.~\ref{fig:phi_iso_iso_pca_vs_sigmah_sigmaw} we compare this expression with the overlaps obtained from the direct diagonalization of the random matrices, for different values of $\sigma_W$. As expected, at fixed $\sigma_H$ the overlap is equal to one in the case of fully correlated noise ($\sigma_W=0$), and decreases monotonically with the strength $\sigma_W$ of the uncorrelated part of the noise. At fixed $\sigma_W$, the overlap also decreases with increasing $\sigma_H$, as the relative contribution of the noise $\sigma=(\sigma_W^2 + \sigma_H^2)^{1/2}$
with respect to the signal $\mu$ increases. For $\sigma_H=0$, the noise in the two sets of measurements is uncorrelated and the overlap converges to the square of \eqref{eq:OvPCA}. As discussed in Sec. \ref{sec:theoretical_res2}, this corresponds to the fact that the estimators ${\bf u}^{a}_{\rm iso}$ are orthogonal in the subspace orthogonal to the signal direction ${\bf e}_N$.

\chapter*{Conclusive remarks}
\addcontentsline{toc}{chapter}{Conclusive remarks}
\markboth{}{}
The works presented in this thesis were motivated by problems of high-dimensional random landscapes, that arise in a variety of domains: from neural networks \cite{ChaosSompo88, mannelli_spiked_matrix_2019} to spin glasses \cite{ros2019complex, folena2020rethinking} and complex ecosystems \cite{MAY_1972, fyodorov2016nonlinear}, just to name a few. The physics of high-dimensional landscapes is a relatively new subject, dating back about 20/30 years \cite{cavagna1998stationary, Fyod2004}, and many questions remain open. Among these questions we were particularly interested in the "statics" versus "dynamics" debate. This question sparks from the interest to understand the system's dynamics (often solved via Dynamical mean-field theory for fully connected systems with many interacting units) 
from a landscape perspective. More precisely, one is interested to understand what is the influence of the landscape's topology and geometry on the dynamics of the system, often modeled as a gradient descent with, possibly, noise (Langevin dynamics). The prototypical model of energy landscape which we explained in Sec.~\ref{sec:p_spin_model}, i.e. the pure spherical $p$-spin model, is arguably the only high-dimensional complex system where a clear relation between the statics and the dynamics has been found \cite{Castellani_2005}. As we saw in Sec.~\ref{sec:mixed_models}, already a slight modification of such model breaks phenomena that were previously though to be general \cite{folena2020rethinking, Folena_gd_simul_2021}. Moreover, in the context of systems with non-reciprocal interactions (such as random neural networks and complex ecosystems), the theory is even more far behind, given that the system's dynamics cannot be seen as the optimization of an energy landscape. In Chapters \ref{chapter:non_reciprocal},\,\ref{chapter:scs} we were particularly interested in this question. We therefore decided to conduct a thorough analysis between the Kac-Rice complexity of equilibria and the (explicit) solution of the chaotic dynamics. We have seen that while some connections are present, in general, we cannot use one to infer the other. The problem of understanding the chaotic attractor in terms of a static calculation remains therefore open. Hopefully, our work provides the first steps and motivates further research to overcome this challenge. \\

\noindent Let us moreover underline that there are still open questions regarding the pure $p$-spin model. A question that has attracted a lot of attention recently is the one of activated dynamics \cite{pacco2024curvature, ros2021dynamical, ros2019complexity, rizzo2021path, folena_rare_2025}, namely to understand how the system restores ergodicity once initialized in a local minimum within the glassy region with exponentially many local minima. In Chapter~\ref{chapter:energy_landscapes} we tried to tackle this question by means of static approaches that leverage the Kac-Rice formalism. With the first approach we studied energetic pathways between fixed points, a problem that is ubiquitous in the context of machine learning \cite{draxler2018essentially, annesi2023star, mannelli_afraid_2019, garipov2018loss}. To compute such paths we also had to solve in Chapter~\ref{chapter:rmt_} a problem about overlaps of eigenvectors between spiked, correlated GOE random matrices. With the second approach, instead, we were able to understand the arrangement of triplets of stationary points in the glassy part of the landscape, and identify both signatures of strong correlations among them, as well as regions where memory of the initial points is forgotten.\\

\noindent One of the most challenging problems for the future of high-dimensional random landscapes is, arguably, the study of non-Gaussian landscapes. These arise especially in the context of deep learning, and are thought to have rather different phenomenology from simple glassy systems \cite{baity2018comparing}, mainly by the presence of flat regions in the bottom of the landscape. Although some directions have been taken \cite{kent_algo_marg_2024, maillard_landscape_2020,tsironic_perceptron_kac_2025,montanari_erm_localmin_2025}, a lot of work and challenges remain open.\\

\noindent Finally, I hope that this thesis and the works that it contains will open new doors, thus motivating further research to better explore high-dimensional landscapes and their fascinating mysteries.


\newgeometry{left=20mm, right=20mm, top=24mm, bottom=24mm}

\cleardoublepage
\addcontentsline{toc}{chapter}{Bibliography} 
\printbibliography

@article{MixedModelNieuwenhuizen95,
  title = {Exactly Solvable Model of a Quantum Spin Glass},
  author = {Nieuwenhuizen, Th. M.},
  journal = {Phys. Rev. Lett.},
  volume = {74},
  issue = {21},
  pages = {4289--4292},
  numpages = {0},
  year = {1995},
  month = {May},
  publisher = {American Physical Society},
  doi = {10.1103/PhysRevLett.74.4289},
  url = {https://link.aps.org/doi/10.1103/PhysRevLett.74.4289}
}

@article{Leuzzi_mixed_2006,
  title = {Spherical $2+p$ spin-glass model: An analytically solvable model with a glass-to-glass transition},
  author = {Crisanti, A. and Leuzzi, L.},
  journal = {Phys. Rev. B},
  volume = {73},
  issue = {1},
  pages = {014412},
  numpages = {20},
  year = {2006},
  month = {Jan},
  publisher = {American Physical Society},
  doi = {10.1103/PhysRevB.73.014412},
  url = {https://link.aps.org/doi/10.1103/PhysRevB.73.014412}
}

@article{Leuzzi_mixed_2004,
  title = {Spherical $2+p$ Spin-Glass Model: An Exactly Solvable Model for Glass to Spin-Glass Transition},
  author = {Crisanti, A. and Leuzzi, L.},
  journal = {Phys. Rev. Lett.},
  volume = {93},
  issue = {21},
  pages = {217203},
  numpages = {4},
  year = {2004},
  month = {Nov},
  publisher = {American Physical Society},
  doi = {10.1103/PhysRevLett.93.217203},
  url = {https://link.aps.org/doi/10.1103/PhysRevLett.93.217203}
}

@article{folena2020rethinking,
  title = {Rethinking Mean-Field Glassy Dynamics and Its Relation with the Energy Landscape: The Surprising Case of the Spherical Mixed $p$-Spin Model},
  author = {Folena, Giampaolo and Franz, Silvio and Ricci-Tersenghi, Federico},
  journal = {Phys. Rev. X},
  volume = {10},
  issue = {3},
  pages = {031045},
  numpages = {26},
  year = {2020},
  month = {Aug},
  publisher = {American Physical Society},
  doi = {10.1103/PhysRevX.10.031045},
  url = {https://link.aps.org/doi/10.1103/PhysRevX.10.031045}
}

@Article{folena_zamponi_weak_2023,
	title={{On weak ergodicity breaking in mean-field spin glasses}},
	author={Giampaolo Folena and Francesco Zamponi},
	journal={SciPost Phys.},
	volume={15},
	pages={109},
	year={2023},
	publisher={SciPost},
	doi={10.21468/SciPostPhys.15.3.109},
	url={https://scipost.org/10.21468/SciPostPhys.15.3.109},
}

@article{Folena_gd_simul_2021,
doi = {10.1088/1742-5468/abe29f},
url = {https://dx.doi.org/10.1088/1742-5468/abe29f},
year = {2021},
month = {mar},
publisher = {IOP Publishing and SISSA},
volume = {2021},
number = {3},
pages = {033302},
author = {Folena, Giampaolo and Franz, Silvio and Ricci-Tersenghi, Federico},
title = {Gradient descent dynamics in the mixed p-spin spherical model: finite-size simulations and comparison with mean-field integration},
journal = {Journal of Statistical Mechanics: Theory and Experiment},
}

@misc{tersenghi_seb_mixed_2025,
      title={Strong ergodicity breaking in dynamical mean-field equations for mixed p-spin glasses}, 
      author={Vincenzo Citro and Federico Ricci-Tersenghi},
      year={2025},
      eprint={2504.12367},
      archivePrefix={arXiv},
      primaryClass={cond-mat.dis-nn},
      url={https://arxiv.org/abs/2504.12367}, 
}

@article{Bongsoo_Kim_num_2001,
doi = {10.1209/epl/i2001-00202-4},
url = {https://dx.doi.org/10.1209/epl/i2001-00202-4},
year = {2001},
month = {mar},
publisher = {},
volume = {53},
number = {5},
pages = {660},
author = {Bongsoo Kim and A. Latz},
title = {The dynamics of the spherical p-spin model:  
From microscopic to asymptotic},
journal = {Europhysics Letters}
}

@article{kraza_following_2012,
doi = {10.1088/1742-5468/2012/07/P07002},
url = {https://dx.doi.org/10.1088/1742-5468/2012/07/P07002},
year = {2012},
month = {jul},
publisher = {IOP Publishing and SISSA},
volume = {2012},
number = {07},
pages = {P07002},
author = {Sun, YiFan and Crisanti, Andrea and Krzakala, Florent and Leuzzi, Luca and Zdeborová, Lenka},
title = {Following states in temperature in the spherical s + p-spin glass model},
journal = {Journal of Statistical Mechanics: Theory and Experiment}
}

@article{kent2024arrangement,
	title={{Arrangement of nearby minima and saddles in the mixed spherical energy landscapes}},
	author={Jaron Kent-Dobias},
	journal={SciPost Phys.},
	volume={16},
	pages={001},
	year={2024},
	publisher={SciPost},
	doi={10.21468/SciPostPhys.16.1.001},
	url={https://scipost.org/10.21468/SciPostPhys.16.1.001},
}

@article{Kent_Dobias_typical_2023,
doi = {10.1209/0295-5075/acf521},
url = {https://dx.doi.org/10.1209/0295-5075/acf521},
year = {2023},
month = {sep},
publisher = {EDP Sciences, IOP Publishing and Società Italiana di Fisica},
volume = {143},
number = {6},
pages = {61003},
author = {Kent-Dobias, Jaron},
title = {When is the average number of saddle points typical?},
journal = {Europhysics Letters},
}

@article{kent_rsb_2023,
  title = {How to count in hierarchical landscapes: A full solution to mean-field complexity},
  author = {Kent-Dobias, Jaron and Kurchan, Jorge},
  journal = {Phys. Rev. E},
  volume = {107},
  issue = {6},
  pages = {064111},
  numpages = {18},
  year = {2023},
  month = {Jun},
  publisher = {American Physical Society},
  doi = {10.1103/PhysRevE.107.064111},
  url = {https://link.aps.org/doi/10.1103/PhysRevE.107.064111}
}

@article{kent_complex_complex_2021,
  title = {Complex complex landscapes},
  author = {Kent-Dobias, Jaron and Kurchan, Jorge},
  journal = {Phys. Rev. Res.},
  volume = {3},
  issue = {2},
  pages = {023064},
  numpages = {5},
  year = {2021},
  month = {Apr},
  publisher = {American Physical Society},
  doi = {10.1103/PhysRevResearch.3.023064},
  url = {https://link.aps.org/doi/10.1103/PhysRevResearch.3.023064}
}

@article{franz1995recipes,
	author = {{Silvio Franz} and {Giorgio Parisi}},
	title = {Recipes for Metastable States in Spin Glasses},
	doi= "10.1051/jp1:1995201",
	url= "https://doi.org/10.1051/jp1:1995201",
	journal = {J. Phys. I France},
	year = 1995,
	volume = 5,
	number = 11,
	pages = "1401-1415",
	month = "",
publisher={EDP Sciences}
}

@article{cavagna1997structure,
doi = {10.1088/0305-4470/30/13/004},
url = {https://doi.org/10.1088/0305-4470/30/13/004},
year = {1997},
month = {jul},
publisher = {},
volume = {30},
number = {13},
pages = {4449},
author = {Andrea Cavagna and Irene Giardina and Giorgio Parisi},
title = {Structure of metastable states in spin glasses by means of a three replica potential},
journal = {Journal of Physics A: Mathematical and General}
}

@article{cavagna1997investigation,
doi = {10.1088/0305-4470/30/20/009},
url = {https://doi.org/10.1088/0305-4470/30/20/009},
year = {1997},
month = {oct},
publisher = {},
volume = {30},
number = {20},
pages = {7021},
author = {Andrea Cavagna and Irene Giardina and Giorgio Parisi},
title = {An investigation of the hidden structure of states in a mean-field spin-glass model},
journal = {Journal of Physics A: Mathematical and General}
}

@article{Georges_Yedidia_1991, title={How to expand around mean-field theory using high-temperature expansions}, volume={24}, DOI={10.1088/0305-4470/24/9/024}, number={9}, journal={Journal of Physics A: Mathematical and General}, publisher={IOP Publishing}, author={Georges, A and Yedidia, J S}, year={1991}, month={may}, pages={2173–2192} }

@article{Biroli_dynTAP_1999,
doi = {10.1088/0305-4470/32/48/301},
url = {https://dx.doi.org/10.1088/0305-4470/32/48/301},
year = {1999},
month = {dec},
publisher = {},
volume = {32},
number = {48},
pages = {8365},
author = {Giulio Biroli},
title = {Dynamical 
TAP approach to mean field glassy systems},
journal = {Journal of Physics A: Mathematical and General},
}

@article{cavagna1998stationary,
  title = {Stationary points of the Thouless-Anderson-Palmer free energy},
  author = {Cavagna, Andrea and Giardina, Irene and Parisi, Giorgio},
  journal = {Phys. Rev. B},
  volume = {57},
  issue = {18},
  pages = {11251--11257},
  numpages = {0},
  year = {1998},
  month = {May},
  publisher = {American Physical Society},
  doi = {10.1103/PhysRevB.57.11251},
  url = {https://link.aps.org/doi/10.1103/PhysRevB.57.11251}
}

@article{ros2019complex,
  title = {Complex Energy Landscapes in Spiked-Tensor and Simple Glassy Models: Ruggedness, Arrangements of Local Minima, and Phase Transitions},
  author = {Ros, Valentina and Ben Arous, Gerard and Biroli, Giulio and Cammarota, Chiara},
  journal = {Phys. Rev. X},
  volume = {9},
  issue = {1},
  pages = {011003},
  numpages = {42},
  year = {2019},
  month = {Jan},
  publisher = {American Physical Society},
  doi = {10.1103/PhysRevX.9.011003},
  url = {https://link.aps.org/doi/10.1103/PhysRevX.9.011003}
}

@article{ghimenti_accelerating_2022,
  title = {Accelerating, to some extent, the $p$-spin dynamics},
  author = {Ghimenti, Federico and van Wijland, Fr\'ed\'eric},
  journal = {Phys. Rev. E},
  volume = {105},
  issue = {5},
  pages = {054137},
  numpages = {15},
  year = {2022},
  month = {May},
  publisher = {American Physical Society},
  doi = {10.1103/PhysRevE.105.054137},
  url = {https://link.aps.org/doi/10.1103/PhysRevE.105.054137}
}

@article{kac_spherical_ferro_52,
  title = {The Spherical Model of a Ferromagnet},
  author = {Berlin, T. H. and Kac, M.},
  journal = {Phys. Rev.},
  volume = {86},
  issue = {6},
  pages = {821--835},
  numpages = {0},
  year = {1952},
  month = {Jun},
  publisher = {American Physical Society},
  doi = {10.1103/PhysRev.86.821},
  url = {https://link.aps.org/doi/10.1103/PhysRev.86.821}
}

@article{Rem_1,
  title = {Random-Energy Model: Limit of a Family of Disordered Models},
  author = {Derrida, B.},
  journal = {Phys. Rev. Lett.},
  volume = {45},
  issue = {2},
  pages = {79--82},
  numpages = {0},
  year = {1980},
  month = {Jul},
  publisher = {American Physical Society},
  doi = {10.1103/PhysRevLett.45.79},
  url = {https://link.aps.org/doi/10.1103/PhysRevLett.45.79}
}

@article{Sompolinsky_zippelius_ed_82,
  title = {Relaxational dynamics of the Edwards-Anderson model and the mean-field theory of spin-glasses},
  author = {Sompolinsky, H. and Zippelius, Annette},
  journal = {Phys. Rev. B},
  volume = {25},
  issue = {11},
  pages = {6860--6875},
  numpages = {0},
  year = {1982},
  month = {Jun},
  publisher = {American Physical Society},
  doi = {10.1103/PhysRevB.25.6860},
  url = {https://link.aps.org/doi/10.1103/PhysRevB.25.6860}
}

@article{dominicis_dmft_78,
  title = {Field-theory renormalization and critical dynamics above ${T}_{c}$: Helium, antiferromagnets, and liquid-gas systems},
  author = {De Dominicis, C. and Peliti, L.},
  journal = {Phys. Rev. B},
  volume = {18},
  issue = {1},
  pages = {353--376},
  numpages = {0},
  year = {1978},
  month = {Jul},
  publisher = {American Physical Society},
  doi = {10.1103/PhysRevB.18.353},
  url = {https://link.aps.org/doi/10.1103/PhysRevB.18.353}
}

@article{sompo_zipp_dyn_phase_1981,
  title = {Dynamic Theory of the Spin-Glass Phase},
  author = {Sompolinsky, H. and Zippelius, Annette},
  journal = {Phys. Rev. Lett.},
  volume = {47},
  issue = {5},
  pages = {359--362},
  numpages = {0},
  year = {1981},
  month = {Aug},
  publisher = {American Physical Society},
  doi = {10.1103/PhysRevLett.47.359},
  url = {https://link.aps.org/doi/10.1103/PhysRevLett.47.359}
}

@article{Castellani_2005,
doi = {10.1088/1742-5468/2005/05/P05012},
url = {https://dx.doi.org/10.1088/1742-5468/2005/05/P05012},
year = {2005},
month = {may},
publisher = {},
volume = {2005},
number = {05},
pages = {P05012},
author = {Castellani, Tommaso and Cavagna, Andrea},
title = {Spin-glass theory for pedestrians},
journal = {Journal of Statistical Mechanics: Theory and Experiment},
}

@article{Folena_notes_2023, title={Introduction to the dynamics of disordered systems: Equilibrium and gradient descent}, volume={631}, DOI={10.1016/j.physa.2022.128152}, journal={Physica A: Statistical Mechanics and its Applications}, publisher={Elsevier BV}, author={Folena, Giampaolo and Manacorda, Alessandro and Zamponi, Francesco}, year={2023}, month={dec.}, pages={128152}}

@misc{zamponi_mean_field_notes_2014,
      title={Mean field theory of spin glasses}, 
      author={Francesco Zamponi},
      year={2014},
      eprint={1008.4844},
      archivePrefix={arXiv},
      primaryClass={cond-mat.stat-mech},
      url={https://arxiv.org/abs/1008.4844}, 
}

@book{DeDominicis_Giardina_book_2006, place={Cambridge}, title={Random Fields and Spin Glasses: A Field Theory Approach}, publisher={Cambridge University Press}, author={De Dominicis, Cirano and Giardina, Irene}, year={2006}, doi={
https://doi.org/10.1017/CBO9780511534836}}

@misc{kurchan_six_2009,
      title={Six out of equilibrium lectures}, 
      author={Jorge Kurchan},
      year={2009},
      eprint={0901.1271},
      archivePrefix={arXiv},
      primaryClass={cond-mat.stat-mech},
      url={https://arxiv.org/abs/0901.1271}, 
}

@book{spin_glass_beyond_86,
author = {Mezard, M and Parisi, G and Virasoro, M},
title = {Spin Glass Theory and Beyond},
publisher = {WORLD SCIENTIFIC},
year = {1986},
doi = {10.1142/0271},
address = {},
edition   = {},
URL = {https://www.worldscientific.com/doi/abs/10.1142/0271},
eprint = {https://www.worldscientific.com/doi/pdf/10.1142/0271}
}

@PHDTHESIS{folena_these_2020,
url = "http://www.theses.fr/2020UPASS060",
title = "The mixed p-spin model : selecting, following and losing states",
author = "Folena, Giampaolo",
year = "2020",
note = "Thèse de doctorat dirigée par Franz, Silvio et Ricci-Tersenghi, Federico Physique université Paris-Saclay 2020",
note = "Thèse de doctorat Physique Università degli studi La Sapienza (Rome) 2020",
note = "2020UPASS060",
url = "http://www.theses.fr/2020UPASS060/document",
}

@article{biroli_berthier_review_2010,
  title = {Theoretical perspective on the glass transition and amorphous materials},
  author = {Berthier, Ludovic and Biroli, Giulio},
  journal = {Rev. Mod. Phys.},
  volume = {83},
  issue = {2},
  pages = {587--645},
  numpages = {0},
  year = {2011},
  month = {Jun},
  publisher = {American Physical Society},
  doi = {10.1103/RevModPhys.83.587},
  url = {https://link.aps.org/doi/10.1103/RevModPhys.83.587}
}

@article{urbani_notes_2017,
  title = {Statistical physics of glassy systems: tools and applications},
  author = {Urbani, Pierfrancesco},
  journal = {IPHT Lecture Notes},
  year = {2017},
  url = {https://courses.ipht.fr/?q=en/node/194}
}

@misc{barra_pspin_97,
      title={The p-spin spherical spin glass model}, 
      author={A. Barrat},
      year={1997},
      eprint={cond-mat/9701031},
      archivePrefix={arXiv},
      primaryClass={cond-mat.dis-nn},
      url={https://arxiv.org/abs/cond-mat/9701031}, 
}

@misc{charobonneau_history_replica_2022,
      title={From the replica trick to the replica symmetry breaking technique}, 
      author={Patrick Charbonneau},
      year={2022},
      eprint={2211.01802},
      archivePrefix={arXiv},
      primaryClass={physics.hist-ph},
      url={https://arxiv.org/abs/2211.01802}, 
}

@article{Talagrand_pspin_2005, title={Free energy of the spherical mean field model}, volume={134}, DOI={10.1007/s00440-005-0433-8}, number={3}, journal={Probability Theory and Related Fields}, publisher={Springer Science and Business Media LLC}, author={Talagrand, Michel}, year={2005}, month={may}, pages={339–382} }

@article{crisanti1992sphericalp,
  title={The spherical p-spin interaction spin glass model: the statics},
  author={Crisanti, Andrea and Sommers, H-J},
  journal={Zeitschrift f{\"u}r Physik B Condensed Matter},
  volume={87},
  number={3},
  pages={341--354},
  year={1992},
  publisher={Springer},
doi = "10.1007/BF01309287",
}

@article{crisanti1993spherical,
  title={The spherical p-spin interaction spin-glass model: the dynamics},
  author={Crisanti, Andrea and Horner, Heinz and Sommers, H -J},
  journal={Zeitschrift f{\"u}r Physik B Condensed Matter},
  volume={92},
  pages={257--271},
  year={1993},
  publisher={Springer},
doi={https://doi.org/10.1007/BF01312184}
}

@article{gross1984simplest,
title = {The simplest spin glass},
journal = {Nuclear Physics B},
volume = {240},
number = {4},
pages = {431-452},
year = {1984},
issn = {0550-3213},
doi = {https://doi.org/10.1016/0550-3213(84)90237-2},
url = {https://www.sciencedirect.com/science/article/pii/0550321384902372},
author = {D.J. Gross and M. Mezard}
}

@article{Thirum_pspin_87,
  title = {Dynamics of the Structural Glass Transition and the $p$-Spin---Interaction Spin-Glass Model},
  author = {Kirkpatrick, T. R. and Thirumalai, D.},
  journal = {Phys. Rev. Lett.},
  volume = {58},
  issue = {20},
  pages = {2091--2094},
  numpages = {0},
  year = {1987},
  month = {May},
  publisher = {American Physical Society},
  doi = {10.1103/PhysRevLett.58.2091},
  url = {https://link.aps.org/doi/10.1103/PhysRevLett.58.2091}
}

@article{Thirum_pspin_87_part2,
  title = {p-spin-interaction spin-glass models: Connections with the structural glass problem},
  author = {Kirkpatrick, T. R. and Thirumalai, D.},
  journal = {Phys. Rev. B},
  volume = {36},
  issue = {10},
  pages = {5388--5397},
  numpages = {0},
  year = {1987},
  month = {Oct},
  publisher = {American Physical Society},
  doi = {10.1103/PhysRevB.36.5388},
  url = {https://link.aps.org/doi/10.1103/PhysRevB.36.5388}
}

@article{Crisanti_TAP_pspin_95,
	author = {{A. Crisanti} and {H.-J. Sommers}},
	title = {Thouless-Anderson-Palmer Approach to the Spherical p-Spin Spin Glass Model},
	DOI= "10.1051/jp1:1995164",
	url= "https://doi.org/10.1051/jp1:1995164",
	journal = {J. Phys. I France},
	year = 1995,
	volume = 5,
	number = 7,
	pages = "805-813",
	month = "",
}

@article{Rieger_TAP_92,
  title = {The number of solutions of the Thouless-Anderson-Palmer equations for p-spin-interaction spin glasses},
  author = {Rieger, H.},
  journal = {Phys. Rev. B},
  volume = {46},
  issue = {22},
  pages = {14655--14661},
  numpages = {0},
  year = {1992},
  month = {Dec},
  publisher = {American Physical Society},
  doi = {10.1103/PhysRevB.46.14655},
  url = {https://link.aps.org/doi/10.1103/PhysRevB.46.14655}
}

@article{kurchan_barriers_93,
	author = {{J. Kurchan} and {G.Parisi} and {M.A. Virasoro}},
	title = {Barriers and metastable states as saddle points in the replica approach},
	doi= "10.1051/jp1:1993217",
	url= "https://doi.org/10.1051/jp1:1993217",
	journal = {J. Phys. I France},
	year = 1993,
	volume = 3,
	number = 8,
	pages = "1819-1838",
	month = "",
}

@article{franz1998effective,
title = {Effective potential in glassy systems: theory and simulations},
journal = {Physica A: Statistical Mechanics and its Applications},
volume = {261},
number = {3},
pages = {317-339},
year = {1998},
issn = {0378-4371},
publisher={Elsevier},
doi = {https://doi.org/10.1016/S0378-4371(98)00315-X},
url = {https://www.sciencedirect.com/science/article/pii/S037843719800315X},
author = {Silvio Franz and Giorgio Parisi}
}

@article{ros2021dynamical,
	title={{Dynamical instantons and activated processes in mean-field glass models}},
	author={Valentina Ros and Giulio Biroli and Chiara Cammarota},
	journal={SciPost Phys.},
	volume={10},
	pages={002},
	year={2021},
	publisher={SciPost},
	doi={10.21468/SciPostPhys.10.1.002},
	url={https://scipost.org/10.21468/SciPostPhys.10.1.002},
}

@article{cavagna1999quenched,
doi = {10.1088/0305-4470/32/5/004},
url = {https://doi.org/10.1088/0305-4470/32/5/004},
year = {1999},
month = {feb},
publisher = {},
volume = {32},
number = {5},
pages = {711},
author = {Andrea Cavagna and Juan P Garrahan and Irene Giardina},
title = {Quenched complexity of the mean-field  p-spin spherical model with external magnetic field},
journal = {Journal of Physics A: Mathematical and General}
}

@article{cugliandolo1993analytical,
  title = {Analytical solution of the off-equilibrium dynamics of a long-range spin-glass model},
  author = {Cugliandolo, L. F. and Kurchan, J.},
  journal = {Phys. Rev. Lett.},
  volume = {71},
  issue = {1},
  pages = {173--176},
  numpages = {0},
  year = {1993},
  month = {Jul},
  publisher = {American Physical Society},
  doi = {10.1103/PhysRevLett.71.173},
  url = {https://link.aps.org/doi/10.1103/PhysRevLett.71.173}
}

@article{Kurchan_phase_space_1996,
doi = {10.1088/0305-4470/29/9/009},
url = {https://dx.doi.org/10.1088/0305-4470/29/9/009},
year = {1996},
month = {may},
publisher = {},
volume = {29},
number = {9},
pages = {1929},
author = {Jorge Kurchan and Laurent Laloux},
title = {Phase space geometry and slow dynamics},
journal = {Journal of Physics A: Mathematical and General}
}

@InProceedings{cugliandolo2002dynamics,
author="Cugliandolo, L. F.",
editor="Barrat, Jean-Louis
and Feigelman, Mikhail
and Kurchan, Jorge
and Dalibard, Jean",
title="Course 7: Dynamics of Glassy Systems",
booktitle="Slow Relaxations and nonequilibrium dynamics in condensed matter",
year="2003",
publisher="Springer Berlin Heidelberg",
address="Berlin, Heidelberg",
pages="367--521",
doi={https://doi.org/10.1007/978-3-540-44835-8_7}
}

@inbook{bouchaud1998out,
author = {Jean-Philippe Bouchaud and Leticia F. Cugliandolo and Jorge Kurchan and Marc Mézard},
title = {Out of equilibrium dynamics in spin-glasses and other glassy systems},
booktitle = {Spin Glasses and Random Fields},
chapter = {},
pages = {161-223},
doi = {10.1142/9789812819437_0006},
publisher={World Scientific Singapore},
URL = {https://www.worldscientific.com/doi/abs/10.1142/9789812819437_0006}
}

@article{leuzzi2003complexity,
  title={The complexity of the spherical p-spin spin glass model, revisited.},
  author={Leuzzi, L and Crisanti, A and Rizzo, T},
  journal={European Physical Journal B: Condensed Matter},
  volume={36},
  number={1},
  year={2003},
doi={https://doi.org/10.1140/epjb/e2003-00325-x}
}

@article{Cugliandolo_1995_weak,
author = {L. F. Cugliandolo and J. Kurchan and},
title = {Weak ergodicity breaking in mean-field spin-glass models},
journal = {Philosophical Magazine B},
volume = {71},
number = {4},
pages = {501--514},
year = {1995},
publisher = {Taylor \& Francis},
doi = {10.1080/01418639508238541},
URL = {https://doi.org/10.1080/01418639508238541}
}

@article{rizzo2021path,
  title = {Path integral approach unveils role of complex energy landscape for activated dynamics of glassy systems},
  author = {Rizzo, Tommaso},
  journal = {Phys. Rev. B},
  volume = {104},
  issue = {9},
  pages = {094203},
  numpages = {26},
  year = {2021},
  month = {Sep},
  publisher = {American Physical Society},
  doi = {10.1103/PhysRevB.104.094203},
  url = {https://link.aps.org/doi/10.1103/PhysRevB.104.094203}
}

@misc{folena_rare_2025,
      title={Rare Trajectories in a Prototypical Mean-field Disordered Model: Insights into Landscape and Instantons}, 
      author={Patrick Charbonneau and Giampaolo Folena and Enrico M. Malatesta and Tommaso Rizzo and Francesco Zamponi},
      year={2025},
      eprint={2505.00107},
      archivePrefix={arXiv},
      primaryClass={cond-mat.dis-nn},
      url={https://arxiv.org/abs/2505.00107}, 
}

@article{lopatin1999instantons,
  title={Instantons in the Langevin dynamics: An application to spin glasses},
  author={Lopatin, AV and Ioffe, LB},
  journal={Physical Review B},
  volume={60},
  number={9},
  pages={6412},
  year={1999},
  publisher={APS},
doi={https://doi.org/10.1103/PhysRevB.60.6412}
}

@article{lopatin_barriers_2000,
  title = {Barriers in the $\mathit{p}$-Spin Interacting Spin-Glass Model: The Dynamical Approach},
  author = {Lopatin, A. V. and Ioffe, L. B.},
  journal = {Phys. Rev. Lett.},
  volume = {84},
  issue = {18},
  pages = {4208--4211},
  numpages = {0},
  year = {2000},
  month = {May},
  publisher = {American Physical Society},
  doi = {10.1103/PhysRevLett.84.4208},
  url = {https://link.aps.org/doi/10.1103/PhysRevLett.84.4208}
}

@article{krzakala_planting_2009,
  title = {Hiding Quiet Solutions in Random Constraint Satisfaction Problems},
  author = {Krzakala, Florent and Zdeborov\'a, Lenka},
  journal = {Phys. Rev. Lett.},
  volume = {102},
  issue = {23},
  pages = {238701},
  numpages = {4},
  year = {2009},
  month = {Jun},
  publisher = {American Physical Society},
  doi = {10.1103/PhysRevLett.102.238701},
  url = {https://link.aps.org/doi/10.1103/PhysRevLett.102.238701}
}

@article{Cavagna_saddles_2001,
doi = {10.1088/0305-4470/34/26/302},
url = {https://dx.doi.org/10.1088/0305-4470/34/26/302},
year = {2001},
month = {jun},
publisher = {},
volume = {34},
number = {26},
pages = {5317},
author = {Andrea Cavagna and Irene Giardina and Giorgio Parisi},
title = {Role of saddles in mean-field dynamics above the
glass transition},
journal = {Journal of Physics A: Mathematical and General}
}

@article{Barrat_bifurcation_95, title={Temperature evolution and bifurcations of metastable states in mean-field spin glasses, with connections with structural glasses}, volume={30}, DOI={10.1088/0305-4470/30/16/006}, number={16}, journal={Journal of Physics A: Mathematical and General}, publisher={IOP Publishing}, author={Barrat, Alain and Franz, Silvio and Parisi, Giorgio}, year={1997}, month={aug.}, pages={5593–5612} }

@article{barrat_dynmeta_96,
doi = {10.1088/0305-4470/29/5/001},
url = {https://dx.doi.org/10.1088/0305-4470/29/5/001},
year = {1996},
month = {mar},
publisher = {},
volume = {29},
number = {5},
pages = {L81},
author = {A Barrat and R Burioni and M Mézard},
title = {Dynamics within metastable states in a mean-field spin glass},
journal = {Journal of Physics A: Mathematical and General}
}

@article{Franz_coupled_97,
  title = {Phase Diagram of Coupled Glassy Systems: A Mean-Field Study},
  author = {Franz, Silvio and Parisi, Giorgio},
  journal = {Phys. Rev. Lett.},
  volume = {79},
  issue = {13},
  pages = {2486--2489},
  numpages = {0},
  year = {1997},
  month = {Sep},
  publisher = {American Physical Society},
  doi = {10.1103/PhysRevLett.79.2486},
  url = {https://link.aps.org/doi/10.1103/PhysRevLett.79.2486}
}

@article{benarous_mixed_2019, title={Geometry and Temperature Chaos in Mixed Spherical Spin Glasses at Low Temperature: The Perturbative Regime}, volume={73}, DOI={10.1002/cpa.21875}, number={8}, journal={Communications on Pure and Applied Mathematics}, publisher={Wiley}, author={Arous, Gérard Ben and Subag, Eliran and Zeitouni, Ofer}, year={2019}, month={nov.}, pages={1732–1828}}

@article{Auffinger_Chen_mixed_landscape_2018, title={On the energy landscape of spherical spin glasses}, volume={330}, DOI={10.1016/j.aim.2018.03.028}, journal={Advances in Mathematics}, publisher={Elsevier BV}, author={Auffinger, Antonio and Chen, Wei-Kuo}, year={2018}, month={may}, pages={553–588} }

@article{subag2017complexity,
  title={The complexity of spherical $ p $-spin models—A second moment approach},
  author={Subag, Eliran},
  journal={The Annals of Probability},
  volume={45},
  number={5},
  pages={3385--3450},
  year={2017},
  publisher={Institute of Mathematical Statistics},
  doi={10.1214/16-AOP1139}
}

@article{auffinger_complexity_2013,
author = {Antonio Auffinger and Gerard Ben Arous},
title = {{Complexity of random smooth functions on the high-dimensional sphere}},
volume = {41},
journal = {The Annals of Probability},
number = {6},
publisher = {Institute of Mathematical Statistics},
pages = {4214 -- 4247},
year = {2013},
doi = {10.1214/13-AOP862},
URL = {https://doi.org/10.1214/13-AOP862}
}

@article{Auffinger_Arous_Černý_spin_2012, title={Random Matrices and Complexity of Spin Glasses}, volume={66}, DOI={10.1002/cpa.21422}, number={2}, journal={Communications on Pure and Applied Mathematics}, publisher={Wiley}, author={Auffinger, Antonio and Arous, Gérard Ben and Černý, Jiří}, year={2012}, month={sep.}, pages={165–201} }

@misc{Auffinger_saddles_2020,
      title={The number of saddles of the spherical $p$-spin model}, 
      author={Antonio Auffinger and Julian Gold},
      year={2020},
      eprint={2007.09269},
      archivePrefix={arXiv},
      primaryClass={math.PR},
      url={https://arxiv.org/abs/2007.09269}, 
}

@article{Subag_tap_2023, title={The free energy of spherical pure p-spin models: computation from the TAP approach}, volume={186}, DOI={10.1007/s00440-023-01200-0}, number={3-4}, journal={Probability Theory and Related Fields}, publisher={Springer Science and Business Media LLC}, author={Subag, Eliran}, year={2023}, month={may}, pages={715–734} }

@article{Subag_Zeitouni_2021, title={Concentration of the complexity of spherical pure p-spin models at arbitrary energies}, volume={62}, DOI={10.1063/5.0070582}, number={12}, journal={Journal of Mathematical Physics}, publisher={AIP Publishing}, author={Subag, Eliran and Zeitouni, Ofer}, year={2021}, month={dec.}}

@article{benarous_detemrinants_2022,
title = "Exponential growth of random detemrinants beyond invariance",
author = "Gérard Ben Arous and Paul Bourgade and Benjamin McKenna",
note = "Publisher Copyright: {\textcopyright} 2022, Mathematical Sciences Publishers. All rights reserved.",
year = "2022",
doi = "10.2140/pmp.2022.3.731",
language = "English (US)",
volume = "3",
pages = "731--789",
journal = "Probability and Mathematical Physics",
issn = "2690-0998",
publisher = "Mathematical Sciences Publishers",
number = "4",
}

@article{benarous_mfd_2024,
title = "Landscape complexity beyond invariance and the elastic manifold",
author = "Gérard Ben Arous and Paul Bourgade and Benjamin McKenna",
note = "Publisher Copyright: {\textcopyright} 2023 Wiley Periodicals LLC.",
year = "2024",
month = feb,
doi = "10.1002/cpa.22146",
language = "English (US)",
volume = "77",
pages = "1302--1352",
journal = "Communications on Pure and Applied Mathematics",
issn = "0010-3640",
publisher = "Wiley-Liss Inc.",
number = "2",
}

@article{SK75,
  title = {Solvable Model of a Spin-Glass},
  author = {Sherrington, David and Kirkpatrick, Scott},
  journal = {Phys. Rev. Lett.},
  volume = {35},
  issue = {26},
  pages = {1792--1796},
  numpages = {0},
  year = {1975},
  month = {Dec},
  publisher = {American Physical Society},
  doi = {10.1103/PhysRevLett.35.1792},
  url = {https://link.aps.org/doi/10.1103/PhysRevLett.35.1792}
}

@article{bray1980metastable,
doi = {10.1088/0022-3719/13/19/002},
url = {https://doi.org/10.1088/0022-3719/13/19/002},
year = {1980},
month = {jul},
publisher = {},
volume = {13},
number = {19},
pages = {L469},
author = {A J Bray and M A Moore},
title = {Metastable states in spin glasses},
journal = {Journal of Physics C: Solid State Physics}
}

@article{TAP_77,
author = {D. J. Thouless and P. W. Anderson and R. G. Palmer and},
title = {Solution of 'Solvable model of a spin glass'},
journal = {The Philosophical Magazine: A Journal of Theoretical Experimental and Applied Physics},
volume = {35},
number = {3},
pages = {593--601},
year = {1977},
publisher = {Taylor \& Francis},
doi = {10.1080/14786437708235992},
URL = {https://doi.org/10.1080/14786437708235992}
}

@article{parisi_order_para_spin_83,
  title = {Order Parameter for Spin-Glasses},
  author = {Parisi, Giorgio},
  journal = {Phys. Rev. Lett.},
  volume = {50},
  issue = {24},
  pages = {1946--1948},
  numpages = {0},
  year = {1983},
  month = {Jun},
  publisher = {American Physical Society},
  doi = {10.1103/PhysRevLett.50.1946},
  url = {https://link.aps.org/doi/10.1103/PhysRevLett.50.1946}
}

@article{Parisi_sk_solution_1980,
doi = {10.1088/0305-4470/13/4/009},
url = {https://dx.doi.org/10.1088/0305-4470/13/4/009},
year = {1980},
month = {apr},
publisher = {},
volume = {13},
number = {4},
pages = {L115},
author = {G Parisi},
title = {A sequence of approximated solutions to the S-K model for spin glasses},
journal = {Journal of Physics A: Mathematical and General}
}

@article{tersenghi_seb_sk_2020,
author = {Massimo Bernaschi  and Alain Billoire  and Andrea Maiorano  and Giorgio Parisi  and Federico Ricci-Tersenghi },
title = {Strong ergodicity breaking in aging of mean-field spin glasses},
journal = {Proceedings of the National Academy of Sciences},
volume = {117},
number = {30},
pages = {17522-17527},
year = {2020},
doi = {10.1073/pnas.1910936117},
URL = {https://www.pnas.org/doi/abs/10.1073/pnas.1910936117},
eprint = {https://www.pnas.org/doi/pdf/10.1073/pnas.1910936117}}

@article{bouchaud1992weak,
	author = {{J. P. Bouchaud}},
	title = {Weak ergodicity breaking and aging in disordered systems},
	DOI= "10.1051/jp1:1992238",
	url= "https://doi.org/10.1051/jp1:1992238",
	journal = {J. Phys. I France},
	year = 1992,
	volume = 2,
	number = 9,
	pages = "1705-1713",
	month = "",
}

@article{dyre1987master,
  title = {Master-equation appoach to the glass transition},
  author = {Dyre, Jeppe C.},
  journal = {Phys. Rev. Lett.},
  volume = {58},
  issue = {8},
  pages = {792--795},
  numpages = {0},
  year = {1987},
  month = {Feb},
  publisher = {American Physical Society},
  doi = {10.1103/PhysRevLett.58.792},
  url = {https://link.aps.org/doi/10.1103/PhysRevLett.58.792}
}

@article{monthus1996models,
doi = {10.1088/0305-4470/29/14/012},
url = {https://doi.org/10.1088/0305-4470/29/14/012},
year = {1996},
month = {jul},
publisher = {},
volume = {29},
number = {14},
pages = {3847},
author = {Cécile Monthus and Jean-Philippe Bouchaud},
title = {Models of traps and glass phenomenology},
journal = {Journal of Physics A: Mathematical and General}
}

@article{bouchaud1995aging,
	author = {{J.-P. Bouchaud} and {D.S. Dean}},
	title = {Aging on Parisi's Tree},
	DOI= "10.1051/jp1:1995127",
	url= "https://doi.org/10.1051/jp1:1995127",
	journal = {J. Phys. I France},
	year = 1995,
	volume = 5,
	number = 3,
	pages = "265-286",
	month = "",
}

@article{stariolo2019activated,
  title={Activated dynamics of the Ising p-spin disordered model with finite number of variables},
  author={Stariolo, Daniel A and Cugliandolo, Leticia F},
  journal={Europhysics Letters},
  volume={127},
  number={1},
  pages={16002},
  year={2019},
  publisher={IOP Publishing},
doi={10.1209/0295-5075/127/16002}
}

@article{stariolo2020barriers,
  title={Barriers, trapping times, and overlaps between local minima in the dynamics of the disordered Ising p-spin model},
  author={Stariolo, Daniel A and Cugliandolo, Leticia F},
  journal={Physical Review E},
  volume={102},
  number={2},
  pages={022126},
  year={2020},
  publisher={APS},
doi={https://doi.org/10.1103/PhysRevE.102.022126}
}

@article{margiotta2018spectral,
  title={Spectral properties of the trap model on sparse networks},
  author={Margiotta, Riccardo Giuseppe and K{\"u}hn, Reimer and Sollich, Peter},
  journal={Journal of Physics A: Mathematical and Theoretical},
  volume={51},
  number={29},
  pages={294001},
  year={2018},
  publisher={IOP Publishing},
  doi={10.1088/1751-8121/aac67a}
}

@article{margiotta2019glassy,
  title={Glassy dynamics on networks: local spectra and return probabilities},
  author={Margiotta, Riccardo Giuseppe and K{\"u}hn, Reimer and Sollich, Peter},
  journal={Journal of Statistical Mechanics: Theory and Experiment},
  volume={2019},
  number={9},
  pages={093304},
  year={2019},
  publisher={IOP Publishing},
  doi={10.1088/1742-5468/ab3aeb}
}

@article{bertin2003cross,
doi = {10.1088/0305-4470/36/43/002},
url = {https://doi.org/10.1088/0305-4470/36/43/002},
year = {2003},
month = {oct},
publisher = {},
volume = {36},
number = {43},
pages = {10683},
author = {Eric M Bertin},
title = {Cross-over from entropic to thermal dynamics in glassy models},
journal = {Journal of Physics A: Mathematical and General}
}

@article{paccoros,
  title = {Overlaps between eigenvectors of spiked, correlated random matrices: From matrix principal component analysis to random Gaussian landscapes},
  author = {Pacco, Alessandro and Ros, Valentina},
  journal = {Phys. Rev. E},
  volume = {108},
  issue = {2},
  pages = {024145},
  numpages = {29},
  year = {2023},
  month = {Aug},
  publisher = {American Physical Society},
  doi = {10.1103/PhysRevE.108.024145},
  url = {https://link.aps.org/doi/10.1103/PhysRevE.108.024145}
}

@article{pacco2024curvature,
doi = {10.1088/1751-8121/ad2039},
url = {https://doi.org/10.1088/1751-8121/ad2039},
year = {2024},
month = {jan},
publisher = {IOP Publishing},
volume = {57},
number = {7},
pages = {07LT01},
author = {Pacco, Alessandro and Biroli, Giulio and Ros, Valentina},
title = {Curvature-driven pathways interpolating between stationary points: the case of the pure spherical 3-spin model},
journal = {Journal of Physics A: Mathematical and Theoretical}
}

@misc{us_non_reciprocal_2025,
      title={Non-reciprocal interactions and high-dimensional chaos: comparing dynamics and statistics of equilibria in a solvable model}, 
      author={Samantha J. Fournier and Alessandro Pacco and Valentina Ros and Pierfrancesco Urbani},
      year={2025},
      eprint={2503.20908},
      archivePrefix={arXiv},
      primaryClass={cond-mat.dis-nn},
      url={https://arxiv.org/abs/2503.20908}, 
}

@article{pacco_triplets_2025,
doi = {10.1088/1742-5468/adbe40},
url = {https://doi.org/10.1088/1742-5468/adbe40},
year = {2025},
month = {apr},
publisher = {IOP Publishing},
volume = {2025},
number = {3},
pages = {033302},
author = {Pacco, Alessandro and Rosso, Alberto and Ros, Valentina},
title = {Triplets of local minima in a high-dimensional random landscape: correlations, clustering, and memoryless activated jumps},
journal = {Journal of Statistical Mechanics: Theory and Experiment}
}

@article{pacco_quenched_triplets_2025,
year = {2025},
author = {Pacco, Alessandro and Rosso, Alberto and Ros, Valentina},
title = {Three-point complexity of the pure spherical p-spin model: the quenched calculation},
journal = {In preparation}
}

@article{pacco_scs_2025,
year = {2025},
journal = {In preparation},
author ={Alessandro Pacco et al.},
title = {Landscape analysis of random neural networks with excitatory
interactions: dynamics, complexity, marginal states and topology trivializations},
}

@book{vivo_book_rmt_2018,
title = "Introduction to Random Matrices: Theory and Practice",
author = "Pierpaolo Vivo and Giacomo Livan and Marcel Novaes",
year = "2018",
doi = "10.1007/978-3-319-70885-0",
language = "English",
series = "SpringerBriefs",
publisher = "Springer",
}

@book{Tao_book_2012, address={Providence, Rhode Island}, title={Topics in Random Matrix Theory}, DOI={10.1090/gsm/132}, journal={Graduate Studies in Mathematics}, publisher={American Mathematical Society}, author={Tao, Terence}, year={2012}, month={mar.}}

@article{bun2018overlaps,
  title = {Overlaps between eigenvectors of correlated random matrices},
  author = {Bun, Jo\"el and Bouchaud, Jean-Philippe and Potters, Marc},
  journal = {Phys. Rev. E},
  volume = {98},
  issue = {5},
  pages = {052145},
  numpages = {12},
  year = {2018},
  month = {Nov},
  publisher = {American Physical Society},
  doi = {10.1103/PhysRevE.98.052145},
  url = {https://link.aps.org/doi/10.1103/PhysRevE.98.052145}
}

@article{bun2016rotational,
  author={Bun, Joël and Allez, Romain and Bouchaud, Jean-Philippe and Potters, Marc},
  journal={IEEE Transactions on Information Theory}, 
  title={Rotational Invariant Estimator for General Noisy Matrices}, 
  year={2016},
  volume={62},
  number={12},
  pages={7475-7490},
  doi={10.1109/TIT.2016.2616132}}

@article{allez_free_2014,
author = {Allez, Romain and Bouchaud, Jean-Philippe},
title = {Eigenvector dynamics under free addition},
journal = {Random Matrices: Theory and Applications},
volume = {03},
number = {03},
pages = {1450010},
year = {2014},
doi = {https://doi.org/10.1142/S2010326314500105}}

@article{peche2006largest,
  title={The largest eigenvalue of small rank perturbations of Hermitian random matrices},
  author={P{\'e}ch{\'e}, Sandrine},
  journal={Probability Theory and Related Fields},
  volume={134},
  number={1},
  pages={127--173},
  year={2006},
  publisher={Springer},
doi={https://doi.org/10.1007/s00440-005-0466-z}
}

@article{benaych2012singular,
  title={The singular values and vectors of low rank perturbations of large rectangular random matrices},
  author={Benaych-Georges, Florent and Nadakuditi, Raj Rao},
  journal={Journal of Multivariate Analysis},
  volume={111},
  pages={120--135},
  year={2012},
  publisher={Elsevier},
doi={https://doi.org/10.1016/j.jmva.2012.04.019}
}

@article{capitaine2011free,
  title={Free convolution with a semicircular distribution and eigenvalues of spiked deformations of Wigner matrices},
  author={Capitaine, Mireille and Donati-Martin, Catherine and F{\'e}ral, Delphine and F{\'e}vrier, Maxime and others},
  journal={Electron. J. Probab},
  volume={16},
  number={64},
  pages={1750--1792},
  year={2011}
}

@misc{capitaine2016spectrum,
      title={Spectrum of deformed random matrices and free probability}, 
      author={M Capitaine and C Donati-Martin},
      year={2016},
      eprint={1607.05560},
      archivePrefix={arXiv},
      primaryClass={math.PR},
      url={https://arxiv.org/abs/1607.05560}, 
}

@article{knowles2014outliers,
  title={The outliers of a deformed Wigner matrix},
  author={Knowles, Antti and Yin, Jun},
  journal={The Annals of Probability},
  volume={42},
  number={5},
  pages={1980--2031},
  year={2014},
  publisher={Institute of Mathematical Statistics}
}

@article{bordenave2016outlier,
  title={Outlier eigenvalues for deformed iid random matrices},
  author={Bordenave, Charles and Capitaine, Mireille},
  journal={Communications on Pure and Applied Mathematics},
  volume={69},
  number={11},
  pages={2131--2194},
  year={2016},
  publisher={Wiley Online Library},
doi={https://doi.org/10.1002/cpa.21629}
}

@phdthesis{rochet2016isolated,
  title={Isolated eigenvalues of non Hermitian random matrices},
  author={Rochet, Jean},
  year={2016},
  school={Universit{\'e} Sorbonne Paris Cit{\'e}}
}

@article{bun2017cleaning,
  title={Cleaning large correlation matrices: tools from random matrix theory},
  author={Bun, Jo{\"e}l and Bouchaud, Jean-Philippe and Potters, Marc},
  journal={Physics Reports},
  volume={666},
  pages={1--109},
  year={2017},
  publisher={Elsevier},
doi={https://doi.org/10.1016/j.physrep.2016.10.005}
}

@article{fyodorov2022extreme, title={Extreme Eigenvalues and the Emerging Outlier in Rank-One Non-Hermitian Deformations of the Gaussian Unitary Ensemble}, volume={25}, DOI={10.3390/e25010074}, number={1}, journal={Entropy}, publisher={MDPI AG}, author={Fyodorov, Yan V. and Khoruzhenko, Boris A. and Poplavskyi, Mihail}, year={2022}, month={dec.}, pages={74} }

@article{ikeda2022bose, title={Bose–Einstein-like condensation of deformed random matrix: a replica approach}, volume={2023}, DOI={10.1088/1742-5468/acb7d6}, number={2}, journal={Journal of Statistical Mechanics: Theory and Experiment}, publisher={IOP Publishing}, author={Ikeda, Harukuni}, year={2023}, month={feb.}, pages={023302} }

@article{montanari2015limitation,
  title={On the limitation of spectral methods: From the gaussian hidden clique problem to rank-one perturbations of gaussian tensors},
  author={Montanari, Andrea and Reichman, Daniel and Zeitouni, Ofer},
  journal={Advances in Neural Information Processing Systems},
  volume={28},
  year={2015}
}

@article{ledoit2011eigenvectors,
  title={Eigenvectors of some large sample covariance matrix ensembles},
  author={Ledoit, Olivier and P{\'e}ch{\'e}, Sandrine},
  journal={Probability Theory and Related Fields},
  volume={151},
  number={1},
  pages={233--264},
  year={2011},
  publisher={Springer},
doi={https://doi.org/10.1007/s00440-010-0298-3}
}

@article{lu2020phase,
  title={Phase transitions of spectral initialization for high-dimensional non-convex estimation},
  author={Lu, Yue M and Li, Gen},
  journal={Information and Inference: A Journal of the IMA},
  volume={9},
  number={3},
  pages={507--541},
  year={2020},
  publisher={Oxford University Press},
doi={https://doi.org/10.1093/imaiai/iaz020}
}

@article{hwang2020force,
  title = {Force balance controls the relaxation time of the gradient descent algorithm in the satisfiable phase},
  author = {Hwang, Sungmin and Ikeda, Harukuni},
  journal = {Phys. Rev. E},
  volume = {101},
  issue = {5},
  pages = {052308},
  numpages = {9},
  year = {2020},
  month = {May},
  publisher = {American Physical Society},
  doi = {10.1103/PhysRevE.101.052308},
  url = {https://link.aps.org/doi/10.1103/PhysRevE.101.052308}
}

@article{bertin_far_eq_2024,
  title = {Far-from-equilibrium complex landscapes},
  author = {Guislain, Laura and Bertin, Eric},
  journal = {Phys. Rev. E},
  volume = {111},
  issue = {6},
  pages = {L062101},
  numpages = {6},
  year = {2025},
  month = {Jun},
  publisher = {American Physical Society},
  doi = {10.1103/93tr-phkw},
  url = {https://link.aps.org/doi/10.1103/93tr-phkw}
}

@article{bertin_discontinuous_oscil_2024,
  title = {Discontinuous phase transition from ferromagnetic to oscillating states in a nonequilibrium mean-field spin model},
  author = {Guislain, Laura and Bertin, Eric},
  journal = {Phys. Rev. E},
  volume = {109},
  issue = {3},
  pages = {034131},
  numpages = {25},
  year = {2024},
  month = {Mar},
  publisher = {American Physical Society},
  doi = {10.1103/PhysRevE.109.034131},
  url = {https://link.aps.org/doi/10.1103/PhysRevE.109.034131}
}

@article{bertin_collective_2024,
doi = {10.1088/1742-5468/ad72dc},
url = {https://doi.org/10.1088/1742-5468/ad72dc},
year = {2024},
month = {sep},
publisher = {IOP Publishing},
volume = {2024},
number = {9},
pages = {093210},
author = {Guislain, Laura and Bertin, Eric},
title = {Collective oscillations in a three-dimensional spin model with non-reciprocal interactions},
journal = {Journal of Statistical Mechanics: Theory and Experiment}
}

@article{bertin_hidden_2024,
doi = {10.1088/1751-8121/ad6ab4},
url = {https://doi.org/10.1088/1751-8121/ad6ab4},
year = {2024},
month = {aug},
publisher = {IOP Publishing},
volume = {57},
number = {37},
pages = {375001},
author = {Guislain, Laura and Bertin, Eric},
title = {Hidden collective oscillations in a disordered mean-field spin model with non-reciprocal interactions},
journal = {Journal of Physics A: Mathematical and Theoretical}
}

@article{fraboul2021artificial,
title = {Artificial selection of communities drives the emergence of structured interactions},
journal = {Journal of Theoretical Biology},
volume = {571},
pages = {111557},
year = {2023},
issn = {0022-5193},
doi = {https://doi.org/10.1016/j.jtbi.2023.111557},
url = {https://www.sciencedirect.com/science/article/pii/S0022519323001546},
author = {Jules Fraboul and Giulio Biroli and Silvia {De Monte}}
}

@article{baron2022non,
  title = {Breakdown of Random-Matrix Universality in Persistent Lotka-Volterra Communities},
  author = {Baron, Joseph W. and Jewell, Thomas Jun and Ryder, Christopher and Galla, Tobias},
  journal = {Phys. Rev. Lett.},
  volume = {130},
  issue = {13},
  pages = {137401},
  numpages = {6},
  year = {2023},
  month = {Mar},
  publisher = {American Physical Society},
  doi = {10.1103/PhysRevLett.130.137401},
  url = {https://link.aps.org/doi/10.1103/PhysRevLett.130.137401}
}

@article{kosterlitz1976spherical,
  title = {Spherical Model of a Spin-Glass},
  author = {Kosterlitz, J. M. and Thouless, D. J. and Jones, Raymund C.},
  journal = {Phys. Rev. Lett.},
  volume = {36},
  issue = {20},
  pages = {1217--1220},
  numpages = {0},
  year = {1976},
  month = {May},
  publisher = {American Physical Society},
  doi = {10.1103/PhysRevLett.36.1217},
  url = {https://link.aps.org/doi/10.1103/PhysRevLett.36.1217}
}

@article{nadler2008finite,
  title={Finite sample approximation results for principal component analysis: A matrix perturbation approach},
  author={Nadler, Boaz},
  journal={The Annals of Statistics},
  volume={36},
  number={6},
  pages={2791--2817},
  year={2008},
  publisher={Institute of Mathematical Statistics},
doi={10.1214/08-AOS618}
}

@book{potters_bouchaud_2020, place={Cambridge}, title={A First Course in Random Matrix Theory: for Physicists, Engineers and Data Scientists}, DOI={10.1017/9781108768900}, publisher={Cambridge University Press}, author={Potters, Marc and Bouchaud, Jean-Philippe}, year={2020}}

@article{baik2005phase,
  title={Phase transition of the largest eigenvalue for nonnull complex sample covariance matrices},
  author={Baik, Jinho and Arous, G{\'e}rard Ben and P{\'e}ch{\'e}, Sandrine},
  journal={The Annals of Probability},
  volume={33},
  number={5},
  pages={1643--1697},
  year={2005},
  publisher={Institute of Mathematical Statistics},
doi={10.1214/009117905000000233}
}

@article{benaych2011eigenvalues,
  title={The eigenvalues and eigenvectors of finite, low rank perturbations of large random matrices},
  author={Benaych-Georges, Florent and Nadakuditi, Raj Rao},
  journal={Advances in Mathematics},
  volume={227},
  number={1},
  pages={494--521},
  year={2011},
  publisher={Elsevier},
doi={https://doi.org/10.1016/j.aim.2011.02.007}
}

@article{edwards1976eigenvalue,
  title={The eigenvalue spectrum of a large symmetric random matrix},
  author={Edwards, Samuel F and Jones, Raymund C},
  journal={Journal of Physics A: Mathematical and General},
  volume={9},
  number={10},
  pages={1595},
  year={1976},
  publisher={IOP Publishing},
doi={10.1088/0305-4470/9/10/011}
}

@article{VERBAARSCHOT1984367,
title = {Evaluation of ensemble averages for simple Hamiltonians perturbed by a GOE interaction},
journal = {Annals of Physics},
volume = {153},
number = {2},
pages = {367-388},
year = {1984},
issn = {0003-4916},
doi = {https://doi.org/10.1016/0003-4916(84)90023-X},
author = {J Verbaarschot and H.A Weidenmüller and M Zirnbauer}}

@article{cipolloni2022thermalisation,
  title={Thermalisation for Wigner matrices},
  author={Cipolloni, Giorgio and Erd{\H{o}}s, L{\'a}szl{\'o} and Schr{\"o}der, Dominik},
  journal={Journal of Functional Analysis},
  volume={282},
  number={8},
  pages={109394},
  year={2022},
  publisher={Elsevier},
doi={https://doi.org/10.1016/j.jfa.2022.109394}
}

@article{bun2018optimal,
author = {Bun, Joel},
year = {2018},
month = {02},
pages = {},
title = {An Optimal Rotational Invariant Estimator for General Covariance Matrices: the outliers}
}

@article{allez_cov_overlaps_2025,
doi = {10.1088/1751-8121/add066},
url = {https://dx.doi.org/10.1088/1751-8121/add066},
year = {2025},
month = {may},
publisher = {IOP Publishing},
volume = {58},
number = {20},
pages = {205003},
author = {Attal, Elie and Allez, Romain},
title = {Eigenvector overlaps of random covariance matrices and their submatrices},
journal = {Journal of Physics A: Mathematical and Theoretical},
}

@article{BiroliGuionnet,
author = {Biroli, Giulio and Guionnet, Alice},
year = {2020},
month = {01},
pages = {},
title = {Large deviations for the largest eigenvalues and eigenvectors of spiked Gaussian random matrices},
volume = {25},
journal = {Electronic Communications in Probability},
doi = {10.1214/20-ECP343}
}

@article{perry2018optimality,
  title={Optimality and sub-optimality of PCA I: Spiked random matrix models},
  author={Perry, Amelia and Wein, Alexander S and Bandeira, Afonso S and Moitra, Ankur},
  journal={The Annals of Statistics},
  volume={46},
  number={5},
  pages={2416--2451},
  year={2018},
  publisher={JSTOR},
doi={https://doi.org/10.1214/17-AOS1625}
}

@article{Girko_1985, title={Circular Law}, volume={29}, DOI={10.1137/1129095}, number={4}, journal={Theory of Probability and Its Applications}, publisher={Society for Industrial & Applied Mathematics (SIAM)}, author={Girko, V. L.}, year={1985}, month={jan.}, pages={694–706} }

@article{girko1986elliptic,
  title={Elliptic law},
  author={Girko, VL},
  journal={Theory of Probability \& Its Applications},
  volume={30},
  number={4},
  pages={677--690},
  year={1986},
  publisher={SIAM},
doi={https://doi.org/10.1137/1130089}
}

@article{nguyen2015elliptic,
  title={The elliptic law},
  author={Nguyen, Hoi H and O’Rourke, Sean},
  journal={International Mathematics Research Notices},
  volume={2015},
  number={17},
  pages={7620--7689},
  year={2015},
  publisher={Oxford University Press},
doi={https://doi.org/10.1093/imrn/rnu174}
}

@article{Ginibre_1965, title={Statistical Ensembles of Complex, Quaternion, and Real Matrices}, volume={6}, DOI={10.1063/1.1704292}, number={3}, journal={Journal of Mathematical Physics}, publisher={AIP Publishing}, author={Ginibre, Jean}, year={1965}, month={mar.}, pages={440–449} }

@article{tao_circularlaw,
author = {Tao, Terence and Vu, Van},
title = {Random Matrices: The Circular Law},
journal = {Communications in Contemporary Mathematics},
volume = {10},
number = {02},
pages = {261-307},
year = {2008},
doi = {10.1142/S0219199708002788},
URL = { https://doi.org/10.1142/S0219199708002788
}
}

@article{tao2013outliers,
  title={Outliers in the spectrum of iid matrices with bounded rank perturbations},
  author={Tao, Terence},
  journal={Probability Theory and Related Fields},
  volume={155},
  number={1},
  pages={231--263},
  year={2013},
  publisher={Springer},
doi={10.1007/s00440-011-0397-9}
}

@book{mehta_rmt,
author = {Mehta, M. L.},
address = {New York},
booktitle = {Random matrices and the statistical theory of energy levels},
lccn = {67023169},
publisher = {Academic Press},
title = {Random matrices and the statistical theory of energy levels},
year = {1967},
doi={https://doi.org/10.1016/C2013-0-12505-6}
}

@article{Wei_spectra_2012,
  title = {Eigenvalue spectra of asymmetric random matrices for multicomponent neural networks},
  author = {Wei, Yi},
  journal = {Phys. Rev. E},
  volume = {85},
  issue = {6},
  pages = {066116},
  numpages = {6},
  year = {2012},
  month = {Jun},
  publisher = {American Physical Society},
  doi = {10.1103/PhysRevE.85.066116},
  url = {https://link.aps.org/doi/10.1103/PhysRevE.85.066116}
}

@article{baron_fine_structure_2024,
  title = {Eigenvalue spectra of finely structured random matrices},
  author = {Poley, Lyle and Galla, Tobias and Baron, Joseph W.},
  journal = {Phys. Rev. E},
  volume = {109},
  issue = {6},
  pages = {064301},
  numpages = {35},
  year = {2024},
  month = {Jun},
  publisher = {American Physical Society},
  doi = {10.1103/PhysRevE.109.064301},
  url = {https://link.aps.org/doi/10.1103/PhysRevE.109.064301}
}

@article{baron_path_2022,
  title = {Eigenvalues of Random Matrices with Generalized Correlations: A Path Integral Approach},
  author = {Baron, Joseph W. and Jewell, Thomas Jun and Ryder, Christopher and Galla, Tobias},
  journal = {Phys. Rev. Lett.},
  volume = {128},
  issue = {12},
  pages = {120601},
  numpages = {6},
  year = {2022},
  month = {Mar},
  publisher = {American Physical Society},
  doi = {10.1103/PhysRevLett.128.120601},
  url = {https://link.aps.org/doi/10.1103/PhysRevLett.128.120601}
}

@article{sommers1988spectrum,
  title = {Spectrum of Large Random Asymmetric Matrices},
  author = {Sommers, H. J. and Crisanti, A. and Sompolinsky, H. and Stein, Y.},
  journal = {Phys. Rev. Lett.},
  volume = {60},
  issue = {19},
  pages = {1895--1898},
  numpages = {0},
  year = {1988},
  month = {May},
  publisher = {American Physical Society},
  doi = {10.1103/PhysRevLett.60.1895},
  url = {https://link.aps.org/doi/10.1103/PhysRevLett.60.1895}
}

@article{Burda_non_herm_2011,
  title = {Multiplication law and $S$ transform for non-Hermitian random matrices},
  author = {Burda, Z. and Janik, R. A. and Nowak, M. A.},
  journal = {Phys. Rev. E},
  volume = {84},
  issue = {6},
  pages = {061125},
  numpages = {17},
  year = {2011},
  month = {Dec},
  publisher = {American Physical Society},
  doi = {10.1103/PhysRevE.84.061125},
  url = {https://link.aps.org/doi/10.1103/PhysRevE.84.061125}
}

@misc{Khoru_Sommers_non_herm_2009,
      title={Non-Hermitian Random Matrix Ensembles}, 
      author={B. A. Khoruzhenko and H. -J. Sommers},
      year={2009},
      eprint={0911.5645},
      archivePrefix={arXiv},
      primaryClass={math-ph},
      url={https://arxiv.org/abs/0911.5645}, 
}

@article{jones1978eigenvalue,
  title={The eigenvalue spectrum of a large symmetric random matrix with a finite mean},
  author={Jones, RC and Kosterlitz, JM and Thouless, DJ},
  journal={Journal of Physics A: Mathematical and General},
  volume={11},
  number={3},
  pages={L45},
  year={1978},
  publisher={IOP Publishing},
doi={10.1088/0305-4470/11/3/002}
}

@article{furedi1981eigenvalues,
  title={The eigenvalues of random symmetric matrices},
  author={F{\"u}redi, Zolt{\'a}n and Koml{\'o}s, J{\'a}nos},
  journal={Combinatorica},
  volume={1},
  number={3},
  pages={233--241},
  year={1981},
  publisher={Springer},
doi={https://doi.org/10.1007/BF02579329}
}

@article{johnstone2001distribution,
  title={On the distribution of the largest eigenvalue in principal components analysis},
  author={Johnstone, Iain M},
  journal={The Annals of statistics},
  volume={29},
  number={2},
  pages={295--327},
  year={2001},
  publisher={Institute of Mathematical Statistics}
}

@misc{ros_lecture_2025,
      title={High-dimensional random landscapes: from typical to large deviations}, 
      author={Valentina Ros},
      year={2025},
      eprint={2502.14084},
      archivePrefix={arXiv},
      primaryClass={cond-mat.dis-nn},
      url={https://arxiv.org/abs/2502.14084}, 
}

@article{Wigner_1,
 ISSN = {0003486X, 19398980},
 URL = {http://www.jstor.org/stable/1970079},
 author = {Eugene P. Wigner},
 journal = {Annals of Mathematics},
 number = {3},
 pages = {548--564},
 publisher = {[Annals of Mathematics, Trustees of Princeton University on Behalf of the Annals of Mathematics, Mathematics Department, Princeton University]},
 title = {Characteristic Vectors of Bordered Matrices With Infinite Dimensions},
 urldate = {2025-06-27},
 volume = {62},
 year = {1955},
doi={https://doi.org/10.2307/1970079}
}

@article{Wigner_2,
 ISSN = {0003486X, 19398980},
 URL = {http://www.jstor.org/stable/1970008},
 author = {Eugene P. Wigner},
 journal = {Annals of Mathematics},
 number = {2},
 pages = {325--327},
 publisher = {[Annals of Mathematics, Trustees of Princeton University on Behalf of the Annals of Mathematics, Mathematics Department, Princeton University]},
 title = {On the Distribution of the Roots of Certain Symmetric Matrices},
 urldate = {2025-06-27},
 volume = {67},
 year = {1958},
doi={https://doi.org/10.2307/1970008}
}

@article{Kac_roots_43,
author = {M. Kac},
title = {{On the average number of real roots of a random algebraic equation}},
volume = {49},
journal = {Bulletin of the American Mathematical Society},
number = {4},
publisher = {American Mathematical Society},
pages = {314 -- 320},
year = {1943},
}

@article{Kac_roots_48,
author = {M. Kac},
title = {{On the average number of real roots of a random algebraic
equation (II)}},
volume = {50},
journal = {Proc. London. Math. Soc.},
number = {1},
publisher = {American Mathematical Society},
pages = {390 -- 408},
year = {1948},
}

@article{Edelman_roots_95,
author = "Edelman, Alan and Kostlan, Eric",
title = "{How many zeros of a random polynomial are real?}",
volume = {32},
journal = {Bulletin (New Scries) Of The
American Mathematical Society},
number = {q},
pages = {314 -- 320},
year = {1995},
}

@ARTICLE{Rice_44,
  author={Rice, S. O.},
  journal={The Bell System Technical Journal}, 
  title={Mathematical analysis of random noise}, 
  year={1944},
  volume={23},
  number={3},
  pages={282-332},
  doi={10.1002/j.1538-7305.1944.tb00874.x}}

@book{Cramer_stationary_67,
author = {Cramér, H. and  Leadbetter,M. R.},
year = {1967},
publisher={John Wiley and Sons, Inc.},
title = {Stationary and related stochastic
processes},
doi={}
}

@article{Higgins,
author={M. S. Longuet-Higgins},
title={The statistical analysis of a random, moving surface}, volume={249}, 
DOI={10.1098/rsta.1957.0002}, 
number={966}, 
journal={The Royal Society}, 
year={1957}, 
month={feb.}, 
pages={321–387}}

@misc{brezin_kac_2022,
      title={Kac-Rice formula: A contemporary overview of the main results and applications}, 
      author={Corinne Berzin and Alain Latour and José León},
      year={2022},
      eprint={2205.08742},
      archivePrefix={arXiv},
      primaryClass={math.CA},
      url={https://arxiv.org/abs/2205.08742}, 
}

@article{ros2019complexity,
doi = {10.1209/0295-5075/126/20003},
url = {https://doi.org/10.1209/0295-5075/126/20003},
year = {2019},
month = {may},
publisher = {EDP Sciences, IOP Publishing and Società Italiana di Fisica},
volume = {126},
number = {2},
pages = {20003},
author = {Ros, V. and Biroli, G. and Cammarota, C.},
title = {Complexity of energy barriers in mean-field glassy systems},
journal = {Europhysics Letters}
}

@book{random_fields_adler,
author={Robert J. Adler , Jonathan E. Taylor},
title={Random Fields and Geometry}, 
address={New York, NY}, 
DOI={10.1007/978-0-387-48116-6}, 
journal={Springer Monographs in Mathematics}, 
publisher={Springer New York}, 
year={2007}, }

@inbook{ros2022high,
author = {Valentina Ros and Yan V. Fyodorov},
title = {The High-dimensional Landscape Paradigm: Spin-Glasses, and Beyond},
booktitle = {Spin Glass Theory and Far Beyond},
pages = {95-114},
doi = {10.1142/9789811273926_0006},
year      = {2022},
publisher = {World Scientific Publishing}
}

@misc{Fyod_highd_rmt_2013,
      title={High-Dimensional Random Fields and Random Matrix Theory}, 
      author={Yan V Fyodorov},
      year={2013},
      eprint={1307.2379},
      archivePrefix={arXiv},
      primaryClass={math-ph},
      url={https://arxiv.org/abs/1307.2379}, 
}

@article{Fyod2004,
  title = {Complexity of Random Energy Landscapes, Glass Transition, and Absolute Value of the Spectral Determinant of Random Matrices},
  author = {Fyodorov, Yan V.},
  journal = {Phys. Rev. Lett.},
  volume = {92},
  issue = {24},
  pages = {240601},
  numpages = {4},
  year = {2004},
  month = {Jun},
  publisher = {American Physical Society},
  doi = {10.1103/PhysRevLett.92.240601},
  url = {https://link.aps.org/doi/10.1103/PhysRevLett.92.240601}
}

@article{ros2020distribution,
doi = {10.1088/1751-8121/ab73ac},
url = {https://doi.org/10.1088/1751-8121/ab73ac},
year = {2020},
month = {mar},
publisher = {IOP Publishing},
volume = {53},
number = {12},
pages = {125002},
author = {Ros, Valentina},
title = {Distribution of rare saddles in the p-spin energy landscape},
journal = {Journal of Physics A: Mathematical and Theoretical}
}

@article{Bray_large_dim_2007,
  title = {Statistics of Critical Points of Gaussian Fields on Large-Dimensional Spaces},
  author = {Bray, Alan J. and Dean, David S.},
  journal = {Phys. Rev. Lett.},
  volume = {98},
  issue = {15},
  pages = {150201},
  numpages = {4},
  year = {2007},
  month = {Apr},
  publisher = {American Physical Society},
  doi = {https://doi.org/10.1103/PhysRevLett.98.150201}
}

@article{Fyodorov_Williams_2007, 
title={Replica Symmetry Breaking Condition Exposed by Random Matrix Calculation of Landscape Complexity}, 
volume={129}, 
DOI={10.1007/s10955-007-9386-x},
number={5-6}, 
journal={Journal of Statistical Physics}, 
publisher={Springer Science and Business Media LLC}, 
author={Fyodorov, Yan V. and Williams, Ian}, 
year={2007}, 
month={sep.}, 
pages={1081–1116}}

@article{Monthus_1d_landscape_2003,
title = {Exact solutions for the statistics of extrema of some random 1D landscapes, application to the equilibrium and the dynamics of the toy model},
journal = {Physica A: Statistical Mechanics and its Applications},
volume = {317},
number = {1},
pages = {140-198},
year = {2003},
issn = {0378-4371},
doi = {https://doi.org/10.1016/S0378-4371(02)01317-1},
url = {https://www.sciencedirect.com/science/article/pii/S0378437102013171},
author = {Pierre {Le Doussal} and Cécile Monthus},
}

@article{fyod_doussal_texier_mfd_2018,
title = {Exponential number of equilibria and depinning threshold for a directed polymer in a random potential},
journal = {Annals of Physics},
volume = {397},
pages = {1-64},
year = {2018},
issn = {0003-4916},
doi = {https://doi.org/10.1016/j.aop.2018.07.029},
url = {https://www.sciencedirect.com/science/article/pii/S0003491618302008},
author = {Yan V. Fyodorov and Pierre {Le Doussal} and Alberto Rosso and Christophe Texier},
}

@article{fyod_doussal_mfd_2020,
  title = {Manifolds in a high-dimensional random landscape: Complexity of stationary points and depinning},
  author = {Fyodorov, Yan V. and Le Doussal, Pierre},
  journal = {Phys. Rev. E},
  volume = {101},
  issue = {2},
  pages = {020101},
  numpages = {6},
  year = {2020},
  month = {Feb},
  publisher = {American Physical Society},
  doi = {10.1103/PhysRevE.101.020101},
  url = {https://link.aps.org/doi/10.1103/PhysRevE.101.020101}
}

@article{Fyod_nadal_2012,
  title = {Critical Behavior of the Number of Minima of a Random Landscape at the Glass Transition Point and the Tracy-Widom Distribution},
  author = {Fyodorov, Yan V. and Nadal, Celine},
  journal = {Phys. Rev. Lett.},
  volume = {109},
  issue = {16},
  pages = {167203},
  numpages = {5},
  year = {2012},
  month = {Oct},
  publisher = {American Physical Society},
  doi = {10.1103/PhysRevLett.109.167203},
  url = {https://link.aps.org/doi/10.1103/PhysRevLett.109.167203}
}

@article{Gillin_Sherrington_2000,
doi = {10.1088/0305-4470/33/16/302},
url = {https://dx.doi.org/10.1088/0305-4470/33/16/302},
year = {2000},
month = {apr},
publisher = {},
volume = {33},
number = {16},
pages = {3081},
author = {Peter Gillin and David Sherrington},
title = {p>2 spin glasses with first-order ferromagnetic transitions},
journal = {Journal of Physics A: Mathematical and General},
}

@article{Nadal_right_tail_TW_2012,
author = {Borot, Ga\"{e}tan and Nadal, C\'{e}line},
title = {Right tail asymptotic expansion of the Tracy-Widom beta laws},
journal = {Random Matrices: Theory and Applications},
volume = {01},
number = {03},
pages = {1250006},
year = {2012},
doi = {10.1142/S2010326312500062}}

@article{Borot_max_eval_2011,
doi = {10.1088/1742-5468/2011/11/P11024},
url = {https://dx.doi.org/10.1088/1742-5468/2011/11/P11024},
year = {2011},
month = {nov},
publisher = {},
volume = {2011},
number = {11},
pages = {P11024},
author = {Borot, G and Eynard, B and Majumdar, S N and Nadal, C},
title = {Large deviations of the maximal eigenvalue of random matrices},
journal = {Journal of Statistical Mechanics: Theory and Experiment},
}

@article{Tracy_Widom_1996, title={On orthogonal and symplectic matrix ensembles}, volume={177}, DOI={10.1007/bf02099545}, number={3}, journal={Communications in Mathematical Physics}, publisher={Springer Science and Business Media LLC}, author={Tracy, Craig A. and Widom, Harold}, year={1996}, month={apr.}, pages={727–754} }

@article{ChaosSompo88,
  title = {Chaos in Random Neural Networks},
  author = {Sompolinsky, H. and Crisanti, A. and Sommers, H. J.},
  journal = {Phys. Rev. Lett.},
  volume = {61},
  issue = {3},
  pages = {259--262},
  numpages = {0},
  year = {1988},
  month = {Jul},
  publisher = {American Physical Society},
  doi = {10.1103/PhysRevLett.61.259},
  url = {https://link.aps.org/doi/10.1103/PhysRevLett.61.259}
}

@book{HeliasBook20,
author = {Helias, Moritz and Dahmen, David},
year = {2020},
month = {01},
pages = {},
title = {Statistical Field Theory for Neural Networks},
isbn = {978-3-030-46443-1},
doi = {10.1007/978-3-030-46444-8}
}

@article{AnnibaleDynamics2024,
title = {Dynamically selected steady states and criticality in non-reciprocal networks},
journal = {Chaos, Solitons and Fractals},
volume = {182},
pages = {114809},
year = {2024},
issn = {0960-0779},
doi = {https://doi.org/10.1016/j.chaos.2024.114809},
url = {https://www.sciencedirect.com/science/article/pii/S0960077924003618},
author = {Carles Martorell and Rubén Calvo and Alessia Annibale and Miguel A. Muñoz},
}

@article{Parisi_asymm_1986,
doi = {10.1088/0305-4470/19/11/005},
url = {https://dx.doi.org/10.1088/0305-4470/19/11/005},
year = {1986},
month = {aug},
publisher = {},
volume = {19},
number = {11},
pages = {L675},
author = {G Parisi},
title = {Asymmetric neural networks and the process of learning},
journal = {Journal of Physics A: Mathematical and General}
}

@article{MastrogiuseppeLink2018,
title = {Linking Connectivity, Dynamics, and Computations in Low-Rank Recurrent Neural Networks},
journal = {Neuron},
volume = {99},
number = {3},
pages = {609-623.e29},
year = {2018},
issn = {0896-6273},
doi = {https://doi.org/10.1016/j.neuron.2018.07.003},
url = {https://www.sciencedirect.com/science/article/pii/S0896627318305439},
author = {Francesca Mastrogiuseppe and Srdjan Ostojic},
}

@article{crisanti_path_2018,
  title = {Path integral approach to random neural networks},
  author = {Crisanti, A. and Sompolinsky, H.},
  journal = {Phys. Rev. E},
  volume = {98},
  issue = {6},
  pages = {062120},
  numpages = {16},
  year = {2018},
  month = {Dec},
  publisher = {American Physical Society},
  doi = {10.1103/PhysRevE.98.062120},
  url = {https://link.aps.org/doi/10.1103/PhysRevE.98.062120}
}

@article{HeliasMemory18,
  title = {Optimal Sequence Memory in Driven Random Networks},
  author = {Schuecker, Jannis and Goedeke, Sven and Helias, Moritz},
  journal = {Phys. Rev. X},
  volume = {8},
  issue = {4},
  pages = {041029},
  numpages = {28},
  year = {2018},
  month = {Nov},
  publisher = {American Physical Society},
  doi = {10.1103/PhysRevX.8.041029},
  url = {https://link.aps.org/doi/10.1103/PhysRevX.8.041029}
}

@article{Sompo_transition_2015,
  title = {Transition to Chaos in Random Neuronal Networks},
  author = {Kadmon, Jonathan and Sompolinsky, Haim},
  journal = {Phys. Rev. X},
  volume = {5},
  issue = {4},
  pages = {041030},
  numpages = {28},
  year = {2015},
  month = {Nov},
  publisher = {American Physical Society},
  doi = {10.1103/PhysRevX.5.041030},
  url = {https://link.aps.org/doi/10.1103/PhysRevX.5.041030}
}

@article{Rajan_spectra_06,
  title = {Eigenvalue Spectra of Random Matrices for Neural Networks},
  author = {Rajan, Kanaka and Abbott, L. F.},
  journal = {Phys. Rev. Lett.},
  volume = {97},
  issue = {18},
  pages = {188104},
  numpages = {4},
  year = {2006},
  month = {Nov},
  publisher = {American Physical Society},
  doi = {10.1103/PhysRevLett.97.188104},
  url = {https://link.aps.org/doi/10.1103/PhysRevLett.97.188104}
}

@article{HeliasFP2022,
  title = {Fixed point geometry in chaotic neural networks},
  author = {Stubenrauch, Jakob and Keup, Christian and Kurth, Anno C. and Helias, Moritz and van Meegen, Alexander},
  journal = {Phys. Rev. Res.},
  volume = {7},
  issue = {2},
  pages = {023203},
  numpages = {11},
  year = {2025},
  month = {May},
  publisher = {American Physical Society},
  doi = {10.1103/PhysRevResearch.7.023203},
  url = {https://link.aps.org/doi/10.1103/PhysRevResearch.7.023203}
}

@article{wainrib2013topological,
  title = {Topological and Dynamical Complexity of Random Neural Networks},
  author = {Wainrib, Gilles and Touboul, Jonathan},
  journal = {Phys. Rev. Lett.},
  volume = {110},
  issue = {11},
  pages = {118101},
  numpages = {4},
  year = {2013},
  month = {Mar},
  publisher = {American Physical Society},
  doi = {10.1103/PhysRevLett.110.118101},
  url = {https://link.aps.org/doi/10.1103/PhysRevLett.110.118101}
}

@article{HuangSteadyState2024,
doi = {10.1088/1572-9494/ad8126},
url = {https://doi.org/10.1088/1572-9494/ad8126},
year = {2024},
month = {dec},
publisher = {IOP Publishing},
volume = {77},
number = {3},
pages = {035601},
author = {Qiu, Junbin and Huang, Haiping},
title = {An optimization-based equilibrium measure describing fixed points of non-equilibrium dynamics: application to the edge of chaos},
journal = {Communications in Theoretical Physics}
}

@article{Fyodorov_2016,
doi = {10.1088/1742-5468/aa511a},
url = {https://dx.doi.org/10.1088/1742-5468/aa511a},
year = {2016},
month = {dec},
publisher = {IOP Publishing and SISSA},
volume = {2016},
number = {12},
pages = {124003},
author = {Y V Fyodorov},
title = {Topology trivialization transition in random non-gradient autonomous ODEs on a sphere},
journal = {Journal of Statistical Mechanics: Theory and Experiment}
}

@article{Fyodorov_resilient_2021,
  title = {Nonlinearity-generated resilience in large complex systems},
  author = {Belga Fedeli, S. and Fyodorov, Y. V. and Ipsen, J. R.},
  journal = {Phys. Rev. E},
  volume = {103},
  issue = {2},
  pages = {022201},
  numpages = {16},
  year = {2021},
  month = {Feb},
  publisher = {American Physical Society},
  doi = {10.1103/PhysRevE.103.022201},
  url = {https://link.aps.org/doi/10.1103/PhysRevE.103.022201}
}

@article{Abbott_lyapunov_2023,
  title = {Lyapunov spectra of chaotic recurrent neural networks},
  author = {Engelken, Rainer and Wolf, Fred and Abbott, L. F.},
  journal = {Phys. Rev. Res.},
  volume = {5},
  issue = {4},
  pages = {043044},
  numpages = {28},
  year = {2023},
  month = {Oct},
  publisher = {American Physical Society},
  doi = {10.1103/PhysRevResearch.5.043044},
  url = {https://link.aps.org/doi/10.1103/PhysRevResearch.5.043044}
}

@article{CugliandoloNonrelax97,
  title = {Glassy behaviour in disordered systems with nonrelaxational dynamics},
  author = {Cugliandolo, Leticia F. and Kurchan, Jorge and Le Doussal, Pierre and Peliti, Luca},
  journal = {Phys. Rev. Lett.},
  volume = {78},
  issue = {2},
  pages = {350--353},
  numpages = {0},
  year = {1997},
  month = {Jan},
  publisher = {American Physical Society},
  doi = {10.1103/PhysRevLett.78.350},
  url = {https://link.aps.org/doi/10.1103/PhysRevLett.78.350}
}

@article{crisanti1987dynamics,
  title = {Dynamics of spin systems with randomly asymmetric bonds: Langevin dynamics and a spherical model},
  author = {Crisanti, A. and Sompolinsky, H.},
  journal = {Phys. Rev. A},
  volume = {36},
  issue = {10},
  pages = {4922--4939},
  numpages = {0},
  year = {1987},
  month = {Nov},
  publisher = {American Physical Society},
  doi = {10.1103/PhysRevA.36.4922},
  url = {https://link.aps.org/doi/10.1103/PhysRevA.36.4922}
}

@article{Schuster_suppression_92,
  title = {Suppressing chaos in neural networks by noise},
  author = {Molgedey, L. and Schuchhardt, J. and Schuster, H. G.},
  journal = {Phys. Rev. Lett.},
  volume = {69},
  issue = {26},
  pages = {3717--3719},
  numpages = {0},
  year = {1992},
  month = {Dec},
  publisher = {American Physical Society},
  doi = {10.1103/PhysRevLett.69.3717},
  url = {https://link.aps.org/doi/10.1103/PhysRevLett.69.3717}
}

@misc{garcia2017numberequilibriagivennumber,
      title={On the number of equilibria with a given number of unstable directions}, 
      author={Xavier Garcia},
      year={2017},
      eprint={1709.04021},
      archivePrefix={arXiv},
      primaryClass={math.PR},
      url={https://arxiv.org/abs/1709.04021}, 
}

@article{Kivimae2024,
author = {Kivimae, Pax},
year = {2024},
month = {11},
pages = {},
title = {Concentration of Equilibria and Relative Instability in Disordered Non-Relaxational Dynamics},
volume = {405},
journal = {Communications in Mathematical Physics},
doi = {10.1007/s00220-024-05158-5}
}

@article{lacroix2022counting,
doi = {10.1088/1751-8121/ac564a},
url = {https://doi.org/10.1088/1751-8121/ac564a},
year = {2022},
month = {mar},
publisher = {IOP Publishing},
volume = {55},
number = {14},
pages = {144001},
author = {Lacroix-A-Chez-Toine, Bertrand and Fyodorov, Yan V},
title = {Counting equilibria in a random non-gradient dynamics with heterogeneous relaxation rates},
journal = {Journal of Physics A: Mathematical and Theoretical}
}

@inproceedings{hertz1986memory,
  title={Memory networks with asymmetric bonds},
  author={Hertz, JA and Grinstein, G and Solla, SA},
  booktitle={AIP Conference Proceedings},
  volume={151},
  number={1},
  pages={212--218},
  year={1986},
  organization={American Institute of Physics},
  doi={https://doi.org/10.1063/1.36259}
}

@article{gardner1989phase,
doi = {10.1088/0305-4470/22/12/005},
url = {https://doi.org/10.1088/0305-4470/22/12/005},
year = {1989},
month = {jun},
publisher = {IOP Publishing},
volume = {22},
number = {12},
pages = {1995},
author = {E Gardner and H Gutfreund and I Yekutieli},
title = {The phase space of interactions in neural networks with definite symmetry},
journal = {Journal of Physics A: Mathematical and General}
}

@article{rajan2010stimulus,
  title = {Stimulus-dependent suppression of chaos in recurrent neural networks},
  author = {Rajan, Kanaka and Abbott, L. F. and Sompolinsky, Haim},
  journal = {Phys. Rev. E},
  volume = {82},
  issue = {1},
  pages = {011903},
  numpages = {5},
  year = {2010},
  month = {Jul},
  publisher = {American Physical Society},
  doi = {10.1103/PhysRevE.82.011903},
  url = {https://link.aps.org/doi/10.1103/PhysRevE.82.011903}
}

@article{marti2018correlations,
  title = {Correlations between synapses in pairs of neurons slow down dynamics in randomly connected neural networks},
  author = {Mart\'{\i}, Daniel and Brunel, Nicolas and Ostojic, Srdjan},
  journal = {Phys. Rev. E},
  volume = {97},
  issue = {6},
  pages = {062314},
  numpages = {19},
  year = {2018},
  month = {Jun},
  publisher = {American Physical Society},
  doi = {10.1103/PhysRevE.97.062314},
  url = {https://link.aps.org/doi/10.1103/PhysRevE.97.062314}
}

@article{aguirre2022satisfiability,
doi = {10.1088/1751-8121/ac79e5},
url = {https://doi.org/10.1088/1751-8121/ac79e5},
year = {2022},
month = {jul},
publisher = {IOP Publishing},
volume = {55},
number = {30},
pages = {305001},
author = {Aguirre-López, Fabián and Pastore, Mauro and Franz, Silvio},
title = {Satisfiability transition in asymmetric neural networks},
journal = {Journal of Physics A: Mathematical and Theoretical}
}

@article{pereira2023forgetting,
  title = {Forgetting Leads to Chaos in Attractor Networks},
  author = {Pereira-Obilinovic, Ulises and Aljadeff, Johnatan and Brunel, Nicolas},
  journal = {Phys. Rev. X},
  volume = {13},
  issue = {1},
  pages = {011009},
  numpages = {20},
  year = {2023},
  month = {Jan},
  publisher = {American Physical Society},
  doi = {10.1103/PhysRevX.13.011009},
  url = {https://link.aps.org/doi/10.1103/PhysRevX.13.011009}
}

@article{sussillo2009generating,
  title={Generating coherent patterns of activity from chaotic neural networks},
  author={Sussillo, David and Abbott, Larry F},
  journal={Neuron},
  volume={63},
  number={4},
  pages={544--557},
  year={2009},
  publisher={Elsevier},
  doi={ 10.1016/j.neuron.2009.07.018}
}

@article{fournier2023statistical,
doi = {10.1088/1742-5468/ad082d},
url = {https://doi.org/10.1088/1742-5468/ad082d},
year = {2023},
month = {nov},
publisher = {IOP Publishing},
volume = {2023},
number = {11},
pages = {113301},
author = {Fournier, Samantha J and Urbani, Pierfrancesco},
title = {Statistical physics of learning in high-dimensional chaotic systems},
journal = {Journal of Statistical Mechanics: Theory and Experiment}
}

@article{fournier2024generative,
  title = {Generative modeling through internal high-dimensional chaotic activity},
  author = {Fournier, Samantha J. and Urbani, Pierfrancesco},
  journal = {Phys. Rev. E},
  volume = {111},
  issue = {4},
  pages = {045304},
  numpages = {7},
  year = {2025},
  month = {Apr},
  publisher = {American Physical Society},
  doi = {10.1103/PhysRevE.111.045304},
  url = {https://link.aps.org/doi/10.1103/PhysRevE.111.045304}
}

@article{garnier2024unlearnable,
  title = {Unlearnable Games and ``Satisficing'' Decisions: A Simple Model for a Complex World},
  author = {Garnier-Brun, J\'er\^ome and Benzaquen, Michael and Bouchaud, Jean-Philippe},
  journal = {Phys. Rev. X},
  volume = {14},
  issue = {2},
  pages = {021039},
  numpages = {38},
  year = {2024},
  month = {Jun},
  publisher = {American Physical Society},
  doi = {10.1103/PhysRevX.14.021039},
  url = {https://link.aps.org/doi/10.1103/PhysRevX.14.021039}
}

@article{rogers2022chaos,
  title={Chaos is not rare in natural ecosystems},
  author={Rogers, Tanya L and Johnson, Bethany J and Munch, Stephan B},
  journal={Nature ecology \& evolution},
  volume={6},
  number={8},
  pages={1105--1111},
  year={2022},
  publisher={Nature Publishing Group UK London},
  doi={https://doi.org/10.1038/s41559-022-01787-y}
}

@article{castedo2024generalised,
  title = {Generalized correlations in disordered dynamical systems: Insights from the many-species Lotka-Volterra model},
  author = {Castedo, Sebastian H. and Holmes, Joshua and Baron, Joseph W. and Galla, Tobias},
  journal = {Phys. Rev. E},
  volume = {111},
  issue = {4},
  pages = {044202},
  numpages = {16},
  year = {2025},
  month = {Apr},
  publisher = {American Physical Society},
  doi = {10.1103/PhysRevE.111.044202},
  url = {https://link.aps.org/doi/10.1103/PhysRevE.111.044202}
}

@article{Galla_random_LV_2023,
  title = {Niche overlap and Hopfield-like interactions in generalized random Lotka-Volterra systems},
  author = {Garcia, Enrique Rozas and Crumpton, Mark J. and Galla, Tobias},
  journal = {Phys. Rev. E},
  volume = {108},
  issue = {3},
  pages = {034120},
  numpages = {12},
  year = {2023},
  month = {Sep},
  publisher = {American Physical Society},
  doi = {10.1103/PhysRevE.108.034120},
  url = {https://link.aps.org/doi/10.1103/PhysRevE.108.034120}
}

@article{mahadevan2024continual,
  title = {Continual Evolution in Nonreciprocal Ecological Models},
  author = {Mahadevan, Aditya and Fisher, Daniel S.},
  journal = {PRX Life},
  volume = {3},
  issue = {3},
  pages = {033008},
  numpages = {37},
  year = {2025},
  month = {Aug},
  publisher = {American Physical Society},
  doi = {10.1103/dv3k-75b9},
  url = {https://link.aps.org/doi/10.1103/dv3k-75b9}
}

@article{arnoulx2024many,
  title = {Many-Species Ecological Fluctuations as a Jump Process from the Brink of Extinction},
  author = {Arnoulx de Pirey, Thibaut and Bunin, Guy},
  journal = {Phys. Rev. X},
  volume = {14},
  issue = {1},
  pages = {011037},
  numpages = {19},
  year = {2024},
  month = {Mar},
  publisher = {American Physical Society},
  doi = {10.1103/PhysRevX.14.011037},
  url = {https://link.aps.org/doi/10.1103/PhysRevX.14.011037}
}

@article{Miller_rmt_2015,
  title = {Properties of networks with partially structured and partially random connectivity},
  author = {Ahmadian, Yashar and Fumarola, Francesco and Miller, Kenneth D.},
  journal = {Phys. Rev. E},
  volume = {91},
  issue = {1},
  pages = {012820},
  numpages = {36},
  year = {2015},
  month = {Jan},
  publisher = {American Physical Society},
  doi = {10.1103/PhysRevE.91.012820},
  url = {https://link.aps.org/doi/10.1103/PhysRevE.91.012820}
}

@article{berthier2000two,
  title = {A two-time-scale, two-temperature scenario for nonlinear rheology},
  author = {Berthier, Ludovic and Barrat, Jean-Louis and Kurchan, Jorge},
  journal = {Phys. Rev. E},
  volume = {61},
  issue = {5},
  pages = {5464--5472},
  numpages = {0},
  year = {2000},
  month = {May},
  publisher = {American Physical Society},
  doi = {10.1103/PhysRevE.61.5464},
  url = {https://link.aps.org/doi/10.1103/PhysRevE.61.5464}
}

@article{berlemont2022glassy,
	author = {Berlemont, Kevin and Mongillo, Gianluigi},
	title = {Glassy phase in dynamically-balanced neuronal networks},
	elocation-id = {2022.03.14.484348},
	year = {2022},
	doi = {10.1101/2022.03.14.484348},
	publisher = {Cold Spring Harbor Laboratory},
	URL = {https://www.biorxiv.org/content/early/2022/03/17/2022.03.14.484348},
	eprint = {https://www.biorxiv.org/content/early/2022/03/17/2022.03.14.484348.full.pdf},
	journal = {bioRxiv}
}

@article{Rosso_review_2021,
   author = "Ferrero, Ezequiel E. and Foini, Laura and Giamarchi, Thierry and Kolton, Alejandro B. and Rosso, Alberto",
   title = "Creep Motion of Elastic Interfaces Driven in a Disordered Landscape", 
   journal= "Annual Review of Condensed Matter Physics",
   year = "2021",
   volume = "12",
   number = "Volume 12, 2021",
   pages = "111-134",
   doi = "https://doi.org/10.1146/annurev-conmatphys-031119-050725",
   url = "https://www.annualreviews.org/content/journals/10.1146/annurev-conmatphys-031119-050725",
   publisher = "Annual Reviews",
   issn = "1947-5462",
   type = "Journal Article",
  }

@article{Mezard_Parisi_mfd_1991,
	author = {{Marc Mézard} and {Giorgio Parisi}},
	title = {Replica field theory for random manifolds},
	DOI= "10.1051/jp1:1991171",
	url= "https://doi.org/10.1051/jp1:1991171",
	journal = {J. Phys. I France},
	year = 1991,
	volume = 1,
	number = 6,
	pages = "809-836",
	month = "",
}

@article{Cugliandolo_mfd_out_eq_1996,
  title = {Large Time Out-of-Equilibrium Dynamics of a Manifold in a Random Potential},
  author = {Cugliandolo, Leticia F. and Kurchan, Jorge and Le Doussal, Pierre},
  journal = {Phys. Rev. Lett.},
  volume = {76},
  issue = {13},
  pages = {2390--2393},
  numpages = {0},
  year = {1996},
  month = {Mar},
  publisher = {American Physical Society},
  doi = {10.1103/PhysRevLett.76.2390},
  url = {https://link.aps.org/doi/10.1103/PhysRevLett.76.2390}
}

@article{ros2023generalized,
  title = {Generalized Lotka-Volterra Equations with Random, Nonreciprocal Interactions: The Typical Number of Equilibria},
  author = {Ros, Valentina and Roy, Felix and Biroli, Giulio and Bunin, Guy and Turner, Ari M.},
  journal = {Phys. Rev. Lett.},
  volume = {130},
  issue = {25},
  pages = {257401},
  numpages = {7},
  year = {2023},
  month = {Jun},
  publisher = {American Physical Society},
  doi = {10.1103/PhysRevLett.130.257401},
  url = {https://link.aps.org/doi/10.1103/PhysRevLett.130.257401}
}

@article{RosEcoQuenched2023,
doi = {10.1088/1751-8121/ace00f},
url = {https://dx.doi.org/10.1088/1751-8121/ace00f},
year = {2023},
month = {jul},
publisher = {IOP Publishing},
volume = {56},
number = {30},
pages = {305003},
author = {Ros, Valentina and Roy, Felix and Biroli, Giulio and Bunin, Guy},
title = {Quenched complexity of equilibria for asymmetric generalized Lotka–Volterra equations},
journal = {Journal of Physics A: Mathematical and Theoretical},
}

@article{BiroliLVMarginal2018,
doi = {10.1088/1367-2630/aada58},
url = {https://dx.doi.org/10.1088/1367-2630/aada58},
year = {2018},
month = {aug},
publisher = {IOP Publishing},
volume = {20},
number = {8},
pages = {083051},
author = {Giulio Biroli and Guy Bunin and Chiara Cammarota},
title = {Marginally stable equilibria in critical ecosystems},
journal = {New Journal of Physics},
}

@article{MAY_1972, title={Will a Large Complex System be Stable?}, volume={238}, DOI={10.1038/238413a0}, number={5364}, journal={Nature}, publisher={Springer Science and Business Media LLC}, author={May, Robert M.}, year={1972}, month={aug.}, pages={413–414} }

@article{fyodorov2016nonlinear,
author = {Yan V. Fyodorov  and Boris A. Khoruzhenko },
title = {Nonlinear analogue of the May−Wigner instability transition},
journal = {Proceedings of the National Academy of Sciences},
volume = {113},
number = {25},
pages = {6827-6832},
year = {2016},
doi = {10.1073/pnas.1601136113}}

@article{ben2021counting,
author = {Gérard Ben Arous  and Yan V. Fyodorov  and Boris A. Khoruzhenko },
title = {Counting equilibria of large complex systems by instability index},
journal = {Proceedings of the National Academy of Sciences},
volume = {118},
number = {34},
pages = {e2023719118},
year = {2021},
doi = {10.1073/pnas.2023719118},
URL = {https://www.pnas.org/doi/abs/10.1073/pnas.2023719118}}

@misc{lorenzana_non_rec_spin_aging_2024,
      title={Non-reciprocal spin-glass transition and aging}, 
      author={Giulia Garcia Lorenzana and Ada Altieri and Giulio Biroli and Michel Fruchart and Vincenzo Vitelli},
      year={2024},
      eprint={2408.17360},
      archivePrefix={arXiv},
      primaryClass={cond-mat.dis-nn},
      url={https://arxiv.org/abs/2408.17360}, 
}

@article{vitelli_non_rec_ising_2025,
  title = {Nonreciprocal Ising Model},
  author = {Avni, Yael and Fruchart, Michel and Martin, David and Seara, Daniel and Vitelli, Vincenzo},
  journal = {Phys. Rev. Lett.},
  volume = {134},
  issue = {11},
  pages = {117103},
  numpages = {9},
  year = {2025},
  month = {Mar},
  publisher = {American Physical Society},
  doi = {10.1103/PhysRevLett.134.117103},
  url = {https://link.aps.org/doi/10.1103/PhysRevLett.134.117103}
}

@misc{vitelli_when_is_nonrecip_2025,
      title={When is nonreciprocity relevant?}, 
      author={Giulia Garcia Lorenzana and David Martin and Yael Avni and Daniel S. Seara and Michel Fruchart and Giulio Biroli and Vincenzo Vitelli},
      year={2025},
      eprint={2509.17972},
      archivePrefix={arXiv},
      primaryClass={cond-mat.stat-mech},
      url={https://arxiv.org/abs/2509.17972}, 
}

@article{Fyod_simplest_optimization_2013, title={Topology Trivialization and Large Deviations for the Minimum in the Simplest Random Optimization}, volume={154}, doi={10.1007/s10955-013-0838-1}, number={1-2}, journal={Journal of Statistical Physics}, author={Fyodorov, Yan V. and Le Doussal, Pierre}, year={2014}, pages={466–490}}

@article{Fyodorov_least_sqaure_2022,
doi = {10.1088/1751-8121/ac6d8e},
url = {https://dx.doi.org/10.1088/1751-8121/ac6d8e},
year = {2022},
month = {may},
publisher = {IOP Publishing},
volume = {55},
number = {24},
pages = {244008},
author = {Fyodorov, Yan V and Tublin, Rashel},
title = {Optimization landscape in the simplest constrained random least-square problem},
journal = {Journal of Physics A: Mathematical and Theoretical},
}

@article{Montanari_landscape_spiked_2018,
author = {Ben Arous, Gerard and Mei, Song and Montanari, Andrea and Nica, Mihai},
year = {2017},
month = {11},
pages = {},
title = {The Landscape of the Spiked Tensor Model},
volume = {72},
journal = {Communications on Pure and Applied Mathematics},
doi = {10.1002/cpa.21861}
}

@inproceedings{Montanari_tensor_pca_2014,
author = {Montanari, Andrea and Richard, Emile},
title = {A statistical model for tensor PCA},
year = {2014},
publisher = {MIT Press},
address = {Cambridge, MA, USA},
booktitle = {Proceedings of the 28th International Conference on Neural Information Processing Systems - Volume 2},
pages = {2897–2905},
numpages = {9},
location = {Montreal, Canada},
series = {NIPS'14}
}

@InProceedings{maillard_landscape_2020,
  title = 	 {Landscape Complexity for the Empirical Risk of Generalized Linear Models},
  author =       {Maillard, Antoine and Ben Arous, G\'erard and Biroli, Giulio},
  booktitle = 	 {Proceedings of The First Mathematical and Scientific Machine Learning Conference},
  pages = 	 {287--327},
  year = 	 {2020},
  editor = 	 {Lu, Jianfeng and Ward, Rachel},
  volume = 	 {107},
  series = 	 {Proceedings of Machine Learning Research},
  month = 	 {20--24 Jul},
  publisher =    {PMLR},
  url = 	 {https://proceedings.mlr.press/v107/maillard20a.html},
}

@article{baity2018comparing,
doi = {10.1088/1742-5468/ab3281},
url = {https://doi.org/10.1088/1742-5468/ab3281},
year = {2019},
month = {dec},
publisher = {IOP Publishing and SISSA},
volume = {2019},
number = {12},
pages = {124013},
author = {Baity-Jesi, Marco and Sagun, Levent and Geiger, Mario and Spigler, Stefano and Ben Arous, Gérard and Cammarota, Chiara and LeCun, Yann and Wyart, Matthieu and Biroli, Giulio},
title = {Comparing dynamics: deep neural networks versus glassy systems*},
journal = {Journal of Statistical Mechanics: Theory and Experiment}
}

@InProceedings{mannelli_spiked_matrix_2019,
  title = 	 {Passed and Spurious: Descent Algorithms and Local Minima in Spiked Matrix-Tensor Models},
  author =       {Mannelli, Stefano Sarao and Krzakala, Florent and Urbani, Pierfrancesco and Zdeborova, Lenka},
  booktitle = 	 {Proceedings of the 36th International Conference on Machine Learning},
  pages = 	 {4333--4342},
  year = 	 {2019},
  editor = 	 {Chaudhuri, Kamalika and Salakhutdinov, Ruslan},
  volume = 	 {97},
  series = 	 {Proceedings of Machine Learning Research},
  month = 	 {09--15 Jun},
  publisher =    {PMLR}
}

@article{mannelli_pitfalls_2020,
  title = {Marvels and Pitfalls of the Langevin Algorithm in Noisy High-Dimensional Inference},
  author = {Sarao Mannelli, Stefano and Biroli, Giulio and Cammarota, Chiara and Krzakala, Florent and Urbani, Pierfrancesco and Zdeborov\'a, Lenka},
  journal = {Phys. Rev. X},
  volume = {10},
  issue = {1},
  pages = {011057},
  numpages = {41},
  year = {2020},
  month = {Mar},
  publisher = {American Physical Society},
  doi = {10.1103/PhysRevX.10.011057},
  url = {https://link.aps.org/doi/10.1103/PhysRevX.10.011057}
}

@inproceedings{mannelli_afraid_2019,
 author = {Sarao Mannelli, Stefano and Biroli, Giulio and Cammarota, Chiara and Krzakala, Florent and Zdeborov\'{a}, Lenka},
 booktitle = {Advances in Neural Information Processing Systems},
 editor = {H. Wallach and H. Larochelle and A. Beygelzimer and F. d\textquotesingle Alch\'{e}-Buc and E. Fox and R. Garnett},
 pages = {},
 publisher = {Curran Associates, Inc.},
 title = {Who is Afraid of Big Bad Minima? Analysis of gradient-flow in spiked matrix-tensor models},
 url = {https://proceedings.neurips.cc/paper_files/paper/2019/file/fbad540b2f3b5638a9be9aa6a4d8e450-Paper.pdf},
 volume = {32},
 year = {2019}
}

@misc{kent_algo_marg_2024,
      title={Algorithm-independent bounds on complex optimization through the statistics of marginal optima}, 
      author={Jaron Kent-Dobias},
      year={2024},
      eprint={2407.02092},
      archivePrefix={arXiv},
      primaryClass={cond-mat.dis-nn},
      url={https://arxiv.org/abs/2407.02092}, 
}

@misc{montanari_erm_localmin_2025,
      title={Local minima of the empirical risk in high dimension: General theorems and convex examples}, 
      author={Kiana Asgari and Andrea Montanari and Basil Saeed},
      year={2025},
      eprint={2502.01953},
      archivePrefix={arXiv},
      primaryClass={stat.ML},
      url={https://arxiv.org/abs/2502.01953}, 
}

@misc{tsironic_perceptron_kac_2025,
      title={Landscape Complexity for the Empirical Risk of Generalized Linear Models: Discrimination between Structured Data}, 
      author={Theodoros G. Tsironis and Aris L. Moustakas},
      year={2025},
      eprint={2503.14403},
      archivePrefix={arXiv},
      primaryClass={cs.LG},
      url={https://arxiv.org/abs/2503.14403}, 
}

@misc{fyod_non_herm_2025,
      title={Kac-Rice inspired approach to non-Hermitian random matrices}, 
      author={Yan V Fyodorov},
      year={2025},
      eprint={2506.21058},
      archivePrefix={arXiv},
      primaryClass={math-ph},
      url={https://arxiv.org/abs/2506.21058}, 
}

@article{Lacroix_Fyodorov_Fedeli_2022, title={Superposition of random plane waves in high spatial dimensions: Random matrix approach to landscape complexity}, volume={63}, DOI={10.1063/5.0086919}, number={9}, journal={Journal of Mathematical Physics}, publisher={AIP Publishing}, author={Lacroix-A-Chez-Toine, Bertrand and Fyodorov, Yan V. and Fedeli, Sirio Belga}, year={2022}, month={sep.}}

@misc{lacroix_fyod_waves_2024,
      title={Superposition of plane waves in high spatial dimensions: from landscape complexity to the deepest minimum value}, 
      author={Bertrand Lacroix-A-Chez-Toine and Yan V. Fyodorov},
      year={2024},
      eprint={2411.09687},
      archivePrefix={arXiv},
      primaryClass={cond-mat.dis-nn},
      url={https://arxiv.org/abs/2411.09687}, 
}

@incollection{jonsson1998nudged,
  title={Nudged elastic band method for finding minimum energy paths of transitions},
  author={J{\'o}nsson, Hannes and Mills, Greg and Jacobsen, Karsten W},
  booktitle={Classical and quantum dynamics in condensed phase simulations},
  pages={385--404},
  year={1998},
  publisher={World Scientific},
doi={https://doi.org/10.1142/9789812839664_0016}
}

@article{bolhuis2002transition,
doi = {10.1088/0953-8984/12/8A/316},
url = {https://doi.org/10.1088/0953-8984/12/8A/316},
year = {2000},
month = {feb},
publisher = {},
volume = {12},
number = {8A},
pages = {A147},
author = {Peter G Bolhuis and Christoph Dellago and Phillip L Geissler and David Chandler},
title = {Transition 
path sampling: throwing 
ropes over mountains in the dark},
journal = {Journal of Physics: Condensed Matter}
}

@inproceedings{freeman2016topology,
  title={Topology and Geometry of Half-Rectified Network Optimization},
  author={Freeman, C Daniel and Bruna, Joan},
  booktitle={International Conference on Learning Representations},
  year={2016}
}

@article{wales1998archetypal,
  title={Archetypal energy landscapes},
  author={Wales, David J and Miller, Mark A and Walsh, Tiffany R},
  journal={Nature},
  volume={394},
  number={6695},
  pages={758--760},
  year={1998},
  publisher={Nature Publishing Group},
doi={https://doi.org/10.1038/29487}
}

@inproceedings{draxler2018essentially,
  title = 	 {Essentially No Barriers in Neural Network Energy Landscape},
  author =       {Draxler, Felix and Veschgini, Kambis and Salmhofer, Manfred and Hamprecht, Fred},
  booktitle = 	 {Proceedings of the 35th International Conference on Machine Learning},
  pages = 	 {1309--1318},
  year = 	 {2018},
  volume = 	 {80},
  series = 	 {Proceedings of Machine Learning Research},
  month = 	 {10--15 Jul},
  publisher =    {PMLR}
}

@article{mauri2022mutational,
  title = {Mutational Paths with Sequence-Based Models of Proteins: From Sampling to Mean-Field Characterization},
  author = {Mauri, Eugenio and Cocco, Simona and Monasson, R\'emi},
  journal = {Phys. Rev. Lett.},
  volume = {130},
  issue = {15},
  pages = {158402},
  numpages = {6},
  year = {2023},
  month = {Apr},
  publisher = {American Physical Society},
  doi = {10.1103/PhysRevLett.130.158402},
  url = {https://link.aps.org/doi/10.1103/PhysRevLett.130.158402}
}

@article{tian2020exploring,
  title={Exploring the sequence fitness landscape of a bridge between protein folds},
  author={Tian, Pengfei and Best, Robert B},
  journal={PLoS computational biology},
  volume={16},
  number={10},
  pages={e1008285},
  year={2020},
  publisher={Public Library of Science San Francisco, CA USA},
doi={https://doi.org/10.1371/journal.pcbi.1008285}
}

@article{debenedetti2001supercooled,
  title={Supercooled liquids and the glass transition},
  author={Debenedetti, Pablo G and Stillinger, Frank H},
  journal={Nature},
  volume={410},
  number={6825},
  pages={259--267},
  year={2001},
  publisher={Nature Publishing Group UK London},
  doi={https://doi.org/10.1038/35065704}
}

@article{annesi2023star,
  title = {Star-Shaped Space of Solutions of the Spherical Negative Perceptron},
  author = {Annesi, Brandon Livio and Lauditi, Clarissa and Lucibello, Carlo and Malatesta, Enrico M. and Perugini, Gabriele and Pittorino, Fabrizio and Saglietti, Luca},
  journal = {Phys. Rev. Lett.},
  volume = {131},
  issue = {22},
  pages = {227301},
  numpages = {7},
  year = {2023},
  month = {Nov},
  publisher = {American Physical Society},
  doi = {10.1103/PhysRevLett.131.227301},
  url = {https://link.aps.org/doi/10.1103/PhysRevLett.131.227301}
}

@article{garipov2018loss,
  title={Loss surfaces, mode connectivity, and fast ensembling of dnns},
  author={Garipov, Timur and Izmailov, Pavel and Podoprikhin, Dmitrii and Vetrov, Dmitry P and Wilson, Andrew G},
  journal={Advances in neural information processing systems},
  volume={31},
  year={2018}
}

@article{xu2010anharmonic,
doi = {10.1209/0295-5075/90/56001},
url = {https://doi.org/10.1209/0295-5075/90/56001},
year = {2010},
month = {jun},
publisher = {},
volume = {90},
number = {5},
pages = {56001},
author = {Xu, N. and Vitelli, V. and Liu, A. J. and Nagel, S. R.},
title = {Anharmonic and quasi-localized vibrations in jammed solids—Modes for mechanical failure},
journal = {Europhysics Letters}
}

@article{widmer2008irreversible,
  title={Irreversible reorganization in a supercooled liquid originates from localized soft modes},
  author={Widmer-Cooper, Asaph and Perry, Heidi and Harrowell, Peter and Reichman, David R},
  journal={Nature Physics},
  volume={4},
  number={9},
  pages={711--715},
  year={2008},
  publisher={Nature Publishing Group UK London},
doi={https://doi.org/10.1038/nphys1025}
}

@article{franz2013quasi,
  title={Quasi-equilibrium in glassy dynamics: an algebraic view},
  author={Franz, Silvio and Parisi, Giorgio},
  journal={Journal of Statistical Mechanics: Theory and Experiment},
  volume={2013},
  number={02},
  pages={P02003},
  year={2013},
  publisher={IOP Publishing},
doi={10.1088/1742-5468/2013/02/P02003}
}

@article{cugliandolo2011effective,
  title={The effective temperature},
  author={Cugliandolo, Leticia F},
  journal={Journal of Physics A: Mathematical and Theoretical},
  volume={44},
  number={48},
  pages={483001},
  year={2011},
  publisher={IOP Publishing},
doi={10.1088/1751-8113/44/48/483001}
}

@article{scalliet2022thirty,
  title = {Thirty Milliseconds in the Life of a Supercooled Liquid},
  author = {Scalliet, Camille and Guiselin, Benjamin and Berthier, Ludovic},
  journal = {Phys. Rev. X},
  volume = {12},
  issue = {4},
  pages = {041028},
  numpages = {28},
  year = {2022},
  month = {Dec},
  publisher = {American Physical Society},
  doi = {10.1103/PhysRevX.12.041028},
  url = {https://link.aps.org/doi/10.1103/PhysRevX.12.041028}
}

@article{ninarello2017models,
  title = {Models and Algorithms for the Next Generation of Glass Transition Studies},
  author = {Ninarello, Andrea and Berthier, Ludovic and Coslovich, Daniele},
  journal = {Phys. Rev. X},
  volume = {7},
  issue = {2},
  pages = {021039},
  numpages = {22},
  year = {2017},
  month = {Jun},
  publisher = {American Physical Society},
  doi = {10.1103/PhysRevX.7.021039},
  url = {https://link.aps.org/doi/10.1103/PhysRevX.7.021039}
}

@article{korchinski_thermal_2025,
  title = {Thermal Avalanches Drive Logarithmic Creep in Disordered Media},
  author = {Korchinski, Daniel J. and Shohat, Dor and Lahini, Yoav and Wyart, Matthieu},
  journal = {Phys. Rev. X},
  volume = {15},
  issue = {3},
  pages = {031024},
  numpages = {12},
  year = {2025},
  month = {Jul},
  publisher = {American Physical Society},
  doi = {10.1103/x7rr-vxnr},
  url = {https://link.aps.org/doi/10.1103/x7rr-vxnr}
}

@article{liu2018creep,
  title={Creep dynamics of athermal amorphous materials: a mesoscopic approach},
  author={Liu, Chen and Ferrero, Ezequiel E and Martens, Kirsten and Barrat, Jean-Louis},
  journal={Soft matter},
  volume={14},
  number={41},
  pages={8306--8316},
  year={2018},
  publisher={Royal Society of Chemistry},
doi={https://doi.org/10.1039/C8SM01392F}
}

@article{ferrero2017spatiotemporal,
  title = {Spatiotemporal Patterns in Ultraslow Domain Wall Creep Dynamics},
  author = {Ferrero, Ezequiel E. and Foini, Laura and Giamarchi, Thierry and Kolton, Alejandro B. and Rosso, Alberto},
  journal = {Phys. Rev. Lett.},
  volume = {118},
  issue = {14},
  pages = {147208},
  numpages = {6},
  year = {2017},
  month = {Apr},
  publisher = {American Physical Society},
  doi = {10.1103/PhysRevLett.118.147208},
  url = {https://link.aps.org/doi/10.1103/PhysRevLett.118.147208}
}

@article{durin2024earthquakelike,
  title = {Earthquakelike dynamics in ultrathin magnetic films},
  author = {Durin, Gianfranco and Schimmenti, Vincenzo Maria and Baiesi, Marco and Casiraghi, Arianna and Magni, Alessandro and Herrera-Diez, Liza and Ravelosona, Dafin\'e and Foini, Laura and Rosso, Alberto},
  journal = {Phys. Rev. B},
  volume = {110},
  issue = {2},
  pages = {L020405},
  numpages = {5},
  year = {2024},
  month = {Jul},
  publisher = {American Physical Society},
  doi = {10.1103/PhysRevB.110.L020405},
  url = {https://link.aps.org/doi/10.1103/PhysRevB.110.L020405}
}

@article{anderson1964hard,
  title={Hard superconductivity: theory of the motion of Abrikosov flux lines},
  author={Anderson, Philip W and Kim, YB},
  journal={Reviews of modern physics},
  volume={36},
  number={1},
  pages={39},
  year={1964},
  publisher={APS},
doi={https://doi.org/10.1103/RevModPhys.36.39}
}

@article{tahaei2023scaling,
  title = {Scaling Description of Dynamical Heterogeneity and Avalanches of Relaxation in Glass-Forming Liquids},
  author = {Tahaei, Ali and Biroli, Giulio and Ozawa, Misaki and Popovi\ifmmode \acute{c}\else \'{c}\fi{}, Marko and Wyart, Matthieu},
  journal = {Phys. Rev. X},
  volume = {13},
  issue = {3},
  pages = {031034},
  numpages = {20},
  year = {2023},
  month = {Sep},
  publisher = {American Physical Society},
  doi = {10.1103/PhysRevX.13.031034},
  url = {https://link.aps.org/doi/10.1103/PhysRevX.13.031034}
}

@article{de2024dynamical,
  title = {Dynamical heterogeneities of thermal creep in pinned interfaces},
  author = {de Geus, Tom W. J. and Rosso, Alberto and Wyart, Matthieu},
  journal = {Phys. Rev. E},
  volume = {111},
  issue = {1},
  pages = {L013503},
  numpages = {6},
  year = {2025},
  month = {Jan},
  publisher = {American Physical Society},
  doi = {10.1103/PhysRevE.111.L013503},
  url = {https://link.aps.org/doi/10.1103/PhysRevE.111.L013503}
}

@article{tapias2020entropic,
  title={From entropic to energetic barriers in glassy dynamics: The Barrat--M{\'e}zard trap model on sparse networks},
  author={Tapias, Diego and Paprotzki, Eva and Sollich, Peter},
  journal={Journal of Statistical Mechanics: Theory and Experiment},
  volume={2020},
  number={9},
  pages={093302},
  year={2020},
  publisher={IOP Publishing},
doi={10.1088/1742-5468/abaecf}
}

@article{khomenko2021relationship,
  title = {Relationship between two-level systems and quasilocalized normal modes in glasses},
  author = {Khomenko, Dmytro and Reichman, David R. and Zamponi, Francesco},
  journal = {Phys. Rev. Mater.},
  volume = {5},
  issue = {5},
  pages = {055602},
  numpages = {7},
  year = {2021},
  month = {May},
  publisher = {American Physical Society},
  doi = {10.1103/PhysRevMaterials.5.055602},
  url = {https://link.aps.org/doi/10.1103/PhysRevMaterials.5.055602}
}

@article{coslovich2019localization,
	title={{A localization transition underlies the mode-coupling crossover of glasses}},
	author={Daniele Coslovich and Andrea Ninarello and Ludovic Berthier},
	journal={SciPost Phys.},
	volume={7},
	pages={077},
	year={2019},
	publisher={SciPost},
	doi={10.21468/SciPostPhys.7.6.077},
	url={https://scipost.org/10.21468/SciPostPhys.7.6.077},
}

@article{baity2021revisiting,
  title={Revisiting the concept of activation in supercooled liquids},
  author={Baity-Jesi, Marco and Biroli, Giulio and Reichman, David R},
  journal={The European Physical Journal E},
  volume={44},
  number={6},
  pages={77},
  year={2021},
  publisher={Springer},
  doi={https://doi.org/10.1140/epje/s10189-021-00077-y}
}

@article{heuer2008exploring,
  title={Exploring the potential energy landscape of glass-forming systems: from inherent structures via metabasins to macroscopic transport},
  author={Heuer, Andreas},
  journal={Journal of Physics: Condensed Matter},
  volume={20},
  number={37},
  pages={373101},
  year={2008},
  publisher={IOP Publishing},
  doi={10.1088/0953-8984/20/37/373101}
}

@article{ioffe1987dynamics,
  title={Dynamics of interfaces and dislocations in disordered media},
  author={Ioffe, Lev B and Vinokur, Valerii M},
  journal={Journal of Physics C: Solid State Physics},
  volume={20},
  number={36},
  pages={6149},
  year={1987},
  publisher={IOP Publishing},
doi={10.1088/0022-3719/20/36/016}
}

@InProceedings{berfin_geometry_2021,
  title = 	 {Geometry of the Loss Landscape in Overparameterized Neural Networks: Symmetries and Invariances},
  author =       {Simsek, Berfin and Ged, Fran{\c{c}}ois and Jacot, Arthur and Spadaro, Francesco and Hongler, Clement and Gerstner, Wulfram and Brea, Johanni},
  booktitle = 	 {Proceedings of the 38th International Conference on Machine Learning},
  pages = 	 {9722--9732},
  year = 	 {2021},
  editor = 	 {Meila, Marina and Zhang, Tong},
  volume = 	 {139},
  series = 	 {Proceedings of Machine Learning Research},
  month = 	 {18--24 Jul},
  publisher =    {PMLR}
}

@misc{martinelli_flat_channels_2025,
      title={Flat Channels to Infinity in Neural Loss Landscapes}, 
      author={Flavio Martinelli and Alexander Van Meegen and Berfin Şimşek and Wulfram Gerstner and Johanni Brea},
      year={2025},
      eprint={2506.14951},
      archivePrefix={arXiv},
      primaryClass={cs.LG},
      url={https://arxiv.org/abs/2506.14951}
}


\cleardoublepage
\setcounter{chapter}{6}    
\renewcommand{\thechapter}{A}  
\chapter*{Appendix}        
\addcontentsline{toc}{chapter}{Appendix}  
\markboth{APPENDIX}{Appendix}             

\renewcommand{\thesection}{\thechapter.\arabic{section}}
\setcounter{section}{0}


\section{Gaussian conditioning}
\label{app:sec_gausian_conditioning}
\noindent We begin by recalling the formula for Gaussian conditioning.
Given a multivariate normal vector distributed as $Z\sim\mathcal{N}({\bm \mu}, \Sigma)$, and given the  partition:
\begin{align*}
    \begin{split}
        &Z=\begin{bmatrix}
            X\\
            Y
        \end{bmatrix}\quad\quad{\bm\mu}=\begin{bmatrix}
            \mu_X\\
            \mu_Y
        \end{bmatrix}\quad\quad\Sigma=\begin{bmatrix}
            \Sigma_{XX} & \Sigma_{XY}\\
            \Sigma_{YX} & \Sigma_{YY}
        \end{bmatrix},
    \end{split}
\end{align*}
then the conditional law of $(X|Y={\bf y})$ is a multivariate normal $\mathcal{N}(\tilde{{\bm\mu}},\tilde{\Sigma})$ with parameters 
\begin{align}
\begin{split}
&\tilde{{\bm\mu}}={\bm\mu}_X+\Sigma_{XY}\Sigma_{YY}^{-1}({\bf y}-{\bm\mu}_Y)\\
&\tilde{\Sigma}=\Sigma_{XX}-\Sigma_{XY}\Sigma_{YY}^{-1}\Sigma_{YX}.
\end{split}
\end{align}

\section{Dynamical mean-field theory 
 of Gaussian fields}\label{app:dynamical_calculations}
 
Taking inspiration from Ref.~\cite{HeliasBook20}, we summarize the derivation of the DMFT equations.\\

\noindent We write our dynamical system as
\begin{align}
\label{app:dynamical_example}
    \partial_t{\bf x}(t)={\bf F}({\bf x}(t))+{\bm\eta}(t),\quad\quad{\bf x}\in\mathbb{R}^N
\end{align}
where ${\bf F}$ is the random Gaussian field in Eq.~\ref{eqapp:mod} of Chapter~\ref{chapter:non_reciprocal} and ${\bm\eta}$ is a mean-zero Gaussian white noise with $\langle \eta_i(t)\eta_j(t')\rangle_\eta=2T\delta_{ij}\delta(t-t')$. We also assume that we start at ${\bf x}(0)={\bf a}$. There are two common ways to pass from the continuous to the discretized version of this dynamical system: the Itô and the Stratonovich conventions. In the case of additive noise (i.e. $\eta$ does not depend on $x$), the two give the same continuous time limit \cite{HeliasBook20}. Consider partitioning time in steps $\Delta t$, then by integrating in time the LHS and RHS of \eqref{app:dynamical_example} we get, in the Itô convention (which implies evaluating functions at the current time only):
\begin{align}
&\int_t^{t+\Delta t}\partial_\tau{\bf x}(\tau)d\tau=\int_t^{t+\Delta t}{\bf F}({\bf x})d\tau+\int_t^{t+\Delta t}{\bm \eta}(\tau)d\tau\\
\Rightarrow& \begin{cases}{\bf x}(t+\Delta t)={\bf x}(t)+{\bf F}({\bf x}(t))\Delta t+\sqrt{2T\Delta t}\,\,{\bm \xi} \\
{\bf x}(0)={\bf a}\,\delta_{t0}
\end{cases}
\end{align}
with ${\bm\xi}\sim\mathcal{N}(0,\mathbb{I})$ and it is meant that at each time step we have a new ${\bm\xi}$. In the last line we used that the integral of the white noise is Gaussian (as a sum of Gaussians), and has zero-mean and variance given by:
\begin{align}
    \left\langle\int_t^{t+\Delta t}\int_t^{t+\Delta t}{\bf \eta}_i(\tau)\eta_j(\tau') \right\rangle_\eta=2T\delta_{ij}\Delta t,
\end{align}
which implies that 
\begin{align}
    \int_t^{t+\Delta t}{\bm\eta}(\tau)d\tau\sim \mathcal{N}(0, 2T\Delta t\mathbb{I})\sim \sqrt{2T\Delta t}\,\,{\bm\xi}.
\end{align}
In short, we may write the above discretized system as:
\begin{align}
    {\bf x}_{l+1}={\bf x}_l+{\bf F}({\bf x}_l)\Delta t+\sqrt{2T\Delta t}\,\,{\bm\xi}_l\equiv {\bf y}({\bf x}_l,{\bm\xi}_l)
\end{align}
for $l=0,\cdots,t/\Delta t$ with ${\bf x}_0={\bf a}$, and ${\bm\xi}_l\overset{iid}{\sim}\mathcal{N}(0,\mathbb{I})$. Now, we can express the joint probability of a certain path $\{{\bf x}_l\}_{l=1,\cdots,t/\Delta t }$ (in short $\{{\bf x}(t)\}$) as the sum over all realizations of the white noise that generate such path:
\begin{align}
    p(\{{\bf x}(t)\})=\int \prod_{l=0}^{t/\Delta t} d{\bm\xi}_{l}\,\rho({\bm\xi}_l)\,\delta({\bf x}_{l+1}-{\bf y}({\bf x}_{l},{\bm\xi}_l)),
\end{align}
where $\rho$ is the pdf \footnote{probability density function} of each Gaussian vector ${\bm\xi}$. Notice that ${\bf y}$ only depends on ${\bf x}_l$ and not on the previous history, a property called Markov property of the process \cite{HeliasBook20}. The next step consists in opening up the Dirac deltas and introduce their Fourier representation:
\begin{align}
    \delta({\bf x})=\frac{1}{(2\pi i)^N}\int_{-i\infty}^{i\infty}d\tilde{{\bf x}}\,e^{{\bf x}\cdot\tilde{{\bf x}}}
\end{align}
where $\tilde{{\bf x}}$ is called the response field. This gives:
\begin{align*}
p(\{{\bf x}(t)\})&=\prod_l\int d{\bm\xi}_l\,\rho({\bm\xi}_l)\, \int_{-i\infty}^{i\infty} \frac{d\tilde{{\bf x}}_l}{(2\pi i)^N}e^{\tilde{{\bf x}}_l\cdot ({\bf x}_{l+1}-{\bf x}_l-{\bf F}({\bf x}_l)\Delta t-\sqrt{2T\Delta t}\,{\bm\xi}_l)}\\
&=\prod_l\int_{-i\infty}^{i\infty} \frac{d\tilde{{\bf x}}_l}{(2\pi i)^N}e^{\tilde{{\bf x}}_l\cdot ({\bf x}_{l+1}-{\bf x}_l-{\bf F}({\bf x}_l)\Delta t)}e^{W[\tilde{\bf x}_l]}
\end{align*}
where 
\begin{align}
    W[\tilde{{\bf x}}_l]=\ln\left\langle e^{-\sqrt{2T\Delta t}\,\,\tilde{{\bf x}}_l\cdot{\bm\xi}_l}\right\rangle_{\bm \xi}=\ln\left[e^{T\Delta t\tilde{{\bf x}}_l^2}\right]=T\Delta t\,\tilde{{\bf x}}_l^2.
\end{align}
Here we also see why the factor 2 multiplying $T$ was necessary, so that now we get rid of it in the path integral. This representation of $p({\bf x})$ takes the name of MSRDJ path integral \cite{HeliasBook20}. By abusing notation, in the limit $\Delta t\to 0$, we may write the probability distribution of a path $\{{\bf x}(\tau)\}_{\tau\leq t}$ (or, in short $\{{\bf x}(t)\}$) starting from ${\bf x}(0)={\bf a}$ as 
\begin{align}
    p(\{{\bf x}(t)\})=\int \mathcal{D}\tilde{{\bf x}}\,e^{\int_0^t ds\,\tilde{{\bf x}}\cdot[\partial_s{\bf x}-{\bf F}({\bf x})]}e^{T\int_0^t ds\, \tilde{{\bf x}}^2 }
\end{align}
where we neglect for simplicity the function's arguments, but it is intended that all the fields are functions of time, and $\mathcal{D}$ contains all normalization factors. We recall that $\mathbb{E}$ indicates the average over the quenched disorder (we haven't used it yet), $\langle\cdot\rangle_\eta$ indicates the average with respect to the gaussian white noise; we define then the average with respect to the probability distribution of paths up to time $t$ as:
\begin{align}
    \langle \cdot \rangle=\int \mathcal{D}{\bf x}\,\,p(\{{\bf x}(t)\})\,\,\cdot
\end{align}
where $\cdot$ indicates any possible functional of ${\bf x}$. In particular, notice that 
\begin{align}
    1=\int \mathcal{D}{\bf x}\int \mathcal{D}\tilde{{\bf x}}\,\,\text{exp}\left[\int dt\left(\, S_0[{\bf x},\tilde{{\bf x}}]-{\bf f}({\bf x})\cdot\tilde{\bf x}\right)\right]
\end{align}
where we have defined $S_0$ to be the action containing only single unit properties, thus excluding the interaction term ${\bf f}$:
\begin{align}
    S_0[{\bf x}(t),\tilde{{\bf x}}(t)]=\tilde{\bf x}(t)\cdot[\partial_t+\lambda({\bf x})]{\bf x}(t)+T\,\tilde{\bf x}^2(t).
\end{align}
The idea underlying the next few steps is that, by taking the quenched average, we change from the variables ${\bf x},\tilde{\bf x}$ to two new fields $Q_1,Q_2$. We will then take a saddle point over these two fields, thus calculating the maximum contribution to the probability mass. By doing so, we will see that the interaction terms between single units decouple, and we get an action over $N$ identical units with an effective self-consistently determined noise term. The idea for doing this is that in the limit $N\to \infty$ all quantities of interest (response functions, autocorrelation functions) are self-averaging and thus converge to their mean values irrespective of the quenched disorder: by taking an average over the quenched disorder each unit will feel the same (effective) noise as any other unit and we can therefore find these average values. We thus obtain:
\begin{align*}
    1&=\int \mathcal{D}{\bf x}\int \mathcal{D}\tilde{{\bf x}}\,\,\mathbb{E}\,\,\text{exp}\left[\int dt\left(\, S_0[{\bf x},\tilde{{\bf x}}]-{\bf f}({\bf x})\cdot\tilde{{\bf x}}\right)\right]=\int \mathcal{D}{\bf x}\int \mathcal{D}\tilde{{\bf x}}\,\,\text{exp}\left[\int dt\, S_0[{\bf x},\tilde{{\bf x}}]\right]\mathbb{E}\left[e^{-\int dt\,{\bf f}({\bf x}(t))\cdot\tilde{\bf x}(t)}\right].
\end{align*}
Now, recall that for a generic Gaussian vector ${\bf X}\sim\mathcal{N}({\bm\mu},{\bf \Sigma})$ we have that the moment generating function reads:
\begin{align}
    \mathbb{E}\left[e^{{\bf X}^\top{\bf t}}\right]=e^{{\bm\mu}^\top{\bf t}+\frac{1}{2}{\bf t}^\top{\bm\Sigma}\,{\bf t}}
\end{align}
for any constant vector ${\bf t}$. Applying this in the case where we have an integral in the exponent (by discretizing time), it follows:
\begin{align*}
\mathbb{E}\left[e^{-\int dt\,{\bf f}({\bf x}(t))\cdot\tilde{\bf x}(t)}\right]&=\text{exp}\left[-J\int dt\,m({\bf x}(t)){\bf 1}\cdot\tilde{\bf x}+\frac{1}{2}\int dt\,ds\,\tilde{\bf x}^\top(t){\bf \Sigma}(t,s)\tilde{\bf x}(s)\right]
\end{align*}
where we used the shortcut ${\bf\Sigma}_{ij}(t,s)\equiv\text{Cov}[f_i({\bf x}(t)),f_j({\bf x}(s))]$. We further compute the last term, which reads:
\begin{align*}
\text{exp}&\left[\frac{1}{2}\int dt\,ds\,\tilde{\bf x}^\top(t){\bf \Sigma}(t,s)\tilde{\bf x}(s)\right]=\text{exp}\left[\frac{1}{2}\int dt\,ds\,\tilde{\bf x}^\top(t)\tilde{\bf x}(s)\Phi_1\left(\frac{{\bf x}(t)\cdot{\bf x}(s)}{N}\right)\right]\times \\&\times\text{exp}\left[\frac{1}{2}\int dt\,ds\,\frac{\tilde{\bf x}^\top(t){\bf x}(s)}{N}\tilde{\bf x}^\top(s){\bf x}(t)\Phi_2\left(\frac{{\bf x}(t)\cdot{\bf x}(s)}{N}\right)\right].
\end{align*}
Let us now introduce the following order parameters:
\begin{align}
\begin{split}
&Q_1(t,s)=\frac{{\bf x}(t)\cdot{\bf x}(s)}{N}\\
&Q_2(t,s)=\frac{\tilde{\bf x}(t)\cdot\tilde{\bf x}(s)}{N}\\
&Q_3(t,s)=\frac{\tilde{\bf x}(t)\cdot{\bf x}(s)}{N}\\
&Q_4(t,s)=\frac{{\bf x}(t)\cdot\tilde{\bf x}(s)}{N}\\
\end{split}
\end{align}
which we enforce through Dirac deltas:
\begin{align*}
&\delta(-NQ_1(t,s)+{\bf x}(t)\cdot{\bf x}(s))\cdots\delta(-NQ_4(t,s)+{\bf x}(t)\cdot\tilde{\bf x}(s))\\=&\int \mathcal{D}\lambda_1\cdots \mathcal{D}\lambda_4\,e^{\int dt\,ds\,\lambda_1(-NQ_1(t,s)+{\bf x}(t)\cdot{\bf x}(s))}e^{\ldots}e^{\int dt\,ds\,\lambda_4(-NQ_4(t,s)+{\bf x}(t)\cdot\tilde{\bf x}(s))}
\end{align*}
By plugging this into our equation for the number 1 we get:
\begin{align}
\label{app:eq:one}
1=\int\prod_{i=1}^4\mathcal{D}Q_i\mathcal{D}\lambda_i\,\text{exp}\left\{N\left[-\sum_{i=1}^4Q_i\lambda_i+\ln Z +\frac{1}{2}Q_2\Phi_1(Q_1)+\frac{1}{2}Q_3Q_4\Phi_2(Q_1)\right]\right\}
\end{align}
where we see that, by introducing the order parameters, the integrals over ${\bf x}$ and $\tilde{\bf x}$ decouple and can be factorized as a unique integral to the power $N$, which we grouped within the variable $Z$:
\begin{align}
Z=\int\mathcal{D}x\mathcal{D}\tilde{x}\,\text{exp}\left\{S_0[x,\tilde{x}]-(J\,m )\,\tilde{x}+x\lambda_1x+\tilde{x}\lambda_2\tilde{x}+\tilde{x}\lambda_3x+x\lambda_4\tilde{x}\right\}
\end{align}
where all integrals have been omitted but are clear from the context. Now, we wish to maximize the action to get the leading order parameters that contribute to the probability mass as $N\to\infty$. We thus get at the Saddle Point:
\begin{align}
\begin{split}
&Q_1(t,s)=C(t,s)\\
&Q_2(t,s)=0\\
&Q_3(t,s)=R(s,t)\\
&Q_4(t,s)=R(t,s)
\end{split}
\end{align}
and 
\begin{align}
\begin{split}
&\lambda_1(t,s)=\frac{1}{2}Q_2(t,s)\Phi_1'(Q_1(t,s))+\frac{1}{2}Q_3(t,s)Q_4(t,s)\Phi_2'(Q_1(t,s))=0\\
&\lambda_2(t,s)=\frac{1}{2}\Phi_1(Q_1(t,s))\\
&\lambda_3(t,s)=\frac{1}{2}Q_4(t,s)\Phi_2(Q_1(t,s))\\
&\lambda_4(t,s)=\frac{1}{2}Q_3(t,s)\Phi_2(Q_1(t,s))
\end{split}
\end{align}
Now, if we plug back these results into \eqref{app:eq:one}, we have that the terms outside of $\ln Z$ cancel, and we are thus left with 
\begin{align}
    1\overset{N\to\infty}{\approx}(Z^*)^N
\end{align}
with 
\begin{align*}
Z^*=\int \mathcal{D}x\mathcal{D}\tilde{x}\,&\text{exp}\Bigg\{\int dt \Bigg(S_0[x(t),\tilde{x}(t)]-\tilde{x}(t)(J\,m)+\int ds\, \frac{1}{2}\tilde{x}(t)\Phi_1(C(t,s))\tilde{x}(s)\\
&+\tilde{x}(t)\int_0^t ds\,R(t,s)\Phi_2(C(t,s))x(s)\Bigg)\Bigg\}
\end{align*}
and we recall that:
\begin{align}
    S_0[x(t),\tilde{x}(t)]=\tilde{x}(t)[\partial_t+\lambda(x)]x(t)+T\int ds\,\tilde{x}(t)\tilde{x}(s)\delta(t-s).
\end{align}
By introducing the $\delta$ function in the white noise term, we can now match this action with the new effective action:
\begin{align}
    \tilde{S}[x(t),\tilde{x}(t)]=\tilde{x}(t)[\partial_t+\lambda(x)]x(t)+\int ds\,\tilde{x}(t)\tilde{x}(s)\left[T\delta(t-s)+\frac{1}{2}\Phi_1(C(t,s))\right]
\end{align}
This finally gives us a SDE for one single unit with self-consistently determined noise:
\begin{align}
\begin{split}
&\partial_tx(t)=-\lambda(x)x(t)+\int_0^t ds\,R(t,s)\Phi_2(C(t,s))x(s)+Jm(t)+\eta(t)\\
&\langle\eta(t)\eta(s)\rangle=2T\delta(t-s)+\Phi_1(C(t,s))
\end{split}.
\end{align}


\subsection{Computation of $\langle \eta(t)x(t')\rangle$}
\label{app:comp_eta_x_t_tp}
The probability distribution of the Gaussian noise $\eta$ corresponds to the continuum limit of a discretized process over time steps $\eta_i$. The integration measure over all paths, denoted by $\mathcal{D}\eta$ corresponds to the continuum limit of the product $\prod d\eta_i$ at each time-step. For simplicity, let us denote $\langle \eta(t)\eta(t')\rangle=\Sigma(t,t')$. Then, it holds:
\begin{align}
    \int dt'\, \Sigma(t,t')\Sigma^{-1}(t',t'')=\delta(t-t'').
\end{align}
We also write compactly the probability distribution as
\begin{align}
    P[\eta]\propto e^{-\frac{1}{2}\int dt dt' \Sigma^{-1}(t,t')\eta(t)\eta(t')}.
\end{align}
Now, let us consider the following identity:
\begin{align}
    \int \mathcal{D}\eta\,\frac{\delta}{\delta\eta(t)}(x(t')P[\eta])=0,
\end{align}
which follows from the fact that the Gaussian distribution decays to zero at infinity. This identity can be expanded and recast in the following form:
\begin{align}
    0=\int\mathcal{D}\eta\left[\frac{\delta x(t')}{\delta \eta(t)}P[\eta]+x(t')\frac{\delta P[\eta]}{\delta \eta(t)}\right]=R(t',t)+\int \mathcal{D}\eta\, x(t')\frac{\delta P[\eta]}{\delta\eta(t)}.
\end{align}
Now, let us analyse the second term. Since $P$ is a Gaussian distribution, we directly find
\begin{align}
    \frac{\delta P[\eta]}{\delta\eta(t)}=-\int dt''\,\Sigma^{-1}(t,t'')\eta(t'')P[\eta],
\end{align}
which ultimately implies
\begin{align*}
&R(t',s)=-\int \mathcal{D}\eta\, x(t')\frac{\delta P[\eta]}{\delta\eta(s)}=\int dt'' \,\Sigma^{-1}(s,t'')\langle \eta(t'')x(t')\rangle \\
\Rightarrow &\int ds\, R(t',s)\Sigma(t,s)=\int dt''\, ds\,\Sigma(t,s)\Sigma^{-1}(s,t'')\langle \eta(t'')x(t')\rangle\\
\Rightarrow &\int ds\, R(t',s)\Sigma(t,s)=\int dt''\delta(t-t'')\langle\eta(t'')x(t')\rangle=\langle\eta(t)x(t')\rangle\\
\Rightarrow& \int_0^{t'}ds\,R(t',s)\Sigma(t,s)=\langle\eta(t)x(t')\rangle 
\end{align*}

\subsection{DMFT equations}
\label{app:dmft_equations}
We recall that the autocorrelation and response functions are given by:
\begin{align}
    C(t,t')=\langle x(t),x(t')\rangle,\quad R(t,t')=\left\langle\frac{\delta x(t)}{\delta\eta(t')}\right\rangle,\quad m(t)=\langle x(t)\rangle.
\end{align}
The equation for the autocorrelation function is obtained by multiplying the SDE at time $t$ by $x(t')$, and taking the average:
\begin{align}
    \partial_tC(t,t')=-\lambda(t)C(t,t')+\int_0^tds\,R(t,s)\Phi_2(C(t,s))C(t',s)+Jm(t)m(t')+\langle \eta(t)x(t')\rangle.
\end{align}
The equation for $R$ is easier to find, as we just need to take the derivative of the SDE with respect to $\eta(t')$, and then take an average:
\begin{align}
\partial_t R(t,t')=-\lambda(t)R(t,t')+\int_{t'}^tds\,R(t,s)\Phi_2(C(t,s))R(s,t')+\delta(t-t')
\end{align}
where now the integral starts from $t'\leq t$ to respect causality. Finally, the equation for $m(t)$ is simply obtained by taking the average in the SDE:
\begin{align}
\partial_t m(t)=-\lambda(t)m(t)+\int_0^t ds\,R(t,s)\Phi_2(C(t,s))m(s)+Jm(t).
\end{align}
Regarding $\lambda(t)$, as we have seen in Chapter~\ref{chapter:non_reciprocal}, either it is a confining potential
\begin{align}
    \lambda(t)=\lambda(C(t,t)),
\end{align}
or it imposes a spherical constraint, that is, $C(t,t)=1$ at all times, and $\partial_tx^2(t)=0$. In order to obtain a formula for $\lambda$, we multiply the LHS and RHS of the SDE by $x(t)$, and exploit the fact that, according to Itô's prescription (and imposing directly the spherical constraint):

\begin{align}
    0=\partial_tx^2(t)=2x(t)\partial_tx(t)+2T\Rightarrow x(t)\partial_t x(t)=-T.
\end{align}

\noindent This imposes the following expression for $\lambda$:
\begin{align}
\lambda(t)=T+\int_0^t ds\,R(t,s)\Phi_2(C(t,s))C(t,s)+Jm^2(t)+\langle\eta(t)x(t)\rangle,
\end{align}
and using the result above we obtain:
\begin{align}
\langle\eta(t)x(t')\rangle=\int_0^{t'}ds\,R(t',s)\langle\eta(t)\eta(s)
\rangle=2TR(t',t)+\int_0^{t'}ds\,R(t',s)\Phi_1(C(t,s))
\end{align}
and we recall that $R(t^-,t)=0$, but $R(t^+,t)>0$. With these results, we can summarize the DMFT equations as follows:
\begin{align}
&\begin{cases}
\lambda(t)=T+ \int_0^t ds\,R(t,s)\Phi_2(C(t,s))C(t,s)+     \int_0^{t}ds\,R(t,s)\Phi_1(C(t,s))+Jm^2(t)\quad\text{SpM}\\
\lambda(t)=\lambda(C(t,t)) \quad\quad\text{CM}
\end{cases}\\
&\partial_t R(t,t')=-\lambda(t)R(t,t')+\int_{t'}^tds\,R(t,s)\Phi_2(C(t,s))R(s,t')+\delta(t-t')\\
\begin{split}
&\partial_tC(t,t')=-\lambda(t)C(t,t')+\int_0^tds\,R(t,s)\Phi_2(C(t,s))C(s,t')+Jm(t)m(t')+2TR(t',t)\\
&\quad\quad\quad\quad +\int_0^{t'}ds\,R(t',s)\Phi_1(C(t,s))
\end{split}
\end{align}
where we recall that $CM$ stands for confined model and $SpM$ for spherical model.

\subsection{Reduction to the pure spherical $p$-spin}
To get the dynamical equation for the pure spherical $p$-spin model, we need to choose $\Phi_1(u)=\frac{p}{2}u^{p-1}$ and $\Phi_2(u)=\Phi_1'(u)=\frac{p(p-1)}{2}u^{p-2}$, and $J=0$ to get the standard landscape without external field. With these choices, we get:
\begin{align}
&\lambda(t)=T+\frac{p^2}{2}\int_0^t ds\,R(t,s)C^{p-1}(t,s)\\
\begin{split}
&\partial_t C(t,t')=-\lambda(t)C(t,t')+\frac{p(p-1)}{2}\int_0^tds\,R(t,s)C(t,s)^{p-2}C(s,t')\\
&\quad\quad\quad\quad+\frac{p}{2}\int_0^{t'}ds\,R(t',s)C^{p-1}(t,s)+2TR(t',t)\\
\end{split}
\\
&\partial_tR(t,t')=-\lambda(t)R(t,t')+\frac{p(p-1)}{2}\int_{t'}^tds\,R(t,s)R(s,t')C^{p-2}(t,s)+\delta(t-t')
\end{align}

\newpage
\section{Quenched complexity of Gaussian fields}
\label{app:rnn_quenched}
Section based on \cite{us_non_reciprocal_2025}.
\subsection{The replicated Kac-Rice formalism}
For the class of models in Chapter~\ref{chapter:non_reciprocal} we compute the quenched complexity of equilibria by means of the replicated Kac-Rice formalism, following the procedure outlined in \cite{ros2019complex, ros2019complexity}. Recall from Chapter~\ref{chapter:non_reciprocal} that we had
\begin{align}
    \mathcal{N}(\lambda,m,q)=\int_{\mathbb{R}^N} d{\bf x}\,\Delta({\bf x})|\det \mathcal{H}({\bf x})|\delta({\bf f}({\bf x})-\lambda\,{\bf x})
\end{align}
with
\begin{align*}
&\Delta({\bf x})=\delta\left(\sum_ix_i-Nm\right)\delta({\bf x}^2-Nq),\quad \quad 
[\mathcal{H}({\bf x})]_{ij}=\frac{\partial f_i({\bf x})}{\partial x_j}-\lambda\delta_{ij}.
\end{align*}
The $N \times N$ matrix $\mathcal{H}({\bf x})$ is the Jacobian \footnote{we neglect the term $d\lambda/dx_i$, since ultimately the result for the complexity does not change, up to finite rank perturbations of the Jacobian.}. The integer powers $\mathcal{N}^n$ are obtained multiplying the integral representation $n$ times: they involve therefore $n$ integration variables, ${\bf x}^a$ with $a=1, \cdots, n$, which we refer to as replicas. Taking the average with respect to the random field ${\bf f}$, one gets the Kac-Rice formula for the moments of $\mathcal{N}$, which reads:
\begin{align}\label{eq:kRM}
\mathbb{E}\mathcal{N}^n=\int \prod d{\bf x}^a \prod_a\Delta({\bf x}^a)\mathbb{E}\left[\prod_a\delta({\bf f}({\bf x}^a)-\lambda{\bf x})\right]\mathbb{E}\left[\prod_a |\det\mathcal{H}({\bf x}^a)|\bigg|{\bf f}({\bf x}^a)=\lambda{\bf x}^a,\, a=1,\ldots,n \right].
\end{align}
In this expression, the expectation value of the product of determinants of the matrices $\mathcal{H}({\bf x}^a)$ is conditioned to ${\bf f}({\bf x}^a)=\lambda {\bf x}^a$. The annealed complexity can be derived setting $n=1$ in this expression. \\
The expectation values in this expression are functions of the configurations $ {\bf x}^a$; below, we show that for large $N$ the dependence on these configurations enters only through the overlaps $Q^{ab}=N^{-1}  {\bf x}^a \cdot {\bf x}^b$ between the different replicas $a \neq b$. This implies a huge dimensionality reduction in the integral \eqref{eq:kRM}, which can be represented in the form:
\begin{align}\label{eq:ActionQ}
\mathbb{E}\mathcal{N}^n(m, \lambda, q)=\int \prod_{a < b} dQ_{ab} \,e^{N n \tilde{\Sigma} \tonde{\lambda,m,q,Q_{ab}}+ o(Nn)}
\end{align}
for some function $\tilde{\Sigma} \tonde{\lambda,m,q,Q_{ab}}$. The leading order exponential behavior of this quantity can be determined via a saddle-point on the variables $Q_{ab}$. We solve the resulting problem within the RS (Replica Symmetric) ansatz, which corresponds to setting
\begin{equation}
    Q_{ab}= q \delta_{ab}+ \tilde Q (1-\delta_{ab}), \quad \quad \tilde{\Sigma} \tonde{\lambda,m,q,Q^{ab}} \to  \tilde{\Sigma}(\lambda,m,q,\tilde{Q}).
\end{equation}
 As a consequence of the saddle-point calculation, the quenched complexity within the RS framework is given in terms of the variational problem \eqref{eq.compvar}.
 In the following subsections, we determine the form of the function $\tilde{\Sigma}(\lambda,m,q,\tilde{Q})$, by studying separately each of the three terms appearing in the integrand in \eqref{eq:kRM}.

\subsection{Joint probability that all replicas are equilibria}
Under the constraints imposed on the ${\bf x}^a$, the probability that all the configurations ${\bf x}^a$ are equilibria with prescribed values of $m, q$ and $\lambda$ reads:
\begin{align}
\label{app:proba_term_initial}
\begin{split}
\mathbb{P}:=\mathbb{E}\left[\prod_a\delta({\bf f}({\bf x}^a)-\lambda{\bf x}^a)\right]&=\frac{1}{\sqrt{(2\pi)^{Nn}\det\hat{C}}}\int \prod_a d{\bf f}^a\delta(-\lambda{\bf x}^a+{\bf f}^a) e^{-\frac{1}{2}\sum_{a,b}({\bf f}^a-Jm\mathbf{1})[\hat{C}^{-1}]^{ab}({\bf f}^b-Jm\mathbf{1})}\\
    &=\frac{1}{\sqrt{(2\pi)^{Nn}\det\hat{C}}}e^{-\frac{1}{2} \sum_{a,b}\tonde{\lambda {\bf x}^a - J m {\bm 1}} [\hat C^{-1}]^{ab}\tonde{\lambda {\bf x}^b - J m {\bm 1}}},
\end{split}
\end{align}
where we have used the fact that the field ${\bf f}({\bf x})$ is Gaussian with the statistics \eqref{app:eq:def_Covariance}, and where we defined the covariance matrix
\begin{align}\label{eq:CovMat}
    \hat{C}^{ab}_{ij}=\text{Cov}[f_i({\bf x}^a),f_j({\bf x}^b)]= \delta_{ij} \Phi_1\tonde{\frac{{\bf x}^a \cdot {\bf x}^b}{N}}+ \frac{x^a_j \, x^b_i}{N} \Phi_2\tonde{\frac{{\bf x}^a \cdot {\bf x}^b}{N}}.
\end{align}
{The joint probability \eqref{app:proba_term_initial} in principle depends on all the configurations ${\bf x}^a$; we now show that, due to the isotropic structure of the covariances of the random fields, the term \eqref{app:proba_term_initial} can actually be written solely as a functions of the overlaps $Q^{ab}= N^{-1} {\bf x}^a \cdot {\bf x}^b$ between the replicas. To show this, we follow the procedure introduced in \cite{ros2019complex}.
The key observation for computing the expression in the exponent of \eqref{app:proba_term_initial} is that there is no need to invert the full $N n \times N n$ covariance matrix $\hat{C}$; instead, it is sufficient to determine its inverse within an appropriately chosen subspace.} To proceed, we decompose the matrix $\hat{C}$ as the sum of a diagonal and non-diagonal part in replica space:
\begin{equation}
    \hat C= \hat D + \hat O, \quad \quad \hat C^{-1}= \hat D^{-1}[\mathbb{I}+ \hat O \hat D^{-1}]^{-1}
\end{equation}
where $\mathbb{I}$ is the $N n \times N n$ identity matrix, and
\begin{equation}
\begin{split}
     D^{ab}_{ij}&= \delta_{ab}\quadre{\delta_{ij} \Phi_1(q) + \frac{x_j^a x_i^a}{N} \Phi_2(q)}\\
     O^{ab}_{ij}&=(1-\delta_{ab})\quadre{\delta_{ij} \Phi_1(Q^{ab})+ \frac{x_j^a x_i^b}{N} \Phi_2(Q^{ab})}.
\end{split}
\end{equation}
We exploit the notation \eqref{eq:Notation}, and additionally define:
\begin{align}\label{eq:not2}
    \Phi_{i}(Q^{ab})\equiv \Phi_i^{ab},\quad\quad i\in\{1,2\},
\end{align}
with $\Phi_i^{aa}=\Phi_i^q$. 
By means of the Shermann-Morrison formula, we get:
\begin{equation}
\begin{split}
  &  [\hat D^{-1}]^{ab}_{ij}= \delta_{ab} \quadre{(\Phi_1^q)^{-1} \delta_{ij}- \eta \frac{x_j^a x_i^a}{N} }, \quad \quad \eta:=\frac{\Phi_2^q/\Phi_1^q}{\Phi_1^q+ \Phi_2^q q}.
    \end{split}
\end{equation} 
We now introduce a set of vectors spanning a sub-space (of the $(Nn)$- dimensional space on which the matrix $\hat C$ acts) that is closed under the action of $\hat C$, and which is relevant to reconstruct the products appearing in the exponent of Eq.~\eqref{app:proba_term_initial}. In the following, we implement explicitly the RS assumption on the overlap matrix. We define the following three $Nn$-dimensional vectors:
\begin{equation}
\begin{split}
    &{\bm \xi}_1:= ({\bf x}^1, \cdots, {\bf x}^n),\\
    &{\bm \xi}_2:=({\bm 1}, \cdots, {\bm 1}),\\
    &{\bm \xi}_3:= \tonde{\sum_{b \neq 1} {\bf x}^b, \cdots, \sum_{b \neq n} {\bf x}^b},
    \end{split}
\end{equation}
and determine the action of the matrix $\hat O \hat D^{-1}$ on these vectors. From Eq.\eqref{app:proba_term_initial} it appears that the expression at the exponent can be expressed in terms of the action of $\hat C^{-1}$ on these vectors, as  $
-\frac{1}{2} \tonde{\lambda {\bm \xi}_1 - J m\, {\bm \xi}_2 } \hat C^{-1}\tonde{\lambda {\bm \xi}_1 - J m\, {\bm \xi}_2}$, 
from which it is evident that determining the action of $\hat{C}$ on this set of vectors suffices for our purposes.
The action on $\hat{D}^{-1}$  is given by the following expressions:
\begin{equation}
\begin{split}
    & [\hat D^{-1}{\bm \xi}_1]^a= (\Phi_1^q)^{-1}{\bm \xi}_1^a - \eta q {\bm \xi}_1^a,\\
    &[\hat D^{-1}{\bm \xi}_2]^a=(\Phi_1^q)^{-1}{\bm \xi}_2^a- \eta m{\bm \xi}_1^a,\\
    &[\hat D^{-1}{\bm \xi}_3]^a= (\Phi_1^q)^{-1}{\bm \xi}_3^a- \eta {\bm \xi}_1^a \sum_{b \neq a} Q^{ab}, 
    \end{split}
\end{equation}
while the action of $\hat{O}$ reads:
\begin{equation}
\begin{split}
    & [\hat O{\bm \xi}_1]^a_i=\sum_{b\neq a}\left[\Phi_1^{ab}+\Phi_2^{ab}Q^{ab}\right]x_i^b, \\
    &[\hat O{\bm \xi}_2]^a_i=\sum_{b\neq a}\left[\Phi_1^{ab}+m\Phi_2^{ab}x_i^b\right],\\
    &[\hat O{\bm \xi}_3]^a_i=\sum_{b\neq a}\left[\Phi_1^{ab}\sum_{c\neq a}x_i^c+x_i^b\left(\Phi_2^{ab}\sum_{c\neq a}Q^{ac}+q\Phi_2^{ab}-\Phi_1^{ab}-Q^{ab}\Phi_2^{ab}\right)+\Phi_1^{ab}x_i^a\right].
    \end{split}
\end{equation}
Now, within the RS ansatz assumption, the previous expressions simplify, taking the form:
\begin{equation}
\label{app:action_B_xi}
\begin{split}
    &\hat D^{-1}{\bm \xi}_1= [(\Phi_1^q)^{-1}- \eta q]{\bm \xi}_1,\\
    &\hat D^{-1}{\bm \xi}_2=(\Phi_1^q)^{-1}{\bm \xi}_2- \eta m{\bm \xi}_1,\\
    &\hat D^{-1}{\bm \xi}_3= (\Phi_1^q)^{-1}{\bm \xi}_3- \eta(n-1)\tilde{Q}{\bm \xi}_1,
    \end{split}
\end{equation}
\begin{equation}
\begin{split}
     \hat O{\bm \xi}_1&=\left[\tilde{\Phi}_1+\tilde{\Phi}_2\tilde{Q}  \right]{\bm\xi}_3, \\
    \hat O{\bm \xi}_2&=(n-1)\tilde{\Phi}_1{\bm\xi}_2+m\tilde{\Phi}_2{\bm\xi}_3,\\
    \hat O{\bm \xi}_3&=(n-1)\tilde{\Phi}_1{\bm\xi}_3+\left(\tilde{\Phi}_2\tilde{Q}(n-1)+q\tilde{\Phi}_2-\tilde{\Phi}_1-\tilde{Q}\tilde{\Phi}_2\right){\bm\xi}_3+(n-1)\tilde{\Phi}_1{\bm\xi}_1\\
    &=\left[(n-2)\tilde{\Phi}_1+\tilde{\Phi}_2\left(\tilde{Q}(n-2)+q\right)\right]{\bm\xi}_3+(n-1)\tilde{\Phi}_1{\bm\xi}_1.
    \end{split}
\end{equation}
Combining these formulas, we get:
\begin{align}
\label{app:action_xi_OD}
\begin{split}
\hat{O}\hat{D}^{-1}{\bm\xi}_1&=[(\Phi_1^q)^{-1}-\eta q][\tilde{\Phi}_1+\tilde{\Phi}_2\tilde{Q}]{\bm\xi}_3\\
\hat{O}\hat{D}^{-1}{\bm\xi}_2&=(\Phi_1^q)^{-1}\tilde{\Phi}_1(n-1){\bm\xi}_2+\left[(\Phi_1^q)^{-1}m\tilde{\Phi}_2
-\eta m[\tilde{\Phi}_1+\tilde{\Phi}_2\tilde{Q}]\right]{\bm\xi}_3\\
\hat{O}\hat{D}^{-1}{\bm\xi}_3&=(\Phi^q_1)^{-1}(n-1)\tilde{\Phi}_1{\bm\xi}_1 +\\&\Bigg\{-\eta(n-1)\tilde{Q}\left[\tilde{\Phi}_1+\tilde{\Phi}_2\tilde{Q}\right]+(\Phi^q_1)^{-1}\left[(n-2)\tilde{\Phi}_1+\left(\tilde{Q}(n-2)+q\right)\tilde{\Phi}_2\right]\Bigg\}{\bm\xi}_3.
\end{split}
\end{align}
In order to invert the matrix $\mathbb{I}+ \hat O \hat D^{-1}$, it is convenient to express it in a basis of the subspace spanned by ${\bm\xi}_1,{\bm\xi}_2,{\bm\xi}_3$, that is made of orthogonal vectors. We thus consider the following orthogonal basis vectors:
\begin{equation}
\label{app:def_v}
    \begin{split}
        {\bf v}_1&= \alpha_1{\bm \xi}_1,\\
       {\bf v}_2&= \alpha_2\tonde{{\bm  \xi}_2- \frac{m}{q} {\bm \xi}_1},\\
       {\bf v}_3&= {\bm \xi}_3 - \alpha_3{\bm \xi}_1 - \alpha_4 {\bm \xi}_2,
    \end{split}
\end{equation}
with
\begin{align}
    \begin{split}
        &\alpha_1=\frac{1}{\sqrt{N  n q}},\\
        &\alpha_2=\frac{\sqrt{q}}{\sqrt{N n (q-m^2)}}, \\
        &\alpha_3=\frac{(n-1) (\tilde{Q}-m^2 )}{q-m^2 },\\
        &\alpha_4=\frac{m (n-1) (q - \tilde{Q})}{q-m^2}.
    \end{split}
\end{align}
Notice that we did not normalize the last vector  ${\bf v}_3$, as the final expression we aim to compute depends only on the components of the matrix along the vectors  ${\bf v}_1$ and ${\bf v}_2$, making the normalization of ${\bf v}_3$ unnecessary. We denote by $\hat{M}^{{\bm\xi}}$ the $3 \times 3$ matrix expressing the action of  $\hat{M}:=\mathbb{I}+\hat{O}\hat{D}^{-1}$ on the subspace spanned by the vectors ${\bm\xi}$, such that $\hat{M} {\bm \xi}_k=\sum_{l=1}^3[\hat{M}^{{\bm\xi}}]_{lk}{\bm \xi}_l$ for $k=1, 2,3$. The components of this matrix can be easily read from Eq.\eqref{app:action_xi_OD}. We also denote by $P$ the $3 \times 3$ matrix encoding for the change of basis from the vectors ${\bm\xi}$ to ${\bf v}$, meaning that ${\bf v}_k=\sum_{l=1}^3 P_{lk} \, {\bm \xi}_l$ $k=1, 2,3$. We have:
\[
P=\begin{pmatrix}
\frac{1}{\sqrt{n \, N \, q}} & -\frac{m}{ \sqrt{n \, N \,q \left(q-m^2\right)}} & -\frac{\left(n-1\right) \, \left(m^2 - \tilde{Q}\right)}{m^2 - q} \\
0 & \frac{\sqrt{q}}{\sqrt{n \, N \, \left(q-m^2\right)}} & \frac{m \, \left(n-1\right) \, \left(q - \tilde{Q}\right)}{m^2 - q} \\
0 & 0 & 1
\end{pmatrix}
\]
which gives us an inverse:
\[
P^{-1}=\begin{pmatrix}
\sqrt{n  N  q} & m\sqrt{\frac{nN}{q}} & (n-1)\tilde{Q}\sqrt{\frac{nN}{q}} \\
0 & \sqrt{\frac{n  N  (q-m^2)}{q}} & m(n-1)(q-\tilde{Q})\sqrt{\frac{nN}{q(q-m^2)}} \\
0 & 0 & 1
\end{pmatrix},
\]
such that ${\bm \xi}_k=\sum_{l=1}^3 P^{-1}_{lk} \, {\bf v}_l$. Then the action of $\hat M$ on the subspace spanned by the vectors ${\bf v}$ is obtained as:
\begin{align}
    \hat{M}^{{\bf v}}=P^{-1}\hat{M}^{{\bm\xi}}P,
\end{align}
meaning that $\hat M {\bf v}_k= \sum_{l=1}^3 [\hat{M}^{{\bf v}}]_{lk} {\bf v_l}$ for $k=1, 2,3$.
Therefore
\begin{align}
    [\hat{M}^{{\bf v}}]^{-1}=P^{-1}[\hat{M}^{{\bm\xi}}]^{-1}P,
\end{align}
and by plugging the expressions for $\hat{M}^{\bm\xi}$ and $P$, we get:
\begin{align}\label{eq:Ys}
    \begin{split}
    Y_{11}:&=[\hat{M}^{{\bf v}}]^{-1}_{11}\Big|_{n=0}=
    \frac{(\Phi_1^q + q \Phi_2^q) (q \Phi_1^q - 
   2 q \tilde{\Phi}_1 + 
   \tilde{Q} \tilde{\Phi}_1 + (q - \tilde{Q})^2 \tilde{\Phi}_2)}{q\mathcal{A}}\\
   Y_{12}:&=[\hat{M}^{{\bf v}}]^{-1}_{12}\Big|_{n=0}=\frac{m (q - \tilde{Q}) ( \Phi_1^q + q \Phi_2^q) (\tilde{\Phi}_1 + 
   \tilde{Q} \tilde{\Phi}_2)}{q\mathcal{A}\sqrt{q-m^2}}\\
   Y_{21}:&=[\hat{M}^{{\bf v}}]^{-1}_{21}\Big|_{n=0}=\frac{m (q - \tilde{Q})\Phi_1^q (\tilde{\Phi}_1 + 
   \tilde{Q} \tilde{\Phi}_2)}{q\mathcal{A}\sqrt{q-m^2}}\\
   Y_{22}:&=[\hat{M}^{{\bf v}}]^{-1}_{22}\Big|_{n=0}=
   \frac{q\,\Phi_1^q}{(\Phi_1^q - \tilde{\Phi}_1)(q-m^2)}
   -\frac{\Phi_1^q\,m^2 (q (\Phi_1^q + q \Phi_2^q) - 
   \tilde{Q}(\tilde{\Phi}_1 + \tilde{Q} \tilde{\Phi}_2))}{q\mathcal{A}(q-m^2)}
    \end{split}
\end{align}
where
\begin{align}
\begin{split}
    \mathcal{A}=(\Phi_1^q)^2 + (\tilde{\Phi}_1)^2 - 2 q \tilde{\Phi}_1 \Phi_2^q + (q - \tilde{Q})^2 \Phi_2^q \tilde{\Phi}_2 + \tilde{Q} \tilde{\Phi}_1 (\Phi_2^q + \tilde{\Phi}_2) + \Phi_1^q (-2 \tilde{\Phi}_1 - 2 \tilde{Q} \tilde{\Phi}_2 + q (\Phi_2^q + \tilde{\Phi}_2)).
\end{split}
\end{align}
In these expressions, we have already taken the limit $n \to 0$ since we are interested in the linear term (in $n$) of the expansion of the exponent in \eqref{app:proba_term_initial}. As we shall show below, a multiplicative factor of order $n$ will be provided by the  normalization of ${\bf v}_1$, allowing us to set $n=0$ in the matrix elements. Now, recall that $\hat{C}^{-1}=\hat{D}^{-1}\hat{M}$. We have the action of $\hat{M}$ on the ${\bf v}$ basis, and so it remains to find the action of $\hat{D}^{-1}$. This is easily done from Eq.~\eqref{app:action_B_xi} and \eqref{app:def_v}. Using that  $\hat{D}$ is symmetric, we find:
\begin{align}
    \begin{split}
        &{\bf v}_1^\top\hat{D}^{-1}=[(\Phi_1^q)^{-1}- \eta q]{\bf v}_1^\top\\
        &{\bf v}_2^\top\hat{D}^{-1}=(\Phi_1^q)^{-1}{\bf v}_2^\top.
    \end{split}
\end{align}
The terms needed to compute the exponent in \eqref{app:proba_term_initial} are
\begin{align}
\begin{split}
U_{11}:= \lim_{n \to 0}\frac{{\bm\xi}_1^\top\hat{C}^{-1}{\bm\xi}_1}{N n}&=\lim_{n \to 0} \frac{1}{\alpha_1^2}\frac{{\bf v}_1^\top\hat{D}^{-1}\hat{M}^{-1}{\bf v}_1}{N n }= q[(\Phi_1^q)^{-1}-\eta q]Y_{11}\\
\end{split}
\end{align}
and 
\begin{align}
    \begin{split}
      U_{22}:&=\lim_{n \to 0} \frac{{\bm\xi}_2^\top\hat{C}^{-1}{\bm\xi}_2}{N n}=\lim_{n \to 0}  \frac{1}{N n}\left(\frac{{\bf v}_2}{\alpha_2}+\frac{m}{q}\frac{{\bf v}_1}{\alpha_1}\right)^\top\hat{C}^{-1}\left(\frac{{\bf v}_2}{\alpha_2}+\frac{m}{q}\frac{{\bf v}_1}{\alpha_1}\right)\\
        &=\lim_{n \to 0} \frac{1}{Nn} \quadre{ \frac{1}{\alpha_2^2}{\bf v}_2\hat{C}^{-1}{\bf v}_2+\frac{m}{q}\frac{1}{\alpha_1\alpha_2}\tonde{{\bf v}_2\hat{C}^{-1}{\bf v}_1+{\bf v}_1\hat{C}^{-1}{\bf v}_2}+\frac{m^2}{q^2}\frac{1}{\alpha_1^2}{\bf v}_1\hat{C}^{-1}{\bf v}_1}\\
        &= 
       \frac{q-m^2}{q}(\Phi_1^q)^{-1}Y_{22}+\frac{m \sqrt{q-m^2}}{q}\left((\Phi_1^q)^{-1}Y_{21}+[(\Phi_1^q)^{-1}- \eta q]Y_{12}\right)+\frac{m^2}{q}[(\Phi_1^q)^{-1}- \eta q]Y_{11}
    \end{split}
\end{align}
and 
\begin{align}
\begin{split}
U_{12}:&=\lim_{n \to 0}\frac{1}{Nn} \tonde{{\bm\xi}_2^\top\hat{C}^{-1}{\bm\xi}_1+{\bm\xi}_1^\top\hat{C}^{-1}{\bm\xi}_2}=\lim_{n \to 0}\frac{1}{Nn} \quadre{\frac{1}{\alpha_1}\left(\frac{{\bf v}_2}{\alpha_2}+\frac{m}{q}\frac{{\bf v}_1}{\alpha_1} \right)\hat{C}^{-1}{\bf v}_1+\frac{1}{\alpha_1}{\bf v}_1\hat{C}^{-1}\left(\frac{{\bf v}_2}{\alpha_2}+\frac{m}{q}\frac{{\bf v}_1}{\alpha_1} \right)}\\
&=\lim_{n \to 0}\frac{1}{Nn}\frac{1}{\alpha_1\alpha_2}\left[{\bf v}_2\hat{C}^{-1}{\bf v}_1+{\bf v}_1\hat{C}^{-1}{\bf v}_2\right]+\lim_{n \to 0}\frac{1}{Nn}\frac{2}{\alpha_1^2}\frac{m}{q}{\bf v}_1\hat{C}^{-1}{\bf v}_1\\
&=\sqrt{q-m^2}\left[(\Phi_1^q)^{-1}Y_{21}+ [(\Phi_1^q)^{-1}- \eta q]Y_{12}\right]+2m [(\Phi_1^q)^{-1}- \eta q]Y_{11}.
\end{split}
\end{align}
Substituting the expressions \eqref{eq:Ys} into these formulas, we finally get:
\begin{align}\label{eq:Ufin}
\begin{split}
U_{11}&=\frac{q \Phi_1^q - 2 q \tilde{\Phi}_1 + 
 \tilde{Q} \tilde{\Phi}_1 + (q - \tilde{Q})^2 \tilde{\Phi}_2}{\mathcal{A}}\\
 U_{22}&=\frac{1}{\Phi_1^q - \tilde{\Phi}_1}-\frac{m^2(\Phi_2^q-\tilde{\Phi}_2)}{\mathcal{A}}\\
U_{12}&=2m\frac{\tilde{\Phi}_2(q-\tilde{Q})+\Phi_1^q-\tilde{\Phi}_1}{\mathcal{A}}.
\end{split}
\end{align}

To complete the calculation of the joint probability \eqref{app:proba_term_initial}, it remains to determine the determinant of the matrix
\[
\hat{C}^{ab}_{ij}=\delta_{ij}\Phi_1(Q^{ab})+\frac{x_j^ax_i^b}{N}\Phi_2(Q^{ab}).
\]
It is useful to decompose $\mathbb{R}^N$ into a subspace $S=\text{Span}\grafe{{\bf x}^1/\sqrt{N},\ldots,{\bf x}^n/\sqrt{N},{\bf 1}/\sqrt{N}}$ with ${\bf 1}=(1, \cdots, 1)^T$, and its orthogonal complement $S^\perp$. By doing so, the matrix $\hat{C}$ in this basis, and within the RS ansatz, reads:
\begin{align}
    \hat{C}^{ab} =(\Phi_1^q\delta_{ab}+(1-\delta_{ab})\tilde{\Phi}_1)\mathbb{I}+\frac{[{\bf x}^b]^\top}{\sqrt{N}}\frac{{\bf x}^a}{\sqrt{N}}\Phi_2^{ab}
\end{align}
Therefore, within the RS ansatz, by the matrix determinant lemma, to leading exponential order in $N$ the determinant reads:
\begin{align*}
    \det\hat{C}&=\quadre{\det\begin{pmatrix}
        \Phi_1^q & \ldots & \tilde{\Phi}_1\\
        \vdots & \ddots & \vdots\\
        \tilde{\Phi}_1 & \ldots &\Phi_1^q
    \end{pmatrix}}^N o(e^N)=\left[(\Phi_1^q-\tilde{\Phi}_1)^{n-1}(\Phi_1^q+(n-1)\tilde{\Phi}_1) \right]^N o(e^N) \\&=e^{\grafe{nN\left[\text{log}(\Phi_1^q-\tilde{\Phi}_1)+\frac{\tilde{\Phi}_1}{\Phi_1^q-\tilde{\Phi}_1}\right]+ o(N n)}}.\\
\end{align*}

\paragraph*{Asymptotic behavior of the joint probability: quenched case}
Combining all the pieces derived above, the contribution of the joint probability to the quenched complexity reads:
\begin{align}
\label{app:P_quench}
\begin{split}
\mathcal{P}:=\lim_{n\to0,N\to\infty}\frac{\log\mathbb{P}}{nN}=-\frac{1}{2}\left\{\log(2\pi)+\log(\Phi_1^q-\tilde{\Phi}_1)+\frac{\tilde{\Phi}_1}{\Phi_1^q-\tilde{\Phi}_1}+\lambda^2U_{11}-\lambda Jm U_{12}+J^2m^2U_{22}\right\},
    \end{split}
\end{align}
where $U_{11}, U_{22}$ and $U_{12}$ are given in \eqref{eq:Ufin}. As we remarked, this contribution depends on the configurations ${\bf x}^a$ only through the overlaps $q$ and $\tilde Q$. \\

\paragraph*{Asymptotic behavior of the joint probability: annealed case}
To get the  annealed case we set $n=1$ and $\tilde{Q}=0$. This gives 
\begin{align}
\det\hat{C}=e^{N\log\Phi_1^q+ o(N)}
\end{align}
and 
\begin{align}
\begin{split}
U_{11}^{A}=\frac{q}{\Phi_1^q + q \Phi_2^q}, \quad \quad
U_{12}^A =\frac{2m}{\Phi_1^q + q \Phi_2^q}, \quad \quad
U_{22}^A=\frac{\Phi_1^q + (q-m^2) \Phi_2^q}{\Phi_1^q (\Phi_1^q + 
   q \Phi_2^q)}
\end{split}
\end{align}
which finally implies
\begin{align}
\label{app:P_ann}\mathcal{P}_A:=\lim_{N\to\infty}\frac{\log\mathbb{P}|_{n=1,\tilde{Q}=0}}{N}=-\frac{1}{2}\left[\log(2\pi)+\log(\Phi_1^q)+\lambda^2U_{11}^A-\lambda J m U_{12}^A+J^2m^2U_{22}^A \right].
\end{align}

\subsection{The conditional expectation of the Jacobians}
\label{app:determinant_calc}
Consider now the term:
\begin{align}\label{eq.det}
   \mathbb{D}:= \mathbb{E}\left[\prod_{a=1}^n|\det\mathcal{H}({\bf x}^a)|\Bigg| \substack{{\bf f}^b=\lambda{\bf x}^b \\ b=1, \cdot, n} \right]=\mathbb{E}\left[ e^{\sum_{a=1}^n\text{Tr} \log |\mathcal{H}({\bf x}^a)|}\Bigg| \substack{{\bf f}^b=\lambda{\bf x}^b \\ b=1, \cdot, n} \right].
\end{align}
We define $\mathcal{H}^a:=\mathcal{H}({\bf x}^a)$  and  $G^a_{ij}:=\partial_jf_i({\bf x}^a)$
where $a$ is the replica index, and $\partial_j$ denotes the derivative with respect to $x_j^a$. Computing the expectation $  \mathbb{D}$ requires a priori to determine the joint distribution of the matrices $G^a$ for $a=1, \cdots, n$ conditioned to the forces ${\bf f}^b:={\bf f}({\bf x}^b)$ taking values $\lambda {\bf x}^b$. In fact, the second equality in \eqref{eq.det} shows that, since the trace can be expressed as a sum over eigenvalues, the quantity we are interested in depends on the random matrices $G^a$ only through their eigenvalue distribution: denoting with $\tilde \rho^a_N(z)$ the eigenvalue distribution of the \emph{conditioned} matrix $G^a$, it holds:
\begin{align}\label{eq.det2}
   \mathbb{D}= \mathbb{E}\left[ e^{N \sum_{a=1}^n \int d\tilde \rho_N^a(z) \log |z-\lambda|} \right].
\end{align}
We now argue that the problem simplifies drastically if one is interested only in the leading order behavior (in $N$) of this quantity.\\

\paragraph*{Statistics of the Jacobian matrices prior to conditioning. } We begin by discussing the statistics of the entries of the random Gaussian matrices $G^a$. It holds
\begin{align}\label{eq:mean}
    \mu_{G}^{aij}:=\mathbb{E}\left[ G_{ij}^a\right] = \frac{J}{N}
\end{align}
and (making use of he compact notation \eqref{eq:not2}):
\begin{align}
\label{app:eq:sigma_G}
\begin{split}
    \Sigma_G^{aij,bkl}=\text{Cov}\left[ G_{ij}^a,G_{kl}^b\right]&=\left(\delta_{ik}\delta_{jl}\frac{[\Phi_1']^{ab}}{N}+\delta_{il}\delta_{jk}\frac{\Phi_2^{ab}}{N} \right)+\delta_{ik}[\Phi_1'']^{ab}\frac{x_l^ax_j^b}{N^2}+\delta_{jl}[\Phi'_2]^{ab}\frac{x_i^b x_k^a}{N^2}\\
    &+\delta_{jk}[\Phi_2']^{ab}\frac{x_l^a x_i^b}{N^2}+\delta_{il}[\Phi'_2]^{ab}\frac{x_k^a x_j^b}{N^2}+[\Phi_2'']^{ab}\frac{x_l^a x_k^a x_i^b x_j^b}{N^3}.
\end{split}
\end{align}
To analyze the expression for the covariances, it is convenient to perform a change of the basis in which the matrices are expressed. We consider the same decomposition of $\mathbb{R}^N$ into the $(n+1)$-dimensional subspace $S$ and its $(N-n-1)$-dimensional orthogonal complement $S^\perp$ that we have made use of in the previous subsection. Consider a set of orthonormal basis vectors ${\bf e}_\alpha$ with $\alpha=1, \cdots, N-n-1$ spanning the subspace $S^\perp$: these vectors are orthogonal to ${\bf x}^a$ and to ${\bf 1}$. Let $G^a_{\alpha \beta}={\bf e}_\alpha \cdot  G^a \cdot  {\bf e}_\beta$ be the components of the matrix $G$ in this basis. From \eqref{eq:mean} and the fact that ${\bf e}_\alpha \perp {\bf 1}$ (for $\alpha=1, \cdots, N-n-1$) it follows that these components have zero average, while \eqref{app:eq:sigma_G} shows that they have covariances 
\begin{align}
\label{app:eq:sigma_G_rotate}
\begin{split}
    \Sigma_G^{a\alpha \beta,b \gamma \delta}=\text{Cov}\left[ G_{\alpha \beta}^a,G_{\gamma \delta}^b\right]&=\left(\delta_{\alpha \gamma}\delta_{\beta \delta}\frac{[\Phi_1']^{ab}}{N}+\delta_{\alpha \delta}\delta_{\beta \gamma}\frac{\Phi_2^{ab}}{N} \right) \quad \quad \alpha, \beta, \gamma, \delta \leq N-n-1.
\end{split}
\end{align}
These covariances between the components in $S^\perp$ thus depend on the configurations ${\bf x}^a$ only through the overlap $Q^{ab}$ between them. Moreover, in this subspace the statistics is isotropic: it is invariant with respect to changes of the basis spanning the subspace. This invariance will be crucial for the subsequent calculation. \\
The covariances of the components with respect to the basis vectors in $S$, as well as the covariances between mixed components, have a more complicated expression which depends explicitly on the choice of the basis vectors in $S$, and which therefore is not basis invariant. These covariances can be computed explicitly for specific choices of the basis vectors, see \cite{ros2019complex, ros2019complexity} for similar examples. Since we are however interested only to the leading order contribution (in $N$) of the expectation value, we can neglect computing such covariances explicitly. Indeed, in the subspace  decomposition of $\mathbb{R}^N$ that we have chosen the matrices $G^a$ have a block structure, 
\begin{equation}\label{eq:Block}
    G^a= 
\begin{pNiceArray}{ccc|c}
  \Block{3-3}<\Large>{E^a}  & && \\
  & & & b^a\\
  & & &\\
    \hline
&[b^a]^T & & c^a
\end{pNiceArray},
\end{equation}
where $E^a$ is a block of dimension $(N-n-1) \times (N-n-1)$, corresponding to the subspace $S^\perp$, where the entries have statistics \eqref{app:eq:sigma_G_rotate}, while $c^a$ and $b^a$ are $(n+1) \times (n+1)$ and $(N-n-1) \times (n+1)$ dimensional blocks with entries with covariances that we have not determined explicitly. To leading order in $N$, the determinant of the matrix $G^a$ equals to the determinant of the block $E^a$, which has dimension scaling with $N$: the continuous part of the eigenvalue distribution of $G^a$ is in fact determined solely by this block in the limit $N \to \infty$. The remaining components have a different statistics, that can be expressed in terms of finite-rank perturbations (both additive and multiplicative) to a $N \times N$ matrix with entries with the same statistics as $E^a$, namely, invariant with respect to changes of basis. These finite-rank perturbations can contribute to the eigenvalues distribution with sub-leading terms in $1/N$, corresponding to isolated eigenvalues (aka, outliers). Since these isolated eigenvalues, being sub-leading, do not contribute to the complexity, we do not perform their calculation in this work. Notice that such outliers must be tracked when discussing the stability of the equilibria counted by the complexity. Indeed, there may be cases where the isolated eigenvalues are the only ones with negative real parts, leading to a dynamical instability of the equilibrium that would be overlooked if these eigenvalues were ignored. However, as we discuss extensively below and in the main text, in the models we consider, most equilibria are already linearly unstable. Therefore, for now, we neglect the calculation of the Jacobian's outliers. \\

\paragraph*{Statistics of the Jacobian matrices after to conditioning. }
We now discuss how the statistics of the matrices $G^a$ is affected by conditioning to ${\bf f}({\bf x}^b)= \lambda {\bf x}^b$ for all $b$. We use the shorthand notation $f^a_i := f_i({\bf x}^a)$. 
It holds:
\begin{align}\label{eq:mixed1}
   \mu_f^{ai}:= \mathbb{E}[f^a_i]=J \,m,\quad\quad \Sigma_{ff}^{ai,bj}:=\text{Cov}[f^a_i,f^b_j]=\delta_{ij}\Phi_1^{ab}+\frac{x_j^ax_i^b}{N}\Phi_2^{ab}.
\end{align}
and
\begin{align}\label{eq:mixed2}
    \text{Cov}\left[f_i^a,G^b_{kl} \right]=\delta_{ik}[\Phi'_1]^{ab}\frac{x^a_l}{N}+[\Phi_2']^{ab}\frac{x_l^ax_j^ax_i^b}{N^2}+\delta_{il}\Phi_2^{ab}\frac{x_j^a}{N}
\end{align}
The conditional law of $G$ can be determined with the standard procedure for Gaussian conditioning. However, Eqs.~\eqref{eq:mixed1} and \eqref{eq:mixed2} show that the only components $G_{\alpha \beta}^a= {\bf e}_\alpha \cdot G \cdot {\bf e}_\beta$ whose statistics is affected by the conditioning are those such that either ${\bf e}_\alpha$ or ${\bf e}_\beta$ belong to the subspace $S$. These are precisely the components whose statistics we are neglecting, since it gives a sub-leading contribution to the eigenvalue density. Again, these sub-leading contributions can contribute to eventual outliers: if one is interested in such outliers, the effect of the conditioning has to be worked out explicitly.\\

\paragraph*{Large-$N$ factorization and concentration. } The final simplifying ingredient to proceed with the calculation consists in the observation that the correlations between the Jacobian matrices evaluated at different ${\bf x}^a$, which are non-zero according to \eqref{app:eq:sigma_G}, are not relevant when computing the expectation value \eqref{eq.det2} to leading exponential order in $N$. In fact, to leading order in $N$ the expectation value factorizes,
\begin{equation}
         \mathbb{D}= \mathbb{E}\left[ e^{N \sum_{a=1}^n \int d\tilde \rho_N^a(z) \log |z-\lambda|} \right]= \prod_{a=1}^n\mathbb{E}\left[ e^{N  \int d\tilde \rho_N^a(z) \log |z-\lambda|} \right]o(e^N).
\end{equation}
An argument for such factorization can be found in \cite{ros2019complex}; it relies on the fact that the eigenvalues distribution of the Gaussian matrices $G^a$ has a large-deviation law with speed higher than $N$. For a rigorous proof of this factorization in the case of symmetric matrices and $n=2$, see \cite{subag2017complexity}.   \\

\paragraph*{The elliptic ensemble determinant. } Combining all the arguments reported above, we conclude that: (a) the joint expectation value of the conditioned Jacobian is, to leading order in $N$, determined only in terms of the single-matrix eigenvalue distributions $\tilde\rho_N^a(x)$: the correlations between the different Jacobians do not enter in the calculation; (b) the matrices $G^a$, both prior and after conditioning to the forces, have the block structure \eqref{eq:Block}; (c) to leading order in $N$, the eigenvalue distribution of matrices of this form coincides with the eigenvalue density of the block $E^a$, which has extensive size in $N$; (d) the block $E^a$ has Gaussian components with the statistics \eqref{app:eq:sigma_G_rotate}: this block is therefore a matrix belonging to the real elliptic ensemble \cite{girko1986elliptic, sommers1988spectrum, nguyen2015elliptic}. Hence, to compute the contribution of $\mathbb{D}$ to the complexity, it suffices to exploit results on the asymptotic eigenvalue density of matrices belonging to the real elliptic ensemble. This eigenvalue density is well known: it is uniform, with a support with elliptic shape in the complex plane. Consider the rescaled matrices $\tilde G^a= G^a/ \sqrt{\Phi_1'(q)}$ having covariances $\mathbb{E}\left[\tilde G^a_{ij}\tilde G^a_{kl}\right]=N^{-1}\left(\delta_{ik}\delta_{jl}+\alpha_q\delta_{il}\delta_{jk}\right)$, where we recall that $ \alpha_q= {\Phi_2^q}/{\dot \Phi_1^q}$. The eigenvalue density of the matrices $\tilde G^a$ is equal for each $a$. It is supported on a domain with elliptic shape in the complex plane $z=x+i y$, 
\begin{align}
\frac{x^2}{(1+\alpha_q)^2}+\frac{y^2}{(1-\alpha_q)^2}\leq 1.
\end{align}
The density is uniform, equal to
\begin{equation}
    \rho_{ \tilde G}(x,y)=\frac{1}{\pi(1-\alpha_q^2)}
\end{equation}
. We therefore find that
\begin{equation}
         \mathbb{D}= [\Phi_1'(q)]^{\frac{N n}{2}}  \tonde{ e^{N  \int dx dy\, \rho_{\tilde G}(x,y) \log |x+i y-\kappa|} }^n o(e^N), \quad \quad \kappa=\frac{\lambda}{\sqrt{\Phi_1'(q)}}.
\end{equation}
Consider now
\begin{equation}
\begin{split}
    I(\kappa)&:= \int dx dy\, \rho_{\tilde G}(x,y) \log |x+i y-\kappa|=\frac{1}{2\pi}\frac{1}{1-\alpha_q^2}\int dx dy \,\log\left[(x-\kappa)^2 + y^2\right]\\
    &= \frac{1}{\pi}\int_{-1}^1 dx\int_0^{\sqrt{1-x^2}} dy\log\left[\left(x(1+\alpha_q)-\kappa
    \right)^2+y^2(1-\alpha_q)^2\right].
    \end{split}
\end{equation}
This integral can be computed explicitly, see Appendix A.3 in \cite{RosEcoQuenched2023}. One finds that for arbitrary real $\kappa$ and $\alpha_q> -1$:
\begin{align}
   I(\kappa)=\begin{cases}
    \frac{1}{2}\left(\frac{\kappa^2}{1+\alpha_q}-1\right)\quad\quad\text{if } |\kappa|\leq 1+\alpha_q\\
    \frac{1}{8\alpha_q}\left(\kappa-\text{sign}(\kappa)\sqrt{\kappa^2-4\alpha_q}\right)^2+\log\left|\frac{\kappa+\text{sign}(\kappa)\sqrt{\kappa^2-4\alpha_q}}{2}\right|\quad\quad\text{if }|\kappa|>1+\alpha_q
    \end{cases},
\end{align}
and therefore we get back Eq.~\eqref{app:det_eqn}:
\begin{align*}
    \Theta:= \lim_{N \to \infty} \lim_{n \to 0} \frac{\log \mathbb{D}}{Nn}= \begin{cases}
    \frac{\log\Phi'_1(q)}{2}+\frac{1}{2}\left(\frac{\kappa^2}{1+\alpha_q}-1\right)\quad\quad\text{if } |\kappa|\leq 1+\alpha_q\\
    \frac{\log\Phi'_1(q)}{2}+\frac{1}{8\alpha_q}\left(\kappa-\text{sign}(\kappa)\sqrt{\kappa^2-4\alpha_q}\right)^2+\log\left|\frac{\kappa+\text{sign}(\kappa)\sqrt{\kappa^2-4\alpha_q}}{2}\right|\quad\quad\text{else}
    \end{cases}
\end{align*}
Notice that one gets the same contribution for both the annealed and quenched case. Whenever $\alpha_q=0$, this reduces to: 
\begin{align}
\Theta|_{\alpha_q=0}=\begin{cases}
    \frac{\log\Phi'_1(q)}{2}+\frac{1}{2}\left(\kappa^2-1\right)\quad\quad\text{if } |\kappa|\leq 1\\
    \frac{\log\Phi'_1(q)}{2}+\log |\kappa|\quad\quad\text{if }|\kappa|>1.\\
    \end{cases}
\end{align}

\paragraph*{A note on the spherical model. } As we have remarked in Chapter~\ref{chapter:non_reciprocal}, the topological complexity of equilibria in the spherical model can be obtained from $\Sigma(\lambda, m,q)$ by setting $q \to 1$ and optimizing over $\lambda$; while in the confined model we have to set $\lambda\to \lambda(q)$. To be more precise, when considering the spherical model we should encode for the spherical constraint in the Jacobian (as in \cite{Fyodorov_2016}), or alternatively project it onto the hypersphere; instead when considering the confined model we should include the derivative of $\lambda$ with respect to ${\bf x}$ in the Jacobian. However, in both cases, we are safe in the $N\to\infty$ limit when computing $\Sigma$.

\subsection{The constrained integral over the replicas}
As illustrated in the sections above, the isotropy of the correlations of the force field implies that the asymptotic behavior of both $\mathbb{P}$ and $\mathbb{D}$ depends on the vectors ${\bf x}^a$ only through the overlap parameters $Q_{ab}$. Therefore, the expression \eqref{eq:kRM} for the moments simplifies to \eqref{eq:ActionQ}. We define the phase space factor:
\begin{align*}
V(Q_{ab})&:=\int \prod_a d{\bf x}^a\Delta({\bf x}^a) \prod_{a\leq b}\delta({\bf x}^a\cdot {\bf x}^b-Q_{ab})\\
&=\int \prod_a d{\bf x}^a \delta\left(\sum_i x^a_i-N m\right)\delta({\bf x}^a\cdot{\bf x}^a-N q) \prod_{a\leq b}\delta({\bf x}^a\cdot {\bf x}^b-Q_{ab}).
\end{align*}
Using the integral representation of the delta distributions we obtain
\begin{align*}
V(Q^{ab})
&=\int \prod_a d{\bf x}^a \int \prod_{a\leq b}\frac{d\lambda_{ab}}{\sqrt{2\pi}}\prod_a \frac{dw_a}{\sqrt{2\pi}}e^{-i\sum_{a\leq b}\lambda_{ab}({\bf x}^a\cdot {\bf x}^b-Q^{ab}N)-i\sum_aw_a(\sum_i x_i-Nm)}\\
&=\int \prod_a d{\bf x}^a \int \prod_{a\leq b}\frac{d\lambda_{ab}}{\sqrt{2\pi}}\prod_a \frac{dw_a}{\sqrt{2\pi}}e^{\frac{N}{2}\text{Tr}(\hat{\Lambda}\hat{Q})+Nm\sum_a\hat{w}_a}\left[\int \prod_a d x^a e^{-\sum_{a\leq b}x^a\lambda_{ab}x^b-\sum_aw_ax^a}  \right]^N,
\end{align*}
where $\Lambda_{ab}=2\lambda_{aa}\delta_{ab}+(1-\delta_{ab})\lambda_{ab}$. We consider the RS ansatz $w_a=\hat{w}$ and 
$\Lambda_{ab}=2\hat{\lambda}_1\delta_{ab}+\hat{\lambda}_0(1-\delta_{ab}).$
Then integral in square brackets then reads 
\begin{align}
    \int d{\bf x}\, e^{-\frac{1}{2}{\bf x}\hat{\Lambda}{\bf x}-\hat{w}{\bf 1}^\top {\bf x}}=(2\pi)^{\frac{n}{2}}(\det \hat{\Lambda})^{-\frac{1}{2}}e^{\frac{1}{2}\hat{w}^2{\bf 1}^\top\hat{\Lambda}^{-1}{\bf 1}},
\end{align}
where now in bold we denote $n$-dimensional vectors.
The inverse of $\hat{\Lambda}$:
\begin{align}
\hat{\Lambda}=
    \begin{pmatrix}
        2\hat{\lambda}_1 & \ldots & \hat{\lambda}_0\\
        \vdots & \ddots & \vdots\\
        \hat{\lambda}_0 & \ldots & 2\hat{\lambda}_1
    \end{pmatrix},
    \quad\quad\quad\quad
    \hat{\Lambda}^{-1}=
    \begin{pmatrix}
        \tilde{h} & \ldots & \tilde{z}\\
        \vdots & \ddots & \vdots\\
        \tilde{z} & \ldots & \tilde{h}
    \end{pmatrix}
\end{align}
has entries
\begin{align*}
\tilde{h}=\frac{2\hat{\lambda}_1+(n-2)\hat{\lambda}_0}{(2\hat{\lambda}_1-\hat{\lambda}_0)(2\hat{\lambda}_1+(n-1)\hat{\lambda}_0)}, \quad \quad
&\tilde{z}=-\frac{\hat{\lambda}_0}{(2\hat{\lambda}_1-\hat{\lambda}_0)(2\hat{\lambda}_1+(n-1)\hat{\lambda}_0)}.
\end{align*}
It follows that 
\begin{align*}
\text{Tr}(\hat{\Lambda}\hat{Q})=n\left[2\hat{\lambda}_1q+(n-1)\tilde{Q}\hat{\lambda}_0\right],\quad
&{\bf 1}\hat{\Lambda}^{-1}{\bf 1}=n\tilde{h}+(n^2-n)\tilde{z}, \quad \log\det\hat\Lambda\approx n\left[\frac{\hat{\lambda}_0}{2\hat{\lambda}_1-\hat{\lambda}_0}+\log(2\hat{\lambda}_1-\hat{\lambda}_0)\right].
\end{align*}
Summing up, we obtain for arbitrary $n$:
\begin{align*}
\frac{\log V}{N}&=\text{extr}_{\hat{\lambda}_0,\hat{\lambda}_1, \hat{w}}\Bigg\{\frac{1}{2}n\left[2\hat{\lambda}_1q+(n-1)\tilde{Q}\hat{\lambda}_0\right]
+n\hat{w}m+\frac{n}{2}\log(2\pi)-\frac{1}{2}n\left[\frac{\hat{\lambda}_0}{2\hat{\lambda}_1-\hat{\lambda}_0}+\log(2\hat{\lambda}_1-\hat{\lambda}_0)\right]\\
&+\frac{1}{2}\hat{w}^2(n\tilde{h}+(n^2-n)\tilde{z})
\Bigg\}.
\end{align*}
The contribution to the annealed complexity is obtained setting $n \to 1$, with $\hat{\lambda}_0,\tilde{Q}\to 0$. It reads:
\begin{align}
\label{app:vol_ann}
    \mathcal{V}_A(m,q):=\frac{1}{2}+\frac{1}{2}\log(2\pi (q-m^2)).
\end{align}
The contribution to the quenched complexity is instead obtained taking $n \to 0$,
\begin{align}
\begin{split}
&\mathcal{V}:=\lim_{N\to\infty,n\to0}\frac{\log V}{nN}\\
&=\frac{1}{2}\text{extr}_{\hat w,\hat\lambda_0,\hat\lambda_1}\Bigg\{
    2 m \hat w - \tilde{Q} \hat\lambda_0 - 
   \frac{\hat w^2}{\hat\lambda_0 - 2 \hat\lambda_1} + \frac{\hat\lambda_0}{\hat\lambda_0 - 
    2\hat\lambda_1} + 2 q \hat\lambda_1 + \log(2 \pi) - 
   \log(2\hat\lambda_1-\hat\lambda_0)
    \Bigg\}.
\end{split}
\end{align}
Optimizing over the parameters results in:
\begin{align}
\label{app:vol_quench}
    \mathcal{V}(m, q , \tilde Q)=
\frac{ q-m^2 + (q - \tilde{Q}) \log(2 \pi) + (q - \tilde{Q}) \log\left(q - \tilde{Q}\right)}{2 (q - \tilde{Q})}.
\end{align}
Notice that quenched and annealed contributions coincide only for $m=0,\tilde{Q}=0$. \\\\
\noindent Finally, we see that we recover all the expressions in Sec.~\ref{sec:rnn_topo_general}.

\newpage
\section{Annealed complexity of a random neural network}
\label{app:scs_ann_compl}
In this section of the Appendix we carry out the computation of the annealed complexity of Chapter~\ref{chapter:scs}. We will be more sloppy with the notation with respect to the main Chapter, in particular we drop all function arguments, to make things cleaner. Let us recall that we have to compute the following quantity:
\begin{align*}
   \mathbb{E}[\mathcal{N}]= \int_{\mathbb{R}^N} d{\bf x} \,
    \Omega({\bf x}) \mathbb{P}({\bf F}({\bf x})=0) \, \mathbb{E} \left[ |\mathrm{det} \,  \partial {\bf F}({\bf x})| \, \Big| {\bf F}({\bf x})=0\right].
\end{align*}
This calculation can be divided in three parts: the phase space term, which is essentially $\int d{\bf x}\,\Omega({\bf x})$, the probability term $\mathbb{P}({\bf F}({\bf x})={\bf 0})$ and the determinant term $\mathbb{E} \left[ |\mathrm{det} \,  \partial {\bf F}({\bf x})| \, \Big| {\bf F}({\bf x})=0\right]$. Let us remark that this separation is only possible because we have identified all the relevant order parameters that make the second two terms ${\bf x}$ independent, and we have fixed such order parameters thanks to $\Omega({\bf x})$, which then is used to compute the "abundance" of vectors ${\bf x}$ that satisfy such constraints (hence the name "phase space term"). Let us start from the probability term.

\subsection{Probability term}
\noindent The components of the vector ${\bf F}$ are Gaussian random variables with mean and variance given by:
\begin{equation}
\begin{split}
&\mu_i\equiv \mathbb{E}[{F_i({\bf x})}] =- x_i + g J_0 M_\phi\\
&\hat{C}_{ij}\equiv \text{Var}[F_i({\bf x}), F_j({\bf x})]=g^2 \left[ \delta_{ij}  Q_\phi + \frac{1}{N}\alpha \, \phi(x_i) \phi(x_j) \right].
\end{split}
\end{equation}

\noindent By the Sherman-Morrison formula, the inverse of the correlation matrix $C$ is 
\begin{equation}
    [\hat{C}^{-1}]_{ij}= \frac{1}{g^2 Q_\phi} \left[\delta_{ij}- \frac{1}{N} \frac{\alpha}{ Q_\phi(1+ \alpha)} \phi(x_i)\, \phi(x_j) \right]
\end{equation}
The second term is a rank-1 projection, hence we can neglect it to leading order in $N$: 
\begin{equation}
  \lim_{N\to\infty}\frac{\log \mathrm{det} \,\hat{C}}{N}= \log (g^2 Q_\phi).
\end{equation}
Thus
\begin{equation}
   \mathbb{P}({\bf F}({\bf x})=0)=\frac{1}{(2\pi)^{N/2}[\det \hat{C}]^{1/2}}e^{-\frac{1}{2}{\bm \mu}^\top \hat{C}\,{\bm\mu}}
\end{equation}
where, in terms of the order parameters:
\begin{equation}
\begin{split}
    \sum_{ij} \mu_i [\hat{C}^{-1}]_{ij} \mu_j= \frac{N}{g^2 Q_\phi} \left( q-2 g J_0 M_\phi \, m+ g^2 J_0^2 M_\phi^2 \right)\\
    -\frac{N \alpha}{g^2 Q_\phi^2 (1+\alpha)}\left( Z^2 - 2 g J_0 \, M^2_\phi Z + g^2 J_0^2 M_\phi^4 \right)
    \end{split}
\end{equation}
Hence, at the exponential scale, the probability term is controlled by the following function, see Eq.~\eqref{eq:scs_p_alpha}, of the order parameters:
\begin{align*}
    p_\alpha(M_\phi, Q_\phi, Z):=\lim_{N\to\infty} \frac{1}{N}\log\mathbb{P}=&-\frac{1}{2} \log (2 \pi g^2 Q_\phi)-\frac{1}{2}\frac{1}{g^2 Q_\phi} \left( q-2 g J_0 M_\phi \, m+ g^2 J_0^2 M_\phi^2 \right)\\
    &+\frac{1}{2}\frac{\alpha}{g^2 Q_\phi^2 (1+\alpha)}\left( Z^2 - 2 g J_0 \, M^2_\phi Z + g^2 J_0^2 M_\phi^4 \right)
\end{align*}

\subsection{Phase space term}
Let us denote $V:=\int_{\mathbb{R}^N} d{\bf x}\,\Omega({\bf x})$, where we recall that $\Omega({\bf x})$ encodes all the different constraints, see Eq.~\eqref{eq:scs_constraints}. We will use the Fourier representation of the deltas as is usually done for these computations:
\begin{align}
    \delta(g(x))=\frac{1}{2\pi}\int d\hat{x}\, e^{i\hat{x}g(x)}.
\end{align}
Hence we get (by neglecting factors that do not contribute exponentially in $N$):
\begin{align*}
    V&\propto \int d\hat{\lambda} \,d\hat{\omega} \,d\hat{\eta}\, d\hat{\xi}\, d\hat{t}\,d\hat{\theta}\, \int d{\bf x}\,e^{i\hat{\lambda}\left(Q_\phi-\frac{\sum_k \phi^2(x_k)}{N}\right) + i\hat{\omega}\left(q-\frac{\sum_k x_k^2}{N}\right)+i\hat{\eta}\left(m-\frac{\sum_k x_k}{N}\right)+i\hat{\xi}\left(M_\phi-\frac{\sum_k \phi(x_k)}{N}\right)}\\
    &\quad\times e^{i\hat{t}\left(D_\phi-\frac{\sum_k \phi'(x_k)}{N}\right)+i\hat{\theta}\left(Z-\frac{\sum_k x_k\phi(x_k)}{N}\right)}
    \\
\end{align*}
let us denote for simplicity $d\hat{\Omega}:=d\hat{\lambda} \,d\hat{\omega} \,d\hat{\eta}\, d\hat{\xi}\, d\hat{t}\,d\hat{\theta}$, which are usually called \textit{conjugate variables}. Then the next step consists in rescaling the conjugate variables (so that, for example, $i\hat{\lambda}\to \hat{\lambda}N$) and neglecting again all those terms that do not contribute exponentially in $N$. Now the integral of the conjugate variables is over $\mathbb{C}$ (the imaginary axis), and we have:
\begin{align*}
    V&\propto \int d\hat{\Omega}\,e^{N\left[\hat{\lambda} Q_\phi+\hat{\omega} q + \hat{\eta} m + \hat{\xi} M_\phi+\hat{t}D_\phi + \hat{\theta} Z\right]}\int d{\bf x}\,e^{-\sum_k\left[\hat{\lambda} \phi^2(x_k) + \hat{\omega} x_k^2+\hat{\eta} x_k +\hat{\xi} \phi(x_k) + \hat{t} \phi'(x_k) + \hat{\theta} x_k\phi(x_k)\right]}\\
    &=\int d\hat{\Omega}\,e^{N\left[\hat{\lambda} Q_\phi+\hat{\omega} q + \hat{\eta} m + \hat{\xi} M_\phi+\hat{t}D_\phi + \hat{\theta} Z\right]}I^N,
\end{align*}
where we have defined the following integral (over $\mathbb{R}$):
\begin{align}
    I=\int dx \,e^{-\left[\hat{\lambda} \phi^2(x) + \hat{\omega} x^2+\hat{\eta} x +\hat{\xi} \phi(x) + \hat{t} \phi'(x) + \hat{\theta}\, x\,\phi(x)\right]}.
\end{align}
Now, it is evident why our choice of $\phi$, see Eq.~\eqref{eq:def_phi_sign}, is convenient. In fact, we can divide this integral into the contribution $|x|\leq 1$ and the contribution $|x|>1$. Hence, we find it useful to write $I=I_<+I_>$, where we define:
\begin{align}
\begin{split}
&I_<=\int_{|x|\leq 1}dx\,e^{-\left[\hat{\lambda} x^2 + \hat{\omega} x^2+\hat{\eta} x +\hat{\xi} x + \hat{t}  + \hat{\theta}\, x^2\right]}\\
&I_>=\int_{|x|>1}dx\,e^{-\left[\hat{\lambda} + \hat{\omega} x^2+\hat{\eta} x +\hat{\xi} \,\text{sign}(x) + \hat{\theta}\, x\,\text{sign}(x)\right]}.
\end{split}
\end{align}
Now, these are Gaussian integrals that are not hard to compute. For the first integral we get:
\begin{align}
I_<=e^{-t}\Theta_1
\end{align}
and for the second integral (by splitting the positive and negative part), we get:
\begin{align}
    I_>=e^t\left[\Theta_2+\Theta_3 \right]
\end{align}
where $\Theta_1,\Theta_2,\Theta_3$ were defined in Eq.~\eqref{eq:scs_thetas}. From this, we get that
\begin{align*}
\lim_{N\to\infty}\frac{1}{N}\log V=\text{extr}_{\hat\lambda,\hat\omega,\hat\xi,\hat\theta,\hat t,\hat\eta}v
\end{align*}
with 
\begin{align*}
v&:=\left\{\hat\lambda Q_\phi+\hat\omega q + \hat\eta m + \hat\xi M_\phi+\hat t(D_\phi-1) + \hat\theta Z+\log({\Theta}_1+{\Theta}_2+{\Theta}_3) \right\},
\end{align*}
where we used the Method of Steepest descent to evaluate the integral in the $N\to\infty$ limit, thus recovering Eq.~\eqref{eq:scs_v}

\subsection{Determinant term}
This terms is found in the same way as in Appendix.~\ref{app:determinant_calc}, since in both cases the Jacobian follows to the elliptic law. Hence Eq.~\eqref{eq:scs_d_alpha} is the same as Eq.~\eqref{app:det_eqn} applied to this case.

\subsection{Quenched complexity}
\label{sec:quenched_scs}
We give a very short summary of the computation of the quenched complexity in the case $\alpha=0$. In this case we will work with $n$ replicas ${\bf x}^a$, and so we define the following order parameters:
\begin{align*}
&M_\phi=\frac{1}{N}\sum_i\phi(x_i^a) \quad\quad&m&=\frac{1}{N}\sum x_i^a\\
&\hat{Q}_\phi^{ab}=\frac{1}{N}\sum_i\phi(x_i^a)\phi(x_i^b)\quad\quad&\hat{q}^{ab}&=\frac{1}{N}\sum_i x_i^ax_i^b\\
&D_\phi=\frac{1}{N}\sum_i\phi'(x_i^a).
\end{align*}
Like before, let us define
\begin{align}
\begin{split}
    \Omega({\bf x}^a)&=
    \delta\left(m-\frac{1}{N}\sum_ix^a_i\right)
    \delta\left(q-\frac{1}{N}\sum_i(x^a_i)^2\right)\delta\left(D_\phi-\frac{1}{N}\sum_i \phi'(x^a_i)\right)\times\\
    &
    \times\delta\left(M_\phi-\frac{1}{N}\sum_i\phi(x^a_i)\right)
    \delta\left(Q_\phi-\frac{1}{N}\sum_i\phi^2(x^a_i)\right)
\end{split}
\end{align}
and moreover we will indicate $F_i^a\equiv F_i({\bf x}^a)$ the force for the $i-$th replica. Hence by means of the replica trick, the complexity reads:
\begin{align}
   \Sigma(m,M_\phi,q,Q_\phi,D_\phi):=\lim_{N\to\infty}\frac{1}{N} \mathbb{E}\log\mathcal{N}(m,M_\phi,q,Q_\phi,D_\phi)=\lim_{N\to\infty,n\to 0}\frac{\mathbb{E}[\mathcal{N}^n]-1}{Nn}
\end{align}
where we have 
\begin{align}
    \mathbb{E}[\mathcal{N}^n]=\int \prod_a d{\bf x}^a\Omega({\bf x}^a)\mathbb{E}\left[\prod_a \delta({\bf F}^a)\right]\mathbb{E}\left[\prod_a|\det\partial {\bf F}^a|\Bigg|{\bf F}^a=0,\quad a=0,\ldots,n\right].
\end{align}

\noindent Like for the annealed computation, we have three terms that contribute to the complexity: the probability term, the determinant term, and the phase-space factor. In order to carry out the calculations, we assume a RS (replica symmetric) ansatz on the distribution of the replicas:
\begin{align}
    \hat{Q}_\phi=\begin{pmatrix}
        Q_\phi & \cdots & \tilde{Q} \\
         & \ddots & \\
        \tilde{Q} & \cdots & Q_\phi\\
    \end{pmatrix},\quad\quad 
    \hat{q}=
     \begin{pmatrix}
        q& \cdots & \tilde{q} \\
         & \ddots & \\
        \tilde{q} & \cdots & q\\
    \end{pmatrix}.
\end{align}
When opening up the delta functions to compute the phase space term, we introduce conjugate parameters $\hat{\Lambda}_{ab},\hat{\Omega}_{ab},\hat{\eta}_a,\hat\xi_a,\hat{t}_a$. The RS ansatz is also applied to these parameters, which are grouped in the following way: 
\begin{align}
    \hat{\Lambda}=\begin{pmatrix}
        2\hat{\lambda}_1 & \cdots & \hat{\lambda}_0 \\
         & \ddots & \\
        \hat{\lambda}_0 & \cdots & 2\hat{\lambda}_1\\
    \end{pmatrix}
    \quad\quad 
    \hat{W}=\begin{pmatrix}
        2\hat{\omega}_1 & \cdots & \hat{\omega}_0 \\
         & \ddots & \\
        \hat{\omega}_0 & \cdots & 2\hat{\omega}_1\\
    \end{pmatrix}
\end{align}
and 
\begin{align}
    \eta_a=\hat\eta,\quad \xi_a=\hat\xi,\quad t_a=\hat{t}.
\end{align}
We may also write 
\begin{align}
\begin{split}
    &\hat\Lambda_{ab}=\hat\lambda_0+\delta_{ab}(2\hat\lambda_1-\hat\lambda_0)\\
    &\hat W_{ab}=\hat w_0+\delta_{ab}(2\hat w_1-\hat w_0).
\end{split}
\end{align}
Let us be a bit sloppy to give a quick recap of the calculation. One proceeds by the introduction of delta functions encoding the overlaps between replicas, and opening them up in Fourier:
\begin{align*}
    \mathbb{E}[\mathcal{N}^n]&=\int\prod dQ_\phi^{ab}dq^{ab}d\Lambda_{ab}dW_{ab}\left[\,\mathbb{P}\cdot \mathbb{D}\cdot\int\prod d{\bf x}^aV(\{{\bf x}^a\})\right]\\
    &=\int\prod dQ_\phi^{ab}dq^{ab}d\Lambda_{ab}dW_{ab}\,e^{Nn(v+p_0+d_0)}
\end{align*}
where $V$ represents the phase space term (containing those expressions that still depend on the ${\bf x}^a$) and that has to be integrated, and $\mathbb{P}$, $\mathbb{D}$ only depend on the order parameters and conjugate variables (the ${\bf x}$ dependence being lost thanks to the introduction of the delta functions encoding the overlaps). At this point we introduce the RS ansatz and apply the Method of Steepest Descent on the action:
\begin{align}
\Sigma(m,M_\phi,q,Q_\phi,D_\phi)=d_0+\text{extr}_{\tilde{Q},\tilde{q},\hat{\lambda}_0,\hat{\lambda}_1,\hat{\omega}_0,\hat{\omega}_1}(v+p_0).
\end{align}

A careful calculation of the integrals reveals the following:

\subsubsection{Probability term}
\begin{align}
\begin{split}
    p_0:&=\lim_{N\to\infty,n\to 0}\frac{\log \mathbb{P}}{Nn}= -\frac{1}{2}\log[2\pi g^2(Q_\phi-\tilde{Q})]-\frac{1}{2}\frac{\tilde{Q}}{Q_\phi-\tilde{Q}}\\&-\frac{1}{2g^2}\frac{q (Q_\phi - 2 \tilde{Q}) + g^2 J_0^2 M_\phi^2 (Q_\phi - \tilde{Q}) + \tilde{q} \tilde{Q} + 
   2 g J_0 m M_\phi (-Q_\phi + \tilde{Q})}{(Q_\phi - \tilde{Q})^2}
\end{split}
\end{align}

\subsubsection{Phase-space term}
The phase space term can be computed to be equal to:
\begin{align}
\begin{split}
    v:&=\lim_{N\to\infty,n\to 0}\frac{\log V}{Nn} =\hat\lambda_1Q_\phi-\frac{\tilde{Q}}{2}\hat\lambda_0+\hat{w}_1q-\frac{\tilde{q}}{2}\hat{w}_0+\hat\eta m+\hat\xi M_\phi+\hat{t}D_\phi\\
    &+\int du_1 du_2\mathcal{G}(u_1,\hat{w}_0)\mathcal{G}(u_2,\hat\lambda_0)\log[I_k(u_1,u_2)+I_u(u_1,u_2)+I_p(u_1,u_2)]
\end{split}
\end{align}
with
\begin{align}
    I_k(u_1,u_2)=\int_{-1}^1dx\,e^{-\frac{1}{2}a_kx^2+b_kx+c_k}=\frac{\sqrt{\pi}e^{\frac{b_k^2}{2a_k}+c_k}}{\sqrt{2a_k}}\left[\text{Erf}\left(\frac{a_k-b_k}{\sqrt{2a_k}}\right)+\text{Erf}\left(\frac{a_k+b_k}{\sqrt{2a_k}}\right)\right]
\end{align}
and
\begin{align}
    I_u(u_1,u_2)=\int_{-\infty}^{-1}dx\, e^{-\frac{1}{2}a_{up}x^2+b_{up}x+c_u}=\frac{\sqrt{\pi}e^{\frac{b_{up}^2}{2a_{up}}+c_u}}{\sqrt{2a_{up}}}\left[1-\text{Erf}\left(\frac{a_{up}+b_{up}}{\sqrt{2a_{up}}}\right)\right]
\end{align}
and 
\begin{align}
    I_p(u_1,u_2)=\int_{1}^{\infty}dx\,e^{-\frac{1}{2}a_{up}x^2+b_{up}x+c_p}=\frac{\sqrt{\pi}e^{\frac{b_{up}^2}{2a_{up}}+c_p}}{\sqrt{2a_{up}}}\left[1-\text{Erf}\left(\frac{a_{up}-b_{up}}{\sqrt{2a_{up}}}\right)\right]
\end{align}
where we defined:
\begin{align}
\begin{split}
    &a_k=2(\hat{w}_1+\hat{\lambda}_1)-(\hat{w}_0+\hat{\lambda}_0)\\
    &b_k=u_1+u_2-\hat{\eta}-\hat\xi\\
    &c_k=-\hat{t}\\
    &a_{up}=2\hat{w}_1-\hat{w}_0\\
    &b_{up}=u_1-\hat\eta\\
    &c_u=\hat\xi-\frac{1}{2}(2\hat{\lambda}_1-\hat{\lambda}_0)-u_2\\
    &c_p=-\hat\xi-\frac{1}{2}(2\hat{\lambda}_1-\hat{\lambda}_0)+u_2.
\end{split}
\end{align}
The function $\mathcal{G}$ is the Gaussian kernel $\mathcal{G}(x,a):=\frac{e^{\frac{1}{2a}x^2}}{\sqrt{2\pi(-a)}}$. 

\subsubsection{Final expression}
The final expression reads:
\begin{align}
\Sigma(D_\phi)=d_0+\text{extr}_{m,M_\phi,Q_\phi,\tilde{Q},q,\tilde{q},\hat\eta,\hat\xi,\hat{t},\hat{\lambda}_1,\hat{\lambda}_0,\hat\omega_0,\hat\omega_1}\left[p_0+v\right].
\end{align}
Note that with our expressions we could use $n=2$ above, to obtain the second moment of $\mathcal{N}$, which could then be used for a second moment approach to prove that when $J_0=0$ the annealed complexity is exact.

\subsubsection{Reduction to the annealed complexity}
In the paramagnetic case the quenched complexity gives back the annealed one with $\tilde{Q}=\tilde{q}=\hat\eta=\hat\xi=\hat\lambda_0=\hat\omega_0=0$ and $\hat\omega_1\to\hat\omega$, $\hat\lambda_1\to\hat\lambda$. The second point where annealed and quenched computations match is the cavity point, defined as the prolongation of the FFP solution into the unstable phase. At this point, the annealed is recovered by imposing $\tilde{Q}=0,\,\,\tilde{q}=(gJ_0M_\phi)^2$ which imply $\hat\omega_0=\hat\lambda_0=0$, thus allowing us to simplify the convolution integral in $v$ by using delta functions on $\hat\omega_0=\hat\lambda_0$. The other order parameters satisfy the same equations as in the annealed case, Eqs.~\eqref{eq:ffp_scs_params}, \eqref{eq:ffp_scs_mphi}, \eqref{eq:ffp_scs_qphi}, and the same result is thus recovered.

\newpage
\section{Gradient along the geodesic path}
\label{app:gradient_max_barr}
In this Appendix, adapted from \cite{pacco2024curvature}, we derive the statistical distribution of the vector $ \nabla h({\bm\sigma}[\gamma;f=0])$ at each configuration along the geodesic path, parametrized as \eqref{eq:path} with $f(\gamma)=0$. We denote each configuration along the path simply as ${\bm \sigma}(\gamma):= {\bm\sigma}[\gamma;0]$, and take $p=3$. For any fixed value of $\gamma\in(0,1)$,  plugging the expression for ${\bm\sigma}(\gamma)$ inside the formula for $\nabla h$ one gets
\begin{align}
    \nabla h({\bm\sigma}(\gamma))=\gamma^2\,\nabla h({\bm\sigma}_1)+\beta^2\,\nabla h({\bm\sigma}_0)+\gamma\,\beta\,\nabla^2h({\bm\sigma}_0)\cdot {\bm\sigma}_1,
\end{align}
which using a change of basis and conditioning on the properties of ${\bm \sigma}_a$ reduces to:
\begin{align}
\begin{split}
    \nabla h({\bm\sigma}(\gamma))&=\gamma^2 3 \sqrt{2N}\epsilon_1{\bm\sigma}_1+\beta^2 3 \sqrt{2N}\epsilon_0{\bm\sigma}_0+\gamma\beta\nabla^2 h({\bm\sigma}_0)\cdot \left(q{\bm\sigma}_0-\sqrt{1-q^2}{\bf e}_{N-1}({\bm\sigma}_0)\right)\\
    &=\gamma^2 3 \sqrt{2N}\epsilon_1{\bm\sigma}_1+\beta^2 3\sqrt{2N}\epsilon_0{\bm\sigma}_0+\gamma\beta\left[ 
    6q\, \sqrt{2N}\epsilon_0{\bm\sigma}_0-\sqrt{1-q^2}\,\nabla^2h({\bm\sigma}_0)\cdot {\bf e}_{N-1}({\bm\sigma}_0)
    \right].
\end{split}
\end{align}
A small reminder: $\beta$ is the parameter that fixes ${\bm\sigma}(\gamma)$ on the surface of the hypersphere. We now consider the vector ${\bf g}(\gamma):={\bf g}({\bm \sigma}(\gamma))$, that is the projection of $\nabla h$ on the $(N-1)$-dimensional tangent plane $\tau[{\bm\sigma}(\gamma)]$. We choose a basis of this tangent plane in such a way that the first $N-2$ elements span the subspace orthogonal to both ${\bm\sigma}_0,{\bm\sigma}_1$; we denote them with ${\bf x}_1,\ldots,{\bf x}_{N-2}$ as in the main text, see Sec.~\ref{sec:en_land_settin_two_point}. The remaining vector equals to
\begin{align}
    {\bf e}_{N-1}(\gamma):= {\bf e}_{N-1}({\bm \sigma}(\gamma))=A({\bm\sigma}_0+C{\bm\sigma}_1), \quad \quad A=(1+2Cq+C^2)^{-1/2}, \quad \quad C=-\frac{\gamma q+\beta}{\gamma +\beta q}.
\end{align}
This vector is orthogonal to  $ {\bm\sigma}(\gamma)$ and has unit norm. It is simple to check that it coincides with the tangent vector to the geodesic path at the point ${\bm \sigma}(\gamma)$. Therefore, the gradient ${\bf g}(\gamma)= {\bf g}^\parallel(\gamma)+ {\bf g}^\perp(\gamma)$, where ${\bf g}^\parallel(\gamma) = ({\bf g}(\gamma) \cdot {\bf e}_{N-1}(\gamma)) {\bf e}_{N-1}(\gamma)$ is the component tangent to the geodesic, while ${\bf g}^\perp(\gamma)$ is the orthogonal one. Using that 
\begin{align}
    {\bf e}_{N-1}(\gamma)=
    (A+ACq){\bm\sigma}_0-AC\sqrt{1-q^2}{\bf e}_{N-1}({\bm\sigma}_0)
\end{align}
 we see that the tangent component equals to:
\begin{align}
    \begin{split}
       \nabla h({\bm\sigma}(\gamma))\cdot {\bf e}_{N-1}(\gamma)&=3\gamma^2  \sqrt{2N}\epsilon_1(A q + A C)+3\beta^2 \sqrt{2N}\epsilon_0(A+ACq)+\gamma\beta\bigg[6q\,\sqrt{2N}\epsilon_0(A+ACq)\\
        &-\sqrt{1-q^2}\,{\bf e}_{N-1}(\gamma)\cdot \nabla^2 h({\bm\sigma}_0)\cdot {\bf e}_{N-1}({\bm\sigma}_0)
        \bigg].
    \end{split}
\end{align}
One can check explicitly that this expression vanishes at the point $\gamma$ where the geodesic energy profile reaches its maximum. Consider now the orthogonal component of the gradient. We see that component-wise in the chosen local basis it holds for $i<N-1$:
\begin{align}
    \frac{\nabla h({\bm\sigma}(\gamma))}{\sqrt{N-1}}\cdot{\bf x}_i=-\gamma\beta\,\sqrt{1-q^2}{\bf x}_i\cdot \frac{\nabla^2 h({\bm\sigma}^0)}{\sqrt{N-1}}\cdot {\bf e}_{N-1}({\bm\sigma}^0)=-\gamma\beta\sqrt{1-q^2}\,m_{i,N-1}^0.
\end{align}
where $m_{i,N-1}^0$ was defined around Eq.~\eqref{eq:app:annealed_matrix_0}. Therefore, the orthogonal component ${\bf g}^\perp(\gamma)$ is proportional to the vector that makes up the last column of the shifted Hessian $\tilde{\mathcal{H}}({\bm \sigma}_0)$ (neglecting the last component of the column). This vector is in general not vanishing at the point that corresponds to the maximum of the energy profile along the geodesic. We define the normalized vector 
\begin{equation}
    {\bf v}_{\rm Hess}=  Z \sum_{i=1}^{N-2} [{\bf x}_i \cdot \tilde{\mathcal{H}}({\bm \sigma}_0) \cdot {\bf e}_{N-1}({\bm \sigma}_0)] \, {\bf x}_i =\frac{1}{\sqrt{\sum_{i=1}^{N-2} [m_{i \, N-1}^0]^2}} \tonde{m_{1 \, N-1}^0, m_{2 \, N-1}^0, \cdots, m_{N-2 \, N-1}^0,0}^T,
\end{equation}
which is orthogonal to both ${\bm \sigma}_a$ ($Z$ is the normalization factor). Plugging this into \eqref{eq:PrePAth} and making use of the fact that the entries $m_{i N-1}^a$ of the Hessians are uncorrelated to all other entries of the Hessian matrices, we see that the only matrix element that survives\footnote{the others are elements of the first block, which have zero average} in Eq. \eqref{eq:PrePAth} is:
\begin{equation}
    \mathbb{E} \quadre{ {\bf v}_{\rm Hess} \cdot \frac{\tilde{\mathcal{H}}({\bm \sigma}_0)}{\sqrt{2N}} \cdot {\bf e}_{N-1}({\bm \sigma}_0) }=  \mathbb{E} \quadre{\sqrt{\frac{\sum_{i=1}^{N-2} [m_{i \, N-1}^0]^2}{2 N}}}= \sqrt{\frac{\Delta^2}{2}}+\mathcal{O}\tonde{\frac{1}{N}}\stackrel{p=3}{=} \sqrt{\frac{3(1-q^2)}{1+q^2}}+\mathcal{O}\tonde{\frac{1}{N}},
\end{equation}
which leads to Eq. \ref{eq:energy_profile_Hess} in the main text. \\

\newpage
\section{Three-point complexity: some calculations}
\label{app:three_point}
In this Appendix we give ideas behind the doubly-annealed three-point complexity calculation \cite{pacco_triplets_2025}.
\subsection{Probability term in the doubly-annealed calculation}
Like in the previous complexity calculations, here we have a phase-space term, a probability term, and a determinant term. The phase space and determinant terms are the easiest: the first one is just an integral over the phase space; the second can be shown to have the same form as for the one-point case when $N\to\infty$ \cite{ros2019complex, ros2019complexity} (since the conditioning creates at most finite-rank perturbations that do not contribute to the determinant). The hardest part is therefore the probability term, call it $P(\epsilon_2, q_0, q_1| \epsilon_1, \epsilon_0, q)$, which is the probability that the point ${\bf s}_2$ is a fixed point of energy density $\epsilon_2$ conditioned to ${\bf s}_0,{\bf s}_1$. If we write this using the field $h$ and ${\bm\sigma}={\bf s}/\sqrt{N}$ as in Sec.~\ref{sec:curvature_driven_paths} we have:
\begin{align}
P(\epsilon_2, q_0, q_1| \epsilon_1, \epsilon_0, q)=\mathbb{E}\left[\delta(h({\bm\sigma}_2)-\sqrt{2N}\epsilon_2)\,\delta \tonde{{\bf g}({\bm\sigma}_2)}\bigg|\grafe{
    \begin{subarray}{l}
    h^a = \sqrt{2N}\epsilon_a\\ 
    {\bf g}^a={\bf 0}, \; \forall  a=0,1 \end{subarray}}\right]
\end{align}
where we are implicitly assuming that we realize the overlaps imposed by delta functions in the integration over phase space. Then we can write:
\begin{align}\label{eq:DistRatio}
   P(\epsilon_2, q_0, q_1| \epsilon_1, \epsilon_0, q)=\frac{P_3({\bm \epsilon}, {\bf q})}{P_2(\epsilon_1, \epsilon_0, q)},\quad{\bm\epsilon}=(\epsilon_0,\epsilon_1,\epsilon_2),\quad{\bf q}=(q_0,q_1,q_2)
\end{align}
where $P_3({\bm \epsilon}, {\bf q})$ is the joint probability density of ${\bf g}^a=0$ and $h^a= \sqrt{2N} \epsilon_a$ for $a=0,1,2$, while $P_2(\epsilon_1, \epsilon_0, q)$ is the analogous quantity for $a=0,1$. The term $P_2(\epsilon_0, \epsilon_1, q)$ can be read from the calculation of the two-point complexity in \cite{ros2019complexity}. We therefore focus on the calculation of the numerator. Since for the pure $p$-spin model it holds $\nabla h({\bm \sigma})= {\bf g}({\bm\sigma}) + p\,h({\bm \sigma}){\bf e}_N({\bm \sigma})$, we can include the conditioning on the value of the energies as part of the conditioning on the values of the (unconstrained) gradients. Hence, we have to compute the joint probability density function of the three gradients  evaluated at
\begin{align*}
\nabla h({\bm \sigma}_a)={\bf 0}+p\sqrt{2N}\epsilon_a\,{\bf e}_N({\bm \sigma}_a).
\end{align*}
Using that the gradients components are centered Gaussian random variables, we get:
\begin{align}
    \begin{split}
   P_3({\bm \epsilon}, {\bf q})=     \mathbb{P}(\{\nabla h({\bm \sigma}_a)=p\sqrt{2N}\epsilon_a\,{\bf e}_N({\bm\sigma}_a)\}_{a=0}^2)=\frac{e^{-\frac{1}{2}\sum_{a,b=0}^2\nabla h^a \cdot\left[\hat{C}^{-1}\right]^{ab}\cdot\nabla h^b}}{(2\pi)^{\frac{3N}{2}}(\det\hat{C})^{\frac{1}{2}}}=\frac{e^{-N\,F({\bm\epsilon},{\bf q})}}
        {(2\pi)^{\frac{3N}{2}}(\det\hat{C})^{\frac{1}{2}}}
    \end{split}
\end{align}
where 
\begin{align}
\begin{split}\label{eq:effe}
F({\bm\epsilon},{\bf q})&=\epsilon_0^2\,Y_{0}^{(p)}({\bf q})+\epsilon_0\epsilon_1\,Y_{01}^{(p)}({\bf q})+\epsilon_0\epsilon_2\,Y_{02}^{(p)}({\bf q})
+\epsilon_1\epsilon_2\,Y_{12}^{(p)}({\bf q})+\epsilon_1^2Y_{1}^{(p)}({\bf q})+\epsilon_2^2Y_{2}^{(p)}({\bf q})
\end{split}
\end{align}
 with $Y_{ab}^{(p)}({\bf q}):=p^2\left[{\bm \sigma}_a\cdot[\hat{C}^{-1}]^{ab}\cdot{\bm \sigma}_b +
 {\bm \sigma}_b\cdot[\hat{C}^{-1}]^{ba}\cdot{\bm \sigma}_a\right]$ for $a<b$ (and only one term for $a=b$) and where the correlation matrix $\hat{C}$ takes the following form (in any orthonormal basis of $\mathbb{R}^N$), by Eq.~\eqref{app:eq:CovGradComp}:
\begin{align}
    C_{ij}^{ab}=\mathbb{E}\left[  \nabla h_i^a\nabla h^b_j\right]=p\, Q^{p-1}_{ab}\delta_{ij}+p(p-1)Q^{p-2}_{ab}\sigma_j^a\sigma_i^b, \quad \quad Q_{ab}={\bm\sigma}_a\cdot{\bm\sigma}_b.
\end{align}
 The matrix $\hat C$ is high-dimensional ($3N \times 3N$); since $F({\bm \epsilon},{\bf q})$ requires to compute only contractions of its inverse with the vectors ${\bm \sigma}_a$, a convenient way to proceed is to identify subsets of vectors (including the ${\bm \sigma}_a$) that are closed under the action of  $\hat{C}$: this allows to invert the correlation matrix only in the corresponding subspace, reducing the complexity of the inversion \cite{ros2019complex}.  We therefore proceed with the following steps: (i) find  a set of vectors that are closed under the action of $\hat{C}$, (ii) orthogonalize them, (iii)  invert the matrix $\hat{C}$ (or better its action on this set of vectors), and (iv) extract the quantities ${\bm\sigma}_a\cdot[\hat{C}]^{ab}\cdot{\bm\sigma}_b$ for any $a,b=0,1,2$.\\
\noindent It can be easily checked that the following set of $(3N)$-dimensional vectors is closed under the action of $\hat{C}$:
\begin{align}
\label{eq:app:xi_vecs}
\begin{split}
    &{\bm \xi}^1=({\bm \sigma}_0,{\bf 0},{\bf 0})\quad\quad{\bm \xi}^2=({\bm \sigma}_1,{\bf 0},{\bf 0})\quad \quad {\bm \xi}^3=({\bm \sigma}_2,{\bf 0},{\bf 0})\\
    &{\bm \xi}^4=({\bf 0},{\bm \sigma}_0,{\bf 0})\quad\quad{\bm \xi}^5=({\bf 0},{\bm \sigma}_1,{\bf 0})\quad \quad {\bm \xi}^6=({\bf 0},{\bm \sigma}_2,{\bf 0})\\
    &{\bm \xi}^7=({\bf 0},{\bf 0},{\bm \sigma}_0)\quad\quad{\bm \xi}^8=({\bf 0},{\bf 0},{\bm \sigma}_1)\quad \quad {\bm \xi}^9=({\bf 0},{\bf 0},{\bm \sigma}_2).
\end{split}
\end{align}
In our notation, the vector ${\bm \xi}^1=({\bm \sigma}_0,{\bf 0},{\bf 0})$ is obtained concatenating the $N$-dimensional vector ${\bm \sigma}_0$ with two other $N$-dimensional vectors with zero entries, and similar for the others. The vectors in this family are linearly independent, but not orthogonal to each others. We therefore introduce a set of orthogonal $(3N)$-dimensional vectors ${\bm\chi}$:
\begin{align}
\begin{aligned}
 {\bm\chi}^{1} & = {\bm\xi}^1 \quad\quad & {\bm\chi}^4 & = {\bm\xi}^4 - q{\bm\xi}^5 \quad\quad & {\bm\chi}^7 & = {\bm\xi}^7 - q_0{\bm\xi}^9 \\
 {\bm\chi}^2 & = {\bm\xi}^2 - q{\bm\xi}^1 \quad\quad & {\bm\chi}^5 & = {\bm\xi}^5 \quad\quad & {\bm\chi}^8 & = {\bm\xi}^8 + c_1{\bm\xi}^7 + d_1{\bm\xi}^9 \\
 {\bm\chi}^3 & = {\bm\xi}^3 + c{\bm\xi}^1 + d{\bm\xi}^2 \quad\quad & {\bm\chi}^6 & = {\bm\xi}^6 + c{\bm\xi}^4 + d{\bm\xi}^5 \quad\quad & {\bm\chi}^9 & = {\bm\xi}^9,
\end{aligned}
\end{align}
with 
\begin{align}
\begin{aligned}
c&=\frac{q q_1-q_0 }{1 - q^2}\quad\quad d=\frac{qq_0 - q_1}{1 - q^2}\quad\quad c_1=\frac{ q_0 q_1-q}{1 - q_0^2}\quad\quad d_1=\frac{qq_0 -  q_1}{1 - q_0^2}.
\end{aligned}
\end{align}
\noindent The key trick to compute the action of $\hat C^{-1}$ on these vectors is to write the correlation matrix as $\hat{C}=p(\hat{D}+\hat{O})$,  where 
\begin{align}
    \begin{split}
        &D_{ij}^{ab}=\delta_{ab}[\delta_{ij}+(p-1)\sigma_{a,i}\sigma_{a,j}]\\
        &O_{ij}^{ab}=(1-\delta_{ab})[\delta_{ij}Q_{ab}^{p-1}+(p-1)Q_{ab}^{p-2}\sigma_{b,i}\sigma_{a,j}].
    \end{split}
\end{align}
This implies that $\hat{C}^{-1}=p^{-1}\hat{D}^{-1}[1+\hat{O}\hat{D}^{-1}]^{-1}$;
moreover, the Sherman-Morrison formula allows us to write
\begin{align}
    [\hat{D}^{-1}]^{ab}_{ij}=\delta_{ab}[\delta_{ij}-(p-1)p^{-1}\sigma_{a,i}\sigma_{a,j}].
\end{align}
\noindent Our goal is to determine what is the action of the matrices $\hat{D}^{-1}$ and $[1+\hat{O}\hat{D}^{-1}]$ on the vectors ${\bm \chi}^i$, and to then invert the matrix $[1+\hat{O}\hat{D}^{-1}]$ restricted to this subspace. 
One might notice that, at variance with \cite{ros2019complex}, we are working here with vectors ${\bm \chi}^i$ that are orthogonal but not necessarily of unit norm, except for those vectors that only contain one configuration ${\bm \sigma}_a$, which are ${\bm\chi}^1,{\bm\chi}^5,{\bm\chi}^9$. This is due to the fact that the quadratic form \eqref{eq:effe} is a function of the matrix elements of $\hat C^{-1}$ only along these three directions: the normalization of the remaining ${\bm \chi}^i$ will not enter in the final result, and it can therefore be safely neglected, simplifying the  formalism.\\

\noindent To compute the action of $\hat{C}$ on the ${\bm\chi}$ vectors, it is convenient to first do it on the ${\bm\xi}$ vectors and then to make a change of basis. For convenience, we introduce the following $(3N)$-dimensional vector:
\begin{align}
    {\bf v}(n,\beta):=(0,\ldots,\underbrace{{\bm\sigma}_n}_{\beta},0,\ldots)\quad\quad n,\beta\in\{0,1,2\}
\end{align}
where we basically put the $n$-th configuration in the location $\beta$ (with $n,\beta$ starting at 0). In this way, each ${\bm\xi}$ vector is written as a ${\bf v}$ vector (notice that this can be generalized to an arbitrary number of configurations). Then, one gets
\begin{align}
    \hat{D}^{-1}{\bf v}(n,\beta)={\bf v}(n,\beta)-\frac{p-1}{p}Q_{n\beta}\,{\bf v}(\beta,\beta)
\end{align}
and similarly
\begin{align}
\hat{O}{\bf v}(n,\beta)=\sum_{\alpha\neq\beta}\left[Q_{\alpha\beta}^{p-1}{\bf v}(n,\alpha)+(p-1)Q_{\alpha\beta}^{p-2}Q_{\alpha n}\,{\bf v}(\beta,\alpha)\right].
\end{align}
From this, one can obtain the explicit form of the matrix $[\hat{D}^{-1}]_{\xi}$\footnote{the subscript indicates the basis in which it is expressed} (see Appendix of Ref.~\cite{pacco_triplets_2025}) and of $[\hat{O}^{-1}]_{\xi}$. Then by denoting $U_{\xi\to\chi}$ the change of basis matrix, we define 
\begin{align}
    \hat{M}:=\mathbb{I}_{9}+[\hat{O}\hat{D}^{-1}]_\chi=U^{-1}_{\xi\to\chi}(\mathbb{I}_9+[\hat{O}\hat{D}^{-1}]_\xi)^{-1}U_{\xi\to\chi}
\end{align}
and obtain that 
\begin{align}
\begin{split}
&Y_{0}^{(p)}({\bf q})=[\hat{M}^{-1}]_{11},\quad Y_{1}^{(p)}({\bf q})=[\hat{M}^{-1}]_{55},\quad Y_{2}^{(p)}({\bf q})=[\hat{M}^{-1}]_{99},\quad Y_{01}^{(p)}({\bf q})=[\hat{M}^{-1}]_{15}+[\hat{M}^{-1}]_{51}\\& Y_{02}^{(p)}({\bf q})=[\hat{M}^{-1}]_{19}+[\hat{M}^{-1}]_{91},\quad Y_{12}^{(p)}({\bf q})=[\hat{M}^{-1}]_{59}+[\hat{M}^{-1}]_{95}.\\
\end{split}
\end{align}
The idea behind this calculation was used by us also to compute the probability term in the quenched case \cite{pacco_quenched_triplets_2025}. However, in that case there are $1+m+k$ replicas: 1 for ${\bm\sigma}_0$ which is annealed and so does not need to be replicated; $m$ for ${\bm\sigma}_1$; $k$ for ${\bm\sigma}_2$. It turns out that for the RS computation we need $18$ vectors such as those in \eqref{eq:app:xi_vecs}, in order for their action on $\hat{C}$ to be closed (for the two-point complexity only 5 were needed \cite{ros2019complexity}). Notice that, moreover, this method can be extended to any $n-$point complexity calculation, but inverting $[\hat{O}\hat{D}^{-1}]_\xi$ will quickly increase in difficulty (you need $n^2$ vectors).

\subsection{The $Y^{(3)}$ variables for the doubly-annealed case}
\label{app:Y_vars}
The expressions above can be easily solved with a mathematical software, but we were not able to find a simple closed form solution of the $Y^{(p)}$ variables as a function of $p$. For $p=3$ the formulas are already cumbersome:
\begin{equation}\label{eq:yp31}
\begin{split}
Y_0^{(3)}({\bf q})=&-((6 q^6 + 6 q_0^6 + 2 q_0^2 (-1 + q_1^2)^3 - (-1 + q_1^2)^3 (1 + q_1^2) - 
     12 q^5 q_0 q_1 (2 + q_1^2) + 3 q_0^4 (1 - 4 q_1^2 + 3 q_1^4)\\
     &- 
     4 q q_0 q_1 (-(-1 + q_1^2)^3 + 3 q_0^4 (2 + q_1^2) + 
        3 q_0^2 (-1 + q_1^4)) + 
     12 q^3 q_0 q_1 (1 - q_1^4\\
     &+ q_0^2 (1 - 12 q_1^2 + q_1^4)) + 
     3 q^4 (1 - 4 q_1^2 + 3 q_1^4 + q_0^2 (-4 + 28 q_1^2 + 6 q_1^4))+ 
     2 q^2 ((-1 + q_1^2)^3\\
     &+ q_0^4 (-6 + 42 q_1^2 + 9 q_1^4) + 
        q_0^2 (3 - 30 q_1^2 + 33 q_1^4 - 6 q_1^6)))/((-1 + q^2 + q_0^2 - 
       2 q q_0 q_1 + q_1^2)^2\\
       &(-1 + q^4 + q_0^4 + 8 q q_0 q_1 - 
       4 q_0^2 q_1^2 + q_1^4 + 2 q^2 (-2 q_1^2 + q_0^2 (-2 + q_1^2)))))
 \end{split}
 \end{equation}

 \begin{equation}
 \begin{split}
Y_1^{(3)}({\bf q})=&
(-6 q^6 + (-1 + q_0^2)^3 (1 + q_0^2) + 12 q^5 q_0 (2 + q_0^2) q_1 - 
   2 (-1 + q_0^2)^3 q_1^2 - 3 (1 - 4 q_0^2 + 3 q_0^4) q_1^4 - 6 q_1^6\\& -
   12 q^3 q_0 q_1 (1 - q_0^4 + (1 - 12 q_0^2 + q_0^4) q_1^2) + 
   4 q q_0 q_1 (-(-1 + q_0^2)^3 + 3 (-1 + q_0^4) q_1^2 + 
      3 (2 + q_0^2) q_1^4) \\
      &+ 
   2 q^2 (-(-1 + q_0^2)^3 + 
      3 (-1 + q_0) (1 + q_0) (1 - 9 q_0^2 + 2 q_0^4) q_1^2 - 
      3 (-2 + 14 q_0^2 + 3 q_0^4) q_1^4) \\&- 3 q^4 (1 - 4 q_1^2 + q_0^4 (3 + 6 q_1^2) + 
      4 q_0^2 (-1 + 7 q_1^2)))/((-1 + q^2 + q_0^2 - 2 q q_0 q_1\\& + 
     q_1^2)^2 (-1 + q^4 + q_0^4 + 8 q q_0 q_1 - 4 q_0^2 q_1^2 + q_1^4 + 
     2 q^2 (-2 q_1^2 + q_0^2 (-2 + q_1^2))))
      \end{split}
 \end{equation}

 \begin{equation}
 \begin{split}
Y_2^{(3)}({\bf q})=&(-1 + 2 q^2 - 2 q^6 + q^8 + 2 q_0^2 - 6 q^2 q_0^2 + 6 q^4 q_0^2 - 
   2 q^6 q_0^2 - 3 q_0^4 + 12 q^2 q_0^4 - 9 q^4 q_0^4 - 6 q_0^6\\& - 
   4 q q_0 ((-1 + q^2)^3 - 3 (-1 + q^4) q_0^2 - 3 (2 + q^2) q_0^4) q_1 + 
   2 (-(-1 + q^2)^3\\
   &+ 3 (-1 + q) (1 + q) (1 - 9 q^2 + 2 q^4) q_0^2 - 
      3 (-2 + 14 q^2 + 3 q^4) q_0^4) q_1^2\\& - 
   12 q q_0 (1 - q^4 + (1 - 12 q^2 + q^4) q_0^2) q_1^3 - 
   3 (1 - 4 q_0^2 + q^4 (3 + 6 q_0^2)\\& + 4 q^2 (-1 + 7 q_0^2)) q_1^4 + 
   12 q (2 + q^2) q_0 q_1^5 - 
   6 q_1^6)/((-1 + q^2 + q_0^2\\& - 2 q q_0 q_1 + q_1^2)^2 (-1 + q^4 + q_0^4 + 
     8 q q_0 q_1 - 4 q_0^2 q_1^2 + q_1^4 + 
     2 q^2 (-2 q_1^2 + q_0^2 (-2 + q_1^2))))
     \end{split}
     \end{equation}

     \begin{equation}\label{eq:yp32}
     \begin{split}
Y_{01}^{(3)}({\bf q})=&(2 (q - q_0 q_1)^2 (3 q^5 - 6 q^4 q_0 q_1 + 2 q_0 q_1 (1 + q_0^2 + q_1^2) - 
     2 q^2 q_0 q_1 (-2 + 3 q_0^2 + 3 q_1^2)\\& - 
     q (-1 + 3 q_0^4 + 2 q_1^2 + 3 q_1^4 + q_0^2 (2 - 4 q_1^2)) + 
     2 q^3 (2 - 3 q_1^2 + q_0^2 (-3 + 9 q_1^2))))/((-1 + q^2 + q_0^2\\& - 
     2 q q_0 q_1 + q_1^2)^2 (-1 + q^4 + q_0^4 + 8 q q_0 q_1 - 4 q_0^2 q_1^2 + 
     q_1^4 + 2 q^2 (-2 q_1^2 + q_0^2 (-2 + q_1^2))))    \end{split}
 \end{equation}

 \begin{equation}
 \begin{split}
Y_{12}^{(3)}({\bf q})=&-((2 (-q q_0 + 
       q_1)^2 (-2 q q_0 (1 + q^2 + q_0^2) + (-1 + 3 q^4 + 2 q_0^2 + 
          3 q_0^4 + q^2 (2 - 4 q_0^2)) q_1\\& + 
       2 q q_0 (-2 + 3 q^2 + 3 q_0^2) q_1^2 + 
       2 (-2 + 3 q_0^2 + q^2 (3 - 9 q_0^2)) q_1^3 + 6 q q_0 q_1^4 - 
       3 q_1^5))/((-1 + q^2 \\&+ q_0^2 - 2 q q_0 q_1 + q_1^2)^2 (-1 + q^4 + 
       q_0^4 + 8 q q_0 q_1 - 4 q_0^2 q_1^2 + q_1^4 + 
       2 q^2 (-2 q_1^2 + q_0^2 (-2 + q_1^2)))))
           \end{split}
 \end{equation}

 \begin{equation}
 \label{app:eq:last_Y}
 \begin{split}
Y_{02}^{(3)}({\bf q})=&(2 (q_0 - q q_1)^2 (3 q_0^5 - 6 q q_0^4 q_1 + 2 q q_1 (1 + q^2 + q_1^2) - 
     2 q q_0^2 q_1 (-2 + 3 q^2 + 3 q_1^2)\\& - 
     q_0 (-1 + 3 q^4 + 2 q_1^2 + 3 q_1^4 + q^2 (2 - 4 q_1^2)) + 
     2 q_0^3 (2 - 3 q_1^2 + q^2 (-3 + 9 q_1^2))))/((-1 + q^2\\& + q_0^2 - 
     2 q q_0 q_1 + q_1^2)^2 (-1 + q^4 + q_0^4 + 8 q q_0 q_1 - 4 q_0^2 q_1^2 + 
     q_1^4 + 2 q^2 (-2 q_1^2 + q_0^2 (-2 + q_1^2)))).
\end{split}
\end{equation}

\end{document}